\documentclass[11pt]{article}
\usepackage{macros}
\usepackage{subfiles}
\graphicspath{{./Figures/}}

\title{Gauge origami and quiver W-algebras }
\author{Taro Kimura and Go Noshita}
\begin{document}
\begin{titlepage}
\vspace*{1cm}
\vskip 12mm
\begin{center}
    {\LARGE \bf{Gauge origami and quiver W-algebras}\par}          
\vskip 3cm    
\begin{center}
{\Large Taro Kimura${}^{\diamondsuit}$ and Go Noshita${}^{\spadesuit}$}
\end{center}
\vskip 2cm 
{\it $\diamondsuit$ Institut de Mathématiques de Bourgogne,}
{\it Université de Bourgogne, CNRS, Dijon, France }\\
{\it $\spadesuit$ Department of Physics, The University of Tokyo, Tokyo, Japan}
\end{center}
\begin{center}
{E-mail: \href{mailto:taro.kimura@u-bourgogne.fr}{\texttt{taro.kimura@u-bourgogne.fr}}, \\ 
 \href{mailto:noshita-go969@g.ecc.u-tokyo.ac.jp}{\texttt{noshita-go969@g.ecc.u-tokyo.ac.jp}}}
\end{center}
\vfill
\begin{abstract}
We explore the quantum algebraic formalism of the gauge origami system in $\mathbb{C}^{4}$, where D2/D4/D6/D8-branes are present. We demonstrate that the contour integral formulas have free field interpretations, leading to the operator formalism of $qq$-characters associated with each D-brane. The $qq$-characters of D2 and D4-branes correspond to screening charges and generators of the affine quiver W-algebra, respectively. On the other hand, the $qq$-characters of D6 and D8-branes represent novel types of $qq$-characters, where monomial terms are characterized by plane partitions and solid partitions. The composition of these $qq$-characters yields the instanton partition functions of the gauge origami system, eventually establishing the BPS/CFT correspondence.

Additionally, we demonstrate that the fusion of $qq$-characters of D-branes in lower dimensions results in higher-dimensional D-brane $qq$-characters. We also investigate quadratic relations among these $qq$-characters. Furthermore, we explore the relationship with the representations, $q$-characters, and the Bethe ansatz equations of the quantum toroidal~$\mathfrak{gl}_{1}$. This connection provides insights into the Bethe/Gauge correspondence of the gauge origami system from both gauge-theoretic and quantum-algebraic perspectives.

We finally present conjectures regarding generalizations to general toric Calabi--Yau four-folds. These generalizations imply the existence of an extensive class of $qq$-characters, which we call BPS $qq$-characters. These BPS $qq$-characters offer a new systematic approach to derive a broader range of BPS/CFT correspondence and Bethe/Gauge correspondence.
\end{abstract}
\vfill
\end{titlepage}

\tableofcontents


\section{Introduction and summary}\label{sec:introduction}
\paragraph{Rise of BPS/CFT correspondence}
Towards a non-perturbative understanding of string theory and quantum field theory, exact computations of physical observables have played significant roles. Among these observables, the instanton partition functions of supersymmetric gauge theories have emerged as significant quantities that can be precisely computed through localization techniques~\cite{Nekrasov:2002qd,Nekrasov:2003rj,Nakajima:2003pg,Nakajima:2003uh,Nakajima:2005fg,Pestun:2007rz,Pestun:2016zxk}. These partition functions exhibit remarkable characteristics, particularly in the realm of rich geometric and algebraic properties. Our specific interest lies in the algebraic structures of these partition functions, which have, in turn, led to the revelation of a novel duality known as the BPS/CFT correspondence~\cite{Nakajima:1999,Nekrasov:2013xda,Nekrasov:2015wsu,Nekrasov:2016qym,Nekrasov:2016ydq,Nekrasov:2017rqy,Nekrasov:2017gzb,Nekrasov:2017cih,Nekrasov:2018xsb}. The BPS/CFT correspondence claims that correlation functions of BPS observables in supersymmetric gauge theories are dual to correlation functions of infinite-dimensional algebras. A well-known example of this duality is the Alday--Gaiotto--Tachikawa (AGT) duality~\cite{Gaiotto:2009we,Alday:2009aq,Wyllard:2009hg} (see also~\cite{LeFloch:2020uop} for a review), which established a connection between instanton partition functions of 4d $\mathcal{N}=2$ theories and conformal blocks of 2d conformal field theories (CFT), where both theories originate from a 6d $\mathcal{N}=(2,0)$ theory compactified on a Riemann surface. A 5d lift-up of the AGT duality called the 5d AGT duality~\cite{Awata:2009ur,Awata:2010yy,Taki:2014fva}, was subsequently discovered and the partition functions are dual to correlation functions of quantum algebras. Over the last decade, efforts have been dedicated to generalizing both the gauge theory side and the infinite-dimensional algebra side to achieve a more comprehensive understanding of the BPS/CFT correspondence~\cite{miki2007q,FFJMM1,Litvinov:2016mgi,Gaiotto:2017euk,Prochazka:2017qum,Prochazka:2018tlo,Rapcak:2018nsl,Rapcak:2020ueh,Li:2020rij,Galakhov:2021vbo,Noshita:2021ldl}.   

\paragraph{Generalized gauge theory}
The usual instantons appear as topologically nontrivial field configurations minimizing the action of the Yang--Mills theory on $\mathbb{R}^{4}$ and the moduli space comes from the ADHM construction \cite{Atiyah:1978ri,Atiyah:1984px}. After introducing a nontrivial background flux called $\Omega$-background, Nekrasov showed that the instanton partition function takes the form of $\mathcal{Z}=\sum_{\lambda}\mathfrak{q}^{|\lambda|}\mathcal{Z}[\lambda]$ where the terms are a summation of two-dimensional partitions \cite{Lossev:1997bz,Moore:1997dj,Nekrasov:2002qd,Nekrasov:2003rj}. Over the last few years, researches on generalizations of the Yang--Mills instantons have been conducted and we now know that there are higher dimensional generalizations of instantons having ADHM-like constructions \cite{Acharya:1997gp,Baulieu:1997jx,Nekrasov:2009JJM,Szabo:2022zyn,Nekrasov:2023xzm}. The instanton partition functions are given as a statistical sum of random partitions as $\mathcal{Z}=\sum_{\Lambda}\mathfrak{q}^{|\Lambda|}\mathcal{Z}[\Lambda]$, where $\Lambda$ is a random partition. Considering 6d and 8d theories, the random partitions $\Lambda$ will be plane partitions \cite{Jafferis:2007sg,Cirafici:2008sn,Nekrasov:2009JJM,Awata:2009dd} and solid partitions \cite{Nekrasov:2017cih,Nekrasov:2018xsb}, respectively (see also \cite{Kanno:2020ybd}). In type IIB string theory on $\mathbb{R}^{10}$, such kind of setup appears as the low energy limit of the worldvolume theory of the D1-branes probing the $\D(2p+1)$-branes $(p=2,3,4)$~\cite{Witten:1994tz,Witten:1995gx,Douglas:1995bn,Douglas:1996uz}: 2d ($p=2$), 3d ($p=3$), 4d $(p=4)$ partitions. Mathematically, the $p=3$ case gives the equivariant Donaldson--Thomas invariants of $\mathbb{C}^{3}$, while the $p=4$ case is called the magnificent four model and gives the equivariant Donaldson--Thomas invariants of $\mathbb{C}^{4}$.

Another direction of generalizations of the Yang--Mills instantons is the \emph{generalized gauge theory}~\cite{Nekrasov:2015wsu,Nekrasov:2016qym,Nekrasov:2016ydq}, which is a theory defined on several, generally, intersecting components as $\mathcal{S}=\bigcup_{i}\mathcal{S}_{i}$ where on each space-time component $\mathcal{S}_{i}$ there is an original field theory. On each space-time component $\mathcal{S}_{i}$, we have a gauge group $G_{i}$ and at the intersection $\mathcal{S}_{i}\cap\mathcal{S}_{j}$ we have bifundamental fields transforming under $G_{i}\times G_{j}$, and thus it could be understood as a generalized quiver gauge theory. From the viewpoint of each field theory on $\mathcal{S}_{i}$, the intersection of other components plays the role of defects. The first example of such general gauge theory is the so-called spiked instanton system, which was introduced in~\cite{Nekrasov:2015wsu,Nekrasov:2016qym}. The spiked instantons arise from the low energy limit of D1-branes probing \emph{intersecting} D5 (and anti-D5)-branes. Later, it was generalized to D1-branes probing intersecting D7-branes in~\cite{Pomoni:2021hkn} (see also~\cite{Fasola:2023ypx,Cao:2023lon}), and the arising instantons are called the tetrahedron instantons. Note that these instantons are generalizations of the higher dimensional instantons introduced in the previous paragraph. 

\paragraph{Gauge origami}
The setups discussed previously are collectively called the \emph{gauge origami} and the arising partition function is called the \emph{gauge origami partition function} \cite{Nekrasov:2016ydq}. Consider a type IIB theory on $Z\times \mathcal{C}$ where $Z$ is a toric Calabi--Yau four-fold and $\mathcal{C}=\mathbb{C},\,\mathbb{C}^{\times},\,\mathbb{T}^{2}$ and the low energy limit of the D1-branes probing $\D(2p+1)$-branes. The D1-branes wrap $\mathcal{C}$ while the $\D(2p+1)$-branes wrap the product of $\mathcal{C}$ and non-compact toric submanifolds of $Z$ in a way preserving a suitable number of supersymmetries. Depending on $\mathcal{C}$, the arising partition function becomes rational, trigonometric, and elliptic, respectively. The gauge origami partition function generally takes the form as
\bea
\mathcal{Z}_{\text{折紙}}=\sum_{\{\Lambda_{i}^{(\alpha)}\}}\underline{\mathfrak{q}}^{|\underline{\vec{\Lambda}}|}\prod_{(i,\alpha)}\mathcal{Z}[\Lambda_{i}^{(\alpha)}]\prod_{(i,\alpha)\neq (j,\beta)}\mathcal{Z}(\Lambda_{i}^{(\alpha)}\,|\,\Lambda_{j}^{(\beta)})
\eea
where $i$ labels the possible types of toric submanifolds and $\alpha$ labels the number of D-branes wrapping them. The partition function is a summation of random \emph{BPS crystals}, which are generalizations of the partitions, and they are denoted $\Lambda^{(\alpha)}_{i}$. These crystals are expected to be, generally, truncations of four-dimensional BPS crystals which are generalizations of the three-dimensional BPS crystals \cite{Szendroi:2007nu,Mozgovoy2008OnTN,Ooguri:2009ijd,Nagao:2010kx}. The $\mathcal{Z}[\Lambda_{i}^{(\alpha)}]$ part comes from the contribution of each $\D(2p+1)$-branes while the $\mathcal{Z}(\Lambda_{i}^{(\alpha)}\,|\,\Lambda_{j}^{(\beta)})$ part comes from the bifundamental contributions at the junctions. In this paper, we mainly focus on the case when $Z=\mathbb{C}^{4}$ and $\mathcal{C}=\mathbb{C}^{\times }\simeq \mathbb{R}\times \mathbb{S}^{1}$ which gives the K-theoretic magnificent four \cite{Nekrasov:2017cih,Nekrasov:2018xsb}, tetrahedron instanton \cite{Pomoni:2021hkn}, and spiked instanton \cite{Nekrasov:2016gud} setups. From the string theory viewpoint, we take the T-duality of the D1--$\D(2p+1)$ system and consider a $\D0$--$\D(2p)$ system, where each $\D(2p)$-brane gives a $(2p+1)$-dimensional gauge theory and the $\D0$-branes play the roles of instantons. 

\paragraph{Quantum algebra of gauge origami}
An interesting property of the gauge origami partition function is the existence of an infinite set of non-perturbative Dyson--Schwinger equations related to the symmetries of adding and removing instantons~\cite{Nekrasov:2015wsu}. The $qq$-characters are physical observables characterizing them, and interestingly, there is an operator formalism of them called the quiver W-algebra \cite{Kimura:2015rgi,Kimura:2016dys,Kimura:2017hez,Kimura:2019xzj} (see \cite{Kimura:2020jxl} for a review). From the gauge theoretic viewpoint, such algebras are associated with the Dynkin diagram corresponding to the quiver structure of the gauge theory. In the gauge origami formalism, the $qq$-characters and quiver W-algebras appear from the gauge origami system in $Z=\mathbb{C}^{2}\times \mathbb{C}^{2}/\Gamma$, where $\Gamma$ denotes the finite subgroup of $\mathrm{SU}(2)$ associated with the quiver structure through the McKay correspondence. Placing $\D4$-branes wrapping the $\mathbb{C}^{2}$ part gives the 5d $\mathcal{N}=1$ (affine) quiver gauge theories and the $\D4$-branes wrapping the $\mathbb{C}^{2}/\Gamma$ give the $qq$-characters or the quiver W-algebras of the theory. Physically, they are codimension four defects of the quiver gauge theory \cite{Kim:2016qqs}.

One of the reasons why such algebras are considered important is because they give the BPS/CFT correspondence for the present case. One can construct screening currents from the quiver structure and the vacuum expectation value of them gives the instanton partition function of the quiver gauge theory. Moreover, after defining the highest weight, the commutativity with the screening charges determines the generator of the quiver W-algebra uniquely and they are the operator version of the $qq$-characters. These quiver W-algebras include the well-known $q$-Virasoro \cite{Shiraishi:1995rp}, $q$-W$_{N}$ \cite{Awata:1995zk,Awata:1996dx}, and Frenkel--Reshetkhin's deformed W-algebras \cite{Frenkel:1997CMP,Frenkel:1998ojj}.

Although the quiver W-algebras indeed gave a way to discuss the BPS/CFT correspondence from the quantum algebraic viewpoint, the applicable theory is still limited and needs to be extended. For example, while the $qq$-characters associated with $Z=\mathbb{C}^{2}/\Upsilon\times \mathbb{C}^{2}/\Gamma$ were studied in terms of partition functions in \cite{Nekrasov:2012xe,Nekrasov:2013xda,Nekrasov:2015wsu,Jeong:2018qpc,Jeong:2021rll}, it seems that the complete operator formalism of such cases is still missing in the literature\footnote{We hope to report it in a future work \cite{Kimura-Noshita} (see also \cite{Bourgine:2019phm}).}. Moreover, the quiver W-algebra is only applicable to discuss two stacks of D-branes in transverse directions while we have multiple intersecting D-branes in the gauge origami system. Based on recent studies such as the tetrahedron instanton system we also have D-branes wrapping not only complex two-dimensional surfaces but also complex three-dimensional manifolds. Most importantly, we need to generalize the operator formalism to describe gauge origami systems associated with general toric Calabi--Yau four-folds~$Z$. See \cite{Cao:2019tvv,Kimura:2022zsm,Cao:2023gvn,Szabo:2023ixw,Nekrasov:2023nai} for discussions along this direction. 

The goal of this paper is to fill in this gap by generalizing the concept of quiver W-algebras and showing the BPS/CFT correspondence of the gauge origami system associated with $\mathbb{C}^{4}$. We will only give some conjectures for generalizations to toric Calabi--Yau four-folds and details are postponed for future work \cite{Kimura-Noshita}.

\subsection*{Summary of the results}
Let us summarize the main results of this paper. 
\paragraph{D2/D4/D6 $qq$-characters}
We introduce a D2 gauge origami system which we call the coupled vortex system (section~\ref{sec:cplvortex_partitionfunct}). This theory comes from the low energy limit of the D0-branes probing the D2-branes spanning $\mathbb{C}_{a}\subset\mathbb{C}^{4}$ ($a=1,2,3,4$; see section~\ref{sec:equiv-index} for notations), which is obtained through a dimensional reduction of the D4-brane gauge origami system.
This coupled vortex system also plays an important role similar to higher-dimensional branes in the gauge origami systems.
    
Considering a stack of D2-branes, D4-branes and D6-branes spanning transversely in $\mathbb{C}^{4}$, we derive the D2/D4/D6 $qq$-characters in terms of partition functions in section~\ref{sec:qqpartitionfunct}. The D4 $qq$-character is known to be characterized by 2d partitions. The D2 $qq$-characters are $qq$-characters where the monomial terms are 1d partitions, while the D6 $qq$-characters are $qq$-characters whose monomial terms are 3d partitions (plane partitions). 
\paragraph{D-brane vertex operators}
We introduce vertex operators $\mathsf{A}(x),\mathsf{S}_{a}(x),\mathsf{X}_{A}(x),\mathsf{W}_{\bar{a}}(x),\mathsf{Z}(x)$ corresponding to D0, D2, D4, D6, and D8-branes wrapping the possible subspaces $\text{pt}$, $\mathbb{C}_{a}$, $\mathbb{C}_{A}^{2}$, $\mathbb{C}_{\bar{a}}^{3}$, $\mathbb{C}^{4}$ ($a\in\four,A\in\six$), respectively.
The main result of section~\ref{sec:freefieldintegral} is as follows.
\begin{theorem}[Theorem~\ref{thm:freefieldconclusion}]
For each D-brane ($\D0,\D2,\D4,\D6,\D8$), we define the corresponding vertex operators as
\begin{align}
    \renewcommand\arraystretch{1.2}{
        \begin{tabular}{|c|c|c|}\hline
            D-brane & space-time & vertex operator \\
           \hline\hline  D0-brane  & $\text{pt}\times\mathbb{S}^{1}$ & $\mathsf{A}(x)$ \\
           \hline D2-brane   &  $\mathbb{C}_{a}\times \mathbb{S}^{1}$ ($a\in\four$)  & $\mathsf{S}_{a}(x)$\\
           \hline D4-brane &  $\mathbb{C}^{2}_{A}\times \mathbb{S}^{1}$ ($A\in\six$) & $\mathsf{X}_{A}(x)$\\
           \hline D6-brane &  $\mathbb{C}^{3}_{\bar{a}}\times \mathbb{S}^{1}$ ($a\in\four$) &  $\mathsf{W}_{\bar{a}}(x)$\\
           \hline D8-brane & $\mathbb{C}^{4}\times \mathbb{S}^{1}$  & $\mathsf{Z}(x)$ \\\hline
         \end{tabular}}
\end{align}
We have multiple copies of vertex operators if there are multiple ways that the D-branes can wrap. Then, the contour integral formula of the $k$-instanton partition function takes the form as
\bea
    \mathcal{Z}_k = \oint \prod_{I=1}^{k}\frac{dx_{I}}{2\pi\iota x_{I}} \bra{0} \prod_{I=1}^{k}\mathsf{A}(x_{I})^{-1} :\prod_{i}\mathsf{V}_{i}(v_{i}): \ket{0},
\eea
where $\mathsf{V}_{i}(x)$ is an operator written from $\{\mathsf{S}_{a}(x),\mathsf{X}_{A}(x),\mathsf{W}_{\bar{a}}(x),\mathsf{Z}(x)\}$.
\end{theorem}
The vertex operators have a $q$-Cartan matrix understanding associated with quivers and thus are generalizations of the conventional quiver W-algebras (see section~\ref{sec:quiverstructure}). 

\paragraph{Operator formalism of $qq$-characters}
We introduce the operator formalism of the $qq$-characters of the gauge origami system and show the BPS/CFT correspondence. For each D$(2p)$-brane, we can associate a $qq$-character. The D2 $qq$-characters $\mathscr{Q}_{a}(x)\,(a\in\four)$ are four copies of the screening charge of the $\widehat{A}_{0}$ quiver W-algebra \cite{Kimura:2015rgi} and the D4 $qq$-characters $\mathsf{T}_{A}(x)\,(A\in\six)$ are six copies of the generator of the $\widehat{A}_{0}$ quiver W-algebra. The D6 $qq$-characters $\mathsf{T}_{\bar{a}}(x)\,(a\in\four)$ are the new $qq$-characters where the monomial terms are labeled by plane partitions. These $qq$-characters represent the quantum algebras associated with complex $1,2,3$-dimensional submanifolds. We will see that their compositions indeed give the gauge origami partition function which shows the BPS/CFT correspondence. 
\begin{theorem}[Theorems~\ref{thm:cplvortexBPSCFT}, \ref{thm:spiked-qq-BPSCFT}, \ref{thm:tetra-origamiBPSCFT}]\label{thm:BPS/CFTintro}
The gauge origami partition function is in general given as a correlation function of the $qq$-characters,
\bea
    \mathcal{Z}_{\text{折紙}}=\bra{0}\prod_{(i,\alpha)}\mathsf{T}_{i}(x_{i,\alpha})\ket{0},\quad i\in\four\oplus\six\oplus\four^{\vee}.
\eea
\end{theorem}
For the case $\mathcal{C} = \mathbb{T}^2$, the corresponding partition function is given by a torus correlator of the $qq$-characters instead of the vacuum expectation value.
We can also use the elliptic version of the vertex operators discussed in section~\ref{sec:ellipticWorigami} to discuss the case $\mathcal{C} = \mathbb{T}^2$.

In the context of quiver W-algebra and also the quantum integrable system, the commutation relation between the vertex operators plays a fundamental role.
We obtain an interesting commutativity among $qq$-characters: The D2, D4, D6 $qq$-characters associated with subspaces of $\mathbb{C}^{4}$ that are transverse with each other commutes. 
\begin{theorem}[Theorem~\ref{thm:tetrascreening}]
The $qq$-characters associated with the elements $i,j\in\four\oplus\six\oplus\four^{\vee}$ commute with each other up to trivial zero modes (see \eqref{eq:weakcommute}) when $i$ and $j$ are transverse with each other:
\beq
    \mathsf{T}_{i}(x)\mathsf{T}_{j}(x')-f_{ij}(x,x')\mathsf{T}_{j}(x')\mathsf{T}_{i}(x)=0 \quad \Longleftrightarrow \quad i \cap j =\emptyset,
\eeq
where $f_{ij}(x,x')$ are zero mode factors.
\end{theorem}    
Moreover, infinite fusion of D2 (D4) $qq$-characters give D4 (D6) $qq$-characters (Thm.~\ref{thm:D2toD4fusion}, \ref{thm:D4toD6fusion}).
Using the fusion process and fusing an infinite number of D6 $qq$-characters, we define the D8 $qq$-characters (section~\ref{sec:fusionD6toD8}). We also show the BPS/CFT correspondence of the magnificent four model up to sign factors. 
\begin{theorem}[Theorem~\ref{thm:D8magnificentBPSCFT}]
The composition of the $\D8$ $qq$-characters gives the partition function of higher rank magnificent four system up to sign factors, which establishes the BPS/CFT correspondence for the magnificent four:
\bea
    \bra{0}\mathsf{T}_{\four;a_{N}}(x_{N})\cdots \mathsf{T}_{\four;a_{1}}(x_{1})\ket{0}&=\sum_{\rho^{(1)},\cdots ,\rho^{(N)}}\mathfrak{q}^{|\rho|}\prod_{i=1}^{N}\mathcal{Z}_{\four;a_{i}}^{\D8}[\rho^{(i)},K_{i}]\prod_{j>i}\mathcal{Z}_{\text{1-loop}}^{\D8\tbar\D8}(x_{i},K_{i}\,|\,x_{j},K_{j})\\
    &\qquad \qquad\times\prod_{j>i}\mathcal{Z}^{\D8\tbar\D8}_{K_{i}|K_{j}}(x_{i},\rho^{(i)}\,|\,x_{j},\rho^{(j)}).
\eea
\end{theorem}
We then give conjectures related to generalizations to toric Calabi--Yau four-folds (section~\ref{sec:CY4}, \ref{sec:CY3timesC}, \ref{sec:CY2CY2}, \ref{sec:generalD6qq}, Conj.~\ref{conj:CY3}, \ref{conj:CY4}, \ref{conj:BPSqq}). The $qq$-characters appearing in such generalizations are called BPS $qq$-characters.
\paragraph{Quantum toroidal algebra, Bethe ansatz}
We also show that the $qq$-characters associated with $\mathbb{C}^4$ geometry have a correspondence with quantum toroidal algebras (section~\ref{sec:QTgl1}). In particular, we show that the D2/D4/D6 system corresponds to the vector/Fock/MacMahon representation of quantum toroidal $\mathfrak{gl}_1$.
We also consider the quiver quantum toroidal algebras associated with toric Calabi--Yau three-folds and construct generic $qq$-characters, that we call the BPS $qq$-characters (section~\ref{sec:BPSqqQTA}).
In section~\ref{sec:Betheansatz}, through the semi-classical analysis of the gauge origami system, we obtain the universal form of the Bethe ansatz equations (BAEs) for the gauge origami system on $\mathbb{C}^4$.
\begin{theorem}[Theorem~\ref{thm:generalBetheansatz}]%
The BAE obtained as the saddle point equation of the D2/D4/D6 system partition function is generally written as follows,
\beq
        1=-\mathfrak{q}\frac{\mathsf{Q}(\mathsf{q}_{1}^{-1}x)\mathsf{Q}(\mathsf{q}_{2}^{-1}x)\mathsf{Q}(\mathsf{q}_{3}^{-1}x)}{\mathsf{Q}(\mathsf{q}_{1}x)\mathsf{Q}(\mathsf{q}_{2}x)\mathsf{Q}(\mathsf{q}_{3}x)},
\eeq
where the three parameters obey $\mathsf{q}_{1}\mathsf{q}_{2}\mathsf{q}_{3}=1$ and the corresponding $\mathsf{Q}$-functions for the D2/D4/D6 system are given by \eqref{eq:Q-func_D2}, \eqref{eq:Q-func_D4}, \eqref{eq:Q-func_D6}, \eqref{eq:D4foldedBethe}, \eqref{eq:D6foldedBethe}.
\end{theorem}
In fact, it turns out that this BAE is associated with quantum toroidal $\mathfrak{gl}_1$. We can further discuss the representation dependence by imposing the flavor contribution.
\begin{theorem}[Theorem~\ref{thm:BAEwithPolynom}]
Let $\mathsf{Q}(x)$, $a(x)$, $d(x)$ be polynomials in $x \in \mathbb{C}^\times$ with the parameters obeying $\mathsf{q}_{1}\mathsf{q}_{2}\mathsf{q}_{3}=1$.
The saddle point equation of the D2 gauge origami system with flavor D8-$\overline{\D8}$ branes gives the BAE involving the additional polynomials $a(x)$ and $d(x)$ specifying the representation of quantum toroidal $\mathfrak{gl}_1$,
\bea
1=-\mathfrak{q}\frac{a(x)}{d(x)}\frac{\mathsf{Q}(\mathsf{q}_{1}^{-1}x)\mathsf{Q}(\mathsf{q}_{2}^{-1}x)\mathsf{Q}(\mathsf{q}_{3}^{-1}x)}{\mathsf{Q}(\mathsf{q}_{1}x)\mathsf{Q}(\mathsf{q}_{2}x)\mathsf{Q}(\mathsf{q}_{3}x)}.
\eea
\end{theorem}
This structure is analogous to Yangian and quantum affine algebra, where the BAE has a universal form, which does not depend on the representation, while the (highest-weight) representation data appear only in the $a$ and $d$ polynomials.

In addition to the saddle point equation, that gives rise to the BEA, we also study the $q$-character, the semi-classical reduction of the $qq$-character.
For generic $q$-characters, we have the following relation.
\begin{theorem}[Theorem~\ref{thm:q-ch_commute}]
For any $q$-characters obtained in the semi-classical limit, $\mathcal{T}(x),\mathcal{T}'(x)\in\{\hat{\mathscr{Q}}_{1,2,3}(x),\hat{\mathsf{T}}_{12,23,13}(x),\hat{\mathsf{T}}^{K}_{123}(x)\}$, we have
\beq
[\mathcal{T}(x),\mathcal{T}'(x')]=0.
\eeq    
\end{theorem}    
From this point of view, one may obtain a wide class of commuting operators from the $q$-character, that would be identified with the commuting Hamiltonians of the corresponding quantum integrable system.
We also mention generalizations of these arguments to other Calabi--Yau geometries and discuss the corresponding BAE.

\paragraph{Geometric realization of $qq$-characters}
Although we have focused on the algebraic aspects so far, the gauge origami construction also provides geometric insights on $qq$-character and the underlying quantum algebraic structures. For example, for the Fock module of quantum toroidal $\mathfrak{gl}_1$, the $qq$-character is given by the equivariant integral over the Hilbert scheme of points on $\mathbb{C}^2$, which is given as a quiver variety associated with $\widehat{A}_0$ quiver (Prop.~\ref{prop:qq_ch_geom1}). Considering the higher rank framing space, the quiver variety gives rise to the Quot scheme, and one can then obtain the tensor product module. Generalizing this construction, we have the geometric realization of the $qq$-character of the MacMahon module of quantum toroidal $\mathfrak{gl}_1$.
\begin{theorem}[Theorem~\ref{thm:qq_ch_geom2}]
    Let $x \in \mathbb{C}^\times$.
    The $qq$-character of the MacMahon module of quantum toroidal $\mathfrak{gl}_1$ is given as follows,
    \begin{equation}
    \mathsf{T}_{1;x}[\mathbf{Y}] = \sum_{v \ge 0} \mathsf{T}_{1,v;x}[\mathbf{Y}]
    \, , \qquad
    \mathsf{T}_{1,v;x}[\mathbf{Y}] =  \int_{[\mathfrak{M}_{1,v}]^\text{vir}} \operatorname{ch} \wedge^\bullet \mathbf{Y}_{1,v}^\vee \mathbf{Y} \operatorname{td} \left( T \mathfrak{M}_{1,v} \right)
    \, ,
    \end{equation}
    where the integral is equivariantly taken over the virtual fundamental cycle of the moduli space of $v$ D0 branes on a single D6 brane, isomorphic to the Hilbert scheme of points on $\mathbb{C}^3$, $\mathfrak{M}_{1,v} \cong \operatorname{Hilb}^v(\mathbb{C}^3)$.
    We denote the observable bundle and the formal bundle over $\mathfrak{M}_{1,v}$ by $\mathbf{Y}_{1,v}$ and $\mathbf{Y}$.    
\end{theorem}
The $qq$-character of the tensor product of the MacMahon module can be similarly given by the equivariant integral over the moduli space of the higher rank D6-brane system, which is isomorphic to the Quot scheme, $\mathfrak{M}_{w,v} \cong \operatorname{Quot}^v_{\mathbb{C}^3}(\mathcal{O}^{\oplus w})$.
Physically, this integral computes the codimension six defect partition function of the tetrahedron instanton system.
This formalism should be straightforwardly extended to generic representations of quantum toroidal $\mathfrak{gl}_1$ by replacing the moduli space with the corresponding one.
    

\subsection*{Organization of this paper}

The paper is organized as follows. In section~\ref{sec:multi-dim-part}, we introduce the notations related to 1d, 2d, 3d, and 4d partitions, highlighting their relationships with lower dimensional partitions and defining $q$-coordinates for boxes of the partitions. In section~\ref{sec:gaugeorigamipartitionfunction}, we delve into explicit formulas and properties of the instanton partition functions of the gauge origami system. This section covers various aspects, including the physical setup, contour integral formulas, partition functions, their decompositions, and the concept of $qq$-characters. Free field realizations of the contour integral formulas for each system are given in section~\ref{sec:freefieldintegral}. Vertex operators associated with the D0, D2, D4, D6, D8-branes are introduced. We also discuss the quiver structure and generalizations to toric Calabi--Yau four-folds. Sections~\ref{sec:D2qqcharacter}, \ref{sec:D4qqcharacter}, \ref{sec:D6qqcharacter}, and \ref{sec:D8qq} are dedicated to the study of D2, D4, D6, and D8-brane $qq$-characters, examining their relationships with various gauge origami systems. We discuss the fusion process, generalizations, quadratic relations, and each system's BPS/CFT correspondence. In section~\ref{sec:toroidal_alg}, we delve into quantum toroidal algebras and discuss the relation with BPS $qq$-characters. Section~\ref{sec:Betheansatz} explores applications of the D-brane $qq$-characters, including their role in semi-classical analysis and the Bethe ansatz equations of the gauge origami system. Geometric realization of the $qq$-characters and elliptic generalizations are discussed in section~\ref{sec:geometryqq} and section~\ref{sec:ellipticWorigami}, respectively. In section~\ref{sec:conclusion}, we finally conclude our findings and engage in discussions about the implications of our work. Additionally, we provide some supplementary formulas used in the main sections of the paper in the Appendix.


\section{Multi-dimensional partitions}\label{sec:multi-dim-part}
Let us summarize the notation we use for multi-dimensional partitions. We follow the convention and the terminology of \cite{Nekrasov:2017cih,Nekrasov:2018xsb,Nekrasov:2023nai}.

\subsection{Partitions}

\paragraph{One-dimensional partition} A one-dimensional partition is a row of boxes stacked only in one direction as 
\begin{equation}\label{eq:1dpartition_positive}
\adjustbox{valign=c}{{\scalebox{1.5}{\yng(5)}}}
\end{equation}
Such a kind of partition is specified by a non-negative integer $k\in\mathbb{Z}_{\geq0}$. 

We can also introduce one-dimensional partitions where the boxes are extended semi-infinitely to the left of a border as 
\begin{equation}\label{eq:1dpartition_integer}
    \begin{tikzpicture}
        \draw[->] (-2,0)--(4,0);
        \draw[thick]   (-0.5,-0.5)--(-0.5,1);
        \draw (-2,0.7)--(3,0.7);
        \draw (3,0)--(3,0.7);
        \draw (0.2,0)--(0.2,0.7);
        \draw (0.9,0)--(0.9,0.7);
         \draw (1.6,0)--(1.6,0.7);
          \draw (2.3,0)--(2.3,0.7);
          \draw (-1.2,0)--(-1.2,0.7);
          \node at (-1.6,0.35) {$\cdots$};
    \end{tikzpicture}
\end{equation}
For these kinds of partitions, an integer $k\in\mathbb{Z}$ specifies it. The integer $k$ denotes the number of boxes counted from the border to the right. If no boxes are on the right and the rightest box is left to the border, then $k$ will be negative.

\paragraph{Two-dimensional partition (Young diagram)}
A two-dimensional partition is a non-decreasing sequence of non-negative integers 
\begin{equation}
\begin{split}
    \lambda=(\lambda_{1}\geq \lambda_{2}\geq\ldots\geq\lambda_{\ell(\lambda)}\geq 0),\quad
    |\lambda|=\lambda_{1}+\cdots+\lambda_{\ell(\lambda)}\label{eq:Youngcond}
\end{split}
\end{equation}
where $|\lambda|$ is the size, and $\ell(\lambda)$ is the length of the partition. A partition will be drawn as a collection of square boxes $\Bbox$, and this collection is called the Young diagram.\footnote{In this paper, Young diagrams and 2d partitions are not distinguished. Sometimes, we will also omit the 2d and just say partitions for 2d partitions when it is obvious. Similarly, plane partitions and 3d partitions, solid partitions, and 4d partitions are not distinguished.} For example, a partition $\lambda=(5,3,2,1)$ will look like
\begin{equation}
\adjustbox{valign=c}{{\scalebox{1.5}{\yng(1,2,3,5)}}}
\end{equation}
A box $\Bbox$ positioned in the $i'$th row counted from the bottom 
and $j'$th column counted from the left in the Young diagram is assigned a coordinate $(i,j)$, $1\leq i\leq \ell(\lambda),\,1\leq j\leq \lambda_{i}$. We denote the set of all possible 2d partitions/Young diagrams as $\mathcal{P}$. The transpose is denoted as $\lambda^{\rmT}$ in the usual sense. We call this type of description the $(1,1)$-type description where we followed the terminology in \cite{Nekrasov:2017cih}. This is because we are assigning a one-dimensional partition with only non-negative integers $\lambda_{i}\geq 0$ to each $i=1,\ldots,$. Note that there are two $(1,1)$-type descriptions of the 2d partition, depending on whether we are using $\lambda$ or the transpose $\lambda^{\rmT}$. 
 
\paragraph{Three-dimensional partition (plane partition)}
The plane partition is a stack of cubes that obeys a generalization of the condition \eqref{eq:Youngcond}. We denote the set of all possible plane partitions as $\mathcal{PP}$. There are two descriptions of the plane partition which we call $(2,1)$-type and $(1,2)$-type.
\begin{figure}
    \centering
    \includegraphics[width=8cm]{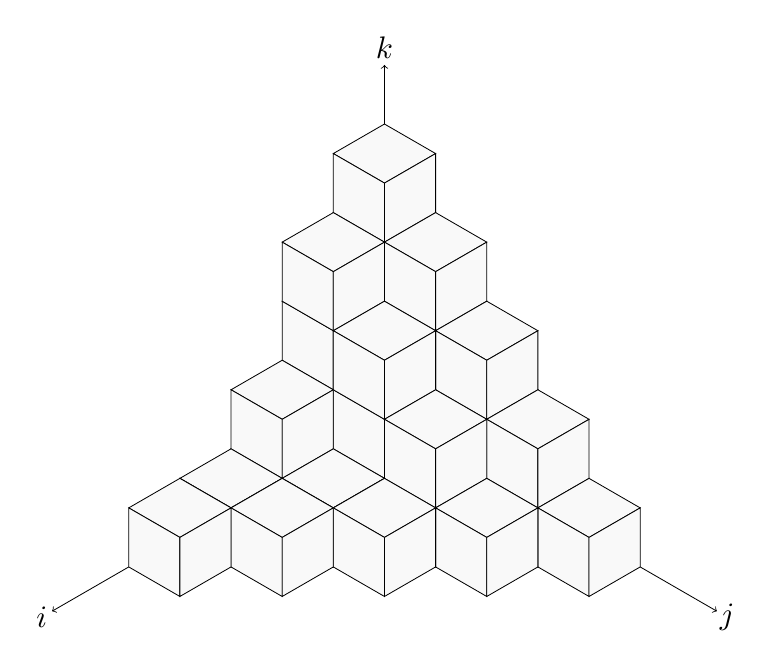}
    \caption{Plane partition}
    \label{fig:pp}
\end{figure}

The $(2,1)$-type description understands the plane partition $\pi$ as a 2d partition $\lambda_{\pi}$ where there is a map mapping each box $\Bbox=(i,j)\in\lambda_{\pi}$ a number $\pi_{i,j}$ obeying the following condition (see Figure \ref{fig:pp} and \ref{fig:pp2112})
\begin{equation}
\pi_{i,j}\geq \pi_{i+1,j},\quad \pi_{i,j}\geq \pi_{i,j+1}.\label{eq:planepartcond}
\end{equation}
The size of the plane partition is defined as the number of cubes included in the configuration 
\begin{equation}
    |\pi|=\sum_{i,j}\pi_{i,j}.
\end{equation}
We have three descriptions depending on which two-dimensional plane we project the plane partition to get the Young diagram. Given three axes $1,2,3$, we use the symbol $\pi$ when the plane partition is projected to a Young diagram in the 12-plane, $\check{\pi}$ when the Young diagram is projected onto the 13-plane, and $\check{\check{\pi}}$ when it is projected onto the 23-plane. This is the 3d analog of the transpose of the Young diagram. Like the Young diagram case, a cube $\cube$ in the plane partition $\pi$ is assigned a coordinate $(i,j,k)\,(i,j,k\geq 1)$ as Figure \ref{fig:pp}. In this description, we have 
\begin{equation}
    (i,j,k)\in\pi\,\,\Leftrightarrow \,\, 1\leq k\leq \pi_{i,j}.
\end{equation}

\begin{figure}
\begin{minipage}{0.45\linewidth}
    \centering
    \includegraphics[width=7cm]{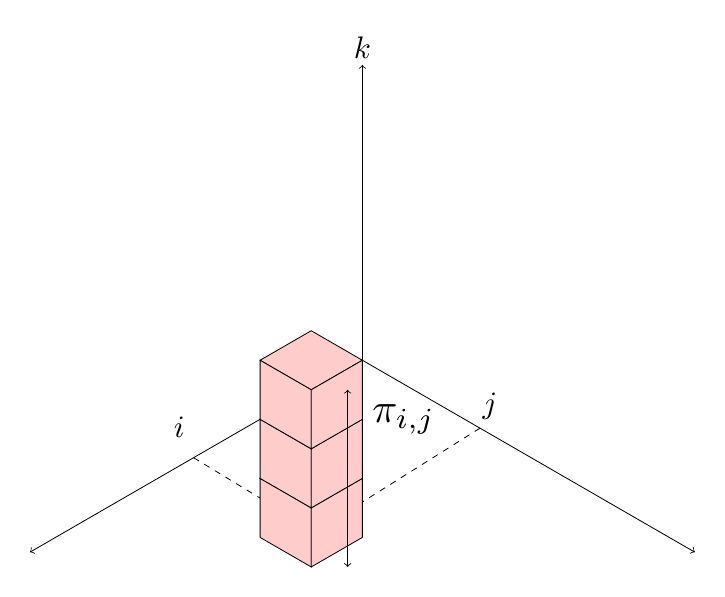}
\end{minipage}
\begin{minipage}
{0.45\linewidth}
    \centering
    \includegraphics[width=7cm]{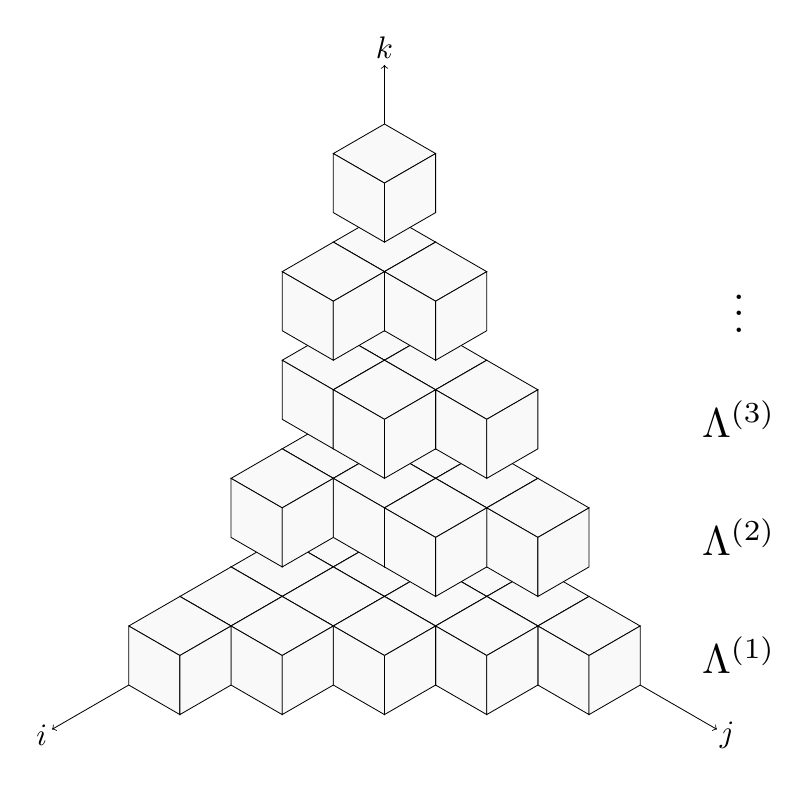}
\end{minipage}
\caption{$(2,1)$-type and $(1,2)$-type description of the plane partition}\label{fig:pp2112}
\end{figure}
The other $(1,2)$-type description understands the plane partition as a non-increasing sequence of Young diagrams (see Figure \ref{fig:pp2112}):
\bea
&\pi=(\Lambda^{(1)},\Lambda^{(2)},\ldots,\Lambda^{(h(\pi))}),\quad \Lambda^{(k)}=(\Lambda_{1}^{(k)},\ldots,\Lambda_{i}^{(k)}\ldots)\\
&\Lambda^{(k)}\succeq\Lambda^{(k+1)},\,\,\forall k
\eea
where $\Lambda^{(k)}\succeq \Lambda^{(k+1)}$ means $\forall (i,j)\in\Lambda^{(k+1)}\Rightarrow (i,j)\in\Lambda^{(k)}$ and $h(\pi)$ is the height of the plane partition which is defined as 
\begin{equation}
    h(\pi)=\min\{k\geq 0\,|\,(1,1,k+1)\not\in\pi\}.
\end{equation}
The size is defined as
\begin{equation}
    |\pi|=\sum_{k}|\Lambda^{(k)}|
\end{equation}
in this description. Similarly, the $\cube=(i,j,k)$ in the plane partition obeys the condition
\begin{equation}
    (i,j,k)\in\pi\,\,\Leftrightarrow \,\, 1\leq j\leq \Lambda_{i}^{(k)}
\end{equation}
Note that we also have three possible descriptions of this type depending on which axis we define the height of the plane partition. 

\paragraph{Four-dimensional partition (solid partition)}
\begin{figure}[ht]
\centering
\scalebox{0.3}{
\begin{tikzpicture}[scale=3]
    \tikzstyle{vertex}=[circle,fill]
	 \node[vertex] (v0) at (0,0) {};
	 \node[vertex] (v1) at (0,1) {};
	 \node[vertex] (v2) at (1,0) {};
	 \node[vertex] (v3) at (1,1) {};
	 \node[vertex] (v4) at (0.23, 0.4) {};
	 \node[vertex] (v5) at (0.23,1.4) {};
	 \node[vertex] (v6) at (1.23,0.4) {};
	 \node[vertex] (v7) at (1.23,1.4) {};
	 \node[vertex] (v8) at (-1,-1) {};
	 \node[vertex] (v9) at (-1,2) {};
	 \node[vertex] (v13) at (-0.66,2.7){ };
	 \node[vertex] (v12) at (-0.66,-0.3){ };
	 \node[vertex] (v10) at (2,-1) {};
	 \node[vertex] (v14) at (2.34,-0.3) {};
	 \node[vertex] (v11) at (2,2) {};
	 \node[vertex] (v15) at (2.34,2.7) {};
  \draw[] (v9) -- (v11) -- (v15) -- (v13) -- (v9);
   \draw[] (v8) -- (v10) -- (v14) ;
    \draw[] (v8) -- (v9);
     \draw[] (v10) -- (v11) ;
      \draw[] (v14) -- (v15);
       \draw[] (v1) -- (v3) -- (v7) -- (v5) -- (v1);
        \draw[] (v0) -- (v2) -- (v6);
        \draw[] (v0) -- (v1);
        \draw[dashed] (v4) -- (v5);
        \draw[dashed] (v4) -- (v6);
        \draw[dashed] (v4) -- (v0);
        \draw[] (v2) -- (v3);
        \draw[] (v6) -- (v7);
\draw[dashed] (v12) -- (v8);
        \draw[dashed] (v12) -- (v14);
        \draw[dashed] (v12) -- (v13);
        \draw[dashed] (v0) -- (v8);
        \draw[dashed] (v2) -- (v10);
        \draw[dashed] (v4) -- (v12);
        \draw[dashed] (v6) -- (v14);
        \draw[dashed] (v1) -- (v9);
        \draw[dashed] (v3) -- (v11);
        \draw[dashed] (v5) -- (v13);
        \draw[dashed] (v7) -- (v15);
        \end{tikzpicture}}
    \caption{4-cube}
    \label{fig:4-cube}
\end{figure}
A solid partition is a four-dimensional analog of the Young diagram and plane partition. It is a stack of 4-cubes (see Figure \ref{fig:4-cube}) obeying similar conditions to \eqref{eq:Youngcond} and \eqref{eq:planepartcond}. We denote the set of all possible solid partitions as $\mathcal{SP}$. We have three ways to describe the solid partition: $(3,1)$, $(2,2)$, and $(1,3)$-types.
\begin{enumerate}
    \item $(3,1)$-type:
    This description is similar to the plane partition's $(2,1)$-type description. We project the solid partition to a plane partition $\pi_{\rho}$ and for each cube $(i,j,k)\in\pi_{\rho}$, a height function $\rho_{i,j,k}$ is defined. The height function obeys the condition
    \begin{equation}
        \rho_{i,j,k}\geq \rho_{i+1,j,k},\quad \rho_{i,j,k}\geq \rho_{i,j+1,k},\quad \rho_{i,j,k}\geq \rho_{i,j,k+1}.
    \end{equation}
    The size is defined as
    \begin{equation}
        |\rho|=\sum_{(i,j,k)\in\pi_{\rho}}\rho_{i,j,k}.
    \end{equation}
    4-cubes $\hcube$ in the solid partition are assigned coordinates in a natural way $(i,j,k,l)$ and obey
    \begin{equation}
    (i,j,k,l)\in \rho \,\,\Leftrightarrow\,\, 1\leq l\leq \rho_{i,j,k}.
    \end{equation}
    Depending on which three axes we project the solid partition to obtain a plane partition, we have four choices for this description. 
    \item
    \begin{figure}
        \centering
        \includegraphics[width=7cm]{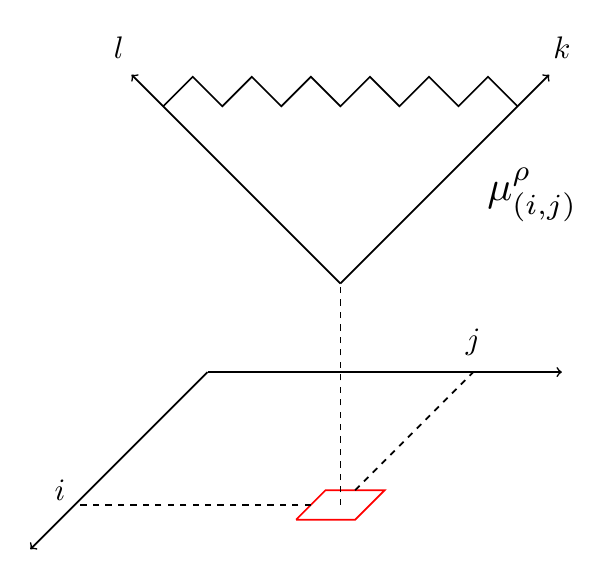}
        \caption{$(2,2)$-type description of solid partition}
        \label{fig:sp22}
    \end{figure}
    $(2,2)$-type: Another way to describe the solid partition is to understand it as a Young diagram $\lambda_{\rho}$ where on each $(i,j)$ there is another Young diagram $\mu^{\rho}_{(i,j)}$ obeying the conditions (see Figure \ref{fig:sp22}): 
    \begin{equation}
        \mu^{\rho}_{(i,j)}\succeq\mu^{\rho}_{(i,j+1)},\quad \mu^{\rho}_{(i,j)}\succeq\mu^{\rho}_{(i+1,j)},
    \end{equation}
    where $\mu^{\rho}_{(i,j)}\succeq\mu^{\rho}_{(i',j')}$ means 
    \begin{equation}
        (k,l)\in \mu^{\rho}_{(i',j')}\,\,\Rightarrow \,\,(k,l)\in\mu^{\rho}_{(i,j)}.
    \end{equation}
    Under this description, the size is defined as
    \begin{equation}
        |\rho|=\sum_{(i,j)\in\lambda_{\rho}}|\mu^{\rho}|
    \end{equation}
    and 
    \begin{equation}
        (i,j,k,l)\in\rho \Leftrightarrow (k,l)\in\mu^{\rho}_{(i,j)}.
    \end{equation}
    We have six possible descriptions of this type depending on the choice of the two axes.
    \item $(1,3)$-type: 
    \begin{figure}
        \centering
        \includegraphics[width=13cm]{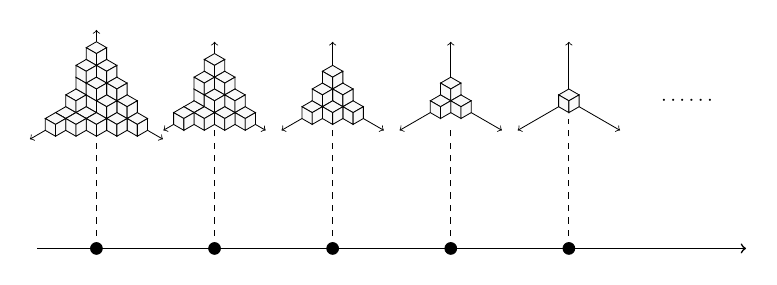}
        \caption{$(1,3)$-type description of the solid partition. The horizontal axis is one of the four axes of the solid partition. The solid partition is decomposed into multiple plane partitions $\Pi^{(1)},\Pi^{(2)},\ldots$.}
        \label{fig:solidplanedecomp}
    \end{figure}
    This description resembles the $(1,2)$-type description of plane partition. The solid partition is understood as non-increasing sequences of plane partitions (see Figure \ref{fig:solidplanedecomp}):
    \begin{equation}
        \rho=(\Pi^{(1)},\Pi^{(2)},\ldots,),\quad \Pi^{(l)}\succeq\Pi^{(l+1)}
    \end{equation}
    where $\Pi^{(l)}\succeq \Pi^{(l+1)}$ means
    \begin{equation}
    (i,j,k)\in \Pi^{(l+1)}\Rightarrow (i,j,k)\in\Pi^{(l)}.
    \end{equation}
    Under this description, the size is defined as
    \begin{equation}
        |\rho|=\sum_{l}|\Pi^{(l)}|
    \end{equation}
    and obviously,
    \begin{equation}
        (i,j,k,l)\in\rho\,\,\Leftrightarrow \,\, (i,j,k)\in\Pi^{(l)}.
    \end{equation}
    We again have four possible descriptions for this type.
\end{enumerate}

\subsection{Coordinates and \texorpdfstring{$q$}{q}-contents}\label{sec:qcoordinate}
For later use, we introduce the $q$-contents ($q$-coordinates) of the multi-dimensional partitions. The boxes (cubes, 4-cubes) of the multi-dimensional partitions are described by at most four components: $(i,j,k,l)$. We introduce four independent parameters\footnote{Later, these parameters are identified as the $\Omega$-background parameters after imposing the Calabi--Yau condition $\sum_{a=1}^{4}\epsilon_{a}=0$.} $\epsilon_{a}\,(a=1,2,3,4)$ and define the additive coordinates of the boxes in the multi-dimensional partitions as 
\begin{equation}
    c(\Bbox)=\begin{dcases}
    \mathfrak{a}+(i-1)\epsilon_{a}\quad(i\in\mathbb{Z}_{\geq 1}\,\,\text{or}\,\,i\in\mathbb{Z})\quad &\text{1d partition}\\
    \mathfrak{a}+(i-1)\epsilon_{a}+(j-1)\epsilon_{b}\quad (i,j\in\mathbb{Z}_{\geq 1})\quad&\text{2d partition}\\
    \mathfrak{a}+(i-1)\epsilon_{a}+(j-1)\epsilon_{b}+(k-1)\epsilon_{c}\quad (i,j,k\in\mathbb{Z}_{\geq} 1)\quad &\text{3d partition}\\
    \mathfrak{a}+(i-1)\epsilon_{a}+(j-1)\epsilon_{b}+(k-1)\epsilon_{c}+(l-1)\epsilon_{d}\quad (i,j,k,l\in\mathbb{Z}_{\geq 1})\quad &\text{4d partition}
    \end{dcases}
\end{equation}
where $\mathfrak{a}$ is the coordinate of the origin of the multi-dimensional partitions and $a,b,c,d$ are all different elements of $\{1,2,3,4\}$. For low dimensions, we have the following figures.
\begin{itemize}
    \item 1d partitions: depending on whether $i\in \mathbb{Z}_{\geq 1}$ or $\mathbb{Z}$ gives the coordinates of the two types in \eqref{eq:1dpartition_positive},\eqref{eq:1dpartition_integer}
    \begin{equation}
        \adjustbox{valign=c}{\begin{tikzpicture}
        \fill[red!20!white] (0.9,0)--(0.9,0.7)--(1.6,0.7)--(1.6,0)--(0.9,0);
        \draw[->] (-2,0)--(4,0);
        \node [right] at (4,0){$\epsilon_{a}$};
        \node[below] at (-0.15,0) {$1$};
        \node [below] at (-0.85,0){$0$};
        \node [below] at (0.55,0){$\cdots$};
        \node [below] at (1.25,0){$i$};
         \draw[thick]   (-0.5,-0.5)--(-0.5,1);
        \draw (-2,0.7)--(3,0.7);
        \draw (3,0)--(3,0.7);
        \draw (0.2,0)--(0.2,0.7);
        \draw (0.9,0)--(0.9,0.7);
         \draw (1.6,0)--(1.6,0.7);
          \draw (2.3,0)--(2.3,0.7);
          \draw (-1.2,0)--(-1.2,0.7);
          \node at (-1.6,0.35) {$\cdots$};
          \draw (1.25,0.35)--++(60:0.8);
          \node at (2.4,1.3) {$\mathfrak{a}+(i-1)\epsilon_{a}$};
          \draw (-0.15,0.35)--(-0.15,0.9);
          \node at (-0.15,1.3){$\mathfrak{a}$};
    \end{tikzpicture}}
    \end{equation}
    \item 2d partitions: 
    \begin{equation}
        \adjustbox{valign=c}{\begin{tikzpicture}
         \fill[red!20!white] (0.9,1.4)--(1.6,1.4)--(1.6,2.1)--(0.9,2.1)--(0.9,1.4);
        \draw[->] (-1,0)--(4,0);
        \node[above] at (-0.5,4){$\epsilon_{b}$};
        \node [right] at (4,0){$\epsilon_{a}$};
        \node[below] at (-0.15,0) {$1$};
        \node [below] at (0.55,0){$\cdots$};
        \node [below] at (1.25,0){$i$};
         \draw[->]   (-0.5,-0.5)--(-0.5,4);
         \draw (0.2,3.5)--(0.2,0.7);
         \draw (0.9,2.8)--(0.9,0.7);
         \draw (1.6,2.1)--(1.6,0.7);
         \draw (2.3,1.4)--(2.3,0.7);
         \draw (2.3,0.7)--(-0.5,0.7);
         \draw (2.3,1.4)--(-0.5,1.4);
         \draw (1.6,2.1)--(-0.5,2.1);
         \draw (0.9,2.8)--(-0.5,2.8);
         \draw (0.2,3.5)--(-0.5,3.5);
        \draw (-0.5,0.7)--(3,0.7);
        \draw (3,0)--(3,0.7);
        \draw (0.2,0)--(0.2,0.7);
        \draw (0.9,0)--(0.9,0.7);
         \draw (1.6,0)--(1.6,0.7);
          \draw (2.3,0)--(2.3,0.7);
          \draw (-0.15,0.35)--++(-0.7,-1);
          \node[left] at (-0.85,-0.65){$\mathfrak{a}$};
          \node [left] at (-0.5,0.35) {$1$};
          \node [left] at (-0.5,1.05){$\vdots$};
          \node [left] at (-0.5,1.75){$j$};
          \draw  (1.25,1.75)--++(0.9,0.9);
          \node[right] at (2.2,2.65) {$\mathfrak{a}+(i-1)\epsilon_{a}+(j-1)\epsilon_{b}$};
        \end{tikzpicture}
        }
    \end{equation}
    \item 3d partitions:
    \begin{equation}
        \adjustbox{valign=c}{\includegraphics[width=11cm]{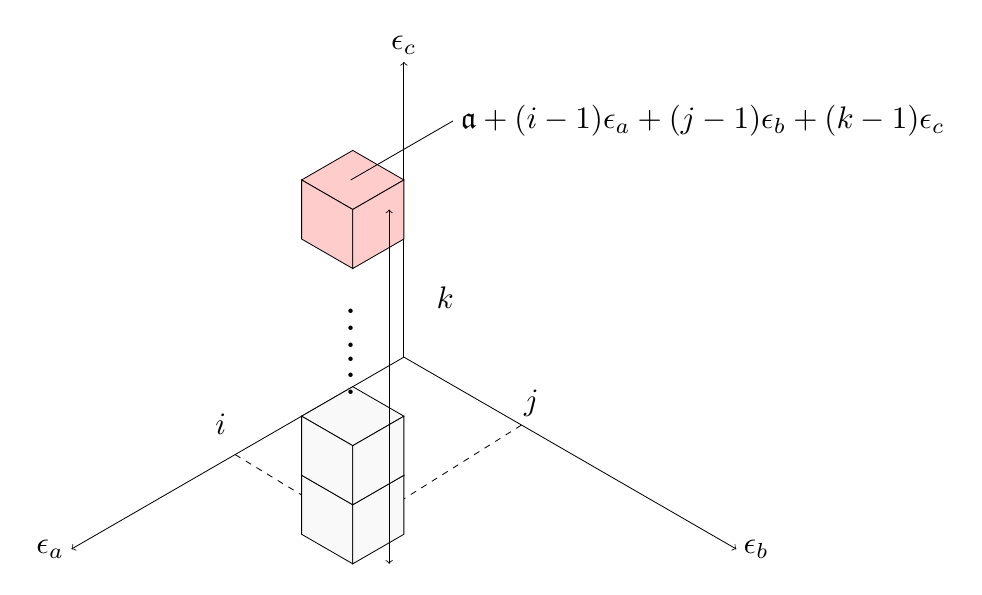}}
    \end{equation}
\end{itemize}

We also define the multiplicative coordinates of the boxes called $q$-contents or $q$-coordinates by taking the exponential:
\begin{equation}
    q(\Bbox)=e^{c(\Abox)}=\begin{dcases}
    uq_{a}^{i-1}\quad(i\in\mathbb{Z}_{\geq 1}\,\,\text{or}\,\,i\in\mathbb{Z})\quad &\text{1d partition}\\
    uq_{a}^{i-1}q_{b}^{j-1}\quad (i,j\in\mathbb{Z}_{\geq 1})\quad&\text{2d partition}\\
    uq_{a}^{i-1}q_{b}^{j-1}q_{c}^{k-1}\quad (i,j,k\in\mathbb{Z}_{\geq} 1)\quad &\text{3d partition}\\
    uq_{1}^{i-1}q_{2}^{j-1}q_{3}^{k-1}q_{4}^{l-1}\quad (i,j,k,l\in\mathbb{Z}_{\geq 1})\quad &\text{4d partition}
    \end{dcases}
\end{equation}
where $q_{a}=e^{\epsilon_{a}}\,(a=1,2,3,4)$ and $u=e^{\mathfrak{a}}$.

\paragraph{Other types of description}
Looking at the $q$-coordinates, we can also see how the higher-dimensional partitions can be decomposed into multiple lower-dimensional partitions.
\begin{itemize}
\item 1d partitions: The 1d partitions can be understood as a collection of the $q$-coordinates as $\{uq_{a}^{i-1}\mid 1\leq i\leq k ,k\in\mathbb{Z}_{\geq 0} \}$ or $\{uq_{a}^{i-1}\mid i\leq k,\,k\in\mathbb{Z}\}$.
    \item 2d partition: The boxes included in the Young diagram are labeled by $(i,j),\,(1\leq j\leq \lambda_{i},\,1\leq i\leq \ell(\lambda))$ with coordinates $\{uq_{a}^{i-1}q_{b}^{j-1}\}$. We can decompose this set as 
    \begin{equation}
        \left\{uq_{a}^{i-1}q_{b}^{j-1}\mid \substack{1\leq j\leq \lambda_{i}\\1\leq i\leq \ell(\lambda)}\right\}=\bigcup_{i=1}^{\ell(\lambda)}\left\{(uq_{a}^{i-1})q_{b}^{j-1}\mid 1\leq j\leq \lambda_{i}\right\}.
    \end{equation}
    This means the 2d partition is a collection of 1d partitions where the origins of the 1d partitions are shifted as $uq_{a}^{i-1}$. This gives the (1,1)-type decomposition of the 2d partitions.
    \item 3d partition: It is understood as a collection of $q$-coordinates $\{uq_{a}^{i-1}q_{b}^{j-1}q_{c}^{k-1}\}$ where $(i,j,k)\in\pi$. The $(2,1)$-type description comes from the decomposition
    \begin{equation}
        \left\{uq_{a}^{i-1}q_{b}^{j-1}q_{c}^{k-1}\mid (i,j,k)\in\pi\right\}=\bigcup_{(i,j)\in\lambda_{\pi}}\left\{(uq_{a}^{i-1}q_{b}^{j-1})\,q_{c}^{k-1}\mid 1\leq k\leq \pi_{ij}\right\},
    \end{equation}
    which means the plane partition is decomposed into 1d partitions with length $\pi_{ij}$ and coordinates of the origins as $uq_{a}^{i-1}q_{b}^{j-1}$. The $(1,2)$-type description comes from
    \begin{align}
        \left\{uq_{a}^{i-1}q_{b}^{j-1}q_{c}^{k-1}\mid (i,j,k)\in\pi\right\}=\bigcup_{k}\left\{(uq_{c}^{k-1})q_{a}^{i-1}q_{b}^{j-1}\mid (i,j)\in\Lambda^{(k)}\right\}.
    \end{align}
    The plane partition is then decomposed into multiple 2d partitions $\Lambda^{(k)}$ whose coordinates of the origins are $uq_{c}^{k-1}$.
    \item 4d partition: The solid partition is a collection of $q$-coordinates $\{uq_{a}^{i-1}q_{b}^{j-1}q_{c}^{k-1}q_{d}^{l-1}\}$ where $(i,j,k,l)\in\rho$. 
    \begin{enumerate}
        \item $(3,1)$-type: It is decomposed into multiple 1d partitions with length $\rho_{i,j,k}$ and origins at $uq_{a}^{i-1}q_{b}^{j-1}q_{c}^{k-1}$.
        \begin{equation}
            \left\{uq_{a}^{i-1}q_{b}^{j-1}q_{c}^{k-1}q_{d}^{l-1}\mid (i,j,k,l)\in\rho\right\}=\bigcup_{(i,j,k)\in\pi_{\rho}}\left\{(uq_{a}^{i-1}q_{b}^{j-1}q_{c}^{k-1})q_{d}^{l-1}\mid 1\leq l\leq\rho_{i,j,k}\right\}
        \end{equation}
        \item $(2,2)$-type: It is decomposed into multiple 2d partitions $\mu^{\rho}_{(i,j)}$ with origins with origins at $uq_{a}^{i-1}q_{b}^{j-1}$. 
        \begin{equation}
            \left\{uq_{a}^{i-1}q_{b}^{j-1}q_{c}^{k-1}q_{d}^{l-1}\mid (i,j,k,l)\in\rho\right\}=\bigcup_{(i,j)\in\lambda_{\rho}}\left\{(uq_{a}^{i-1}q_{b}^{j-1})q_{c}^{k-1}q_{d}^{l-1}\mid (k,l)\in\mu^{\rho}_{(i,j)}\right\}
        \end{equation}
        \item $(1,3)$-type: It is decomposed into multiple 3d partitions $\Pi^{(l)}$ with origins at $uq_{d}^{l-1}$.
        \begin{equation}
            \left\{uq_{a}^{i-1}q_{b}^{j-1}q_{c}^{k-1}q_{d}^{l-1}\mid (i,j,k,l)\in\rho\right\}=\bigcup_{l}\left\{(uq_{d}^{l-1})q_{a}^{i-1}q_{b}^{j-1}q_{c}^{k-1}\mid (i,j,k)\in\Pi^{(l)}\right\}
        \end{equation}
    \end{enumerate}
\end{itemize}


\section{Instanton partition functions of the gauge origami system}\label{sec:gaugeorigamipartitionfunction}
In this section, we review the physical setup (section~\ref{sec:physicalsetup}), the equivariant index formalism (section~\ref{sec:equiv-index}), contour integrals and explicit instanton partition functions (section~\ref{sec:M4partitionfunction}, \ref{sec:Tetrahedron_inst}, \ref{sec:spiked_partitionfunct}, \ref{sec:cplvortex_partitionfunct}) of the gauge origami system introduced first by Nekrasov in the context of BPS/CFT correspondence \cite{Nekrasov:2015wsu,Nekrasov:2016ydq,Nekrasov:2016qym,Nekrasov:2017cih,Nekrasov:2018xsb,Nekrasov:2017gzb,Nekrasov:2017rqy}. We are, in particular, interested in the following three setups: spiked instantons \cite{Nekrasov:2016ydq,Nekrasov:2016qym}, tetrahedron instantons \cite{Pomoni:2021hkn,Pomoni:2023nlf,Fasola:2023ypx}, and the magnificent four \cite{Nekrasov:2017cih,Nekrasov:2018xsb,Nekrasov:2023nai}. We additionally introduce a system that we call a \emph{coupled vortex system} whose characteristics are yet to be studied.

We also study the decompositions of the partition functions and show that partition functions of higher dimensional theory are obtained by infinite products of the partition functions of lower dimensional theories in section~\ref{sec:decomp_partition}. The decomposition depends on how the multi-dimensional partitions are described as discussed in detail in section~\ref{sec:multi-dim-part}. We finally introduce the $qq$-characters in terms of partition functions in section~\ref{sec:qqpartitionfunct}. 

\subsection{Physical setup}\label{sec:physicalsetup}

In this section, we review the gauge origami construction which is a generalization of the ADHM construction of instantons of supersymmetric gauge theories. Roughly speaking, the gauge origami system is a \emph{generalized gauge theory} whose space-time $\mathcal{S}$ contains several intersecting components as
\bea
\mathcal{S}=\bigcup_{i}\mathcal{S}_{i}.
\eea
We have a gauge group $G_{i}$ for each component $\mathcal{S}_{i}$ and matter fields in the intersection $\mathcal{S}_{i}\cap \mathcal{S}_{j}$ transforming under the gauge group $G_{i}\times G_{j}$ and thus are bifundamental multiplets. In this sense, we can understand such theories as generalized quiver gauge theories.

The gauge origami system is described as a low-energy-limit of the world volume theory on multiple intersecting D-branes in type II string theory. Let $Z\times \mathbb{R}^{2}$ be a ten-dimensional space-time where $Z=\mathbb{C}^{4}$. Generally, $Z$ can be a toric Calabi--Yau four-fold. We denote the four complex coordinates of $\mathbb{C}^{4}$ as $\{z_{a}\}_{a\in\four}$ where $\four=\{1,2,3,4\}$ (see section \ref{sec:equiv-index} for the notations). We have four types of subspaces: $\mathbb{C}$, $\mathbb{C}^{2}$, $\mathbb{C}^{3}$, $\mathbb{C}^{4}$. There are four possible $\mathbb{C}$ and $\mathbb{C}^{3}$ subspaces:
\bea
\mathbb{C}_{a},\qquad \mathbb{C}^{3}_{\bar{a}}\quad (a\in\four),
\eea
where $\bar{a}$ is the complement of $a\in\four$: $\bar{a}\in\{123,124,134,234\}$ (see section \ref{sec:equiv-index} for details of the notations). For $\mathbb{C}^{2}$, there are six possible subspaces:
\bea
\mathbb{C}^{2}_{A}\subset \mathbb{C}^{4},\quad A\in \six=\{12,13,14,23,24,34\}.
\eea

Depending on which type of subspaces the D-branes wrap different gauge origami systems appear. In this paper, we are interested in the following four setups: spiked instantons, tetrahedron instantons, magnificent four, and the coupled vortex system. The spiked instanton setup is a gauge origami system whose gauge theory is described by multiple D-branes wrapping the six possible subspaces $\mathbb{C}^{2}_{A}\,(A\in\six)$. On the other hand, the tetrahedron instanton setup is a setup appearing on D-branes wrapping the four possible subspaces $\mathbb{C}^{3}_{\bar{a}}\,(a\in\four)$. We can also consider a system where D-branes wrap the entire $\mathbb{C}^{4}$ and this is the magnificent four system. These three setups were previously studied, but for future use, we introduce a system we call the coupled vortex system. We expect such theories to be described by D-branes wrapping the four possible subspaces $\mathbb{C}_{a}\,(a\in \four)$ but a detailed analysis of this system is postponed for future work.

Let us review each system in more detail.
\subsection*{Spiked instanton}
The spiked instanton system comes from intersecting D-brane configurations\footnote{Note that this is a system obtained by taking T-duality of the setup where D1-branes probe $\D(2p+1)$-branes as explained in section~\ref{sec:introduction}.} \cite{Nekrasov:2015wsu,Nekrasov:2016qym,Nekrasov:2016ydq,Nekrasov:2016gud}:
\bea\renewcommand{\arraystretch}{1.2}\label{eq:spikedhierarchy}
    \begin{tabular}{c|c|c}
      Type IIB $\D(-1)$-D3-$\overline{\D3}$   & $\mathbb{C}^{4}\times \mathbb{R}^{2}$& rational\\
    \hline Type IIA $\D0$-$\D4$-$\overline{\D4}$ & $\mathbb{C}^{4}\times \mathbb{R}\times \mathbb{S}^{1}$& trigonometric\\
    \hline Type IIB $\D1$-$\D5$-$\overline{\D5}$ & $\mathbb{C}^{4}\times \mathbb{T}^{2}$& elliptic
    \end{tabular}
\eea
Depending on the total space-time $\mathbb{C}^{4}\times \mathcal{C}$ ($\mathcal{C}=\mathbb{R}^{2}$, $\mathbb{R}\times \mathbb{S}^{1}$, $\mathbb{T}^{2}$), the generalized gauge theory will be a combination of 4d ($\mathcal{C}=\mathbb{R}^{2}$), 5d ($\mathcal{C}=\mathbb{R}\times\mathbb{S}^{1}$), 6d ($\mathcal{C}=\mathbb{T}^{2}$) theories and the arising instanton partition function will be rational, trigonometric, and elliptic.

We focus on the type IIA setup as in Table \ref{t:spikedinstanton}. To do the supersymmetric localization, the $\mathbb{C}^{4}$ part of the system is placed under an $\Omega$-background. We label each stack of $n_{A}$ D4 ($\overline{\D4}$)-branes by $\D4_{A}$ ($\overline{\D4}_{A}$) where the subindex indicates which $\mathbb{C}^{2}_{A}$ subspace they wrap. For each stack of $\D4$ ($\overline{\D4}$)-branes wrapping $\mathbb{C}^{2}_{A}\times \mathbb{S}^{1}$, we have a 5d $\mathcal{N}=1^{\ast}$ $\U(n_{A})$ gauge theory. On the junctions of $\D4_{A}$ and $\D4_{B}$ defined on $\mathbb{C}_{A}\cap \mathbb{C}_{B}\,(A\neq B\in \six)$, the open strings connecting them give bifundamental contributions of $\U(n_{A})\times \U(n_{B})$. Note that for the intersecting brane configuration to preserve two supercharges, we need to include two stacks of anti D-branes in the system \cite{Nekrasov:2016qym,Nekrasov:2016gud}. In later sections, we will not distinguish the D4-branes and anti D4-branes.

The D0-branes in the system play the roles of instantons. There are three types of instantons in the spiked instanton system. The first type is the instantons in the $\hat{A}_{0}$ quiver gauge theory coming from each stack of $\D4_{A}\,(A\in \six)$-branes. The instanton contribution coming from the bifundamental multiples connecting 5d theories such as $\D4_{ab}$ and $\D4_{ac}$ are called folded instantons. The last type is the crossed instantons coming from two orthogonal stacks of D-branes such as $\D4_{12}$ and $\D4_{34}$. Namely, we have 
\bea
\text{spiked instantons}=\begin{dcases}
\text{instantons of $\hat{A}_{0}$ quiver theory},\\
\text{folded instantons},\\
\text{crossed instantons}.
\end{dcases}
\eea

\begin{table}[t]
\centering
\begin{tabular}{|c|c|c|c|c|c|c|c|c|c|c|}
\hline
& \multicolumn{2}{c|}{$\mathbb{C}_{1}$} & \multicolumn{2}{c|}{$\mathbb{C}_{2}$} & \multicolumn{2}{c|}{$\mathbb{C}_{3}$} & \multicolumn{2}{c|}{$\mathbb{C}_{4}$} & \multicolumn{2}{c|}{$\mathbb{R}\times \mathbb{S}^{1}$} \\
\cline{2-11}  & 1 & 2 & 3 & 4& 5 & 6 & 7 & 8 & 9& 0\\
\hline $\D0$& $\bullet$ & $\bullet$  & $\bullet$  & $\bullet$  & $\bullet$  & $\bullet$   & $\bullet$  & $\bullet$  & $\bullet$   & $-$\\
\hline
$\D4_{12} $& $-$ & $-$ & $-$ & $-$ & $\bullet$ & $\bullet$ & $\bullet$ & $\bullet$ & $\bullet$ & $-$ \\
\hline
\raisebox{-0.6mm}{$\overline{\D4}_{13}$} & $-$ & $-$& $\bullet$ & $\bullet$  & $-$ & $-$ & $\bullet$ & $\bullet$ & $\bullet$ & $-$ \\
\hline
$\D4_{14} $& $-$ & $-$  & $\bullet$ & $\bullet$ & $\bullet$ & $\bullet$& $-$ & $-$ & $\bullet$ & $-$ \\
\hline $\D4_{23}$ & $\bullet$ & $\bullet$ & $-$ & $-$ & $-$ & $-$ & $\bullet$ & $\bullet$ & $\bullet$ & $-$ \\
\hline
    \raisebox{-0.6mm}{$\overline{\D4}_{24}$} & $\bullet$ & $\bullet$ & $-$ & $-$ & $\bullet$ & $\bullet$ & $-$ & $-$ & $\bullet$ & $-$ \\
\hline
$\D4_{34} $& $\bullet$ & $\bullet$ & $\bullet$ & $\bullet$& $-$ & $-$ & $-$ & $-$  & $\bullet$ & $-$ \\
\hline
\end{tabular}
\caption{Brane configuration of gauge origami of spiked instanton. Point-like directions of the D-branes are denoted $\bullet$, while the directions where the D-branes are extending are denoted $-$.}
\label{t:spikedinstanton}
\end{table}

Let us briefly review the generalized ADHM construction of the spiked instantons. For each gauge group we associate a vector space $\bfN_{A}=\mathbb{C}^{n_{A}}\,(A\in \six)$. Another vector space $\bfK=\mathbb{C}^{k}$ is associated to the $k$ $\D0$-branes corresponding to the instantons. Similar to the ADHM construction, there are maps acting on $\{\bfN_{A}\}_{A\in \six}$ and $\bfK$:
\bea
I_{A}:\bfN_{A}\rightarrow \bfK,\quad J_{A}: \bfK\rightarrow \bfN_{A},\quad B_{a}:\bfK\rightarrow \bfK
\eea
where $A\in\six,\,\,a\in\four$. From the world-volume theory of the instantons, the maps $I_{A},J_{A}$ are understood as open strings connecting the $\D4_{A}$ and $\D0$-branes. The maps $B_{a}$ correspond to the four transverse directions of the $\D0$ theory which are the four complex coordinates of $\mathbb{C}^{4}$. 

We then introduce the following moment maps:
\bea\label{eq:spikedmomentdef}
&\mu_{A}=[B_{a},B_{b}]+I_{A}J_{A},\quad A=ab,\,\,(a<b),\\
&s_{A}=\mu_{A}+\epsilon_{A\bar{A}}\mu_{\bar{A}}^{\dagger},\quad A\in\six,\\
&\mu_{\mathbb{R}}=\sum_{a\in\four}[B_{a},B_{a}^{\dagger}]+\sum_{A\in\six}(I_{A}I_{A}^{\dagger}-J_{A}^{\dagger}J_{A})
\eea
where $\epsilon_{A\bar{A}}=\epsilon_{abcd}\,\,(A=ab,\,\bar{A}=cd)$ for $a<b,\,c<d$ is the total antisymmetric tensor with $\epsilon_{1234}=1$. We also have the property $s_{A}^{\dagger}=\epsilon_{A\bar{A}}s_{\bar{A}}$ which implies $s_{A}$ is a real map giving six real conditions. Note also that $\mu_{A},s_{A},\mu_{\mathbb{R}}\in\operatorname{Hom}\,(\bfK,\bfK)$. The moment map equations are then given as 
\bea\label{eq:spikedKK}
\{s_{A}=0\}/\U(k),\quad (A\in\six),\qquad 
\{\mu_{\mathbb{R}}=\zeta\cdot1_{k}\}/\U(k),\quad (\zeta>0)
\eea
where we turned on the FI parameter. We additionally have the following condition
\bea\label{eq:spikedNK}
\{\sigma_{aA}=0\}/\U(k),\quad \sigma_{aA}=B_{a}I_{A}+\epsilon_{abA}B_{b}^{\dagger}J_{A}^{\dagger}
\eea
where $a,b\in\bar{A},\,(a\neq b)$ and $\sigma_{aA}\in\operatorname{Hom}\,(\bfN_{A},\bfK)$. Note that these conditions do not appear in the original ADHM construction. They are conditions appearing only when we consider D0-branes and intersecting $\D4$-branes. We also impose the condition 
\beq\label{eq:spikedNN}
\{\Upsilon_{A}=0\}/\U(k),\quad \Upsilon_{A}=J_{\bar{A}}I_{A}-I^{\dagger}_{\bar{A}}J_{A}^{\dagger}:\bfN_{A}\rightarrow \bfN_{\bar{A}}
\eeq
obeying $\Upsilon_{A}=-\Upsilon_{\bar{A}}$. This condition comes from the open strings connecting the $\D4_{A}$ and $\D4_{\bar{A}}$-branes. The instanton moduli space is then defined as\footnote{In order to apply the equivariant localization formalism, the real moment map equation would be replaced with the stability condition, and the quotient is then accordingly replaced with the geometric invariant theory (GIT) quotient with the automorphism group $\mathrm{GL}(\mathbf{K})$, that is a complexification of the unitary group $\U(k)$.} 
\bea\label{eq:spikedinstantonmoduli}
\mathfrak{M}_{\vec{n},k}=\left.\{(\vec{B},\vec{I},\vec{J})\mid s_{A}=0,\,\,\mu_{\mathbb{R}}=\zeta\cdot 1_{k},\,\, \sigma_{aA}=0,\,\,\Upsilon_{A}=0\}\right/\U(k).
\eea
Using the following identity \cite[eq.~(54)]{Nekrasov:2016qym}:
\bea
&\sum_{A\in\six}\text{Tr}\,s_{A}s_{A}^{\dagger}+\sum_{A\in\six,\,a\in\bar{A}}\Tr\sigma_{aA}\sigma_{aA}^{\dagger}+\sum_{A\in\six}\Tr\Upsilon_{A}\Upsilon_{A}^{\dagger} \\
&\quad =2\sum_{A\in\six}(||\mu_{A}||^{2}+||J_{\bar{A}}I_{A}||^{2})+\sum_{A\in\six,\,a\in\bar{A}}(||B_{a}I_{A}||^{2}+||J_{A}B_{a}||^{2})
\eea
the conditions \eqref{eq:spikedKK}, \eqref{eq:spikedNK}, \eqref{eq:spikedNN}, we have
\bea\label{eq:spikedJEterm}
s_{A}=0&\longrightarrow \mu_{A}=0\,\,(A\in\six),\\
\sigma_{aA}=0 &\longrightarrow  B_{a}I_{A}=0,\quad J_{A}B_{a}=0\,\,(A\in\six,\,a\in\bar{A}),\\
\Upsilon_{A}=0&\longrightarrow J_{\bar{A}}I_{A}=0\,\,(A\in\six).
\eea
From the world-volume theory of the instantons viewpoint, the condition $\mu_{\mathbb{R}}=0$ is equivalent to the D-term condition and the conditions are equivalent to the $J$-term and E-term conditions. 

When there is only one stack of D4-branes, say $\D4_{12}$-branes, we have $\bfN_{A}=0\,(A\neq 12)$ and we can simply set $I_{A}=J_{A}=0\,(A\neq 12)$. From \eqref{eq:spikedJEterm}, we have $B_{3},B_{4}=0$ on the solutions, and then the moduli space indeed reduces to the standard one \cite{Nekrasov:2016qym}:
\bea
\left.\left\{(B_{12},I_{12},J_{12})\left|\begin{array}{c}
     [B_{1},B_{2}]+I_{12}J_{12}=0, \\
     \sum\limits_{a=1,2}[B_{a},B_{a}^{\dagger}]+(I_{12}I_{12}^{\dagger}-J_{12}^{\dagger}J_{12})=\zeta\cdot 1_{k}
\end{array}\right.\right\}\right/\U(k)
\eea

Let us discuss the symmetries of the ADHM variables $(\vec{B},\vec{I},\vec{J})$. We first have the $\U(k)$ symmetry coming from the gauge symmetries of the $\D0$-branes:
\beq\label{eq:instantongaugesymmetry}
(B_{a},I_{A},J_{A})_{a\in\four,A\in\six}\longmapsto (g^{-1}B_{a}g,g^{-1}I_{A},J_{A}g),\quad g\in\U(k). 
\eeq
We also have the $\prod_{A\in\six}\U(n_{A})$ symmetry acting as
\beq\label{eq:D4flavorsymmetry}
(B_{a},I_{A},J_{A})\longmapsto (B_{a},I_{A}h_{A},h_{A}^{-1}J_{A}),\quad \underline{h}=(h_{A})_{A\in\six}\in \prod_{A\in\six}\U(n_{A}).
\eeq
This symmetry is called the framing rotation in \cite{Nekrasov:2016qym} because they are flavor symmetries from the world-volume theory of the $\D0$-branes viewpoint. We finally have a symmetry corresponding to the spatial rotations of $\mathbb{C}^{4}$:
\beq\label{eq:D4ADHMvariablerotation}
(B_{a},I_{A},J_{A})\longmapsto (q_{a}B_{a}, I_{A},q_{A}J_{A})
\eeq
where $q_{a}\,(a=1,2,3,4)$ and $q_{A}=q_{a}q_{b}$. Obviously, this symmetry preserves the conditions coming from $\mu_{\mathbb{R}}$. For the conditions coming from $s_{A}, \sigma_{aA},\Upsilon_{A}$ to have this symmetry we need
\beq\label{eq:CY4cond}
q_{1}q_{2}q_{3}q_{4}=1.
\eeq
For example, we have $\mu_{A}\rightarrow q_{A}\mu_{A}$ and for $A=12$, $s_{12}\rightarrow q_{12}\mu_{12}+q_{34}^{-1}\mu_{34}^{\dagger}$. Thus, with the condition \eqref{eq:CY4cond}, we have $s_{12}\rightarrow q_{12}s_{12}$. Other conditions can be checked similarly and the symmetry transformation of the generalized ADHM constraints are
\bea\label{eq:D4ADHMconstraintsymmetry}
s_{A}&\longmapsto q_{A}g^{-1}s_{A}g,\quad \sigma_{aA}\longmapsto q_{a}g^{-1}\sigma_{aA}h_{A},\quad \Upsilon_{A}\longmapsto q_{A}^{-1}h_{\bar{A}}^{-1}\Upsilon_{A}h_{A},
\eea
where $g\in\U(k),\,h_{A}\in\U(n_{A})$.
The reason why it is called the spatial rotation is because $\{B_{a}\}_{a\in\four}$ correspond to the four complex coordinates of $\mathbb{C}^{4}$ and $\{q_{a}\}_{a\in\four}$ act as the maximal torus $\U(1)^{3}$ of the group $\SU(4) \subset \mathrm{Spin}(8)$ rotating the $\mathbb{C}^{4}$.

\subsubsection*{Tetrahedron instanton}
The spiked instanton system is a system where D-branes wrapping $\mathbb{C}^{2}\subset \mathbb{C}^{4}$ appeared. Recently, the gauge origami system was generalized to stacks of D-branes wrapping $\mathbb{C}^{3}\subset\mathbb{C}^{4}$ and is called the tetrahedron instanton system \cite{Pomoni:2021hkn,Pomoni:2023nlf,Fasola:2023ypx}. The intersecting D-brane configuration is given by
\bea\label{eq:tetrahedronhierarchy}
\renewcommand{\arraystretch}{1.2}
    \begin{tabular}{c|c|c}
      Type IIB $\D(-1)$-D5  & $\mathbb{C}^{4}\times \mathbb{R}^{2}$& rational\\
    \hline Type IIA $\D0$-$\D6$ & $\mathbb{C}^{4}\times \mathbb{R}\times \mathbb{S}^{1}$& trigonometric\\
    \hline Type IIB $\D1$-$\D7$ & $\mathbb{C}^{4}\times \mathbb{T}^{2}$& elliptic
    \end{tabular}
\eea

We focus on the type IIA theory as in Table \ref{t:tetrahedroninstanton} with the $\Omega$-background in $\mathbb{C}^{4}$. Similar to the spiked instanton case, we label each stack of $n_{\bar{a}}$ D6-branes as $\D6_{\bar{a}}(a\in\four)$ depending on which $\mathbb{C}^{3}_{\bar{a}}$ subspace they wrap. For each D6$_{\bar{a}}$-brane, we have a 7d $\mathcal{N}=1$ $\U(n_{\bar{a}})$ theory. We also have bifundamental contributions coming from the junctions of $\D6_{\bar{a}}$ and $\D6_{\bar{b}}$, for $a\neq b$. The D0-branes play the roles of instantons and are called the tetrahedron instantons. In this case, there are only two types of instantons. The first type of instantons are the instantons on $\mathbb{C}^{3}$ coming from the $\D6_{\bar{a}}$ theory. The second type is a folded instanton contribution coming from $\D6_{\bar{a}}$ and $\D6_{\bar{b}}$, where $a\neq b$. In this case, we only have this type because any two stacks of D6-branes spanning different subspaces will share 4+1 dimensions. Therefore, focusing on one of the stacks of D6-branes, other stacks of D6-branes are codimension-two defects.
\begin{table}[t]
\centering
\begin{tabular}{|c|c|c|c|c|c|c|c|c|c|c|}
\hline
& \multicolumn{2}{c|}{$\mathbb{C}_{1}$} & \multicolumn{2}{c|}{$\mathbb{C}_{2}$} & \multicolumn{2}{c|}{$\mathbb{C}_{3}$} & \multicolumn{2}{c|}{$\mathbb{C}_{4}$} & \multicolumn{2}{c|}{$\mathbb{R}\times \mathbb{S}^{1}$} \\
\cline{2-11}  & 1 & 2 & 3 & 4& 5 & 6 & 7 & 8 & 9& 0\\
\hline $\D0$& $\bullet$ & $\bullet$  & $\bullet$  & $\bullet$  & $\bullet$  & $\bullet$   & $\bullet$  & $\bullet$  & $\bullet$   & $-$\\
\hline
$\D6_{123} $& $-$ & $-$ & $-$ & $-$ & $-$ & $-$ & $\bullet$ & $\bullet$ & $\bullet$ & $-$ \\
\hline
$\D6_{124} $& $-$ & $-$& $-$ & $-$  & $\bullet$ & $\bullet$ & $-$ & $-$ & $\bullet$ & $-$ \\
\hline
$\D6_{134} $& $-$ & $-$  & $\bullet$ & $\bullet$ & $-$ & $-$& $-$ & $-$ & $\bullet$ & $-$ \\
\hline $\D6_{234}$ & $\bullet$ & $\bullet$ & $-$ & $-$ & $-$ & $-$ & $-$ & $-$ & $\bullet$ & $-$ \\
\hline
\end{tabular}
\caption{Brane configuration of gauge origami of tetrahedron instanton.}
\label{t:tetrahedroninstanton}
\end{table}

Let us review the instanton moduli space of the tetrahedron instanton system discussed in \cite{Pomoni:2021hkn} (see also \cite{Pomoni:2023nlf,Fasola:2023ypx}). For each gauge group $\U(n_{\bar{a}})\,(a\in\four)$ coming from the $\D6_{\bar{a}}$-branes, we associate a vector space $\bfN_{\bar{a}}=\mathbb{C}^{n_{\bar{a}}}\,(a\in\four)$. We also associate the $k$ D0-branes with a vector space $\bfK=\mathbb{C}^{k}$. The open strings connecting the D0-branes give $B_{a}\in\operatorname{Hom}\,(\bfK,\bfK)$ which is the same with the spiked instanton setup. Maps $I_{\bar{a}}\in\operatorname{Hom}\,(\bfN_{\bar{a}},\bfK)$ correspond to the open strings connecting the D0 and D6-branes. We then introduce the following moment maps:
\bea
\mu_{A}&=[B_{a},B_{b}],\quad A=ab\,(a<b)\\
s_{A}&=\mu_{A}+\epsilon_{A\bar{A}}\mu_{\bar{A}}^{\dagger}=[B_{a},B_{b}]+\frac{1}{2}\epsilon_{abcd}[B_{c}^{\dagger},B_{d}^{\dagger}],\quad A=ab\\
\mu_{\mathbb{R}}&=\sum_{a\in\four}[B_{a},B_{a}^{\dagger}]+\sum_{a\in\four}I_{\bar{a}}I_{\bar{a}}^{\dagger}
\eea
where all of them belong to $\operatorname{Hom}\,(\bfK,\bfK)$. The moment map equations are given 
\bea\label{eq:D6KKcondition}
\{\mu_{\mathbb{R}}=\zeta\cdot 1_{k}\}/\U(k)\,\, (\zeta>0),\quad \{s_{A}=0\}/\U(k). 
\eea
We also have the contributions coming from the open strings connecting the D0 and D6-branes:
\bea\label{eq:D6NKcondition}
\{\sigma_{\bar{a}}=0\}_{a\in\four}/\U(k),\quad \sigma_{\bar{a}}=B_{a}I_{\bar{a}}\in\operatorname{Hom}\,(\bfN_{\bar{a}},\bfK).
\eea
The instanton moduli space of the tetrahedron instanton system is defined as 
\bea
\mathfrak{M}_{\vec{n},k}=\left.\left\{(\vec{B},\vec{I})\mid \mu_{\mathbb{R}}-\zeta\cdot 1_{k}=s_{A}=\sigma_{\bar{a}}=0\right\}\right/\U(k).
\eea
Note that similar to the spiked instanton case, using 
\bea\label{eq:D6Ftermconverse}
\sum_{a<b}\Tr[B_{a},B_{b}][B_{a},B_{b}]^{\dagger}=\frac{1}{2}\sum_{a<b}\Tr s_{ab}s_{ab}^{\dagger},
\eea
the condition $s_{A}=0$ is replaced with a stronger condition $\mu_{A}=0$, which coincides with the F-term condition of the world-volume theory of the $\D0$-branes. When there is only one stack of D6-branes, say $\D6_{123}$, we can set $I_{\bar{a}}=0\,(a\neq 4)$ and indeed obtain the standard moduli space of the D0-D6 theory \cite{Nekrasov:2009JJM,Jafferis:2007sg,Cirafici:2008sn,Kanno:2020ybd}.

We also have similar gauge symmetries of $\D0$ and flavor symmetries coming from $\D6_{\bar{a}}$. The gauge symmetry is 
\bea\label{eq:D6gaugesymmetry}
(B_{a},I_{\bar{b}})_{a,b\in\four}\longmapsto(g^{-1}B_{a}g,g^{-1}I_{\bar{b}}),\quad g\in\U(k),
\eea
while the flavor symmetries are given 
\bea\label{eq:D6flavorsymmetry}
(B_{a},I_{\bar{b}})\longmapsto (B_{a},I_{\bar{b}}h_{\bar{b}}),\quad \underline{h}=(h_{\bar{a}})_{a\in\four}\in\prod_{a\in\four}\U(n_{\bar{a}}).
\eea
The rotational symmetry
is given as
\bea\label{eq:D6rotationalsymmetry}
(B_{a},I_{\bar{b}})\longmapsto (q_{a}B_{a},I_{\bar{b}}).
\eea
Note that the condition \eqref{eq:CY4cond} comes from the invariance of $s_{A}=0$. The symmetry transformation of the generalized ADHM constraints are
\bea\label{eq:D6ADHMconstrsymmetry}
s_{A}\longmapsto q_{A}g^{-1}s_{A}g,\quad \sigma_{\bar{a}}\longmapsto q_{a}g^{-1}\sigma_{\bar{a}}h_{\bar{a}}
\eea
where $g\in\U(k),\,h_{\bar{a}}\in\U(n_{\bar{a}})$.

\subsubsection*{Magnificent four}
Instead of considering D-branes wrapping subspaces of $\mathbb{C}^{4}$, we can consider D-branes wrapping the entire $\mathbb{C}^{4}$ \cite{Nekrasov:2017cih,Nekrasov:2018xsb,Nekrasov:2023nai}. It is called the magnificent four system. The brane configuration is given
\bea\label{eq:m4hierarchy}
\renewcommand{\arraystretch}{1.2}
    \begin{tabular}{c|c|c}
      Type IIB $\D(-1)$-D7-$\overline{\D7}$   & $\mathbb{C}^{4}\times \mathbb{R}^{2}$& rational\\
    \hline Type IIA $\D0$-$\D8$-$\overline{\D8}$ & $\mathbb{C}^{4}\times \mathbb{R}\times \mathbb{S}^{1}$& trigonometric\\
    \hline Type IIB $\D1$-$\D9$-$\overline{\D9}$ & $\mathbb{C}^{4}\times \mathbb{T}^{2}$& elliptic
    \end{tabular}
\eea
Again, we focus on the type IIA case. The D8 and $\overline{\D8}$ branes wrap the $\mathbb{C}^{4}\times \mathbb{S}^{1}$. In order to stabilize the configuration, the D8-branes need to appear in pair with $\overline{\D8}$-branes. Considering $n$ D8-$\overline{\D8}$ branes, we get a $\U(n|n)$ supergroup gauge theory. Similar to other cases, the $\D0$-branes play the roles of $\mathbb{C}^{4}$-instantons. 

\begin{table}[t]
\centering
\begin{tabular}{|c|c|c|c|c|c|c|c|c|c|c|}
\hline
& \multicolumn{2}{c|}{$\mathbb{C}_{1}$} & \multicolumn{2}{c|}{$\mathbb{C}_{2}$} & \multicolumn{2}{c|}{$\mathbb{C}_{3}$} & \multicolumn{2}{c|}{$\mathbb{C}_{4}$} & \multicolumn{2}{c|}{$\mathbb{R}\times \mathbb{S}^{1}$} \\
\cline{2-11}  & 1 & 2 & 3 & 4& 5 & 6 & 7 & 8 & 9& 0\\
\hline $\D0$& $\bullet$ & $\bullet$  & $\bullet$  & $\bullet$  & $\bullet$  & $\bullet$   & $\bullet$  & $\bullet$  & $\bullet$   & $-$\\
\hline
$\D8 $& $-$ & $-$ & $-$ & $-$ & $-$ & $-$ & $-$ & $-$ & $\bullet$ & $-$ \\
\hline
\raisebox{-0.6mm}{$\overline{\D8}$}& $-$ & $-$& $-$ & $-$  & $-$ & $-$ & $-$ & $-$ & $\bullet$ & $-$ \\
\hline
\end{tabular}
\caption{Brane configuration of magnificent four.}
\label{t:magnificentfour}
\end{table}

We denote the vector space coming from the D8-branes as $\mathbf{n}=\mathbb{C}^{n}$ and the vector space coming from the instantons as $\bfK=\mathbb{C}^{k}$. We have maps $B_{a}\in\operatorname{Hom}\,(\bfK,\bfK)$ and $I\in\operatorname{Hom}\,(\mathbf{n},\bfK)$ corresponding to the D0-D0 and D0-D8 strings. We define the moment maps as 
\bea
\mu_{\mathbb{R}}&=\sum_{a\in\four}[B_{a},B_{a}^{\dagger}]+II^{\dagger},\\
\mu_{A}&=[B_{a},B_{b}],\quad A=ab \,\,(a<b),\\
s_{ab}&=[B_{a},B_{b}]+\frac{1}{2}\epsilon_{abcd}[B_{c}^{\dagger},B_{d}^{\dagger}]
\eea
and the instanton moduli space is defined as 
\bea
\mathfrak{M}_{n,k}=\left.\left\{(\vec{B},I)\mid \mu_{\mathbb{R}}-\zeta\cdot 1_{k}=s_{A}=0\right\}\right/\U(k).
\eea
Similarly, using \eqref{eq:D6Ftermconverse}, the condition $s_{A}=0$ is replaced with $\mu_{A}=0$. The difference with the tetrahedron system is that there are no conditions coming from $\sigma_{\bar{a}}=0$.

The symmetries of the ADHM variables and ADHM constraints are similar to the tetrahedron case:
\bea\label{eq:D8ADHMvariablesymmetry}
(B_{a},I)_{a\in\four}&\longmapsto(g^{-1}B_{a}g,g^{-1}I),\quad g\in\U(k),\\
(B_{a},I)&\longmapsto (B_{a},Ih),\quad h\in\U(n),\\
(B_{a},I)&\longmapsto (q_{a}B_{a},I)
\eea
and
\bea\label{eq:D8ADHMconstrsymmetry}
s_{A}\longmapsto q_{A}g^{-1}s_{A}g,\quad g\in\U(k).
\eea

\subsubsection*{Coupled vortex system}
Analogous to the setup introduced before, one would like to consider a system where D-branes wrapping $\mathbb{C}\subset\mathbb{C}^{4}$ appear. For the moment, we do not know what kind of system it would be but based on the previous discussions, we expect it will be a theory of intersecting 2d/3d/4d gauge theories coming from the following:
\bea\label{eq:cplvorthierarchy}
\renewcommand{\arraystretch}{1.2}
    \begin{tabular}{c|c|c}
      Type IIB $\D(-1)$-D1  & $\mathbb{C}^{4}\times \mathbb{R}^{2}$& rational\\
    \hline Type IIA $\D0$-$\D2$ & $\mathbb{C}^{4}\times \mathbb{R}\times \mathbb{S}^{1}$& trigonometric\\
    \hline Type IIB $\D1$-$\D3$ & $\mathbb{C}^{4}\times \mathbb{T}^{2}$& elliptic
    \end{tabular}
\eea
We expect that vortex-like objects appear and we call it the coupled vortex system. Focusing on the $\D0\tbar\D2$ setup, the 
brane set-up should look like Table~\ref{t:cplvortex}. We will have a 3d $\U(n_{a})$ theory on each $\mathbb{C}_{a}\times \mathbb{S}^{1}$ and there should be bifundamentals connecting such 3d theories. Note that in this case, any two stack of D2-branes $\D2_{a}$ and $\D2_{b}$, where $a\neq b$ will share a one-dimensional space $\mathbb{S}^{1}$. 

For the moment, we do not know how to define the moduli space of this system. The discussion here is still a conjecture and a detailed analysis is postponed for future work(see~\cite[eq.~(60)]{Shadchin:2006yz}, \cite{Nekrasov:2009JJM,Rapcak:2021hdh} for related works). 

We denote the vector space coming from the D2-branes as $\bfN_{a}=\mathbb{C}^{n_{a}}$ and the vector space coming from the instantons as $\bfK=\mathbb{C}^{k}$. The contributions coming from the $\D0\tbar \D0$ and $\D0\tbar \D2$ strings give $B_{a}\in\operatorname{Hom}\,(\bfK,\bfK)$ and $I_{a}\in\operatorname{Hom}\,(\bfN_{a},\bfK)$. We also have $J_{a}^{(b)}\in\operatorname{Hom}\,(\bfK,\bfN_{a})$ where $b\neq a$ corresponding to open strings in the opposite direction.

We conjecture that the D-term condition for this system is\footnote{Like other cases, $\mu_{\mathbb{R}}$ is a sum of the D-term conditions of the vortex moduli space after setting $J_{a}^{(b)}=0$ on the moduli space. Recall that the standard vortex theory comes from the moduli space $\{(B,I)\mid[B,B^{\dagger}]+II^{\dagger}=\zeta\cdot 1_{k}\}/\U(k)$ (see \cite[eq.~(60)]{Shadchin:2006yz}). }
\bea
\mu_{\mathbb{R}}=\sum_{a\in\four}[B_{a},B_{a}^{\dagger}]+\sum_{a\in\four}I_{a}I_{a}^{\dagger}-\sum_{a\in\four}\sum_{b\neq a}J_{a}^{(b)\dagger}J_{a}^{(b)}\in\operatorname{Hom}\,(\bfK,\bfK)
\eea
and
\bea
\{\mu_{\mathbb{R}}=\zeta\cdot 1_{k}\}/\U(k)\,\,\,(\zeta>0).
\eea
Additional contributions coming from the F-term are necessary. We conjecture that there is a condition generalizing \eqref{eq:spikedmomentdef} and \eqref{eq:spikedKK}:
\bea
\{s_{A}=0\}/\U(k),
\eea
where
\bea
s_{A}=\mu_{A}+\epsilon_{A\bar{A}}\mu_{\bar{A}}^{\dagger}\in\operatorname{Hom}\,(\bfK,\bfK),\quad \mu_{A}=[B_{a},B_{b}]+\cdots.
\eea
The explicit form of the right-hand side of $\mu_{A}$ is unknown. Moreover, we expect that we have two more sets of conditions
\bea
\{\sigma_{b;a}=0\}/\U(k)\,\,(b\neq a),\qquad \{\tilde{\sigma}_{a}=0\}/\U(k)
\eea
where $\sigma_{b;a},\tilde{\sigma}_{a}\in\operatorname{Hom}\,(\bfN_{a},\bfK)$ are generalizations of \eqref{eq:spikedNK}. The conditions above should also be invariant under the following symmetries
\bea
(B_{a},I_{a},J_{a}^{(b)})_{a\in\four,b\in\bar{a}}&\longmapsto(g^{-1}B_{a}g,g^{-1}I_{a},J_{a}^{(b)}g),\quad g\in\U(k),\\
(B_{a},I_{a},J_{a}^{(b)})&\longmapsto (B_{a},I_{a}h_{a},h_{a}^{-1}J_{a}^{(b)}),\quad \underline{h}=(h_{a})_{a\in\four}\in\prod_{a\in\four}\U(n_{a}),\\
(B_{a},I_{a},J_{a}^{(b)})&\longmapsto (q_{a}B_{a},I_{a},q_{ab}J_{a}^{(b)}),
\eea
where the $\U(1)^{3}$ transformation of $J_{a}^{(b)}$ is still a conjecture. The transformation of the generalized ADHM constraints is expected to be\footnote{\label{footnote:D2ADHMintegral}The charges were assigned so that we can reproduce the contour integral formula that will be proposed in section~\ref{sec:cplvortex_partitionfunct}. The contour integral formula in \eqref{eq:D2integral} is rewritten as
\bea
\mathcal{Z}_{k}^{\D2}\propto \oint \prod_{I=1}^{k}\frac{dx_{I}}{2\pi\iota x_{I}}\prod_{a\in\four}\prod_{\alpha=1}^{n_{a}}\prod_{I=1}^{k}g_{\bar{a}}\left(\frac{v_{a,\alpha}}{x_{I}}\right)\prod_{I<J}\mathcal{A}_{\mathbb{C}^{4}}\left(\frac{x_{I}}{x_{J}}\right)^{-1}.
\eea
The ADHM variables $(B_{a},I_{a},J_{a}^{(b)})\,(b\neq a)$ give the denominators
\bea
B_{a}:\prod_{a\in\four}\prod_{I\neq J}(1-q_{a}x_{I}/x_{J})^{-1},\quad I_{a}:\prod_{a\in\four}\prod_{\alpha=1}^{n_{a}}\prod_{I=1}^{k}(1-v_{a,\alpha}/x_{I})^{-1},\quad J_{a}^{(b)}:\prod_{a\in\four}\prod_{b\neq a}\prod_{\alpha=1}^{n_{a}}\prod_{I=1}^{k}(1-q_{ab}x_{I}/v_{a,\alpha})^{-1}
\eea
while the ADHM constraints $s_{A}=0,\,\sigma_{b;a}=0,\,\tilde{\sigma}_{a}=0$ give the numerators
\bea
s_{A}=0:\prod_{I<J}\prod_{A\in\six}(1-q_{A}x_{I}/x_{J}),\quad \sigma_{b;a}=0\,(b\neq a):\prod_{I=1}^{k}\prod_{a\in\four}\prod_{\alpha=1}^{n_{a}}\prod_{b\neq a}(1-q_{b}v_{a,\alpha}/x_{I}),\quad \tilde{\sigma}_{a}:\prod_{a\in\four}\prod_{I=1}^{k}\prod_{\alpha=1}^{n_{a}}(1-q_{a}^{-1}v_{a,\alpha}/x_{I}).
\eea
Including the Haar measure contribution
\bea
\U(k):\prod_{I\neq J}(1-x_{I}/x_{J}),
\eea
we obtain the contour integral formula. See~\cite{Kanno:2020ybd,Kimura:2020jxl} for example for a review on how to relate the ADHM variables and constraints with the contour integral formula.} 
\bea
s_{A}\longmapsto q_{A}g^{-1}s_{A}g,\quad \sigma_{b;a}\longmapsto q_{b}g^{-1}\sigma_{b;a}h_{a},\quad
\tilde{\sigma}_{a}\longmapsto q_{a}^{-1}g^{-1}\tilde{\sigma}_{a}h_{a},
\eea
where $g\in\U(k),h_{a}\in\U(n_{a})$.

\begin{table}[t]
\centering
\begin{tabular}{|c|c|c|c|c|c|c|c|c|c|c|}
\hline
& \multicolumn{2}{c|}{$\mathbb{C}_{1}$} & \multicolumn{2}{c|}{$\mathbb{C}_{2}$} & \multicolumn{2}{c|}{$\mathbb{C}_{3}$} & \multicolumn{2}{c|}{$\mathbb{C}_{4}$} & \multicolumn{2}{c|}{$\mathbb{R}\times \mathbb{S}^{1}$} \\
\cline{2-11}  & 1 & 2 & 3 & 4& 5 & 6 & 7 & 8 & 9& 0\\
\hline $\D0$& $\bullet$ & $\bullet$  & $\bullet$  & $\bullet$  & $\bullet$  & $\bullet$   & $\bullet$  & $\bullet$  & $\bullet$   & $-$\\
\hline
$\D2_{1} $& $-$ & $-$ & $\bullet$  & $\bullet$  & $\bullet$ & $\bullet$ & $\bullet$ & $\bullet$ & $\bullet$ & $-$ \\
\hline
$\D2_{2} $& $\bullet$ & $\bullet$ & $-$ & $-$  & $\bullet$  & $\bullet$  & $\bullet$ & $\bullet$ & $\bullet$ & $-$ \\
\hline
$\D2_{3} $& $\bullet$ & $\bullet$   & $\bullet$ & $\bullet$ & $-$ & $-$& $\bullet$ & $\bullet$  & $\bullet$ & $-$ \\
\hline 
$\D2_{4} $& $\bullet$ & $\bullet$ & $\bullet$ & $\bullet$  & $\bullet$ & $\bullet$ & $-$ & $-$ & $\bullet$ & $-$ \\
\hline
\end{tabular}
\caption{Brane configuration of coupled vortex system.}
\label{t:cplvortex}
\end{table}

\subsection{Equivariant index formalism}\label{sec:equiv-index}
In this section, we briefly review the equivariant index formalism and introduce the notation we use in this paper. Some of the notations were already used in the previous section. We basically follow the notations of \cite{Nekrasov:2016ydq}. Explicit applications to the gauge origami system introduced in the previous section will be discussed in later sections. 
\paragraph{Index functor}
For a vector bundle with the virtual character
\bea
    \text{ch}\,\bfX=\sum_{i}n_{i}e^{x_{i}},
\eea
where $n_{i}\in\mathbb{Z}$ here is the multiplicity and $x_{i}$'s are the Chern roots, the index functor to convert the additive components to multiplicative components is defined as 
\bea
    \mathbb{I}\left[\bfX\right]=\prod_{i}\llbracket x_{i}\rrbracket ^{n_{i}},\qquad \llbracket x\rrbracket=\begin{dcases}
        x,\quad &(\text{4d})\\
        1-e^{-x},\quad &(\text{5d})\\
        \theta(e^{-x};p),\quad &(\text{6d})
    \end{dcases}\label{eq:rat_trig_ell}
\eea
The theta function here is defined in Appendix \ref{sec:ellipticformula}. The hierarchical structure between rational, trigonometric, and elliptic functions appears here by taking the limit as 
\bea
    \theta(e^{-x};p)\xrightarrow{p\rightarrow0} 1-e^{-x}=x+\mathcal{O}(x^{2}).
\eea
Depending on the type $Z\times \mathcal{C}$ where $Z$ is a toric Calabi--Yau four-fold and $\mathcal{C}=\mathbb{C},\mathbb{C}^{\times },\mathbb{T}^{2}$, the index is obtained from \eqref{eq:rat_trig_ell}.

Most of the computations in this paper will be done explicitly using the trigonometric notation, so when not mentioned we are using the following convention:
\bea
    \mathbb{I}[x]=(1-x^{-1})=\exp\left(-\sum_{n=1}^{\infty}\frac{1}{n}x^{-n}\right)
\eea
but using the formula in \eqref{eq:rat_trig_ell} one can convert the results to rational and elliptic ones. For the elliptic case, we distinguish the index as $\mathbb{I}_{p}[x]$ and a brief computation will be done in section~\ref{sec:ellipticWorigami}.

Note that the character of the dual of bundle $\bfX$ is defined as
\bea
    \text{ch}\,\bfX^{\vee}=\sum_{i}n_{i}e^{-x_{i}}
\eea
and we have the reflection property
\bea
    \mathbb{I}\left[\bfX^{\vee}\right]=(-1)^{\text{rk}\bfX}\det\bfX\,\,\mathbb{I}[\bfX]\label{eq:index_reflecprop}
\eea
where $\text{rk}\bfX=\sum_{i}n_{i}$ and $\det\bfX=\prod_{i}e^{n_{i}x_{i}}$. The index in the 5d notation can be written using the character of the anti-symmetrization,
\begin{equation}
        \wedge^\bullet \mathbf{X} = \sum_{n=0}^\infty (-1)^n \wedge^n \mathbf{X}
        \, ,\quad \mathbb{I}[\mathbf{X}] = \operatorname{ch} \wedge^\bullet \mathbf{X}^\vee. 
\end{equation}
Since $\wedge^n \bfX = 0$ for $n > \text{rk}\bfX$, the summation is finite if $\text{rk}\bfX < \infty$.
Then, the twisted index gives rise to the characteristic polynomial
\bea
\mathbb{I}[\bfX y^{\vee}]=\prod_{i}\left(1-y/y_{i}\right)=\sum_{k=0}^{\infty}(-y)^{k}\operatorname{ch}\wedge^k \mathbf{X}^\vee=:\operatorname{ch}\wedge^\bullet_y \mathbf{X} ,\quad \operatorname{ch} \mathbf{X}=\sum_{i}y_{i}.
\eea
We also have the $p$-th Adams operation on $\bfX$ defined as
\beq
    \text{ch}\,\bfX^{[p]}=\sum_{i}n_{i}e^{px_{i}}.
\eeq
From now on, the characters and the bundles will be identified and we omit the $\text{ch}$ and simply write it as
\beq
    \bfX=\sum_{i}n_{i}e^{x_{i}}.
\eeq

\paragraph{Four and Six}
The set of non-negative integers are denoted as $[n]=\{1,2,\ldots,n\}$ where $n\in\mathbb{N}$. Let $\underline{\textbf{4}}$ denote the set $[4]$, $\underline{\textbf{6}}$ denote the set of 2-element subsets of $\four$, and $\four^{\vee}$ denote the set of 3-element subsets of $\four$ as
\bea
    \four=\{1,2,3,4\},\quad \six=\{12,13,14,23,24,34\},\quad \four^{\vee}=\{123,124,134,234\}.
\eea
The order is defined in the lexicographic order as $12<13<14<23<24<34$. The complement $\bar{A}$ of $A\in\six$ is defined for example as
\bea
    A=12,\quad \bar{A}=34.
\eea
Note that $\four\simeq \four^{\vee}$ under the map $a\in\four\leftrightarrow\bar{a}\in\four^{\vee} $. We introduce the set $\three$ as the quotient $\six/\sim$ where $A\sim\bar{A}$ is
\bea
    (12)\sim (34),\quad (13)\sim(24),\quad (23)\sim(14)
\eea
and choose the representative as $A=a4,\,\, a\in[3]$. We also use 
\bea
    A=(ab),\quad \text{sup}(A)=b,\quad \text{inf}(A)=a
\eea
for $a<b$.

From the geometric viewpoint, the set $\four$ denotes the complex dimension 1 and 3 submanifolds of the Calabi--Yau four-fold, while the set $\six$ denotes the complex dimension 2 submanifolds. We can summarize the data of this in a tetrahedron (see Figure \ref{fig:complex}).
\begin{figure}[t]
    \centering
    \includegraphics[width=7cm]{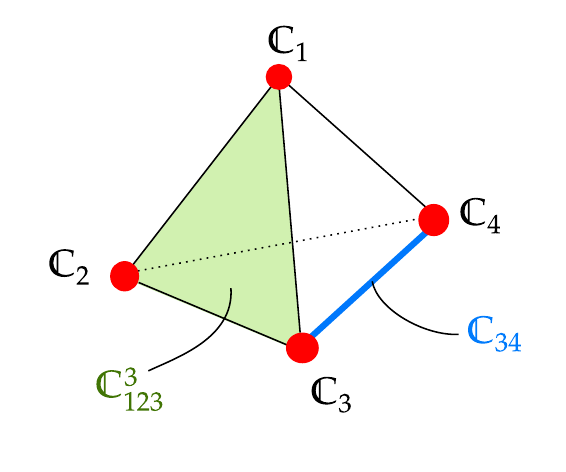}
    \caption{The four vertices of the tetrahedron correspond to the $\mathbb{C}_{a}\,\,(a\in\four)$, the six edges connecting two vertices of the tetrahedron correspond to the $\mathbb{C}^{2}_{A}\,\,(A\in\six)$, the four faces surrounded by three vertices and the three edges connecting them correspond to the complex 3d subspace $\mathbb{C}^{3}_{\bar{a}}\,\,(a\in\four)$, and the whole tetrahedron correspond to the $\mathbb{C}^{4}_{\four}$. }
    \label{fig:complex}
\end{figure}

\paragraph{Data}
The coordinates of $\mathbb{C}^4$ 
are denoted $z_{i}\,\,(i=1,2,3,4)$. The $\Omega$-deformation parameters are 
\bea
    q_{a}=e^{\epsilon_{a}},\quad \sum_{a\in\four}\epsilon_{a}=0,
\eea
where we simply omit the $\mathbb{S}^{1}$ radius. This is the Cartan torus of $\SU(4)$, which acts on the four complex coordinates as
\bea
  (z_{1},z_{2},z_{3},z_{4})\mapsto (q_{1}z_{1},q_{2}z_{2},q_{3}z_{3},q_{4}z_{4})  
\eea
with the condition $q_{1}q_{2}q_{3}q_{4}=1$.

We also use the notations: 
\bea
    \bfQ_{a}=q_{a}=e^{\epsilon_{a}},\quad \bfP_{a}=\wedge \bfQ_{a}=1-q_{a},\quad \bfP_{a}^{\vee}=1-q_{a}^{-1},\quad a\in\four
\eea
and for any subset $S\subseteq\four$
\bea
    \bfQ_{S}=\prod_{a\in S}\bfQ_{a},\quad \bfP_{S}=\prod_{a\in S}\bfP_{a}.
\eea
For example, we have
\bea
    &\bfP_{A}=(1-q_{1})(1-q_{2}),\quad A=12,\\
    &\bfP_{abc}=(1-q_{a})(1-q_{b})(1-q_{c}),\quad a,b,c\in\four.
\eea
Some properties of the index are
\bea
    q_{\four}=q_{\emptyset}=1,\quad \bfP_{\four}=\bfP_{1}\bfP_{2}\bfP_{3}\bfP_{4},\quad q_{\bar{S}}=q_{S}^{-1},\quad
    \bfP^{\vee}_{S}=(-1)^{|S|}q_{S}^{-1}\bfP_{S},\label{eq:dualorigamiprop}
\eea
for any subset $S\subseteq\four$. Note here, we denote $\bar{S}$ as the complement of the subset $S$. For later use, we also define $\bfP_{\bar{S}}=\prod_{a\in\bar{S}}\bfP_{a}$.

For example, we have 
\bea
    \bfP_{\bar{4}}=\bfP_{123}=(1-q_{1})(1-q_{2})(1-q_{3}).
\eea
For the subsets $S\subset\four$, we use the following notation
\bea
    S=\begin{dcases}
        a,\quad a=1,2,3,4,\\
        A,\quad A\in\six,\\
        \bar{a},\quad \bar{a}\in\{123,124,134,234\}
    \end{dcases}
\eea
We can also decompose $\bfP_{\four}$ as 
\bea
\bfP_{\four}=\bfP_{\three}+\bfP_{\three}^{\vee},\quad \bfP_{\four}=\bfP_{\four}^{\vee}.
\eea
Using the reflection property of the index in \eqref{eq:index_reflecprop}, we obtain the reflection properties of the index of the bundles $\bfP_{a}$ as 
\bea
    \mathbb{I}[\bfP_{a}x]=q_{a}^{-1}\mathbb{I}[\bfP_{a}^{\vee}x^{\vee}],\quad \mathbb{I}[\bfP_{ab}x]=\mathbb{I}[\bfP_{ab}^{\vee}x^{\vee}],\quad \mathbb{I}[\bfP_{abc}x]=\mathbb{I}[\bfP_{abc}^{\vee}x^{\vee}],\quad \mathbb{I}[\bfP_{\four}x]=\mathbb{I}[\bfP_{\four}^{\vee}x^{\vee}] \label{eq:reflec_origamiindex}
\eea
where $a,b,c\in\four$ and $a\neq b, b\neq c, a\neq c$.

\paragraph{Structure functions}
Let us introduce special functions which we call structure functions that will be used frequently later:\footnote{\label{footnote:structure-function}The terminology \emph{structure functions} comes from an observation that if we take $\bar{a}=\bar{4}=123$, and take the limit $q_{4}\rightarrow 1$, the function $g_{\bar{a}}(x)$ will be \bea
    g_{123}(x)=\frac{(1-q_{1}x)(1-q_{2}x)(1-q_{3}x)(1-q_{4}^{-1}x)}{(1-x)(1-q_{4}^{-1}q_{1}^{-1}x)(q_{4}^{-1}q_{2}^{-1}x)(1-q_{4}^{-1}q_{3}^{-1}x)}\xrightarrow{q_{4}\rightarrow 1} \frac{(1-\sfq_{1}x)(1-\sfq_{2}x)(1-\sfq_{3}x)}{(1-\sfq_{1}^{-1}x)(1-\sfq_{2}^{-1}x)(1-\sfq_{3}^{-1}x)}
\eea where $q_{1,2,3}\rightarrow \sfq_{1,2,3}$ and $\sfq_{1}\sfq_{2}\sfq_{3}=1$. It resembles the structure function of the quantum toroidal $\mathfrak{gl}_{1}$ (see section \ref{sec:QTgl1}).}
\bea\label{eq:struct_funct}
\mathscr{V}_{a}(x)&=\mathbb{I}[-\bfP_{a}^{\vee}x^{\vee}]=\frac{1-q_{a}x}{1-x},\quad a\in\four,\\
\mathscr{S}_{ab}(x)&=\mathbb{I}[-\bfP_{ab}^{\vee}x^{\vee}]=\frac{(1-q_{a}x)(1-q_{b}x)}{(1-x)(1-q_{a}q_{b}x)},\quad ab\in\six,\\
g_{\bar{a}}(x)&=\mathbb{I}[-\bfP_{\bar{a}}^{\vee}x^{\vee}]=\frac{\prod_{i\neq a}(1-q_{i}x)(1-q_{a}^{-1}x)}{(1-x)\prod_{i\neq a}(1-q_{a}^{-1}q_{i}^{-1}x)},\quad a\in\four,\\
\mathcal{A}_{\mathbb{C}^{4}}(x)&=\mathbb{I}[-\bfP_{\four}^{\vee}x^{\vee}]=\frac{\prod_{i=1}^{4}(1-q_{i}x)\prod^{4}_{i=1}(1-q_{i}^{-1}x)}{(1-x)^{2}\prod_{i\neq j}(1-q_{i}q_{j}x)}.
\eea
Obviously, we have the following properties:
\bea
&g_{abc}(x)=\frac{\mathscr{S}_{ab}(x)}{\mathscr{S}_{ab}(q_{c}x)}=\frac{\mathscr{S}_{bc}(x)}{\mathscr{S}_{bc}(q_{a}x)}=\frac{\mathscr{S}_{ac}(x)}{\mathscr{S}_{ac}(q_{b}x)},\\
&\mathscr{S}_{ab}(x)=\frac{\mathscr{V}_{a}(x)}{\mathscr{V}_{a}(q_{b}x)}=\frac{\mathscr{V}_{b}(x)}{\mathscr{V}_{b}(q_{a}x)},\quad \mathcal{A}_{\mathbb{C}^{4}}(x)=\frac{g_{\bar{a}}(x)}{g_{\bar{a}}(q_{a}x)},\qquad a,b,c\in\four.
\eea
Using \eqref{eq:index_reflecprop}, \eqref{eq:dualorigamiprop}, and \eqref{eq:reflec_origamiindex}, we also have the following reflection formulas:
\bea\label{eq:reflec_structfunc}
\mathscr{V}_{a}(x)&=\mathbb{I}[-\bfP_{a}^{\vee}x^{\vee}]=q_{a}\mathbb{I}[-\bfP_{a}x]=q_{a}\mathscr{V}_{a}(q_{a}^{-1}x^{-1})^{-1},\\
\mathscr{S}_{ab}(x)&=\mathbb{I}[-\bfP_{ab}^{\vee}x^{\vee}]=\mathbb{I}[-\bfP_{ab}x]=\mathscr{S}_{ab}(q_{a}^{-1}q_{b}^{-1}x^{-1}) \\
g_{\bar{a}}(x)&=\mathbb{I}\left[-\bfP_{\bar{a}}^{\vee}x^{\vee}\right]=\mathbb{I}\left[-\bfP_{\bar{a}}x\right]=g_{\bar{a}}(q_{a}x^{-1})^{-1},\\
\mathcal{A}_{\mathbb{C}^{4}}(x)&=\mathbb{I}\left[-\bfP_{\four}^{\vee}x^{\vee}\right]=\mathbb{I}[-\bfP_{\four}x]=\mathcal{A}_{\mathbb{C}^{4}}(x^{-1}).
\eea

\paragraph{Example: $\widehat{A}_0$ quiver gauge theory on $\mathbb{C}^{2}_{12}\times \mathbb{S}^{1}$ (5d $\mathcal{N}=1^*$ theory)}\mbox{}

The $k$-instanton moduli space of $\U(n)$ gauge theory on $\mathbb{C}^{2}_{12}$ denoted by $\mathfrak{M}_{n,k}$ is parametrized by the topological data $(n,k)$ where $k$ is the instanton number. The framing space and instanton space are 
\bea
\bfN=\mathbb{C}^{n},\quad \bfK=\mathbb{C}^{k}
\eea
and the automorphism groups $\GL(\bfN)$, $\GL(\bfK)$ give the characters
\bea
\bfN=\sum_{\alpha=1}^{n}e^{a_{\alpha}}=\sum_{\alpha=1}^{n}v_{\alpha},\qquad \bfK=\sum_{I=1}^{k}e^{\phi_{I}}=\sum_{I=1}^{k}x_{I}.
\eea
The observable sheaf, which is obtained from the universal sheaf via localization, is defined as
\bea
    \bfY\equiv\bfN-\bfP_{12}\bfK
\eea
and the vector multiplet contribution, which is obtained by the tangent bundle over the instanton moduli space $T\mathfrak{M}_{n,k}$, is given as 
\bea
    \mathbf{v}_{\text{vect.}}=\frac{\bfY^{\vee}\bfY}{\bfP_{12}}=\frac{\bfN^{\vee}\bfN}{\bfP_{12}}-\bfN^{\vee}\bfK-\bfQ_{12}^{\vee}\bfK^{\vee}\bfN+\bfP_{12}^{\vee}\bfK^{\vee}\bfK.
\eea
Without adding any other multiplets, the index of this contribution computes the partition function of the pure super Yang--Mills theory. To add an adjoint hypermultiplet with mass $m$, we need to add the following term to the vector multiplet contribution
\begin{equation}
\bfH_{\text{adj}}=-\mathbf{M}_{\text{adj}}\frac{\bfY^{\vee}\bfY}{\bfP_{12}},\quad \mathbf{M}_{\text{adj}}=e^{m}\eqqcolon\mu. 
\end{equation}
Then, the 5d $\mathcal{N}=1^{\ast}$ theory partition function will be given by the index of
\begin{equation}
    \mathbf{v}=\frac{\bfY^{\vee}(1-\bfM_{\text{adj}})\bfY}{\bfP_{12}}.
\end{equation}
Using the convention $q_{3}=\mu^{-1},\,q_{4}=\mu q_{12}^{-1}$ obeying $q_1 q_2 q_3 q_4 = 1$, we have 

\bea
\label{eq:affinequiver}
    &\mathbf{v}=\frac{\bfY^{\vee}\bfP_{3}^{\vee}\bfY}{\bfP_{12}}=\frac{\bfN^{\vee}\bfP_{3}^{\vee}\bfN}{\bfP_{12}}-\bfP_{3}^{\vee}\bfN^{\vee}\bfK-\bfQ_{12}^{\vee}\bfP_{3}^{\vee}\bfK^{\vee}\bfN+\bfP_{123}^{\vee}\bfK^{\vee}\bfK=\mathring{\mathbf{v}}+\mathbf{v}_{\text{inst}},\\
    &\mathring{\mathbf{v}}=\frac{\bfN^{\vee}\bfP_{3}^{\vee}\bfN}{\bfP_{12}},\quad \mathbf{v}_{\text{inst}}=-\bfP_{3}^{\vee}\bfN^{\vee}\bfK-\bfQ_{12}^{\vee}\bfP_{3}^{\vee}\bfK^{\vee}\bfN+\bfP_{123}^{\vee}\bfK^{\vee}\bfK.
\eea
The instanton partition function will be $\mathcal{Z}_{(n,k)}=\mathbb{I}[\mathbf{v}_{\text{inst}}]$ which gives the contour integral formula.
\begin{proposition}[Losev--Moore--Nekrasov--Shatashvili (LMNS) formula~\cite{Losev:1997tp,Moore:1997dj,Lossev:1997bz}]
The $k$-instanton partition function of 5d $\mathcal{N}=1^*$ $\U(n)$ gauge theory defined on $\mathbb{C}^2 \times \mathbb{S}^1$ is given by the following contour integral,
\bea\label{eq:affinecontour}
    \mathcal{Z}_{(n,k)}=&\frac{1}{k!}\left(\frac{(1-q_{12})}{(1-q_{1})(1-q_{2})}\mathscr{S}_{12}(q_{3})\right)^{k}\\
    &\quad \times\oint\prod_{I=1}^{k}\frac{\text{d}x_{I}}{2\pi\iota x_{I}}\prod_{I=1}^{k}\frac{P(q_{3}^{-1}x_{I})\widetilde{P}(q_{4}^{-1}x_{I})}{P(x_{I})\widetilde{P}(q_{34}^{-1}x_{I})}\prod_{I\neq J}^{k}\frac{\mathscr{S}_{12}(q_{3}x_{I}/x_{J})}{\mathscr{S}_{12}(x_{I}/x_{J})}
\eea
where $\iota=\sqrt{-1}$ and the gauge polynomials
\begin{equation}
    P(x)=\mathbb{I}[\bfN^{\vee}x]=\prod_{\alpha=1}^{n}\left(1-\frac{v_{\alpha}}{x}\right),\quad \widetilde{P}(x)=\mathbb{I}[x^{-1}\mathbf{N}]=\prod_{\alpha=1}^{n}\left(1-\frac{x}{v_{\alpha}}\right).
\end{equation}
This integral is taken over the Cartan torus of $\mathrm{GL}(\mathbf{K})$, which is interpreted as the Jeffrey--Kirwan residue~\cite{Jeffrey1993LocalizationFN} of the integrand.
\end{proposition}
The total instanton partition function is then given as
\begin{equation}
    \mathcal{Z}_{\text{inst.}}=\sum_{k=0}^{\infty}\mathfrak{q}^{k}\mathcal{Z}_{(n,k)},
\end{equation}
where $\mathfrak{q}^{k} = \exp(2 k \pi \iota \tau)$ is the topological term with the complexified coupling constant, $\tau = \frac{\theta}{2\pi} + \frac{4 \pi \iota}{g^2}$, with $\theta$ the theta angle and $g^2$ the gauge coupling.

The contour integral will actually localize on the fixed points classified by $n$-tuples of Young diagrams~\cite{Nakajima:1999,Nekrasov:2002qd,Nekrasov:2003rj}:
\begin{equation}
\begin{split}
&\vec{v}=(v_{\alpha})_{\alpha=1,\ldots,n},\quad  \vec{\lambda}=(\lambda^{(\alpha)})_{\alpha=1,\ldots,n},\quad |\vec{\lambda}|=\sum_{\alpha=1}^{n}|\lambda^{(\alpha)}|=k,\\
&\{x_{I}\}_{I=1,\ldots,k}\longrightarrow \{\chi_{12,v_{\alpha}}(\Bbox)=v_{\alpha}q_{1}^{i-1}q_{2}^{j-1}\}_{\alpha=1,\ldots,n,\,\Abox=(i,j)\in\lambda^{(\alpha)}}.
\end{split}
\end{equation}
The instanton partition function is then obtained by summing the residues coming from each pole. The character of the instanton bundle $\bfK$ at each fixed point $\vec{\lambda}$ is given by 
\bea
\bfK|_{\vec{\lambda}}=\sum_{\alpha=1}^{n}\sum_{\Abox\in\lambda^{(\alpha)}}\chi_{12,v_{\alpha}}(\Bbox).
\eea
Inserting this in \eqref{eq:affinequiver} and taking the index, we obtain the contribution to the instanton partition function of 5d $\mathcal{N}=1^*$ theory from each fixed point $\vec{\lambda}$ as
\bea
\mathcal{Z}_{12}^{\D4}[\vec{v},\vec{\lambda}\,;q_{3}]=\prod_{\alpha,\beta=1}^{n}\frac{\mathsf{N}_{12}(q_{3}v_{\alpha},\lambda^{(\alpha)}\,|\,v_{\beta},\lambda^{(\beta)})}{\mathsf{N}_{12}(v_{\alpha},\lambda^{(\alpha)}\,|\,v_{\beta},\lambda^{(\beta)})}
\eea
where we defined the Nekrasov factor as 
\bea
    \mathsf{N}_{12}(v_{1},\lambda^{(1)}\,|\,v_{2},\lambda^{(2)})&=\prod_{\Abox\in\lambda^{(1)}}\left(1-\frac{q_{12}\chi_{12,v_{1}}(\Bbox)}{v_{2}}\right)\prod_{\Abox\in\lambda^{(2)}}\left(1-\frac{v_{1}}{\chi_{12,v_{2}}(\Bbox)}\right)\prod_{\substack{\Abox\in\lambda^{(1)}\\\AboxF\in\lambda^{(2)}}}\mathscr{S}_{12}\left(\frac{\chi_{12,v_{1}}(\Bbox)}{\chi_{12,v_{2}}(\BboxF)}\right).
\eea
The instanton partition function will then be given as follows.
\begin{proposition}[\cite{Nekrasov:2002qd,Nekrasov:2003rj}]
The instanton partition function of 5d $\mathcal{N}=1^*$ $\U(n)$ gauge theory is given by summation over $n$-tuple partitions,
\bea
    \mathcal{Z}_{\text{inst.}}=\sum_{\vec{\lambda} }\mathfrak{q}^{|\vec{\lambda}|}\mathcal{Z}_{12}^{\D4}[\vec{v},\vec{\lambda}\,;q_{3}].
\eea
The index $\D4$ comes from the fact that this theory comes from $n$ $\D4$-branes in the $\mathbb{C}^{4}$ set-up with $\Omega$-background.     
\end{proposition}

\paragraph{Nekrasov factors for 5d theory}
Using the quadrality in $q_{a}\,(a\in\four)$, we define similar Nekrasov factors\footnote{The Nekrasov factors can be written using the arm and leg length of each box in the Young diagrams, $a_\lambda(\Abox) = \lambda_{i}-j$ and $\ell_\lambda(\Abox) = \lambda_j^{\rmT} - i$ for $\Abox = (i,j)$, such as \begin{equation}\label{eq:D4armlegformula}
    \mathsf{N}_{ab}(v_{1},\lambda\,|\,v_{2},\nu)=\prod_{\Abox\in\lambda}\left(1-Qq_{a}^{\ell_{\lambda}(\Abox)}q_{b}^{-a_{\nu}(\Abox)}\right)\prod_{\Abox\in\nu}\left(1-Qq_{a}^{-\ell_{\nu}(\Abox)}q_{b}^{a_{\lambda}(\Abox)+1}\right),\quad Q=v_{1}/v_{2}.
\end{equation} From the algebraic viewpoint, using the product form in \eqref{eq:D4Nekrasovfactor} is much useful so we will use it rather than the simplified form in \eqref{eq:D4armlegformula}. } which can be used to discuss gauge theories on $\mathbb{C}^{2}_{A}\times \mathbb{S}^{1}\,(A\in\six)$:
\bea
 \mathsf{N}_{A}(v_{1},\lambda^{(1)}\,|\,v_{2},\lambda^{(2)})&=\prod_{\Abox\in\lambda^{(1)}}\left(1-\frac{q_{A}\chi_{A,v_{1}}(\Bbox)}{v_{2}}\right)\prod_{\Abox\in\lambda^{(2)}}\left(1-\frac{v_{1}}{\chi_{A,v_{2}}(\Bbox)}\right)\prod_{\substack{\Abox\in\lambda^{(1)}\\\AboxF\in\lambda^{(2)}}}\mathscr{S}_{A}\left(\frac{\chi_{A,v_{1}}(\Bbox)}{\chi_{A,v_{2}}(\BboxF)}\right),\label{eq:D4Nekrasovfactor}
\eea
where we defined the box content as
\bea
    \chi_{ab,v}(\Bbox)=vq_{a}^{i-1}q_{b}^{j-1},\quad A=ab\in\six,\quad \Bbox=(i,j)\in\lambda.
\eea
\begin{lemma}
The recursion formulas of the 5d Nekrasov factors are given as follows,
\bea\label{eq:5dNekrasovrecursion}
    &\frac{\mathsf{N}_{A}(v_{1},\lambda^{(1)}+\Bbox\,|\,v_{2},\lambda^{(2)})}{\mathsf{N}_{A}(v_{1},\lambda^{(1)}\,|\,v_{2},\lambda^{(2)})}=\mathscr{Y}^{A\,\vee}_{\lambda^{(2)},v_{2}}(q_{A}\chi_{A,v_{1}}(\Bbox))\\
   &\frac{\mathsf{N}_{A}(v_{1},\lambda^{(1)}\,|\,v_{2},\lambda^{(2)}+\BboxF)}{\mathsf{N}_{A}(v_{1},\lambda^{(1)}\,|\,v_{2},\lambda^{(2)})}=\mathscr{Y}_{\lambda^{(1)},v_{1}}^{A}(\chi_{A,v_{2}}(\BboxF)),
\eea
where we define the $\mathscr{Y}$-functions,
\bea
   \mathscr{Y}_{\lambda,v}^{A}(x)&=\left(1-\frac{v}{x}\right)\prod_{\Abox\in\lambda}\mathscr{S}_{A}\left(\frac{\chi_{A,v}(\Bbox)}{x}\right)=\frac{\prod\limits_{\Abox\in A(\lambda)}\left(1-\chi_{A,v}(\Bbox)/x\right)}{\prod\limits_{\Abox\in R(\lambda)}\left(1-q_{A}\chi_{A,v}(\Bbox)/x\right)},\\
   \mathscr{Y}^{A\,\vee}_{\lambda,v}(x)&=\left(1-\frac{x}{v}\right)\prod_{\Abox\in \lambda}\mathscr{S}_{A}\left(q_{A}^{-1}\frac{x}{\chi_{A,v}(\Bbox)}\right)=\frac{\prod\limits_{\Abox\in A(\lambda)}\left(1-x/\chi_{A,v}(\Bbox)\right)}{\prod\limits_{\Abox\in R(\lambda)}\left(1-q_{A}^{-1}x/\chi_{A,v}(\Bbox)\right)}
\eea    
\end{lemma}

This $\mathscr{Y}$-functions have been introduced to describe the Seiberg--Witten curve of generic quiver gauge theory~\cite{Nekrasov:2012xe}.
Physically, it can be understood as the contribution of the codimension four defect operator (Wilson loop along $\mathbb{S}^1$ in the current setup) under the instanton background~\cite{Kim:2016qqs}.
See also section~\ref{sec:qqpartitionfunct}.

\subsection{Magnificent four}\label{sec:M4partitionfunction}
Let us first consider the magnificent four setup introduced and discussed in \cite{Nekrasov:2017cih,Nekrasov:2018xsb} (see also section~\ref{sec:physicalsetup}). Mathematically, the partition function is interpreted as the generating function the Donaldson-Thomas invariants of the Calabi--Yau four-fold $\mathbb{C}^{4}$~\cite{Cao:2017swr,Cao:2019tvv} and physically, it corresponds to the setup discussed in \cite{Billo:2009gc,Billo:2009di,Bonelli:2020gku,Billo:2021xzh}. Let us review how to derive the instanton partition function. 

The associated cotangent bundle of $\mathbb{C}^{4}$ is decomposed into 
\bea
   \mathbf{Q}=\bigoplus_{i=1}^{4}\mathbf{Q}_{i},\quad \text{ch}\,\mathbf{Q}_{i}=q_{i}
\eea
where the Calabi--Yau condition is imposed $q_{1}q_{2}q_{3}q_{4}=1$. The framing and instanton bundles $ \bfN=\mathbf{n}-\bar{\mathbf{n}}$, $ \mathbf{K}$ with the characters are given by
\bea
    \mathbf{N}=\mathbf{n}-\bar{\mathbf{n}}=\sum_{\alpha=1}^{n}(e^{a_{\alpha}}-e^{b_{\alpha}})=\sum_{\alpha=1}^{n}(v_{\alpha}-\bar{v}_{\alpha}),\quad \mathbf{K}=\sum_{I=1}^{k}e^{\phi_{I}}=\sum_{I=1}^{k}x_{I}.
\eea
Note that the contribution of $\bar{\mathbf{n}}$ in $\mathbf{N}$ corresponds to $\overline{\D8}$-branes, that are indispensable to stabilize the brane configuration. The observable sheaf is then defined as
\bea\label{eq:D8Ybundle}
    \mathbf{Y}=\mathbf{N}-\mathbf{P}_{\four}\mathbf{K}
\eea
and the vector multiplet contribution is given as 
\bea
    \mathbf{V}&=\frac{\mathbf{Y}^{\vee}\mathbf{Y}}{\mathbf{P}_{\four}}=\mathring{\mathbf{V}}+\mathbf{V}_{\text{inst}},\quad\mathring{\mathbf{V}}=\frac{\mathbf{N}^{\vee}\mathbf{N}}{\mathbf{P}_{\four}},\quad
    \mathbf{V}_{\text{inst}}=-\mathbf{N}^{\vee}\mathbf{K}-\mathbf{K}^{\vee}\mathbf{N}+\mathbf{P}_{\four}^{\vee}\mathbf{K}^{\vee}\mathbf{K}.
\eea
Due to the lack of a perfect obstruction theory for the CY four-folds, we need to take the ``square root" part of the total bundle $\bfV$ (see e.g.,~\cite{Borisov:2017GT,Oh:2020rnj}), which we choose
\bea
\mathbf{V}_{\text{inst}}&=\mathbf{v}_{\text{inst}}+\mathbf{v}_{\text{inst}}^{\vee},\quad
\mathbf{v}_{\text{inst}}=-\mathbf{N}^{\vee}\mathbf{K}+\mathbf{P}_{123}^{\vee}\mathbf{K}^{\vee}\mathbf{K}\label{eq:mag4ch}
\eea
where we used 
\bea
    \mathbf{P}_{\four}=\mathbf{P}_{123}+\mathbf{P}_{123}^{\vee}.
\eea
We note that other decompositions using $\mathbf{P}_{234},\,\mathbf{P}_{134}$, and $\mathbf{P}_{124}$ gives the same result up to sign factors \cite{Nekrasov:2017cih,Nekrasov:2018xsb}.

Following the procedure in the previous section, we have the following expression for the instanton partition function.
\begin{proposition}[\cite{Nekrasov:2017cih,Nekrasov:2018xsb}]
The total partition function of the magnificent four system is given by
\bea
\mathcal{Z}^{\D8}_{\text{inst.}}=\sum_{k=0}^{\infty}\mathfrak{q}^{k}\mathcal{Z}^{\D8}_{k} .
\eea    
Each contribution is given by the contour integral over the Cartan torus of $\mathrm{GL}(\mathbf{K})$, which is interpreted as the Jeffrey-Kirwan residue of the integrand,
\bea
    \mathcal{Z}^{\D8}_{k}=\mathbb{I}[\mathbf{v}_{\text{inst}}]=\frac{\mathcal{G}_{\bar{4}}^{k}}{k!}\oint \prod_{I=1}^{k}\frac{dx_{I}}{2\pi\iota x_{I}}\prod_{I=1}^{k}\frac{\overline{P}(x_{I})}{P(x_{I})}\prod_{I\neq J}g_{\bar{4}}\left(\frac{x_{I}}{x_{J}}\right)^{-1}\label{eq:D8integral}
\eea
where we define
\bea
&\mathcal{G}_{\bar{a}}=\frac{\prod_{i\neq a}(1-q_{a}^{-1}q_{i}^{-1})}{(1-q_{a}^{-1})\prod_{i\neq a}(1-q_{i})},\quad a\in\four\\
&P(x)=\prod_{\alpha=1}^{n}\left(1-\frac{v_{\alpha}}{x}\right),\quad \overline{P}(x)=\prod_{\alpha=1}^{n}\left(1-\frac{\bar{v}_{\alpha}}{x}\right).
\eea
\end{proposition}
Note that this contour integral formula is compatible with the symmetries of the moduli space given in \eqref{eq:D8ADHMvariablesymmetry}, \eqref{eq:D8ADHMconstrsymmetry} (see also the discussion in footnote~\ref{footnote:D2ADHMintegral}).

Actually, the poles (the equivariant fixed point in the moduli space) of the multi-integral are classified by $n$-tuples of solid partitions and we can rewrite the instanton partition function in a sum of solid partitions:
\begin{align}
\begin{split}
&\vec{v}=(v_{\alpha})_{\alpha=1,\ldots,n},\quad \vec{\rho}=(\rho^{(\alpha)})_{\alpha=1,\ldots,n},\quad |\rho|=\sum_{\alpha=1}^{n}|\rho^{(\alpha)}|=k,\\
&\{x_{I}\}_{I=1,\ldots,k}\longrightarrow\{\chi_{\four,v_{\alpha}}(\hcube)=v_{\alpha}q_{1}^{i-1}q_{2}^{j-1}q_{3}^{k-1}q_{4}^{l-1}\}_{\alpha=1,\ldots,n,\,\shcube=(i,j,k,l)}.
\end{split}
\end{align}
Note that the hyper cubes in the origin of the solid partitions have coordinates $v_{\alpha}\,(\alpha=1,\ldots,n)$ but not $\bar{v}_{\alpha}\,(\alpha=1,\ldots,n)$. This is because in the contour integral formula \eqref{eq:D8integral}, only the factors coming from $P(x)$ give poles but $\overline{P}(x)$ will not. At the fixed point $\vec{\rho}$, the character of the bundle $\bfK$ will be
\bea
\bfK|_{\vec{\rho}}=\sum_{\alpha=1}^{n}\sum_{\shcube\in\rho^{(\alpha)}}\chi_{\four,v_{\alpha}}(\hcube).
\eea
Since a $\D8$-brane will appear with a $\overline{\D8}$ as a pair, it is natural to use the following notation
\bea
    \bar{v}_{\alpha}=K_{\alpha}v_{\alpha},\quad \alpha=1,\ldots,n,\quad 
    \bfN=\sum_{\alpha=1}^{n}(1-K_{\alpha})v_{\alpha}
\label{eq:antiD8parameter}
\eea
where\footnote{Later in section \ref{sec:QTgl1}, we will see that this notation of $\{K_\alpha\}_{\alpha}$ is motivated from the central charge of the MacMahon representation of quantum toroidal $\mathfrak{gl}_{1}$. } $K_{\alpha}$ physically represents the distance between the $\D8_{\alpha}$ and $\overline{\D8}_{\alpha}$ branes. Inserting this to \eqref{eq:mag4ch} and taking the index, we obtain
\bea
    \mathcal{Z}^{\D8}_{\four;4}[\vec{v},\vec{\rho}\,;\vv{K}]&=\prod_{\alpha=1}^{n}\mathcal{Z}_{\four;4}^{\D8}[\rho^{(\alpha)};K_{\alpha}]\prod_{\beta>\alpha}\mathcal{Z}^{\D8\tbar\D8}_{K_{\alpha}|K_{\beta}}(v_{\alpha},\rho^{(\alpha)}\,|\,v_{\beta},\rho^{(\beta)}),\\
    \mathcal{Z}^{\D8}_{\four;4}[\rho\,;K]&=\prod_{\shcube\in\rho}\frac{(1-Kv/\chi_{\four,v}(\shcube))}{(1-v/\chi_{\four,v}(\shcube))}\prod_{\shcube,\shcube'\in\rho}g_{\bar{4}}\left(\frac{\chi_{\four,v}(\hcube)}{\chi_{\four,v}(\hcube')}\right)^{-1},\\
    \mathcal{Z}^{\D8\tbar\D8}_{K_{1}|K_{2}}(x_{1},\rho^{(1)}\,|\,x_{2},\rho^{(2)})&=\prod_{\shcube'\in\rho^{(2)}}\left(\frac{1-K_{1}x_{1}/\chi_{\four,x_{2}}(\hcube')}{1-x_{1}/\chi_{\four,x_{2}}(\hcube')}\right)\prod_{\shcube\in\rho^{(1)}}\left(K_{2}\frac{1-K_{2}^{-1}\chi_{\four,x_{1}}(\hcube)/x_{2}}{1-\chi_{\four,x_{1}}(\hcube)/x_{2}}\right)\\
    &\qquad\qquad\times \prod_{\substack{\shcube\in\rho^{(1)}\\\shcube'\in\rho^{(2)}}}\mathcal{A}_{\mathbb{C}^{4}}\left(\frac{\chi_{\four,x_{1}}(\hcube)}{\chi_{\four,x_{2}}(\hcube')}\right)^{-1}.
\label{eq:mag4Nekrasovfact}
\eea
The factor $\mathcal{Z}_{\four;4}^{\D8}[\rho\,;K]$ comes from the $\U(1|1)$ gauge theory for each $\D8\tbar\overline{\D8}$ stack. The subindex $4$ comes from the fact we are using $\bfP_{\bar{4}}=\bfP_{123}$ for the square root part. For later use, we introduce 
\bea
\mathcal{Z}^{\D8}_{\four;a}[\rho\,;K]&=\prod_{\shcube\in\rho}\frac{(1-Kv/\chi_{\four,v}(\shcube))}{(1-v/\chi_{\four,v}(\shcube))}\prod_{\shcube,\shcube'\in\rho}g_{\bar{a}}\left(\frac{\chi_{\four,v}(\hcube)}{\chi_{\four,v}(\hcube')}\right)^{-1},\quad a\in\four.\label{eq:mag4Nekrasovfact2}
\eea
The factor $\mathcal{Z}^{\D8\tbar\D8}_{K_{1}|K_{2}}(x_{1},\rho^{(1)}\,|\,x_{2},\rho^{(2)})$ comes from the open strings connecting the two different stacks $\D8_{\alpha}\tbar\overline{\D8}_{\alpha}$ and $\D8_{\beta}\tbar\overline{\D8}_{\beta}$.

\begin{proposition}[\cite{Nekrasov:2017cih,Nekrasov:2018xsb}]
The total instanton partition function of K-theoretic $\U(n|n)$ magnificent four system is given by summation over $n$-tuple solid partitions, 
\bea
\mathcal{Z}^{\D8}_{\text{inst.}}=\sum_{\vec{\rho}\in\mathcal{SP}}\mathfrak{q}^{|\vec{\rho}|}(-1)^{\sigma_{a}(\vec{\rho})}\mathcal{Z}_{\four;a}^{\D8}[\vec{v},\vec{\rho}\,;\vv{K}],\quad a\in\four\label{eq:mag4inst}
\eea
where $\sigma_{a}(\vec{\rho})$ is the sign factor depending on the choice of the orientation associated with the square root bundle.     
\end{proposition}

\paragraph{One-loop perturbative part}
For the one-loop perturbative part, we need to take the index of the square root part of $\mathring{\bfV}$. To do this, we first omit the singular part, specify an ordering in the parameters $\{v_{\alpha}\}$ and then take half of them: 
\bea
    &\mathring{\bfV}=\frac{\bfN^{\vee}\bfN}{\bfP_{\four}}=\frac{1}{\bfP_{\four}}\sum_{\alpha,\beta}(1-K_{\alpha}^{\vee})(1-K_{\beta})v_{\beta}/v_{\alpha}\rightarrow \mathring{\mathbf{v}}+\mathring{\mathbf{v}}^{\vee},\\
    &\mathring{\mathbf{v}}=\frac{1}{\mathbf{P}_{\four}}\sum_{v_{\beta}>v_{\alpha}}(1-K_{\alpha}^{\vee})(1-K_{\beta})v_{\beta}/v_{\alpha}.
\eea
The index will then be
\bea\label{eq:D8oneloop}
    \mathbb{I}[\mathring{\mathbf{v}}]&=\prod_{\alpha<\beta}\mathcal{Z}_{\text{1-loop}}^{\D8\tbar\D8}(v_{\alpha},K_{\alpha}\,|\,v_{\beta},K_{\beta})\eqqcolon\mathcal{Z}_{\text{1-loop}}^{\D8},\\
    \mathcal{Z}_{\text{1-loop}}^{\D8\tbar\D8}(v_{1},K_{1}\,|\,v_{2},K_{2})&=\exp\left(-\sum_{n>0}\frac{1}{n}\frac{(1-K_{2}^{-n})(1-K_{1}^{n})}{\bfP_{\four}^{[n]}}\left(\frac{v_{1}}{v_{2}}\right)^{n}\right).
\eea
The one-loop perturbative part can be written using the $q$-shifted factorial or $q$-deformed multi-gamma functions in Appendix \ref{sec:q-functions} and \eqref{eq:D8oneloop-qgamma}.

\subsection{Tetrahedron instanton }\label{sec:Tetrahedron_inst}
The instanton partition function of the tetrahedron instanton was first computed in \cite{Pomoni:2021hkn}. We also note that the contour integral formula of the $\mathbb{C}^3$-partition function was given in~\cite{Cirafici:2008sn,Kanno:2020ybd}. Let us review the explicit formulas. 

The total index of the tetrahedron instanton system resembles the magnificent four system. The different part is that the total observable sheaf $\bfY$ will be a sum of the observable sheaves\footnote{Note here that the subindex $\bar{a}$ of $\bfY_{\bar{a}},\bfN_{\bar{a}},\bfK_{\bar{a}}$ does not mean that these bundles are products of bundles with index $a$ such as $\bfN_{123}=\bfN_{1}\bfN_{2}\bfN_{3}$. } of $\bfY_{\bar{a}}$ corresponding to the theory on $\mathbb{C}^{3}_{\bar{a}}\times \mathbb{S}^{1}$:
\bea\label{eq:D6Ybundle}
\bfV=\frac{\bfY^{\vee}\bfY}{\bfP_{\four}},\quad \bfY=\sum_{a\in\four}\bfP_{a}\bfY_{\bar{a}},\quad \bfY_{\bar{a}}=\bfN_{\bar{a}}-\bfP_{\bar{a}}\bfK_{\bar{a}}
\eea
which gives 
\bea
    \bfV&=\mathring{\bfV}+\bfV_{\text{inst.}},\quad 
    \mathring{\bfV}=\sum_{a,b\in\four}\frac{\bfP_{a}^{\vee}\bfP_{b}}{\bfP_{\four}}\bfN_{\bar{a}}^{\vee}\bfN_{\bar{b}},\\
    \bfV_{\text{inst.}}&=\sum_{a,b\in\four}\left(-\bfP_{b}\bfN_{\bar{b}}\bfK_{\bar{a}}^{\vee}-\bfP_{a}^{\vee}\bfN_{\bar{a}}^{\vee}\bfK_{\bar{b}}+\bfP_{\four}\bfK_{\bar{a}}^{\vee}\bfK_{\bar{b}}\right),
\eea
where 
\bea
\bfN_{\bar{a}}=\sum_{\alpha=1}^{n_{\bar{a}}}e^{a_{\bar{a},\alpha}}=\sum_{\alpha=1}^{n_{\bar{a}}}v_{\bar{a},\alpha},\quad \bfK_{\bar{a}}=\sum_{I=1}^{k_{\bar{a}}}e^{\phi_{\bar{a},I}}=\sum_{I=1}^{k_{\bar{a}}}x_{\bar{a},I},\quad a\in\four. \label{eq:D6chdef}
\eea
Note that the observable sheaves can be rewritten as 
\bea\label{eq:tetratotalinst}
\bfY=\sum_{a\in\four}\bfP_{a}\bfN_{\bar{a}}-\bfP_{\four}\sum_{a\in\four}\bfK_{\bar{a}}=\sum_{a\in\four}\bfP_{a}\bfN_{\bar{a}}-\bfP_{\four}\bfK
\eea
where $\bfK=\sum_{a\in\four}\bfK_{\bar{a}}$ is the vector space introduced in section~\ref{sec:physicalsetup}. For later use, we decomposed the instanton bundle into components $\bfK_{\bar{a}}$ so that the instanton partition function after localization has a nice decomposition formula.
Similar to the magnificent four setup, we need to take the square root part of $\bfV$. For the instanton part, we choose the following square root part:
\bea
    \mathbf{v}_{\text{inst.}}&=\sum_{a\in\four}\left(-\mathbf{P}_{a}^{\vee}\bfN_{\bar{a}}^{\vee}\bfK_{\bar{a}}+\frac{\bfP_{\four}}{\bfP^{\vee}_{
a}}\bfK_{\bar{a}}^{\vee}\bfK_{\bar{a}}\right)-\sum_{a\in\four}\sum_{b\neq a}\bfP_{a}^{\vee}\bfN_{\bar{a}}^{\vee}\bfK_{\bar{b}}+\sum_{a<b}\bfP_{\four}\bfK_{\bar{a}}^{\vee}\bfK_{\bar{b}},\\
\bfV_{\text{inst.}}&=\mathbf{v}_{\text{inst.}}+\mathbf{v}_{\text{inst.}}^{\vee}
\label{eq:D6tetrainstch}
\eea
where the first two terms give the partition function of the theory on $\mathbb{C}^{3}_{\bar{a}}\times \mathbb{S}^{1}$ and the last two terms give the bifundamental contributions coming from open strings connecting $\mathbb{C}^{3}_{\bar{a}}\times \mathbb{S}^{1}$ and $\mathbb{C}^{3}_{\bar{b}}\times \mathbb{S}^{1}$.

The contour integral formula is obtained by inserting \eqref{eq:D6chdef} into \eqref{eq:D6tetrainstch} and taking the index.
\begin{proposition}[\cite{Pomoni:2021hkn}]
The $\underline{k}$-instanton partition function of the tetrahedron instanton system is given by the contour integral,
\bea
    \mathcal{Z}^{\D6}_{\underline{k}}=\mathbb{I}[\mathbf{v}_{\text{inst.}}]&=\frac{\underline{\mathcal{G}}^{\underline{k}}}{\underline{k}!}\oint\prod_{a\in\four}\prod_{I=1}^{k_{\bar{a}}}\frac{dx_{\bar{a},I}}{2\pi\iota x_{\bar{a},I}}\prod_{a,b\in\four}\prod_{\alpha=1}^{n_{\bar{a}}}\prod_{I=1}^{k_{\bar{b}}}\mathscr{V}_{a}\left(\frac{v_{\bar{a},\alpha}}{x_{\bar{b},I}}\right)\\
    &\qquad \times\prod_{a\in\four}\prod_{I\neq J}^{k_{\bar{a}}}g_{\bar{a}}\left(\frac{x_{\bar{a},J}}{x_{\bar{a},I}}\right)^{-1}\prod_{a<b}\prod_{I=1}^{k_{\bar{a}}}\prod_{J=1}^{k_{\bar{b}}}\mathcal{A}_{\mathbb{C}^{4}}\left(\frac{x_{\bar{a},I}}{x_{\bar{b},J}}\right)^{-1} 
\label{eq:D6integral}
\eea
where 
\bea
\underline{\mathcal{G}}^{\underline{k}}=\prod_{a\in\four}\mathcal{G}_{\bar{a}}^{k_{\bar{a}}},\qquad\underline{k}!=\prod_{a\in\four}k_{\bar{a}}!.
\eea
The total partition function is then given by
\bea
\mathcal{Z}^{\D6}_{\text{inst.}}=\sum_{k=0}^{\infty}\mathfrak{q}^{k}\sum_{\substack{(k_{\bar{a}})_{a\in\four},\\\sum_{a}k_{\bar{a}}=k}}\mathcal{Z}^{\D6}_{\underline{k}}.
\eea    
\end{proposition}
We also note that this contour integral formula is compatible with the symmetries given in \eqref{eq:D6gaugesymmetry}, \eqref{eq:D6flavorsymmetry}, \eqref{eq:D6rotationalsymmetry}, \eqref{eq:D6ADHMconstrsymmetry}.

The poles of the contour integral are classified by plane partitions:
\bea
&\underline{\vec{v}}=(\vec{v}_{\bar{a}})_{a\in\four}=(v_{\bar{a},\alpha})_{a\in\four}^{\alpha=1,\ldots,n_{\bar{a}}},\quad \underline{\vec{\pi}}=(\vec{\pi}_{\bar{a}})_{a\in\four}=(\pi^{(\alpha)}_{\bar{a}})_{a\in\four}^{\alpha=1,\ldots, n_{\bar{a}}},\quad |\underline{\vec{\pi}}|=\sum_{a\in\four}\sum_{\alpha=1}^{n_{\bar{a}}}|\pi_{\bar{a}}^{(\alpha)}|,\\
&\{x_{\bar{a},I}\}_{a\in\four}^{I=1,\ldots,k_{\bar{a}}}\longrightarrow \{\chi_{\bar{a},v_{\bar{a},\alpha}}(\cube)\}_{a\in\four,\,\,\scube\,\in \pi_{\bar{a}}^{(\alpha)}}^{\alpha=1,\ldots,n_{\bar{a}}},\quad \chi_{abc,v}(\cube)=vq_{a}^{i-1}q_{b}^{j-1}q_{c}^{k-1}.
\eea
At the fixed points, the character of $\bfK_{\bar{a}}$ will be 
\bea
\left.\bfK_{\bar{a}}\right|_{\vec{\pi}_{\bar{a}}}=\sum_{\alpha=1}^{n_{\bar{a}}}\sum_{\scube\in\pi_{\bar{a}}^{(\alpha)}}\chi_{\bar{a},v_{\bar{a},\alpha}}(\cube),\quad a\in\four.
\eea
Inserting this in \eqref{eq:D6tetrainstch} and taking the index, we have 
\bea\label{eq:D6tetinst_partfunct}
    \mathcal{Z}_{\text{tet.inst.}}^{\D6}[\vec{\underline{v}},\vec{\underline{\pi}}]&=\prod_{a\in\four}\prod_{\alpha=1}^{n_{\bar{a}}}\widetilde{\mathcal{Z}}^{\D6}_{\bar{a}}[\pi_{\bar{a}}^{(\alpha)}]\prod_{a\in\four}\prod_{1\leq\alpha<\beta\leq n_{\bar{a}}}\mathcal{Z}^{\D6\tbar\D6}_{\bar{a}|\bar{a}}(v_{\bar{a},\alpha},\pi_{\bar{a}}^{(\alpha)}\,|\,v_{\bar{a},\beta},\pi_{\bar{a}}^{(\beta)}) \\
    &\qquad \times \prod_{a<b}\prod_{\alpha=1}^{n_{\bar{a}}}\prod_{\beta=1}^{n_{\bar{b}}}\mathcal{Z}^{\D6\tbar\D6}_{\bar{a}|\bar{b}}(v_{\bar{a},\alpha},\pi_{\bar{a}}^{(\alpha)}\,|\,v_{\bar{b},\beta},\pi_{\bar{b}}^{(\beta)}),
\eea
where
\bea
    \widetilde{Z}^{\D6}_{\bar{a}}[\pi]&=\prod_{\scube\in\pi}\frac{1-q_{a}v/\chi_{\bar{a},v}(\cube)}{1-v/\chi_{\bar{a},v}(\cube)}\prod_{\substack{\scube\in\pi\\\scubeF\in\pi}}g_{\bar{a}}\left(\frac{\chi_{\bar{a},v}(\cube)}{\chi_{\bar{a},v}(\cubeF)}\right)^{-1},\\
    \mathcal{Z}_{\bar{a}|\bar{b}}^{\D6\tbar\D6}(v_{1},\pi^{(1)}\,|\,v_{2},\pi^{(2)})&=\prod_{\scube\in\pi^{(1)}}\left(q_{b}\frac{1-q_{b}^{-1}\chi_{\bar{a},v_{1}}(\cube)/v_{2}}{1-\chi_{\bar{a},v_{1}}(\cube)/v_{2}}\right)\prod_{\scubeF\in\pi^{(2)}}\left(\frac{1-q_{a}v_{1}/\chi_{\bar{b},v_{2}}(\cubeF)}{1-v_{1}/\chi_{\bar{b},v_{2}}(\cubeF)}\right) \\
    &\qquad \times\prod_{\substack{\scube\in\pi^{(1)}\\\scubeF\in\pi^{(2)}}}\mathcal{A}_{\mathbb{C}^{4}}\left(\frac{\chi_{\bar{a},v_{1}}(\cube)}{\chi_{\bar{b},v_{2}}(\cubeF)}\right)^{-1}.
\eea
\begin{proposition}
The total instanton partition function of the tetrahedron instanton system is given by summation over plane partitions,
\bea
\mathcal{Z}_{\text{inst.}}^{\D6}=\sum_{\vec{\underline{\pi}}}\mathfrak{q}^{|\vec{\underline{\pi}}|}\mathcal{Z}^{\D6}_{\text{tet.inst.}}[\underline{\vec{v}},\vec{\underline{\pi}}]
\eea 
\end{proposition}
We remark that, for this case, compared to the magnificent four system, we do not need any sign factor as proved by~\cite{Fasola:2023ypx}. Later we will also see this property in the algebraic formalism.   

\paragraph{Nekrasov factors for 7d theory}
Following the 5d Nekrasov factors in \eqref{eq:D4Nekrasovfactor}, we may define the Nekrasov factors for the tetrahedron instanton system as 
\bea
    \mathsf{N}_{\bar{a}}(v_{1},\pi_{1}\,|\,v_{2},\pi_{2})=\frac{\prod\limits_{\scube\in\pi_{2}}\left(1-v_{1}/\chi_{\bar{a},v_{2}}(\cube)\right)}{\prod\limits_{\scube\in\pi_{1}}\left(1-q_{a}^{-1}\chi_{\bar{a},v_{1}}(\cube)/v_{2}\right)}\prod_{\substack{\scube\in\pi_{1}\\\scubeF\in\pi_{2}}}g_{\bar{a}}\left(\frac{\chi_{\bar{a},v_{1}}(\cube)}{\chi_{\bar{a},v_{2}}(\cubeF)}\right).
\eea
It is a part of the following index
\bea
    \mathbb{I}\left[\bfN_{\bar{a},1}^{\vee}\bfK_{\bar{a},2}-q_{a}\bfK_{\bar{a},1}^{\vee}\bfN_{\bar{a},2}-\bfP_{\bar{a}}^{\vee}\bfK_{\bar{a},1}^{\vee}\bfK_{\bar{a},2}\right]
\eea
which represents a different square root from what we used in \eqref{eq:D6tetrainstch}. The square root we use here is symmetric in the sense that we have both $\bfN_{\bar{a},i}\,(i=1,2)$ and $\bfK_{\bar{a},i}\,(i=1,2)$.

\begin{lemma}
We have the following recursion relations for the 7d Nekrasov factor:
\bea\label{eq:D6Nekrasovrecursion}
    \frac{\mathsf{N}_{\bar{a}}(v_{1},\pi_{1}+\cube\,|\,v_{2},\pi_{2})}{\mathsf{N}_{\bar{a}}(v_{1},\pi_{1}\,|\,v_{2},\pi_{2})}&=\mathscr{W}^{\bar{a}\,\vee}_{\pi_{2},v_{2}}(q_{a}^{-1}\chi_{\bar{a},v_{1}}(\cube))^{-1}\\
    \frac{\mathsf{N}_{\bar{a}}(v_{1},\pi_{1}\,|\,v_{2},\pi_{2}+\cubeF)}{\mathsf{N}_{\bar{a}}(v_{1},\pi_{1}\,|\,v_{2},\pi_{2})}&=\mathscr{W}^{\bar{a}}_{\pi_{1},v_{1}}(\chi_{\bar{a},v_{2}}(\cubeF)),
\eea
where 
\bea\label{eq:D6Nekrasov-shell}
    \mathscr{W}_{\pi,v}^{\bar{a}}(x)=(1-v/x)\prod_{\scube\in\pi}g_{\bar{a}}\left(\chi_{\bar{a},v}(\cube)/x\right)\propto \prod_{\cube\in A(\pi)}\left(1-\chi_{\bar{a},v}(\cube)/x\right)\prod_{\scube\in R(\pi)}\left(1-q_{a}^{-1}\chi_{\bar{a},v}(\cube)/x\right)
\eea
and 
\bea
\mathscr{W}^{\bar{a}\,\vee}_{\pi,v}(x)=\left(1-\frac{x}{v}\right)\prod_{\scube\in\pi}g_{\bar{a}}\left(q_{a}\frac{x}{\chi_{\bar{a},v}(\cube)}\right)^{-1}.
\eea
Moreover, using the results in Appendix \ref{app:D6U1partitionfunction}, we have 
\bea
    \frac{\mathsf{N}_{\bar{a}}(v,\pi+\cube\,|\,v,\pi+\cube)}{\mathsf{N}_{\bar{a}}(v,\pi\,|\,v,\pi)}=\left(\frac{q_{a}v}{\chi_{\bar{a},v}(\cube)}\right)\frac{\underset{x=\chi_{\bar{4},v}(\scube)}{\Res}x^{-1}\mathscr{W}^{\bar{4}}_{\pi+\scube,v}(q_{4}^{-1}x)^{-1}}{\underset{x=\chi_{\bar{4},v}(\scube)}{\Res}x^{-1}\mathscr{W}^{\bar{4}}_{\pi,v}(x)^{-1}}.
\eea    
\end{lemma}
We can define the vector $\U(1)$ contribution of the $\D6$ theory as 
\bea
    \mathcal{Z}^{\D6}_{\bar{a}}[\pi]=\frac{1}{\mathsf{N}_{\bar{a}}(v,\pi\,|\,v,\pi)}
\eea
which resembles the partition function of the pure SYM in the 5d theory. The two factors $\mathcal{Z}_{\bar{a}}^{\D6}[\pi]$ and $\widetilde{\mathcal{Z}}_{\bar{a}}^{\D6}[\pi]$ differ by extra Chern--Simons like term and topological term
\bea
    \mathcal{Z}_{\bar{a}}^{\D6}[\pi]=\prod_{\scube\in\pi}\left(-\frac{\chi_{\bar{a},v}(\cube)}{q_{a}v}\right)\widetilde{\mathcal{Z}}_{\bar{a}}^{\D6}[\pi].
\eea
The recursion relation of $\widetilde{\mathcal{Z}}^{\D6}_{\bar{a}}[\pi]$ is then (see Thm.~\ref{eq:app-thm-D6U1recursionformula})
\bea
    \frac{\widetilde{\mathcal{Z}}^{\D6}_{\bar{a}}[\,\pi+\cube\,]}{\widetilde{\mathcal{Z}}^{\D6}_{\bar{a}}[\pi]}=-\frac{\underset{x=\chi_{\bar{4},v}(\scube)}{\Res}x^{-1}\mathscr{W}^{\bar{4}}_{\pi,v}(x)^{-1}}{\underset{x=\chi_{\bar{4},v}(\scube)}{\Res}x^{-1}\mathscr{W}^{\bar{4}}_{\pi+\scube,v}(q_{4}^{-1}x)^{-1}}.
\eea

\paragraph{One-loop perturbative part}
Similar to the magnificent four system, we omit the singular part and choose the following square root part for the perturbative part:
\bea
    \mathring{\mathbf{v}}=\sum_{(b,\beta)>(a,\alpha)}\frac{\bfP_{a}^{\vee}\bfP_{b}}{\bfP_{\four}}v_{\bar{b},\beta}/v_{\bar{a},\alpha}
\eea
where we specified an order in the pair of indices $(\bar{a},\alpha)_{a\in\four}^{\alpha=1,\ldots,n_{\bar{a}}}$. Taking the index gives 
\begin{equation}
\begin{split}
&\mathbb{I}[\mathring{\mathbf{v}}]=\prod_{(b,\beta)>(a,\alpha)}\mathcal{Z}^{\D6\tbar\D6}_{\text{1-loop}}(v_{\bar{a},\alpha},\bar{a}\,|\,v_{\bar{b},\beta},\bar{b})\eqqcolon\mathcal{Z}^{\D6}_{\text{1-loop}},\\
&\mathcal{Z}^{\D6\tbar\D6}_{\text{1-loop}}(x_{1},\bar{a}\,|\,x_{2},\bar{b})=\exp\left(-\sum_{n>0}\frac{1}{n}\frac{\bfP_{a}^{[n]}\bfP_{b}^{[-n]}}{\bfP_{\four}^{[n]}}\left(\frac{x_{1}}{x_{2}}\right)^{n}\right).
\label{eq:D6oneloop}
\end{split}
\end{equation}
For $a=b$, the one-loop factor is written using the multi $q$-shifted factorial or the $q$-deformed triple gamma function. See Appendix \ref{sec:q-functions} and \eqref{eq:D6oneloop-qgamma}.

\subsection{Spiked instanton}\label{sec:spiked_partitionfunct}
The spiked instanton system introduced in \cite{Nekrasov:2016ydq,Nekrasov:2016qym} can be understood similar to the previous two setups. The total observable sheaf $\bfY$ will be a sum of the observable sheaves of $\bfY_{A}\,(A\in\six)$ corresponding to the theory of $\D4$-branes on $\mathbb{C}^{2}_{A}\times \mathbb{S}^{1}$:
\begin{equation}\label{eq:D4Ybundle}
    \bfV=\frac{\bfY^{\vee}\bfY}{\bfP_{\four}},\quad \bfY=\sum_{A\in\six}\bfP_{\bar{A}}\bfY_{A},\quad \bfY_{A}=\bfN_{A}-\bfP_{A}\bfK_{A}\quad A\in\six
\end{equation}
which gives $\bfV=\mathring{\bfV}+\bfV_{\text{inst}}$ where%
\footnote{%
One of the defining conditions of the spiked instanton moduli space $\{\Upsilon_A = 0\}_{A \in \underline{\mathbf{6}}}$ gives rise to the contribution $\sum_{A \in \underline{\mathbf{6}}}q_{A} \mathbf{N}_A^\vee \mathbf{N}_{\bar{A}}$, which is now incorporated in the perturbative part.}
\begin{equation}
\mathring{\bfV}=\sum_{A,B\in\six}\bfN_{A}^{\vee}\frac{\bfP_{\bar{A}}^{\vee}\bfP_{\bar{B}}}{\bfP_{\four}}\bfN_{B},\quad \mathbf{V}_{\text{inst.}}=\sum_{A,B\in\six}\left(-\bfN_{A}^{\vee}\bfP_{\bar{A}}^{\vee}\bfK_{B}-\bfK_{A}^{\vee}\bfP_{\bar{B}}\bfN_{B}+\bfP_{\four}\bfK_{A}^{\vee}\bfK_{B}\right),
\end{equation}
and
\begin{equation}
\bfN_{A}=\sum_{\alpha=1}^{n_{A}}e^{a_{A,\alpha}}=\sum_{\alpha=1}^{n_{A}}v_{A,\alpha},\quad \bfK_{A}=\sum_{I=1}^{k_{A}}e^{\phi_{A,I}}=\sum_{I=1}^{k_{A}}x_{A,I},\quad A\in\six.
\end{equation}
Note again that the total observable sheaf can be rewritten as 
\bea\label{eq:spikedtotalinst}
\bfY=\sum_{A\in\six}\bfP_{\bar{A}}\bfN_{A}-\bfP_{\four}\bfK
\eea
where $\bfK=\sum_{A\in\six}\bfK_{A}$ is the total instanton bundle introduced in section~\ref{sec:physicalsetup}. This is a similar situation as \eqref{eq:tetratotalinst}.

Similar to the previous cases, we actually have to extract the square root part. We choose the following decomposition
\bea\label{eq:D4squareroot}
\mathbf{v}_{\text{inst.}}&=-\sum_{A\in\six}\left(\bfP_{\bar{A}}^{\vee}\bfN_{A}^{\vee}\bfK_{A}-\bfP_{\text{inf}(\bar{A})}^{\vee}\bfP^{\vee}_{A}\bfK_{A}^{\vee}\bfK_{A}\right)\\
&-\sum_{A\in\six}\sum_{B\neq A}\bfN_{A}^{\vee}\bfP_{\bar{A}}^{\vee}\bfK_{B}+\sum_{A<B}\bfP_{\four}\bfK_{A}^{\vee}\bfK_{B}
\eea
where the first line gives the contribution of the affine quiver gauge theory with adjoint mass $q_{\text{inf}(\bar{A})}$ on each $\mathbb{C}^{2}_{A}\times \mathbb{S}^{1}$ and the second lines gives the folded and crossed instantons contributions.

\begin{proposition}
The $\underline{k}$-instanton partition function of the spiked instanton system is given by the following contour integral,
\bea
    \mathcal{Z}^{\D4}_{\underline{k}}=\mathbb{I}[\mathbf{v}_{\text{inst.}}]&=\frac{\underline{\mathcal{G}}^{\underline{k}}}{\underline{k}!}\oint \prod_{A\in\six}\prod_{I=1}^{k_{A}}\frac{dx_{A,I}}{2\pi\iota x_{A,I}}\prod_{A\in\six}\prod_{\alpha=1}^{n_{A}}\prod_{I=1}^{k_{A}}\mathscr{S}_{\bar{A}}\left(\frac{v_{A,\alpha}}{x_{A,I}}\right)
    \prod_{A\in\six}\prod_{I\neq J}g_{\overbar{\scalebox{0.7}{sup$(\bar{A})$}}}\left(\frac{x_{A,I}}{x_{A,J}}\right)^{-1}\\
    &\qquad\times\prod_{A\in\six}\prod_{B\neq A}\prod_{\alpha=1}^{n_{A}}\prod_{I=1}^{k_{B}}\mathscr{S}_{\bar{A}}\left(\frac{v_{A,\alpha}}{x_{B,I}}\right)\prod_{A<B}\prod_{I=1}^{k_{A}}\prod_{J=1}^{k_{B}}\mathcal{A}_{\mathbb{C}^{4}}\left(\frac{x_{A,I}}{x_{B,J}}\right)^{-1}
    \label{eq:D4integral}
\eea
where we define
\bea
\underline{\mathcal{G}}^{\underline{k}}=\prod_{A\in\six}\mathcal{G}_{\text{sup}(\bar{A})}^{k_{A}},\quad \underline{k}!=\prod_{A\in\six}k_{A}!.
\eea    
The total instanton partition function is
\bea
\mathcal{Z}_{\text{inst.}}^{\D4}=\sum_{k=0}^{\infty}\mathfrak{q}^{k}\sum_{\substack{(k_{A})_{A\in\six}\\ \sum_{A}k_{A}=k}}\mathcal{Z}^{\D4}_{\underline{k}}.
\eea
\end{proposition}
The contour integral formula is compatible with the symmetries given in \eqref{eq:instantongaugesymmetry}, \eqref{eq:D4flavorsymmetry}, \eqref{eq:D4ADHMvariablerotation}, \eqref{eq:D4ADHMconstraintsymmetry}.

Note here that the index in \eqref{eq:D4squareroot} is slightly different from the one used in \eqref{eq:affinequiver}. The term we used is $\bfP_{\bar{A}}^{\vee}\bfN_{A}^{\vee}\bfK_{A}$ while a straightforward generalization of \eqref{eq:affinequiver} will be 
\begin{equation}
\bfP_{\text{inf}(\bar{A})}^{\vee}\left(\bfN_{A}^{\vee}\bfK_{A}+q_{A}^{-1}\bfK_{A}^{\vee}\bfN_{A}\right).
\end{equation}
After taking the index, the contour integral formula differs with \eqref{eq:affinecontour}:
\bea
\prod_{A\in\six}\prod_{\alpha=1}^{n_{A}}\prod_{I=1}^{k_{A}}\mathscr{S}_{\bar{A}}\left(\frac{v_{A,\alpha}}{x_{A,I}}\right)=\prod_{A\in\six}q_{\text{inf}(\bar{A})}^{-n_{A}k_{A}}\prod_{A\in\six}\prod_{\alpha=1}^{n_{A}}\prod_{I=1}^{k_{A}}\mathscr{V}_{\text{inf}(\bar{A})}\left(\frac{v_{A,\alpha}}{x_{A,I}}\right)\mathscr{V}_{\text{inf}(\bar{A})}\left(\frac{q_{A}x_{A,I}}{v_{A,\alpha}}\right).
\eea
This overall factor will be eventually related to the topological term for each affine quiver gauge theory. It seems that using the index in \eqref{eq:D4squareroot} is natural from the quantum algebraic viewpoint so we will use this. Note that in the limit $\mathbb{S}^1 \to \text{pt}$, such terms will disappear and have no effect.

The poles of the contour integral~\eqref{eq:D4integral} are classified by 2d partitions:
\bea
    &\vec{\underline{v}}=(v_{A,\alpha})^{\alpha=1,\ldots, n_{A}}_{A\in\six},\quad \vec{\underline{\lambda}}=(\vec{\lambda}_{A})_{A\in\six}=(\lambda_{A}^{(\alpha)})^{\alpha=1,\ldots,n_{A}}_{A\in\six},\quad |\vec{\underline{\lambda}}|=\sum_{A\in\six}\sum_{\alpha=1}^{n_{A}}|\lambda_{A}^{(\alpha)}|,\\
    &\{x_{A,I}\}_{A\in\six}^{I=1,\ldots,k_{A}}\longrightarrow \{\chi_{A,v_{A,\alpha}}(\Bbox)\}^{\alpha=1,\ldots,n_{A}}_{A\in\six,\,\Abox\in\lambda_{A,\alpha}},\quad \chi_{ab,v}(\Bbox)=vq_{a}^{i-1}q_{b}^{j-1}.
\eea
At the fixed points, the character of $\bfK_{A}$ will be 
\bea
\left.\bfK_{A}\right|_{\vec{\lambda}_{A}}=\sum_{\alpha=1}^{n_{A}}\sum_{\Abox\in\lambda_{A}^{(\alpha)}}\chi_{A,v_{A,\alpha}}(\Bbox),\quad A\in\six.
\eea
Inserting this and taking the index, we have
\bea\label{eq:D4spikedpartition1}
    \mathcal{Z}^{\D4}_{\text{spk.inst.}}[\underline{\vec{v}},\vec{\underline{\lambda}}]&=\prod_{A\in\six}\prod_{\alpha=1}^{n_{A}}\widetilde{\mathcal{Z}}^{\D4}_{A}[\lambda_{A}^{(\alpha)}]\prod_{A\in\six}\prod_{\alpha<\beta}\mathcal{Z}^{\D4\tbar\D4}_{A|A}(v_{A,\alpha},\lambda_{A}^{(\alpha)}\,|\,v_{A,\beta},\lambda_{A}^{(\beta)}) \\
    &\qquad \times \prod_{A<B}\prod_{\alpha=1}^{n_{A}}\prod_{\beta=1}^{n_{B}}\mathcal{Z}_{A|B}^{\D4\tbar\D4}(v_{A,\alpha},\lambda_{A}^{(\alpha)}\,|\,v_{B,\beta},\lambda_{B}^{(\beta)}),
\eea
where
\bea\label{eq:D4spikedpartition2}
    \widetilde{\mathcal{Z}}^{\D4}_{A}[\lambda]&=q_{\text{inf}(\bar{A})}^{-|\lambda|}\mathcal{Z}^{\D4}_{A}[\lambda;q_{\text{inf}(\bar{A})}]=q_{\text{inf}(\bar{A})}^{-|\lambda|}\frac{\mathsf{N}_{A}(q_{\text{inf}(\bar{A})}v,\lambda\,|\,v,\lambda)}{\mathsf{N}_{A}(v,\lambda\,|\,v,\lambda)},\\
    \mathcal{Z}_{A|B}^{\D4\tbar\D4}(v_{1},\lambda^{(1)}\,|\,v_{2},\lambda^{(2)})&=\prod_{\Abox\in\lambda^{(1)}}\mathscr{S}_{\bar{B}}\left(q_{B}\frac{\chi_{A,v_{1}}(\Bbox)}{v_{2}}\right)\prod_{\AboxF\in\lambda^{(2)}}\mathscr{S}_{\bar{A}}\left(\frac{v_{1}}{\chi_{B,v_{2}}(\BboxF)}\right)\prod_{\substack{\Abox\in\lambda^{(1)}\\\AboxF\in\lambda^{(2)}}}\mathcal{A}_{\mathbb{C}^{4}}\left(\frac{\chi_{A,v_{1}}(\Bbox)}{\chi_{B,v_{2}}(\BboxF)}\right)^{-1}.
\eea
The property and low levels of the $\U(1)$ contribution are given in Appendix \ref{app:D4U1partitionfunction}. 

\begin{proposition}
The total instanton partition function of the spiked instanton system is given by summation over 2d partitions,
\bea
\mathcal{Z}_{\text{inst.}}^{\D4}=\sum_{\vec{\underline{\lambda}}}\mathfrak{q}^{|\vec{\underline{\lambda}}|}\mathcal{Z}^{\D4}_{\text{spk.inst.}}[\vec{\underline{v}},\vec{\underline{\lambda}}].
\eea
\end{proposition}
\paragraph{One-loop perturbative part}
Omitting the singular part, we use the following square root part for the one-loop perturbative part:
\bea
    \mathring{\mathbf{v}}=\sum_{(B,\beta)>(A,\alpha)}\frac{\bfP_{\bar{A}}^{\vee}\bfP_{\bar{B}}}{\bfP_{\four}}v_{B,\beta}/v_{A,\alpha},
\eea
where we specified an order in $(A,\alpha)^{\alpha=1\ldots,n_{A}}_{A\in\six}$. Taking the index, we have 
\bea\label{eq:D4oneloop}
&\mathbb{I}[\mathring{\mathbf{v}}]=\prod_{(B,\beta)>(A,\alpha)}\mathcal{Z}^{\D4\tbar\D4}_{\text{1-loop}}(v_{A,\alpha},A\,|\,v_{B,\beta},B)\eqqcolon\mathcal{Z}^{\D4}_{\text{1-loop}},\\
&\mathcal{Z}_{\text{1-loop}}^{\D4\tbar\D4}(x_{1},A\,|\,x_{2},B)=\exp\left(-\sum_{n=1}^{\infty}\frac{1}{n}\frac{\bfP_{\bar{A}}^{[n]}\bfP_{\bar{B}}^{[-n]}}{\bfP_{\four}^{[n]}}\left(\frac{x_{1}}{x_{2}}\right)^{n}\right).
\eea
For $a=b$, the one-loop factor is written using the multi $q$-shifted factorial or the $q$-deformed double gamma function. See Appendix \ref{sec:q-functions} and \eqref{eq:D4oneloop-qgamma}.

\paragraph{Supergroup analogue}
Actually, the contour integral formulas and instanton partition functions for the 5d gauge theories have a supergroup generalization introduced in \cite{Kimura:2019msw} (see also \cite{Vafa:2001qf,Okuda:2006fb,Dijkgraaf:2016lym,Kimura:2023iup,Kimura:2023ndz} for related works). The explicit contour integral formula for the 5d $\mathcal{N}=1^{\ast}$ $\U(n_{+}|n_{-})$ gauge theory is given in Appendix~\ref{app:supergroup} and \eqref{eq:app-supergroupLMNS}. The strategy to derive it is to change all of the characters of $\bfN,\bfK$ to supercharacters:
\bea
\operatorname{sch}\bfN=\operatorname{ch}\bfN^{+}-\operatorname{ch}\bfN^{-},\quad \operatorname{sch}\bfK=\operatorname{ch}\bfK^{+}-\operatorname{ch}\bfK^{-}.
\eea
This generalization gives extra vector-like and bifundamental-like contributions which enables us to understand the supergroup gauge theory as a type of quiver gauge theory (see \cite{Dijkgraaf:2016lym,Kimura:2019msw,Kimura:2023ndz}). To obtain the supergroup gauge theory using D-branes, one needs to include \emph{ghost}/\emph{negative} $\D4$-branes, denoted as $\D4^{-}$, to the system additional to the normal\footnote{We also call it positive D-branes.} $\D4$-branes, which we denote $\D4^{+}$. For such cases, $\D0^{+}$-branes play the role of positive instantons while $\D0^{-}$-branes play the role of negative instantons (see \cite{Kimura:2023iup}). 

After changing all the characters to supercharacters, one can obtain the supergroup analog of the gauge origami of the spiked instanton system. For the simplest case, where there are only D-branes on $\mathbb{C}^{2}_{12}\times \mathbb{S}^{1}$, see \eqref{eq:app-supergroupLMNS} for the contour integral formula. See for example \cite{Kimura:2019msw} or \cite[section 2]{Noshita:2022dxv} for the explicit formulas for general quiver gauge theories. The generalization to the gauge origami system is straightforward so we omit it.

\subsection{Coupled vortex system}\label{sec:cplvortex_partitionfunct}
Following the previous sections, it is natural to consider a theory where $\D2$-branes intersect. We propose a $\D2$-analogue partition function of the gauge origami system. The total index is written similarly with $\bfY$, but this time it is a sum of $\bfY_{a}\,(a\in\four)$:
\bea\label{eq:D2Ybundle}
    \bfV&=\frac{\bfY^{\vee}\bfY}{\bfP_{\four}}\quad \bfY=\sum_{a\in\four}\bfP_{\bar{a}}\bfY_{a},\quad \bfY_{a}=\bfN_{a}-\bfP_{a}\bfK_{a},
\eea
which gives
\bea
\bfV&=\mathring{\bfV}+\bfV_{\text{inst}},\quad \mathring{\mathbf{V}}=\sum_{a,b\in\four}\bfN_{a}^{\vee}\frac{\bfP_{\bar{a}}^{\vee}\bfP_{\bar{b}}}{\bfP_{\four}}\bfN_{b}\\
    \bfV_{\text{inst.}}=&\sum_{a,b\in\four}(-\bfN_{a}^{\vee}\bfP_{\bar{a}}^{\vee}\bfK_{b}-\bfK_{a}^{\vee}\bfP_{\bar{b}}\bfN_{b}+\bfP_{\four}\bfK_{a}^{\vee}\bfK_{b})
\eea
where 
\bea  \bfN_{a}=\sum_{\alpha=1}^{n_{a}}e^{a_{a,\alpha}}=\sum_{\alpha=1}^{n_{a}}v_{a,\alpha},\quad \bfK_{a}=\sum_{I=1}^{k_{a}}e^{\phi_{a,I}}=\sum_{I=1}^{k_{a}}x_{a,I}.
\eea
Similar to the tetrahedron instanton \eqref{eq:tetratotalinst} and spiked instanton \eqref{eq:spikedtotalinst} cases, the total observable sheaf can be rewritten using the total instanton bundle
\bea
\bfY=\sum_{a\in\four}\bfP_{\bar{a}}\bfN_{a}-\bfP_{\four}\bfK
\eea
where $\bfK=\sum_{a\in\four}\bfK_{a}$.

Mimicking the other cases, we decompose the instanton part as $\bfV_{\text{inst}}=\mathbf{v}_{\text{inst}}+\mathbf{v}_{\text{inst}}^{\vee}$ and choose the square root part as 
\bea
    \mathbf{v}_{\text{inst}}&=-\sum_{a\in\four}\bfP^{\vee}_{\text{i}(a)\text{j}(a)}\left(\bfN_{a}^{\vee}\bfK_{a}-q_{a}^{-1}\bfN_{a}\bfK_{a}^{\vee}-\bfP_{a}^{\vee}\bfK_{a}^{\vee}\bfK_{a}\right)\\
    &-\sum_{a\in\four}\sum_{ b\neq a}\bfN_{a}^{\vee}\bfP_{\bar{a}}^{\vee}\bfK_{b}+\sum_{a<b}\bfP_{\four}\bfK_{a}^{\vee}\bfK_{b}.
\eea
Here, we denote the three indices in $\four\setminus\{a\}$ as $\text{i}(a),\text{j}(a),\text{k}(a)$ and chose two of them for each $a\in\four$ to define the index. The first line gives the 3d theory on each plane $\mathbb{C}_{a}\times \mathbb{S}^{1}$ and the other terms give the bifundamental interactions between the 3d theories.

\begin{proposition}
The total instanton partition function of the coupled vortex system is given by
\bea
\mathcal{Z}_{\text{inst.}}^{\D2}=\sum_{k=0}^{\infty}\mathfrak{q}^{k}\sum_{\sum_{a}k_{a}=k}\mathcal{Z}_{\underline{k}}^{\D2},
\eea    
where each contribution is given by the following contour integral,
\bea
\mathcal{Z}_{\underline{k}}^{\D2}=\mathbb{I}[\mathbf{v}_{\text{inst.}}]&=\frac{\underline{\mathcal{G}}^{\underline{k}}}{\underline{k}!}\oint\prod_{a\in\four}\prod_{I=1}^{k_{a}}\frac{dx_{a,I}}{2\pi\iota x_{a,I}}\prod_{a,b\in\four}\prod_{\alpha=1}^{n_{a}}\prod_{I=1}^{k_{b}}g_{\bar{a}}\left(\frac{v_{a,\alpha}}{x_{b,I}}\right)
 \prod_{I\neq J}g_{a\,\text{i}(a)\text{j}(a)}\left(\frac{x_{a,I}}{x_{a,J}}\right)^{-1},\\
&\qquad \times \prod_{a<b}\prod_{I=1}^{k_{a}}\prod_{J=1}^{k_{b}}\mathcal{A}_{\mathbb{C}^{4}}\left(\frac{x_{a,I}}{x_{b,J}}\right)^{-1}
\label{eq:D2integral}
\eea
with
\bea
\underline{\mathcal{G}}^{\underline{k}}=\prod_{a\in\four}\mathcal{G}_{a\,\text{i}(a)\text{j}(a)}^{k_{a}}\,,\quad \underline{k}!=\prod_{a\in\four}k_{a}.
\eea
\end{proposition}

For simplicity, let us consider the case of a $\U(n)$ theory on $\mathbb{C}_{4}\times \mathbb{S}^{1}$. We choose $a=4$, $\text{i}(4)=q_{2}$, $\text{j}(4)=q_{3}$ to be the masses and then we obtain
\bea
\mathcal{Z}_{k}^{\D2}&=\frac{\mathcal{G}_{\bar{1}}^{k}}{k!}\oint\prod_{I=1}^{k}\frac{dx_{I}}{2\pi\iota x_{I}}\prod_{\alpha=1}^{n}\prod_{I=1}^{k}g_{\bar{4}}\left(\frac{v_{\alpha}}{x_{I}}\right)\prod_{I\neq J}g_{\bar{1}}\left(\frac{x_{I}}{x_{J}}\right)^{-1},\\
&=\frac{\mathcal{G}_{\bar{1}}^{k}}{k!}\oint\prod_{I=1}^{k}\frac{dx_{I}}{2\pi\iota x_{I}}\prod_{\alpha=1}^{n}\prod_{I=1}^{k}\frac{(1-q_{4}^{-1}v_{\alpha}/x_{I})\prod_{i=1}^{3}(1-q_{i}v_{\alpha}/x_{I})}{(1-v_{\alpha}/x_{I})\prod_{i=1}^{3}(1-q_{4}^{-1}q_{i}^{-1}v_{\alpha}/x_{I})}\\
&\qquad \times \prod_{I\neq J}\frac{(1-x_{I}/x_{J})\prod_{i=2}^{4}(1-q_{i}^{-1}q_{1}^{-1}x_{I}/x_{J})}{(1-q_{1}^{-1}x_{I}/x_{J})\prod_{i=2}^{4}(1-q_{i}x_{I}/x_{J})}
\eea
The pole structure implies that the fixed points should be classified by the one-dimensional partitions in section \ref{sec:multi-dim-part}. Looking at the $g_{\bar{4}}(x)$ part, the pole at $x_{I}=v_{\alpha}$ starts the growth of the one-dimensional partition, while the zero at $x_{I}=q_{4}^{-1}v_{\alpha}$ terminates the growth in the opposite direction. The $g_{\bar{1}}(x)^{-1}$ part is decomposed into 
\bea
\frac{(1-x_{I}/x_{J})}{(1-q_{4}x_{I}/x_{J})}\times \frac{(1-q_{4}q_{2}x_{I}/x_{J})(1-q_{4}q_{3}x_{I}/x_{J})(1-q_{1}^{-1}q_{4}^{-1}x_{I}/x_{J})}{(1-q_{2}x_{I}/x_{J})(1-q_{3}x_{I}/x_{J})(1-q_{1}^{-1}x_{I}/x_{J})}
\eea
where the first term gives the contribution of the $\mathcal{N}=2$ vector multiplet and the second term gives the contributions of the three $\mathcal{N}=2$ adjoint chiral multiples with adjoint masses $q_{1},q_{2},q_{3}$ (in the multiplicative notation).

Generally, we assume that the fixed points are classified by one-dimensional partitions labeled by $\mathbb{Z}_{\geq 0}$ as the vortex partition function \cite{Shadchin:2006yz,Nekrasov:2009JJM,Yoshida:2011au,Fujitsuka:2013fga,Yoshida:2014ssa}:
\bea
    &\vec{\underline{v}}=(v_{a,\alpha})_{a\in\four}^{\alpha=1,\ldots,n_{a}},\quad \underline{\vec{k}}=(\vec{k}_{a})_{a\in\four}=(k_{a}^{(\alpha)})_{a\in\four}^{\alpha=1,\ldots,n_{a}},\quad |\vec{\underline{k}}|=\sum_{a\in\four}\sum_{\alpha=1}^{n_{a}}k_{a}^{(\alpha)},\\
    &\{x_{a,I}\}_{a\in\four}^{I=1,\ldots,k_{a}}\longrightarrow \{\chi_{a,v_{a,\alpha}}(\Bbox)\}_{a\in\four}^{\alpha=1,\ldots,n_{a}},\quad \chi_{a,v_{a,\alpha}}(\Bbox)=v_{a,\alpha}q_{a}^{i-1}\,\,(i=1,\ldots,k_{a}^{(\alpha)}).
\eea
The character of $\bfK_{a}$ at the fixed points will then be 
\bea
\left.\bfK_{a}\right|_{\vec{k}_{a}}=\sum_{\alpha=1}^{n_{a}}\sum_{\Abox\in k_{a}^{(\alpha)}}\chi_{a,v_{a,\alpha}}(\Bbox)=\sum_{\alpha=1}^{n_{a}}v_{a,\alpha}\frac{1-q_{a}^{k_{a}^{(\alpha)}}}{1-q_{a}}.
\eea
To evaluate the partition function, we need to insert this and take the index. Before doing that let us introduce the $\D2$-version of the Nekrasov factors.

\paragraph{Nekrasov factors for 3d theory}We define the $\D2$-version of the Nekrasov factors as
\bea
    \mathsf{N}_{a}(v_{1},k_{1}|v_{2},k_{2})&=\frac{\prod\limits_{\AboxF\in k_{2}}\left(1-v_{1}/\chi_{a,v_{2}}(\BboxF)\right)}{\prod\limits_{\Abox\in k_{1}}\left(1-q_{a}\chi_{a,v_{1}}(\Bbox)/v_{2}\right)}\prod_{\substack{\Abox\in k_{1}\\\AboxF\in k_{2}}}\mathscr{V}_{a}\left(\frac{\chi_{a,v_{1}}(\Bbox)}{\chi_{a,v_{2}}(\BboxF)}\right).
\eea

\begin{lemma}
The following recursion formulas hold for the 3d Nekrasov factors,
\bea
\frac{\mathsf{N}_{a}(v_{1},k_{1}+\Bbox\,|\,v_{2},k_{2})}{\mathsf{N}_{a}(v_{1},k_{1}\,|\,v_{2},k_{2})}&=\mathscr{U}^{a\vee}_{k_{2},v_{2}}\left(q_{a}\chi_{a,v_{1}}(\Bbox)\right)^{-1},\quad
\frac{\mathsf{N}_{a}(v_{1},k_{1}\,|\,v_{2},k_{2}+\BboxF)}{\mathsf{N}_{a}(v_{1},k_{1}\,|\,v_{2},k_{2})}=\mathscr{U}^{a}_{k_{1},v_{1}}\left(\chi_{a,v_{2}}(\BboxF)\right)
\eea
with 
\bea\label{eq:D2Ufunctiondef}
    \mathscr{U}^{a}_{k,v}(x)&=\left(1-\frac{v}{x}\right)\prod_{\Abox\in k}\mathscr{V}_{a}\left(\frac{\chi_{a,v}(\Bbox)}{x}\right)=\left(1-\frac{vq_{a}^{k}}{x}\right),\\
    \mathscr{U}^{a\vee}_{k,v}(x)&=\left(1-\frac{x}{v}\right)\prod_{\Abox\in k}\mathscr{V}_{a}\left(q_{a}^{-1}\frac{x}{\chi_{a,v}(\Bbox)}\right)^{-1}=\left(1-\frac{x}{vq_{a}^{k}}\right)
\eea
where note that $k+\Bbox =k+1$.    
\end{lemma}
We note that the $\U(1)$ contributions coming from $\mathsf{N}_{a}(v,k\,|\,v,k)$ is trivial (see Appendix~\ref{app:D2U1partitionfunction}): \beq\mathsf{N}_{a}(v,k\,|\,v,k)=1.\eeq

\begin{proposition}
The coupled vortex partition function is given as follows,
\bea
\mathcal{Z}_{\text{vort.}}^{\D2}=\sum_{\vec{\underline{k}}}\mathfrak{q}^{|\vec{\underline{k}}|}\mathcal{Z}^{\D2}_{\text{cpl.vort.}}[\vec{\underline{v}},\vec{\underline{k}}],
\eea
where
\bea \label{eq:D2cplvortex}
    \mathcal{Z}^{\D2}_{\text{cpl.vort.}}[\vec{\underline{v}},\vec{\underline{k}}]&=\prod_{a\in\four}\prod_{\alpha=1}^{n_{a}}\mathcal{Z}_{a}^{\D2}[k_{a}^{(\alpha)};q_{\text{i}(a)},q_{\text{j}(a)}]\prod_{a\in\four}\prod_{\alpha<\beta}\mathcal{Z}_{a\,|\,a}^{\D2\tbar\D2}(v_{a,\alpha},k_{a}^{(\alpha)}\,|\,v_{a,\beta},k_{a}^{(\beta)})\\
    &\qquad \times \prod_{a<b}\prod_{\alpha=1}^{n_{a}}\prod_{\beta=1}^{n_{b}}\mathcal{Z}_{a\,|\,b}^{\D2\tbar\D2}(v_{a,\alpha},k_{a}^{(\alpha)}\,|\,v_{b,\beta},k_{b}^{(\beta)}),\\
    \mathcal{Z}_{a}^{\D2}[k\,;q_{i},q_{j}]&=\frac{\mathsf{N}_{a}(q_{i}v,k\,|\,v,k)\mathsf{N}_{a}(q_{j}v,k\,|\,v,k)}{\mathsf{N}_{a}(v,k\,|\,v,k)\mathsf{N}_{a}(q_{i}q_{j}v,k\,|\,v,k)},\quad a\in\four,\,\, i,j\in\four\setminus \{a\},\\
\mathcal{Z}^{\D2\tbar\D2}_{a|b}\left(v_{1},k_{1}\,|\,v_{2},k_{2}\right)&=\prod_{\AboxF\in k_{2}}g_{\bar{a}}\left(\frac{v_{1}}{\chi_{b,v_{2}}(\BboxF)}\right)\prod_{\Abox\in k_{1}}g_{\bar{b}}\left(\frac{q_{b}\chi_{a,v_{1}}(\Bbox)}{v_{2}}\right)^{-1}\prod_{\substack{\Abox\in k_{1}\\\AboxF\in k_{2}}}\mathcal{A}_{\mathbb{C}^{4}}\left(\frac{\chi_{a,v_{1}}(\Bbox)}{\chi_{b,v_{2}}(\BboxF)}\right)^{-1}.
\eea
\end{proposition}

\paragraph{One-loop perturbative part}
We choose the following square root part of the perturbative part $\mathring{\mathbf{V}}$ as 
\bea
    \mathring{\mathbf{v}}=\sum_{(b,\beta)>(a,\alpha)}\frac{\bfP_{\bar{a}}^{\vee}\bfP_{\bar{b}}}{\bfP_{\four}}v_{b,\beta}/v_{a,\alpha}
\eea
which gives 

\bea\label{eq:D2oneloop}
    &\mathbb{I}[\mathring{\mathbf{v}}]=\prod_{(b,\beta)>(a,\alpha)}\mathcal{Z}_{\text{1-loop}}^{\D2\tbar\D2}(v_{a,\alpha},a\,|\,v_{b,\beta},b)\eqqcolon\mathcal{Z}^{\D2}_{\text{1-loop}},\\
    &\mathcal{Z}_{\text{1-loop}}^{\D2\tbar\D2}(x,a\,|\,x',b)=\exp\left(-\sum_{n=1}^{\infty}\frac{1}{n}\frac{\bfP_{\bar{a}}^{[n]}\bfP_{\bar{b}}^{[-n]}}{\bfP_{\four}^{[n]}}\left(\frac{x}{x'}\right)^{n}\right).
\eea
When $a=b$, the one-loop factor is written using the $q$-shifted factorial or $q$-deformed gamma functions. See Appendix \ref{sec:q-functions} and \eqref{eq:D2oneloop-qgamma} for the explicit formulas.

\subsection{Decomposition of partition functions}\label{sec:decomp_partition}
 The partition functions of the 3d, 5d, 7d, and 9d theories explained in the previous section can be written in different forms. For each setup, the fixed points are labeled by multi-dimensional partitions introduced in section \ref{sec:multi-dim-part}. Each multi-dimensional partition has different descriptions (see section \ref{sec:multi-dim-part}), which eventually leads to different decompositions of the partition functions. A similar discussion decomposing the partition functions of 5d theories into 3d theories was done in \cite{Nieri:2017ntx}. 
\subsubsection{3d theory}\label{sec:decomp3d}
Let us focus on the case when we only have $n$-stacks of $\D2$-branes on $\mathbb{C}_{a}\times \mathbb{S}^{1}$ leading to a $\U(n)$ theory. After localization, we obtain
\bea
   &\bfY_{a}=\bfN_{a}-\bfP_{a}\bfK_{a}=\sum_{\alpha=1}^{n_{a}}v_{a,\alpha}-(1-q_{a})\sum_{\alpha=1}^{n_{a}}v_{a,\alpha}\frac{1-q_{a}^{k_{a}^{(\alpha)}}}{1-q_{a}}=\sum_{\alpha=1}^{n_{a}}v_{a,\alpha}q_{a}^{k_{a}^{(\alpha)}}\eqqcolon\bfX_{a},\\
    &\bfX_{a}=\sum_{x\in\mathcal{X}_{a}}x,\quad \mathcal{X}_{a}=\{v_{a,\alpha}q_{a}^{k_{a}^{(\alpha)}}\,|\,\alpha=1,\ldots,n_{a}\}
\eea
The total index of the coupled vortex system will be rewritten as
\bea
    \bfV=\frac{\bfY^{\vee}\bfY}{\bfP_{\four}}=\sum_{a,b\in\four}\frac{\bfP_{\bar{a}}^{\vee}\bfP_{\bar{b}}}{\bfP_{\four}}\bfX_{a}^{\vee}\bfX_{b}.
\eea
Choosing a square root part and taking the index gives the representation using the $\bfX_{a}$-variables. Note that the partition function obtained here will be the same as the partition function in \eqref{eq:D2cplvortex} up to perturbative factors and topological terms.

\subsubsection{5d theory}\label{sec:decomp5d}
The two-dimensional partition (Young diagram) has two $(1,1)$-type descriptions depending on which axis we project the Young diagram. We can rewrite the character of $\bfY_{ab}\,(a<b)$ as
\bea
    &\bfY_{ab}=\bfN_{ab}-\bfP_{a}\bfP_{b}\bfK_{ab}=\bfP_{a}\bfX_{ab}=\bfP_{b}\check{\bfX}_{ab},\\
    &\bfX_{ab}=\sum_{x\in\mathcal{X}_{ab}}x,\quad \mathcal{X}_{ab}=\left\{v_{ab,\alpha}q_{a}^{i-1}\left.q_{b}^{\lambda^{(\alpha)}_{ab,i}}\,\right|\,\substack{\alpha=1,\ldots,n_{ab}\\i=1,\ldots,\infty}\right\},\\
    &\check{\bfX}_{ab}=\sum_{x\in\check{\mathcal{X}}_{ab}}x,\quad \check{\mathcal{X}}_{ab}=\left\{ v_{ab,\alpha}q_{b}^{j-1}\left.q_{a}^{\lambda^{(\alpha)\rmT}_{ab,j}}\,\right|\,\substack{\alpha=1,\ldots,n_{ab}\\j=1,\ldots,\infty} \right\}.
\eea
This comes from the following identities
\bea
    v-(1-q_{a})(1-q_{b})\sum_{i=1}^{\infty}\sum_{j=1}^{\lambda_{i}}vq_{a}^{i-1}q_{b}^{j-1}=(1-q_{a})\sum_{i=1}^{\infty}vq_{a}^{i-1}q_{b}^{\lambda_{i}},\quad |q_{a}|<1,\\
v-(1-q_{a})(1-q_{b})\sum_{j=1}^{\infty}\sum_{i=1}^{\lambda^{\rmT}_{j}}vq_{a}^{i-1}q_{b}^{j-1}=(1-q_{b})\sum_{j=1}^{\infty}vq_{b}^{j-1}q_{a}^{\lambda^{\rmT}_{j}},\quad |q_{b}|<1.
\eea
Note that depending on the analytic region of the $q$-parameters, the decomposition will differ. Since we have the condition $q_{1}q_{2}q_{3}q_{4}=1$, we can not keep all of the parameters as $|q_{a}|<1$, but at least one of the parameters should be $|q_{a}|>1$. For example, if we choose an analytic region\footnote{In this paper, we assume in most cases that only one of the $q$-parameters will be $|q_{a}|>1$.} $|q_{1,2,3}|<1,\,|q_{4}|>1$, then we can rewrite $\bfY_{A}\,(A\in\six)$ as
\bea
    \bfY_{i4}=\bfP_{i}\bfX_{i4},\quad \bfY_{ij}=\bfP_{i}\bfX_{ij}=\bfP_{j}\check{\bfX}_{ij},\quad i,j=1,2,3,
\eea
where for $\bfY_{ij}$, we can either use the normal Young diagram or the transpose of it.

For $A\in\six$, assume we have the decomposition
\bea
    \bfY_{A}=\bfP_{s(A)}\bfX_{A},
\eea
where $s(A)$ is one of the indices in $A$ and $\bfX_{A}$ is the corresponding $\bfX$-bundle after the decomposition, then the total index of the spiked instanton is
\bea
    \bfV=\frac{\bfY^{\vee}\bfY}{\bfP_{\four}}=\sum_{A,B\in\six}\frac{(\bfP_{\bar{A}}\bfP_{s(A)})^{\vee}(\bfP_{\bar{B}}\bfP_{s(B)})}{\bfP_{\four}}\bfX_{A}^{\vee}\bfX_{B}=\sum_{A,B\in\six}\frac{\bfP_{\four}}{\bfP_{\bar{s}(A)}^{\vee}\bfP_{\bar{s}(B)}}\bfX_{A}^{\vee}\bfX_{B}\label{eq:spiked_D4toD2}
\eea
where $\bar{s}(A)$ is the other index in $A$, namely $A=s(A)\bar{s}(A)$. Taking the square root part and the index gives the $(1,1)$-type description of the spiked instanton partition function.

\subsubsection{7d theory}\label{sec:decomp7d}
For the tetrahedron instanton system, since the fixed points are classified by plane partitions, we have two types of descriptions of the partition functions: $(2,1)$ and $(1,2)$-type.

\paragraph{$(2,1)$-type}
We can rewrite the character of $\bfY_{abc}\,(a<b<c)$ as 
\bea &\bfY_{abc}=\bfP_{a}\bfP_{b}\bfX_{abc}=\bfP_{a}\bfP_{c}\check{\bfX}_{abc}=\bfP_{b}\bfP_{c}\check{\check{\bfX}}_{abc},\\
&\bfX_{abc}=\sum_{x\in\mathcal{X}_{abc}}x,\quad \check{\bfX}_{abc}=\sum_{x\in\check{\mathcal{X}}_{abc}}x,\quad \check{\check{\bfX}}_{abc}=\sum_{x\in\check{\check{\mathcal{X}}}_{abc}}x,\\
&\mathcal{X}_{abc}=\{v_{abc,\alpha}q_{a}^{i-1}q_{b}^{j-1}q_{c}^{\pi^{(\alpha)}_{abc,ij}}\,|\,\substack{\alpha=1,\ldots,n_{abc}\\
    i,j=1,\ldots,\infty}\},\quad \check{\mathcal{X}}_{abc}=\{v_{abc,\alpha}q_{a}^{i-1}q_{c}^{k-1}q_{b}^{\check{\pi}^{(\alpha)}_{abc,ik}}\,|\,\substack{\alpha=1,\ldots,n_{abc}\\
    i,k=1,\ldots,\infty}\},\\
&\check{\check{\mathcal{X}}}_{abc}=\{v_{abc,\alpha}q_{b}^{j-1}q_{c}^{k-1}q_{a}^{\check{\check{\pi}}^{(\alpha)}_{abc,jk}}\,|\,\substack{\alpha=1,\ldots,n_{abc}\\
    j,k=1,\ldots,\infty}\},
\eea
where $\pi,\check{\pi},\check{\check{\pi}}$ are the three possible $(2,1)$-type descriptions depending on which 2d plane the plane partition is projected. The above decomposition comes from the following identity
\bea
    v-(1-q_{a})(1-q_{b})(1-q_{c})\sum_{i=1}^{\infty}\sum_{j=1}^{\infty}\sum_{k=1}^{\pi_{i,j}}vq_{a}^{i-1}q_{b}^{j-1}q_{c}^{k-1}=(1-q_{a})(1-q_{b})\sum_{i,j=1}^{\infty}vq_{a}^{i-1}q_{b}^{j-1}q_{c}^{\pi_{i,j}}
\eea
where we assumed $|q_{a}|,|q_{b}|<1$. Under the condition that only one of the $q$-parameters obey $|q_{a}|>1\,(\exists!\,a)$, we can always do this $(2,1)$-description. 

For $a\in\four$, let $\text{i}(a)$, $\text{j}(a)$, $\text{k}(a)$ be the three indices of $\bar{a}$ and assume we have the following decomposition
\bea
    \bfY_{\bar{a}}=\bfP_{\text{i}(a)}\bfP_{\text{j}(a)}\bfX_{\bar{a}}
\eea
where we omit the check mark on $\bfX_{\bar{a}}$ for simplicity. Then, the total index of the tetrahedron instanton system will be 
\bea
    \bfV=\frac{\bfY^{\vee}\bfY}{\bfP_{\four}}=\sum_{a,b\in\four}\frac{(\bfP_{a}\bfP_{\text{i}(a)}\bfP_{\text{j}(a)}\bfX_{\bar{a}})^{\vee}(\bfP_{b}\bfP_{\text{i}(b)}\bfP_{\text{j}(b)}\bfX_{\bar{b}})}{\bfP_{\four}}=\sum_{a,b\in\four}\frac{\bfP_{\four}}{\bfP_{\text{k}(a)}^{\vee}\bfP_{\text{k}(b)}}\bfX_{\bar{a}}^{\vee}\bfX_{\bar{b}}.\label{eq:7dscreeningdecomp}
\eea

\paragraph{$(1,2)$-type} Under this description, the plane partition is decomposed into multiple Young diagrams extending in one of the three directions. Let us focus on a $\U(1)$ theory on $\mathbb{C}^{3}_{123}\times \mathbb{S}^{1}$ for simplicity. For $k\in\mathbb{Z}_{>0}$ we have a Young diagram $\lambda^{(k)}$ and the instanton bundle is written as 
\bea
\bfK_{123}=\sum_{k=1}^{\infty}\sum_{\Abox\in\lambda^{(k)}}\chi_{12,vq_{3}^{k-1}}(\Bbox),\quad |q_{3}|<1
\eea
Then, we obtain
\bea
    \bfY_{123}&=v-\bfP_{123}\sum_{k=1}^{\infty}\sum_{\Abox\in\lambda^{(k)}}\chi_{12,vq_{3}^{k-1}}(\Bbox)=\bfP_{3}\left(\frac{v}{1-q_{3}}-\bfP_{12}\sum_{k=1}^{\infty}\sum_{\Abox\in\lambda^{(k)}}\chi_{12,vq_{3}^{k-1}}(\Bbox)\right)\\
    &=\bfP_{3}\left(\sum_{k=1}^{\infty}vq_{3}^{k-1}-\bfP_{12}\sum_{\Abox\in\lambda^{(k)}}\chi_{12,vq_{3}^{k-1}}(\Bbox)\right)=\bfP_{3}\sum_{k=1}^{\infty}\bfY_{12}[vq_{3}^{k-1},\lambda^{(k)}],\label{eq:7dYoungdiagramdecomp}
\eea
where we used 
\bea
    \bfY_{ab}[v,\lambda]=v-\bfP_{ab}\sum_{\Abox\in\lambda}\chi_{ab,v}(\Bbox).
\eea
This identity shows that the 7d $\U(1)$ theory can be understood as a 5d $\U(\infty)$-theory stacked in one of the transverse direction.\footnote{Combining with the discussion in section~\ref{sec:QTgl1}, one will see that this property is related to the fact that the MacMahon representation of quantum toroidal $\mathfrak{gl}_{1}$ is obtained as infinite tensor products of the Fock representations.} Doing this decomposition for other cases and inserting it in the gauge origami index, we obtain the description in lower dimensional partitions.

\subsubsection{9d theory}\label{sec:decomp9d}
For simplicity, let us consider the $\U(1|1)$ theory of the magnificent four system: 
\bea
    \bfY=\bfN-\bfP_{\four}\bfK,\quad \bfN=v-\bar{v},\quad \bfK=\sum_{\shcube \in\rho}\chi_{\four,v}(\hcube),
\eea
where $\rho$ is the solid partition. We have three possible descriptions: $(3,1),\, (2,2),\, (1,3)$.
\paragraph{$(3,1)$-type}
We choose the fourth direction to be the direction where the height function is defined and then obtain for $|q_{1,2,3}|<1$
\begin{equation}
\begin{split}
\bfK&=\sum_{i,j,k=1}^{\infty}\sum_{l=1}^{\rho_{i,j,k}}vq_{1}^{i-1}q_{2}^{j-1}q_{3}^{k-1}q_{4}^{l-1}=\frac{v}{\prod_{a\in\four}(1-q_{a})}-\frac{1}{1-q_{4}}\sum_{i,j,k=1}^{\infty}vq_{1}^{i-1}q_{2}^{j-1}q_{3}^{k-1}q_{4}^{\rho_{i,j,k}},\\
\bfY&=\bfP_{123}\sum_{i,j,k=1}^{\infty}vq_{1}^{i-1}q_{2}^{j-1}q_{3}^{k-1}q_{4}^{\rho_{i,j,k}}-\bar{v}=\bfP_{123}\bfX-\bar{v}.
\end{split}
\end{equation}
For higher rank theories such as $\U(n|n)$, we have 
\bea
&\bfY=\bfP_{123}\bfX-\bar{\mathbf{n}},\quad \bfX=\sum_{x\in\mathcal{X}_{\four}}x,\quad \bar{\mathbf{n}}=\sum_{\alpha=1}^{n}\bar{v}_{\alpha}\quad \mathcal{X}_{\four}=\left\{v_{\alpha}q_{1}^{i-1}q_{2}^{j-1}\left.q_{3}^{k-1}q_{4}^{\rho^{(\alpha)}_{i,j,k}}\right| \substack{i,j,k=1,\ldots,\infty\\\alpha=1,\ldots,n} \right\}.
\eea

\paragraph{$(2,2)$-type}
For example, in this description, we can decompose the solid partition into two Young diagrams $\rho=(\lambda_{12};\mu_{34})$. The $\lambda_{12},\mu_{34}$ will be Young diagrams extending in the 12 and 34 directions respectively. We choose this decomposition for simplicity but of course, we can use the quadrality to choose other decompositions. For each box in the Young diagram $(i,j)\in\lambda_{12}$, we have a Young diagram $\mu^{(i,j)}_{34}$ ($\mu^{\Abox}_{34}$ for $\Bbox\in\lambda_{12}$) and thus 
\bea
\bfK=\sum_{(i,j)\in\lambda_{12}}\sum_{\Abox\in\mu^{(i,j)}_{34}}\chi_{34,vq_{1}^{i-1}q_{2}^{j-1}}(\Bbox).
\eea
We then have 
\bea
    \bfY_{\four}+\bar{\mathbf{n}}&=v-\bfP_{12}\bfP_{34}\sum_{(i,j)\in\lambda_{12}}\sum_{\Abox\in\mu^{(i,j)}_{34}}\chi_{34,vq_{1}^{i-1}q_{2}^{j-1}}(\Bbox)\\
&=\bfY_{12}[v,\lambda_{12}]+\bfP_{12}\sum_{\Abox\in\lambda_{12}}\bfY_{34}[\chi_{12,v}(\Bbox),\mu_{34}^{\Abox}].
\eea

\paragraph{$(1,3)$-type}
We decompose the solid partition into sequences of plane partitions as
\bea
    \rho=(\pi^{(1)},\ldots,\pi^{(l)},\ldots)
\eea
after choosing the specific direction to be 4. Namely, multiple plane partitions are piled up in the fourth direction. We then have
\bea
\bfK_{\four}=\sum_{l=1}^{\infty}\sum_{\scube\in\pi^{(l)}}\chi_{\bar{4},vq_{4}^{l-1}}(\cube),\quad |q_{4}|<1,
\eea
and the character of $\bfY_{\four}$ is rewritten as 
\bea
    \bfY+\bar{\mathbf{n}}&=\bfN-\bfP_{\four}\bfK=\bfP_{4}\left(\frac{v}{1-q_{4}}-\sum_{l=1}^{\infty}\sum_{\scube\in\pi^{(l)}}\chi_{\bar{4},vq_{4}^{l-1}}(\cube)\right)\\
&=\bfP_{4}\left(\sum_{l=1}^{\infty}\left(vq_{4}^{l-1}-\sum_{l=1}\sum_{\scube\in\pi^{(l)}}\chi_{\bar{4},vq_{4}^{l-1}}(\cube)\right)\right)=\bfP_{4}\sum_{l=1}^{\infty}\bfY_{123}[vq_{4}^{l-1},\pi^{(l)}].
\eea
This description resembles the 7d $(1,2)$-type description where 2d partitions were piled up to get 3d partitions.

\subsection{\texorpdfstring{$qq$}{qq}-characters from instanton partition functions}\label{sec:qqpartitionfunct}
\subsubsection{5d theory}
Let us review the $qq$-character of the 5d $\mathcal{N}=1^{\ast}$ theory on $\mathbb{C}^{2}_{12}\times \mathbb{S}^{1}$. Using the recursion formulas of the 5d Nekrasov factors in \eqref{eq:5dNekrasovrecursion}, the recursion relation of the $\U(1)$ partition function $\widetilde{\mathcal{Z}}^{\D4}_{12}[\lambda]$ is (see the derivation and the discussion in \eqref{eq:app-D4U1recursionformula})
\bea\label{eq:D4qqrecursion}
    \frac{\widetilde{\mathcal{Z}}^{\D4}_{12}[\lambda+\Bbox]}{\widetilde{\mathcal{Z}}
^{\D4}_{12}[\lambda]}&=-\frac{\mathscr{Y}_{\lambda,v}^{12}(q_{3}^{-1}\chi_{12,v}(\Bbox))\mathscr{Y}^{12}_{\lambda+\Abox,v}(q_{4}^{-1}\chi_{12,v}(\Bbox))}{\mathscr{Y}^{12}_{\lambda,v}(\chi_{12,v}(\Bbox))\mathscr{Y}^{12}_{\lambda+\Abox,v}(q_{34}^{-1}\chi_{12,v}(\Bbox))}\\
&=-\mathscr{S}_{12}\left(q_{4}\right)\frac{\mathscr{Y}_{\lambda,v}^{12}(q_{3}^{-1}\chi_{12,v}(\Bbox))\mathscr{Y}^{12}_{\lambda,v}(q_{4}^{-1}\chi_{12,v}(\Bbox))}{\mathscr{Y}^{12}_{\lambda,v}(\chi_{12,v}(\Bbox))\mathscr{Y}^{12}_{\lambda+\Abox,v}(q_{34}^{-1}\chi_{12,v}(\Bbox))}
\eea
where $\Bbox\in A(\lambda)$ and we used $\mathscr{Y}_{\lambda+\Abox,v}^{12}(x)=\mathscr{S}_{12}\left(\chi_{12,v}(\Bbox)/x\right)\mathscr{Y}_{\lambda,v}^{12}(x)$. Strictly speaking, the $\mathscr{Y}$-functions in the denominator are singular. The term $\mathscr{Y}^{12}_{\lambda,v}(x)$ has a zero in $x=\chi_{12,v}(\Bbox)$ while the term $\mathscr{Y}^{12}_{\lambda+\Abox}(q_{12}x)$ has  a pole at $x=\chi_{12,v}(\Bbox)$. However, since the zero and pole will cancel with each other, we can write in the above way.

We can also write the recursion formula using the dual $\mathscr{Y}$-functions as
\bea\label{eq:D4qqdualrecursion}
    \frac{\widetilde{\mathcal{Z}}^{\D4}_{12}[\lambda+\Bbox]}{\widetilde{\mathcal{Z}}
^{\D4}_{12}[\lambda]}&=-\frac{\mathscr{Y}_{\lambda,v}^{12\,\vee}(q_{3}^{-1}\chi_{12,v}(\Bbox))\mathscr{Y}^{12\,\vee}_{\lambda+\Abox,v}(q_{4}^{-1}\chi_{12,v}(\Bbox))}{\mathscr{Y}^{12\,\vee}_{\lambda,v}(\chi_{12,v}(\Bbox))\mathscr{Y}^{12\,\vee}_{\lambda+\Abox,v}(q_{34}^{-1}\chi_{12,v}(\Bbox))}\\
&=-\mathscr{S}_{12}\left(q_{3}\right)\frac{\mathscr{Y}_{\lambda,v}^{12\,\vee}(q_{3}^{-1}\chi_{12,v}(\Bbox))\mathscr{Y}^{12\,\vee}_{\lambda,v}(q_{4}^{-1}\chi_{12,v}(\Bbox))}{\mathscr{Y}^{12\,\vee}_{\lambda,v}(\chi_{12,v}(\Bbox))\mathscr{Y}^{12\,\vee}_{\lambda+\Abox,v}(q_{34}^{-1}\chi_{12,v}(\Bbox))}
\eea
Note here that we have the property $\mathscr{S}_{12}(q_{3})=\mathscr{S}_{12}(q_{4})$. 

The recursion relations are then rewritten as
\bea\label{eq:D4qqanalytic}
\underset{x=\chi_{12,v}(\Bbox)}{\Res}\left[\mathscr{Y}^{12}_{\lambda+\Abox,v}(q_{12}x)\widetilde{\mathcal{Z}}^{\D4}_{12}[\lambda+\Bbox]+\mathscr{S}_{12}(q_{4})\frac{\mathscr{Y}^{12}_{\lambda,v}(q_{3}^{-1}x)\mathscr{Y}^{12}_{\lambda,v}(q_{4}^{-1}x)}{\mathscr{Y}^{12}_{\lambda,v}(x)}\widetilde{\mathcal{Z}}^{\D4}_{12}[\lambda]\right]=0,\\
\underset{x=\chi_{12,v}(\Bbox)}{\Res}\left[\mathscr{Y}^{12\,\vee}_{\lambda+\Abox,v}(q_{12}x)\widetilde{\mathcal{Z}}^{\D4}_{12}[\lambda+\Bbox]+\mathscr{S}_{12}(q_{3})\frac{\mathscr{Y}^{12\,\vee}_{\lambda,v}(q_{3}^{-1}x)\mathscr{Y}^{12\,\vee}_{\lambda,v}(q_{4}^{-1}x)}{\mathscr{Y}^{12\,\vee}_{\lambda,v}(x)}\widetilde{\mathcal{Z}}^{\D4}_{12}[\lambda]\right]=0
\eea
which means that the singularity of the first term at $x=\chi_{12,v}(\Bbox)$ is cancelled by the second term. Including the topological term and taking the instanton summation, we obtain
\bea
&\left\langle\widehat{\mathscr{Y}}^{12}(q_{12}x)\right\rangle_{12}+\mathfrak{q}\mathscr{S}_{12}(q_{4})\left\langle\frac{\widehat{\mathscr{Y}}^{12}(q_{3}^{-1}x)\widehat{\mathscr{Y}}^{12}(q_{4}^{-1}x)}{\widehat{\mathscr{Y}}^{12\,\vee}(x)}\right\rangle_{12},\\
&\left\langle\widehat{\mathscr{Y}}^{12\,\vee}(q_{12}x)\right\rangle_{12}+\mathfrak{q}\mathscr{S}_{12}(q_{4})\left\langle\frac{\widehat{\mathscr{Y}}^{12\,\vee}(q_{3}^{-1}x)\widehat{\mathscr{Y}}^{12\,\vee}(q_{4}^{-1}x)}{\widehat{\mathscr{Y}}^{12\,\vee}(x)}\right\rangle_{12}
\eea
where
\bea
\langle\widehat{\mathcal{O}}\rangle_{12}=\frac{1}{\mathcal{Z}_{12}}\sum_{\lambda}\mathfrak{q}^{|\lambda|}\widetilde{\mathcal{Z}}^{\D4}_{12}[\lambda]\mathcal{O}_{\lambda,v},\quad \mathcal{Z}_{12}=\sum_{\lambda}\mathfrak{q}^{|\lambda|}\widetilde{\mathcal{Z}}_{12}^{\D4}[\lambda].
\eea
Note here that we are understanding $\widehat{\mathcal{O}}$ as an operator acting on the fixed points labeled by $(v,\lambda)$, where $v$ is the exponentiated Coulomb branch parameter of the $\U(1)$ gauge theory on $\mathbb{C}^{2}_{12}\times \mathbb{S}^{1}$.  The subindex $12$ denotes on which gauge theory the instanton expectation value is taken. Although we were considering when the bulk theory is the $\U(1)$ theory, we can generalize it to $\U(n)$ theories. In such cases, the $\widehat{\mathscr{Y}}^{12}(x)$ operator will act as
\bea
\widehat{\mathscr{Y}}^{12}(x)\longrightarrow \prod_{\alpha=1}^{n}\mathscr{Y}^{12}_{\lambda^{(\alpha)},v_{\alpha}}(x),\quad \widehat{\mathscr{Y}}^{12\,\vee}(x)\longrightarrow \prod_{\alpha=1}^{n}\mathscr{Y}^{12\,\vee}_{\lambda^{(\alpha)},v_{\alpha}}(x)
\eea
under an instanton background labeled by $(\vec{v},\vec{\lambda})$.

Due to the property in \eqref{eq:D4qqanalytic}, the singularities coming from the first term are canceled by the second term. This procedure to add a term and compensate the pole singularity is called the \textbf{iWeyl reflection} \cite{Nekrasov:2012xe} and is defined as
\bea
\widehat{\mathscr{Y}}^{12}(q_{12}x)&\longmapsto \mathfrak{q}\mathscr{S}_{12}(q_{4})\widehat{\mathscr{Y}}^{12}(q_{12}x)\widehat{\mathscr{A}}(x)^{-1},\\
\widehat{\mathscr{Y}}^{12\,\vee}(q_{12}x)&\longmapsto \mathfrak{q}\mathscr{S}_{12}(q_{3})\widehat{\mathscr{Y}}^{12\,\vee}(q_{12}x){\widehat{\mathscr{A}}}^{\,\vee}(x)^{-1},
\eea
where 
\bea
\widehat{\mathscr{A}}(x)=\frac{\widehat{\mathscr{Y}}^{12}(x)\widehat{\mathscr{Y}}^{12}(q_{3}^{-1}q_{4}^{-1}x)}{\widehat{\mathscr{Y}}^{12}(q_{3}^{-1}x)\widehat{\mathscr{Y}}^{12}(q_{4}^{-1}x)},\quad {\widehat{\mathscr{A}}}^{\,\vee}(x)=\frac{\widehat{\mathscr{Y}}^{12\,\vee}(x)\widehat{\mathscr{Y}}^{12\,\vee}(q_{3}^{-1}q_{4}^{-1}x)}{\widehat{\mathscr{Y}}^{12\,\vee}(q_{3}^{-1}x)\widehat{\mathscr{Y}}^{12\,\vee}(q_{4}^{-1}x)}
\eea
After this iWeyl reflection, new singularities will appear from the new $\mathscr{Y}$-functions coming from $\widehat{\mathscr{A}}(x)$ and thus, we need to recursively do this procedure. After summing up all the terms appearing after the iWeyl reflection, we will obtain a pole-free function in $x$ which is called the \textbf{$qq$-character}:
\bea
\widehat{\mathscr{T}}^{12}(x)&=\widehat{\mathscr{Y}}^{12}(q_{12}x)+\mathfrak{q}\mathscr{S}_{12}(q_{4})\widehat{\mathscr{Y}}^{12}(q_{12}x)\widehat{\mathscr{A}}(x)^{-1}+\cdots \\
\widehat{\mathscr{T}}^{12\,\vee}(x)&=\widehat{\mathscr{Y}}^{12}(q_{12}x)+\mathfrak{q}\mathscr{S}_{12}(q_{3})\widehat{\mathscr{Y}}^{12\,\vee}(q_{12}x){\widehat{\mathscr{A}}}^{\,\vee}(x)^{-1}+\cdots
\eea
Doing the iWeyl reflection recursively, one can show that the full formula of this $qq$-character is  
\bea
\widehat{\mathscr{T}}^{12}(x)&=\sum_{\mu\in\mathcal{P}}\mathfrak{q}^{|\mu|}\widetilde{\mathcal{Z}}^{\D4}_{34}[\mu]\widehat{\mathscr{Y}}^{12}(q_{12}x)\prod_{\Abox\in\mu}\widehat{\mathscr{A}}(\chi_{34,x}(\Bbox))^{-1}\\
&=\sum_{\mu\in\mathcal{P}}\mathfrak{q}^{|\mu|}\widetilde{\mathcal{Z}}^{\D4}_{34}[\mu]\frac{\prod_{\Abox\in A(\mu)}\widehat{\mathscr{Y}}^{12}(\chi_{34,x}(\Bbox))}{\prod_{\Abox\in R(\mu)}\widehat{\mathscr{Y}}^{12}(q_{34}\chi_{34,x}(\Bbox))},\\
\widehat{\mathscr{T}}^{12\,\vee}(x)&=\sum_{\mu\in\mathcal{P}}\mathfrak{q}^{|\mu|}\widetilde{\mathcal{Z}}^{\D4}_{34}[\mu]\widehat{\mathscr{Y}}^{12\,\vee}(q_{12}x)\prod_{\Abox\in\mu}{\widehat{\mathscr{A}}}^{\,\vee}(\chi_{34,x}(\Bbox))^{-1}\\
&=\sum_{\mu\in\mathcal{P}}\mathfrak{q}^{|\mu|}\widetilde{\mathcal{Z}}^{\D4}_{34}[\mu]\frac{\prod_{\Abox\in A(\mu)}\widehat{\mathscr{Y}}^{12\,\vee}(\chi_{34,x}(\Bbox))}{\prod_{\Abox\in R(\mu)}\widehat{\mathscr{Y}}^{12\,\vee}(q_{34}\chi_{34,x}(\Bbox))}.
\eea
Obviously, $\mathscr{S}_{12}(q_{3})=\mathscr{S}_{12}(q_{4})$ is the one-instanton contribution of the $\U(1)$ affine quiver gauge theory on $\mathbb{C}^{2}_{34}\times \mathbb{S}^{1}$ (see \eqref{app-eq:one-instanton}).

Generally, we can do the same calculation for other 5d theories on $\mathbb{C}^{2}_{A}\times \mathbb{S}^{1}\,(A\in\six)$. 
\begin{proposition}
    The $qq$-characters with respect with the 5d $\mathcal{N}=1^{\ast}$ gauge theory on $\mathbb{C}_{A}^{2}\times \mathbb{S}^{1}$ where $A\in\six$ are given by
    \bea\label{eq:D4qqpart}
    \widehat{\mathscr{T}}^{A}(x)&=\sum_{\lambda\in\mathcal{P}}\mathfrak{q}^{|\mu|}\widetilde{\mathcal{Z}}^{\D4}_{\bar{A}}[\lambda]\widehat{\mathscr{Y}}^{A}(q_{A}x)\prod_{\Abox \in \lambda}\widehat{\mathscr{A}}(\chi_{\bar{A},x}(\Bbox))^{-1},\\
     \widehat{\mathscr{T}}^{A\,\vee}(x)&=\sum_{\lambda\in\mathcal{P}}\mathfrak{q}^{|\mu|}\widetilde{\mathcal{Z}}^{\D4}_{\bar{A}}[\lambda]\widehat{\mathscr{Y}}^{A\,\vee}(q_{A}x)\prod_{\Abox \in \lambda}{\widehat{\mathscr{A}}}^{\,\vee}(\chi_{\bar{A},x}(\Bbox))^{-1}.\\
    \eea
    The iWeyl reflection of $\widehat{\mathscr{Y}}^{A}(x)$, $\widehat{\mathscr{Y}}^{A\,\vee}(x)$ is defined as 
   \bea\label{eq:D4iWeylpartition}
\widehat{\mathscr{Y}}^{A}(q_{A}x)&\longmapsto \mathfrak{q}\mathscr{S}_{A}(q_{c})\widehat{\mathscr{Y}}^{A}(q_{A}x)\widehat{\mathscr{A}}(x)^{-1},\\
\widehat{\mathscr{Y}}^{A\,\vee}(q_{A}x)&\longmapsto \mathfrak{q}\mathscr{S}_{A}(q_{d})\widehat{\mathscr{Y}}^{A\,\vee}(q_{A}x){\widehat{\mathscr{A}}}^{\,\vee}(x)^{-1},
\eea
where $\{c,d\}=\bar{A}$ and 
\bea
\widehat{\mathscr{A}}(x)=\frac{\widehat{\mathscr{Y}}^{A}(x)\widehat{\mathscr{Y}}^{A}(q_{A}x)}{\widehat{\mathscr{Y}}^{A}(q_{c}^{-1}x)\widehat{\mathscr{Y}}^{A}(q_{d}^{-1}x)},\quad {\widehat{\mathscr{A}}}^{\,\vee}(x)=\frac{\widehat{\mathscr{Y}}^{A\,\vee}(x)\widehat{\mathscr{Y}}^{A\vee}(q_{A}x)}{\widehat{\mathscr{Y}}^{A\,\vee}(q_{c}^{-1}x)\widehat{\mathscr{Y}}^{A\,\vee}(q_{d}^{-1}x)}.
\eea
\end{proposition}
Note that the explicit form of the operator $\widehat{\mathscr{A}}(x)$ depends on the theory we consider.\footnote{Later, we will see in the operator formalism that it will be a unique operator (see \eqref{eq:oprelationwithD0}).} We call the $qq$-characters $\widehat{\mathscr{T}}^{A}(x),\widehat{\mathscr{T}}^{A\,\vee}(x)$ the $\D4_{\bar{A}}$-brane $qq$-characters (shortly $\D4_{\bar{A}}$ $qq$-characters). Namely, the $qq$-characters associated with the 5d theory on the $\D4_{A}$-brane are the $\D4_{\bar{A}}$ $qq$-characters.

\paragraph{D4 $qq$-characters as codimension four defects}
The $qq$-character is understood as a codimension four defect with respect to the corresponding bulk theory. Let us consider $\mathbb{C}^{2}_{12}\times \mathbb{S}^{1}$ as the bulk theory. The $qq$-character is understood as a $\D4$-brane spanning the codimension four subspace $\mathbb{C}^{2}_{34}\times\mathbb{S}^{1}$ in the transverse direction:
\bea
\includegraphics[width=8cm]{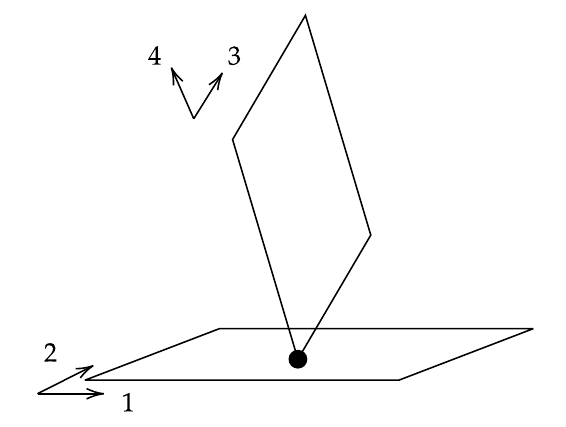}
\eea
\begin{table}[t]
\begin{center}
\begin{tabular}{|c|c|c|c|c|c|c|c|c|c|c|}
\hline
& \multicolumn{2}{c|}{$\mathbb{C}_{1}$} & \multicolumn{2}{c|}{$\mathbb{C}_{2}$} & \multicolumn{2}{c|}{$\mathbb{C}_{3}$} & \multicolumn{2}{c|}{$\mathbb{C}_{4}$} & \multicolumn{2}{c|}{$\mathbb{R}\times \mathbb{S}^{1}$} \\
\cline{2-11}  & 1 & 2 & 3 & 4& 5 & 6 & 7 & 8 & 9& 0\\
\hline
$\D4_{12} $& $-$ & $-$ & $-$ & $-$ & $\bullet$ & $\bullet$ & $\bullet$ & $\bullet$ & $\bullet$ & $-$ \\
\hline
$\D4_{34}$ & $\bullet$ & $\bullet$ & $\bullet$ & $\bullet$ & $-$ & $-$ & $-$ & $-$ & $\bullet$ & $-$ \\
\hline
\end{tabular}
\caption{The configuration of the two sets of $\D4$-branes wrapping the $12$-plane and $34$-plane. From the $\D4_{12}$ view point, the $\D4_{34}$-brane plays the role of the $qq$-character.}
\label{t:D4defectbraneweb}
\end{center}
\end{table}
The brane web construction is given in Table \ref{t:D4defectbraneweb}. This is the reason why we called the $qq$-character of the D4 theory on $\mathbb{C}^{2}_{A}\times \mathbb{S}^{1}$ the $\D4_{\bar{A}}$ $qq$-character.

The contour integral formula of the $qq$-character is intuitive. The contour integral comes from the following character:
\bea
&\bfN_{34}^{\vee}\frac{\bfP_{12}^{\vee}\bfP_{34}}{\bfP_{\four}}\bfN_{12}-\bfP_{34}^{\vee}\bfN_{12}^{\vee}\bfK-\bfN_{34}^{\vee}\bfP_{12}^{\vee}\bfK+\bfP_{123}^{\vee}\bfK_{12}^{\vee}\bfK_{12}+\bfP_{134}^{\vee}\bfK^{\vee}_{34}\bfK_{34}+\bfP_{\four}\bfK_{12}^{\vee}\bfK_{34}\\
=&\bfN_{34}^{\vee}q_{12}^{-1}\left(\bfN_{12}-\bfP_{12}\bfK\right)-\bfP_{34}^{\vee}\bfN_{12}^{\vee}\bfK+\bfP_{123}^{\vee}\bfK_{12}^{\vee}\bfK_{12}+\bfP_{134}^{\vee}\bfK^{\vee}_{34}\bfK_{34}+\bfP_{\four}\bfK_{12}^{\vee}\bfK_{34}
\eea
where $\bfK=\bfK_{12}+\bfK_{34}$ and we included the perturbative part for convention. Note that this character is just a part of \eqref{eq:D4squareroot} where the terms including $\bfN_{34},\,\bfK_{34}$ are kept, namely the crossed instanton configuration.

Consider the case when there is only one transverse D4-brane with $\bfN_{34}=x$. The contour integral is then written as
\bea
\frac{\mathcal{G}_{4}^{k}}{k!}\oint\prod_{I=1}^{k}\frac{dx_{I}}{2\pi\iota x_{I}}\textcolor{red}{\prod_{\alpha=1}^{n}\left(1-\frac{q_{12}x}{v_{\alpha}}\right)\prod_{I=1}^{k}\mathscr{S}_{12}\left(\frac{x}{x_{I}}\right)}\textcolor{blue}{
\prod_{\alpha=1}^{n}\prod_{I=1}^{k}\mathscr{S}_{34}\left(\frac{v_{\alpha}}{x_{I}}\right)\prod_{I<J}\mathcal{A}_{\mathbb{C}^{4}}\left(\frac{x_{I}}{x_{J}}\right)^{-1}}.
\eea
The red term comes from the transverse $\D4_{34}$-brane while the blue term comes from the $\D4_{12}$-brane. One can see that after localization the red term gives the contribution coming from the $\mathscr{Y}$-functions. When $n=1$, after evaluating the poles, one can see we have
\bea
\left\langle \widehat{\mathscr{T}}^{12}(x)\right\rangle _{12}=\mathcal{Z}^{12|34}_{\text{cross. inst.}},\quad \left\langle \widehat{\mathscr{T}}^{12\,\vee}(x)\right\rangle _{12}=\mathcal{Z}^{12|34}_{\text{cross. inst.}}
\eea
Namely, integrating the transverse $\D4_{34}$-brane first, the operators $\widehat{\mathscr{T}}^{12}(x),\widehat{\mathscr{T}}^{12\vee}(x)$ representing the defect brane arises in the bulk theory and the expectation value of it gives the partition function of the gauge origami system.

\subsubsection{7d theory}
Let us do a similar analysis for the 7d theory. We focus first on the 7d theory on $\mathbb{C}^{3}_{123}\times \mathbb{S}^{1}$. The recursion formula of $\widetilde{\mathcal{Z}}^{\D6}_{123}[\pi]$ is given from \eqref{eq:app-D6U1recursionformula}:
\bea
\frac{\widetilde{\mathcal{Z}}^{\D6}_{123}[\pi+\cube]}{\widetilde{\mathcal{Z}}^{\D6}_{123}[\pi]}=-\frac{\mathscr{W}^{\bar{4}}_{\pi+\scube,v}(q_{123}\chi_{123,v}(\cube))}{\mathscr{W}_{\pi,v}^{\bar{4}}(\chi_{123,v}(\cube))}=-q_{4}\frac{\mathscr{W}^{\bar{4}\,\vee}_{\pi+\scube,v}(q_{123}\chi_{123,v}(\cube))}{\mathscr{W}_{\pi,v}^{\bar{4}\,\vee}(\chi_{123,v}(\cube))},
\eea
where $\cube\in A(\pi)$. It is then rewritten as 
\bea
\underset{x=\chi_{123,v}(\scube)}{\Res}\left[ \mathscr{W}_{\pi+\scube,v}^{\bar{4}}(q_{123}x)^{-1}\widetilde{\mathcal{Z}}^{\D6}_{123}[\pi+\cube]+\mathscr{W}_{\pi,v}^{\bar{4}}(x)^{-1}\widetilde{\mathcal{Z}}^{\D6}_{123}[\pi] \right]=0,\\
\underset{x=\chi_{123,v}(\scube)}{\Res}\left[ \mathscr{W}_{\pi+\scube,v}^{\bar{4}\,\vee}(q_{123}x)^{-1}\widetilde{\mathcal{Z}}^{\D6}_{123}[\pi+\cube]+q_{4}\mathscr{W}_{\pi,v}^{\bar{4}\,\vee}(x)^{-1}\widetilde{\mathcal{Z}}^{\D6}_{123}[\pi] \right]=0
\eea
which shows that the second term cancels the singularity coming from the pole at $x=\chi_{123,v}(\scube)$ of the first term. Including the topological term and taking the instanton summation, we obtain 
\bea
&\left\langle\widehat{\mathscr{W}}^{\bar{4}}(q_{123}x)^{-1}\right\rangle_{123}+\mathfrak{q}\left\langle\widehat{\mathscr{W}}^{\bar{4}}(x)^{-1}\right\rangle_{123},\\
&\left\langle\widehat{\mathscr{W}}^{\bar{4}\,\vee}(q_{123}x)^{-1}\right\rangle_{123}+\mathfrak{q}q_{4}\left\langle\widehat{\mathscr{W}}^{\bar{4}\,\vee}(x)^{-1}\right\rangle_{123}
\eea
where 
\bea
\left\langle \mathcal{O}   \right\rangle_{123}=\frac{1}{\mathcal{Z}_{123}}\sum_{\pi}\mathfrak{q}^{|\pi|}\widetilde{\mathcal{Z}}^{\D6}_{123}[\pi]\mathcal{O}_{\pi,v},\quad \mathcal{Z}_{123}=\sum_{\pi}\mathfrak{q}^{|\pi|}\widetilde{\mathcal{Z}}^{\D6}_{123}[\pi].
\eea
This combination is a pole-free function at $x=\chi_{123,v}(\cube)$ because the singularity coming from the first term is canceled by the second term. We can define the iWeyl reflection analogous to the D4 case as
\bea\label{eq:D2iWeylpartition}
\widehat{\mathscr{W}}^{\bar{4}}(q_{123}x)^{-1}&\longmapsto \mathfrak{q}\widehat{\mathscr{W}}^{\bar{4}}(q_{123}x)^{-1}\widehat{\mathscr{A}}(x)^{-1},\\
\widehat{\mathscr{W}}^{\bar{4}\,\vee}(q_{123}x)^{-1}&\longmapsto \mathfrak{q}q_{4}\widehat{\mathscr{W}}^{\bar{4}\,\vee}(q_{123}x)^{-1}{\widehat{\mathscr{A}}}^{\vee}(x)^{-1}
\eea
where 
\bea
\widehat{\mathscr{A}}(x)=\frac{\widehat{\mathscr{W}}^{\bar{4}}(x)}{\widehat{\mathscr{W}}^{\bar{4}}(q_{4}^{-1}x)},\quad {\widehat{\mathscr{A}}}^{\vee}(x)=\frac{\widehat{\mathscr{W}}^{\bar{4}\,\vee}(x)}{\widehat{\mathscr{W}}^{\bar{4}\,\vee}(q_{4}^{-1}x)} .
\eea
We can do this procedure recursively and obtain a pole-free function in $x$ which is a generalization of the 5d case:
\bea
\widehat{\mathscr{T}}^{123}(x)&=\widehat{\mathscr{W}}^{\bar{4}}(q_{123}x)^{-1}+\mathfrak{q}\widehat{\mathscr{W}}^{\bar{4}}(q_{123}x)^{-1}\widehat{\mathscr{A}}(x)^{-1}+\cdots\\
&=\widehat{\mathscr{W}}^{\bar{4}}(q_{123}x)^{-1}+\mathfrak{q}\widehat{\mathscr{W}}^{\bar{4}}(x)^{-1}+\cdots,\\
\widehat{\mathscr{T}}^{123\,\vee}(x)&=\widehat{\mathscr{W}}^{\bar{4}\,\vee}(q_{123}x)^{-1}+\mathfrak{q}q_{4}\widehat{\mathscr{W}}^{\bar{4}}(q_{123}x)^{-1}{\widehat{\mathscr{A}}}^{\,\vee}(x)^{-1}+\cdots\\
&=\widehat{\mathscr{W}}^{\bar{4}\,\vee}(q_{123}x)^{-1}+\mathfrak{q}q_{4}\widehat{\mathscr{W}}^{\bar{4}\,\vee}(x)^{-1}+\cdots.
\eea
The full formula is given as
\bea
\widehat{\mathscr{T}}^{123}(x)&=\sum_{k=0}^{\infty}\mathfrak{q}^{k}\,\widehat{\mathscr{W}}^{\bar{4}}(q_{123}x)^{-1}\prod_{i=1}^{k}\widehat{\mathscr{A}}(q_{4}^{i-1}x)^{-1}=\sum_{k=0}^{\infty}\mathfrak{q}^{k}\, \widehat{\mathscr{W}}^{\bar{4}}(q_{4}^{k-1}x)^{-1},\\
\widehat{\mathscr{T}}^{123\,\vee}(x)&=\sum_{k=0}^{\infty}(\mathfrak{q}q_{4})^{k}\,\widehat{\mathscr{W}}^{\bar{4}\,\vee}(q_{123}x)^{-1}\prod_{i=1}^{k}{\widehat{\mathscr{A}}}^{\,\vee}(q_{4}^{i-1}x)^{-1}=\sum_{k=0}^{\infty}(\mathfrak{q}q_{4})^{k}\, \widehat{\mathscr{W}}^{\bar{4}\,\vee}(q_{4}^{k-1}x)^{-1}
\eea

A similar discussion for other 7d theories on $\mathbb{C}^{3}_{\bar{a}}\times \mathbb{S}^{1}$ can be done and we have the following statement.
\begin{proposition}
    The $qq$-characters with respect to the 7d $\mathcal{N}=1$ theory on $\mathbb{C}^{3}_{\bar{a}}\times \mathbb{S}^{1}$ where $a\in\four$ are
    \bea\label{eq:D2qqpart}
\widehat{\mathscr{T}}^{\bar{a}}(x)&=\sum_{k=0}^{\infty}\mathfrak{q}^{k}\,\widehat{\mathscr{W}}^{\bar{a}}(q_{\bar{a}}x)^{-1}\prod_{i=1}^{k}\widehat{\mathscr{A}}(q_{a}^{i-1}x)^{-1}=\sum_{k=0}^{\infty}\mathfrak{q}^{k}\, \widehat{\mathscr{W}}^{\bar{a}}(q_{a}^{k-1}x)^{-1},\\
\widehat{\mathscr{T}}^{\bar{a}\,\vee}(x)&=\sum_{k=0}^{\infty}(\mathfrak{q}q_{a})^{k}\,\widehat{\mathscr{W}}^{\bar{a}\,\vee}(q_{\bar{a}}x)^{-1}\prod_{i=1}^{k}{\widehat{\mathscr{A}}}^{\,\vee}(q_{a}^{i-1}x)^{-1}=\sum_{k=0}^{\infty}(\mathfrak{q}q_{a})^{k}\, \widehat{\mathscr{W}}^{\bar{a}\,\vee}(q_{a}^{k-1}x)^{-1},
    \eea
    where the iWeyl reflection is defined as 
    \bea\label{eq:D2iWeylpartgeneral}
\widehat{\mathscr{W}}^{\bar{a}}(q_{\bar{a}}x)^{-1}&\longmapsto \mathfrak{q}\widehat{\mathscr{W}}^{\bar{a}}(q_{\bar{a}}x)^{-1}\widehat{\mathscr{A}}(x)^{-1},\\
\widehat{\mathscr{W}}^{\bar{a}\,\vee}(q_{\bar{a}}x)^{-1}&\longmapsto \mathfrak{q}q_{a}\widehat{\mathscr{W}}^{\bar{a}\,\vee}(q_{\bar{a}}x)^{-1}{\widehat{\mathscr{A}}}^{\,\vee}(x)^{-1}
\eea
and
\bea
\widehat{\mathscr{A}}(x)=\frac{\widehat{\mathscr{W}}^{\bar{a}}(x)}{\widehat{\mathscr{W}}^{\bar{a}}(q_{a}^{-1}x)},\quad {\widehat{\mathscr{A}}}^{\,\vee}(x)=\frac{\widehat{\mathscr{W}}^{\bar{a}\,\vee}(x)}{\widehat{\mathscr{W}}^{\bar{a}\,\vee}(q_{a}^{-1}x)}.
\eea  
\end{proposition}
We call these $qq$-characters the $\D2$-brane $qq$-characters, shortly D2 $qq$-characters. Namely, the $qq$-character with respect to the $\D6_{\bar{a}}$ theory on $\mathbb{C}^{3}_{\bar{a}}\times \mathbb{S}^{1}$ is the $\D2_{a}$ $qq$-character.

\begin{figure}[t]
    \centering
    \includegraphics[width=9cm]{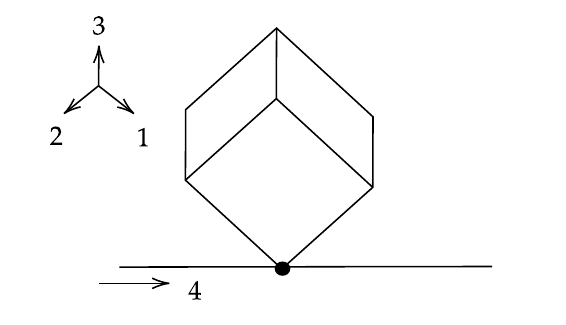}
    \caption{$\D2_{4}$-brane and $\D6_{123}$-brane. The $\D2_{4}$-brane wraps the $\mathbb{C}_{4}$ inside $\mathbb{C}^{4}$ while the $\D6_{123}$-brane wraps the $\mathbb{C}^{3}_{123}$. From the $\D2$-brane viewpoint, the $\D6$ shares the $\mathbb{S}^{1}$ part and is a 1d line defect which is a codimension two defect. On the other hand, from the $\D6$-brane viewpoint, the $\D2$-brane is a 1d line defect sharing the $\mathbb{S}^{1}$ part and thus is a codimension six defect.}
    \label{fig:D2D6defectqq}
\end{figure}
\begin{table}[t]
\begin{center}
\begin{tabular}{|c|c|c|c|c|c|c|c|c|c|c|}
\hline
& \multicolumn{2}{c|}{$\mathbb{C}_{1}$} & \multicolumn{2}{c|}{$\mathbb{C}_{2}$} & \multicolumn{2}{c|}{$\mathbb{C}_{3}$} & \multicolumn{2}{c|}{$\mathbb{C}_{4}$} & \multicolumn{2}{c|}{$\mathbb{R}\times \mathbb{S}^{1}$} \\
\cline{2-11}  & 1 & 2 & 3 & 4& 5 & 6 & 7 & 8 & 9& 0\\
\hline
$\D2_{4} $& $\bullet$& $\bullet$ & $\bullet$ & $\bullet$ & $\bullet$ & $\bullet$ & $-$ & $-$ & $\bullet$ & $-$ \\
\hline
$\D6_{123}$ & $-$  & $-$  & $-$  & $-$  & $-$ & $-$ & $\bullet$ & $\bullet$ & $\bullet$ & $-$ \\
\hline
\end{tabular}
\caption{The $\D2_{4}$-brane and $\D6_{123}$ are transverse to each other. The $\D6$ ($\D2$)-brane plays the role of the $qq$-character of the $\D2$ ($\D6$) theory.}
\label{t:D2D6defectbraneweb}
\end{center}
\end{table}

\paragraph{D2 $qq$-characters as codimension six defects}
Similar to the $\D4$ theories, we can understand the codimension six defects coming from $\D2$-branes in the transverse direction as $\D2$ $qq$-characters of the $\D6$ theory (see Figure~\ref{fig:D2D6defectqq}). The brane web configuration is given in Table~\ref{t:D2D6defectbraneweb}. The D2-brane in the transverse direction shares the $\mathbb{S}^{1}$ and thus is a line defect, which is a codimension six defect from the D6-brane theory viewpoint.

Mimicking the construction of the $\D4$ case, we propose that the contour integral formula comes from the character
\bea\label{eq:D6bulkD2defect}
&\bfN_{4}^{\vee}\frac{\bfP_{123}^{\vee}\bfP_{4}}{\bfP_{\four}}\bfN_{123}-\bfP_{4}^{\vee}\bfN_{123}^{\vee}\bfK-\bfP_{123}^{\vee}\bfN_{4}^{\vee}\bfK+\bfP_{123}^{\vee}\bfK_{123}^{\vee}\bfK_{123}+\bfP_{124}^{\vee}\bfK_{4}^{\vee}\bfK_{4}+\bfP_{\four}\bfK_{123}^{\vee}\bfK_{4}\\
=&-q_{4}\bfN^{\vee}_{4}\left(\bfN_{123}-\bfP_{123}\bfK\right)-\bfP_{4}^{\vee}\bfN_{123}^{\vee}\bfK+\bfP_{123}^{\vee}\bfK_{123}^{\vee}\bfK_{123}+\bfP_{124}^{\vee}\bfK_{4}^{\vee}\bfK_{4}+\bfP_{\four}\bfK_{123}^{\vee}\bfK_{4},
\eea
where $\bfK=\bfK_{123}+\bfK_{4}$. Note that although considering a subpart of the gauge origami system of the spiked instantons gives the $\D4$ $qq$-characters, for the $\D6$ case, we need to consider a gauge origami system where both $\D2$ and $\D6$-branes appear together. Such generalization is straightforward by using 
\bea
\bfY=\sum_{a\in\four}\bfP_{\bar{a}}\bfY_{a}+\sum_{a\in\four}\bfP_{a}\bfY_{\bar{a}},
\eea
inserting into $\bfY^{\vee}\bfY/\bfP_{\four}$, and doing the similar analysis in sections \ref{sec:M4partitionfunction}, \ref{sec:Tetrahedron_inst}, \ref{sec:spiked_partitionfunct}, \ref{sec:cplvortex_partitionfunct}. We omit such discussions and only focus only on \eqref{eq:D6bulkD2defect}.

Let us consider the case when we have one $\D2_{4}$-brane $\bfN_{4}=x$ and $n$ $\D6$-branes $\bfN_{123}=\sum_{\alpha=1}^{n}v_{\alpha}$. The contour integral is written as 
\bea
\frac{\mathcal{G}_{123}^{k}}{k!}\oint \prod_{I=1}^{k}\frac{dx_{I}}{2\pi\iota x_{I}}\textcolor{red}{\prod_{\alpha=1}^{n}\frac{1}{\left(1-q_{123}x/v_{\alpha}\right)}\prod_{I=1}^{k}g_{123}\left(\frac{x}{x_{I}}\right)}\textcolor{blue}{
\prod_{\alpha=1}^{n}\prod_{I=1}^{k}\mathscr{V}_{4}\left(\frac{v_{\alpha}}{x_{I}}\right)\prod_{I< J}^{k}\mathcal{A}_{\mathbb{C}^{4}}\left(\frac{x_{J}}{x_{I}}\right)^{-1}}.
\eea
The red term comes from the transverse $\D2$-brane, while the blue term comes from the bulk $\D6$-branes. After localization, the red term gives the $\mathscr{W}$-functions. Similar to the D4-case, integrating out the $\D2_{4}$-brane leads to a codimension six defect operator $\widehat{\mathscr{T}}^{123}(x),\widehat{\mathscr{T}}^{123\,\vee}(x)$ and the expectation value of it with respect to the bulk $\D6_{123}$-theory gives the partition function of the gauge origami system where there are $\D2_{4}$ and $\D6_{123}$-branes spanning transversely.

\subsubsection{3d theory}
Let us consider 3d theories on $\mathbb{C}_{a}\times \mathbb{S}^{1}$. We focus on the $a=4$ case. The recursion formula of the 3d partition function is \eqref{eq:app-D2U1recursionformula} (see also \eqref{eq:app-D2Unrecursion}):
\bea
\frac{\mathcal{Z}^{\D2}_{4}[k+\Bbox;q_{1},q_{2}]}{\mathcal{Z}^{\D2}_{4}[k;q_{1},q_{2}]}
&=-\frac{\mathscr{V}_{4}(q_{3})}{\mathscr{V}_{4}(q_{23})\mathscr{V}_{4}(q_{13})}\frac{\mathscr{U}^{4}_{k,v}(q_{1}^{-1}\chi_{4,v}(\Bbox))}{\mathscr{U}^{4}_{k,v}(q_{23}^{-1}\chi_{4,v}(\Bbox))}\frac{\mathscr{U}^{4}_{k,v}(q_{2}^{-1}\chi_{4,v}(\Bbox))}{\mathscr{U}^{4}_{k,v}(q_{13}^{-1}\chi_{4,v}(\Bbox))}\frac{\mathscr{U}^{4}_{k,v}(q_{3}^{-1}\chi_{4,v}(\Bbox))}{\mathscr{U}^{4}_{k,v}(q_{12}^{-1}\chi_{4,v}(\Bbox))}\frac{\mathscr{U}^{4}_{k+\Abox,v}(q_{123}^{-1}\chi_{4,v}(\Bbox))}{\mathscr{U}^{4}_{k,v}(\chi_{4,v}(\Bbox))},\\
\frac{\mathcal{Z}^{\D2}_{4}[k+\Bbox;q_{1},q_{2}]}{\mathcal{Z}^{\D2}_{4}[k;q_{1},q_{2}]}&=-\frac{\mathscr{V}_{4}(q_{1})\mathscr{V}_{4}(q_{2})}{\mathscr{V}_{4}(q_{12})}\frac{\mathscr{U}^{4\,\vee}_{k,v}(q_{1}^{-1}\chi_{4,v}(\Bbox))}{\mathscr{U}^{4\,\vee}_{k,v}(q_{23}^{-1}\chi_{4,v}(\Bbox))}\frac{\mathscr{U}^{4\,\vee}_{k,v}(q_{2}^{-1}\chi_{4,v}(\Bbox))}{\mathscr{U}^{4\,\vee}_{k,v}(q_{13}^{-1}\chi_{4,v}(\Bbox))}\frac{\mathscr{U}^{4\,\vee}_{k,v}(q_{3}^{-1}\chi_{4,v}(\Bbox))}{\mathscr{U}^{4\,\vee}_{k,v}(q_{12}^{-1}\chi_{4,v}(\Bbox))}\frac{\mathscr{U}^{4\,\vee}_{k+\Abox,v}(q_{123}^{-1}\chi_{4,v}(\Bbox))}{\mathscr{U}^{4\,\vee}_{k,v}(\chi_{4,v}(\Bbox))}
\eea
where we used $\mathscr{U}^{4}_{k+1,v}(x)=\mathscr{V}_{4}\left(\chi_{4,v}(\Bbox)/x\right)\mathscr{U}^{4}_{k,v}(x)$ and $\mathscr{U}^{4\,\vee}_{k+1,v}(x)=\mathscr{V}_{4}\left(q_{4}^{-1}x/\chi_{4,v}(\Bbox)\right)^{-1}\mathscr{U}^{4\,\vee}_{k,v}(x)$. This is rewritten as
\bea
&\underset{x=\chi_{4,v}(\Abox)}{\Res}\left[\mathscr{U}_{k+\Abox}^{4}(q_{4}x)^{-1}\mathcal{Z}^{\D2}_{4}[k+\Bbox\,;q_{1},q_{2}]\right.\\
&\qquad \left.+\frac{\mathscr{V}_{4}(q_{3})}{\mathscr{V}_{4}(q_{23})\mathscr{V}_{4}(q_{13})}\frac{\mathscr{U}^{4}_{k,v}(q_{1}^{-1}x)}{\mathscr{U}^{4}_{k,v}(q_{23}^{-1}x)}\frac{\mathscr{U}^{4}_{k,v}(q_{2}^{-1}x)}{\mathscr{U}^{4}_{k,v}(q_{13}^{-1}x)}\frac{\mathscr{U}^{4}_{k,v}(q_{3}^{-1}x)}{\mathscr{U}^{4}_{k,v}(q_{12}^{-1}x)}\frac{1}{\mathscr{U}^{4}_{k,v}(x)}\mathcal{Z}^{\D2}_{4}[k;q_{1},q_{2}]\right]=0,\\
&\underset{x=\chi_{4,v}(\Abox)}{\Res}\left[\mathscr{U}_{k+\Abox}^{4\,\vee}(q_{4}x)^{-1}\mathcal{Z}^{\D2}_{4}[k+\Bbox\,;q_{1},q_{2}]\right.\\
&\qquad \left.+\frac{\mathscr{V}_{4}(q_{1})\mathscr{V}_{4}(q_{2})}{\mathscr{V}_{4}(q_{12})}\frac{\mathscr{U}^{4\,\vee}_{k,v}(q_{1}^{-1}x)}{\mathscr{U}^{4\,\vee}_{k,v}(q_{23}^{-1}x)}\frac{\mathscr{U}^{4\,\vee}_{k,v}(q_{2}^{-1}x)}{\mathscr{U}^{4\,\vee}_{k,v}(q_{13}^{-1}x)}\frac{\mathscr{U}^{4\,\vee}_{k,v}(q_{3}^{-1}x)}{\mathscr{U}^{4\,\vee}_{k,v}(q_{12}^{-1}x)}\frac{1}{\mathscr{U}^{4\,\vee}_{k,v}(x)}\mathcal{Z}^{\D2}_{4}[k;q_{1},q_{2}]\right]=0
\eea
Including the topological term and taking the instanton summation, the following function is pole-free in $x=\chi_{4,v}(\cube)$:
\bea
&\left\langle \widehat{\mathscr{U}}^{4}(q_{4}x)^{-1} \right\rangle_{4}+\mathfrak{q}\frac{\mathscr{V}_{4}(q_{3})}{\mathscr{V}_{4}(q_{23})\mathscr{V}_{4}(q_{13})}\left\langle \frac{\widehat{\mathscr{U}}^{4}(q_{1}^{-1}x)\widehat{\mathscr{U}}^{4}(q_{2}^{-1}x)\widehat{\mathscr{U}}^{4}(q_{3}^{-1}x)}{\widehat{\mathscr{U}}^{4}(q_{12}^{-1}x)\widehat{\mathscr{U}}^{4}(q_{13}^{-1}x)\widehat{\mathscr{U}}^{4}(q_{23}^{-1}x)\widehat{\mathscr{U}}^{4}(x)}\right\rangle_{4},\\
&\left\langle \widehat{\mathscr{U}}^{4\vee}(q_{4}x)^{-1} \right\rangle_{4}+\mathfrak{q}\frac{\mathscr{V}_{4}(q_{2})\mathscr{V}_{4}(q_{1})}{\mathscr{V}_{4}(q_{12})}\left\langle \frac{\widehat{\mathscr{U}}^{4\vee}(q_{1}^{-1}x)\widehat{\mathscr{U}}^{4\vee}(q_{2}^{-1}x)\widehat{\mathscr{U}}^{4\vee}(q_{3}^{-1}x)}{\widehat{\mathscr{U}}^{4\vee}(q_{12}^{-1}x)\widehat{\mathscr{U}}^{4\vee}(q_{13}^{-1}x)\widehat{\mathscr{U}}^{4\vee}(q_{23}^{-1}x)\widehat{\mathscr{U}}^{4\vee}(x)}\right\rangle_{4},
\eea
where 
\bea
\left\langle \mathcal{O}   \right\rangle_{4}=\frac{1}{\mathcal{Z}_{4}}\sum_{k\geq 0}\mathfrak{q}^{k}\widetilde{\mathcal{Z}}^{\D2}_{4}[k\,;q_{1},q_{2}]\mathcal{O}_{k,v},\quad \mathcal{Z}_{4}=\sum_{k}\mathfrak{q}^{k}\widetilde{\mathcal{Z}}^{\D2}_{4}[k;q_{1},q_{2}].
\eea
Hence, using 
\bea
\frac{\mathscr{V}_{4}(q_{3})}{\mathscr{V}_{4}(q_{23})\mathscr{V}_{4}(q_{13})}=-\prod_{i=1}^{3}\frac{1-q_{4}^{-1}q_{i}^{-1}}{1-q_{i}},\quad \frac{\mathscr{V}_{4}(q_{1})\mathscr{V}_{4}(q_{2})}{\mathscr{V}_{4}(q_{12})}=-q_{4}\prod_{i=1}^{3}\frac{1-q_{4}^{-1}q_{i}^{-1}}{1-q_{i}}
\eea
the iWeyl reflection is analogously defined as 
\bea\label{eq:D6iWeylpartition}
\widehat{\mathscr{U}}^{4}(q_{4}x)^{-1}&\longmapsto -\mathfrak{q}\prod_{i=1}^{3}\frac{1-q_{4}^{-1}q_{i}^{-1}}{1-q_{i}}\widehat{\mathscr{U}}^{4}(q_{4}x)^{-1}\widehat{\mathscr{A}}(x)^{-1},\\
\widehat{\mathscr{U}}^{4\vee}(q_{4}x)^{-1}&\longmapsto -q_{4}\mathfrak{q}\prod_{i=1}^{3}\frac{1-q_{4}^{-1}q_{i}^{-1}}{1-q_{i}}\widehat{\mathscr{U}}^{4\vee}(q_{4}x)^{-1}{\widehat{\mathscr{A}}}^{\,\vee}(x)^{-1}
\eea
where 
\bea
\widehat{\mathscr{A}}(x)&=\frac{\widehat{\mathscr{U}}^{4}(q_{12}^{-1}x)\widehat{\mathscr{U}}^{4}(q_{13}^{-1}x)\widehat{\mathscr{U}}^{4}(q_{23}^{-1}x)\widehat{\mathscr{U}}^{4}(x)}{\widehat{\mathscr{U}}^{4}(q_{1}^{-1}x)\widehat{\mathscr{U}}^{4}(q_{2}^{-1}x)\widehat{\mathscr{U}}^{4}(q_{3}^{-1}x)\widehat{\mathscr{U}}^{4}(q_{4}x)},\\{\widehat{\mathscr{A}}}^{\,\vee}(x)&=\frac{\widehat{\mathscr{U}}^{4\vee}(q_{12}^{-1}x)\widehat{\mathscr{U}}^{4\vee}(q_{13}^{-1}x)\widehat{\mathscr{U}}^{4\vee}(q_{23}^{-1}x)\widehat{\mathscr{U}}^{4\vee}(x)}{\widehat{\mathscr{U}}^{4\vee}(q_{1}^{-1}x)\widehat{\mathscr{U}}^{4\vee}(q_{2}^{-1}x)\widehat{\mathscr{U}}^{4\vee}(q_{3}^{-1}x)\widehat{\mathscr{U}}^{4\vee}(q_{4}x)}
\eea
Using the iWeyl reflection recursively, we can construct the $qq$-characters as
\bea
\widehat{\mathscr{T}}^{4}(x)&=\widehat{\mathscr{U}}^{4}(q_{4}x)^{-1}+\mathfrak{q}\left(-\prod_{i=1}^{3}\frac{1-q_{4}^{-1}q_{i}^{-1}}{1-q_{i}}\right)\widehat{\mathscr{U}}^{4}(q_{4}x)^{-1}\widehat{\mathscr{A}}(x)^{-1}+\cdots,\\
\widehat{\mathscr{T}}^{4\vee}(x)&=\widehat{\mathscr{U}}^{4\vee}(q_{4}x)^{-1}+\mathfrak{q}\left(-q_{4}\prod_{i=1}^{3}\frac{1-q_{4}^{-1}q_{i}^{-1}}{1-q_{i}}\right)\widehat{\mathscr{U}}^{4\vee}(q_{4}x)^{-1}{\widehat{\mathscr{A}}}^{\,\vee}(x)^{-1}+\cdots,\\
\eea
Generally, the complete formula will be 
\bea
\widehat{\mathscr{T}}^{4}(x)&=\sum_{\pi\in\mathcal{PP}}(\mathfrak{q}q_{4}^{-1})^{|\pi|}\widetilde{\mathcal{Z}}^{\D6}_{123}[\pi]\widehat{\mathscr{U}}^{4}(q_{4}x)^{-1}\prod_{\scube\in \pi}\widehat{\mathscr{A}}(\chi_{\bar{4},x}(\cube))^{-1},\\
\widehat{\mathscr{T}}^{4\vee}(x)&=\sum_{\pi\in\mathcal{PP}}\mathfrak{q}^{|\pi|}\widetilde{\mathcal{Z}}^{\D6}_{123}[\pi]\widehat{\mathscr{U}}^{4\vee}(q_{4}x)^{-1}\prod_{\scube\in \pi}{\widehat{\mathscr{A}}}^{\,\vee}(\chi_{\bar{4},x}(\cube))^{-1}.
\eea
A derivation of this will be done using the operator formalism in section \ref{sec:D6qqcharacter}.

We can do the same analysis for other 3d theories and obtain the following statement.
\begin{proposition}
    The $qq$-characters with respect with the 3d theory on $\mathbb{C}_{a}\times \mathbb{S}^{1}$ where $a\in\four$ are
    \bea\label{eq:D6qqchpartition}
    \widehat{\mathscr{T}}^{a}(x)&=\sum_{\pi\in\mathcal{PP}}(\mathfrak{q}q_{a}^{-1})^{|\pi|}\widetilde{\mathcal{Z}}^{\D6}_{\bar{a}}[\pi]\widehat{\mathscr{U}}^{a}(q_{a}x)^{-1}\prod_{\scube\in \pi}\widehat{\mathscr{A}}(\chi_{\bar{a},x}(\cube))^{-1},\\
\widehat{\mathscr{T}}^{a\vee}(x)&=\sum_{\pi\in\mathcal{PP}}\mathfrak{q}^{|\pi|}\widetilde{\mathcal{Z}}^{\D6}_{\bar{a}}[\pi]\widehat{\mathscr{U}}^{a\vee}(q_{a}x)^{-1}\prod_{\scube\in \pi}{\widehat{\mathscr{A}}}^{\,\vee}(\chi_{\bar{a},x}(\cube))^{-1},
    \eea
    where the iWeyl reflection is defined as
    \bea\label{eq:D6iWeylgeneral}
\widehat{\mathscr{U}}^{a}(q_{a}x)^{-1}&\longmapsto -\mathfrak{q}\prod_{i=1}^{3}\frac{1-q_{a}^{-1}q_{i}^{-1}}{1-q_{i}}\widehat{\mathscr{U}}^{a}(q_{a}x)^{-1}\widehat{\mathscr{A}}(x)^{-1},\\
\widehat{\mathscr{U}}^{a\vee}(q_{a}x)^{-1}&\longmapsto -q_{a}\mathfrak{q}\prod_{i=1}^{3}\frac{1-q_{a}^{-1}q_{i}^{-1}}{1-q_{i}}\widehat{\mathscr{U}}^{a\vee}(q_{a}x)^{-1}{\widehat{\mathscr{A}}}^{\,\vee}(x)^{-1}
\eea
where 
\bea
\widehat{\mathscr{A}}(x)&=\frac{\widehat{\mathscr{U}}^{a}(x)}{\widehat{\mathscr{U}}^{a}(q_{a}x)}\prod_{i=1}^{3}\frac{\widehat{\mathscr{U}}^{a}(q_{a}q_{i}x)}{\widehat{\mathscr{U}}^{a}(q_{i}^{-1}x)},\quad 
{\widehat{\mathscr{A}}}^{\,\vee}(x)=\frac{\widehat{\mathscr{U}}^{a\vee}(x)}{\widehat{\mathscr{U}}^{a\vee}(q_{a}x)}\prod_{i=1}^{3}\frac{\widehat{\mathscr{U}}^{a\vee}(q_{a}q_{i}x)}{\widehat{\mathscr{U}}^{a\vee}(q_{i}^{-1}x)}.
\eea
\end{proposition}
We call these $qq$-characters the $\D6$-brane $qq$-characters, shortly D6 $qq$-characters. Namely, the $qq$-character with respect to the $\D2_{a}$ theory on $\mathbb{C}_{a}\times \mathbb{S}^{1}$ is the $\D6_{\bar{a}}$ $qq$-character.

\paragraph{D6 $qq$-characters as codimension two defects}
We can switch the roles of the D-branes for the 7d theory case and consider the $\D2$ theory as the bulk theory and the $\D6$ theory as the defect theory (see Figure \ref{fig:D2D6defectqq} and Table \ref{t:D2D6defectbraneweb}). This setup gives rise to a codimension two defect, a line defect in the 3d theory. The $\D6_{\bar{a}}$-brane will then give the $qq$-character of the $\D2_{a}$-theory. Let us focus on the case $a=4$ and $\bfN_{123}^{\vee}=x$. The contour integral formula comes from
\bea
&\bfN_{123}^{\vee}\frac{\bfP_{4}^{\vee}\bfP_{123}}{\bfP_{\four}}\bfN_{4}-\bfP_{4}^{\vee}\bfN_{123}^{\vee}\bfK-\bfP_{123}^{\vee}\bfN_{4}^{\vee}\bfK+\bfP_{123}^{\vee}\bfK_{123}^{\vee}\bfK_{123}+\bfP_{124}^{\vee}\bfK_{4}^{\vee}\bfK_{4}+\bfP_{\four}\bfK_{123}^{\vee}\bfK_{4}\\
=&-q_{123}\bfN^{\vee}_{123}\left(\bfN_{4}-\bfP_{4}\bfK\right)-\bfP_{123}^{\vee}\bfN_{4}^{\vee}\bfK+\bfP_{123}^{\vee}\bfK_{123}^{\vee}\bfK_{123}+\bfP_{124}^{\vee}\bfK_{4}^{\vee}\bfK_{4}+\bfP_{\four}\bfK_{123}^{\vee}\bfK_{4},
\eea
Similar to the $\D2$ case, this will come from a gauge origami system where both $\D2$ and $\D6$-branes are included. The contour integral is written as
\bea
\frac{\mathcal{G}_{\bar{1}}^{k}}{k!}\oint\prod_{I=1}^{k}\frac{dx_{I}}{2\pi\iota x_{I}}\textcolor{red}{\prod_{\alpha=1}^{n}\frac{1}{(1-q_{4}x/v_{\alpha})}\prod_{I=1}^{k}\mathscr{V}_{4}\left(\frac{x}{x_{I}}\right)}\textcolor{blue}{\prod_{\alpha=1}^{n}\prod_{I=1}^{k}g_{\bar{4}}\left(\frac{v_{\alpha}}{x_{I}}\right)\prod_{I< J}\mathcal{A}_{\mathbb{C}^{4}}\left(\frac{x_{I}}{x_{J}}\right)^{-1}}.
\eea
The red term comes from the transverse $\D6$-brane, while the blue term comes from the bulk $\D2$-brane. After localization, the red term gives the $\mathscr{U}$-functions. Similar to other cases, integrating out the $\D6_{\bar{4}}$-brane leads to a codimension two defect operator $\widehat{\mathscr{T}}^{4}(x),\widehat{\mathscr{T}}^{4\,\vee}(x)$ and the expectation value of it with respect to the bulk $\D2_{4}$-theory gives the partition function of the gauge origami system where there are $\D2_{4}$ and $\D6_{123}$-branes spanning transversely.

\section{Free field realizations of contour integral formulas}\label{sec:freefieldintegral}
In the previous section, we introduced the contour integral formulas and non-perturbative partition functions of the magnificent four \eqref{eq:D8integral}, tetrahedron instantons \eqref{eq:D6integral}, spiked instantons \eqref{eq:D6integral}, and the coupled vortex system \eqref{eq:D2integral}. Since we expect BPS/CFT correspondence of these systems, the contour integrals themselves should be related to vertex operators. In this section, we determine the explicit forms of the vertex operators and show that their compositions give free field realizations of the contour integral formulas. The main strategy is to relate D-branes in the physical setup with vertex operators. We discuss the operator formalism of the magnificent four in section~\ref{sec:M4LMNS}, the tetrahedron instanton in section~\ref{sec:tetraLMNS}, the spiked instanton in section~\ref{sec:spikedLMNS}, and the coupled vortex system in section~\ref{sec:cplvortLMNS}. Zero modes of the vertex operators are discussed in section~\ref{sec:zeromodes}. In section \ref{sec:quiverstructure}, we discuss the relation with graded quivers and that the vertex operators are defined from the quiver $q$-Cartan matrices. We also generalize the discussion to toric Calabi--Yau four-folds and give some conjectures. 

The main statement is as follows.
\begin{theorem}\label{thm:freefieldconclusion}
    For each D-brane ($\D0,\D2,\D4,\D6,\D8$), we can define the corresponding vertex operators as
    \begin{align}
    \renewcommand\arraystretch{1.2}{
        \begin{tabular}{|c|c|c|c|}\hline
            D-brane & space-time & vertex operator & reference\\
           \hline\hline  D0-brane  & $\mathbb{S}^{1}$ & $\mathsf{A}(x)$ & \eqref{eq:D0op}\\
           \hline D2-brane   &  $\mathbb{C}_{a}\times \mathbb{S}^{1}$ ($a\in\four$)  & $\mathsf{S}_{a}(x)$&\eqref{eq:D2op}\\
           \hline D4-brane &  $\mathbb{C}^{2}_{A}\times \mathbb{S}^{1}$ ($A\in\six$) & $\mathsf{X}_{A}(x)$&\eqref{eq:D4op}\\
           \hline D6-brane &  $\mathbb{C}^{3}_{\bar{a}}\times \mathbb{S}^{1}$ ($a\in\four$) &  $\mathsf{W}_{\bar{a}}(x)$& \eqref{eq:D6op}\\
           \hline D8-brane & $\mathbb{C}^{4}\times \mathbb{S}^{1}$  & $\mathsf{Z}(x)$& \eqref{eq:D8op} \\\hline
         \end{tabular}}
    \end{align}
    In the contour integral formula, the D0-branes giving instanton contributions arise from\footnote{In this paper, when we write $\mathsf{V}(x)^{-1}$ for a vertex operator $\mathsf{V}(x)$, we are meaning $:\mathsf{V}(x)^{-1}:$. The normal ordering of them is implicitly imposed.} $\mathsf{A}^{-1}(x)$, while other D-branes arise from $\mathsf{S}_{a}(x),\mathsf{X}_{A}(x),\mathsf{W}_{\bar{a}}(x),\mathsf{Z}(x)$. To include anti D-branes, we need to reverse the power of the operators as $\mathsf{S}_{a}(x)^{-1},\mathsf{X}_{A}(x)^{-1},\mathsf{W}_{\bar{a}}(x)^{-1},\mathsf{Z}(x)^{-1}$. 
\end{theorem}

\subsection{Magnificent four}\label{sec:M4LMNS}
Let us start with the magnificent four system. Actually, the vertex operator of this system was already introduced in \cite{Kimura:2022zsm} in the context of double quiver gauge theory. See also \cite{Kimura:2019hnw}.

We first introduce the vertex operator corresponding to the $\D0$-brane which represents the instantons and vortices in the physical setup (see section \ref{sec:physicalsetup}). In the algebraic context, it is called the \emph{root current}:
\beq
    \mathsf{A}(x)=\mathsf{a}_{0}(x):\exp\left(\sum_{n\neq 0}\mathsf{a}_{n}x^{-n}\right):,\quad [\mathsf{a}_{n},\mathsf{a}_{m}]=-\frac{1}{n}\bfP_{\four}^{[n]}\delta_{n+m,0},\label{eq:D0op}
\eeq
where $\mathsf{a}_{0}(x)$ is the zero mode that will be determined later. The right-hand side of the commutation relation of the root current represents the CY$_{4}$ geometry $\mathbb{C}^{4}$. In the context of quiver W-algebra \cite{Kimura:2015rgi} and double quiver gauge theory \cite{Kimura:2022zsm}, this is just the root current of the affine quiver W-algebra which is denoted by the $(\widehat{A}_{0},\widehat{A}_{0})$ theory in the terminology of \cite{Kimura:2022zsm}. The operator product of the $\mathsf{A}$-operator is
\begin{equation}
    \mathsf{A}(x)\mathsf{A}(x')=\wick{\c{\mathsf{a}_{0}(x)}\c{\mathsf{a}_{0}(x')}}\mathcal{A}_{\mathbb{C}^{4}}\left(\frac{x'}{x}\right)^{-1}:\mathsf{A}(x)\mathsf{A}(x'):,\label{eq:D0ope}
\end{equation}
where we used the convention
\begin{equation}
   \mathsf{a}_{0}(x)\mathsf{a}_{0}(x')=\wick{\c{\mathsf{a}_{0}(x)}\c{\mathsf{a}_{0}(x')}}:\mathsf{a}_{0}(x)\mathsf{a}_{0}(x'):.
\end{equation}
For the zero modes $\mathsf{a}_{0}(z)$, we impose the condition that the OPE factor will be the same rational function. Namely, using \eqref{eq:reflec_structfunc}, we impose 
\begin{equation}
    \wick{\c{\mathsf{a}_{0}(x)}\c{\mathsf{a}_{0}(x')}}=\wick{\c{\mathsf{a}_{0}(x')}\c{\mathsf{a}_{0}(x)}}\label{eq:D0D0zero}
\end{equation}
and then we have the OPE factor symmetric in $x$ and $x'$. We set $\wick{\c{\mathsf{a}_{0}(x)}\c{\mathsf{a}_{0}(x')}}=1$ (see section~\ref{sec:zeromodes} for the derivation and explicit form of the zero modes).

To discuss the operator formalism of the magnificent four system, we need to introduce a vertex operator corresponding to the $\D8$-brane (and $\overbar{\D8}$).
We denote this vertex operator as $\mathsf{Z}(x)$:
\begin{equation}
    \mathsf{Z}(x)=\mathsf{z}_{0}(x):\exp\left(\sum_{n\neq 0}\mathsf{z}_{n}x^{-n}\right):,\quad [\mathsf{z}_{n},\mathsf{z}_{m}]=-\frac{1}{n}\frac{1}{\bfP_{\four}^{[n]}}\delta_{n+m,0},\label{eq:D8op}
\end{equation}
where the relation with the root current is
\begin{equation}
    [\mathsf{a}_{n},\mathsf{z}_{m}]=-\frac{1}{n}\delta_{n+m,0},\quad \mathsf{z}_{n}=\frac{\mathsf{a}_{n}}{\bfP_{\four}^{[n]}}.
\end{equation}
Explicitly, we have
\bea
    \mathsf{Z}(x)\mathsf{A}(x')&=\wick{\c{\mathsf{z}_{0}(x)}\c{\mathsf{a}_{0}(x')}}\left(1-x'/x\right):\mathsf{Z}(x)\mathsf{A}(x'):\\
    \mathsf{A}(x')\mathsf{Z}(x)&=\wick{\c{\mathsf{a}_{0}(x')}\c{\mathsf{z}_{0}(x)}}\left(1-x/x'\right):\mathsf{A}(x')\mathsf{Z}(x):
\eea
Since, in the magnificent four, the $\D8$-brane appears with the $\overbar{\D8}$-brane, we introduce a vertex operator corresponding to the brane-antibrane coupled system as 
\begin{equation}
    \widetilde{\mathsf{Z}}^{K}(x)\coloneqq{:\frac{\mathsf{Z}(x)}{\mathsf{Z}(K x)}:}=\tilde{\mathsf{z}}^{K}_{0}(x):\exp\left(\sum_{n\neq 0}\tilde{\mathsf{z}}^{K}_{n}x^{-n}\right):,\quad \tilde{\mathsf{z}}^{K}_{n}=(1-K^{-n})\mathsf{z}_{n}=\frac{1-K^{-n}}{\bfP_{\four}^{[n]}}\mathsf{a}_{n},\label{eq:D8D8barop}
\end{equation}
where $K\in\mathbb{C}^{\times}$ is a generic parameter. This parameter is the parameter introduced in \eqref{eq:antiD8parameter} and corresponds to the distance between the $\D8$ and $\overbar{\D8}$ branes physically. This gives 
\bea
    \widetilde{\mathsf{Z}}^{K}(x)\mathsf{A}(x')&=\wick{\c{\widetilde{\mathsf{z}}^{K}_{0}(x)}\c{\mathsf{a}_{0}(x')}}\frac{1-x'/x}{1-K^{-1}x'/x}:\widetilde{\mathsf{Z}}^{K}(x)\mathsf{A}(x'):,\\
    \mathsf{A}(x')\widetilde{\mathsf{Z}}^{K}(x)&=\wick{\c{\mathsf{a}_{0}(x')}\c{\widetilde{\mathsf{z}}^{K}_{0}(x)}}\frac{1-x/x'}{1-Kx/x'}:\widetilde{\mathsf{Z}}^{K}(x)\mathsf{A}(x'):
\eea
For the zero modes, we impose that the contraction with the root current will be the same rational function after analytic continuation:
\begin{equation}
    \wick{\c{\tilde{\mathsf{z}}^{K}_{0}(x)}\c{\mathsf{a}_{0}(x')}}=K^{-1}\wick{\c{\mathsf{a}_{0}(x')}\c{\tilde{\mathsf{z}}^{K}_{0}(x)}}.\label{eq:D0D8zero}
\end{equation}
We impose $\wick{\c{\mathsf{a}_{0}(x')}\c{\tilde{\mathsf{z}}^{K}_{0}(x)}}=1$ (see section \ref{sec:zeromodes}). Note that we are relaxing the conditions and only imposing conditions on the zero modes of the brane anti-brane coupled vertex operator $\widetilde{\mathsf{Z}}^{K}$ but not the $\mathsf{Z}$-operator. 

\begin{proposition}
The contour integral formula of the magnificent four system shown in \eqref{eq:D8integral} is given as
\begin{equation}\label{eq:D8op_integral}
    \mathfrak{q}^{k}\mathcal{Z}_{k}^{\D8}=\frac{\mathfrak{q}^{k}\mathcal{G}_{\bar{a}}^{k}}{k!}\oint \prod_{I=1}^{k}\frac{dx_{I}}{2\pi\iota x_{I}}\langle \mathsf{A}_{k}^{-1}\widetilde{\mathsf{Z}}^{\underline{K}}_{\underline{n}}\rangle
\end{equation}
where we used
\begin{equation}
\begin{split}
    &\mathsf{A}_{k}^{-1}=\prod_{I=1}^{k}\mathsf{A}(x_{I})^{-1},\quad \widetilde{\mathsf{Z}}^{\underline{K}}_{\underline{n}}=:\prod_{\alpha=1}^{n}\widetilde{\mathsf{Z}}^{K_{\alpha}}(v_{\alpha}):,\quad \langle \mathsf{A}_{k}^{-1}\widetilde{\mathsf{Z}}^{\underline{K}}_{\underline{n}}\rangle=\prod_{\alpha=1}^{n}\prod_{I=1}^{k}\frac{1-K_{\alpha}v_{\alpha}/x_{I}}{1-v_{\alpha}/x_{I}}\prod_{I<J}\mathcal{A}_{\mathbb{C}^{4}}\left(\frac{x_{I}}{x_{J}}\right)^{-1}
\end{split}
\end{equation}
and $\langle \mathcal{O}\rangle=\bra{0}\mathcal{O}\ket{0}$.     
\end{proposition}
We can generalize the correlation function by including an extra parameter $p=e^{2\pi\iota\tau}$ as $\tr(p^{L_{0}}\mathcal{O})$ which will give an elliptic deformation (a torus correlator). $L_0$ is the degree counting operator. Taking the limit $p\rightarrow 0$ from this elliptic deformed correlator gives the vacuum expectation value: $\tr(p^{L_{0}}\mathcal{O})\xrightarrow{p\rightarrow \infty}\bra{0}\mathcal{O}\ket{0}$. Under this elliptic deformation, the arising algebra will be quiver elliptic W-algebra and will be discussed in section \ref{sec:ellipticWorigami}. For the trigonometric case, we simply use the vacuum expectation value of the operators.    

Note that we need to determine an order in the operators when doing explicit computations. In most of the cases, we simply assume $\prod_{I=1}^{k}\mathcal{O}(x_{I})=\mathcal{O}(x_{k})\cdots\mathcal{O}(x_{2}) \mathcal{O}(x_{1})$ for an operator $\mathcal{O}(x)$. Actually, because of the zero modes conditions \eqref{eq:D0D0zero} and \eqref{eq:D0D8zero}, the order of the operators is not relevant after analytic continuation. Note also that the OPE of the $\widetilde{\mathsf{Z}}^{K}(x)$ operators give the perturbative factor introduced in \eqref{eq:D8oneloop}:
\begin{equation}
    \widetilde{\mathsf{Z}}^{K_{n}}(v_{n})\cdots \widetilde{\mathsf{Z}}^{K_{1}}(v_{1})=\prod_{\beta>\alpha}\mathcal{Z}_{\text{1-loop}}^{\D8\tbar\D8}(v_{\alpha},K_{\alpha}\,|\,v_{\beta},K_{\beta}):\prod_{\alpha=1}^{n}\widetilde{\mathsf{Z}}^{K_{\alpha}}(v_{\alpha}):.
\end{equation}

\subsection{Tetrahedron instanton}\label{sec:tetraLMNS}
The tetrahedron system is obtained by adding $\D6$-branes and $\D0$-branes. To reproduce the free field realization of the integral formula in \eqref{eq:D6integral}, we introduce the following vertex operator that represents the $\D6$-brane wrapping $\mathbb{C}^{3}_{\bar{a}}\times \mathbb{S}^{1}$:
\begin{equation}
\mathsf{W}_{\bar{a}}(x)=\mathsf{w}_{\bar{a},0}(x):\exp\left(\sum_{n\neq 0}\mathsf{w}_{\bar{a},n}x^{-n}\right):
,\quad[\mathsf{w}_{\bar{a},n},\mathsf{w}_{\bar{b},m}]=-\frac{1}{n}\frac{\bfP_{\four}^{[n]}}{\bfP_{\bar{a}}^{[-n]}\bfP_{\bar{b}}^{[n]}}\delta_{n+m,0},\label{eq:D6op}
\end{equation}
where $a\in\four$ and $\mathsf{w}_{\bar{a},0}(x)$ are zero modes. The relation with the root current is
\begin{equation}
[\mathsf{a}_{n},\mathsf{w}_{\bar{a},m}]=-\frac{1}{n}\bfP_{a}^{[n]}\delta_{n+m,0},\quad \mathsf{w}_{\bar{a},n}=\frac{\mathsf{a}_{n}}{\bfP_{\bar{a}}^{[-n]}},
\end{equation}
which gives the operator product
\begin{equation}
\begin{split}
    \mathsf{A}(x)\mathsf{W}_{\bar{a}}(x')&=\wick{\c{\mathsf{a}_{0}(x)}\c{\mathsf{w}_{\bar{a},0}(x')}}\mathscr{V}_{a}(x'/x)^{-1}:\mathsf{A}(x)\mathsf{W}_{\bar{a}}(x'):\\
    \mathsf{W}_{\bar{a}}(x')\mathsf{A}(x)&=\wick{\c{\mathsf{w}_{\bar{a},0}(x')}\c{\mathsf{a}_{0}(x)}}\mathscr{V}_{a}(q_{a}^{-1}x/x'):\mathsf{W}_{\bar{a}}(x')\mathsf{A}(x):.
\end{split}
\end{equation}
Using \eqref{eq:reflec_structfunc}, we impose the zero mode conditions as
\begin{equation}
    \wick{\c{\mathsf{a}_{0}(x)}\c{\mathsf{w}_{\bar{a},0}(x')}}=q_{a}\wick{\c{\mathsf{w}_{\bar{a},0}(x')}\c{\mathsf{a}_{0}(x)}}.\label{eq:D0D6zero}
\end{equation}
Explicitly, we impose $\wick{\c{\mathsf{a}_{0}(x)}\c{\mathsf{w}_{\bar{a},0}(x')}}=1$ (see section \ref{sec:zeromodes}). 
The tetrahedron instanton system is then given as follows.
\begin{proposition}
The tetrahedron instanton partition function~\eqref{eq:D6integral} is equivalent to the following vertex operator correlation function after analytic continuation,
\begin{equation}\label{eq:D6op_integral}
    \mathfrak{q}^{\underline{k}}\mathcal{Z}_{\underline{k}}^{\D6}=\frac{\mathfrak{q}^{\underline{k}}\underline{\mathcal{G}}^{\underline{k}}}{\underline{k}!}\oint\prod_{a\in\four}\prod_{I=1}^{k_{\bar{a}}}\frac{dx_{\bar{a},I}}{2\pi\iota x_{\bar{a},I}}\langle \mathsf{A}_{\underline{k}}^{-1}\mathsf{W}_{\underline{n}}\rangle
\end{equation}
where
\begin{subequations}
\begin{align}
&\mathfrak{q}^{\underline{k}}=\prod_{a\in\four}\mathfrak{q}^{k_{\bar{a}}},\quad\mathsf{A}_{\underline{k}}=\prod_{a\in\four}\prod_{I=1}^{k_{\bar{a}}}\mathsf{A}(x_{\bar{a},I}),\quad \mathsf{W}_{\underline{n}}={:\prod_{a\in\four}\prod_{\alpha=1}^{n_{\bar{a}}}\mathsf{W}_{\bar{a}}(v_{\bar{a},\alpha}):},\\
&\langle\mathsf{A}_{\underline{k}}^{-1}\mathsf{W}_{\underline{n}}\rangle=\prod_{a,b\in\four}\prod_{\alpha=1}^{n_{\bar{a}}}\prod_{I=1}^{k_{\bar{b}}}\mathscr{V}_{a}\left(\frac{v_{\bar{a},\alpha}}{x_{\bar{b},I}}\right)\prod_{(a,I)<(b,J)}\mathcal{A}_{\mathbb{C}^{4}}\left(\frac{x_{\bar{a},I}}{x_{\bar{b},J}}\right)^{-1}
\end{align}
\end{subequations}
\end{proposition}
Note that the order of the operator product in $\mathsf{A}_{\bar{k}}$ does not change the result. The OPE of the $\mathsf{W}$-operators give the one loop perturbative part in \eqref{eq:D6oneloop}:
\begin{equation}
    \mathsf{W}_{\bar{b}}(v_{\bar{b},\beta})\mathsf{W}_{\bar{a}}(v_{\bar{a},\alpha})=\mathcal{Z}_{\text{1-loop}}^{\D6\tbar\D6}(v_{\bar{a},\alpha},\bar{a}\,|\,v_{\bar{b},\beta},\bar{b}):\mathsf{W}_{\bar{b}}(v_{\bar{b},\beta})\mathsf{W}_{\bar{a}}(v_{\bar{a},\alpha}):.
\end{equation}

\paragraph{Relation with magnificent four}
Let us focus on the 7d $\U(1)$ theory on $\mathbb{C}^{3}_{\bar{a}}\times \mathbb{S}^{1}$. Starting from the 9d $\U(1|1)$ theory of the magnificent four and tuning the parameter $K=q_{a}$, we have 
\begin{equation}\label{eq:D8D6reduction}
    \widetilde{\mathsf{Z}}^{q_{a}}(x)={:\frac{\mathsf{Z}(x)}{\mathsf{Z}(q_{a}x)}:}\simeq\mathsf{W}_{\bar{a}}(x),
\end{equation}
where the equality $\simeq$ is up to extra zero modes depending on the explicit form. In our notation in section \ref{sec:zeromodes}, this becomes an exact identity. Generally, starting from a 9d $\U(n|n)$ magnificent four theory with parameters $(K_{\alpha})_{\alpha=1}^{n}$, $n=\sum_{a\in\four}n_{\bar{a}}$ and tuning 
\begin{equation}
   (K_{\alpha})_{\alpha=1}^{n}\longrightarrow (K_{\bar{a},\alpha})_{\alpha=1}^{n_{\bar{a}}},\quad K_{\bar{a},\alpha}=q_{a}
\end{equation}
we have
\begin{equation}
    \langle \mathsf{A}_{k}^{-1}\widetilde{\mathsf{Z}}^{\underline{K}}_{\underline{n}}\rangle=\langle \mathsf{A}_{\underline{k}}^{-1}\mathsf{W}_{\underline{n}}\rangle,
\end{equation}
and thus, we obtain the tetrahedron instanton system. Note also that setting $K=q_{a}^{-1}$ gives 
\bea
\widetilde{\mathsf{Z}}^{q_{a}^{-1}}(x)={:\frac{\mathsf{Z}(x)}{\mathsf{Z}(q_{a}^{-1}x)}:}=\mathsf{W}_{\bar{a}}(q_{a}^{-1}x)^{-1}
\eea
which allows $\mathsf{W}_{\bar{a}}(x)$ to appear in the denominator. 

Physically, this property suggests that the $\D8$-branes and anti-$\D8$-branes annihilate in a specific distance under the $\Omega$-background and eventually reproduce the, generally intersecting, $\D6$-branes system \cite{Nekrasov:2017cih,Nekrasov:2018xsb,Pomoni:2021hkn,Pomoni:2023nlf}, which is also interpreted as a tachyon condensation~\cite{Sen:1998sm}.

\paragraph{Supergroup generalization}
Following the construction of the magnificent four system in~\eqref{eq:D8D8barop}, we can write down the contour integral formula with D6 operators appearing in the denominator:
\beq\label{eq:D6supergroupdef}
{:\frac{\mathsf{W}_{\bar{a}}(v_{1})\cdots \mathsf{W}_{\bar{a}}(v_{n})}{\mathsf{W}_{\bar{a}}(u_{1})\cdots \mathsf{W}_{\bar{a}}(u_{m})}:}
\eeq
The contour integral formula is proportional to
\beq\label{eq:D6supergroupLMNS1}
\oint \prod_{a\in\four}\prod_{I=1}^{k_{\bar{a}}}\frac{dx_{\bar{a},I}}{2\pi\iota x_{\bar{a},I}}\langle\mathsf{A}_{\underline{k}}^{-1}\mathsf{W}_{\underline{n}|\underline{m}}\rangle
\eeq
where 
\begin{subequations}\label{eq:D6supergroupLMNS2}
\begin{align}
\mathsf{A}_{\underline{k}}&=\prod_{a\in\four}\prod_{I=1}^{k_{\bar{a}}}\mathsf{A}(x_{\bar{a},I}),\quad \mathsf{W}_{\underline{n}|\underline{m}}={:\prod_{a\in\four}\frac{\prod_{\alpha=1}^{n_{\bar{a}}}\mathsf{W}_{\bar{a}}(v_{\bar{a},\alpha})}{\prod_{\beta=1}^{m_{\bar{a}}}\mathsf{W}_{\bar{a}}(u_{\bar{a},\beta})}:},\\
\langle \mathsf{A}_{\underline{k}}^{-1}\mathsf{W}_{\underline{n}|\underline{m}}\rangle&=\prod_{a,b\in\four}\prod_{\alpha=1}^{n_{\bar{a}}}\prod_{I=1}^{k_{\bar{b}}}\mathscr{V}_{a}\left(\frac{v_{\bar{a},\alpha}}{x_{\bar{b},I}}\right)\prod_{a,b\in\four}\prod_{\beta=1}^{m_{\bar{a}}}\prod_{J=1}^{k_{\bar{b}}}\mathscr{V}_{a}\left(\frac{u_{\bar{a},\alpha}}{x_{\bar{b},J}}\right)^{-1}\prod_{(a,I)<(b,J)}\mathcal{A}_{\mathbb{C}^{4}}\left(\frac{x_{\bar{a},I}}{x_{\bar{b},J}}\right)^{-1}.
\end{align}
\end{subequations}
We expect that the operators in the denominators of \eqref{eq:D6supergroupdef} correspond to $\overbar{\D6}$-branes similar to the situation of the magnificent four system. We leave a detailed analysis of the evaluation of this contour integral formula and its relation with the 7d supergroup gauge theory for future work. We will see in later sections, that after tuning the parameters $\{u_{\bar{a},\beta}\}$ to special values, we can further reduce the system and obtain the contour integral formula of the spiked instanton system.

\subsection{Spiked instanton}\label{sec:spikedLMNS}
Let us next consider the spiked instanton system where $\D4$-branes wrapping $\mathbb{C}^{2}_{A}\times \mathbb{S}^{1},\,(A\in\six)$ appear. We introduce the following vertex operators for $A\in\six$:
\begin{equation}
\begin{split}
    &\mathsf{X}_{A}(x)=\mathsf{x}_{A,0}(x):\exp\left(\sum_{n\neq 0}\mathsf{x}_{A,n}x^{-n}\right):,\quad [\mathsf{x}_{A,n},\mathsf{x}_{B,m}]=-\frac{1}{n}\frac{\bfP_{\four}^{[n]}}{\bfP_{A}^{[-n]}\bfP_{B}^{[n]}}\delta_{n+m,0},
\end{split}\label{eq:D4op}
\end{equation}
where $\mathsf{x}_{A,0}(x)$ is the zero mode. The relation with the $\mathsf{A}$-operator is
\begin{equation}
    \mathsf{x}_{A,n}=\frac{\mathsf{a}_{n}}{\bfP_{A}^{[-n]}},\qquad [\mathsf{a}_{n},\mathsf{x}_{A,m}]=-\frac{1}{n}\bfP_{\bar{A}}^{[n]}\delta_{n+m,0}.
\end{equation}
Explicitly, the contraction formulas are 
\begin{equation}
\begin{split}
    \mathsf{A}(x)\mathsf{X}_{A}(\nu)&=\mathscr{S}_{\bar{A}}(\nu/x)^{-1}\wick{\c{\mathsf{a}_{0}(x)}\c{\mathsf{x}_{A,0}(x')}}:\mathsf{A}(x)\mathsf{X}_{A}(\nu):,\\
    \mathsf{X}_{A}(\nu)\mathsf{A}(x)&=\mathscr{S}_{\bar{A}}(q_{A}x/\nu)^{-1}\wick{\c{\mathsf{x}_{A,0}(x')}\c{\mathsf{a}_{0}(x)}}:\mathsf{X}_{A}(\nu)\mathsf{A}(x):.
\end{split}
\end{equation}
We impose the following condition on the zero modes so that the operator product of the right-hand side is the same rational function after analytic continuation \eqref{eq:reflec_structfunc}:
\begin{equation}
\begin{split}
    \wick{\c{\mathsf{a}_{0}(x)}\c{\mathsf{x}_{A,0}(x')}}=\wick{\c{\mathsf{x}_{A,0}(x')}\c{\mathsf{a}_{0}(x)}}.
\end{split}\label{eq:D0D4zero}
\end{equation}
We will use $\wick{\c{\mathsf{a}_{0}(x)}\c{\mathsf{x}_{A,0}(x')}}=1$ (see section \ref{sec:zeromodes} for explicit forms).

The vertex operators $\mathsf{X}_{A}(x)$ we introduced here are just the $\mathsf{Y}$-operators in \cite{Kimura:2015rgi} up to shift of variables
\begin{equation}
   \mathsf{X}_{A}(x)=\mathsf{Y}_{A}(q_{A}^{-1}x).
\end{equation}
Explicitly, if we consider the affine quiver $\widehat{A}_{0}$ gauge theory on $\mathbb{C}^{2}_{12}\times \mathbb{S}^{1}$, the $\mathsf{Y}$-operator\footnote{The $\mathsf{Y}$-operators are operators representing defects on the bulk gauge theory we are focusing on. For the gauge theory on the $\D4$-brane wrapping $\mathbb{C}^{2}_{12}\times \mathbb{S}^{1}$, the point defect is the transverse $\D4$-brane wrapping $\mathbb{C}^{2}_{34}\times \mathbb{S}^{1}$. Thus, the $\mathsf{Y}$-operator is related with the operator $\mathsf{X}_{34}(x)$. } in~\cite{Kimura:2015rgi} will be described as $\mathsf{X}_{34}(q_{12}^{-1}x)$. 

\begin{proposition}
The contour integral formula~\eqref{eq:D4integral} of the spiked instanton configuration is rewritten as
\begin{equation}\label{eq:D4op_integral}
\mathfrak{q}^{\underline{k}}\mathcal{Z}_{\underline{k}}^{\D4}=\frac{\mathfrak{q}^{\underline{k}}\underline{\mathcal{G}}^{\underline{k}}}{\underline{k}!}\oint\prod_{A\in\six}\prod_{I=1}^{k_{A}}\frac{dx_{A,I}}{2\pi\iota x_{A,I}}\langle\mathsf{A}^{-1}_{\underline{k}}\mathsf{X}_{\underline{n}}\rangle
\end{equation}
where
\begin{subequations}
\begin{align}
&\mathfrak{q}^{\underline{k}}=\prod_{A\in\six}\mathfrak{q}^{k_{A}},\quad \underline{\mathcal{G}}^{\underline{k}}=\prod_{A\in\six}\left(\mathcal{G}_{\text{sup}(\bar{A})}\right)^{k_{A}},\quad 
\mathsf{A}_{\underline{k}}=\prod_{A\in\six}\prod_{I=1}^{k_{A}}\mathsf{A}(x_{A,I}),\quad \mathsf{X}_{\underline{n}}={:\prod_{A\in\six}\prod_{\alpha=1}^{n_{A}}\mathsf{X}_{A}(v_{A,\alpha}):},\\
&\langle\mathsf{A}^{-1}_{\underline{k}}\mathsf{X}_{\underline{n}}\rangle=\prod_{A,B\in\six}\prod_{I=1}^{k_{B}}\prod_{\alpha=1}^{n_{A}}\mathscr{S}_{\bar{A}}\left(\frac{v_{A,\alpha}}{x_{B,I}}\right)\prod_{(A,I)<(B,J)}\mathcal{A}_{\mathbb{C}^{4}}\left(\frac{x_{A,I}}{x_{B,J}}\right)^{-1}.
\end{align}
\end{subequations}
The one-loop perturbative part in \eqref{eq:D4oneloop} is obtained by the OPE of the $\mathsf{X}$ operators:
\begin{equation}
    \mathsf{X}_{B}(v_{B,\beta})\mathsf{X}_{A}(v_{A,\alpha})=\mathcal{Z}^{\D4\tbar\D4}_{\text{1-loop}}(v_{A,\alpha},A\,|\,v_{B,\beta},B):\mathsf{X}_{B}(v_{B,\beta})\mathsf{X}_{A}(v_{A,\alpha}):.
\end{equation}
\end{proposition}

\paragraph{Relation with tetrahedron instanton}
Similar to the situation in \eqref{eq:D8D6reduction}, where the $\D8\tbar\overbar{\D8}$ coupled system is reduced to the tetrahedron system, using \eqref{eq:D6supergroupdef} and specializing the parameters, we can obtain the contour integral formula of the spiked instanton system. This comes from the following relation:
\begin{equation}
    \mathsf{X}_{ab}(x)\simeq {:\frac{\mathsf{W}_{abc}(x)}{\mathsf{W}_{abc}(q_{c}x)}:}
\end{equation}
where $\simeq $ means they are equivalent up to zero modes. In our notation of the zero modes, it will become an exact identity (see \eqref{eq:oprelation1}). For example, we have 
\bea
{:\prod_{\alpha=1}^{n_{12}}\mathsf{X}_{12}(v_{12,\alpha}):}={:\prod_{\alpha=1}^{n_{12}}\frac{\mathsf{W}_{123}(v_{12,\alpha})}{\mathsf{W}_{123}(q_{3}v_{12,\alpha})}:}
\eea
which means by considering the $\D6\tbar \overbar{\D6}$ system spanning $\mathbb{C}^{3}_{123}\times \mathbb{S}^{1}$ and tuning the positions of them with the parameter $q_{3}$, we can obtain the 5d theory on $\mathbb{C}_{12}^{2}\times \mathbb{S}^{1}$.

\paragraph{Supergroup generalization}
As mentioned in section \ref{sec:spiked_partitionfunct}, to obtain supergroup analogs of the gauge origami system of the spiked instanton, we need to include negative D-branes $\D0^{-},\D4^{-}$ to the system. Due to the fact that the corresponding vertex operators of $\D0^{+}$ and $\D4^{+}$ are $\mathsf{A}(x)^{-1}$ and $\mathsf{X}_{A}(x)$ in the contour integral formula, it is natural to relate the $\D0^{-}$ and $\D4^{-}$ branes with $\mathsf{A}(x)$ and $\mathsf{X}_{A}(x)^{-1}$. Therefore, the contour integral formula for the supergroup analog of the spiked instanton system should be proportional to 
\bea
\oint\prod_{A\in\six}\prod_{I=1}^{k_{A}^{+}}\prod_{J=1}^{k_{A}^{-}}\frac{dx_{A,I}^{+}}{2\pi\iota x_{A,I}^{+}}\frac{dx_{A,J}^{-}}{2\pi\iota x_{A,J}^{-}}\langle\mathsf{A}_{\underline{k}_{+}}^{-1}\mathsf{A}_{\underline{k}_{-}}\mathsf{X}_{\underline{n}_{+}|\underline{n}_{-}} \rangle
\eea
where 
\bea
\mathsf{A}_{\underline{k}_{+}}^{-1}=\prod_{A\in\six}\prod_{I=1}^{k_{A,+}}\mathsf{A}(x^{+}_{A,I})^{-1},\quad \mathsf{A}_{\underline{k}_{-}}=\prod_{A\in\six}\prod_{J=1}^{k_{A,-}}\mathsf{A}(x^{-}_{A,J}),\quad \mathsf{X}_{\underline{n}_{+}|\underline{n}_{-}}={:\prod_{A\in\six}\frac{\prod_{\alpha=1}^{n_{A}^{+}}\mathsf{X}_{A}(v_{A,\alpha}^{+})}{\prod_{\beta=1}^{n_{A}^{-}}\mathsf{X}_{A}(v_{A,\beta}^{-})}:}.
\eea
This complete formula indeed reproduces the contour integral formula in \eqref{eq:app-supergroupLMNS} for the affine quiver supergroup gauge theory case. By evaluating the poles properly, we expect this gives the supergroup analog of the gauge origami system.

\subsection{Coupled vortex system}\label{sec:cplvortLMNS}
Similar to the previous cases, let us consider the $\D2$ coupled vortex system in \eqref{eq:D2integral}. We introduce an operator corresponding to the $\D2$-brane wrapping $\mathbb{C}_{a}\times \mathbb{S}^{1}(a\in\four)$ as
\bea
    \mathsf{S}_{a}(x)=\mathsf{s}_{a,0}(x):\exp\left(\sum_{n\neq 0}\mathsf{s}_{a,n}x^{-n}\right):,\quad [\mathsf{s}_{a,n},\mathsf{s}_{b,m}]=-\frac{1}{n}\frac{\bfP_{\four}^{[n]}}{\bfP_{a}^{[-n]}\bfP_{b}^{[n]}}\delta_{n+m,0},\label{eq:D2op}
\eea
where $\mathsf{s}_{a,0}(x)$ is the zero modes and the relation with the $\mathsf{A}$-operator is
\bea
\relax[\mathsf{a}_{n},\mathsf{s}_{a,m}]=-\frac{1}{n}\bfP_{\bar{a}}^{[n]}\delta_{n+m,0},\quad \mathsf{s}_{a,n}=\frac{\mathsf{a}_{n}}{\bfP_{a}^{[-n]}}.
\eea
The operator product formulas are
\bea
    \mathsf{A}(x)\mathsf{S}_{a}(x')&=\wick{\c{\mathsf{a}_{0}(x)}\c{\mathsf{s}_{a,0}(x')}}g_{\bar{a}}(x'/x)^{-1}: \mathsf{A}(x)\mathsf{S}_{a}(x'):,\\
    \mathsf{S}_{a}(x')\mathsf{A}(x)&=\wick{\c{\mathsf{s}_{a,0}(x')}\c{\mathsf{a}_{0}(x)}}g_{\bar{a}}(q_{a}x/x'): \mathsf{A}(x)\mathsf{S}_{a}(x'):,
\eea
where we impose the zero modes condition as
\bea
    \wick{\c{\mathsf{a}_{0}(x)}\c{\mathsf{s}_{a,0}(x')}}=\wick{\c{\mathsf{s}_{a,0}(x')}\c{\mathsf{a}_{0}(x)}}.\label{eq:D0D2zero}
\eea
so that the rational function arising on the right-hand side will be the same after using \eqref{eq:reflec_structfunc}. Explicitly, we have $\wick{\c{\mathsf{a}_{0}(x)}\c{\mathsf{s}_{a,0}(x')}}=1$ (see section \ref{sec:zeromodes}). Under this condition, the contour integral in \eqref{eq:D2integral} is written as follows.
\begin{proposition}
The contour integral formula for the coupled vortex system shown in \eqref{eq:D2integral} is given by the following correlation function of the vertex operators,
\bea\label{eq:D2op_integral}
\mathfrak{q}^{\underline{k}}\mathcal{Z}_{\underline{k}}^{\D2}=\frac{\mathfrak{q}^{\underline{k}}\underline{\mathcal{G}}^{\underline{k}}}{\underline{k}!}\oint\prod_{a\in\four}\prod_{I=1}^{k_{a}}\frac{dx_{a,I}}{2\pi\iota x_{a,I}}\langle\mathsf{A}_{\underline{k}}^{-1}\mathsf{S}_{\underline{n}}\rangle,
\eea
where
\begin{subequations}
\begin{align}
&\mathfrak{q}^{\underline{k}}=\prod_{a\in\four}\mathfrak{q}^{k_{a}},\quad \mathsf{A}_{\underline{k}}=\prod_{a\in\four}\prod_{I=1}^{k_{a}}\mathsf{A}(x_{a,I}),\quad \mathsf{S}_{\underline{n}}={:\prod_{a\in\four}\prod_{\alpha=1}^{n_{a}}\mathsf{S}_{a}(v_{a,\alpha}):},\\
    &\langle\mathsf{A}_{\underline{k}}^{-1}\mathsf{S}_{\underline{n}}\rangle=\prod_{a,b\in\four}\prod_{\alpha=1}^{n_{a}}\prod_{I=1}^{k_{b}}g_{\bar{a}}\left(\frac{v_{a,\alpha}}{x_{b,I}}\right)\prod_{(a,I)<(b,J)}\mathcal{A}_{\mathbb{C}^{4}}\left(\frac{x_{a,I}}{x_{b,J}}\right)^{-1}.
\end{align}
\end{subequations}
The one-loop perturbative part in \eqref{eq:D2oneloop} is obtained from the $\mathsf{S}$-operators:
\bea
    \mathsf{S}_{b}(v_{b,\beta})\mathsf{S}_{a}(x_{a,\alpha})=\mathcal{Z}_{\text{1-loop}}^{\D2\tbar\D2}(v_{a,\alpha},a\,|\,v_{b,\beta},b):\mathsf{S}_{b}(v_{b,\beta})\mathsf{S}_{a}(x_{a,\alpha}):.
\eea
\end{proposition}

The $\D2$-brane $\mathsf{S}$-operators are related to the screening currents of quiver W-algebras~\cite{Kimura:2015rgi}. Let us consider the two screening currents $\mathsf{S}_{1}(x)$ and $\mathsf{S}_{2}(x)$. Focusing on $\mathsf{S}_{2}(x)$, we have 
\bea
\relax    [\mathsf{s}_{2,n},\mathsf{s}_{2,m}]=-\frac{1}{n}\frac{1-q_{1}^{n}}{1-q_{2}^{-n}}(1-q_{3}^{n})(1-q_{4}^{n})\delta_{n+m,0}
\eea
which gives one of the screening currents of the affine quiver W-algebra \cite[eq.~(3.33)]{Kimura:2015rgi}. The screening current $\mathsf{S}_{1}(x)$ gives the other screening current. The other two screening currents $\mathsf{S}_{3}(x),\,\mathsf{S}_{4}(x)$ are introduced in a symmetric way using the quadrality. Thus, using two screening currents $\mathsf{S}_{a}(x),\,\mathsf{S}_{b}(x')\,(a\neq b)$ we will obtain six affine quiver W-algebras.

\begin{remark}
Observing the operators \eqref{eq:D8op}, \eqref{eq:D6op}, \eqref{eq:D4op}, \eqref{eq:D2op}, one can see that the $\D2$, $\D4$, $\D6$, $\D8$ vertex operators $\mathsf{S}_{a}(x),\mathsf{X}_{A}(x),\mathsf{W}_{\bar{a}}(x),\mathsf{Z}(x)$ are all related with the D0-brane operator in a similar way as
\bea
    \mathcal{A}_{n}=\frac{\mathsf{a}_{n}}{\bfP_{\mathcal{S}}^{[-n]}},\quad \mathcal{A}_{n}=\mathsf{s}_{a,n},\,\,\mathsf{x}_{A,n},\,\,\mathsf{w}_{\bar{a},n},\,\,\mathsf{z}_{n}\label{eq:relationwithD0}
\eea
where $\mathcal{S}$ is a subset of $\four$ depending on the subspace the D-branes are spanning. For example, we have the following formal expansions:
\bea\label{eq:expansioninrootcurrent}
\mathsf{s}_{a,n}&=-q_{a}^{n}\sum_{m=0}^{\infty}q_{a}^{nm}\mathsf{a}_{n},\quad |q_{a}|<1,\\
\mathsf{x}_{ab,n}&=q_{ab}^{n}\sum_{l,m=0}^{\infty}q_{a}^{ln}q_{b}^{mn}\mathsf{a}_{n},\quad |q_{a,b}|<1,\\
\mathsf{w}_{abc,n}&=-q_{abc}^{n}\sum_{k,l,m=0}^{\infty}q_{a}^{kn}q_{b}^{ln}q_{c}^{mn}\mathsf{a}_{n},\quad |q_{a,b,c}|<1,
\eea
where we used 
\bea\label{eq:qexp-analyticcont}
\frac{1}{1-q_{a}}=\begin{dcases}
    \sum_{l=0}^{\infty}q_{a}^{l},\quad |q_{a}|<1,\\
    -q_{a}^{-1}\sum_{l=0}^{\infty}q_{a}^{-l},\quad |q_{a}|>1.
\end{dcases}
\eea
Note that we need to be careful of the analytic region of the $q$-parameters when doing such formal expansions. Due to the condition $q_{1}q_{2}q_{3}q_{4}=1$, we can only impose at most three parameters as $|q_{a}|<1$ and thus, the $\mathsf{z}_{n}$ will be expanded as 
\bea
\mathsf{z}_{n}=-q_{123}^{n}\sum_{j,k,l,m=0}^{\infty}q_{1}^{jn}q_{2}^{kn}q_{3}^{ln}q_{4}^{-mn}\mathsf{a}_{n},\quad |q_{1,2,3}|<1,\quad |q_{4}|>1.
\eea
Other formal expansions in different analytic regions are obtained using \eqref{eq:qexp-analyticcont}.
\end{remark}

\subsection{Zero modes conditions}\label{sec:zeromodes}
Let us impose some conditions on the zero modes $\mathsf{a}_{0}(x),\widetilde{\mathsf{z}}_{0}^{K}(x),\mathsf{w}_{\bar{a},0}(x),\mathsf{x}_{A,0}(x),\mathsf{s}_{a,0}(x)$ and determine the free field realizations of them. Using the observation \eqref{eq:relationwithD0}, we can see that the operator product with operators associated with D-branes intersecting only at a point will give rational functions. We impose the zero modes so that the operator product will be the same rational functions after the analytic continuation. For the cases, when the $\mathsf{A}$-operator involves, the conditions are given in \eqref{eq:D0D0zero}, \eqref{eq:D0D8zero}, \eqref{eq:D0D6zero}, \eqref{eq:D0D4zero}, \eqref{eq:D0D2zero}. For the cases when the $\mathsf{S}$-operators involve, we impose the following conditions:
\begin{itemize}
    \item D2$_{a}$-D2$_{b}$ $(a\neq b)$: 
    \begin{subequations}
    \begin{align}
        \mathsf{S}_{a}(x)\mathsf{S}_{b}(x')&=\wick{\c{\mathsf{s}_{a,0}(x)}\c{\mathsf{s}_{b,0}(x')}}\mathscr{S}_{\overbar{ab}}(q_{a}x'/x):\mathsf{S}_{a}(x)\mathsf{S}_{b}(x'):,\\
        \mathsf{S}_{b}(x')\mathsf{S}_{a}(x)&=\wick{\c{\mathsf{s}_{b,0}(x')}\c{\mathsf{s}_{a,0}(x)}}\mathscr{S}_{\overbar{ab}}(q_{b}x'/x):\mathsf{S}_{a}(x)\mathsf{S}_{b}(x'):
    \end{align}
    \end{subequations}
    which gives the zero mode condition
    \bea
        \wick{\c{\mathsf{s}_{a,0}(x)}\c{\mathsf{s}_{b,0}(x')}}=\wick{\c{\mathsf{s}_{b,0}(x')}\c{\mathsf{s}_{a,0}(x)}},\quad a\neq b\label{eq:D2D2zero}
    \eea
    \item D4$_{A}$-D2$_{c}$ $(c,d\in\bar{A})$:
    \begin{subequations}
    \begin{align}
        \mathsf{X}_{A}(x)\mathsf{S}_{c}(x')&=\wick{\c{\mathsf{x}_{A,0}(x)}\c{\mathsf{s}_{c,0}(x')}}\frac{1-q_{A}x'/x}{1-q_{A}q_{d}x'/x}: \mathsf{X}_{A}(x)\mathsf{S}_{c}(x'):,\\
        \mathsf{S}_{c}(x')\mathsf{X}_{A}(x)&=\wick{\c{\mathsf{s}_{c,0}(x')}\c{\mathsf{x}_{A,0}(x)}}\frac{1-q_{A}^{-1}x/x'}{1-q_{A}^{-1}q_{d}^{-1}x/x'}:\mathsf{X}_{A}(x)\mathsf{S}_{c}(x'):
    \end{align}
    \end{subequations}
    which gives the zero mode condition
    \bea
    \wick{\c{\mathsf{x}_{A,0}(x)}\c{\mathsf{s}_{c,0}(x')}}=q_{A}^{-1}q_{c}^{-1}\wick{\c{\mathsf{s}_{c,0}(x')}\c{\mathsf{x}_{A,0}(x)}}\label{eq:D2D4zero}
    \eea
    \item D2$_{a}$-D6$_{\bar{a}}$:
    \begin{subequations}
    \begin{align}
        \mathsf{W}_{\bar{a}}(x)\mathsf{S}_{a}(x')&=\wick{\c{\mathsf{w}_{\bar{a},0}(x)}\c{\mathsf{s}_{a,0}(x')}}\frac{1}{1-q_{a}^{-1}x'/x}:\mathsf{W}_{\bar{a}}(x)\mathsf{S}_{a}(x'):,\\
        \mathsf{S}_{a}(x')\mathsf{W}_{\bar{a}}(x)&=\wick{\c{\mathsf{s}_{a,0}(x')}\c{\mathsf{w}_{\bar{a},0}(x)}}\frac{1}{1-q_{a}x/x'}:\mathsf{W}_{\bar{a}}(x)\mathsf{S}_{a}(x'):
    \end{align}
    \end{subequations}
    which gives the zero mode condition
    \begin{equation}
        \wick{\c{\mathsf{w}_{\bar{a},0}(x)}\c{\mathsf{s}_{a,0}(x')}}=\left(-\frac{x'}{q_{a}x}\right)\wick{\c{\mathsf{s}_{a,0}(x')}\c{\mathsf{w}_{\bar{a},0}(x)}}\label{eq:D2D6zero}
    \end{equation}
\end{itemize}

We can do the same analysis for the $\D4$-brane operators as
\begin{subequations}
\begin{align}
    \mathsf{X}_{A}(x)\mathsf{X}_{\bar{A}}(x')=\wick{\c{\mathsf{x}_{A,0}(x)}\c{\mathsf{x}_{\bar{A},0}(x')}}\left(1-q_{A}\frac{x'}{x}\right):\mathsf{X}_{A}(x)\mathsf{X}_{\bar{A}}(x'):,\\
    \mathsf{X}_{\bar{A}}(x')\mathsf{X}_{A}(x')=\wick{\c{\mathsf{x}_{\bar{A},0}(x')}\c{\mathsf{x}_{A,0}(x)}}\left(1-q_{\bar{A}}\frac{x}{x'}\right):\mathsf{X}_{A}(x)\mathsf{X}_{\bar{A}}(x'):
\end{align}
\end{subequations}
and impose the condition
\beq
   \mathsf{x}_{A,0}(x)\mathsf{x}_{\bar{A},0}(x')=\left(-q_{A}^{-1}\frac{x}{x'}\right)\mathsf{x}_{\bar{A},0}(x')\mathsf{x}_{A,0}(x)\label{eq:D4D4zero}
\eeq
but to make the discussion simple, we do not impose this condition\footnote{\label{note:D4footnote}This condition only affects the zero modes when we are considering the quadratic relations of the $qq$-characters which will be derived in section \ref{sec:D4quadraticrelation}.}.

Under these conditions, the free field realizations of the zero modes are given as 
\begin{subequations}\label{eq:zeromodes1}
\begin{align}
    &\mathsf{a}_{0}(x)=e^{\mathsf{a}_{0}},\quad \mathsf{s}_{a,0}(x)=x^{\mathsf{s}_{a,0}}e^{\widetilde{\mathsf{s}}_{a,0}},\quad \mathsf{w}_{\bar{a}}(x)=x^{\mathsf{w}_{\bar{a},0}}e^{\widetilde{\mathsf{w}}_{\bar{a},0}}e^{\widetilde{\widetilde{\mathsf{w}}}_{\bar{a},0}},\\
    &\mathsf{x}_{A,0}(x)=e^{\mathsf{x}_{A,0}}, \quad \widetilde{\mathsf{z}}_{0}^{K}(x)=x^{\mathsf{z}^{K}_{0}}e^{\widetilde{\mathsf{z}}^{K}_{0}}e^{\widetilde{\widetilde{\mathsf{z}}}^{K}_{0}}
\end{align}
\end{subequations}
with
\begin{subequations}\label{eq:zeromodes2}
\begin{align}
& \mathsf{a}_{0}=\mathsf{t}_{0},\quad \mathsf{s}_{a,0}=-(\log q_{a})^{-1}\mathsf{t}_{0},\quad \widetilde{\mathsf{s}}_{a,0}=-(\log q_{a})^{-1}\widetilde{\partial}_{\mathsf{t}},\\
    &\mathsf{w}_{\bar{a},0}=-\log q_{a}\,\widetilde{\mathsf{t}}_{0},\quad \widetilde{\mathsf{w}}_{\bar{a},0}=-\log q_{a}\log (-q_{a})\,\widetilde{\mathsf{t}}_{0},\quad \widetilde{\widetilde{\mathsf{w}}}_{\bar{a},0}=-\log q_{a}\partial_{\mathsf{t}},\\
    &\mathsf{x}_{A,0}=\log q_{c}\log q_{d}\,\widetilde{\mathsf{t}}_{0},\,\,\quad (\bar{A}=cd),\\
    &\mathsf{z}^{K}_{0}=-\log K \tilde{\mathsf{t}}_{0},\quad \widetilde{\mathsf{z}}^{K}_{0}=-\log K\log(-K)\widetilde{\mathsf{t}}_{0},\quad \tilde{\tilde{\mathsf{z}}}^{K}_{0}=-\log K\partial_{\mathsf{t}}
\end{align}
\end{subequations}
where we introduced two independent sets of zero modes
\beq
    [\partial_{\mathsf{t}},\mathsf{t}_{0}]=[\widetilde{\partial}_{\mathsf{t}},\tilde{\mathsf{t}}_{0}]=1,\quad [\mathsf{t}_{0},\tilde{\mathsf{t}}_{0}]=[\partial_{\mathsf{t}},\widetilde{\partial}]=[\mathsf{t}_{0},\widetilde{\partial}_{\mathsf{t}}]=[\tilde{\mathsf{t}}_{0},\partial_{\mathsf{t}}]=0.\label{eq:zeromodesdef}
\eeq
The normal ordering is defined as 
\beq
{:\partial_{\mathsf{t}}\,\mathsf{t}_{0}:}=\mathsf{t}_{0}\partial_{\mathsf{t}},\quad {:\tilde{\partial}_{\mathsf{t}}\,\tilde{\mathsf{t}}_{0}:}=\tilde{\mathsf{t}}_{0}\tilde{\partial}_{\mathsf{t}}.
\eeq
See Lem.~\ref{app:lem:zero-modes-prf1} for the proof that the above explicit form of zero modes obeys the zero modes conditions.

Under the above conditions, we actually have extra relations 
\beq
    \mathsf{a}_{0}(x)={:\frac{\mathsf{s}_{a,0}(x)}{\mathsf{s}_{a,0}(q_{a}x)}:},\quad \mathsf{x}_{ab,0}(x)={:\frac{\mathsf{w}_{abc,0}(x)}{\mathsf{w}_{abc,0}(q_{c}x)}:},\quad \mathsf{w}_{\bar{a},0}(x)=\tilde{\mathsf{z}}^{q_{a}}_{0}(x)\label{eq:zeromodesrelation}
\eeq
which imply
\beq
    \mathsf{A}(x)={:\frac{\mathsf{S}_{a}(x)}{\mathsf{S}_{a}(q_{a}x)}:},\quad \mathsf{X}_{ab}(x)={:\frac{\mathsf{W}_{abc}(x)}{\mathsf{W}_{abc}(q_{c}x)}:},\quad \mathsf{W}_{\bar{a}}(x)={:\frac{\mathsf{Z}(x)}{\mathsf{Z}(q_{a}x)}:}=\widetilde{\mathsf{Z}}^{q_{a}}(x).\label{eq:oprelation1}
\eeq
See Cor.~\ref{app:cor:zero-modes-prf2} for the derivation of these relations. Note that in our notation, the relation between the $\D2$ and $\D4$ operators are 
\begin{align}
    \mathsf{S}_{a}(x)=\mathsf{s}_{a,0}(x):\frac{\mathsf{X}_{ab}(x)}{\mathsf{X}_{ab}(q_{b}x)}:
\end{align}
where extra zero modes appear in front. Using these relations, we also have 
\begin{subequations}
\begin{align}
\mathsf{A}(x)&=\mathsf{a}_{0}(x):\frac{\mathsf{X}_{ab}(x)\mathsf{X}_{ab}(q_{ab}x)}{\mathsf{X}_{ab}(q_{a}x)\mathsf{X}_{ab}(q_{b}x)}:,\label{eq:oprelationwithD0-D4}\\
&=\mathsf{a}_{0}(x):\frac{\mathsf{W}_{abc}(x)\mathsf{W}_{abc}(q_{ab}x)\mathsf{W}_{abc}(q_{ac}x)\mathsf{W}_{abc}(q_{bc}x)}{\mathsf{W}_{abc}(q_{a}x)\mathsf{W}_{abc}(q_{b}x)\mathsf{W}_{abc}(q_{c}x)\mathsf{W}_{abc}(q_{abc}x)}:\label{eq:oprelationwithD0-D6}\\
&=\mathsf{a}_{0}(x):\frac{\mathsf{Z}(x)^{2}\prod_{a<b}\mathsf{Z}(q_{ab}x)}{\prod_{a\in\four}\mathsf{Z}(q_{a}x)\prod_{a\in\four}\mathsf{Z}(q_{a}^{-1}x)}:\label{eq:oprelationwithD0-D8}
\end{align}\label{eq:oprelationwithD0}
\end{subequations}
The contraction formulas of the D-brane vertex operators obtained after using the explicit form of the zero modes are summarized in Prop.~\ref{app:prop:contraction_formula}

\subsection{Quiver structure and generalizations}\label{sec:quiverstructure}
The vertex operators $\mathsf{A}(x),\mathsf{S}_{a}(x),\mathsf{X}_{A}(x),\mathsf{W}_{\bar{a}}(x),\mathsf{Z}(x)$ are all defined by the commutation relations with the root current $\mathsf{A}(x)$. Let us show that we can understand the commutation relations using \emph{graded quivers} \cite{Oppermann2015QuiversFS, Buan2008ColouredQM,Franco:2017lpa,Closset:2018axq} (see also the references therein) and that the commutation relations have a $q$-Cartan matrix understanding. We also briefly discuss how to generalize our construction to other complicated geometries.

\subsubsection{\texorpdfstring{$\mathbb{C}^{4}$}{C4} geometry}
We first denote $\bfP_{a},\bfP_{A},\bfP_{\bar{a}},\bfP_{\four}\,(a\in\four, A\in\six)$ as $c_{a},c_{A},c_{\bar{a}},c_{1234}$ and call them the \emph{full $q$-Cartan matrix}:
\bea
&c_{a}=1-q_{a},\quad c_{ab}=(1-q_{a})(1-q_{b}),\quad c_{abc}=(1-q_{a})(1-q_{b})(1-q_{c}),\\ &c_{1234}=(1-q_{1})(1-q_{2})(1-q_{3})(1-q_{4}).
\eea
The full $q$-Cartan matrix can be decomposed into the \emph{half $q$-Cartan matrix} $c_{a,A,\bar{a},1234}^{+}$ as 
\bea\label{eq:halfCartan}
&c_{a}=c_{a}^{+}+c_{a}^{-},\quad c_{a}^{+}=1,\quad c_{a}^{-}=-q_{a}c_{a}^{+\vee},\\
&c_{ab}=c_{ab}^{+}+c_{ab}^{-},\quad  c_{ab}^{+}=1-q_{a},\quad c_{ab}^{-}=q_{ab}c_{ab}^{+\vee},\\
&c_{abc}=c_{abc}^{+}+c_{abc}^{-},\quad c_{abc}^{+}=1-q_{a}-q_{b}-q_{c},\quad c_{abc}^{-}=-q_{a}q_{b}q_{c}c_{abc}^{+\vee},\\
&c_{1234}=c_{1234}^{+}+c_{1234}^{-},\quad c_{1234}^{+}=1-\sum_{a\in\four}q_{a}+q_{4}\sum_{i=1}^{3}q_{i},\quad c_{1234}^{-}=c_{1234}^{+\vee}.
\eea
We can write this decomposition in a compact form as 
\bea
c_{\mathcal{S}}=\bfP_{\mathcal{S}}=c_{\mathcal{S}}^{+}+c_{\mathcal{S}}^{-},\quad c_{\mathcal{S}}^{-}=\prod_{i\in\mathcal{S}}(-q_{i})c_{\mathcal{S}}^{+\vee},\quad \mathcal{S}\subseteq \four.
\eea
Note that the decomposition is not unique and is related to how we choose the square root part of the index as what we did in section \ref{sec:equiv-index}. Since each vertex operator corresponds to the D-brane wrapping subspaces $\mathbb{C}_{a},\mathbb{C}^{2}_{A},\mathbb{C}^{3}_{\bar{a}},\mathbb{C}^{4}$, the $q$-Cartan matrices should include the information of these geometries. Let us show that the $q$-Cartan matrices can be read from quivers dual to each of the geometry $\mathbb{C}_{a},\mathbb{C}^{2}_{A},\mathbb{C}^{3}_{\bar{a}},\mathbb{C}^{4}$. To show this, let us first introduce the concept of graded quivers following the convention of \cite[section 2]{Franco:2017lpa,Closset:2018axq}.

\paragraph{Graded quivers} A quiver is denoted $Q=(Q_{0},Q_{1})$, where $Q_{0}$ is a set of nodes and $Q_{1}$ is a set of arrows. For each arrow $I:i\rightarrow j,\,(i,j\in Q_{0})$, we have the source node $s(I)=i$ and the target node $t(I)=j$. 

We fix $m$ to be a nonnegative integer and denote the graded quiver as $\overbar{Q}=(\overbar{Q}_{0},\overbar{Q}_{1})$. The set of nodes does not change $\overbar{Q}_{0}=Q_{0}$ but the set of arrows changes to a set of \emph{graded} arrows. A graded arrow $I:i\rightarrow j$ has a source node $s(I)=i\in Q_{0}$, a target node $t(I)=j\in Q_{0}$ and additionally a degree $|I|$ which is an integer in $\{0,1,\ldots,m\}$. For every graded arrow $I\in \overbar{Q}_{1}$, we also introduce its opposite $I_{\text{op}}:j\rightarrow i$, where the source node and target node are reversed and the degree is given $|I_{\text{op}}|=m-|I|$. Namely, we will focus on a particular kind of graded quiver, such that every arrow has an opposite arrow.

For every node $i\in \overbar{Q}_{0}$, additionally, there is a loop $\ell_{i}$ with $s(\ell_{i})=i$, $t(\ell_{i})=i$ and degree~$-1$. We also introduce the opposite of them as $\bar{\ell}_{i}$ with degree $m+1$. These are the only arrows with degrees not in $\{0,1,\ldots,m\}$ and will be not drawn explicitly in the quiver diagram.

For later use, we denote an arrow from $i$ to $j$ with degree $c$ as 
\beq
\Phi^{(c)}_{ij}:i\rightarrow j
\eeq
and the opposite of it as 
\beq
(\Phi^{(c)}_{ij})_{\text{op}}\equiv \bar{\Phi}^{(m-c)}_{ji}.
\eeq
We then can pair the double arrows as $(\Phi^{(c)}_{ij},\bar{\Phi}_{ji}^{(m-c)})$ and illustrate them as 
\bea\label{eq:doublearrows}
\newcommand{\arrowHeadPosition}{0.7}
\begin{tikzpicture}[decoration={markings,mark=at position \arrowHeadPosition with {\arrow{latex}}}]
    \draw[thick, postaction={decorate}](0,0.2)--(3,0.2);
    \draw[thick, postaction={decorate}](3,-0.2)--(0,-0.2);
    \draw[fill=black!20!white,thick](0,0) circle(0.4cm);
    \node[thick] at (0,0){$i$};
    \draw[fill=black!20!white,thick](3,0) circle (0.4cm);
    \node[thick] at (3,0){$j$};
    \node[above] at (1.5, 0.3){$\Phi^{(c)}_{ij}$};
    \node [below] at (1.5,-0.3){$\bar{\Phi}^{(m-c)}_{ji}$};
\end{tikzpicture}
\eea
Following the terminology in \cite{Franco:2017lpa}, we refer to this pair of arrows as $(c,m-c)$ arrow. The different types of arrows are then labeled by $0\leq c\leq m/2$. 

From the physical viewpoint, the graded quiver contains information of the components included in a quiver gauge theory. The quiver nodes are identified with the gauge groups\footnote{We may assign integers $N_{i}$ to each quiver node and then the gauge group is identified as $\U(N_{i})$.}, the loops with degree $-1,m+1$ are identified with vector multiples, and the arrows with degrees $\{0,1,\ldots,m\}$ are matter fields. Different degrees of arrows represent different matter fields, and thus the graded quiver enables us to describe a larger class of quiver gauge theories in various dimensions. Each arrow connecting quiver nodes $i,j$ will change under the (anti-)fundamental representation of the source and target nodes. Note that a self-loop arrow $\Phi_{ii}^{(c)}$ will be a matter multiplet transforming in the adjoint representation.

We may choose an orientation of the double arrow by drawing either $\Phi_{ij}^{(c)}$ or $\bar{\Phi}^{(m-c)}_{ji}$. In the quiver diagram, only one of them will be explicitly drawn.\footnote{For a detailed discussion of this choice see \cite[section 2]{Franco:2017lpa}.} When $m$ is even, we have a pair of degree $c=m/2$ arrows $(\Phi^{(m/2)}_{ij},\bar{\Phi}^{(m/2)}_{ji})$. These arrows will be drawn in unoriented arrows. Note again that the self-loops $\ell_{i}$ are not explicitly drawn in the quiver diagram.\footnote{In addition to the quiver diagram, we can assign \emph{graded quiver superpotentials} which imposes relations on the path algebra obtained from the graded quiver. When relating the quiver arrows with the $q$-deformation parameters of the $q$-Cartan matrix, they will play important roles, but in this paper, we omit the discussion of it and postpone it for future work.} 

For later use, let us briefly explain the connection between graded quivers and supersymmetric gauge theories. Graded quivers of degree $m$ correspond to $(6-2m)$-dimensional gauge theories with $2^{3-m}$ supercharges. Consider a Type $\IIB$ string theory where $\D(5-2m)$-branes probe the CY $(m+2)$-folds. The Calabi--Yau manifold arises as the classical moduli spaces of the gauge theories \cite{Douglas:1996sw, Franco:2005sm,Franco:2005rj,Franco:2015tya}. Depending on $m$, the arising gauge theory is given as 
\bea\renewcommand{\arraystretch}{1.2}
\begin{array}{|c|c|c|}\hline
    \,m\, & \text{gauge theory} & \text{geometry}  \\
\hline    \,0\, & \text{ 6d $\mathcal{N}=(1,0)$ theory }   &  \text{ D5-branes  probing $\CY_{2}$ }  \\
 \hline   \,1\, &   \text{ 4d $\mathcal{N}=1$ theory } & \text{ D3-branes probing $\CY_{3}$ } \\
 \hline   \,2\, &  \text{ 2d $\mathcal{N}=(0,2)$ theory 
 }&  \text{ D1-branes probing $\CY_{4}$ }\\\hline
\end{array}
\eea

\begin{itemize}
    \item $m=2$ case: The $m=2$ graded quivers correspond to 2d $\mathcal{N}=(0,2)$ quiver gauge theories. The graded arrows are described in $(c,m-c)$ double arrows with $0\leq c\leq m/2$. In this case, we have $(0,2)$ and $(1,1)$ double arrows. From the gauge theoretic viewpoint, we have vector superfields, chiral superfields, and the Fermi superfields. The vector superfields are associated with each node of the quiver. The chiral (Fermi) superfields correspond with the degree $c=0$ ($c=1$) arrows of the graded quiver. We denote the chiral superfields as $X_{ij}$ and they are identified with the arrows as $(X_{ij},\overline{X}_{ij})\leftrightarrow (0,2)$ arrow. The Fermi superfields $\Lambda_{ij}$ are identified with the arrows as $(\Lambda_{ij},\bar{\Lambda}_{ij})\leftrightarrow (1,1)$ arrow. The quiver diagram is then described as 
    \bea
\newcommand{\arrowHeadPosition}{0.7}
\begin{tikzpicture}[decoration={markings,mark=at position \arrowHeadPosition with {\arrow{latex}}}]
\begin{scope}
    \draw[thick, postaction={decorate}](0,0.2)--(3,0.2);
    \draw[thick, postaction={decorate}](3,-0.2)--(0,-0.2);
    \draw[fill=black!20!white,thick](0,0) circle(0.4cm);
    \node[thick] at (0,0){$i$};
    \draw[fill=black!20!white,thick](3,0) circle (0.4cm);
    \node[thick] at (3,0){$j$};
    \node[above] at (1.5, 0.3){$(0)$};
    \node [below] at (1.5,-0.3){$(2)$};
\end{scope}
\begin{scope}[yshift=-2cm]
    \draw[thick, postaction={decorate}](0,0.2)--(3,0.2);
    \draw[thick, postaction={decorate}](3,-0.2)--(0,-0.2);
    \draw[fill=black!20!white,thick](0,0) circle(0.4cm);
    \node[thick] at (0,0){$i$};
    \draw[fill=black!20!white,thick](3,0) circle (0.4cm);
    \node[thick] at (3,0){$j$};
    \node[above] at (1.5, 0.3){$(1)$};
    \node [below] at (1.5,-0.3){$(1)$};
\end{scope}
\begin{scope}[xshift=4cm]
\draw[ultra thick,<->] (0,0)--(1,0);
\end{scope}
\begin{scope}[xshift=4cm, yshift=-2cm]
\draw[ultra thick,<->] (0,0)--(1,0);
\end{scope}
\begin{scope}[xshift=6cm]
    \draw[thick, postaction={decorate}](0,0)--(3,0);
    \draw[fill=black!20!white,thick](0,0) circle(0.4cm);
    \node[thick] at (0,0){$i$};
    \draw[fill=black!20!white,thick](3,0) circle (0.4cm);
    \node[thick] at (3,0){$j$};
    \node[above] at (1.5, 0){$X_{ij}$};
\end{scope}
\begin{scope}[xshift=6cm,yshift=-2cm]
    \draw[thick,red](0,0)--(3,0);
    \draw[fill=black!20!white,thick](0,0) circle(0.4cm);
    \node[thick] at (0,0){$i$};
    \draw[fill=black!20!white,thick](3,0) circle (0.4cm);
    \node[thick] at (3,0){$j$};
    \node[above] at (1.5, 0){$\Lambda_{ij}$};
\end{scope}
\end{tikzpicture}
\eea
The left diagram is the quiver diagram with double arrows as \eqref{eq:doublearrows}, while the right diagram is the quiver diagram with single arrows. The Fermi superfields are drawn in unoriented arrows because the degrees are self-conjugate. For examples of quivers having multiple nodes, see \cite{Franco:2015tya,Franco:2015tna}. Quivers of these types are denoted as $\overbar{Q}=(\overbar{Q}_{0},\overbar{Q}_{1})$ in later discussions.

\item $m=1$ case: The $m=1$ graded quivers correspond to 4d $\mathcal{N}=1$ quiver gauge theories. We only have one type of double arrows which is $(0,1)$. They correspond to the chiral superfields of the gauge theory which we denote $X_{ij}$. The quiver diagram is illustrated as 
\bea
\newcommand{\arrowHeadPosition}{0.7}
\begin{tikzpicture}[decoration={markings,mark=at position \arrowHeadPosition with {\arrow{latex}}}]
\begin{scope}
    \draw[thick, postaction={decorate}](0,0.2)--(3,0.2);
    \draw[thick, postaction={decorate}](3,-0.2)--(0,-0.2);
    \draw[fill=black!20!white,thick](0,0) circle(0.4cm);
    \node[thick] at (0,0){$i$};
    \draw[fill=black!20!white,thick](3,0) circle (0.4cm);
    \node[thick] at (3,0){$j$};
    \node[above] at (1.5, 0.3){$(0)$};
    \node [below] at (1.5,-0.3){$(1)$};
\end{scope}
\begin{scope}[xshift=4cm]
\draw[ultra thick,<->] (0,0)--(1,0);
\end{scope}
\begin{scope}[xshift=6cm]
    \draw[thick, postaction={decorate}](0,0)--(3,0);
    \draw[fill=black!20!white,thick](0,0) circle(0.4cm);
    \node[thick] at (0,0){$i$};
    \draw[fill=black!20!white,thick](3,0) circle (0.4cm);
    \node[thick] at (3,0){$j$};
    \node[above] at (1.5, 0){$X_{ij}$};
\end{scope}
\end{tikzpicture}
\eea    
Some examples are
\bea
\newcommand{\arrowHeadPosition}{0.7}
\begin{tikzpicture}[decoration={markings,mark=at position \arrowHeadPosition with {\arrow{latex}}}]
\begin{scope}
    \draw[postaction={decorate}, black,thick,scale=1.3] (0,0) arc(270:-90:0.20 and 0.6) ;
    \draw[postaction={decorate}, black,thick,scale=1.3,rotate=-120] (0,0) arc(270:-90:0.20 and 0.6) ;
    \draw[postaction={decorate}, black,thick,scale=1.3,rotate=120] (0,0) arc(270:-90:0.20 and 0.6) ;
    \draw[fill=black!20!white,thick](0,0) circle(0.3cm);
    \node at (0,-1.5){$\mathbb{C}^{3}$};
\end{scope}
\begin{scope}[xshift=4cm]
    \draw[postaction={decorate}, black,thick,scale=1.3] (0.65,0) arc(0:-180:0.65 and 0.3) ;
    \draw[postaction={decorate}, black,thick,scale=1.3] (0.75,0) arc(0:-180:0.75 and 0.4) ;
    \draw[postaction={decorate}, black,thick,scale=1.3] (-0.65,0) arc(180:0:0.65 and 0.3) ;
    \draw[postaction={decorate}, black,thick,scale=1.3] (-0.75,0) arc(180:0:0.75 and 0.4) ;
    \draw[postaction={decorate}, black,thick,scale=1.3] (-0.8,0) arc(360:0:0.4 and 0.3) ;
    \draw[postaction={decorate}, black,thick,scale=1.3] (0.8,0) arc(-180:180:0.4 and 0.3) ;
    \draw[fill=black!20!white,thick](-0.95,0) circle(0.3cm);
    \draw[fill=black!20!white,thick](0.9,0) circle(0.3cm);
    \node at (0,-1.5){$\mathbb{C}^{2}/\mathbb{Z}_{2}\times \mathbb{C}$};
\end{scope}
\begin{scope}[xshift=8cm]
    \draw[postaction={decorate}, black,thick,scale=1.3] (0.65,0) arc(0:-180:0.65 and 0.3) ;
    \draw[postaction={decorate}, black,thick,scale=1.3] (0.75,0) arc(0:-180:0.75 and 0.4) ;
    \draw[postaction={decorate}, black,thick,scale=1.3] (-0.65,0) arc(180:0:0.65 and 0.3) ;
    \draw[postaction={decorate}, black,thick,scale=1.3] (-0.75,0) arc(180:0:0.75 and 0.4) ;
    \draw[fill=black!20!white,thick](-0.95,0) circle(0.3cm);
    \draw[fill=black!20!white,thick](0.9,0) circle(0.3cm);
    \node at (0,-1.5){conifold};
\end{scope}
\end{tikzpicture}
\eea
Such kind of quivers are denoted as $Q=(Q_{0},Q_{1})$ in later sections.

\item  $m=0$ case: The $m=0$ graded quivers correspond to 6d $\mathcal{N}=(1,0)$ theories. The nodes correspond to the vector supermultiplets, while the $(0,0)$ double arrow corresponds to the hypermultiplets. The discussion for the tensor multiplets is omitted. The quiver diagram is described as
\bea
\newcommand{\arrowHeadPosition}{0.7}
\begin{tikzpicture}[decoration={markings,mark=at position \arrowHeadPosition with {\arrow{latex}}}]
\begin{scope}
    \draw[thick, postaction={decorate}](0,0.2)--(3,0.2);
    \draw[thick, postaction={decorate}](3,-0.2)--(0,-0.2);
    \draw[fill=black!20!white,thick](0,0) circle(0.4cm);
    \node[thick] at (0,0){$i$};
    \draw[fill=black!20!white,thick](3,0) circle (0.4cm);
    \node[thick] at (3,0){$j$};
    \node[above] at (1.5, 0.3){$(0)$};
    \node [below] at (1.5,-0.3){$(0)$};
\end{scope}
\begin{scope}[xshift=4cm]
\draw[ultra thick,<->] (0,0)--(1,0);
\end{scope}
\begin{scope}[xshift=6cm]
    \draw[thick](0,0)--(3,0);
    \draw[fill=black!20!white,thick](0,0) circle(0.4cm);
    \node[thick] at (0,0){$i$};
    \draw[fill=black!20!white,thick](3,0) circle (0.4cm);
    \node[thick] at (3,0){$j$};
    \node[above] at (1.5, 0){$X_{ij}$};
\end{scope}
\end{tikzpicture}
\eea    
where we denoted the hypermultiplets using $X_{ij}$.
\end{itemize}

Some examples of these quivers are  
\bea
\begin{tikzpicture}[scale=0.8]
\begin{scope}[scale=0.9]
    \draw[black,thick] (0,0) arc(-90:270:0.6 and 1.0) ;
    \draw[fill=black!20!white,thick](0,0) circle(0.3cm);
    \node at (0,-5){$\mathbb{C}^{2}$ ($\widehat{A}_{0}$)};
\end{scope}
\begin{scope}[scale=0.9,xshift=8cm]
\foreach \ang\lab in {90/1,45/2,180/{n-1},135/n}{
  \draw[-,thick] ($(0,0)+(\ang:2.9)$) arc (\ang:\ang-45:2.9);
}
\draw[-,shorten <=7pt,thick] ($(0,0)+(320:2.9)$) arc (320:355:2.9);
\draw[-,shorten >=7pt,thick] ($(0,0)+(305:2.9)$) arc (305:270:2.9);
\draw[-,shorten <=7pt,thick] ($(0,0)+(270:2.9)$) arc (270:235:2.9);
\draw[-,shorten >=7pt,thick] ($(0,0)+(215:2.9)$) arc (215:180:2.9);
\foreach \ang in {310,315,320,220,225,230}{
  \draw[fill=black] ($(0,0)+(\ang:2.9)$) circle (.02);
}
\foreach \ang\lab\anch in {90/1/north, 45/2/{north east}, 0/3/east, 270/i/south, 180/{n-1}/west, 135/n/{north west}}{
  \draw[fill=black!20!white,thick] ($(0,0)+(\ang:2.9)$) circle (.28cm);
}
\node at (0,-5){$\mathbb{C}^{2}/\mathbb{Z}_{n}$ ($\widehat{A}_{n-1}$)};
\end{scope}
\end{tikzpicture}
\eea
The quivers coming from the toric Calabi--Yau two-fold are the affine A-type Dynkin quivers. Note that we can also consider affine D and E-type quivers though they are not toric. Moreover, by decoupling one of the arrows of the quivers, we can also consider finite-type A, D, and E Dynkin quivers. 

These types of quivers coming from $m=0$ are denoted as $\Gamma=(\Gamma_{0},\Gamma_{1})$ or $\Upsilon=(\Upsilon_{0},\Upsilon_{1})$ in later sections.

\paragraph{Gauge origami viewpoint} The gauge origami system can be understood a 2d $\mathcal{N}=(0,2)$ quiver gauge theory with flavor branes. Consider the type IIB theory on $Z\times \mathcal{C}$ where $Z$ is a toric $\CY_{4}$ and $\mathcal{C}=\mathbb{C},\,\mathbb{C}^{\times},\,\mathbb{T}^{2}$. Consider a situation where there are multiple D1-branes\footnote{We may add fractional D1-branes to make the gauge group for each quiver node differ.} wrapping $\mathcal{C}$ probing $Z$. The arising theory is a 2d $\mathcal{N}=(0,2)$ quiver gauge theory. We then consider D3, D5, D7, D9-branes wrapping non-compact cycles of $Z$ while preserving two supersymmetries. These branes will play the role of flavor branes from the D1-brane theory viewpoint. On the other hand, from the viewpoint of the theory inside $Z$, we have a generalized gauge theory where the D-branes inside $Z$ intersect with each other, and the D1-branes play the role of the instantons of the theory. Depending on $\mathcal{C}$, we can study the dimensional reduction (matrix models/supersymmetric quantum mechanics/field theory), and the gauge origami partition function is the index of the system. When $Z=\mathbb{C}^{4}$, this is obvious from the discussion in section~\ref{sec:physicalsetup}. Taking T-duality of the setups in \eqref{eq:spikedhierarchy}, \eqref{eq:tetrahedronhierarchy}, \eqref{eq:m4hierarchy}, \eqref{eq:cplvorthierarchy}, for all cases, we indeed have a D1-brane theory with flavor branes wrapping the subspaces of $\mathbb{C}^{4}$.

\paragraph{$\mathbb{C}^{4}$-case} 
Let us consider the case when $\CY_{4}=\mathbb{C}^{4}_{1234}$. The dual gauge theory is just the maximally supersymmetric Yang--Mills in 2d which is the 2d $\mathcal{N}=(8,8)$ SYM. This is because placing the D1-brane in the transverse direction of $\mathbb{C}^{4}$ will give a SYM theory with 16 supercharges. Note that it can be obtained by a dimensional reduction of the 4d $\mathcal{N}=4$ SYM. In the 2d $\mathcal{N}=(0,2)$ language, we have one vector superfield, four adjoint chiral superfields $X_{1},X_{2},X_{3},X_{4}$, and three Fermi superfields $\Lambda_{1},\Lambda_{2},\Lambda_{3}$. The quiver diagram is then described as 
\bea
\newcommand{\arrowHeadPosition}{0.7}
\begin{tikzpicture}[decoration={markings,mark=at position \arrowHeadPosition with {\arrow{latex}}}]
\begin{scope}
    \draw[postaction={decorate}, black,thick,scale=1.3] (0,0) arc(360:0:0.80 and 0.4) ;
    \draw[postaction={decorate}, black,thick,scale=1.3] (0,0) arc(360:0:0.70 and 0.3) ;
    \draw[postaction={decorate}, black,thick,scale=1.3] (0,0) arc(360:0:0.60 and 0.2) ;
    \draw[postaction={decorate}, black,thick,scale=1.3] (0,0) arc(360:0:0.50 and 0.1) ;
    \draw[red,thick,scale=1.3] (0,0) arc(180:-180:0.80 and 0.4) ;
    \draw[red,thick,scale=1.3] (0,0) arc(180:-180:0.70 and 0.3) ;
    \draw[red,thick,scale=1.3] (0,0) arc(180:-180:0.60 and 0.2) ;
    \draw[fill=black!20!white,thick](0,0) circle(0.4cm);
    \node [thick,left] at (-2.3,0) {$X_{1},X_{2},X_{3},X_{4}$};
    \node [thick,right,red] at (2.3,0) {$\Lambda_{1},\Lambda_{2},\Lambda_{3}$};
\end{scope}
\end{tikzpicture}
\eea
Note that the chiral superfields $X_{1},X_{2},X_{3},X_{4}$ are identified with the four complex coordinates of the transverse $\mathbb{C}_{1234}^{4}$.

After identifying the superfields with the commutative parameters in the $q$-Cartan matrix, we can define the half $q$-Cartan matrix of this quiver as
\bea
c_{1234}^{+}=1-(X_{1}+X_{2}+X_{3}+X_{4})+\textcolor{red}{(\Lambda_{1}+\Lambda_{2}+\Lambda_{3})}.
\eea
The term $1$ corresponds with the vector superfield, the terms $X_{i}\,(i=1,2,3,4)$ correspond with the chiral superfields, and the red term $\Lambda_{i}\,(i=1,2,3)$ corresponds with the Fermi superfields. 

To relate with \eqref{eq:halfCartan}, we need the identification
\bea\label{eq:2dquiverident}
&X_{1}\leftrightarrow q_{1},\quad X_{2}\leftrightarrow q_{2},\quad X_{3}\leftrightarrow q_{3},\quad X_{4}\leftrightarrow q_{4},\\
&\Lambda_{1}\leftrightarrow q_{1}q_{4},\quad \Lambda_{2}\leftrightarrow q_{2}q_{4},\quad \Lambda_{3}\leftrightarrow q_{3}q_{4}.
\eea
For the conjugated fields such as $\bar{\Lambda}_{1}$, the parameter $q_{1}^{-1}q_{4}^{-1}$ is assigned.

Actually, this identification can be understood from the superpotential of the theory. The potential of the 2d $\mathcal{N}=(0,2)$ theories takes the form 
\bea
W=\sum_{a}\left(\Lambda_{a}J_{a}(X)+\bar{\Lambda}_{a}E_{a}(x)\right)
\eea
where $a$ runs all the Fermi fields. For the $\mathbb{C}^{4}$ case, the $J$ and $E$-terms are 
\bea
J \qquad \qquad &\qquad\qquad  E\\
\Lambda_{1}:X_{2}X_{3}-X_{3}X_{2}\quad  & \quad  X_{4}X_{1}-X_{1}X_{4}\\
\Lambda_{2}:X_{3}X_{1}-X_{1}X_{3}\quad  & \quad  X_{4}X_{2}-X_{2}X_{4}\\
\Lambda_{3}:X_{1}X_{2}-X_{2}X_{1}\quad  & \quad X_{4}X_{3}-X_{3}X_{4}.
\eea
The vacuum of this potential is obtained from the vanishing $J$ and $E$-terms which come from $\partial W/\partial \Lambda_{a}=0$. This is called the F-term condition.

Physically, the parameters $q_{i}\,(i=1,2,3,4)$ are identified with the $\U(1)^{4} \subset \mathrm{Spin}(8)$ rotational symmetries of the $\mathbb{C}^{4}$. For the vacuum defined from $J_{a}=E_{a}=0$ to be invariant under the rotational symmetry, we need $q_{i}q_{j}=q_{j}q_{i}$. Each monomial term of the $W$ also needs to be invariant under this symmetry. For example, from $\Lambda_{1}X_{2}X_{3}$ we can see that the charge of $\Lambda_{1}$ should be $q_{2}^{-1}q_{3}^{-1}$. Then, the charge of $\bar{\Lambda}_{1}$ will be $q_{2}q_{3}$ and for the term $\bar{\Lambda}_{1}X_{4}X_{1}$ coming from the $E$-term to be invariant, we need $q_{2}q_{3}q_{1}q_{4}=1$. Other terms will give the same condition and this is the Calabi--Yau condition imposed on the $q$-deformation parameters.

The other half $q$-Cartan is obtained as 
\bea
c^{-}_{1234}=1-\sum_{a\in\four}q_{a}^{-1}+q_{4}^{-1}\sum_{i=1}^{3}q_{i}^{-1}.
\eea
Obviously, looking at the parameters assigned, we can see that the second term corresponds with the conjugate of the chiral superfields, while the third term corresponds with the conjugate of the Fermi superfields. Therefore, we can say that the total $q$-Cartan matrix $c_{1234}=c^{+}_{1234}+c^{-}_{1234}$ is associated with the graded quiver illustrated in the double arrows notation as in \eqref{eq:doublearrows}.

\paragraph{$\mathbb{C}^{3}$-case}
Let us move on to the $\mathbb{C}^{3}$-case. The dual gauge theory is the maximally supersymmetric Yang--Mills in 4d which is the 4d $\mathcal{N}=4$ theory. In the 4d $\mathcal{N}=1$ language, there are one vector superfield and three adjoint chiral superfields $X'_{1},X'_{2},X'_{3}$. The quiver diagram is described as 
\bea
\newcommand{\arrowHeadPosition}{0.7}
\begin{tikzpicture}[decoration={markings,mark=at position \arrowHeadPosition with {\arrow{latex}}}]
\begin{scope}
    \draw[postaction={decorate}, black,thick,scale=1.3] (0,0) arc(360:0:0.80 and 0.4) ;
    \draw[postaction={decorate}, black,thick,scale=1.3] (0,0) arc(360:0:0.70 and 0.3) ;
    \draw[postaction={decorate}, black,thick,scale=1.3] (0,0) arc(360:0:0.60 and 0.2) ;
    \draw[fill=black!20!white,thick](0,0) circle(0.4cm);
    \node [thick,left] at (-2.3,0) {$X'_{1},X'_{2},X'_{3}$};
\end{scope}
\end{tikzpicture}
\eea
Given this quiver diagram, we can read the half $q$-Cartan matrix as 
\bea
c_{\mathbb{C}^{3}}^{+}=1-(X'_{1}+X'_{2}+X'_{3}).
\eea
The three chiral superfields are identified with the three complex coordinates of the transverse $\mathbb{C}^{3}$. Depending on which subspace $\mathbb{C}^{3}_{\bar{a}}\,(a=1,2,3,4)$ we are considering, the identification with the $q$-deformation parameters will be different:
\bea
\mathbb{C}^{3}=\mathbb{C}^{3}_{123}:\quad c_{123}^{+}=1-(q_{1}+q_{2}+q_{3}),\\
\mathbb{C}^{3}=\mathbb{C}^{3}_{124}:\quad c_{124}^{+}=1-(q_{1}+q_{2}+q_{4}),\\
\mathbb{C}^{3}=\mathbb{C}^{3}_{134}:\quad c_{134}^{+}=1-(q_{1}+q_{3}+q_{4}),\\
\mathbb{C}^{3}=\mathbb{C}^{3}_{234}:\quad c_{234}^{+}=1-(q_{2}+q_{3}+q_{4}).
\eea

Since, in the gauge origami system, we are not imposing the Calabi--Yau condition on $\mathbb{C}^{3}$ but only on $\mathbb{C}^{4}$, the three $q$-deformation parameters of the half $q$-Cartan matrix are independent. In other words, we start from $\mathbb{C}^{4}$ and impose the CY condition only on the 2d $\mathcal{N}=(0,2)$ graded quiver. The 4d $\mathcal{N}=1$ quivers are understood as a subquiver of the CY$_{4}$ quiver and with $q$-parameters associated. We then read the $q$-Cartan matrix from it. 

The other half $q$-Cartan matrix and the total $q$-Cartan matrix can be understood similarly but we need to be careful about how we assign the $\U(1)$ charges since we are not imposing the CY condition on $\mathbb{C}^{3}$. We will not discuss it in this paper. A detailed discussion of this part will be done in \cite{Kimura-Noshita}.

Let us comment on what will happen when we impose the CY condition on $\mathbb{C}^{3}_{123}$. The superpotential of the theory is given by 
\bea
W=X'_{1}[X'_{2},X'_{3}].
\eea
After identifying the chiral superfields as $X'_{i}\leftrightarrow \mathsf{q}_{i}\,\,(i=1,2,3)$, a similar analysis as the 2d case shows that we need the conditions $\mathsf{q}_{i}\mathsf{q}_{j}=\mathsf{q}_{j}\mathsf{q}_{i}$ and $\mathsf{q}_{1}\mathsf{q}_{2}\mathsf{q}_{3}=1$. Therefore, we only have two independent parameters. Actually, such kind of parameter assignment was done similarly in \cite{Li:2020rij,Galakhov:2020vyb} where the authors defined the quiver Yangian using melting crystal methods \cite{Ooguri:2009ijd,Iqbal:2003ds,Iqbal:2007ii,Aganagic:2003db,Okounkov:2003sp} and brane tiling techniques \cite{Douglas:1996sw, Kennaway:2007tq,Hanany:1997tb,Franco:2005rj,Hanany:2005ss,Ooguri:2009ijd}. Discussion on the relationship with such algebras will be discussed in section~\ref{sec:toroidal_alg}.

\paragraph{$\mathbb{C}^{2}$-case}
The dual gauge theory is the 6d $\mathcal{N}=(1,0)$ SYM. The quiver is a single node with a single unoriented arrow. The node corresponds to the vector multiplet, while the unoriented arrow corresponds to the hypermultiplet denoted as $\Phi$:
\bea
\newcommand{\arrowHeadPosition}{0.7}
\begin{tikzpicture}[decoration={markings,mark=at position \arrowHeadPosition with {\arrow{latex}}}]
\begin{scope}
    \draw[ black,thick,scale=1.3] (0,0) arc(360:0:0.80 and 0.4) ;
    \draw[fill=black!20!white,thick](0,0) circle(0.4cm);
    \node [thick,left] at (-2.3,0) {$\Phi$};
\end{scope}
\end{tikzpicture}
\eea
The hypermultiplet corresponds to either the two complex coordinates of the transverse $\mathbb{C}^{2}$. Let us consider the case that the transverse geometry is $\mathbb{C}_{12}^{2}$. We then can identify the half $q$-Cartan matrix as 
\bea
c_{12}^{+}=1-q_{1}.
\eea
Note for this case, since the arrow is unoriented, it is arbitrary to choose either $c_{12}^{+}=1-q_{1}$ or $c_{12}^{+}=1-q_{2}$ to be the half $q$-Cartan matrix. 

Depending on which subspace $\mathbb{C}^{2}_{A}\,(A\in \six)$ we are considering, we have the following six types of $q$-Cartan matrices:
\bea
\mathbb{C}^{2}=\mathbb{C}^{2}_{ab}:\quad c_{ab}^{+}=1-q_{a},
\eea
where we chose one of the indices to define the $q$-Cartan matrix.

We also note that when we impose the Calabi--Yau condition on the $\mathbb{C}^{2}$ part, we have only one deformation parameter and the total $q$-Cartan matrix is $c=(1-\mathsf{q})(1-\mathsf{q}^{-1})$. This corresponds to the unrefined limit in the context of 4d (5d) gauge theory.

\paragraph{$\mathbb{C}$-case} For this case, we do not know the corresponding gauge theory and the gauge theoretic origin. However, from the analogy of the discussions before, we expect the quiver should be drawn with one node and no arrows:
\bea
\adjustbox{valign=c}{\begin{tikzpicture}
\begin{scope}
    \draw[fill=black!20!white,thick](0,0) circle(0.4cm);
\end{scope}
\end{tikzpicture}}
\eea
We assign the half $q$-Cartan matrix to this quiver as $c^{+}=1$.

\paragraph{Quiver to algebra}
Given the quiver structure and the corresponding $q$-Cartan matrix of each subspace of $\mathbb{C}^{4}$, we can construct the vertex operators introduced in sections \ref{sec:M4LMNS}, \ref{sec:tetraLMNS}, \ref{sec:spikedLMNS}, \ref{sec:cplvortLMNS}. The $q$-Cartan matrix dual to the entire $\mathbb{C}^{4}$ geometry gives the commutation relations of the root current:
\bea
\relax [\mathsf{a}_{n},\mathsf{a}_{m}]=-\frac{1}{n}\delta_{n+m,0}c_{1234}^{[n]}.
\eea
For general toric CY$_{4}$, the modes of the root currents will be modified as $\mathsf{a}_{n}\rightarrow \{\mathsf{a}_{i,n}\}_{i\in \overbar{Q}_{0}}$ where $\overbar{Q}_{0}$ is the set of nodes of the corresponding quiver (see the next subsection) and the right-hand side of the commutation relations will be the $q$-Cartan matrix. Modes of other vertex operators corresponding to a subspace $\mathcal{S}$ comes from 
\bea
\mathsf{a}_{n}=\mathcal{A}_{n}c_{\mathcal{S}}^{[-n]}
\eea
where $c_{\mathcal{S}}$ is the corresponding $q$-Cartan matrix. In this sense, the definition of all of the vertex operators just comes from the corresponding graded quiver $q$-Cartan matrix, and thus it is a generalization of the quiver W-algebra introduced in \cite{Kimura:2015rgi}.

\subsubsection{CY$_{4}$}\label{sec:CY4}
Let us study what will happen for toric Calabi--Yau four-folds. Let $Z$ be a toric CY$_{4}$ and $\overbar{Q}=(\overbar{Q}_{0},\overbar{Q}_{1})$ the corresponding graded quiver. The dual quiver is classified\footnote{Such theories can be compactly summarized in a 3d model called brane brick models \cite{Franco:2015tna,Franco:2015tya,Franco:2016nwv,Franco:2017cjj,Franco:2017lpa} which are generalizations of brane box models \cite{Garcia-Compean:1998sla,Garcia-Compean:1998sla}. These brane brick models are generalizations of the brane tilings \cite{Douglas:1996sw, Kennaway:2007tq,Hanany:1997tb,Franco:2005rj,Hanany:2005ss,Ooguri:2009ijd,Franco:2005sm} used to describe 4d $\mathcal{N}=1$ quiver gauge theories dual to toric Calabi--Yau three-folds (see \cite{Yamazaki:2008bt} for a review and references).} by 2d $\mathcal{N}=(0,2)$ quiver gauge theories. The oriented arrows of this graded quiver correspond to chiral superfields while the unoriented arrows correspond to Fermi superfields. The isometry group of a toric CY$_{4}$ contains $\U(1)^{4}$ and the superfields have charges of them. For each matter superfield, we can assign the $\U(1)^{4}$ charges of the toric CY$_{4}$. One linear combination of them corresponds to the $\U(1)$ R-symmetry of the 2d theory while the left non-R $\U(1)^{3}$ symmetries are called the \emph{mesonic flavor symmetry} \cite{Franco:2015tna,Franco:2015tya,Franco:2016nwv,Franco:2017cjj,Franco:2017lpa}. For each superfield, we associate the $q$-deformation parameters corresponding to this mesonic flavor symmetry. Namely, if we have a superfield with a $\U(1)^{3}$ charge $(a,b,c)$, we associate $q_{1}^{a}q_{2}^{b}q_{3}^{c}$, where $q_{1},q_{2},q_{3}$ represent the three independent $\U(1)$ charges. 

These $q$-deformation parameters associated with the arrows are denoted as $\{q_{I}\}_{I\in \overbar{Q}_{1}}$. We denote the set of oriented arrows (chiral superfields) as $\overbar{Q}_{1}^{(0)}$ and the set of unoriented arrows (Fermi superfields) as $\overbar{Q}_{1}^{(1)}$. Under this decomposition, the $q$-deformation parameters are decomposed into $\{q_{I}\}_{\overbar{Q}_{1}}=\{q_{I}^{(0)}\}_{I\in \overbar{Q}_{1}^{(0)}}\cup \{q_{I}^{(1)}\}_{I\in \overbar{Q}_{1}^{(1)}}$.

The half $q$-Cartan matrix is then given as 
\bea
c_{Z,ij}^{+}=\delta_{ij}-\sum_{I\in\{j\rightarrow i\}}q_{I}^{(0)}+\sum_{I\in \{j\rightarrow i\}}q_{I}^{(1)}.
\eea
The total $q$-Cartan matrix will be 
\bea
c_{Z,ij}&=2\delta_{ij}-\sum_{I\in\{j\rightarrow i\}}q_{I}^{(0)}+\sum_{I\in \{j\rightarrow i\}}q_{I}^{(1)}\\
&\quad -\sum_{I\in\{i\rightarrow j\}}{q_{I}^{(0)}}^{-1}+\sum_{I\in \{i\rightarrow j\}}{q_{I}^{(1)}}^{-1}
\eea
which obeys $c_{Z,ij}^{\vee}=c_{Z,ji}$. We can also obtain the structure function as 
\bea
\mathcal{A}_{Z,ij}(x)=\mathbb{I}[-c_{Z,ij}^{\vee}x^{\vee}]=\frac{\prod\limits_{I\in\{j\rightarrow i\}}(1-q_{I}^{(0)}x)\prod\limits_{I\in\{i\rightarrow j\}}(1-{q_{I}^{(0)}}^{-1}x)}{(1-x)^{2\delta_{ij}}\prod\limits_{I\in\{j\rightarrow i\}}(1-q_{I}^{(1)}x)\prod\limits_{I\in\{i\rightarrow j\}}(1-{q_{I}^{(1)}}^{-1}x)}.
\eea

Let us move on to the operator formalism. When considering general toric CY$_{4}$, we have multiple quiver nodes, and thus the operators will be labeled by the quiver nodes. Based on the discussion of the $\mathbb{C}^{4}$ case, the D0-brane operator (root current) is defined as
\bea\label{eq:CY4D0vertexop}
\mathsf{A}_{i}(x)=\mathsf{a}_{i,0}(x):\exp\left(\sum_{n\neq 0}\mathsf{a}_{i,n}x^{-n}\right):,\quad [\mathsf{a}_{i,n},\mathsf{a}_{j,m}]=-\frac{1}{n}\delta_{n+m,0}c_{Z,ij}^{[n]},
\eea
where $i\in \overbar{Q}_{0}$.
 We can also introduce a $\D8$-brane wrapping CY$_{4}$ whose corresponding vertex operator is defined as 
 \bea\label{eq:CY4D8vertexop}
\mathsf{Z}_{i}(x)=\mathsf{z}_{i,0}(x):\exp\left(\sum_{n\neq 0}\mathsf{z}_{i,n}x^{-n}\right):,\quad \mathsf{a}_{i,n}=\sum_{j\in\overbar{Q}_{0}}\mathsf{z}_{j,n}c_{Z,ji}^{[-n]}.
 \eea
We omit the discussion of the zero-modes. A detailed discussion will be done in a future publication.

The operator product formulas are given as
\bea
\mathsf{A}_{i}(x)\mathsf{A}_{j}(x')&=\wick{\c{\mathsf{a}_{i,0}(x)}\c{\mathsf{a}_{j,0}(x')}}\mathcal{A}_{Z,ij}\left(x'/x\right)^{-1}:\mathsf{A}_{i}(x)\mathsf{A}_{j}(x'):,\\
\mathsf{Z}_{i}(x)\mathsf{A}_{j}(x')&=\wick{\c{\mathsf{z}_{i,0}(x)}\c{\mathsf{a}_{j,0}(x')}}\left(1-x'/x\right)^{\delta_{ij}}:\mathsf{Z}_{i}(x)\mathsf{A}_{j}(x'):,\\
\mathsf{A}_{j}(x')\mathsf{Z}_{i}(x)&=\wick{\c{\mathsf{a}_{j,0}(x')}\c{\mathsf{z}_{i,0}(x)}}\left(1-x/x'\right)^{\delta_{ij}}:\mathsf{Z}_{i}(x)\mathsf{A}_{j}(x'):.
\eea
Using them, we can define a similar contour integral formula as the magnificent four system which we propose to be the generalization of the magnificent four system to general toric CY$_{4}$.
\begin{conjecture}[\cite{Kimura-Noshita}]\label{conj:CY4}
    Let $Z$ be a toric Calabi--Yau four-fold. We denote the corresponding graded quiver as $\overbar{Q}=(\overbar{Q}_{0},\overbar{Q}_{1})$ and the associated $q$-deformation parameters as $\{q_{I}\}_{I\in \overbar{Q}_{1}}$. We then have the following.
    \begin{enumerate}[topsep=1pt,itemsep=-1ex,partopsep=1ex,parsep=1ex]
        \item After imposing suitable conditions, the $q$-deformation parameters reduce up to three independent parameters.
        \item The corresponding D0 and D8-brane operators are defined as in \eqref{eq:CY4D0vertexop} and \eqref{eq:CY4D8vertexop}.
        \item The partition function where there are D8-branes wrapping the entire $Z$ is proportional to
        \bea
        \frac{1}{\underline{k}!}\oint \prod_{i\in \overbar{Q}_{0}}\prod_{I=1}^{k_{i}}\frac{dx_{i,I}}{2\pi\iota}\langle\mathsf{A}_{\underline{k}}^{-1}\mathsf{Z}_{\underline{n}}\rangle
        \eea
        where 
        \bea
\mathsf{A}_{\underline{k}}=\prod_{i\in\overbar{Q}_{0}}\prod_{I=1}^{k_{i}}\mathsf{A}_{i}(x_{i,I}),\quad \mathsf{Z}_{\underline{n}}={:\prod_{i\in \overbar{Q}_{0}}\prod_{\alpha=1}^{n_{i}}\mathsf{Z}_{i}(v_{i,\alpha}):}.
        \eea
        This gives the BPS/CFT correspondence of this setup.
        \item The poles are classified by 4d analogs of the 3d BPS crystals \cite{Ooguri:2009ijd}.
    \end{enumerate}
\end{conjecture}

Following the discussion of the $\mathbb{C}^{4}$-case, one would like to study the gauge origami system where lower dimensional D-branes appear. To do this practically, we need to study the subspaces of the toric CY$_{4}$ case-by-case. However, for special cases when the total CY$_{4}$ is decomposed into a product of smaller spaces, we can give a general discussion. These will be studied in the following subsections.

We note that there is another way to construct the magnificent four system where $Z$ is general toric $\CY_{4}$. One can understand $Z$ as a combination of $\mathbb{C}^{4}$ patches. The partition function of $Z$ is understood by gluing the results of $\mathbb{C}^{4}$. Studies from this approach were done in~\cite{Cao:2019tvv,Nekrasov:2023nai} where the authors introduced a vertex formalism. The relation with our construction here is not so clear for the moment.

\begin{remark}
    We expect the above conjecture is also applicable to \emph{non-toric} Calabi--Yau four-folds such as the non-Abelian orbifolds. Given the quiver and proper parameter association, we can define the $q$-Cartan matrix and the contour integral formula, but the evaluation of the poles might be non-trivial. 
\end{remark}

\subsubsection{CY$_{3}$ $\times$ $\mathbb{C}$}\label{sec:CY3timesC}
Consider the Calabi--Yau four-fold with the form $Z=X\times \mathbb{C}$, where $X$ is a toric CY three-fold (see also~\cite{Cao:2019tvv,Cao:2023lon,Piazzalunga:2023qik}). We denote the corresponding quiver\footnote{Strictly speaking, the quiver here should be the quiver of the entire CY four-fold $Z$. The quiver of $X$ is understood as a subquiver of $Z$. However, at the level of $q$-Cartan matrix, we only need to take the product of the $q$-Cartan matrix associated with $X$ and $\mathbb{C}$, so we use the quiver of $X$ which is associated with 4d $\mathcal{N}=1$ theories. } of $X$ as $Q=(Q_{0},Q_{1})$. For each arrow $I\in Q_{1}$ of the quiver diagram, we assign a $q$-deformation parameter $q_{I}$. Given the geometry, the F-term condition of the superpotential determines the relations between the $q$-deformation parameters. After this condition, the number of independent parameters will be three. In this paper, we will not give the procedure to determine them. Practically, we need to study case by case.  

We denote the corresponding $q$-Cartan matrix as 
\bea
&c_{Z,ij}=c_{X,ij}(1-q_{4}),\quad c_{X,ij}=c_{X,ij}^{+}+c_{X,ij}^{-},\quad c_{X,ij}^{-}=-q_{4}^{-1}c^{+\vee}_{X,ji},
\eea
where 
\beq
c_{X,ij}^{+}=\delta_{ij}-\sum_{I\in\{j\rightarrow i\}}q_{I}.
\eeq
We are assuming that the $q$-deformation parameters satisfy the F-term conditions. Note also that we have the property $c_{Z,ij}=c_{Z,ji}^{\vee}$ which is the reality condition of the $q$-Cartan matrix. The structure functions associated with $Z$ and $X$ are
\bea
&\mathcal{A}_{Z,ij}(x)=\mathbb{I}[-c_{Z,ij}^{\vee}x^{\vee}]=\frac{(1-q_{4}x)^{\delta_{ij}}(1-q_{4}^{-1}x)^{\delta_{ij}}\prod\limits_{I\in\{j\rightarrow i\}}(1-q_{I}x)\prod\limits_{I\in\{i\rightarrow j\}}(1-q_{I}^{-1}x)}{(1-x)^{2\delta_{ij}}\prod\limits_{I\in\{i\rightarrow j\}}(1-q_{I}^{-1}q_{4}^{-1}x)\prod\limits_{I\in\{j\rightarrow i\}}(1-q_{4}q_{I}x)},\\
&\varphi_{X,ij}(x)=\mathbb{I}[-c_{X,ij}^{\vee}x^{\vee}]=\frac{(1-q_{4}^{-1}x)^{\delta_{ij}}\prod\limits_{I\in\{j\rightarrow i \}}(1-q_{I}x)}{(1-x)^{\delta_{ij}}\prod\limits_{I\in\{i\rightarrow j\}}(1-q_{4}^{-1}q_{I}^{-1}x)}.
\eea

We then can introduce the D0-brane operators (root currents) and the D6-brane operator associated with $X$. The operators are labeled by the quiver nodes $i\in Q_{0}$ as
\bea\label{eq:CY3D0D6vertexoperator}
\mathsf{A}_{i}(x)=\mathsf{a}_{i,0}(x):\exp\left(\sum_{n\neq 0}\mathsf{a}_{i,n}x^{-n}\right):,\quad 
\mathsf{W}_{i}(x)=\mathsf{w}_{i,0}(x):\exp\left(\sum_{n\neq 0}\mathsf{w}_{i,n}x^{-n}\right):
\eea
where 
\bea
\relax[\mathsf{a}_{i,n},\mathsf{a}_{j,m}]=-\frac{1}{n}\delta_{n+m,0}c_{Z,ij}^{[n]},\quad \mathsf{a}_{i,n}=\sum_{j\in Q_{0}}\mathsf{w}_{j,n}c_{X,ji}^{[-n]}.
\eea
For later use, we also define the screening current corresponding to the $\mathbb{C}$ part of $Z$:
\bea
\mathsf{S}_{i}(x)=\mathsf{s}_{i,0}(x):\exp\left(\sum_{n\neq 0}\mathsf{s}_{i,n}x^{-n}\right):,\quad \mathsf{s}_{i,n}=\frac{\mathsf{a}_{i,n}}{1-q_{4}^{-n}}.
\eea
The zero-modes need to obey similar conditions as what we did for the $\mathbb{C}^{4}$ case in section~\ref{sec:zeromodes}, but we do not discuss them in this paper. A detailed derivation will be done in \cite{Kimura-Noshita}.

The operator products of the operators are
\bea
\mathsf{W}_{i}(x)\mathsf{S}_{j}(x')&=\wick{\c{\mathsf{w}_{i,0}(x)}\c{\mathsf{s}_{j}(x')}}\frac{1}{(1-q_{4}^{-1}x'/x)^{\delta_{ij}}}:\mathsf{W}_{i}(x)\mathsf{S}_{j}(x'):,\\
\mathsf{S}_{j}(x')\mathsf{W}_{i}(x)&=\wick{\c{\mathsf{s}_{j,0}(x')}\c{\mathsf{w}_{i,0}(x)}}\frac{1}{(1-q_{4}x/x')^{\delta_{ij}}}:\mathsf{W}_{i}(x)\mathsf{S}_{j}(x'):,\\
\mathsf{W}_{i}(x)\mathsf{A}_{j}(x')&=\wick{\c{\mathsf{w}_{i,0}(x)}\c{\mathsf{a}_{j,0}(x')}}\left(\frac{1-x'/x}{1-q_{4}^{-1}x'/x}\right)^{\delta_{ij}}:\mathsf{W}_{i}(x)\mathsf{A}_{j}(x'): \\
\mathsf{A}_{j}(x')\mathsf{W}_{i}(x)&=\wick{\c{\mathsf{a}_{j,0}(x')}\c{\mathsf{w}_{i,0}(x)}}\left(\frac{1-x/x'}{1-q_{4}x/x'}\right)^{\delta_{ij}}:\mathsf{W}_{i}(x)\mathsf{A}_{j}(x'):,\\
\mathsf{A}_{i}(x)\mathsf{A}_{j}(x')&=\wick{\c{\mathsf{a}_{i,0}(x')}\c{\mathsf{a}_{j,0}(x')}}\mathcal{A}_{Z,ij}\left(x'/x\right)^{-1}:\mathsf{A}_{i}(x)\mathsf{A}_{j}(x'):,\\
\mathsf{A}_{i}(x)\mathsf{S}_{j}(x')&=\wick{\c{\mathsf{a}_{i,0}(x')}\c{\mathsf{s}_{j,0}(x')}}\varphi_{X,ij}\left(x'/x\right)^{-1}:\mathsf{A}_{i}(x)\mathsf{S}_{j}(x'):,\\
\mathsf{S}_{j}(x')\mathsf{A}_{i}(x)&=\wick{\c{\mathsf{s}_{j,0}(x')}\c{\mathsf{a}_{i,0}(x')}}\varphi_{X,ji}\left(q_{4}x/x'\right):\mathsf{S}_{j}(x')\mathsf{A}_{i}(x):.
\eea
Using this, we can define the contour integral formula that computes the partition function of the system where multiple D6-branes wrap the subspace $X$ of the total CY$_{4}$ $Z$. Let us summarize what we have obtained as a conjecture. 

\begin{conjecture}[\cite{Kimura-Noshita}]\label{conj:CY3}
Let $Z$ be the Calabi--Yau four-fold with the form $Z=X\times \mathbb{C}$, where $X$ is a toric CY$_{3}$. Denote the corresponding quiver as $Q=(Q_{0},Q_{1})$ and the associated $q$-deformation parameters as $\{q_{I}\}_{I\in Q_{1}}$. We then have the following.
\begin{enumerate}[topsep=1pt,itemsep=-1ex,partopsep=1ex,parsep=1ex]
    \item After imposing suitable conditions, we can reparametrize the parameters and reduce them up to three independent $q$-deformation parameters.
\item The corresponding D0 and D6-brane operators are $\mathsf{A}_{i}(x)$ and $\mathsf{W}_{i}(x)$ defined in \eqref{eq:CY3D0D6vertexoperator}.
\item The partition function of the tetrahedron instanton system where multiple D6-branes wrap the $X$ subspace is proportional to
\bea
\frac{1}{\underline{k}!}\oint \prod_{i\in Q_{0}}
\prod_{I=1}^{k_{i}}\frac{dx_{i,I}}{2\pi\iota x_{i,I}}\langle\mathsf{A}_{\underline{k}}^{-1}\mathsf{W}_{\underline{n}}\rangle
\eea
where 
\bea 
\mathsf{A}_{\underline{k}}=\prod_{i\in Q_{0}}\prod_{I=1}^{k_{i}}\mathsf{A}_{i}(x_{i,I}),\quad \mathsf{W}_{\underline{n}}={:\prod_{i\in Q_{0}}\prod_{\alpha=1}^{n_{i}}\mathsf{W}_{i}(v_{i,\alpha}):}.
\eea
Moreover, this gives the BPS/CFT correspondence of this system.
\item The poles of this contour integral formula are labeled by the 3d BPS crystals of \cite{Ooguri:2009ijd,Yamazaki:2008bt}. 
\end{enumerate}
\end{conjecture}

\begin{remark}
After taking the limit $q_{4}\rightarrow 1$, assume the $q$-deformation parameters are transformed to $\{q_{I}\}_{I\in Q_{1}}\rightarrow \{\mathsf{q}_{I}\}_{I\in Q_{1}}$. There will be only two independent parameters in $\{\mathsf{q}_{I}\}_{I\in Q_{1}}$. Under this limit, we have 
\beq\label{eq:CY3timesCstructurefunct}
\varphi_{X,ij}(x)\xrightarrow{q_{4}\rightarrow 1} \varphi^{i\Rightarrow j}(x)\coloneqq\frac{\prod_{I\in\{j\rightarrow i\}}(1-\mathsf{q}_{I}x)}{\prod_{I\in\{i\rightarrow j\}}(1-\mathsf{q}_{I}^{-1}x)},
\eeq
where the right-hand side is the structure functions\footnote{The notation of the structure functions differs from the original papers. They come from how we define the explicit zero-modes of the operators.} of the toroidal quiver algebra/quiver quantum toroidal algebra introduced in \cite{Li:2020rij,Galakhov:2020vyb,Galakhov:2021vbo,Galakhov:2021xum,Noshita:2021ldl,Noshita:2021dgj}.
\end{remark}

\subsubsection{\texorpdfstring{$\CY_2\times \CY_2$}{CY2 CY2} }\label{sec:CY2CY2}
Let us next consider the case when the CY$_{4}$ $Z$ takes the form of $\CY_2\times \CY_2$. The corresponding quiver structure for such geometry is decomposed into a product of two quivers: $\overbar{Q}=\Upsilon\times \Gamma $. We denote the nodes and arrows of $\Gamma,\Upsilon$ as $\Gamma_{0,1},\Upsilon_{0,1}$. Each node of $\overbar{Q}$ is written as $(i,j)\in \Upsilon_{0}\times \Gamma_{0}$. Each arrow of $\overbar{Q}$ is decomposed into the following two types:
\bea
I:(i,j)\rightarrow (i',j'),\,\,(j=j'), \qquad I:(i,j)\rightarrow (i',j'),\,\,(i=i').
\eea
The first arrow is identified with an element of $\Upsilon_{1}$ while the second arrow is identified with an element of $\Gamma_{1}$. Note that the quivers $\Upsilon,\Gamma$ are dual to 6d $\mathcal{N}=(1,0)$ theories. We focus on the case when the quiver structures are affine A, D, E type and the CY$_{4}$ is denoted as $\mathbb{C}^{2}/\Upsilon\times \mathbb{C}^{2}/\Gamma$. See \cite{Nekrasov:2016ydq} for the formulas of the gauge origami partition function.

We denote the $q$-deformation parameters associated with the arrows of $\Upsilon_{1}, \Gamma_{1}$ as $\{q_{I}\}_{I\in \Upsilon_{1}\cup\Gamma_{1}}$. Following the discussions of previous subsections, the $q$-Cartan matrix can be written as
\bea
&c_{Z,ab}=c_{\Upsilon,ii'}c_{\Gamma,jj'},\quad a=(i,j),\quad b=({i'},{j'}),\\
&c_{\Upsilon,ii'}=(1+q_{12})\delta_{ii'}-\sum_{I\in \Upsilon_{1}:i'\rightarrow i}q_{I}-\sum_{I\in\Upsilon_{1}:i\rightarrow i'}q_{12}q_{I}^{-1},\\
&c_{\Gamma,jj'}=(1+q_{34})\delta_{jj'}-\sum_{I\in \Gamma_{1}:j'\rightarrow j}q_{I}-\sum_{I\in \Gamma_{1}:j\rightarrow j'} q_{34}q_{I}^{-1}
\eea
where the CY condition is imposed as $q_{12}q_{34}=1$. Explicitly, we have
\bea
c_{Z,ab}&=(2+q_{12}+q_{34})\delta_{ii'}\delta_{jj'}\\
&+\sum_{\substack{I:i'\rightarrow i\\J:j'\rightarrow j}}q_{I}q_{J}+\sum_{\substack{I:i'\rightarrow i\\J:j\rightarrow j'}}q_{I}q_{J}^{-1}q_{34}+\sum_{\substack{I:i\rightarrow i'\\J:j'\rightarrow j}}q_{12}q_{I}^{-1}q_{J}+\sum_{\substack{I:i\rightarrow i'\\J:j\rightarrow j'}}q_{I}^{-1}q_{J}^{-1}\\
&-\sum_{I:i'\rightarrow i}q_{I}(1+q_{34})\delta_{jj'}-\sum_{I:i\rightarrow i'}(1+q_{12})q_{I}^{-1}\delta_{jj'}-\sum_{I:j'\rightarrow j}q_{I}(1+q_{12})\delta_{ii'}-\sum_{I:j\rightarrow j'}(1+q_{34})q_{I}^{-1}\delta_{ii'}
\eea
Note that we have the property $c_{Z,ab}^{\vee}=c_{Z,ba}$. We also assume the other parameters $q_{I}$ obey the F-term conditions. The $q$-Cartan matrices arising here are the $q$-Cartan matrices of the double quiver gauge theory recently introduced in \cite{Kimura:2022zsm}. 

The structure functions are given as 
\bea
\mathcal{A}_{Z,ab}\left(x\right)&=\mathbb{I}[-c_{Z,ab}^{\vee}x^{\vee}],\\
\mathscr{S}_{\Upsilon,ii'}(x)&=\mathbb{I}[-c_{\Upsilon,ii'}^{\vee}x^{\vee}]=\frac{\prod\limits_{I:i'\rightarrow i}(1-q_{I}x)\prod\limits_{I:i\rightarrow i'}(1-q_{12}q_{I}^{-1}x)}{(1-x)^{\delta_{ii'}}(1-q_{12}x)^{\delta_{ii'}}},\\
\mathscr{S}_{\Gamma,jj'}(x)&=\mathbb{I}[-c_{\Gamma,jj'}^{\vee}x^{\vee}]=\frac{\prod\limits_{I:j'\rightarrow j}(1-q_{I}x)\prod\limits_{I:j\rightarrow j'}(1-q_{34}q_{I}^{-1}x)}{(1-x)^{\delta_{jj'}}(1-q_{34}x)^{\delta_{jj'}}}.
\eea

Let us introduce the D0 and D4-brane operators. This time, the operators will have a double index labeled by $\Upsilon_{0}$ and $\Gamma_{0}$. The D0-brane operators are introduced as 
\bea
\mathsf{A}_{j}^{i}(x)=\mathsf{a}_{j,0}^{i}(x):\exp\left(\sum_{n\neq 0}\mathsf{a}_{j,n}^{i}x^{-n}\right):,\quad [\mathsf{a}_{j,n}^{i},\mathsf{a}_{j',m}^{i'}]=-\frac{1}{n}\delta_{n+m,0}c_{\Upsilon,ii'}^{[n]}c_{\Gamma,jj'}^{[n]}
\eea
for $i,i'\in \Upsilon_{0},\,\,j,j'\in\Gamma_{0}$. For the D4-brane operators, we can introduce two types depending on whether they wrap either $\mathbb{C}^{2}/\Upsilon$ or $\mathbb{C}^{2}/\Gamma$:
\bea\label{eq:D4doublequiverdef}
\mathsf{X}_{\Upsilon,j}^{i}(x)&=\mathsf{x}_{\Upsilon,j,0}^{i}(x):\exp\left(\sum_{n\neq 0}\mathsf{x}_{\Upsilon,j,n}^{i}x^{-n}\right):,\quad \mathsf{a}_{i,n}^{j}=\sum_{k\in\Upsilon_{0}}\mathsf{x}_{\Upsilon,i,n}^{k}c_{\Upsilon,kj}^{[-n]},\\
\mathsf{X}_{\Gamma,j}^{i}(x)&=\mathsf{x}_{\Gamma,j,0}^{i}(x):\exp\left(\sum_{n\neq 0}\mathsf{x}_{\Gamma,j,n}^{i}x^{-n}\right):,\quad \mathsf{a}_{i,n}^{j}=\sum_{k\in\Gamma_{0}}\mathsf{x}_{\Gamma,k,n}^{j}c_{\Gamma,ki}^{[-n]}.
\eea
The operator products are given as
\bea
\mathsf{A}^{i}_{j}(x)\mathsf{A}^{i'}_{j'}(x')&=\wick[sep=5pt,offset=1.2em]{\c{\mathsf{a}_{j,0}^{i}(x)}\c{\mathsf{a}_{j',0}^{i'}(x')}}\mathcal{A}_{Z,ab}(x'/x)^{-1}:\mathsf{A}^{i}_{j}(x)\mathsf{A}^{i'}_{j'}(x'):,\quad a=(i,j),\,\,b=(i',j'),\\
\mathsf{A}^{i}_{j}(x)\mathsf{X}^{i'}_{\Upsilon,j'}(x')&=\wick[sep=5pt,offset=1.2em]{\c{\mathsf{a}_{j,0}^{i}(x)}\c{\mathsf{x}_{\Upsilon,j',0}^{i'}(x')}}\mathscr{S}_{\Gamma,jj'}(x'/x)^{-\delta_{ii'}}:\mathsf{A}^{i}_{j}(x)\mathsf{X}^{i'}_{\Upsilon,j'}(x'):,\\
\mathsf{A}^{i}_{j}(x)\mathsf{X}^{i'}_{\Gamma,j'}(x')&=\wick[sep=5pt,offset=1.2em]{\c{\mathsf{a}_{j,0}^{i}(x)}\c{\mathsf{x}_{\Gamma, j',0}^{i'}(x')}}\mathscr{S}_{\Upsilon,ii'}(x'/x)^{-\delta_{jj'}}:\mathsf{A}^{i}_{j}(x)\mathsf{X}^{i'}_{\Gamma,j'}(x'):.
\eea
The gauge origami system of this setup was studied for some examples in \cite{Kimura:2022zsm} and the free field realization of the contour integral formula is
\bea
\frac{1}{\underline{k}!}\oint \prod_{\substack{j\in \Gamma_{0}\\i\in \Upsilon_{0}}}\prod_{I=1}^{k_{j}^{i}}\frac{dx_{j,I}^{i}}{2\pi\iota x^{i}_{j,I}}\langle \mathsf{A}_{\underline{k}}^{-1}\mathsf{X}_{\Upsilon,\underline{n}}\mathsf{X}_{\Gamma,\underline{n'}}\rangle,
\eea
where 
\bea
\mathsf{A}_{\underline{k}}=\prod_{\substack{i\in \Upsilon_{0}\\j\in \Gamma_{0}}}\prod_{I=1}^{k^{i}_{j}}\mathsf{A}^{i}_{j}(x^{i}_{j,I}),\quad \mathsf{X}_{\Upsilon,\underline{n}}={:\prod_{\substack{i\in \Upsilon_{0}\\j\in \Gamma_{0}}}\prod_{\alpha=1}^{n^{i}_{j}}\mathsf{X}^{i}_{j}(\nu^{i}_{j,\alpha}):},\quad \mathsf{X}_{\Gamma,\underline{n}'}={:\prod_{\substack{i\in \Upsilon_{0}\\j\in \Gamma_{0}}}\prod_{\alpha=1}^{{n'}^{i}_{j}}\mathsf{X}^{i}_{j}({\nu'}^{i}_{j,\alpha}):}.
\eea

\begin{remark}
Denoting the above system as $\mathbb{C}_{12}^{2}/\Upsilon\times \mathbb{C}^{2}_{34}/\Gamma$, since we are assigning stacks of D-branes in a transverse way, the arising instantons are crossed instantons. Let us focus on the case $\Upsilon=\mathbb{Z}_{p}$ and $\Gamma=\mathbb{Z}_{q}$ whose orbifold action on the complex four coordinates are
\bea
(z_{1},z_{2},z_{3},z_{4})\mapsto (\omega^{r_{1}} z_{1},\omega^{r_{2}}z_{2}, \sigma ^{s_{1}}z_{3},\sigma^{s_{2}}z_{4}),\quad \omega=e^{2\pi\iota/p},\quad \sigma=e^{2\pi\iota/q}
\eea
where $r_{1}+r_{2}\equiv 0\,\,(\operatorname{mod}p)$ and $s_{1}+s_{2}\equiv0\,\,(\operatorname{mod}q)$. When there are only D4-branes on the $\mathbb{C}^{2}_{12}/\mathbb{Z}_{p}$ surface, the arising instanton partition function is the $\mathbb{Z}_{p}$ invariant part of the partition function of the $\widehat{A}_{0}$ quiver gauge theory. This is the same situation for the D4-branes wrapping $\mathbb{C}^{2}_{34}/\mathbb{Z}_{q}$, where the $\mathbb{Z}_{q}$ invariant part appears as the result. Generally, we can also consider the case where D4-branes wrap other complex two-dimensional surfaces such as $\mathbb{C}_{2}/\mathbb{Z}_{p}\times \mathbb{C}_{3}/\mathbb{Z}_{q}$, where the orbifold action mixes with each other. For example when $q=1$ ($\mathbb{C}_{2}/\mathbb{Z}_{p}\times \mathbb{C}_{3}$), the arising partition function corresponds to instantons with surface operators \cite{Kanno:2011fw}. Including such D-branes enables us to consider the most general spiked instanton setup of the entire $\mathbb{C}_{12}^{2}/\mathbb{Z}_{p}\times \mathbb{C}^{2}_{34}/\mathbb{Z}_{q}$. The D4 vertex operators appearing for this case can not be written as \eqref{eq:D4doublequiverdef} but both the $q$-Cartan matrices $c_{\Upsilon,ij},\,c_{\Gamma,ij}$ mix with each other. A detailed analysis of this will be done in \cite{Kimura-Noshita}. We also note that starting from affine type $q$-Cartan matrices and taking a specific limit of the $q$-parameters, one can obtain the free field realizations of the contour integral formulas for finite type $q$-Cartan matrices (see \cite{Kimura:2022zsm} for examples and details).
\end{remark}


\section{D2-brane \texorpdfstring{$qq$}{qq}-characters}\label{sec:D2qqcharacter}
We introduce the screening charges and discuss the relation with the D2 $qq$-characters in section~\ref{sec:D2qqscreening}. We then use these D2 $qq$-characters to show that they reproduce the partition function of the gauge origami system of D2-branes in section~\ref{sec:cplvortD2qq}.

\subsection{D2 \texorpdfstring{$qq$}{qq}-characters and screening charges}\label{sec:D2qqscreening}
\begin{definition}
    We define the \textbf{screening charges} as 
    \beq
        \mathscr{Q}_{a}(x)=\sum_{k\in\mathbb{Z}}\mathsf{S}_{a}(q_{a}^{k}x)
    \eeq
    for $a\in\four$. We also call these screening charges the \textbf{$\D2$ $qq$-characters} whose motivation will be discussed later. 
\end{definition}
\begin{theorem}\label{thm:D2qq-commute}
    The screening charges with different indices commute with each other: 
    \beq
        \relax[\mathscr{Q}_{a}(x),\mathscr{Q}_{b}(x')]=0,\quad a\neq b,
    \eeq
    for $a,b\in\four$.
\end{theorem}
 \begin{proof}
     Let us focus on the commutation relation between $\mathscr{Q}_{1}(x)$ and $\mathscr{Q}_{4}(x)$. Other cases are obtained by using the quadrality. The operator products between $\mathsf{S}_{1}(x)$ and $\mathsf{S}_{4}(x)$ are 
     \begin{align}
     \begin{split}
    \mathsf{S}_{4}(x)\mathsf{S}_{1}(x')&=\mathscr{S}_{23}\left(q_{4}x'/x\right):\mathsf{S}_{4}(x)\mathsf{S}_{1}(x'):,\\
    \mathsf{S}_{1}(x')\mathsf{S}_{4}(x)&=\mathscr{S}_{23}(q_{1}x/x'):\mathsf{S}_{1}(x')\mathsf{S}_{4}(x):.
    \end{split}
\end{align}
We then have 
\begin{align}
\begin{split}
    &[\mathsf{S}_{4}(q_{4}^{k}x),\mathsf{S}_{1}(x')]\\
    =&\frac{(1-q_{2})(1-q_{3})}{(1-q_{2}q_{3})}\left(\delta\left(\frac{x'}{q_{4}^{k-1}x}\right)\textcolor{red}{:\mathsf{S}_{4}(q_{4}^{k}x)\mathsf{S}_{1}(q_{4}^{k-1}x):}-\delta\left(\frac{x'}{q_{1}q_{4}^{k}x}\right):\mathsf{S}_{4}(q_{4}^{k}x)\mathsf{S}_{1}(q_{1}q_{4}^{k}x):\right)\\
    =&\frac{(1-q_{2})(1-q_{3})}{(1-q_{2}q_{3})}\left(\delta\left(\frac{x'}{q_{4}^{k-1}x}\right)\textcolor{red}{:\mathsf{S}_{4}(q_{4}^{k-1}x)\mathsf{A}^{-1}(q_{4}^{k-1}x)\mathsf{S}_{1}(q_{4}^{k-1}x):}-\delta\left(\frac{x'}{q_{1}q_{4}^{k}x}\right):\mathsf{S}_{4}(q_{4}^{k}x)\mathsf{S}_{1}(q_{1}q_{4}^{k}x):\right)\\
    =&\frac{(1-q_{2})(1-q_{3})}{(1-q_{2}q_{3})}\left(\delta\left(\frac{x'}{q_{4}^{k-1}x}\right)\textcolor{red}{:\mathsf{S}_{4}(q_{4}^{k-1}x)\mathsf{S}_{1}(q_{1}q_{4}^{k-1}x):}-\delta\left(\frac{x'}{q_{1}q_{4}^{k}x}\right):\mathsf{S}_{4}(q_{4}^{k}x)\mathsf{S}_{1}(q_{1}q_{4}^{k}x):\right)
\end{split}
\end{align}
where we used 
\beq
    \mathsf{A}(x)={:\frac{\mathsf{S}_{1}(x)}{\mathsf{S}_{1}(q_{1}x)}:}={:\frac{\mathsf{S}_{4}(x)}{\mathsf{S}_{4}(q_{4}x)}:}.
\eeq
The commutation relation between the screening charge $\mathscr{Q}_{4}(x)$ and $\mathsf{S}_{1}(x')$ is 
\beq
    \relax[\mathscr{Q}_{4}(x),\mathsf{S}_{1}(x')]=\frac{(1-q_{2})(1-q_{3})}{(1-q_{2}q_{3})}\sum_{k\in\mathbb{Z}}\left(\delta\left(\frac{x'}{q_{4}^{k}x}\right)-\delta\left(\frac{x'}{q_{1}q_{4}^{k}x}\right)\right):\mathsf{S}_{4}(q_{4}^{k}x)\mathsf{S}_{1}(q_{1}q_{4}^{k}x):
\eeq
which is a total difference. Thus, we finally have 
\beq
   \relax [\mathscr{Q}_{4}(x),\mathscr{Q}_{1}(x')]=0.
\eeq
\end{proof}
We call the screening charges D2-brane $qq$-characters because each term is related by the iWeyl reflection (see \eqref{eq:D2iWeylpartition} for the iWeyl reflection in terms of partition functions).  We define the operator version of the iWeyl reflection in \eqref{eq:D2iWeylpartition} of the $\D2$-brane operator as the following.
\begin{definition}
    The \textbf{iWeyl reflection} of the $\D2$-brane operator is 
    \beq
        \mathsf{S}_{a}(x)\rightarrow {:\mathsf{S}_{a}(x)\mathsf{A}^{-1}(x):}=\mathsf{S}_{a}(q_{a}x)
    \eeq
    where we used \eqref{eq:oprelation1} in the second identitiy.
\end{definition}
Starting from the screening current $\mathsf{S}_{a}(x)$ and using the iWeyl reflection sequentially, we get a sequence 

\beq
    \mathsf{S}_{a}(x)\xrightarrow{\mathsf{A}^{-1}(x)} \mathsf{S}_{a}(q_{a}x)\xrightarrow{\mathsf{A}^{-1}(q_{a}x)} \cdots\cdots\mathsf{S}_{a}(q_{a}^{k}x)\xrightarrow{\mathsf{A}^{-1}(q_{a}^{k}x)} \cdots
\eeq
 which generates half of the screening charge:
 \beq
     \mathsf{T}^{+}_{a}(x)=\sum_{k\geq 0}\mathsf{S}_{a}(q_{a}^{k}x),\quad a\in\four.
 \eeq

To define the screening charge, the above $\mathsf{T}_{a}^{+}(x)$ is not enough and we need the other half part. It is generated by the iWeyl reflection in the opposite direction as
 \bea
     \cdots \xleftarrow{\mathsf{A}(q_{a}^{-l+1}x)}\mathsf{S}_{a}(q_{a}^{-l}x)\cdots\cdots\xleftarrow{\mathsf{A}(q_{a}^{-2}x)}\mathsf{S}_{a}(q_{a}^{-1}x)\xleftarrow{\mathsf{A}(q_{a}x)}\mathsf{S}_{a}(x)
 \eea
and gives the other half of the screening charge as
\bea
    \mathsf{T}_{a}^{-}(x)=\sum_{k<0}\mathsf{S}_{a}(q_{a}^{k}x),\qquad 
    \mathscr{Q}_{a}(x)&=\mathsf{T}_{a}^{+}(x)+\mathsf{T}^{-}_{a}(x).
\eea
We still call $\mathsf{T}_{a}^{+}(x)$ and $\mathsf{T}_{a}^{-}(x)$ D2 $qq$-characters. We may rescale the root currents as $\mathsf{A}(x)\rightarrow \mathfrak{q}^{-1}\mathsf{A}(x)$ and then obtain 
\bea
\mathscr{Q}_{a}(x)=\sum_{k\in\mathbb{Z}}\mathfrak{q}^{k}\mathsf{S}_{a}(q_{a}^{k}x),\quad \mathsf{T}_{a}^{+}(x)=\sum_{k\geq 0}\mathfrak{q}^{k}\mathsf{S}_{a}(q_{a}^{k}x),\quad \mathsf{T}_{a}^{-}(x)=\sum_{k<0}\mathfrak{q}^{k}\mathsf{S}_{a}(q_{a}^{k}x).
\eea

Each term of the screening charge has a nice pictorial interpretation using the one-dimensional partition labeled by $k\in\mathbb{Z}$ in \eqref{eq:1dpartition_integer}:
\bea
\mathsf{S}_{a}(q_{a}^{k}x),\,\,\,k\in\mathbb{Z} \quad \Longleftrightarrow \quad 
   \adjustbox{valign=c}{ \begin{tikzpicture}
        \draw[->] (-2,0)--(4,0);
        \draw[thick]   (-0.5,-0.5)--(-0.5,1);
        \draw (-2,0.7)--(3,0.7);
        \draw (3,0)--(3,0.7);
        \draw (0.2,0)--(0.2,0.7);
        \draw (0.9,0)--(0.9,0.7);
         \draw (1.6,0)--(1.6,0.7);
          \draw (2.3,0)--(2.3,0.7);
          \draw (-1.2,0)--(-1.2,0.7);
          \node at (-1.6,0.35) {$\cdots$};
          \node at (4,0) [right]{$q_{a}$};
          \node at (-0.15,0)[below]{1};
          \node at (0.55,0)[below]{2};
          \node at (1.25,0)[below]{$\cdots$};
          \node at (1.95,0)[below]{$\cdots$};
          \node at (2.65,0)[below]{$k$};
    \end{tikzpicture}}\label{eq:D2vectorfigure1}
\eea
The screening charge $\mathscr{Q}_{a}(x)$ is understood as a collection of possible one-dimensional partitions labeled by $k\in\mathbb{Z}$ with the $q$-coordinate of the origin $x$:
\bea
\mathscr{Q}_{a}(x) \quad \Longleftrightarrow \left\{\quad\left. 
   \adjustbox{valign=c}{ \begin{tikzpicture}
        \draw[->] (-2,0)--(4,0);
        \draw[thick]   (-0.5,-0.5)--(-0.5,1);
        \draw (-2,0.7)--(3,0.7);
        \draw (3,0)--(3,0.7);
        \draw (0.2,0)--(0.2,0.7);
        \draw (0.9,0)--(0.9,0.7);
         \draw (1.6,0)--(1.6,0.7);
          \draw (2.3,0)--(2.3,0.7);
          \draw (-1.2,0)--(-1.2,0.7);
          \node at (-0.15,0.35) {$x$};
          \node at (-1.6,0.35) {$\cdots$};
          \node at (4,0) [right]{$q_{a}$};
          \node at (-0.15,0)[below]{1};
          \node at (0.55,0)[below]{2};
          \node at (1.25,0)[below]{$\cdots$};
          \node at (1.95,0)[below]{$\cdots$};
          \node at (2.65,0)[below]{$k$};
    \end{tikzpicture}}\quad \right|\,\, k\in\mathbb{Z} \right\}\label{eq:D2vectorfigure2}
\eea
Using \eqref{eq:oprelation1}, we have 
\bea
    \mathsf{S}_{a}(q_{a}^{k}x)={:\prod_{i=-\infty}^{k}\mathsf{A}^{-1}(q_{a}^{i-1}x):}
\eea
which means we can interpret each box of the one-dimensional partition with the $q$-coordinates $xq_{a}^{i-1}$ as the operator $\mathsf{A}^{-1}(xq_{a}^{i-1})$. The screening charge $\mathscr{Q}_{a}(x)$ is then understood as a collection of the possible one-dimensional partitions labeled by $k\in\mathbb{Z}$. 

We can do the same discussion for the D2 $qq$-characters $\mathsf{T}_{a}^{\pm}(x)$. We omit the discussion for $\mathsf{T}^{-}_{a}(x)$ and focus on $\mathsf{T}_{a}^{+}(x)$. In this case, each term of $\mathsf{T}^{+}_{a}(x)$ is represented as a one-dimensional partition labeled by $k\in\mathbb{Z}_{\geq 0}$ in \eqref{eq:1dpartition_positive} and $\mathsf{T}_{a}^{+}(x)$ is interpreted as a collection of the possible one-dimensional partitions labeled by $k\in\mathbb{Z}_{\geq 0}$ with the $q$-coordinate of the origin $x$:
\bea
    \mathsf{T}^{+}_{a}(x) \quad \Longleftrightarrow \quad \left\{\left.\quad
   \adjustbox{valign=c}{ \begin{tikzpicture}
        \draw[->] (-0.5,0)--(4,0);
        \draw[thick]   (-0.5,-0.5)--(-0.5,1);
        \draw (-0.5,0.7)--(3,0.7);
        \draw (3,0)--(3,0.7);
        \draw (0.2,0)--(0.2,0.7);
        \draw (0.9,0)--(0.9,0.7);
         \draw (1.6,0)--(1.6,0.7);
          \draw (2.3,0)--(2.3,0.7);
          \node at (-0.15,0.35) {$x$};
          \node at (4,0) [right]{$q_{a}$};
          \node at (-0.15,0)[below]{1};
          \node at (0.55,0)[below]{2};
          \node at (1.25,0)[below]{$\cdots$};
          \node at (1.95,0)[below]{$\cdots$};
          \node at (2.65,0)[below]{$k$};
    \end{tikzpicture}}\quad \right| \quad k\in\mathbb{Z}_{\geq 0}\quad \right\}.
\eea

In the context of quiver W-algebra, the iWeyl reflection is usually generated from the highest weight only in one direction. For the screening charge, actually, the $qq$-character is associated with the vector representation of quantum toroidal $\mathfrak{gl}_{1}$ (see section~\ref{sec:vectorrep}) which is not a highest weight representation, and thus, we need to generate the monomial terms in two directions. 

We also note that the coefficient factors of the screening charges are 1, which is identical to the $\U(1)$ partition function of the $\D2$ theory introduced in section~\ref{sec:cplvortex_partitionfunct} (see Appendix \ref{app:D2U1partitionfunction} and \eqref{eq:D2U1partition2}): $\mathcal{Z}^{\D2}_{a}[k;q_{i},q_{j}]=1$.

\begin{remark}
Note that the D2 $qq$-characters $\mathsf{T}^{+}_{a}(x)$ are the operator version of the $\D2$ $qq$-characters introduced in \eqref{eq:D2qqpart}:
\bea
\mathsf{T}_{a}^{+}(x)\quad \longleftrightarrow  \quad\widehat{\mathscr{T}}^{\bar{a}}(x),\,\widehat{\mathscr{T}}^{\bar{a}\vee}(x),\quad a\in \four.
\eea
The highest weight and the root current have the following correspondence:
\bea
\mathsf{S}_{a}(x)\quad &\longleftrightarrow \quad \widehat{\mathscr{W}}^{\bar{a}}(q_{\bar{a}}x)^{-1},\widehat{\mathscr{W}}^{\bar{a}\vee}(q_{\bar{a}}x)^{-1},\quad a\in\four\\
\mathsf{A}(x)^{-1}\quad &\longleftrightarrow \quad  \widehat{\mathscr{A}}(x)^{-1},\widehat{\mathscr{A}}^{\,\vee}(x)^{-1}.
\eea
Note also that the $qq$-characters $\mathsf{T}^{\pm}_{a}(x)\,(a\in\four)$ do not commute with the screening charges $\mathscr{Q}_{a}(x)$ nor themselves: 
    \beq[\mathsf{T}_{a}^{\pm}(x),\mathscr{Q}_{b}(x')]\neq 0,\quad [\mathsf{T}_{a}^{\pm}(x),\mathsf{T}_{b}^{\pm}(x')]\neq 0,\quad a,b\in\four.
    \eeq
    For the commutativity of the screening charges, we need the sum of both $\mathsf{T}_{a}^{\pm}(x)$.
\end{remark}
\begin{remark}
Although the above D2-brane $qq$-character (screening charge) is a sum of monomial terms labeled by an integer $k\in\mathbb{Z}$, using the results of section~\ref{sec:D6qqcharacter}, we can introduce a D2 and $\overline{\D6}$ coupled $qq$-character where the monomial terms are labeled by $k\in\mathbb{Z}_{\geq 0}$. In this case, the generating current is 
\begin{equation}
    :\frac{\mathsf{S}_{a}(x)}{\mathsf{W}_{\bar{b}}(q_{b}^{-1}q_{a}^{-1}x)}:,\quad b\neq a\in \four.
\end{equation}
Focusing on $a=1,\,b=4$, one can show that the $qq$-character that commutes with the screening charge $\mathscr{Q}_{4}(x')$ is
\begin{equation}
    \mathsf{T}_{\D2_{1}\text{-}\overline{\D6}_{\bar{4}}}(x)=\sum_{k\geq 0} \mathcal{Z}^{\D2_{1}\tbar\overline{\D6}_{\bar{4}}}[k]:\frac{\mathsf{S}_{1}(q_{1}^{k}x)}{\mathsf{W}_{\bar{4}}(q_{1}^{-1}q_{4}^{-1}x)}:,\quad \mathcal{Z}^{\D2_{1}\tbar\overline{\D6}_{\bar{4}}}[k]=\prod_{l=1}^{k}\frac{1-q_{4}^{-1}q_{1}^{-l}}{1-q_{1}^{-l}},
\end{equation}
with the property
\beq
    \left[\mathsf{T}_{\D2_{1}\text{-}\overline{\D6}_{\bar{4}}}(x),\mathscr{Q}_{4}(x')\right]=0.
\eeq
The coefficient factor here gives the partition function of a 3d $\mathcal{N}=2^{\ast}$ theory with an adjoint matter whose mass is $q_{4}^{-1}$ described in \cite[eq.~(3.10)]{Nekrasov:2009JJM}. A detailed analysis of such systems is postponed for future work.
\end{remark}

\subsection{Coupled vortex system and D2 \texorpdfstring{$qq$}{qq}-characters} \label{sec:cplvortD2qq} 
The contour integral formula in \eqref{eq:D2integral} can be expressed using vertex operators as in \eqref{eq:D2op_integral}. Let us see that the expanded form of the partition function \eqref{eq:D2cplvortex} can be expressed using the $\D2$ $qq$-characters.

Let us consider finite products of the screening charges $\mathsf{S}_{a}(v_{a,\alpha}),\,(a\in\four,\,\alpha=1,\ldots, n_{a})$. The operator product of two screening currents is
\begin{align}
\begin{split}
&\mathsf{S}_{b}\left(v_{b,\beta}q_{b}^{k_{b}^{(\beta)}}\right)\mathsf{S}_{a}\left(v_{a,\alpha}q_{a}^{k^{(\alpha)}_{a}}\right)\\
=&{:\mathsf{S}_{b}(v_{b,\beta})\prod_{\AboxF\in k^{(\beta)}_{b}}\mathsf{A}^{-1}(\chi_{b,v_{b,\beta}}(\BboxF)):}{:\mathsf{S}_{a}(v_{a,\alpha})\prod_{\Abox\in k_{a}^{(\alpha)}}\mathsf{A}^{-1}(\chi_{a,v_{a,\alpha}}(\Bbox)):}\\
    =&\mathcal{Z}_{\text{1-loop}}^{\D2\tbar\D2}(v_{a,\alpha},a\,|\,v_{b,\beta},b)\mathcal{Z}^{\D2\tbar\D2}_{a|b}[v_{a,\alpha},k^{(\alpha)}_{a}\,|\,v_{b,\beta},k^{(\beta)}_{b}],
\end{split}
\end{align}
where $a\neq b$. Generally, for finite products of screening currents, we have 
\bea
    \prod_{a=1}^{\four}\prod_{\alpha=1}^{n_{a}}\mathsf{S}_{a}(v_{a,\alpha}q_{a}^{k^{(\alpha)}_{a}})&=\prod_{(b,\beta)>(a,\alpha)}\mathcal{Z}^{\D2\tbar\D2}_{\text{1-loop}}(v_{a,\alpha},a\,|\,v_{b,\beta},b)\prod_{(b,\beta)>(a,\alpha)}\mathcal{Z}^{\D2\tbar\D2}_{a|b}(v_{a,\alpha},k_{a}^{(\alpha)}\,|\,v_{b,\beta},k_{b}^{(\beta)})\\
    &\qquad \times:\prod_{a=1}^{\four}\prod_{\alpha=1}^{n_{a}}\mathsf{S}_{a}(v_{a,\alpha}q_{a}^{k^{(\alpha)}_{a}}): 
\eea
where the left-hand side is an ordered operator product and the ordering is defined as 
\beq
    \cdots \mathsf{S}_{b}(v_{b,\beta}q_{b}^{k_{b}^{(\beta)}})\cdots \mathsf{S}_{a}(v_{a,\alpha}q_{a}^{k^{(\alpha)}_{a}})\cdots\quad  \Longleftrightarrow\quad  (b,\beta)>(a,\alpha).
\eeq
Therefore, we have the following claim.
\begin{theorem}\label{thm:cplvortexBPSCFT}
The gauge origami partition function involving $\D2$-branes is given by the correlation function of a finite number of screening currents,
\beq\label{eq:D2screening}
    \mathcal{Z}^{\D2}_{\text{1-loop}}\mathcal{Z}^{\D2}_{\text{cpl.vort.}}[\,\underline{\vec{v}},\,\underline{\vec{k}}\,]=\bra{0}\prod_{a=1}^{\four}\prod_{\alpha=1}^{n_{a}}\mathsf{S}_{a}(v_{a,\alpha}q_{a}^{k^{(\alpha)}_{a}})\ket{0}.
\eeq
\end{theorem}
\begin{corollary}
The total partition function can be written using half of the screening charge as
\beq
    \mathcal{Z}^{\D2}_{\text{1-loop}}\mathcal{Z}^{\D2}_{\text{vort.}}=\sum_{\underline{\vec{k}}}\mathfrak{q}^{|\underline{\vec{k}}|}\mathcal{Z}^{\D2}_{\text{1-loop}}\mathcal{Z}^{\D2}_{\text{cpl.vort.}}[\underline{\vec{v}},\underline{\vec{k}}]=\bra{0}\prod_{a\in\four}\prod_{\alpha=1}^{n_{a}}\mathsf{T}^{+}_{a}(v_{a,\alpha})\ket{0}.
\eeq
\end{corollary}

\begin{remark}
    Realization of the 3d partition function on $\mathbb{C}_a \times \mathbb{S}^1$ with finite number insertion of the screening charges has been discussed in~\cite{Aganagic:2013tta,Aganagic:2014oia,Aganagic:2015cta}.
    See also~\cite{Nieri:2017ntx,Kimura:2021ngu} for the case with different types of the screening charges.
\end{remark}

\paragraph{Coulomb branch formula}
The vortex partition function discussed above is given by summation over infinitely many topologically distinct sectors (Higgs branch formula).
There is another description of the 3d partition function, which is given by the contour integral over the Cartan torus of the gauge group (Coulomb branch formula)~\cite{Beem:2012mb,Yoshida:2014ssa}.
Let us show that this Coulomb branch formula is also obtained from the vertex operator formalism.

The screening charges $\mathscr{Q}_{a}(x)$ can be written using the Jackson integral with the base $x$ as 
\begin{equation}
    \mathscr{Q}_{a}(x)=\oint_{\mathcal{C}_{a,x}}d_{q_{a}}z\,\mathsf{S}_{a}(z),
\end{equation}
where the contour integral is denoted $\mathcal{C}_{x}$. Let us first consider the 3d theory on $\mathbb{C}_{a}\times \mathbb{S}^{1}$. Finite products of the screening charges are then given as
\bea
    \mathbf{Z}^{n_{a}}_{a,\vec{x}}\coloneqq \overleftarrow{\prod_{\alpha=1}^{n_{a}}}\mathscr{Q}_{a}(x_{\alpha})&=\oint_{\mathcal{C}_{a,x_{n_{a}}}}\cdots\oint_{\mathcal{C}_{a,x_{2}}}\oint_{\mathcal{C}_{a,x_{1}}}d_{q_{a}}z_{n_{a}}\cdots d_{q_{a}}z_{1}\,\mathsf{S}_{a}(z_{n_{a}})\cdots \mathsf{S}_{a}(z_{1})\\
    &=\oint_{\mathcal{C}_{a,x_{1}},\ldots\mathcal{C}_{a,x_{n_{a}}}}\prod_{\alpha=1}^{n_{a}}d_{q_{a}}z_{i}\,\,{\Delta}_{a}(\vec{z};q_{a}):\prod_{\alpha=1}^{n_{a}}\mathsf{S}_{a}(z_{\alpha}):
\eea
where 
\beq
\Delta_{a}(\vec{z};q_{a})=\prod_{i<j}\Delta_{a}(z_{i}/z_{j}),\quad
    \Delta_{a}(x'/x)=\frac{(x'/x;q_{a})_{\infty}\prod_{i\neq a}(q_{i}q_{a}x'/x;q_{a})_{\infty}}{(q_{a}x'/x;q_{a})_{\infty}\prod_{i\neq a}(q_{i}^{-1}x'/x;q_{a})_{\infty}},
\eeq
and we assumed $|q_{a}|<1$.
Note that $\Delta_{a}(x'/x)$ is just $\mathcal{Z}^{\D2\tbar\D2}_{\text{1-loop}}(x',a\,|\,x,a)$. 

One-loop contributions of 3d $\mathcal{N}=2$ on $\mathbb{C}_a \times \mathbb{S}^1$ are given by
\beq
    z_\text{vec} = \prod_{i \neq j} (x_i/x_j;q_a)_\infty
    \, , \qquad
    z_\text{adj} = \prod_{i\neq j} (\mu_\text{adj} x_i/x_j;q_a)_\infty^{-1} ,
\eeq
where $\mu_{\text{adj}} = e^{m_\text{adj}}$. Meanwhile, the OPE factor of the screening currents is rewritten after using the theta function in \eqref{eq:ellipticthetadef} as
\bea
\Delta_{a}(\vec{z};q_{a})&=\prod_{i<j}\frac{(z_{i}/z_{j};q_{a})_{\infty}\prod_{b\neq a}(q_{b}q_{a}z_{i}/z_{j};q_{a})_{\infty}}{(q_{a}z_{i}/z_{j};q_{a})_{\infty}\prod_{b\neq a}(q_{b}^{-1}z_{i}/z_{j};q_{a})_{\infty}}\\
&=\prod_{i<j}\frac{\prod_{b\neq a}\theta(q_{b}^{-1}z_{j}/z_{i};q_{a})}{\theta(z_{j}/z_{i};q_{a})}\prod_{i\neq j}\frac{(z_{i}/z_{j};q_{a})_{\infty}}{\prod_{b\neq a}(q_{b}^{-1}z_{i}/z_{j};q_{a})_{\infty}}.
\eea
which is thus identified with 1 vector multiplet and 3 adjoint chiral multiples with masses $q_{b}^{-1},\,\,(b\neq a)$.\footnote{Hence, this factor $\Delta_{a}(\vec{z};q_{a})$ agrees with (a trigonometric version of) the measure part of the vertex function associated with the Hilbert scheme of points on $\mathbb{C}^3$ \cite[Prop.~7.12]{Cao:2023lon} up to the boundary contribution.} Additional theta function factors are identified with the boundary contribution $\partial \mathbb{C}_a \times \mathbb{S}^1 = \mathbb{T}^2$. 
In addition to the adjoint matters, we can also add the fundamental matters using the flavor brane vertex operators (see section~\ref{sec:D2_BetheEq}).

Note that due to the condition $q_{1}q_{2}q_{3}q_{4}=1$, we can only choose up to three of the $q$-parameters to be $|q_{a}|<1$. For the analytic region of $|q_{a}|>1$, the above formulas will be modified after using the reflection formula in \eqref{eq:app-qPochreflec} (see also \eqref{eq:D2oneloop-qgamma}).

Let us next consider the contribution between $\mathbb{C}_{a}\times \mathbb{S}^{1}$ and $\mathbb{C}_{b}\times \mathbb{S}^{1}$ where $a \neq b$. For simplicity, let us focus on $a=1,b=2$:
\bea
    \mathscr{Q}_{1}(x)\mathscr{Q}_{2}(x')&=\oint_{\mathcal{C}_{1,x}}\oint_{\mathcal{C}_{2,x'}}d_{q_{1}}z\,d_{q_{2}}z' \mathscr{S}_{34}\left(q_{1}z'/z\right)\\
    &=\oint_{\mathcal{C}_{1,x}}\oint_{\mathcal{C}_{2,x'}}d_{q_{1}}z\,d_{q_{2}}z' \frac{(1-q_{13}z'/z)(1-q_{14}z'/z)}{(1-q_{1}z'/z)(1-q_{2}^{-1}z'/z)}
\eea
The contribution coming from $\mathscr{S}_{34}(z)$ is understood as the chiral and Fermi multiplet living on the 1d intersection $\mathbb{S}^{1}$ \cite{Nieri:2017ntx,Kimura:2021ngu}. 
Therefore, we conclude that the partition function obtained by $\prod_{a\in\four}\mathbf{Z}_{a,\vec{x}_{a}}^{n_{a}}$ is an intersecting gauge theory. Choosing the contour integral properly, we expect that we can reproduce the coupled vortex partition function in~\eqref{eq:D2cplvortex}.
\begin{conjecture}
    By specifying the contour integrals properly, we have the following identity:
    \beq
        \mathcal{Z}^{\D2}_{\text{1-loop}}\mathcal{Z}^{\D2}_{\text{vort.}}=\bra{0}\prod_{a\in\four}\mathbf{Z}_{a,\vec{x}_{a}}^{n_{a}}\ket{0}.
    \eeq
    Namely, the coupled vortex system is obtained by finite products of screening charges.
\end{conjecture}


\section{D4-brane \texorpdfstring{$qq$}{qq}-characters}\label{sec:D4qqcharacter}
In this section, we introduce the operator version of the D4 $qq$-characters introduced in section~\ref{sec:qqpartitionfunct}. These $qq$-characters are not new $qq$-characters but they were already introduced in~\cite{Nekrasov:2015wsu} and the corresponding algebraic structure is known to be the affine quiver W-algebra~\cite{Kimura:2015rgi}. We introduce six types of $\D4$ $qq$-characters corresponding to the six possible configurations of D4-branes in $\mathbb{C}^{4}$ in an equal footing in section~\ref{sec:D4qqaffinequiver}. We then discuss the relation with the spiked instantons in section~\ref{sec:D4qqandspikedinst} and show that their compositions give the gauge origami partition function of the spiked instantons. We also show that the partition function can be rewritten using the screening currents in section~\ref{sec:spikedandscreening}. We then extend the $\D4$ $qq$-characters to general D4 $qq$-characters which give a supergroup analog of the gauge origami partition function in section~\ref{sec:generalD4qq}. The quadratic relations of the $\D4$ $qq$-characters are discussed in section~\ref{sec:D4quadraticrelation}.

\subsection{D4 \texorpdfstring{$qq$}{qq}-characters and affine quiver W-algebra}\label{sec:D4qqaffinequiver}
Let us consider the $qq$-character generated by the D4 operators $\mathsf{X}_{A}(x)\,(A\in\six)$. The D4 $qq$-character generated here is identified with the generator of the affine quiver W-algebra of $\Gamma = \widehat{A}_0$ in \cite{Kimura:2015rgi}. Let us review the derivation of it.

Each term of the $qq$-character is decomposed into two parts, the operator part and the coefficient part. The operator part is determined by the iWeyl reflection which is defined as the following (see also \eqref{eq:D4iWeylpartition}). 
\begin{definition}
    The iWeyl reflection of the $\D4$ vertex operator $\mathsf{X}_{A}(x)\,\,(A\in\six)$ is 
    \begin{equation}
        \mathsf{X}_{A}(x)\longrightarrow {:\mathsf{X}_{A}(x)\mathsf{A}^{-1}(x):},\quad A\in\six.
    \end{equation}
    Using the property in \eqref{eq:oprelationwithD0-D4}, we have 
\beq
    \mathsf{X}_{ab}(x)\longrightarrow {:\mathsf{a}_{0}(x)^{-1}\frac{\mathsf{X}_{ab}(q_{a}x)\mathsf{X}_{ab}(q_{b}x)}{\mathsf{X}_{ab}(q_{ab}x)}:},\quad a\neq b\in\four.
\eeq
\end{definition}
 The operator part of the $qq$-character is obtained by changing the root current to the $\mathsf{X}$-operators and applying the iWeyl reflection recursively. The operators will be classified by two-dimensional Young diagrams as
\beq\label{eq:D4monomialterm}
    {:\mathsf{X}_{ab}(x)\prod_{\Abox\in\lambda}\mathsf{A}^{-1}(\chi_{ab,x}(\Bbox)):}={:\mathsf{a}_{0}(x)^{-1}\frac{\prod_{\Abox\in A(\lambda)}\mathsf{X}_{ab}(\chi_{ab,x}(\Bbox))}{\prod_{\Abox\in R(\lambda)}\mathsf{X}_{ab}(q_{ab}\chi_{ab,x}(\Bbox))}:}.
\eeq
Following the correspondence in \eqref{eq:D2vectorfigure1} and \eqref{eq:D2vectorfigure2}, we can visualize each monomial terms of the $\D4$ $qq$-character using the Young diagrams:
\bea\label{eq:D4Youngcorrespondence}
{:\mathsf{X}_{ab}(x)\prod_{\Abox \in\lambda}\mathsf{A}^{-1}(\chi_{ab,x}(\Bbox)):}\quad \longleftrightarrow \quad \adjustbox{valign=c}{\begin{tikzpicture}
 \fill[red!20!white] (0.9,1.4)--(1.6,1.4)--(1.6,2.1)--(0.9,2.1)--(0.9,1.4);
        \draw[->] (-1,0)--(4,0);
        \node[above] at (-0.5,4){$q_{b}$};
        \node [right] at (4,0){$q_{a}$};
        \node[below] at (-0.15,0) {$1$};
        \node [below] at (0.55,0){$\cdots$};
        \node [below] at (1.25,0){$i$};
         \draw[->]   (-0.5,-0.5)--(-0.5,4);
         \draw (0.2,3.5)--(0.2,0.7);
         \draw (0.9,2.8)--(0.9,0.7);
         \draw (1.6,2.1)--(1.6,0.7);
         \draw (2.3,1.4)--(2.3,0.7);
         \draw (2.3,0.7)--(-0.5,0.7);
         \draw (2.3,1.4)--(-0.5,1.4);
         \draw (1.6,2.1)--(-0.5,2.1);
         \draw (0.9,2.8)--(-0.5,2.8);
         \draw (0.2,3.5)--(-0.5,3.5);
        \draw (-0.5,0.7)--(3,0.7);
        \draw (3,0)--(3,0.7);
        \draw (0.2,0)--(0.2,0.7);
        \draw (0.9,0)--(0.9,0.7);
         \draw (1.6,0)--(1.6,0.7);
          \draw (2.3,0)--(2.3,0.7);
          \draw (-0.15,0.35)--++(-0.7,-1);
          \node[left] at (-0.85,-0.65){$x$};
          \node [left] at (-0.5,0.35) {$1$};
          \node [left] at (-0.5,1.05){$\vdots$};
          \node [left] at (-0.5,1.75){$j$};
           \draw  (1.25,1.75)--++(0.9,0.9);
          \node[right] at (2.2,2.65) {$xq_{a}^{i-1}q_{b}^{j-1}$};
        \end{tikzpicture}
        }
\eea
Similar to the D2-case, each $\mathsf{A}^{-1}(\chi_{ab,x}(\Bbox))$ corresponds to the box of the Young diagram. The operator $\mathsf{X}_{ab}(x)$ defines the vacuum and physically gives the one-loop perturbative part (see section \ref{sec:D4qqandspikedinst}). In the algebraic context, it is called the \emph{highest weight} of the $qq$-character. Moreover, as will be shown in \ref{thm:D4qq-commute}, it uniquely determines the $qq$-character, we also call it the \emph{generating current} of the $qq$-character.

The $qq$-character is defined by adding the monomial terms in \eqref{eq:D4monomialterm} for all possible Young diagrams with specific coefficients. The coefficients are determined by the commutativity with the screening currents.
\begin{definition}
    We define the $\D4$ $qq$-character for $A\in\six$ as
    \beq
        \mathsf{T}_{A}(x)=\sum_{\lambda\in\mathcal{P}}\widetilde{\mathcal{Z}}_{A}^{\D4}[\lambda]:\Lambda_{A,\lambda}(x):
    \eeq
    where and 
    \beq
       \Lambda_{A,\lambda}(x)={:\mathsf{X}_{A}(x)\prod_{\Abox\in\lambda}\mathsf{A}^{-1}(\chi_{A,x}(\Bbox)):}.
    \eeq
    Note that the coefficients $\widetilde{\mathcal{Z}}^{\D4}_{A}[\lambda]$ do not depend on the choice of $a\in\bar{A}$. 
    One may redefine the zero-modes of the root currents as $\mathsf{A}(x)\rightarrow \mathfrak{q}^{-1}\mathsf{A}(x)$ and obtain
\beq
    \mathsf{T}_{A}(x)=\sum_{\lambda\in\mathcal{P}}\mathfrak{q}^{|\lambda|}\widetilde{\mathcal{Z}}_{A}^{\D4}[\lambda]:\Lambda_{A,\lambda}(x):,\quad A\in\six.
\eeq
\end{definition}

\begin{theorem}\label{thm:D4qq-commute}
The $\D4$ $qq$-characters commute with the screening charges associated with the transverse directions
\beq
    [\mathsf{T}_{A}(x),\mathscr{Q}_{a}(x')]=0,\quad \forall a\in\bar{A}.
\eeq
\end{theorem}
\begin{proof}
    Let us focus on $A=12$ and $a=4$ and derive the $\D4$ $qq$-character. Other cases are obtained using the quadrality. Using the formulas in \eqref{eq:app-contractions}, we have 
    \bea\label{eq:D4-D2contraction}
        \Lambda_{12,\lambda}(x)\mathsf{S}_{4}(x')=\left[q_{3}^{-1}\frac{\mathscr{Y}^{12}_{\lambda,x}(q_{12}x')}{\mathscr{Y}^{12}_{\lambda,x}(q_{123}x')}\right]^{x'}_{-} :\Lambda_{12,\lambda}(x)\mathsf{S}_{4}(x'):,\\
        \mathsf{S}_{4}(x')\Lambda_{12,\lambda}(x)=\left[q_{3}^{-1}\frac{\mathscr{Y}^{12}_{\lambda,x}(q_{12}x')}{\mathscr{Y}^{12}_{\lambda,x}(q_{123}x')}\right]^{x'}_{+}:\Lambda_{12,\lambda}(x)\mathsf{S}_{4}(x'):
    \eea
    where $\left[\,f(x)\,\right]^{x}_{\pm}$ means expansions of $f(x)$ in ${x}^{\mp}$ respectively. Assume that the $qq$-character takes the form of 
    \beq
        \mathsf{T}_{12}(x)=\sum_{\lambda\in\mathcal{P}}F^{\D4}_{12}[\lambda]\Lambda_{12,\lambda}(x),\quad F_{12}[\emptyset]=1
    \eeq
    where $F^{\D4}_{12}[\lambda]$ are some coefficients, then the commutation relation is 
    \begin{align}
    \begin{split}
    &[\mathsf{T}_{12}(x),\mathsf{S}_{4}(x')]\\
    =&-q_{3}^{-1}\left[\sum_{\lambda}F^{\D4}_{12}[\lambda]\sum_{\Abox\in A(\lambda)}\delta\left(\frac{\chi_{12,x}(\Bbox)}{q_{4}^{-1}x'}\right)\underset{x'=q_{3}^{-1}\chi_{12,x}(\Abox)}{\Res}{x'}^{-1}\frac{\mathscr{Y}_{\lambda,x}^{12}(x')}{\mathscr{Y}_{\lambda,x}^{12}(q_{3}x')}:\Lambda_{12,\lambda}(x)\mathsf{S}_{4}(x'):\right.\\
    &+\left.\sum_{\lambda}F^{\D4}_{12}[\lambda]\sum_{\Abox\in R(\lambda)}\delta\left(\frac{\chi_{12,x}(\Bbox)}{x'}\right)\underset{x'=q_{34}^{-1}\chi_{12,x}(\Abox)}{\Res}{x'}^{-1}\frac{\mathscr{Y}_{\lambda,x}^{12}(x')}{\mathscr{Y}_{\lambda,x}^{12}(q_{3}x')}:\Lambda_{12,\lambda}(x)\mathsf{S}_{4}(x'):\right].
\end{split}
\end{align}
Shifting the second term as $\lambda=\lambda'+\Bbox$, it will be
\beq
    -q_{3}^{-1}\sum_{\lambda'}\sum_{\Abox\in A(\lambda')}F^{\D4}_{12}[\lambda'+\Bbox\,]\delta\left(\frac{\chi_{12,x}(\Bbox)}{x'}\right)\underset{x'=q_{34}^{-1}\chi_{12,x}(\Abox)}{\Res}{x'}^{-1}\frac{\mathscr{Y}_{\lambda'+\Abox,x}^{12}(x')}{\mathscr{Y}_{\lambda'+\Abox,x}^{12}(q_{3}x')}:\Lambda_{12,\lambda'}(x)\mathsf{A}^{-1}(\chi_{12,x}(\Bbox))\mathsf{S}_{4}(x').
\eeq
Using 
\beq
    {:\Lambda_{12,\lambda'}(x)\mathsf{A}^{-1}(\chi_{12,x}(\Bbox))\mathsf{S}_{4}(\chi_{12,x}(\Bbox)):}={:\Lambda_{12,\lambda'}(x)\mathsf{S}_{4}(q_{4}\chi_{12,x}(\Bbox)):}
\eeq
and imposing the condition (see the recursion formula in Thm. \ref{app:thm-D4recursion})
\beq
    \frac{F^{\D4}_{12}[\lambda+\Bbox]}{F^{\D4}_{12}[\lambda]}=-\frac{\underset{x'=q_{3}^{-1}\chi_{12,x}(\Abox)}{\Res}{x'}^{-1}\frac{\mathscr{Y}_{\lambda,x}^{12}(x')}{\mathscr{Y}_{\lambda,x}^{12}(q_{3}x')}}{\underset{x'=q_{34}^{-1}\chi_{12,x}(\Abox)}{\Res}{x'}^{-1}\frac{\mathscr{Y}_{\lambda+\Abox,x}^{12}(x')}{\mathscr{Y}_{\lambda+\Abox,x}^{12}(q_{3}x')}}=q_{3}^{-1}\frac{\mathcal{Z}^{\D4}_{12}[\lambda+\Bbox\,;q_{3}]}{\mathcal{Z}_{12}^{\D4}[\lambda\,;q_{3}]}=\frac{\widetilde{\mathcal{Z}}_{12}^{\D4}[\lambda+\Bbox]}{\widetilde{\mathcal{Z}}_{12}^{\D4}[\lambda]}
\eeq
we then obtain
\bea
\relax[\mathsf{T}_{12}(x),\mathsf{S}_{4}(x')]&=-q_{3}^{-1}\sum_{\lambda\in\mathcal{P}}\sum_{\Abox\in A(\lambda)}F^{\D4}_{12}[\lambda]\underset{x'=q_{3}^{-1}\chi_{12,x}(\Abox)}{\Res}{x'}^{-1}\frac{\mathscr{Y}_{\lambda,x}^{12}(x')}{\mathscr{Y}_{\lambda,x}^{12}(q_{3}x')}:\Lambda_{12,\lambda}(x)\mathsf{S}_{4}(q_{4}\chi_{12,x}(\Bbox)):\\
&\qquad\qquad \times \left(\delta\left(\frac{q_{4}\chi_{12,x}(\Bbox)}{x'}\right)-\delta\left(\frac{\chi_{12,x}(\Bbox)}{x'}\right)\right)
\eea
which is a total difference. Therefore, under the condition $F_{12}[\lambda]=\widetilde{\mathcal{Z}}_{12}^{\D4}[\lambda]$, the $qq$-character commutes with the screening charge $\mathscr{Q}_{4}(x')$:
\beq
    [\mathsf{T}_{12}(x),\mathscr{Q}_{4}(x')]=0.
\eeq
\end{proof}

\begin{remark}
Note that these D4 $qq$-characters $\mathsf{T}_{A}(x)\,(A\in\six)$ are the operator version of the $\D4$ $qq$-characters introduced in \eqref{eq:D4qqpart}:
\bea
\mathsf{T}_{\bar{A}}(x)\quad \longleftrightarrow  \quad\widehat{\mathscr{T}}^{A}(x),\,\widehat{\mathscr{T}}^{A\vee}(x),\quad A\in \six.
\eea
The highest weight and the root current have the following correspondence:
\bea
\mathsf{X}_{A}(x)\quad &\longleftrightarrow \quad \widehat{\mathscr{Y}}^{A}(q_{A}x),\widehat{\mathscr{Y}}^{A\vee}(q_{A}x),\quad A\in\six\\
\mathsf{A}(x)^{-1}\quad &\longleftrightarrow \quad  \widehat{\mathscr{A}}(x)^{-1},\widehat{\mathscr{A}}^{\,\vee}(x)^{-1}.
\eea
\end{remark}

\subsection{Spiked instantons and D4 \texorpdfstring{$qq$}{qq}-characters}\label{sec:D4qqandspikedinst}
Similar to the integral formula, the expanded version of the spiked instanton partition function can be expressed using the $\D4$ $qq$-characters. 
\begin{lemma}\label{lemm:D4ope}
The operator product of $\{\Lambda_{A,\lambda}(x)\}_{A\in\six}$ is
\beq
    \Lambda_{B,\mu}(x')\Lambda_{A,\lambda}(x)=\mathcal{Z}^{\D4\text{-}\D4}_{\text{1-loop}}(x,A\,|\,x',B)\mathcal{Z}^{\D4\text{-}\D4}_{A|B}(x,\lambda\,|\,x',\mu):\Lambda_{B,\mu}(x')\Lambda_{A,\lambda}(x):.
\eeq
When $A=B$, it gives the vector and adjoint hypermultiplet contributions, while when $A\neq B$, it gives the bifundamental-like contributions connecting gauge theories defined on different subspaces. 
\end{lemma}
\begin{theorem}\label{thm:spiked-qq-BPSCFT}
The gauge origami partition function of the spiked instantons is written using the $\D4$ $qq$-characters:
\bea
    \mathcal{Z}^{\D4}_{\text{1-loop}}\mathcal{Z}^{\D4}_{\text{inst.}}=\sum_{\underline{\vec{\lambda}}}\mathfrak{q}^{|\underline{\vec{\lambda}}|}\mathcal{Z}^{\D4}_{\text{1-loop}}\mathcal{Z}_{\text{spk.inst.}}^{\D4}[\underline{\vec{v}},\underline{\vec{\lambda}}]=\bra{0}\prod_{A\in\six}\prod_{\alpha=1}^{n_{A}}\mathsf{T}_{A}(v_{A,\alpha})\ket{0}
\eea    
Depending on the value of $A,B\in\six$, we get different instanton contributions:
\begin{itemize}
    \item $A=B$: instantons in $\widehat{A}_0$ quiver gauge theory
    \item $A\cap B \in\four$: folded instantons
    \item $B=\bar{A}$: crossed instantons
\end{itemize}
To summarize, we have the following table of BPS/CFT correspondence.
\begin{align}
    \renewcommand\arraystretch{1.2}{
    \begin{tabular}{c|c}\toprule
        BPS &  CFT\\
     \hline  5d $\mathcal{N}=1^{\ast}$ U(1) on $\mathbb{C}^{2}_{ab}\times \mathbb{S}^{1}$ ($\D4_{ab}\times 1$) & $\bra{0}\mathsf{T}_{ab}(v)\ket{0}$ \\
      5d $\mathcal{N}=1^{\ast}$ U($n_{ab}$) on $\mathbb{C}^{2}_{ab}\times \mathbb{S}^{1}$ ($\D4_{ab}\times n_{ab}$)   & $\bra{0}\mathsf{T}_{ab}(v_{n_{ab}})\cdots\mathsf{T}_{ab}(v_{2})\mathsf{T}_{ab}(v_{1})\ket{0}$\\
      crossed instanton: D4$_{12}$-D4$_{34}$-D0 & $\bra{0}\mathsf{T}_{12}(v)\mathsf{T}_{34}(v')\ket{0}$\\
      folded instanton: D4$_{12}$-D4$_{13}$-$\D0$ & $\bra{0}\mathsf{T}_{12}(v)\mathsf{T}_{13}(v')\ket{0}$\\
      gauge origami of spiked instantons & $\bra{0}\prod\limits_{A\in\six}\prod\limits_{\alpha=1}^{n_{A}}\mathsf{T}_{A}(v_{A,\alpha})\ket{0}$
      \\ \bottomrule
    \end{tabular}}
\end{align}
\end{theorem}

\begin{remark}
    It has been argued in \cite{Konno:2021zvl} that the partition function of the 5d $\mathcal{N}=1^{\ast}$ theory on $\mathbb{C}^{2} \times \mathbb{S}^{1}$ is given by the $qq$-character correlation function, and the 6d theory partition function is given by the corresponding torus correlation function.
\end{remark}

\subsection{Spiked instantons and screening currents}\label{sec:spikedandscreening}
Since the two-dimensional partition can be understood using the $(1,1)$-type description, it is natural to expect that the partition functions have an operator representation using the screening currents. Using the character form in \eqref{eq:spiked_D4toD2}, we omit the singular terms that might occur and extract the square root part of the total index by specifying an order:
\beq
    \mathbf{V}=\sum_{A,B\in\six}\frac{\bfP_{\four}}{\bfP_{\bar{s}(A)}^{\vee}\bfP_{\bar{s}(B)}}\mathbf{X}_{A}^{\vee}\bfX_{B}\rightarrow \mathbf{v}=\sum_{\substack{(x,A)<(x',B)\\x\in\mathcal{X}_{A},x'\in\mathcal{X}_{B}}}\frac{\bfP_{\four}}{\bfP_{\bar{s}(A)}^{\vee}\bfP_{\bar{s}(B)}}\left(\frac{x'}{x}\right).\label{eq:D5toD3halfch}
\eeq
Up to one-loop perturbative factors, $\mathbb{I}[\mathbf{v}]$ is identical to the partition function in \eqref{eq:D4spikedpartition1}, \eqref{eq:D4spikedpartition2}. Using \eqref{eq:D2op}, we then have the following BPS/CFT correspondence.
\begin{proposition}
The index of \eqref{eq:D5toD3halfch} is the vacuum expectation value of the screening currents and is equivalent to the partition function in \eqref{eq:D4spikedpartition1} and \eqref{eq:D4spikedpartition2} up to one-loop perturbative factors.
\beq\label{eq:D4screening}
   \mathbb{I}[\mathbf{v}]=\bra{0}\prod_{A\in\six,x\in\mathcal{X}_{A}}\mathsf{S}_{\bar{s}(A)}(x)\ket{0}\simeq \mathcal{Z}^{\D4}_{\text{1-loop}}\mathcal{Z}^{\D4}_{\text{spk.inst.}}[\underline{\vec{v}},\underline{\vec{\lambda}}].
\eeq
Note that we are implicitly defining an ordering in the products of the screening currents as 
\beq
    \cdots \mathsf{S}_{\bar{s}(B)}(x')\cdots \mathsf{S}_{\bar{s}(A)}(x)\cdots \Longleftrightarrow (x,A)<(x',B).
\eeq
The symbol ``$\simeq$'' means that they are identical up to one-loop perturbative factors.
\end{proposition}

\paragraph{Fusion of $\D2$ $qq$-characters}
Since the gauge origami partition function can be written using both the screening currents and the $\D4$ $qq$-characters, one would like to ask if we have any relation connecting both descriptions. An interesting fact is that up to one-loop perturbative factors, the $\D4$ $qq$-characters are just infinite products of the screening charges.
\begin{theorem}\label{thm:D2toD4fusion}
    The $\D4$ $qq$-characters are infinite products of $\D2$ $qq$-characters (screening charges):
    \beq
        \mathsf{T}_{ab}(x)\simeq  \overrightarrow{\prod_{i=1}^{\infty}}\mathscr{Q}_{b}(xq_{a}^{i-1}),\quad a\neq b\in\four,
    \eeq
    where the product is $\overrightarrow{\prod\limits_{i=1}^{\infty}}f(x_{i})=f(x_{1})f(x_{2})\cdots$.
\end{theorem}
\begin{proof}
We only give a sketch of the proof (see \cite{Kimura:2015rgi} and \cite[eq.~(6.2.23)--(6.2.27)]{Kimura:2020jxl} for a similar discussion). The vertex operator $\Lambda_{ab,\lambda}(x)$ of the $\D4$ $qq$-character satisfies
\beq
    \Lambda_{ab,\lambda}(x)\sim {:\prod_{i=1}^{\infty}\mathsf{S}_{b}(xq_{a}^{i-1}q_{b}^{\lambda_{i}}):}
\eeq
where the equality is up to zero modes. By direct computation, one can show that up to one-loop perturbative factors, we have\footnote{Strictly speaking, we have to be careful of the zero-modes appearing on both hand sides. Moreover, the infinite product of the screening currents needs to be regularized properly.}
\beq
    \mathsf{T}_{ab}(x)\simeq\sum_{\lambda\in\mathcal{P}}\overrightarrow{\prod_{i=1}^{\infty}}\mathsf{S}_{b}(xq_{a}^{i-1}q_{b}^{\lambda_{i}}).\label{eq:D4screeningrep}
\eeq
Using the property that for $i\leq j$ and $\lambda_{i}<\lambda_{j}$
\beq
    \mathsf{S}_{b}(xq_{a}^{i-1}q_{b}^{\lambda_{i}})\mathsf{S}_{b}(xq_{a}^{j-1}q_{b}^{\lambda_{j}})=0,
\eeq
we obtain
\beq
    \mathscr{Q}_{b}(x)\mathscr{Q}_{b}(q_{a}x)=\sum_{k\in\mathbb{Z}}\mathsf{S}_{b}(xq_{b}^{k})\sum_{l\in\mathbb{Z}}\mathsf{S}_{b}(xq_{a}q_{b}^{l})=\sum_{k\geq l}\mathsf{S}_{b}(xq_{b}^{k})\mathsf{S}_{b}(xq_{a}q_{b}^{l}).
\eeq
By computing the nontrivial coefficients appearing after taking the contractions of the screening currents, one obtains the statement.
\end{proof}
\begin{corollary}
    The total partition function is written using infinite products of screening charges
    \beq
        \mathcal{Z}^{\D4}_{\text{1-loop}}\mathcal{Z}^{\D4}_{\text{inst.}}=\sum_{\underline{\vec{\lambda}}}\mathfrak{q}^{|\underline{\vec{\lambda}}|}\mathcal{Z}^{\D4}_{\text{1-loop}}\mathcal{Z}_{\text{spk.inst.}}^{\D4}[\underline{\vec{v}},\underline{\vec{\lambda}}]\simeq \bra{0}\prod_{A\in\six}\prod_{\alpha=1}^{n_{A}}\overrightarrow{\prod_{i=1}^{\infty}}\mathscr{Q}_{\bar{s}(A)}(v_{A,\alpha}q_{s(A)}^{i-1})\ket{0}.
    \eeq
\end{corollary}
    
\begin{remark}
We can visualize the above procedure using the correspondence in \eqref{eq:D2vectorfigure1} and \eqref{eq:D2vectorfigure2}. The equation \eqref{eq:D4screeningrep} is the consequence of the fact that the 2d partition has a $(1,1)$-type description. Each term $\mathsf{S}_{b}(xq_{a}^{i-1}q_{b}^{\lambda_{i}})$ is understood as a one-dimensional partition extending in the $q_{b}$ direction as mentioned in section \ref{sec:multi-dim-part}. Piling the one-dimensional partitions in the orthogonal direction, we get the two-dimensional partition and the corresponding vertex operator $\Lambda_{ab,\lambda}(x)$. We call this the \textbf{fusion} process. Namely, fusions of $\D2$ $qq$-characters (screening charges) give the $\D4$ $qq$-character.
\begin{equation}
    \adjustbox{valign=c}{\includegraphics[width=14cm]{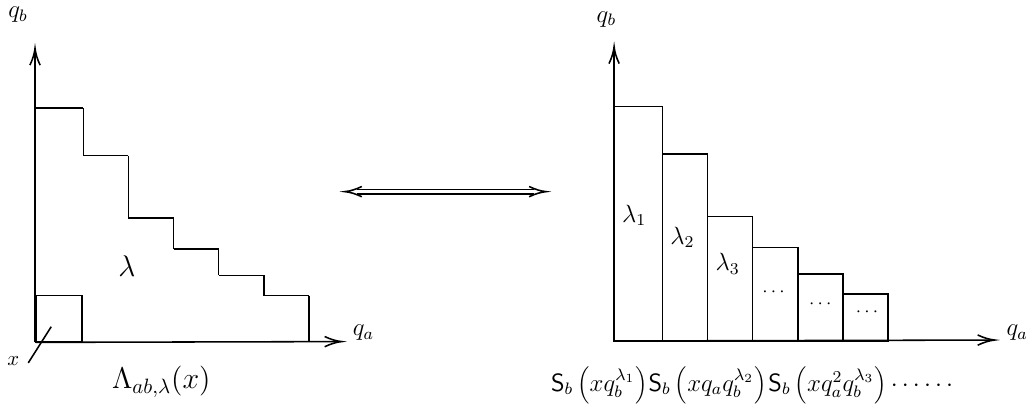}}
\end{equation}

\end{remark}

\subsection{General D4 \texorpdfstring{$qq$}{qq}-characters}\label{sec:generalD4qq}
The highest weight of the $\D4$ $qq$-character discussed in the previous section was $\mathsf{X}_{A}(x)\,(A\in\six)$ and the coefficients appearing are the $\U(1)$ partition function of the affine quiver gauge theory on $\mathbb{C}^{2}_{A}\times \mathbb{S}^{1}$. Higher rank D4 $qq$-characters are obtained as 
\begin{equation}\label{eq:ranknD4qqch}
    \mathsf{T}^{(n)}_{A}(\underline{x})={:\mathsf{X}_{A}(x_{1})\cdots\mathsf{X}_{A}(x_{n}):}+\cdots
\end{equation}
where $\underline{x}=(x_{1},\ldots,x_{n})$ are generic parameters. Imposing the commutativity with the screening charges $\mathscr{Q}_{a}(x')\,(a\in\bar{A})$, the coefficients appearing will be the partition function of the 5d $\mathcal{N}=1^{\ast}$ $\U(n)$ gauge theory on $\mathbb{C}^{2}_{A}\times \mathbb{S}^{1}$.

We can also include negative highest weights as 
\begin{equation}
    \mathsf{T}^{(n|m)}_{A}(\underline{x}|\underline{y})={:\frac{\mathsf{X}_{A}(x_{1})\cdots \mathsf{X}_{A}(x_{n})}{\mathsf{X}_{A}(y_{1})\cdots \mathsf{X}_{A}(y_{m})}:}+\cdots
\end{equation}
where $\underline{x}=(x_{1},\ldots,x_{n})$ and $\underline{y}=(y_{1},\ldots,y_{m})$ are all generic parameters. 
One can show that after imposing the commutativity with the screening charges $\mathscr{Q}_{a}(x')(a\in\bar{A})$, the coefficients will be the partition function of the $\U(n|m)$ supergroup gauge theory. The derivation can be similarly done as the non-supergroup gauge theory (see Appendix \ref{app:supergroup}).
\begin{definition}\label{def:D4negativeiWeyl}
    The iWeyl reflection of the operator $\mathsf{X}_{A}(x)^{-1}$ is defined as
    \begin{equation}
        \mathsf{X}_{A}(x)^{-1}\longrightarrow {:\mathsf{X}_{A}(x)^{-1}\mathsf{A}(q_{A}^{-1}x):}
    \end{equation}
    which we call the \textbf{negative iWeyl reflection}. The normal iWeyl reflection will be called positive iWeyl reflection in contrast\footnote{We omit the \emph{positive} when obvious.}.
\end{definition}
Doing the negative iWeyl reflection recursively, we obtain the monomial terms
\begin{equation}
    \bar{\Lambda}_{ab,\lambda}(x)={:\mathsf{X}_{ab}(x)^{-1}\prod_{\Abox\in\lambda}\mathsf{A}(\bar{\chi}_{ab,x}(\Bbox)):},\quad \bar{\chi}_{ab,x}(\Bbox)=xq_{a}^{-i}q_{b}^{-j},
\end{equation}
which gives the following $\D4$ $qq$-character.
\begin{theorem}\label{thm:D4negativeqqch}
    The $\D4$ $qq$-character generated by $\mathsf{X}_{A}^{-1}(x)$ is 
    \begin{equation}
        \mathsf{T}^{(0|1)}_{A}(x)=\sum_{\lambda\in\mathcal{P}}\widetilde{\mathcal{Z}}^{\D4^{-}}_{A}[\lambda]:\bar{\Lambda}_{A,\lambda}(x):,
    \end{equation}
    where $\widetilde{\mathcal{Z}}^{\D4^{-}}_{A}[\lambda]$ is the $\U(0|1)$ partition function (see Appendix \ref{app:supergroup}). It commutes with the screening charges $\mathscr{Q}_{a}(x')\,(a\in\bar{A})$:
    \begin{equation}
        [\mathsf{T}^{(0|1)}_{A}(x),\mathscr{Q}_{a}(x')]=0,\quad a\in\bar{A}.
    \end{equation}
    By rescaling the root current as $\mathsf{A}(x)\rightarrow \mathfrak{q}^{-1}\mathsf{A}(x)$, the $qq$-character is 
    \begin{equation}
        \mathsf{T}^{(0|1)}_{A}(x)=\sum_{\lambda\in\mathcal{P}}\mathfrak{q}^{-|\lambda|}\widetilde{\mathcal{Z}}^{\D4^{-}}_{A}[\lambda]:\bar{\Lambda}_{A,\lambda}(x):.
    \end{equation}
\end{theorem}

Note that the higher rank $\D4$ $qq$-characters are simply obtained by taking compositions of $\mathsf{T}_{A}^{(1|0)}(x)$ and $\mathsf{T}_{A}^{(0|1)}(x)$ up to one-loop perturbative factors:
\begin{equation}
    \mathsf{T}_{A}^{(n|m)}(\underline{x}\,|\,\underline{y})\propto\prod_{i=1}^{n}\mathsf{T}^{(1|0)}_{A}(x_{i})\prod_{j=1}^{m}\mathsf{T}^{(0|1)}_{A}(y_{j}).
\end{equation}
Obviously, for fixed $A\in\six$, arbitrary products of $\mathsf{T}_{A}^{(1|0)}(x)$ and $\mathsf{T}_{A}^{(0|1)}(x)$ commute with the screening charges $\mathscr{Q}_{a}(x)\,(a\in\bar{A})$ because each of the $qq$-characters commute with the screening charges. For example, consider the product $\mathsf{T}_{12}(x)\mathsf{T}_{12}(x')$ and the commutation relation with $\mathsf{S}_{4}(x'')$. We then have
\bea
\relax[\mathsf{T}_{12}(x)\mathsf{T}_{12}(x'),\mathscr{Q}_{4}(x'')]&=[\mathsf{T}_{12}(x),\mathscr{Q}_{4}(x'')]\mathsf{T}_{12}(x')+\mathsf{T}_{12}(x)[\mathsf{T}_{12}(x'),\mathscr{Q}_{4}(x'')]=0.
\eea
Using
\bea
\mathsf{T}_{12}(x)\mathsf{T}_{12}(x')=\mathcal{Z}_{\text{1-loop}}^{\D4\tbar\D4}(x,12\,|\,x',12)\left(:\mathsf{X}_{12}(x)\mathsf{X}_{12}(x'')+\cdots \right)
\eea
we can see that $\mathsf{T}_{12}(x)\mathsf{T}_{12}(x')/\mathcal{Z}_{\text{1-loop}}^{\D4\tbar\D4}(x,12|x',12)$ is a $qq$-character whose highest weight is $:\mathsf{X}_{12}(x)\mathsf{X}_{12}(x'):$.

More generally, we can introduce $qq$-characters starting from operators like $:\mathsf{X}_{12}(x)\mathsf{X}_{13}(x'):$ where $\D4$ vertex operators associated with different $\mathbb{C}^{2}$ subspaces appear.  The resulting $qq$-character is proportional to the product $\mathsf{T}_{12}(x)\mathsf{T}_{13}(x')$.

\begin{theorem}\label{thm:D4generalqqcharacter}
    Let $a,b,c,d\in\four$ be the four elements in the set $\four$. We choose one of the screening charges $\mathscr{Q}_{d}(x)$. The highest weight of the $\D4$ $qq$-character commuting with this screening charge is generally written as
    \begin{equation}
        \mathsf{T}_{ab:bc:ac}^{(\vec{n}|\vec{m})}(\underline{\vec{x}}\,|\,\underline{\vec{y}})={:\frac{\prod\limits_{\alpha=1}^{n_{ab}}\mathsf{X}_{ab}(x_{ab,\alpha})\prod\limits_{\beta=1}^{n_{bc}}\mathsf{X}_{bc}(x_{bc,\beta})\prod\limits_{\gamma=1}^{n_{ac}}\mathsf{X}_{ac}(x_{ac,\gamma})}{\prod\limits_{\alpha=1}^{m_{ab}}\mathsf{X}_{ab}(y_{ab,\alpha})\prod\limits_{\beta=1}^{m_{bc}}\mathsf{X}_{bc}(y_{bc,\beta})\prod\limits_{\gamma=1}^{m_{ac}}\mathsf{X}_{ac}(y_{ac,\gamma})}:}+\cdots 
    \end{equation}
    where $\vec{n}=(n_{ab},n_{bc},n_{ac}),\,\vec{m}=(m_{ab},m_{bc},m_{ca})$, $\underline{\vec{x}}=(x_{A,\alpha})_{A\in\{ab,bc,ac\}}^{\alpha=1,\ldots,n_{A}}$, $\underline{\vec{y}}=(y_{A,\alpha})_{A\in\{ab,bc,ac\}}^{\alpha=1,\ldots,m_{A}}$. The monomial terms are obtained by doing the positive and negative iWeyl reflections respectively, while the coefficients are obtained by imposing the commutativity with the screening charge $\mathscr{Q}_{d}(x)$:
    \begin{equation}
        [\mathsf{T}_{ab:bc:ac}^{(\vec{n}|\vec{m})}(\underline{\vec{x}}\,|\,\underline{\vec{y}}),\mathscr{Q}_{d}(x')]=0.
    \end{equation}
    Moreover, the $qq$-character obtained here is proportional to the products of the $\D4$ $qq$-characters associated with $A=ab,bc,ac$:
    \begin{equation}
        \mathsf{T}_{ab:bc:ac}^{(\vec{n}|\vec{m})}(\underline{\vec{x}}\,|\,\underline{\vec{y}})\propto \prod_{A\in\{ab,bc,ac\}}\mathsf{T}_{A}^{(n_{A}|m_{A})}(\underline{x_{A}}|\underline{y_{A}})
    \end{equation}
    where $\underline{x_{A}}=(x_{A,\alpha})_{\alpha=1,\ldots,n_{A}}$, $\underline{y_{A}}=(y_{A,\alpha})_{\alpha=1,\ldots,m_{A}}$. Note that the coefficients appearing are the supergroup analog of the gauge origami partition function of vector and folded instantons.
\end{theorem}
\begin{remark}
    The above theorem claims that for example if we fix the screening charge $\mathscr{Q}_{4}(x)$, finite products of $\D4$ $qq$-characters of $A=12,23,13$ will give higher rank $qq$-characters. Such kind of $qq$-characters have nice physical meanings because the appearing coefficients are related to the gauge origami partition function of intersecting $\D4$-branes inside the $\mathbb{C}^{3}_{123}$ subspace. However, since the $qq$-characters such as $\mathsf{T}_{A}(x)\,(A=14,24,34)$ do not commute with $\mathscr{Q}_{4}(x)$, we can not use $\mathsf{X}_{A}(x)^{\pm1}(A=14,24,34)$ as highest weights. Thus, the $qq$-character whose coefficients are the most general gauge origami partition (e.g. crossed instantons) does not appear by just imposing the commutativity with one of the screening charges.
\end{remark}
\begin{remark}
    Using the fact that screening charges of different types commute with each other, $\mathsf{S}_{1,2,3}(x)$ can also be the highest weight of the $qq$-characters commuting with $\mathscr{Q}_{4}(x)$. Such kinds of $qq$-characters are expected to be physically related with a $\D2\tbar\D4$ coupled system. We leave the analysis of these cases for future work.
\end{remark}

\subsection{Quadratic relations of crossed instantons}\label{sec:D4quadraticrelation}
The $qq$-characters $\{\mathsf{T}_{A}(x)\}_{A\in\six}$ are expected to generate a larger algebra than the affine quiver W-algebra. As usual deformed W-algebras, studying the quadratic relations of the generators is another way to understand the algebraic structure. Since deriving the complete quadratic relations of the $\{\mathsf{T}_{A}(x)\}_{A\in\six}$ is beyond the scope of this paper, we just give a part of the quadratic relations which can be derived easily.
\begin{theorem}
    The $\D4$ $qq$-characters $\mathsf{T}_{A}(x)$ and $\mathsf{T}_{\bar{A}}(x)$ for $A\in\six$ (anti-)commute with each other up to trivial zero-modes factors:
    \begin{equation}
       x\mathsf{T}_{A}(x)\mathsf{T}_{\bar{A}}(x')+q_{A}x'\mathsf{T}_{\bar{A}}(x')\mathsf{T}_{A}(x)=0.
    \end{equation}
\end{theorem}
\begin{proof}
Let us consider the commutation relations between $\mathsf{T}_{12}(x)$ and $\mathsf{T}_{34}(x')$:
\begin{align}
    \mathsf{T}_{12}(x)&=\sum_{\lambda}F^{\D4}_{12}[\lambda]:\Lambda_{12,\lambda}(x):,\\
    \mathsf{T}_{34}(x')&=\sum_{\mu}F^{\D4}_{34}[\mu]:\Lambda_{34,\mu}(x'):.
\end{align}
The contraction formulas give 
\bea
    \Lambda_{12,\lambda}(x)\Lambda_{34,\mu}(x')&=\left[\mathcal{Z}_{34\,|\,12}^{\text{tot}}(x',\mu\,|\,x,\lambda)\right]_{|x'/x|\ll1}:\Lambda_{12,\lambda}(x)\Lambda_{34,\mu}(x'):,\\
    \Lambda_{34,\mu}(x')\Lambda_{12,\lambda}(x)&=\left[\mathcal{Z}_{12\,|\,34}^{\text{tot}}(x,\lambda\,|\,x',\mu)\right]_{|x/x'|\ll1}:\Lambda_{12,\lambda}(x)\Lambda_{34,\mu}(x'):,
\eea
where
\begin{equation}
\mathcal{Z}_{\bar{A}|A}^{\text{tot.}}(x',\mu|x,\lambda)=\mathcal{Z}_{\text{1-loop}}^{\D4\tbar\D4}(x',\bar{A}|x,A)\mathcal{Z}_{\bar{A}|A}^{\D4\tbar\D4}(x',\mu|x,\lambda).
\end{equation}
Noting that the one-loop factors of the crossed instantons are rational functions
\beq
    \mathcal{Z}_{\text{1-loop}}^{\D4\tbar\D4}(x',\bar{X}|x,X)=\left(1-q_{A}x'/x\right),
\eeq
we then obtain
\begin{equation}
\begin{split}
    &\mathsf{T}_{12}(x)\mathsf{T}_{34}(x')-\left(-q_{12}\frac{x'}{x}\right)\mathsf{T}_{34}(x')\mathsf{T}_{12}(x)\\
    =&\sum_{\lambda,\mu}F^{\D4}_{12}[\lambda]F^{\D4}_{34}[\mu]\left\{\sum_{\substack{\AboxF\in A(\mu)\\\Abox\in R(\lambda)}}\underset{\substack{\chi_{12,x}(\Abox)\\=\chi_{34,x'}(\AboxF)}}{\Res}\left(\frac{x}{x'}\right)^{-1}\mathcal{Z}_{34\,|\,12}^{\text{tot}}(x',\mu\,|\,x,\lambda)\delta\left(\frac{\chi_{12,x}(\Bbox)}{\chi_{34,x'}(\BboxF)}\right):\Lambda_{12,\lambda}(x)\Lambda_{34,\mu}(x'):\right.\\
    &\left.+\sum_{\substack{\AboxF\in R(\mu)\\\Abox\in A(\lambda)}}\underset{\substack{\chi_{12,x}(\Abox)\\=\chi_{34,x'}(\AboxF)}}{\Res}\left(\frac{x}{x'}\right)^{-1}\mathcal{Z}_{34\,|\,12}^{\text{tot}}(x',\mu\,|\,x,\lambda)\delta\left(\frac{\chi_{12,x}(\Bbox)}{\chi_{34,x'}(\BboxF)}\right):\Lambda_{12,\lambda}(x)\Lambda_{34,\mu}(x'):\right\}.
\end{split}
\end{equation}
Similar to the proof in Thm.~\ref{thm:D4qq-commute}, we redefine the sum of the first term as $\lambda\rightarrow \lambda=\lambda'+\Bbox$, $\mu\rightarrow \mu=\mu'+\BboxF$. After this, the first term will be 
\begin{equation}
\begin{split}
    &\sum_{\lambda',\mu'}\sum_{\substack{\AboxF\in R(\mu')\\\Abox\in A(\lambda')}} F^{\D4}_{12}[\lambda'+\Bbox]F^{\D4}_{34}[\mu'-\BboxF]\underset{\substack{\chi_{12,x}(\Abox)\\=\chi_{34,x'}(\AboxF)}}{\Res}\left(\frac{x}{x'}\right)^{-1}\mathcal{Z}_{34\,|\,12}^{\text{tot}}(x',\mu'-\BboxF\,|\,x,\lambda'+\Bbox)\\
    &\qquad\times\delta\left(\frac{\chi_{12,x}(\Bbox)}{\chi_{34,x'}(\BboxF)}\right):\Lambda_{12,\lambda'+\Abox}(x)\Lambda_{34,\mu'-\AboxF}(x'):.
\end{split}
\end{equation}
Using 
\begin{equation}
    \frac{F^{\D4}_{12}[\lambda+\Bbox]F^{\D4}_{34}[\mu-\BboxF]}{F^{\D4}_{12}[\lambda]F^{\D4}_{34}[\mu]}=-\frac{\underset{\substack{\chi_{12,x}(\Abox)\\=\chi_{34,x'}(\AboxF)}}{\Res}\left(\frac{x}{x'}\right)^{-1}\mathcal{Z}_{34\,|\,12}^{\text{tot}}(x',\mu\,|\,x,\lambda)}{\underset{\substack{\chi_{12,x}(\Abox)\\=\chi_{34,x'}(\AboxF)}}{\Res}\left(\frac{x}{x'}\right)^{-1}\mathcal{Z}_{34\,|\,12}^{\text{tot}}(x',\mu-\BboxF\,|\,x,\lambda+\Bbox)}.
\end{equation}
and
\begin{equation}
    \delta\left(\frac{\chi_{12,x}(\Bbox)}{\chi_{34,x'}(\BboxF)}\right):\Lambda_{12,\lambda+\Abox}(x)\Lambda_{34,\mu-\AboxF}(x'):=\delta\left(\frac{\chi_{12,x}(\Bbox)}{\chi_{34,x'}(\BboxF)}\right):\Lambda_{12,\lambda}(x)\Lambda_{34,\mu}(x'):
\end{equation}
we obtain the claim.
\end{proof}

\begin{remark}
    The extra factors in front of the $qq$-characters come from the operator product of the zero modes $\mathsf{x}_{A,0}(x)$. We may modify the zero modes so that they obey \eqref{eq:D4D4zero} (see also footnote \ref{note:D4footnote}) and then get an exact commuting relation $[\mathsf{T}_{A}(x),\mathsf{T}_{\bar{A}}(x')]=0$. Instead of doing that, we simply relax the commutativity condition and say that operators $\mathcal{O}(x)$ and $\mathcal{O}'(x')$ commute when they satisfy 
    \begin{align}
        \mathcal{O}(x)-f(x,x')\mathcal{O'}(x')=0\label{eq:weakcommute}
    \end{align}
    where $f(x,x')$ are zero modes factors.
\end{remark}


\section{D6-brane \texorpdfstring{$qq$}{qq}-characters}\label{sec:D6qqcharacter}
We introduce the operator version of the D6 $qq$-characters in section~\ref{sec:D6qqchdef}. We show that monomial terms of the $qq$-character with the highest weight $\mathsf{W}_{\bar{a}}(x)$ are classified by plane partitions and that it commutes with the screening charge $\mathscr{Q}_{a}(x')$. We then show that the D6 $qq$-characters reproduce the tetrahedron instanton partition function in section~\ref{sec:tetrainstD6qq}. The D6 $qq$-characters can be obtained by fusion of an infinite number of D4 $qq$-characters (see section~\ref{sec:D4fusiontoD6}). Description in lower dimensional $qq$-characters are discussed in section~\ref{sec:D6lowerdim}. Finally, in section~\ref{sec:generalD6qq}, we introduce general D6 $qq$-characters where truncations of plane partitions appear. We also give a conjecture of the D6 $qq$-characters associated with toric CY$_{4}$ taking the form toric $\CY_{3}\times\mathbb{C}$. The $qq$-characters we introduce in this section imply the existence of a large class of $qq$-characters which we call BPS $qq$-characters.

\subsection{D6 \texorpdfstring{$qq$}{qq}-characters and plane partition}\label{sec:D6qqchdef}
Following the construction of the $\D4$ $qq$-characters and the affine quiver W-algebra, let us construct the $\D6$ $qq$-characters. The $\D6$ $qq$-characters are $qq$-characters whose generating currents are the $\D6$ vertex operators $\mathsf{W}_{\bar{a}}(x)\,(a\in\four)$ in \eqref{eq:D6op}. Similar to the $\D4$ case, each term of the $\D6$ $qq$-character is decomposed into the vertex operator and the coefficient parts. The operator part is determined by the iWeyl reflection of the operator part defined as the following.
\begin{definition}
    The iWeyl reflection of the $\D6$ vertex operator $\mathsf{W}_{\bar{a}}(x)\,(a\in\four)$ is
    \beq
        \mathsf{W}_{\bar{a}}(x)\longrightarrow {:\mathsf{W}_{\bar{a}}(x)\mathsf{A}^{-1}(x):},\quad a\in\four.
    \eeq
    Using the property in \eqref{eq:oprelationwithD0-D6}, the root current is rewritten in the $\D6$ vertex operators as
\beq
    \mathsf{W}_{\bar{a}}(x)\longrightarrow {:\mathsf{a}_{0}(x)^{-1}\mathsf{W}_{\bar{a}}(q_{a}^{-1}x)\frac{\prod\limits_{i\in\four\setminus \{a\}}\mathsf{W}_{\bar{a}}(q_{i}x)}{\prod\limits_{i\in\four\setminus \{a\}}\mathsf{W}_{\bar{a}}(q_{i}^{-1}q_{a}^{-1}x)}:}
\eeq
\end{definition}
The operator part is obtained by changing the root currents to the $\D6$-operators and applying the iWeyl reflection to the numerators recursively. After iWeyl reflections, the operators are classified by plane partitions. Let us study some terms after applying the iWeyl reflection. The operators obtained after $n$-times of iWeyl reflections are called operators at level $n$. We focus on the case $\mathsf{W}_{\bar{4}}(x)$ and then have
\beq
    \mathsf{W}_{\bar{4}}(x)\longrightarrow {:\mathsf{a}_{0}(x)^{-1}\frac{\mathsf{W}_{\bar{4}}(q_{123}x)\prod_{i=1}^{3}\mathsf{W}_{\bar{4}}(q_{i}x)}{\prod_{1\leq i<j\leq 3}\mathsf{W}_{\bar{4}}(q_{i}q_{j}x)}:}.
\eeq
\begin{itemize}
    \item Level 0: We only have one operator 
\beq
    \mathsf{W}_{\bar{4}}(x).
\eeq
We can associate this operator with a vacuum configuration of the plane partition in the space (1,2,3) where there is no box:
\bea
    \mathsf{W}_{\bar{4}}(x)\quad\longleftrightarrow\quad  \adjustbox{valign=c}{\begin{tikzpicture}[scale=0.5]
    \draw[->] (0,0)--(-30:2); 
    \draw[->] (0,0)--(210:2);
    \draw[->] (0,0)--(90:2);
    \node at (210:2.4){$1$};
    \node at (-30:2.4){$2$};
    \node at (90:2.4){$3$};
\end{tikzpicture}}
\eea
The spectral parameter $x$ corresponds to the $q$-coordinates in the origin. The operator $\mathsf{W}_{\bar{4}}(x)$ represents the addable box of this plane partition configuration which is the box in the origin.
\item Level 1: After applying the iWeyl reflection once, we have 
\beq
    {:\mathsf{W}_{\bar{4}}(x)\mathsf{A}^{-1}(x):}={:\mathsf{a}_{0}(x)^{-1}\frac{\mathsf{W}_{\bar{4}}(q_{123}x)\prod_{i=1}^{3}\mathsf{W}_{\bar{4}}(q_{i}x)}{\prod_{i<j\leq 3}\mathsf{W}_{\bar{4}}(q_{i}q_{j}x)}:}.
\eeq
The level 1 current can be described as 
\bea
:\mathsf{a}_{0}(x)^{-1}\frac{\textcolor{blue}{\mathsf{W}_{\bar{4}}(q_{4}^{-1}x)}\textcolor{red}{\prod_{i=1}^{3}\mathsf{W}_{\bar{4}}(q_{i}x)}}{\prod_{i<j\leq 3}\mathsf{W}_{\bar{4}}(q_{i}q_{j}x)}:\quad \longleftrightarrow \quad
\adjustbox{valign=c}{
\begin{tikzpicture}[scale=0.5]
    \draw[->] (0,0)--(-30:2); 
    \draw[->] (0,0)--(210:2);
    \draw[->] (0,0)--(90:2);
    \planepartition{{1}} 
    \node at (210:2.4){$1$};
    \node at (-30:2.4){$2$};
    \node at (90:2.4){$3$};
\end{tikzpicture}}
\eea
An observation is that the red terms $\mathsf{W}_{\bar{4}}(q_{i}x)\,(i=1,2,3)$ correspond to the addable boxes of this plane partition configuration. The variables $q_{i}x\,(i=1,2,3)$ correspond to the $q$-coordinates of the addable boxes. The blue term $\mathsf{W}_{\bar{4}}(q_{4}^{-1}x)$ corresponds to the removable box of the configuration, which is the box in the origin with coordinate $x$.

\item Level 2: Since we have three numerators in the level 1 operator, we will have three possible level 2 operators depending on which term we do the iWeyl reflection.

\begin{enumerate}
    \item $\pi_{1,1}=1,\,\,\pi_{2,1}=1$
    \bea
        {:\mathsf{W}_{\bar{4}}(x)\mathsf{A}^{-1}(x)\mathsf{A}^{-1}(q_{1}x):}\propto \frac{\textcolor{red}{(2)(3)(1^{2})}\textcolor{blue}{(14^{-1})}}{(1^{2}2)(1^{2}3)}\quad \longleftrightarrow \quad \adjustbox{valign=c}{
\begin{tikzpicture}[scale=0.5]
    \draw[->] (0,0)--(-30:2); 
    \draw[->] (0,0)--(210:2.8);
    \draw[->] (0,0)--(90:2);
    \planepartition{{1},{1}} 
    \node at (210:3.1){$1$};
    \node at (-30:2.4){$2$};
    \node at (90:2.4){$3$};
\end{tikzpicture}}
     \eea
    where we simply denote $\mathsf{W}_{\bar{4}}(q_{1}^{a}q_{2}^{b}q_{3}^{c}q_{4}^{d}x)$ as $(1^{a}2^{b}3^{c}4^{d})$. The $q$-coordinates of the addable boxes of this configuration are $q_{1}^{2}x, q_{2}x,q_{3}x$ and correspond to the red terms in the numerator. The $q$-coordinate of the removable box of this configuration is $xq_{1}$ and correspond to the blue term with variable $q_{4}^{-1}q_{1}x$.
    
    \item $\pi_{1,1}=1,\pi_{1,2}=1$
    \bea
    {:\mathsf{W}_{\bar{4}}(x)\mathsf{A}^{-1}(x)\mathsf{A}^{-1}(q_{2}x):}\propto \frac{\textcolor{red}{(1)(3)(2^{2})}\textcolor{blue}{(12^{2}3)}}{(12^{2})(2^{2}3)}\quad \longleftrightarrow \quad \adjustbox{valign=c}{
\begin{tikzpicture}[scale=0.5]
    \draw[->] (0,0)--(-30:2.8); 
    \draw[->] (0,0)--(210:2);
    \draw[->] (0,0)--(90:2);
    \planepartition{{1,1}} 
    \node at (210:2.4){$1$};
    \node at (-30:3.1){$2$};
    \node at (90:2.4){$3$};
\end{tikzpicture}}
    \eea
    Similarly, the red terms correspond to the addable boxes and the blue terms correspond to the removable boxes of the plane partition configuration.
    \item $\pi_{1,1}=2$
    \bea
    {:\mathsf{W}_{\bar{4}}(x)\mathsf{A}^{-1}(x)\mathsf{A}^{-1}(q_{3}x):}\propto \frac{\textcolor{red}{(1)(2)(3^{2})}\textcolor{blue}{(123^{2})}}{(13^{2})(23^{2})}\quad \longleftrightarrow \quad \adjustbox{valign=c}{
\begin{tikzpicture}[scale=0.5]
    \draw[->] (0,0)--(-30:2); 
    \draw[->] (0,0)--(210:2);
    \draw[->] (0,0)--(90:2.6);
    \planepartition{{2}} 
    \node at (210:2.4){$1$};
    \node at (-30:2.4){$2$};
    \node at (90:3.1){$3$};
\end{tikzpicture}}
    \eea
    Again, the red terms correspond to the addable boxes and the blue terms correspond to the removable boxes of the plane partition configuration.
\end{enumerate}
\end{itemize}
We can do this procedure recursively and then obtain the following statement.
\begin{lemma}\label{lemma:D6_iWeyl_ref}
The operators generated from the iWeyl reflection starting from $\mathsf{W}_{\bar{a}}(x)$ are classified by plane partitions:
\begin{equation}
    \Lambda_{\bar{a},\pi}(x)\coloneqq {:\mathsf{W}_{\bar{a}}(x)\prod_{\scube\in\pi}\mathsf{A}^{-1}(\chi_{\bar{a},x}(\cube)):},\quad a\in\four.
\end{equation}
Converting the root currents into the $\D6$-operators, we have 
\begin{equation}
    \mathsf{W}_{\bar{a}}(x)\prod\limits_{\scube\in\pi}\mathsf{A}^{-1}(\chi_{\overbar{a},x}(\cube))\propto \frac{\prod\limits_{\scube\in A(\pi)}\mathsf{W}_{\bar{a}}(\chi_{\overbar{a},x}(\cube))\prod\limits_{\scube\in R(\pi)}\mathsf{W}_{\bar{a}}(q_{a}^{-1}\chi_{\overbar{a},x}(\cube))}{\#}
\end{equation}
where we have some extra denominators determined recursively.\footnote{The explicit form of the right-hand side is related to the shell formula of the plane partition \cite{Feigin2011plane}. In this note, the information of the denominator is not necessary so we will not write the explicit form. }  
\end{lemma}

The $\D6$ $qq$-character is a sum of the vertex operators $\Lambda_{\bar{a},\pi}(x)$ with some specific coefficients:
\beq
    \mathsf{T}_{\bar{a}}(x)=\sum_{\pi\in\mathcal{PP}}F^{\D6}_{\bar{a}}[\pi]\Lambda_{\bar{a},\pi}(x).
\eeq
We impose the condition that this $qq$-character commutes with the screening charge $\mathscr{Q}_{a}(x)$. After imposing this condition, the coefficient is determined uniquely.
\begin{definition}\label{def:D6_qq-ch}
    We define the $\D6$ $qq$-character for $a\in\four$ as 
    \beq
        \mathsf{T}_{\bar{a}}(x)=\sum_{\pi\in\mathcal{PP}}\widetilde{\mathcal{Z}}^{\D6}_{\bar{a}}[\pi]\Lambda_{\bar{a},\pi}(x),\quad a\in\four,
    \eeq
    where the coefficients $\widetilde{\mathcal{Z}}^{\D6}_{\bar{a}}[\pi]$ are identified with the $\U(1)$ partition function of the 7d gauge theory on $\mathbb{C}^{3}_{\bar{a}}\times \mathbb{S}^{1}$ in \eqref{eq:D6tetinst_partfunct}. Rescaling the zero-modes $\mathsf{A}(x)\rightarrow \mathfrak{q}^{-1}\mathsf{A}(x)$, we have 
    \beq
        \mathsf{T}_{\bar{a}}(x)=\sum_{\pi\in\mathcal{PP}}\mathfrak{q}^{|\pi|}\widetilde{\mathcal{Z}}^{\D6}_{\bar{a}}[\pi]\Lambda_{\bar{a},\pi}(x).
    \eeq
    
\end{definition}

\begin{theorem}\label{thm:D6qq-commute}
    The $\D6$ $qq$-characters $\mathsf{T}_{\bar{a}}(x)\,(a\in\four)$ commutes with the screening charge $\mathscr{Q}_{a}(x)$:
    \beq
       \relax [\mathsf{T}_{\bar{a}}(x),\mathscr{Q}_{a}(x')]=0.
    \eeq
\end{theorem}
\begin{proof}
Let us focus on $\mathsf{T}_{\bar{4}}(x)$ and see how the commutativity appears. Using the formulas in Thm. \ref{eq:app-contractions} and the property in \eqref{eq:oprelation1}, we have
\begin{align}
\begin{split}
    [\mathsf{W}_{\bar{4}}(x),\mathsf{S}_{4}(x')]&=q_{4}x\delta\left(q_{4}x/x'\right):\mathsf{W}_{\bar{4}}(x)\mathsf{S}_{4}(q_{4}x):\\
    [{:\mathsf{W}_{\bar{4}}(x)\mathsf{A}^{-1}(x):},\mathsf{S}_{4}(x')]&=x\delta\left(x'/x\right)\frac{\prod_{i=1}^{3}(1-q_{i})}{\prod_{1\leq i<j\leq 3}(1-q_{i}q_{j})}:\mathsf{W}_{123}(x)\mathsf{S}_{4}(q_{4}x):+\cdots,
    \end{split}
\end{align}
where for the second term, we only extracted the pole coming from $x'=x$. Thus, for the pole coming from $\mathsf{W}_{\bar{4}}(x)$ to disappear up to a total difference, we need the combination
\beq
    \mathsf{W}_{\bar{4}}(x)-q_{4}\frac{\prod_{1\leq i<j\leq 3}(1-q_{i}q_{j})}{\prod_{i=1}^{3}(1-q_{i})}:\mathsf{W}_{\bar{4}}(x)\mathsf{A}^{-1}(x):
\eeq
where the coefficient is $\widetilde{Z}_{\bar{4}}^{\D6}[\,\cube\,]$ (see Appendix \eqref{eq:app-D6level1}). Generally, using 
\begin{align}
\begin{split}
  \Lambda_{\bar{a},\pi}(x)\mathsf{S}_{4}(x') &=-q_{4}x\left[\mathscr{W}^{\overbar{4}}_{\pi,x}(q_{4}^{-1}x')^{-1}\right]^{x'}_{-}:\Lambda_{\bar{a},\pi}(x)\mathsf{S}_{4}(x'):,\\
  \mathsf{S}_{4}(x')\Lambda_{\bar{a},\pi}(x)&=-q_{4}x\left[\mathscr{W}_{\pi,x}^{\overbar{4}}(q_{4}^{-1}x')^{-1}\right]^{x'}_{+}:\Lambda_{\bar{a},\pi}(x)\mathsf{S}_{4}(x'):
\end{split}
\end{align}
and the property in \eqref{eq:D6Nekrasov-shell}, we obtain
\begin{align}
\begin{split}
    &[\mathsf{T}_{\bar{4}}(x),\mathsf{S}_{4}(x')]\\
=&q_{4}x\sum_{\pi\in\mathcal{PP}}\widetilde{\mathcal{Z}}_{\bar{4}}^{\text{D6}}[\pi]\left(\sum_{\scube\,\in A(\pi)}\underset{x'=q_{4}\chi_{\overbar{4},x}(\scube)}{\Res}{x'}^{-1}\mathscr{W}^{\overbar{4}}_{\pi,x}(q_{4}^{-1}x')^{-1}\delta\left(\frac{x'}{q_{4}\chi_{\overbar{4},x}(\cube)}\right):\Lambda_{\bar{4},\pi}(x)\mathsf{S}_{4}(q_{4}\chi_{\overbar{4},x}(\cube)):\right.\\
    &\qquad\left.+\sum_{\scube\in R(\pi)}\underset{x'=\chi_{\overbar{4},x}(\scube)}{\Res}{x'}^{-1}\mathscr{W}^{\overbar{4}}_{\pi,x}(q_{4}^{-1}x')^{-1}\delta\left(\frac{x'}{\chi_{\overbar{4},x}(\cube)}\right):\Lambda_{\bar{4},\pi}(x)\mathsf{S}_{4}(\chi_{\overbar{4},x}(\cube)):\right).
\end{split}
\end{align}
Shifting the second term as $\pi'=\pi-\scube$, the second term will be rewritten as
\bea
\sum_{\pi'\in\mathcal{PP}}\widetilde{\mathcal{Z}}_{\bar{4}}^{\text{D6}}[\pi'+\cube]\sum_{\scube\in A(\pi')}\underset{x'=\chi_{\overbar{4},x}(\scube)}{\Res}{x'}^{-1}\mathscr{W}^{\overbar{4}}_{\pi'+\scube,x}(q_{4}^{-1}x')^{-1}\delta\left(\frac{x'}{\chi_{\overbar{4},x}(\cube)}\right):\Lambda_{\bar{4},\pi'+\scube}(x)\mathsf{S}_{4}(\chi_{\overbar{4},x}(\cube)):.
\eea
Using Thm.~\ref{eq:app-thm-D6U1recursionformula}, we have
\begin{align}
\begin{split}
    {:\Lambda_{\bar{4},\pi'+\scube}(x)\mathsf{S}_{4}(\chi_{\overbar{4},x}(\cube)):}={:\Lambda_{\bar{4},\pi'}(x)\mathsf{S}_{4}(q_{4}\chi_{\bar{4},x}(\cube)):},\\
    \frac{\widetilde{\mathcal{Z}}^{\D6}_{\bar{4}}[\pi+\cube]}{\widetilde{\mathcal{Z}}^{\D6}_{\bar{4}}[\pi]}=-\frac{\underset{x'=q_{4}\chi_{\overbar{4},x}(\scube)}{\Res}{x'}^{-1}\mathscr{W}^{\overbar{4}}_{\pi,x}(q_{4}^{-1}x')^{-1}}{\underset{x'=\chi_{\overbar{4},x}(\scube)}{\Res}{x'}^{-1}\mathscr{W}^{\overbar{4}}_{\pi+\scube,x}(q_{4}^{-1}x')^{-1}}
\end{split}
\end{align}
and then we obtain
\begin{align}
\begin{split}
    &[\mathsf{T}_{123}(x),\mathsf{S}_{4}(x')]\\
    =&q_{4}x\sum_{\pi\in\mathcal{PP}}\widetilde{\mathcal{Z}}_{\bar{4}}^{\text{D6}}[\pi]\sum_{\scube\,\in A(\pi)}\underset{x'=q_{4}\chi_{\overbar{4},x}(\scube)}{\Res}{x'}^{-1}\mathscr{W}^{\overbar{4}}_{\pi,x}(q_{4}^{-1}x')^{-1}:\mathsf{W}_{\bar{4}}(x)\prod_{\scube\in\pi}\mathsf{A}^{-1}(\chi_{\overbar{4},x}(\cube))\mathsf{S}_{4}(q_{4}\chi_{\overbar{4},x}(\cube))\\
    &\times \left(\delta\left(\frac{x'}{\chi_{\overbar{4},x}(\cube)}\right)-\delta\left(\frac{x'}{q_{4}\chi_{\overbar{4},x}(\cube)}\right)\right)
\end{split}
\end{align}
which gives the claim.
\end{proof}

The Thm.~\ref{thm:D2qq-commute}, \ref{thm:D4qq-commute}, \ref{thm:D6qq-commute} are summarized in the following theorem.
\begin{theorem}\label{thm:tetrascreening}
Let $\mathscr{T}$ be the tetrahedron corresponding to the $\mathbb{C}^{4}$ geometry (see Figure \ref{fig:complex}). We denote the set of vertices, edges, and faces of $\mathscr{T}$ as $\mathsf{v}=\{a\mid a\in\four\}$, $\mathsf{e}=\{A\mid A\in\six\}$, $\mathsf{f}=\{\bar{a}\mid a\in\four\}$ respectively. We also introduce the union of them as $\mathscr{S}=\mathsf{v}\cup\mathsf{e}\cup\mathsf{f}$. For each element of $i\in\mathscr{S}$, we can associate a $qq$-character. If $i\in\mathsf{v}$, we associate the $\D2$ $qq$-character (screening charge), if $i\in\mathsf{e}$, we associate the $\D4$ $qq$-character, if $i\in\mathsf{f}$, we associate the $\D6$ $qq$-character. The $qq$-character associated with the elements $i,j\in\mathscr{T}$ commute with each other up to trivial zero modes (see \eqref{eq:weakcommute}) when $i$ and $j$ do not intersect in $\mathscr{T}$:
\beq
    \mathsf{T}_{i}(x)\mathsf{T}_{j}(x')-f_{ij}(x,x')\mathsf{T}_{j}(x')\mathsf{T}_{i}(x)=0 \quad \Longleftrightarrow \quad i \cap j =\emptyset,
\eeq
where $f_{ij}(x,x')$ are zero modes.
\begin{align*}\adjustbox{valign=c}{\includegraphics[width=6cm]{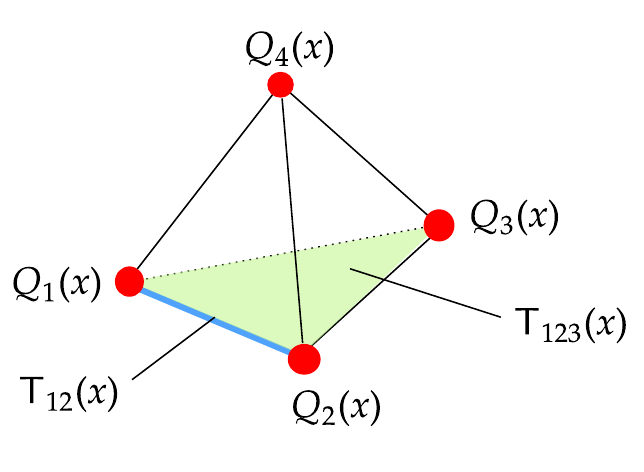}}\end{align*}
\end{theorem}

\begin{remark}
Note that the D6 $qq$-characters $\mathsf{T}_{\bar{a}}(x)\,(a\in\four)$ are the operator version of the $\D6$ $qq$-characters introduced in \eqref{eq:D6qqchpartition}:
\bea
\mathsf{T}_{\bar{a}}(x)\quad \longleftrightarrow  \quad\widehat{\mathscr{T}}^{a}(x),\,\widehat{\mathscr{T}}^{a\vee}(x),\quad a\in \four.
\eea
The highest weight and the root current have the following correspondence:
\bea
\mathsf{W}_{\bar{a}}(x)\quad &\longleftrightarrow \quad \widehat{\mathscr{U}}^{a}(q_{a}x)^{-1},\widehat{\mathscr{U}}^{a\vee}(q_{a}x)^{-1},\quad a\in\four\\
\mathsf{A}(x)^{-1}\quad &\longleftrightarrow \quad  \widehat{\mathscr{A}}(x)^{-1},\widehat{\mathscr{A}}^{\,\vee}(x)^{-1}.
\eea
\end{remark}

\begin{remark}
    Choosing the screening charge $\mathscr{Q}_{4}(x)$, Thm.~\ref{thm:tetrascreening} claims that the kernel of $\mathscr{Q}_{4}(x)$ is generated by $\mathscr{Q}_{1,2,3}(x),\mathsf{T}_{12,23,13}(x), \mathsf{T}_{123}(x)$. We expect that this will give a larger algebra compared with the affine quiver W-algebra given by the kernel of two screening charges. We leave a detailed analysis of this for future work.
\end{remark}

\begin{remark}
We have shown here that the $\D6$ $qq$-character $\mathsf{T}_{\bar{a}}(x)$ is uniquely determined by the commutativity with the screening charge $\mathscr{Q}_{a}(x)$ and that the coefficient factor is the $\U(1)$ partition function of the 7d theory on $\mathbb{C}^{3}_{\bar{a}}\times \mathbb{S}^{1}$. Compared with the magnificent four system where we need to take care of the sign problem~\cite{Nekrasov:2017cih,Nekrasov:2018xsb,Nekrasov:2023nai}, our discussion here gives an algebraic proof showing that there will be no sign problem for the 7d case. This is indeed compatible with a mathematical proof given in~\cite{Fasola:2023ypx}.
\end{remark}

\subsection{Tetrahedron instantons and D6 \texorpdfstring{$qq$}{qq}-characters}\label{sec:tetrainstD6qq}
The expanded version of the tetrahedron instanton partition function can be expressed using the $\D6$ $qq$-characters.
\begin{lemma}\label{lem:D6operatorproduct}
    The operator product of $\{\Lambda_{\bar{a},\pi}(x)\}_{a\in\four}$ is 
    \bea
    \Lambda_{\bar{b},\pi^{(2)}}(x_{2})\Lambda_{\bar{a},\pi^{(1)}}(x_{1})&=\mathcal{Z}^{\D6\tbar\D6}_{1\tbar\text{loop}}(x_{1},\bar{a}\,|\,x_{2},\bar{b})\mathcal{Z}^{\D6\tbar\D6}_{\bar{a}\,|\,\bar{b}}(x_{1},\pi^{(1)}\,|\,x_{2},\pi^{(2)}):\Lambda_{\bar{a},\pi^{(1)}}(x_{1})\Lambda_{\bar{b},\pi^{(2)}}(x_{2}):.
\eea
When $a=b$, it gives the vector multiplet contributions, while when $a\neq b$, it gives the bifundamental contribution connecting gauge theories defined on different $\mathbb{C}^{3}$ subspaces.
\end{lemma}
\begin{theorem}\label{thm:tetra-origamiBPSCFT}
    The gauge origami partition function of the tetrahedron instanton is written using the $\D6$ $qq$-characters:
    \bea
        \mathcal{Z}_{\text{1-loop}}^{\D6}\mathcal{Z}_{\text{inst.}}^{\D6}=\sum_{\underline{\vec{\pi}}}\mathfrak{q}^{|\underline{\vec{\pi}}|}\mathcal{Z}^{\D6}_{\text{1-loop}}\mathcal{Z}^{\D6}_{\text{tet.inst.}}[\underline{\vec{v}},\underline{\vec{\pi}}]=\bra{0}\prod_{a\in\four}\prod_{\alpha=1}^{n_{\bar{a}}}\mathsf{T}_{\bar{a}}(v_{\bar{a},\alpha})\ket{0}.
    \eea
    Explicitly, we have the following table of BPS/CFT correspondence.
    \begin{align}
    \renewcommand\arraystretch{1.2}{
    \begin{tabular}{c|c}\toprule
        BPS &  CFT\\
     \hline  7d U(1) theory on $\mathbb{C}^{3}_{abc}\times \mathbb{S}^{1}$ ($\D6_{abc}\times 1$) & $\bra{0}\mathsf{T}_{abc}(v)\ket{0}$ \\
      7d U($n_{ab}$) theory on $\mathbb{C}^{3}_{abc}\times \mathbb{S}^{1}$ ($\D6_{abc}\times n_{abc}$)   & $\bra{0}\mathsf{T}_{abc}(v_{n_{abc}})\cdots\mathsf{T}_{abc}(v_{2})\mathsf{T}_{abc}(v_{1})\ket{0}$\\
      generalized folded instantons: D6$_{123}$-D6$_{234}$-$\D0$ & $\bra{0}\mathsf{T}_{123}(v)\mathsf{T}_{234}(v')\ket{0}$\\
      gauge origami of tetrahedron instantons & $\bra{0}\prod\limits_{a\in\four}\prod\limits_{\alpha=1}^{n_{\bar{a}}}\mathsf{T}_{\bar{a}}(v_{\bar{a},\alpha})\ket{0}$ \\ \toprule
    \end{tabular}}
\end{align}
\end{theorem}

\subsection{Fusion of D4 \texorpdfstring{$qq$}{qq}-characters to D6 \texorpdfstring{$qq$}{qq}-characters}\label{sec:D4fusiontoD6}
Let us see that the $\D6$ $qq$-characters are obtained by fusion of the $\D4$ $qq$-characters. By studying the zeros and pole structure, one can show the following lemma.
\begin{lemma}\label{lem:D6planecond}
 Given two Young diagrams $\lambda^{(1)},\lambda^{(2)}$, the contraction of the operators $\Lambda_{A,\lambda}(x)$ are given in Lemma \ref{lemm:D4ope} as
 \beq
     \Lambda_{A,\lambda^{(2)}}(x_{2})\Lambda_{A,\lambda^{(1)}}(x_{1})=\mathcal{Z}_{\text{1-loop}}^{\D4\tbar\D4}(x_{1},A\,|\,x_{2},A)\mathcal{Z}_{A|A}^{\D4\tbar\D4}(x_{1},\lambda^{(1)}\,|\,x_{2},\lambda^{(2)}): \Lambda_{A,\lambda^{(2)}}(x_{2})\Lambda_{A,\lambda^{(1)}}(x_{1}):.
 \eeq
 When the Young diagrams obey $\lambda^{(2)}\succ\lambda^{(1)}$ with the parameters as $x_{2}=q_{a}x_{1}\,(a\in\bar{A})$, we have 
 \beq
     \mathcal{Z}_{A|A}^{\D4\tbar\D4}(x_{1},\lambda^{(1)}\,|\,q_{a}x_{1},\lambda^{(2)})=0,\quad a\in\bar{A}
 \eeq
 which gives
 \beq
     \Lambda_{A,\lambda^{(2)}}(x_{2})\Lambda_{A,\lambda^{(1)}}(x_{1})=0.\label{eq:D4fusionproperty}
 \eeq
\end{lemma}
Using the above lemma, finite products of the $\D4$ $qq$-characters are given as 
\begin{align}
\begin{split}
     \mathsf{T}_{A}(q_{a}^{N-1}x)\cdots\mathsf{T}_{A}(q_{a}x)\mathsf{T}_{A}(x)&=\sum_{\lambda^{(N)}\preceq\cdots \lambda^{(2)}\preceq\lambda^{(1)}}\mathfrak{q}^{\sum_{i=1}^{N}|\lambda^{(i)}|}\widetilde{\mathcal{Z}}^{\D4}_{A}[\lambda^{(1)}]\cdots\widetilde{\mathcal{Z}}^{\D4}_{A}[\lambda^{(N)}]\\
     &\qquad\times\Lambda_{A,\lambda^{(N)}}(q_{a}^{N-1}x)\cdots\Lambda_{A,\lambda^{(2)}}(q_{a}x)\Lambda_{A,\lambda^{(1)}}(x)
\end{split}
\end{align}
for $\forall a\in\bar{A}$. Using the operator product and extracting the one-loop perturbative factor, we define the renormalized $N$-fusion $\D4$ $qq$-character as 
\beq
    \overbar{\mathsf{T}}^{(N)}_{A;a}(x)=\sum_{\lambda^{(N)}\preceq\cdots \preceq\lambda^{(1)}}\mathfrak{q}^{\sum_{i=1}^{N}|\lambda^{(i)}|}\prod_{i=1}^{N}\widetilde{\mathcal{Z}}^{\D4}_{A}[\lambda^{(i)}]\prod_{i<j}\mathcal{Z}_{A|A}^{\D4\tbar\D4}(q_{a}^{i-1}x,\lambda^{(i)}\,|\,q_{a}^{j-1}x,\lambda^{(j)}):\prod_{i=1}^{N}\Lambda_{A,\lambda^{(i)}}(q_{a}^{i-1}x):
\eeq
Note that this is the rank $N$ $qq$-character in \eqref{eq:ranknD4qqch} with spectral parameters tuned as $\bar{x}=(x_{i})_{i=1}^{N}$ and $x_{i}=q_{a}^{i-1}x$ for $a\in\bar{A}$. Taking the limit $N\rightarrow \infty$ and considering infinite products of the $qq$-characters, we can see that they are related to the $\D6$ $qq$-character.
\begin{theorem}\label{thm:D4toD6fusion}
    Taking the limit $N\rightarrow \infty$ of the renormalized $N$-fusion $\D4$ $qq$-characters $\overbar{\mathsf{T}}^{(N)}_{ab;c}(x)$ give the $\D6$ $qq$-character $\mathsf{T}_{abc}(x)$:
    \beq
        \overbar{\mathsf{T}}^{(N)}_{ab;c}(x)\xrightarrow{N\rightarrow \infty} \mathsf{T}_{abc}(x).
    \eeq
    Equivalently, we have 
    \beq
        \overleftarrow{\prod_{i=1}^{\infty}}\mathsf{T}_{ab}(xq_{c}^{i-1})\simeq \mathsf{T}_{abc}(x)
    \eeq
    where the symbol ``$\simeq$'' means the equality is true up to one-loop perturbative factors.
\end{theorem}
\begin{proof}
Let us focus on the case $\overbar{\mathsf{T}}^{(N)}_{12;3}(x)$ and $\mathsf{T}_{123}(x)$. Since the infinite products diverge, we need to regularize it properly (see for example \cite{Awata:2018svb}). Moreover, we need to take the inductive limit so that at large $N$, $\overbar{\mathsf{T}}_{12;3}^{(N)}(x)=\overbar{\mathsf{T}}^{(N+1)}_{12;3}(x)$ and we denote this $\overbar{\mathsf{T}}^{(\infty)}_{12;3}(x)$. We only give a sketch of the proof so see \cite{Awata:2018svb} for details. The operator $\overbar{\mathsf{T}}^{(\infty)}_{12;3}(x)$ is expanded as 
\bea
    \overbar{\mathsf{T}}^{(\infty)}_{12;3}(x)=\sum_{\emptyset\preceq\cdots \preceq\lambda^{(2)}\preceq\lambda^{(1)}}\mathfrak{q}^{\sum_{i=1}^{\infty}|\lambda^{(i)}|}\prod_{i=1}^{\infty}\widetilde{\mathcal{Z}}^{\D4}_{12}[\lambda^{(i)}]\prod_{i<j}\mathcal{Z}_{12|12}^{\D4\tbar\D4}(x_{i},\lambda^{(i)}\,|\,x_{j},\lambda^{(j)}):\prod_{i=1}^{\infty}\Lambda_{12,\lambda^{(i)}}(q_{3}^{i-1}x):
\eea
where $x_{i}=q_{3}^{i-1}x$. Since the right-hand side is expanded in all possible Young diagrams $\{\lambda^{(i)}\}$ obeying the condition $\lambda^{(i)}\succeq\lambda^{(i+1)}$, the right-hand side is expanded in all possible plane partitions (see the $(1,2)$ description in section \ref{sec:multi-dim-part}): $\pi=(\lambda^{(1)},\lambda^{(2)},\ldots,\emptyset,\emptyset,\ldots)$. Let us show that the coefficients and the operator part indeed can be written using the plane partition. The infinite product of the $\U(1)$ part is given by
\beq
\prod_{i=1}^{\infty}\widetilde{\mathcal{Z}}^{\D4}_{12}[\lambda^{(i)}]=\prod_{i=1}^{\infty}\prod_{\Abox\in\lambda^{(i)}}\mathscr{S}_{34}\left(\frac{x_{i}}{\chi_{12,x_{i}}(\Bbox)}\right)\prod_{\substack{\Abox\in\lambda^{(i)}\\\AboxF\in\lambda^{(i)}}}g_{\bar{4}}\left(\frac{\chi_{12,x_{i}}(\Bbox)}{\chi_{12,x_{i}}(\BboxF)}\right)^{-1}.
\eeq
We can rewrite the vector multiplet part as
\beq
    \prod_{i<j}\mathcal{Z}_{12|12}^{\D4\tbar\D4}(x_{i},\lambda^{(i)}\,|\,x_{j},\lambda^{(j)})=\prod_{ i<j}\prod_{\Abox\in\lambda^{(i)}}\mathscr{S}_{34}\left(\frac{x_{j}}{\chi_{12,x_{i}}(\Bbox)}\right)\prod_{\AboxF\in\lambda^{(j)}}\mathscr{S}_{34}\left(\frac{x_{i}}{\chi_{12,x_{j}}(\BboxF)}\right)\prod_{\substack{\Abox\in\lambda^{(i)}\\\AboxF\in\lambda^{(j)}}}\mathcal{A}_{\mathbb{C}^{4}}\left(\frac{\chi_{12,x_{i}}(\Bbox)}{\chi_{12,x_{j}}(\BboxF)}\right)^{-1}.
\eeq
Using 
\beq
\prod_{1\leq i<j\leq N}\prod_{\substack{\Abox\in\lambda^{(i)}\\\AboxF\in\lambda^{(j)}}}\mathcal{A}_{\mathbb{C}^{4}}\left(\frac{\chi_{12,x_{i}}(\Bbox)}{\chi_{12,x_{j}}(\BboxF)}\right)^{-1}=\prod_{i\neq j}\prod_{\substack{\Abox\in \lambda^{(i)}\\\AboxF\in\lambda^{(j)}}}g_{\bar{4}}\left(\frac{\chi_{12,x_{i}}(\Bbox)}{\chi_{12,x_{j}}(\BboxF)}\right)^{-1}
\eeq
and 
 \begin{align}
 \begin{split}
        &\prod_{1\leq i<j\leq \infty}\prod_{\Abox\in\lambda^{(i)}}\mathscr{S}_{34}\left(\frac{x_{j}}{\chi_{12,x_{i}}(\Bbox)}\right)\prod_{\AboxF\in\lambda^{(j)}}\mathscr{S}_{34}\left(\frac{x_{i}}{\chi_{12,x_{j}}(\BboxF)}\right)\\
&=\prod_{i=1}^{\infty}\prod_{\Abox\in\lambda^{(i)}}\mathscr{S}_{34}\left(\frac{x_{i}}{\chi_{12,x_{i}}(\Bbox)}\right)^{-1}\prod_{i=1}^{\infty}\prod_{\Abox\in\lambda^{(i)}}\frac{(1-q_{4}x/\chi_{12,x_{i}}(\Bbox))}{(1-x/\chi_{12,x_{i}}(\Bbox))},
    \end{split}
    \end{align}
the coefficient part will be 
\beq
    \prod_{i=1}^{\infty}\widetilde{\mathcal{Z}}^{\D4}_{12}[\lambda^{(i)}]\prod_{i<j}\mathcal{Z}_{12|12}^{\D4\tbar\D4}(x_{i},\lambda^{(i)}\,|\,x_{j},\lambda^{(j)})=\prod_{\scube\in\pi}\frac{(1-q_{4}x/\chi_{\bar{4},x}(\cube))}{(1-x/\chi_{\bar{4},x}(\cube))}\prod_{\substack{\cube\in\pi\\\cubeF\in\pi}}g_{\bar{4}}\left(\frac{\chi_{\bar{4},x}(\cube)}{\chi_{\bar{4},x}(\cubeF)}\right)^{-1}=\widetilde{\mathcal{Z}}_{\bar{4}}^{\D6}[\pi].
\eeq
The operator part is obtained from
\beq
    {:\prod_{i=1}^{\infty}\Lambda_{12,\lambda^{(i)}}(q_{3}^{i-1}x):}={:\prod_{i=1}^{\infty}\mathsf{X}_{12}(q_{3}^{i-1}x)\prod_{i=1}^{\infty}\prod_{\Abox\in\lambda^{(i)}}\mathsf{A}(\chi_{12,x_{i}}(\Bbox))^{-1}:}={:\mathsf{W}_{\bar{4}}(x)\prod_{\scube\in\pi}\mathsf{A}(\chi_{\bar{4},x}(\cube))^{-1}:}
\eeq
where we used 
\beq
    {:\prod_{i=1}^{\infty}\mathsf{X}_{12}(q_{3}^{i-1}x):}={:\prod_{i=1}^{\infty}\frac{\mathsf{W}_{\bar{4}}(q_{3}^{i-1}x)}{\mathsf{W}_{\bar{4}}(q_{3}^{i}x)}:}=\mathsf{W}_{\bar{4}}(x).
\eeq
We then get the identity $\overbar{\mathsf{T}}^{(\infty)}_{12;3}(x)=\mathsf{T}_{123}(x)$.
\end{proof}

\begin{remark}
The fusion procedure can be visualized as the case of the fusion of $\D2$ to $\D4$. Using the correspondence in \eqref{eq:D4Youngcorrespondence} (see also section \ref{sec:multi-dim-part}), we have 
\bea
\adjustbox{valign=c}{\includegraphics[width=5cm]{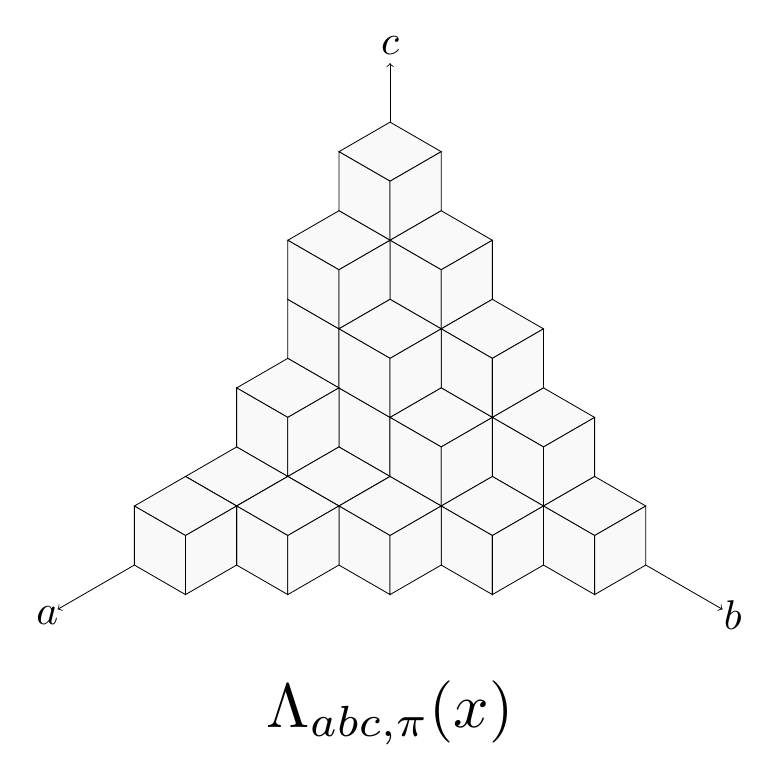}}\quad \longleftrightarrow \quad \adjustbox{valign=c}{\includegraphics[width=7cm]{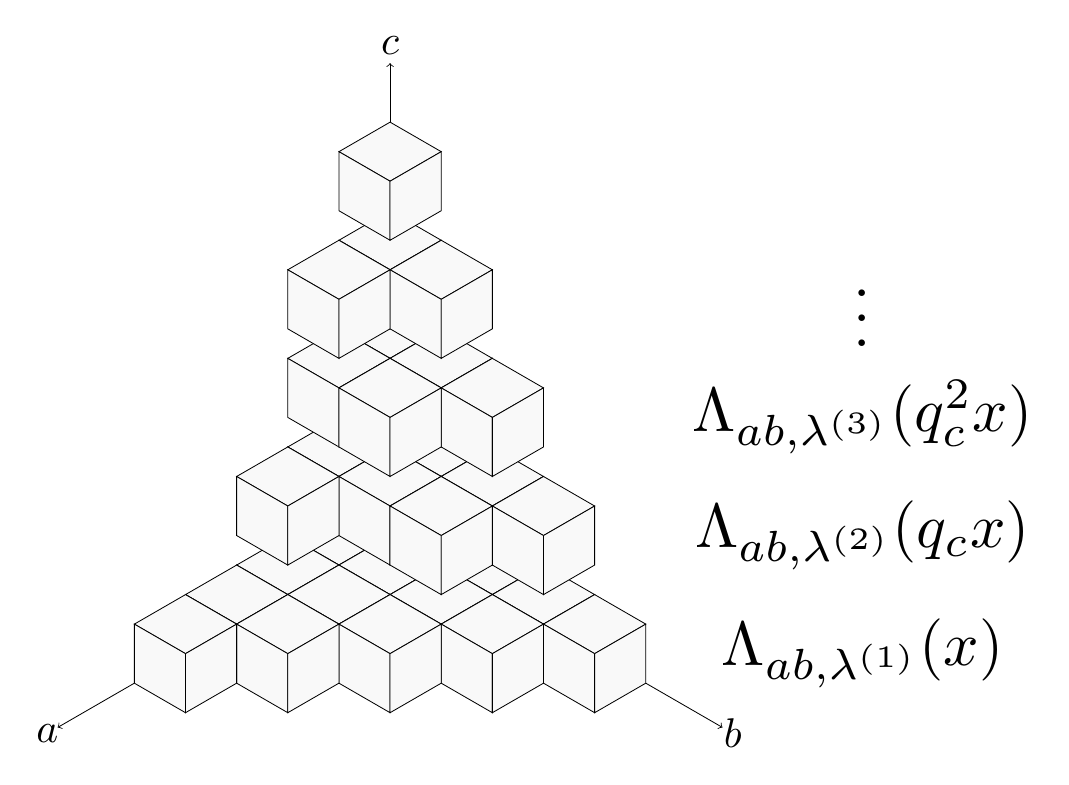}}
\eea

\end{remark}

\subsection{Description in lower dimensional \texorpdfstring{$qq$}{qq}-characters}\label{sec:D6lowerdim}
The previous theorem already implies that the gauge origami partition function of tetrahedron instantons can be reproduced algebraically using infinite products of lower dimensional $qq$-characters. Let us explicitly see this procedure.
\paragraph{\texorpdfstring{$(1,2)$}{(1,2)}-type description and D4 $qq$-characters}
Using the decomposition of the character in \eqref{eq:7dYoungdiagramdecomp}, the gauge origami partition function can be rewritten using the vertex operators $\Lambda_{A,\lambda}(x)$ as follows.
\begin{lemma}
    We have the following identity up to one-loop perturbative factors:
    \beq
\mathcal{Z}^{\D6}_{\text{1-loop}}\mathcal{Z}^{\D6}_{\text{tet.inst.}}[\underline{\vec{v}},\underline{\vec{\pi}}]\simeq \bra{0}\prod_{(abc)\in\four^{\vee}}\prod_{\alpha=1}^{n_{abc}}\prod_{k=1}^{\infty}\widetilde{\mathcal{Z}}^{\D4}_{ab}[\lambda_{abc,\alpha}^{(k)}]\Lambda_{ab,\lambda_{abc,\alpha}^{(k)}}(xq_{c}^{k-1})\ket{0} ,
    \eeq
    where the plane partitions on the right-hand side are rewritten in the $(1,2)$-type description as Young diagrams on the left-hand side.
\end{lemma}
Using the property \eqref{eq:D4fusionproperty}, we also have the following theorem.
\begin{theorem}
    The gauge origami partition function of the tetrahedron instanton system is written using infinite products of $\D4$ $qq$-characters.
    \beq
        \mathcal{Z}^{\D6}_{\text{1-loop}}\mathcal{Z}^{\D6}_{\text{inst.}}=\sum_{\underline{\vec{\pi}}}\mathfrak{q}^{|\underline{\vec{\pi}}|}\mathcal{Z}^{\D6}_{\text{1-loop}}\mathcal{Z}^{\D6}_{\text{tet.inst.}}[\underline{\vec{v}},\underline{\vec{\pi}}]\simeq \bra{0}\prod_{(abc)\in\four^{\vee}}\prod_{\alpha=1}^{n_{abc}}\overleftarrow{\prod_{i=1}^{\infty}}\mathsf{T}_{ab}(v_{\alpha}q_{c}^{i-1})\ket{0}
    \eeq
\end{theorem}

\paragraph{Description in screening currents}
We can also represent the partition function of the tetrahedron instanton system using the screening currents. Using the character form in \eqref{eq:7dscreeningdecomp}, we omit the singular terms that might occur and extract the square root part of the total index by specifying an order in the elements as
\beq
    \bfV=\sum_{a,b\in\four}\frac{\bfP_{\four}}{\bfP^{\vee}_{\text{k}(a)}\bfP_{\text{k}(b)}}\bfX_{\bar{a}}^{\vee}\bfX_{\bar{b}}\longrightarrow \mathbf{v}=\sum_{\substack{(x,a)<(x',b)\\x\in\mathcal{X}_{\bar{a}},x'\in\mathcal{X}_{\bar{b}}}}\frac{\bfP_{\four}}{\bfP_{\text{k}(a)}^{\vee}\bfP_{\text{k}(b)}}\left(\frac{x'}{x}\right).
\eeq
Using the contraction formulas of the screening currents in \eqref{eq:D2op}, the gauge origami partition function of the tetrahedron instantons is rewritten using the screening currents.
\begin{theorem}
    We have the following identity up to one-loop perturbative factors:
    \beq\label{eq:D6screening}
        \mathbb{I}[\mathbf{v}]=\bra{0}\prod_{\substack{a\in\four\\x\in\mathcal{X}_{\bar{a}}}}\mathsf{S}_{\text{k}(a)}(x)\ket{0}\simeq \mathcal{Z}^{\D6}_{\text{1-loop}}\mathcal{Z}^{\D6}_{\text{tet.inst.}}[\underline{\vec{v}},\underline{\vec{\pi}}].
    \eeq
    Note that we are implicitly defining an order in the products of the screening currents as
    \beq
        \cdots \mathsf{S}_{\text{k}(b)}(x')\cdots \mathsf{S}_{\text{k}(a)}(x)\cdots \quad \Longleftrightarrow \quad (x,a)<(x',b).
    \eeq
\end{theorem}

\subsection{General D6 \texorpdfstring{$qq$}{qq}-characters}\label{sec:generalD6qq}
We can consider higher rank analogs of them similar to the $\D4$ $qq$-character
\beq
    \mathsf{T}^{(n)}_{\bar{a}}(\underline{x})={:\mathsf{W}_{\bar{a}}(x_{1})\cdots \mathsf{W}_{\bar{a}}(x_{n}):}+\cdots.
\eeq
This time the coefficients appearing in the expansion of the right-hand side correspond to the partition function of the $\U(n)$ gauge theory on the D6-brane on $\mathbb{C}^{3}_{\bar{a}}\times \mathbb{S}^{1}$. 

Analogous to the $\D4$ case, one would like to add negative weights as the highest weights. General $\D6$ $qq$-characters including negative weights are written as
\beq
    \mathsf{T}^{(n|m)}_{\bar{a}}(\underline{x}\,|\,\underline{y})={:\frac{\mathsf{W}_{\bar{a}}(x_{1})\cdots \mathsf{W}_{\bar{4}}(x_{n})}{\mathsf{W}_{\bar{a}}(y_{1})\cdots \mathsf{W}_{\bar{a}}(y_{m})}:}+\cdots.\label{eq:D6supergroupqqcharacter}
\eeq
The explicit coefficients are then obtained by imposing the commutativity with the screening charge $\mathscr{Q}_{a}(x')$. 

Similar to the $\D4$ case, one might think we need to introduce the negative iWeyl reflection introduced to reproduce the supergroup analog of $\D4$ $qq$-characters. However, noticing that 
\beq
    \mathsf{W}_{\bar{a}}(x)^{-1}\mathsf{S}_{a}(x')={x'}^{-1}(1-q_{a}^{-1}x'/x):\mathsf{W}_{\bar{a}}(x)^{-1}\mathsf{S}_{a}(x'):
\eeq
gives no new poles, $\mathsf{W}_{a}(x)^{-1}$ and $\mathsf{S}_{a}(x')$ commute with each other. Thus, compared to the $\D4$ case where we need to introduce negative iWeyl reflections to cancel the poles arising from the operators in the denominators, for the $\D6$ case, we do not need to do such kind of procedure. The right-hand side of \eqref{eq:D6supergroupqqcharacter} is expanded by the plane partitions generated by $\mathsf{W}_{\bar{a}}(x_{i})\,(i=1,\ldots,n)$ and no plane partitions generated by $\mathsf{W}_{\bar{a}}(y_{j})^{-1}\,(j=1,\ldots,m)$. We expect the coefficients appearing on the right-hand side give the instanton partition function of the 7d $\U(n|m)$ theory on $\mathbb{C}^{3}_{\bar{a}}\times \mathbb{S}^{1}$. We leave a detailed analysis of this for future work.

For later use, let us consider the case when there is only one positive and one negative weight:
\beq
    \mathsf{T}_{\bar{a}}^{(1|1)}(x\,|\,Kx)={:\frac{\mathsf{W}_{\bar{a}}(x)}{\mathsf{W}_{\bar{a}}(Kx)}:}+\cdots
\eeq
where $K$ is a generic parameter. The contraction $\mathsf{W}_{\bar{4}}(Kx)^{-1}$ with the screening current $\mathsf{S}_{4}(x')$ will give a pole free rational function but when $K$ is generic, no poles will be canceled and the coefficients are only modified:
\begin{align}
\begin{split}
    &\mathsf{S}_{a}(x')\Lambda_{\bar{a},\pi}^{K}(x)=K^{-1}\left[\left(1-\frac{Kq_{a}x}{x'}\right)\mathscr{W}_{\pi,x}^{\bar{a}}(q_{a}^{-1}x')^{-1}\right]_{+}:\mathsf{S}_{a}(x')\Lambda_{\bar{a},\pi}^{K}(x):,\\
    &\mathscr{W}_{\pi,x}^{\bar{a}}(x')\rightarrow \frac{1}{(1-Kx/x')}\mathscr{W}_{\pi,x}^{\bar{a}}(x')\coloneqq \mathscr{W}_{\pi,x}^{\bar{a},K}(x'),\quad \widetilde{Z}^{\D6}_{\bar{a}}[\pi]\rightarrow \prod_{\scube\in\pi}\frac{1-Kx/\chi_{\bar{a},x}(\cube)}{1-Kq_{a}x/\chi_{\bar{a},x}(\cube)}\widetilde{Z}^{\D6}_{\bar{a}}[\pi],
\end{split}
\end{align}
where
\bea
    &\Lambda_{\bar{a},\pi}^{K}(x)={:\frac{\mathsf{W}_{\bar{a}}(x)}{\mathsf{W}_{\bar{a}}(Kx)}\prod_{\scube\in\pi}\mathsf{A}^{-1}(\chi_{\bar{a},x}(\cube)):}.
\eea
Thus, we will obtain
\bea
    \mathsf{T}_{\bar{4}}^{(1|1)}(x\,|\,Kx)&=\sum_{\pi\in\mathcal{PP}}\widetilde{\mathcal{Z}}_{\bar{a}}^{\D6}[K,\pi]\Lambda_{\bar{a},\pi}^{K}(x),\\
    \widetilde{\mathcal{Z}}_{\bar{a}}^{\D6}[K,\pi]&=\prod_{\scube\in\pi}\frac{\left(1-Kx/\chi_{\bar{a},x}(\cube)\right)\left(1-q_{a}x/\chi_{\bar{a},x}(\cube)\right)}{\left(1-Kq_{a}x/\chi_{\bar{a},x}(\cube)\right)\left(1-x/\chi_{\bar{a},x}(\cube)\right)}\prod_{\substack{\scube\in\pi\\\scubeF\in\pi}}g_{\bar{a}}\left(\frac{\chi_{\bar{a},x}(\cube)}{\chi_{\bar{a},x},(\cubeF)}\right)^{-1}.
\eea
Note that after rescaling $\mathsf{A}(x)\rightarrow \mathfrak{q}^{-1}\mathsf{A}(x)$ we can change the topological term to $\mathfrak{q}^{|\pi|}$. For later use, let us list some properties of these operators. The operator products of $\{\Lambda_{\bar{a},\pi}^{K}(x)\}$ are 
\begin{align}\label{eq:D6supergroupcontraction1}
\begin{split}
        \Lambda_{\bar{b},\pi^{(2)}}^{K_{2}}(x_{2})\Lambda^{K_{1}}_{\bar{a},\pi^{(1)}}(x_{1})&=\mathcal{Z}^{\D6\tbar\D6}_{\text{1-loop}}(x_{1},\bar{a},K_{1}\,|\,x_{2},\bar{b},K_{2})\mathcal{Z}^{\D6\tbar\D6}_{\bar{a};K_{1}|\bar{b};K_{2}}(x_{1},\pi^{(1)}\,|\,x_{2},\pi^{(2)})\\
        &\qquad\times :\Lambda_{\bar{b},\pi^{(2)}}^{K_{2}}(x_{2})\Lambda^{K_{1}}_{\bar{a},\pi^{(1)}}(x_{1}):
\end{split}
\end{align}
where 
\begin{subequations}\label{eq:D6supergroupcontraction2}
\begin{align}
    \mathcal{Z}^{\D6\tbar\D6}_{\text{1-loop}}(x_{1},\bar{a},K_{1}\,|\,x_{2},\bar{b},K_{2})&=\frac{\mathcal{Z}_{\text{1-loop}}^{\D6\tbar\D6}(x_{1},\bar{a}\,|\,x_{2},\bar{b})\mathcal{Z}_{\text{1-loop}}^{\D6\tbar\D6}(K_{1}x_{1},\bar{a}\,|\,K_{2}x_{2},\bar{b})}{\mathcal{Z}_{\text{1-loop}}^{\D6\tbar\D6}(K_{1}x_{1},\bar{a}\,|\,x_{2},\bar{b})\mathcal{Z}_{\text{1-loop}}^{\D6\tbar\D6}(x_{1},\bar{a}\,|\,K_{2}x_{2},\bar{b})},\\
    \mathcal{Z}^{\D6\tbar\D6}_{\bar{a};K_{1}|\bar{b};K_{2}}(x_{1},\pi^{(1)}\,|\,x_{2},\pi^{(2)})&=\mathcal{Z}_{\bar{a}|\bar{b}}^{\D6\tbar\D6}(x_{1},\pi^{(1)}\,|\,x_{2},\pi^{(2)})\nonumber\\
    &\qquad \times\prod_{\scube\in\pi^{(1)}}\left(q_{b}^{-1}\mathscr{V}_{b}\left(\frac{\chi_{\bar{a},x_{1}}(\cube)}{q_{b}K_{2}x_{2}}\right)\right)\prod_{\scubeF\in\pi^{(2)}}\mathscr{V}_{a}\left(\frac{K_{1}x_{1}}{\chi_{\bar{b},x_{2}}(\cubeF)}\right)^{-1}.
\end{align}
\end{subequations}

\paragraph{Pit reduction of $\D6$ $qq$-characters}
When the parameter $K$ is generic, the coefficients are modified slightly without changing the structure of the $qq$-character. However, when we tune $K$ to specific values, the zeros appearing will cancel the poles, the iWeyl reflection will be restricted, and the right-hand side will not be expanded with arbitrary plane partitions but only by specified plane partitions. This can be understood also from the extra factors in the coefficients. When $K$ is tuned, the coefficient $\mathcal{Z}_{\bar{a}}^{\D6}[K,\pi]$ will be zero for some plane partition configurations and such terms will disappear from the $qq$-character. This procedure is well known in the literature of MacMahon representations and is called the \emph{pit reduction} \cite{Feigin2011plane,bershtein2018plane}. A plane partition with a pit~P is a plane partition that does not contain a box at the position~P.

Let us focus on the case when $a=4$. For example, when we tune the parameter $K$ as $K=q_{3}x$, we have 
\begin{align}
    {:\frac{\mathsf{W}_{\bar{4}}(x)}{\mathsf{W}_{\bar{4}}(q_{3}x)}:}=\mathsf{X}_{12}(x)
\end{align}
and the $\D6$ $qq$-character will reduce to the $\D4$ $qq$-character. This process is just placing a pit in $q_{3}x$ and reducing the plane partition in $(123)$ and restricting it to a Young diagram in the $(12)$-plane:
\begin{align}
    \adjustbox{valign=c}{\includegraphics[width=5cm]{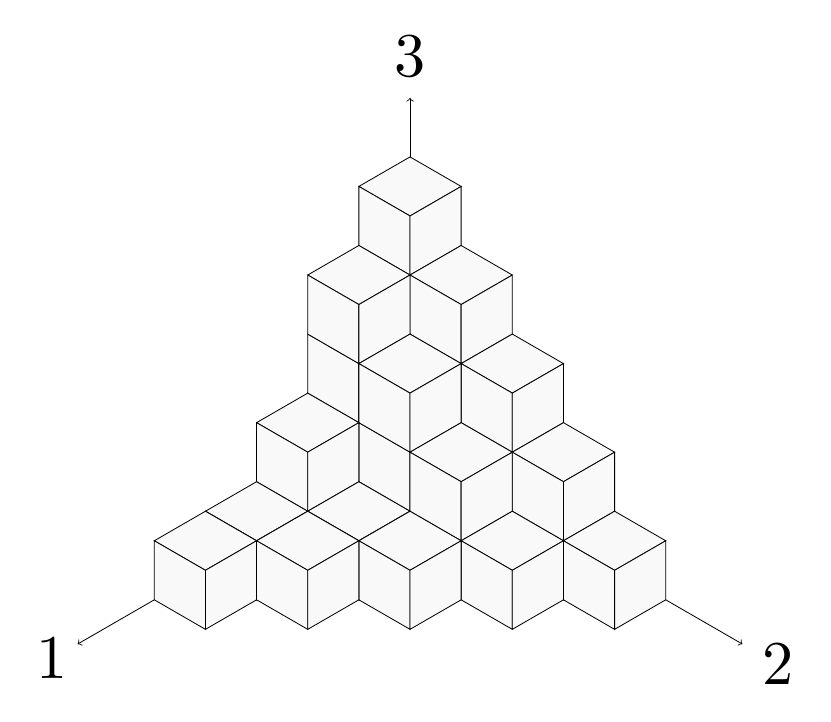}} \quad \Longrightarrow \quad \adjustbox{valign=c}{\includegraphics[width=5cm]{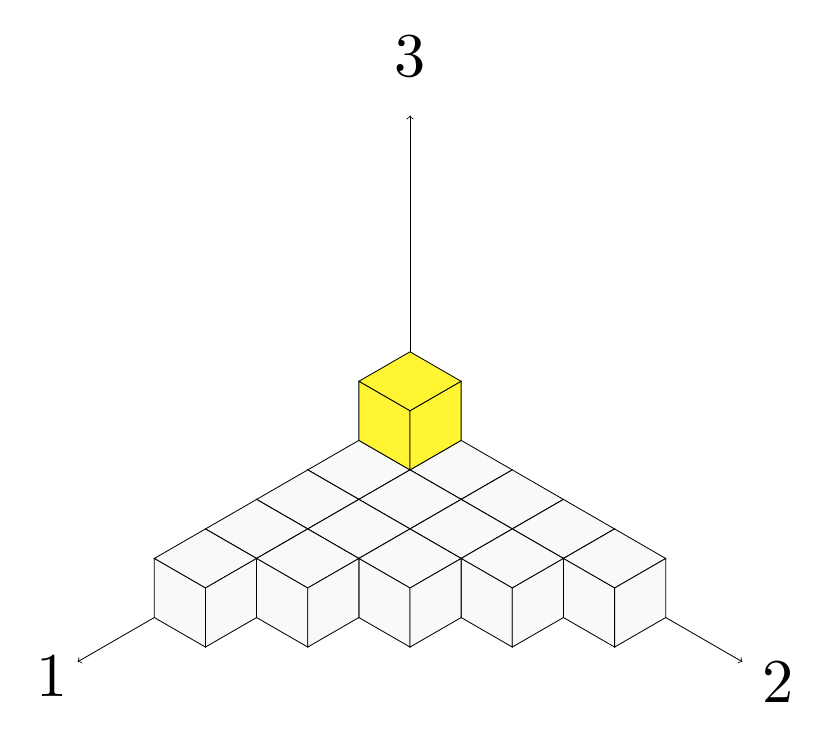}} 
\end{align}
Physically, this is interpreted as the Higgs mechanism and also the tachyon condensation as discussed in section~\ref{sec:tetraLMNS}.

When $K=\chi_{\bar{4},x}(\pitcube)/x=q_{1}^{L-1}q_{2}^{M-1}q_{3}^{N-1}$, we will get a pit reduction of the plane partition:
\begin{align}
\begin{split}
    \mathsf{T}_{\bar{4}}^{\pitcube}(x)&={:\frac{\mathsf{W}_{\bar{4}}(x)}{\mathsf{W}_{\bar{4}}(\chi_{\bar{4},x}(\hspace{-0.5mm}\scalebox{1.2}{\pitcube}))}:}+\cdots=\sum_{\substack{\pi:\text{plane partitions}\\\text{with a pit at $\pitcube$} }}\cdots
\end{split}
\end{align}
The highest weight has the following structure
\begin{align}
\begin{split}
    {:\frac{\mathsf{W}_{\bar{4}}(x)}{\mathsf{W}_{\bar{4}}(q_{1}^{L-1}q_{2}^{M-1}q_{3}^{N-1}x)}:}&={:\prod_{i=1}^{L}\frac{\mathsf{W}_{\bar{4}}(q_{1}^{i-1}q_{2}^{M-1}q_{3}^{N-1}x)}{\mathsf{W}_{\bar{4}}(q_{1}^{i}q_{2}^{M-1}q_{3}^{N-1}x)}\prod_{j=1}^{M}\frac{\mathsf{W}_{\bar{4}}(q_{2}^{j-1}q_{3}^{N-1}x)}{\mathsf{W}_{\bar{4}}(q_{2}^{j}q_{3}^{N-1}x)}\prod_{k=1}^{N}\frac{\mathsf{W}_{\bar{4}}(xq_{3}^{k-1})}{\mathsf{W}_{\bar{4}}(xq_{3}^{k})}:}\\
    &={:\prod_{i=1}^{L}\mathsf{X}_{23}(q_{1}^{i-1}q_{2}^{M-1}q_{3}^{N-1}x)\prod_{j=1}^{M}\mathsf{X}_{13}(q_{2}^{j-1}q_{3}^{N-1}x)\prod_{k=1}^{N}\mathsf{X}_{12}(xq_{3}^{k-1}):}
\end{split}
\end{align}
which implies the pit reduction of the plane partition is related to the general $\D4$ $qq$-characters $\mathsf{T}^{(\vec{n}|\vec{0})}_{12:23:13}(\underline{\vec{x}})$ in Thm.~\ref{thm:D4generalqqcharacter} after specializing the parameters as 
\begin{equation}
    n_{23}=L,\,\,n_{13}=M,\,\,n_{12}=N,\,\,\vec{x}_{12}=(xq_{3}^{k-1})_{k=1}^{N},\,\,\vec{x}_{13}=(q_{2}^{j-1}q_{3}^{N-1}x)_{j=1}^{M},\,\, \vec{x}_{12}=(q_{1}^{i-1}q_{2}^{M-1}q_{3}^{N-1}x)_{i=1}^{L}.
\end{equation}
This fact is nothing special from the plane partition viewpoint. This is because we can pile $L$ Young diagrams spanning the $(23)$-plane, $M$ Young diagrams spanning the $(13)$-plane, and $N$ Young diagrams spanning the $(12)$-plane on top of each other to obtain all possible plane partitions with a pit at $(L,M,N)$. Note also that this decomposition in Young diagrams is not unique and so we have multiple descriptions in the $\D4$ $qq$-characters. Physically, the system corresponding to the $\D6$ $qq$-character with a pit-reduced plane partition is just the gauge origami system with folded instantons where the Coulomb branch parameters are tuned in a specific way. 

Generally, we may add another pit to the plane partition and this is called the double-constrained plane partition which was introduced in \cite{Harada:2018bkb} (see also the references there) to discuss minimal models of W-algebras. Let $(L_{1},M_{1},N_{1})$ and $(L_{2},M_{2},N_{2})$ be the coordinates of the two pits. The parameter $K$ needs to obey the conditions of the two pits
\begin{align}
    K=q_{1}^{L_{1}-1}q_{2}^{M_{1}-1}q_{3}^{N_{1}-1}=q_{1}^{L_{2}-1}q_{2}^{M_{2}-1}q_{3}^{N_{2}-1}.
\end{align}
Imposing this condition causes the $q$-parameters to be not generic anymore. The physical meaning of these types of $qq$-characters and their relation with minimal models are still unclear for the moment. We note that the condition above is just the Burge condition~\cite{Belavin:2015ria,Alkalaev:2014sma}, and thus the BPS/CFT correspondence arising should be an analog of the AGT dual of minimal models of W-algebras. See also \cite{Kimura:2022spi} where some examples of these truncations were studied from the $qq$-character viewpoint.

\paragraph{Plane partitions with boundary conditions}
The $qq$-characters are uniquely determined from the highest weight after imposing the commutativity with the screening charges. 
We further can construct a $qq$-character where each term corresponds to plane partitions with asymptotic Young diagrams $\lambda,\mu,\nu$ in the three axes:
\begin{align}
    \adjustbox{valign=c}{\includegraphics[width=7cm]{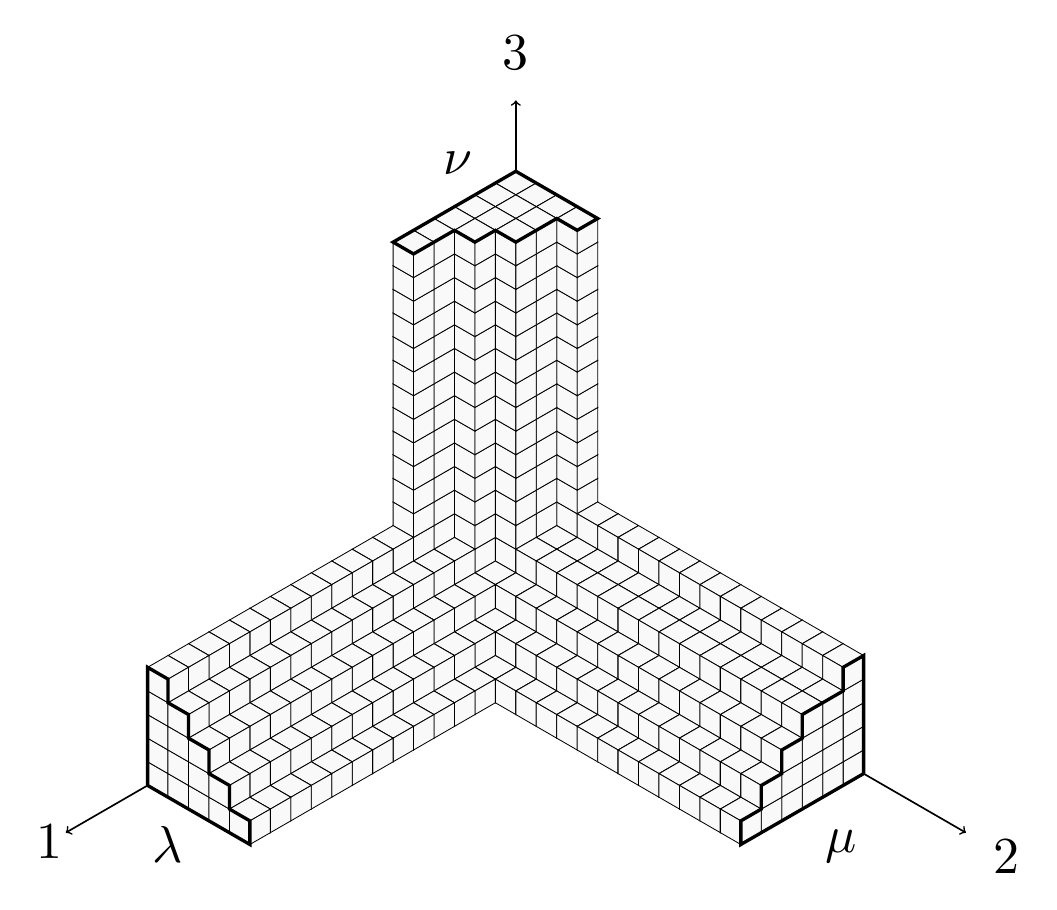}}\quad\adjustbox{valign=c}{\includegraphics[width=7cm]{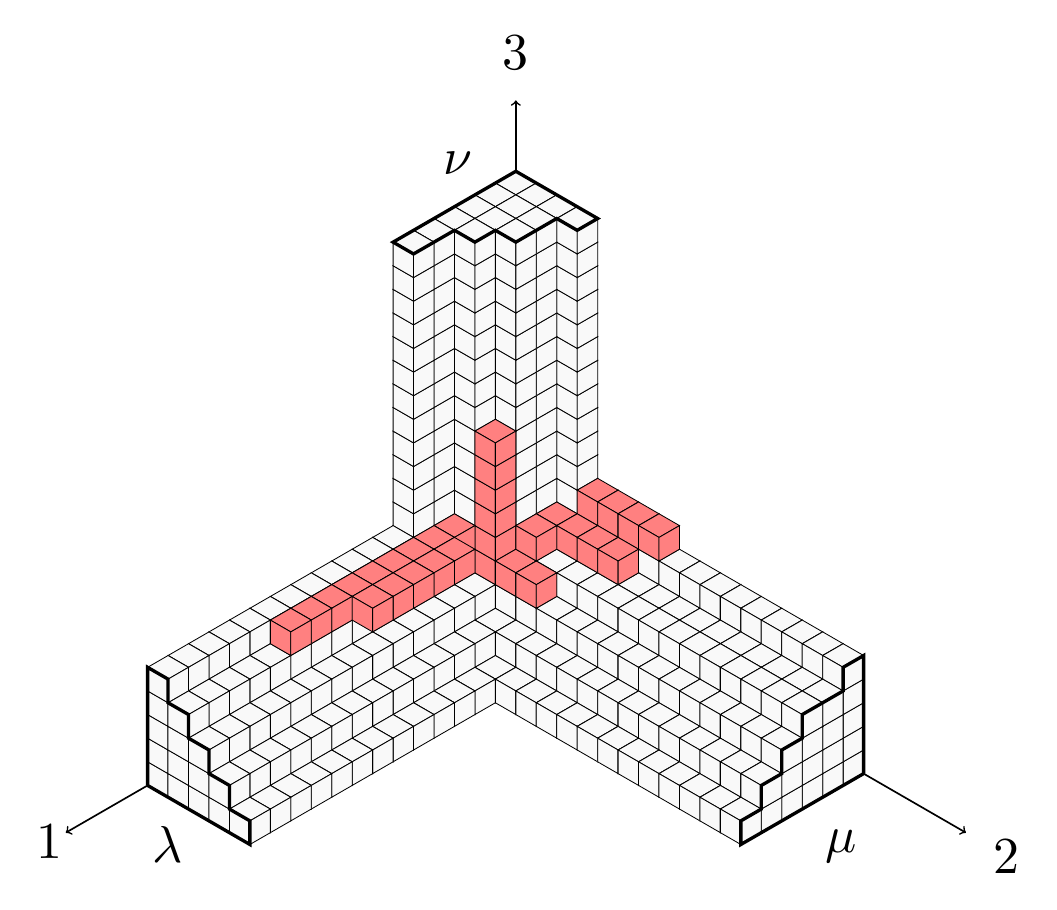}} 
\end{align}
The left figure is the vacuum configuration of the plane partition with boundary conditions and the right figure is the configuration with boxes added to the vacuum configuration. The highest weight reproducing this configuration is 
\beq
    :\mathsf{W}_{\bar{a}}(x)\prod_{\scube\in\mathcal{P}_{\lambda,\mu,\nu}}\mathsf{A}^{-1}(\chi_{\bar{a},x}(\cube)):,\label{eq:D6hw_boundarypartition}
\eeq
where $\mathcal{P}_{\lambda,\mu,\nu}$ is the set of boxes in the vacuum configuration. We denote the $qq$-character obtained from this highest weight as
\begin{equation}
    \mathsf{T}_{\bar{a},\lambda\mu\nu}(x)={:\mathsf{W}_{\bar{a}}(x)\prod_{\scube\in\mathcal{P}_{\lambda,\mu,\nu}}\mathsf{A}^{-1}(\chi_{\bar{a},x}(\cube)):}+\cdots.
\end{equation}
As an example, let us focus on $a=4$ and consider the case where the boundary conditions are $(\lambda,\mu,\nu)=(\Bbox,\emptyset,\emptyset)$. The highest weight is 
\beq
    {:\mathsf{W}_{\bar{4}}(x)\prod_{i=1}^{\infty}\mathsf{A}^{-1}(xq_{1}^{i-1}):}={:\mathsf{W}_{\bar{4}}(x)\prod_{i=1}^{\infty}\frac{\mathsf{S}_{1}(q_{1}^{i}x)}{\mathsf{S}_{1}(q_{1}^{i-1}x)}:}={:\frac{\mathsf{W}_{\bar{4}}(x)}{\mathsf{S}_{1}(x)}:}.
\eeq
The contraction is 
\beq
    :\frac{\mathsf{W}_{\bar{4}}(x)}{\mathsf{S}_{1}(x)}:\mathsf{S}_{4}(q_{4}x')=x'\frac{(1-q_{1}q_{4}x'/x)}{(1-q_{3}^{-1}x'/x)(1-q_{2}^{-1}x'/x)}:\frac{\mathsf{W}_{\bar{4}}(x)}{\mathsf{S}_{1}(x)}\mathsf{S}_{4}(q_{4}x'):
\eeq
which gives poles at $x'=q_{2}x,q_{3}x$. One can show that the terms that cancel the poles coming from these terms are 
\begin{equation}
:\mathsf{W}_{\bar{4}}(x)\mathsf{A}^{-1}(q_{2}x)\mathsf{S}_{1}(x)^{-1}:, \quad :\mathsf{W}_{\bar{4}}(x)\mathsf{A}^{-1}(q_{3}x)\mathsf{S}_{1}(x)^{-1}:
\end{equation}
which correspond to the following plane partitions:
\begin{equation}
   \hspace{-3cm} \adjustbox{valign=c}{\includegraphics[width=12cm]{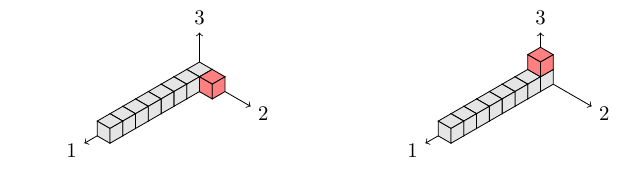}}
\end{equation}
By doing this procedure recursively, one will see that each term of the expansion of the $qq$-character indeed corresponds with the elements in $\mathcal{P}_{\Abox,\emptyset,\emptyset}$.

\paragraph{Web of $qq$-characters}
We may represent the $qq$-character $\mathsf{T}_{\bar{a},\lambda\mu\nu}$ using  a trivalent vertex as
\begin{align}\label{eq:trivalentC3}
    \adjustbox{valign=c}{\begin{tikzpicture}[thick]
		\begin{scope}[scale=1.5,xscale = 1]
		\draw[] (0,0) -- (-1,0) -- (-1.7,-0.7);
		\draw[] (-1,1) -- (-1,0);
  \node at (-1.9,-0.9){$\lambda$};
  \node at (0.2,0){$\mu$};
  \node at (-1,1.2){$\nu$};
		\end{scope}
    \end{tikzpicture}}
\end{align}
where the three legs correspond to the three axes of the plane partitions with boundary conditions $\lambda,\mu,\nu$. This reminds us of the topological vertex \cite{Aganagic:2003db,Okounkov:2003sp,Iqbal:2007ii,Awata:2005fa,Awata:2008ed,Taki:2007dh}. In the context of the topological vertex and the plane partition, one can glue the legs with the same boundary conditions. We expect that we can do the same gluing procedure for the $\D6$ $qq$-characters and obtain $qq$-characters associated with complicated webs of the trivalent vertices such as 
\begin{align}\label{eq:generalbraneweb}
    \adjustbox{valign=c}{\begin{tikzpicture}[thick]
    \begin{scope}[scale=0.8,xscale = 1]
	\draw (0,0)--++(0.8,0);
        \draw (0,0)--++(0,0.8);
        \draw (0,0)--++(-0.6,-0.6)--++(1,0);
        \draw (-0.6,-0.6)--++(-0.8,-0.6)--++(1,0);
        \draw(-1.4,-1.2)--++(-1,-0.6)--++(1,0);
        \draw (-2.4,-1.8)--++(-1.2,-0.6);
    \end{scope}
    \end{tikzpicture}},\hspace{2cm} \adjustbox{valign=c}{\begin{tikzpicture}[thick]
    \begin{scope}[scale=1,xscale = 1]
	\draw (0,0)--++(0.8,0);
        \draw (0,0)--++(0,0.8);
        \draw (0,0)--++(-0.6,-0.6)--++(-0.8,0);
        \draw (-0.6,-0.6)--++(0,-0.8);
    \end{scope}
    \end{tikzpicture}}
\end{align}
Note that the trivalent vertex in \eqref{eq:trivalentC3} and the left and right diagrams in \eqref{eq:generalbraneweb} are the brane webs dual to the toric diagrams of the $\mathbb{C}^{3}$, $\mathbb{C}^{2}/\mathbb{Z}_{4}\times \mathbb{C}$ and the resolved conifold, respectively. Since the $qq$-character associated with the trivalent vertex captures the partition function of the D-branes on the $\mathbb{C}^{3}$ subspace of $\mathbb{C}^{4}$, we expect that the glued $qq$-characters will capture the partition functions of D-branes on general toric CY$_{3}$ manifolds included in the $\text{CY$_{3}$}\times \mathbb{C}$ geometry. Such kind of generalizations are left for future work. We may also add boundary Young diagrams to each of the edges generally. They are expected to be $qq$-characters obtained after gluing $\mathsf{T}_{\bar{a},\lambda\mu\nu}(x)$.

We can combine the two kinds of $qq$-characters $\mathsf{T}_{\bar{a}}^{\pitcube}(x), \mathsf{T}_{\bar{a},\lambda\mu\nu}(x)$ and construct $qq$-characters where plane partitions with boundary conditions and a pit occur. The $qq$-character is uniquely determined from the highest weight as
\beq
    \mathsf{T}_{abc,\lambda\mu\nu}^{L,M,N}(x)={:\frac{\mathsf{W}_{abc}(x)}{\mathsf{W}_{abc}(q_{a}^{L-1}q_{b}^{M-1}q_{c}^{N-1}x)}\prod_{\scube\in\mathcal{P}_{\lambda,\mu,\nu}}\mathsf{A}^{-1}(\chi_{abc,x}(\cube)):}+\cdots.
\eeq
We represent this $qq$-character using a trivalent vertex as 
\begin{align}\label{eq:branewebalgebra1}
    \adjustbox{valign=c}{\begin{tikzpicture}[thick]
		\begin{scope}[scale=1.5,xscale = 1]
		  \node[scale=1] at (-0.5,0.5){$L$};
            \node[scale=1] at (-0.5,-0.4){$N$};
            \node[scale=1] at (-1.5,0.2){$M$};
		\draw[] (0,0) -- (-1,0) -- (-1.7,-0.7);
		\draw[] (-1,1) -- (-1,0);
  \node at (-1.9,-0.9){$\lambda$};
  \node at (0.2,0){$\mu$};
  \node at (-1,1.2){$\nu$};
		\end{scope}
    \end{tikzpicture}}
\end{align}
and gluing of these $qq$-characters should lead to a larger class of $qq$-characters. We call this large class of $qq$-characters, the \textbf{web of $qq$-characters}\footnote{There is a similar concept called \emph{web of W-algebras} \cite{Prochazka:2017qum} where the authors associated brane webs including integers inserted in the faces surrounded by the legs with W-algebras \eqref{eq:branewebalgebra1}. Each trivalent vertex corresponds to the corner VOA \cite{Gaiotto:2017euk} and the glued algebra gives the web of W-algebras. See also \cite{Harada:2020woh} where examples for the gluing process were explicitly done in the trigonometric language.}. For the moment, we do not know how to glue these $qq$-characters explicitly and we leave it for future work.

\paragraph{BPS $qq$-characters}
Another way to construct $qq$-characters associated with the CY$_{4}$ with the form $\CY 3\times \mathbb{C}$ is to use the vertex operators introduced in \eqref{eq:CY3D0D6vertexoperator} (see the notations there). We can define the screening charge corresponding to the $\mathbb{C}$-part as
\bea
\mathscr{Q}_{i}(x)=\sum_{k\in\mathbb{Z}}\mathsf{S}_{i}(q_{4}^{k}x),\quad i\in Q_{0}.
\eea
We then can introduce D6 $qq$-characters associated with the CY$_{3}$ part as
\bea
\mathsf{T}_{i}(x)=\mathsf{W}_{i}(x)+\cdots
\eea
satisfying
\bea
\relax [\mathsf{T}_{i}(x),\mathsf{Q}_{j}(x')]=0,\quad \forall i,j\in Q_{0}.
\eea
An interesting property is that the monomial terms of this $qq$-character are classified by the 3d BPS crystals \cite{Ooguri:2009ijd,Yamazaki:2008bt}. We present this property as a conjecture which will be clarified in a future publication. Details will be explained in \cite{Kimura-Noshita}.

\begin{conjecture}[\cite{Kimura-Noshita}]\label{conj:BPSqq}
Let $Z$ be a toric CY$_{4}$ which takes the form as $Z=X\times \mathbb{C}$ where $X$ is a toric CY$_{3}$. The corresponding quiver of $X$ is $Q=(Q_{0},Q_{1})$ with $q$-deformation parameters $\{q_{I}\}_{I\in Q_{1}}$. Given this quiver, we can construct a three-dimensional crystal called BPS crystals (see \cite{Ooguri:2009ijd,Yamazaki:2008bt} for details). The BPS crystals are sets of colored \emph{atoms} where the colors are labeled by $Q_{0}$. Namely, given a 3d BPS crystal $\Lambda$, we have a color projection map $\text{c}:\Lambda\rightarrow Q_{0}$. The BPS crystal has an atom in the origin which we denote $\mathfrak{o}$. Each atom $\cube$ of the BPS crystal is associated with a coordinate function\footnote{\label{foot:coordinatefunct}This coordinate function is a one-parameter deformation of the coordinate function of \cite{Galakhov:2021vbo,Noshita:2021ldl}. We will discuss the derivation of this in \cite{Kimura-Noshita}. For this paper, we only note that after taking the limit $q_{4}\rightarrow 1$, this coordinate function will become the coordinate functions defined in \cite{Galakhov:2021vbo,Noshita:2021ldl}.} $\chi_{X,x}(\cube)$:
\bea
\chi_{X,x}(\cube )=x\times \prod_{I\in \text{path}[\mathfrak{o}\rightarrow \cube]} q_{I}.
\eea
Given two atoms with color $i$ and $j$, the difference of the coordinates comes from the parameter $q_{I:i\rightarrow j}$ which is associated with the arrow connecting the two atoms. 

Under this condition, the $qq$-characters are given
    \bea
\mathsf{T}_{i}(x)&=\sum_{\Lambda^{(i)}}\mathfrak{q}_{i}^{|\Lambda^{(i)}|}\mathcal{Z}^{\D6}_{i}[\Lambda^{(i)}]:\mathsf{W}_{i}(x)\prod_{\scube\in \Lambda^{(i)}}\mathsf{A}_{\text{c}(\scube)}^{-1}(\chi_{X,x}(\cube)):,\\
\mathsf{T}^{K}_{i}(x)&=\sum_{\Lambda^{(i)}}\mathfrak{q}_{i}^{|\Lambda^{(i)}|}\mathcal{Z}^{\D6}_{i}[K,\Lambda^{(i)}]:\frac{\mathsf{W}_{i}(x)}{\mathsf{W}_{i}(Kx)}\prod_{\scube\in \Lambda^{(i)}}\mathsf{A}_{\text{c}(\scube)}^{-1}(\chi_{X,x}(\cube)):.
 \eea
where $\Lambda^{(i)}$ is the crystal whose atom at the origin is with color $i\in Q_{0}$. The coefficients $\mathcal{Z}_{i}^{\D6}[\Lambda^{(i)}],\mathcal{Z}_{i}^{\D6}[K,\Lambda^{(i)}]$ are determined from the commutativity with the screening charges:
\bea
\relax[\mathsf{T}_{i}(x),\mathscr{Q}_{j}(x')]=0,\quad \relax[\mathsf{T}^{K}_{i}(x),\mathscr{Q}_{j}(x')]=0
\eea
\end{conjecture}
Similar to the $\mathbb{C}^{4}$-case, we can introduce pit reductions by adding extra $\mathsf{W}_{i}(x)$ at the denominators or tuning the parameter $K$. The monomial terms will then be classified by the subcrystals of the 3d BPS crystals. Such types of subcrystal were studied in \cite{Galakhov:2021xum,Noshita:2021dgj}. Since the $qq$-characters obtained in this way have a strong relationship with BPS crystals, we call them \textbf{BPS $qq$-characters}. Later in section~\ref{sec:BPSqqQTA}, we will see that such BPS $qq$-characters have a relation with the quiver quantum toroidal algebras/toroidal quiver BPS algebras \cite{Galakhov:2021vbo,Noshita:2021ldl}.


\section{Towards D8-brane \texorpdfstring{$qq$}{qq}-characters}\label{sec:D8qq}
We use the fusion process of the D6 $qq$-characters to define the D8 $qq$-characters in section~\ref{sec:fusionD6toD8}. We then study the contractions of the D8 $qq$-characters and show that they reproduce the instanton partition function of the magnificent four system up to sign factors in section~\ref{sec:M4rankN}.

\subsection{Fusion of D6 \texorpdfstring{$qq$}{qq}-characters to D8 \texorpdfstring{$qq$}{qq}-characters}\label{sec:fusionD6toD8}
After constructing $\D2,\D4,\D6$ $qq$-characters, we would like to obtain $\D8$ $qq$-characters that reproduce the magnificent four partition function in \eqref{eq:mag4Nekrasovfact}. In the lower-dimensional cases, thanks to Thm.~\ref{thm:tetrascreening}, we can define the $qq$-characters by choosing the highest weight and imposing the commutativity with the screening charges. The highest weight and the corresponding screening charges were chosen so that the associated subspaces do not intersect in the $\mathbb{C}^{4}$ geometry. However, for the $\D8$ operators $\mathsf{Z}(x),\widetilde{\mathsf{Z}}^{K}(x)$, the only operator that makes the operator products become rational functions is the root current $\mathsf{A}(x)$. In this sense, it is natural to construct a screening charge related to the D0 operator. For the moment, we do not know how to construct such kind of screening currents. Instead, we will use the fusion process discussed in Thm.~\ref{thm:D2toD4fusion}, \ref{thm:D4toD6fusion} to define the D8 $qq$-characters. Since we are interested in studying the relation between the magnificent four system, where D8 and anti D8-branes appear, we use $\widetilde{\mathsf{Z}}^{K}(x)$ as the highest weight.
\begin{definition}
We define the $\D8$ $qq$-characters as
\begin{equation}
    \mathsf{T}_{\four;a}^{K}(x)=\sum_{\rho\in\mathcal{SP}}\mathcal{Z}_{\four;a}^{\D8}[\rho,K]\Lambda^{K}_{\four,\rho}(x),\quad a\in\four,\label{eq:D8qqdef1}
\end{equation}
where 
\begin{equation}
    \Lambda^{K}_{\four,\rho}(x)={:\Tilde{\mathsf{Z}}^{K}(x)\prod_{\shcube\in\rho}\mathsf{A}^{-1}(\chi_{\four,x}(\hcube)):}
\end{equation}
Rescaling the root current as $\mathsf{A}(x)\rightarrow \mathfrak{q}^{-1}\mathsf{A}(x)$, we can modify the topological terms as
\begin{equation}
    \mathsf{T}_{\four;a}^{K}(x)=\sum_{\rho\in\mathcal{SP}}\mathfrak{q}^{|\rho|}\mathcal{Z}_{\four;a}^{\D8}[\rho,K]:\Tilde{\mathsf{Z}}^{K}(x)\prod_{\shcube\in\rho}\mathsf{A}^{-1}(\chi_{\four,x}(\hcube)):,\quad a\in\four.\label{eq:D8qqdef2}
\end{equation}
\end{definition}
Let us show that this is obtained by taking the infinite products of the $\D6$ $qq$-characters. Using Lem.~\ref{lem:D6operatorproduct}, one can obtain the following property which is a $\D8$ analog of Lem.~\ref{lem:D6planecond}.
\begin{lemma}
    For $a\in\four$, we have 
    \begin{equation}
        \Lambda_{\bar{a},\pi^{(2)}}(x_{2})\Lambda_{\bar{a},\pi^{(1)}}(x_{1})=0
    \end{equation}
    when $x_{2}=q_{a}x_{1}$ and $\pi^{(2)}\succ \pi^{(1)}$.
\end{lemma}
\begin{proof}
    Let us focus on the case $a=4$. It is enough to study the zeros of $\mathcal{Z}_{\bar{4}|\bar{4}}^{\D6\tbar\D6}(x_{1},\pi^{(1)}\,|\,x_{2},\pi^{(2)})$.
Using the formulas in \eqref{eq:D6tetinst_partfunct}, \eqref{eq:D6Nekrasov-shell}, we have\footnote{Strictly speaking, we need the information of the denominators of $\mathscr{W}^{\bar{{a}}}_{\pi,v}(x)$ in \eqref{eq:D6Nekrasov-shell}, which is omitted here.}
\bea
&\mathcal{Z}_{\bar{4}|\bar{4}}^{\D6\tbar\D6}(x_{1},\pi^{(1)}\,|\,x_{2},\pi^{(2)})\\
\propto&\prod_{\scube\in\pi^{(1)}}\frac{\left(1-q_{4}^{-1}\chi_{\bar{4},x_{1}}(\cube)/x_{2}\right)}{\left(1-\chi_{\bar{4},x_{1}}(\cube)/x_{2}\right)}\prod\limits_{\scubeF\in\pi^{(2)}}\frac{\prod\limits_{\scube\in A(\pi^{(1)})}(1-q_{4}\chi_{\overbar{4},x_{1}}(\cube)/\chi_{\overbar{4},x_{2}}(\cubeF))\prod\limits_{\scube\in R(\pi^{(1)})}(1-\chi_{\overbar{4},x_{1}}(\cube)/\chi_{\overbar{4},x_{2}}(\cubeF))}{\prod\limits_{\scube\in A(\pi^{(1)})}(1-\chi_{\overbar{4},x_{1}}(\cube)/\chi_{\overbar{4},x_{2}}(\cubeF))\prod\limits_{\scube\in R(\pi^{(1)})}(1-q_{4}^{-1}\chi_{\overbar{4},x_{1}}(\cube)/\chi_{\overbar{4},x_{2}}(\cubeF))}.
\eea
When $\pi^{(2)}\succ\pi^{(1)}$, there is a cube $\cube\,'=(i,j,k)$ that is addable to $\pi^{(1)}$ that is included in $\pi^{(2)}$:
\bea
    \exists\,\cube\,'=(i,j,k)\in A(\pi^{(1)}),\quad \cube\,'\in\pi^{(2)}.
\eea
When $x_{2}=q_{4}x_{1}$, the term $(1-q_{4}\chi_{\overbar{4},x_{1}}(\cube\,')/\chi_{\overbar{4},x_{2}}(\cube\,'))=(1-q_{4}x_{1}/x_{2})$ then gives the zero. Thus, we obtain the claim.
\end{proof}
The above discussion is also true for $\Lambda^{K}_{\bar{a},\pi}(x)$, where $K$ is generic.
\begin{lemma}
    For $a\in\four$, after using \eqref{eq:D6supergroupcontraction1} and \eqref{eq:D6supergroupcontraction2}, we have 
    \begin{equation}
        \Lambda_{\bar{a},\pi^{(2)}}^{K_{2}}(x_{2})\Lambda^{K_{1}}_{\bar{a},\pi^{(1)}}(x_{1})=0,
    \end{equation}
    when $x_{2}=q_{a}x_{1}$ and $\pi^{(2)}\succ \pi^{(a)}$ and $K_{1,2}$ are generic.
\end{lemma}
The finite products of the $\D6$ $qq$-characters $\mathsf{T}^{K}_{\bar{a}}(x)\coloneqq\mathsf{T}^{(1|1)}_{\bar{a}}(x|Kx)$ are then expanded as 
\begin{equation}
\mathsf{T}_{\bar{a}}^{K}(x_{N})\dots \mathsf{T}_{\bar{a}}^{K}(x_{1})=\sum_{\pi^{(N)}\preceq\cdots\preceq\pi^{(1)}}\mathfrak{q}^{\sum_{i=1}^{N}|\pi^{(i)}|}\prod_{i=1}^{N}\widetilde{\mathcal{Z}}_{\bar{a}}^{\D6}[K,\pi^{(i)}]\Lambda^{K}_{\bar{a},\pi^{(N)}}(x_{N})\cdots \Lambda^{K}_{\bar{a},\pi^{(1)}}(x_{1}),
\end{equation}
where $x_{i}=q_{a}^{i-1}x\,(i=1,\ldots,N)$. Using the operator product and extracting the one-loop perturbative part, we define the normalized $N$-fusion $\D6$ $qq$-character as 
\begin{equation}\label{eq:D6Nfusion}
    \overline{\mathsf{T}}^{K\,(N)}_{\bar{a}}(x)\coloneqq \sum_{\pi^{(N)}\preceq\cdots\preceq\pi^{(1)}}\mathfrak{q}^{\sum_{i=1}^{N}|\pi^{(i)}|}\prod_{i=1}^{N}\widetilde{\mathcal{Z}}_{\bar{a}}^{\D6}[K,\pi^{(i)}]\prod_{i<j}\mathcal{Z}^{\D6\tbar\D6}_{\bar{a};K\,|\,\bar{a};K}(x_{i},\pi^{(i)}\,|\,x_{j},\pi^{(j)}):\prod_{i=1}^{N}\Lambda_{\bar{a},\pi^{(i)}}^{K}(x_{i}):.
\end{equation}
Taking the limit $N\rightarrow \infty$ and considering the infinite products, we can obtain the D8 $qq$-character defined in \eqref{eq:D8qqdef1}, \eqref{eq:D8qqdef2}.
\begin{theorem}
    The normalized $N$-fusion $\D6$ $qq$-characters $\overline{\mathsf{T}}^{K(N)}_{\bar{a}}(x)\,(a\in\four)$ give the $qq$-characters $\mathsf{T}^{K}_{\four;a}(x)$:
    \begin{equation}
        \overline{\mathsf{T}}^{K(N)}_{\bar{a}}(x)\xrightarrow{N\rightarrow \infty}\mathsf{T}^{K}_{\four;a}(x).
    \end{equation}
    Equivalently, we have 
    \begin{equation}
    \overleftarrow{\prod_{i=1}^{\infty}}\mathsf{T}_{\bar{a}}^{K}(q_{a}^{i-1}x)\simeq \mathsf{T}^{K}_{\four;a}(x),
    \end{equation}
    where the symbol $\simeq $ means the equality is up to one-loop perturbative factors.
\end{theorem}
\begin{proof}
    Let us focus on the case $a=4$. Similar to what we have done in the fusion of $\D4$ $qq$-characters, we need to regularize properly the infinite products. We only give a sketch of the proof of how to take the infinite products:
    \begin{equation}
         \overline{\mathsf{T}}^{K\,(\infty)}_{\bar{4}}(x)\coloneqq \sum_{\emptyset\preceq\cdots\preceq\pi^{(2)}\preceq\pi^{(1)}}\mathfrak{q}^{\sum_{i=1}^{\infty}|\pi^{(i)}|}\prod_{i=1}^{\infty}\widetilde{\mathcal{Z}}_{\bar{4}}^{\D6}[K,\pi^{(i)}]\prod_{i<j}\mathcal{Z}^{\D6\tbar\D6}_{\bar{4};K\,|\,\bar{4};K}(x_{i},\pi^{(i)}\,|\,x_{j},\pi^{(j)}):\prod_{i=1}^{\infty}\Lambda_{\bar{4},\pi^{(i)}}^{K}(x_{i}):.
    \end{equation}
    Since the right-hand side is expanded in arbitrary plane partitions where the $q$-coordinates are tuned properly, we may understand the right-hand side as a sum of arbitrary solid partitions. This comes from the $(1,3)$-type description of the solid partition: $\rho=(\pi^{(1)},\pi^{(2)},\ldots,\emptyset,\ldots)$. The topological term is then rewritten as $\mathfrak{q}^{\sum_{i=1}^{\infty}\pi^{(i)}}=\mathfrak{q}^{|\rho|}$. Let us consider the coefficient parts. The infinite product of the $\U(1|1)$ part is given
    \begin{equation}
        \prod_{i=1}^{\infty}\widetilde{Z}_{\bar{4}}^{\D6}[K,\pi^{(i)}]=\prod_{i=1}^{\infty}\prod_{\scube\in\pi^{(i)}}\frac{(1-Kx_{i}/\chi_{\bar{4},x_{i}}(\cube))(1-x_{i+1}/\chi_{\bar{4},x_{i}}(\cube))}{(1-Kx_{i+1}/\chi_{\bar{4},x_{i}}(\cube))(1-x_{i}/\chi_{\bar{4},x_{i}}(\cube))}\prod_{\substack{\scube\in\pi^{(i)}\\\scubeF\in\pi^{(i)}}}g_{\bar{4}}\left(\frac{\chi_{\bar{4},x_{i}}(\cube)}{\chi_{\bar{4},x_{i}}(\cubeF)}\right)^{-1}.
    \end{equation}
    The other part is given by 
    \bea
        \prod_{i<j}\mathcal{Z}^{\D6\tbar\D6}_{\bar{4};K\,|\,\bar{4};K}(x_{i},\pi^{(i)}\,|\,x_{j},\pi^{(j)})&=\prod_{j>i}\prod_{\scube\in\pi^{(i)}}\frac{\mathscr{V}_{4}\left(x_{j}/\chi_{\bar{4},x_{i}}(\cube)\right)}{\mathscr{V}_{4}\left(Kx_{j}/\chi_{\bar{4},x_{i}}(\cube)\right)}\prod_{\scubeF\in\pi^{(j)}}\frac{\mathscr{V}_{4}\left(x_{i}/\chi_{\bar{4},x_{j}}(\cubeF)\right)}{\mathscr{V}_{4}\left(Kx_{i}/\chi_{\bar{4},x_{j}}(\cubeF)\right)}\\
        &\times \prod_{j>i}\prod_{\substack{\scube\in\pi^{(i)}\\\scubeF\in\pi^{(j)}}}\mathcal{A}_{\mathbb{C}^{4}}\left(\frac{\chi_{\bar{4},x_{i}}(\cube)}{\chi_{\bar{4},x_{j}}(\cubeF)}\right)^{-1}.
    \eea
    Using 
\bea\prod_{j>i}\prod_{\substack{\scube\in\pi^{(i)}\\\scubeF\in\pi^{(j)}}}\mathcal{A}_{\mathbb{C}^{4}}\left(\frac{\chi_{\bar{4},x_{i}}(\cube)}{\chi_{\bar{4},x_{j}}(\cubeF)}\right)^{-1}&=\prod_{i\neq j}\prod_{\substack{\scube\in\pi^{(i)}\\\scubeF\in\pi^{(j)}}}g_{\bar{4}}\left(\frac{\chi_{\bar{4},x_{i}}(\cube)}{\chi_{\bar{4},x_{j}}(\cubeF)}\right)^{-1},\\
     \prod_{i=1}^{\infty}\prod_{\scube\in\pi^{(i)}}\prod_{j=i+1}^{\infty}\frac{\mathscr{V}_{4}\left(x_{j}/\chi_{\bar{4},x_{i}}(\cube)\right)}{\mathscr{V}_{4}\left(Kx_{j}/\chi_{\bar{4},x_{i}}(\cube)\right)}&=\prod_{i=1}^{\infty}\prod_{\scube\in\pi^{(i)}}\frac{\left(1-Kx_{i+1}/\chi_{\bar{4},x_{i}}(\cube)\right)}{\left(1-x_{i+1}/\chi_{\bar{4},x_{i}}(\cube)\right)},\\
     \prod_{j=1}^{\infty}\prod_{\scubeF\in\pi^{(j)}}\prod_{i=1}^{j-1}\frac{\mathscr{V}_{4}\left(x_{i}/\chi_{\bar{4},x_{j}}(\cubeF)\right)}{\mathscr{V}_{4}\left(Kx_{i}/\chi_{\bar{4},x_{j}}(\cubeF)\right)}&=\prod_{j=1}^{\infty}\prod_{\scubeF\in\pi^{(j)}}\frac{\left(1-x_{j}/\chi_{\bar{4},x_{j}}(\cubeF)\right)\left(1-Kx/\chi_{\bar{4},x_{j}}(\cubeF)\right)}{(1-x/\chi_{\bar{4},x_{j}}(\cubeF))(1-Kx_{j}/\chi_{\bar{4},x_{j}}(\cubeF))},
    \eea
    we obtain
    \begin{align}
    \begin{split}
        \prod_{i=1}^{\infty}\widetilde{Z}_{\bar{4}}^{\D6}[K,\pi^{(i)}]\prod_{i<j}\mathcal{Z}^{\D6\tbar\D6}_{\bar{4};K\,|\,\bar{4};K}(x_{i},\pi^{(i)}\,|\,x_{j},\pi^{(j)})&=\prod_{i=1}^{\infty}\prod_{\scube\in\pi^{(i)}}\frac{\left(1-Kx/\chi_{\bar{4},x_{i}}(\cube)\right)}{\left(1-x/\chi_{\bar{4},x_{i}}(\cube)\right)}\prod_{i,j}\prod_{\substack{\scube\in\pi^{(i)}\\\scubeF\in\pi^{(j)}}}g_{\bar{4}}\left(\frac{\chi_{\bar{4},x_{i}}(\cube)}{\chi_{\bar{4},x_{j}}(\cubeF)}\right)^{-1}\\
        &=\prod_{\shcube\in\rho}\frac{\left(1-Kx/\chi_{\four,x}(\hcube)\right)}{\left(1-x/\chi_{\four,x}(\hcube)\right)}\prod_{\shcube,\shcube'\in\rho}g_{\bar{4}}\left(\frac{\chi_{\four,x}(\shcube)}{\chi_{\four,x}(\hcube')}\right)^{-1},
        \end{split}
    \end{align}
    which is exactly $\mathcal{Z}^{\D8}_{\four;4}[\rho;K]$. For the operator part, we can also show
    \begin{equation}
         {:\prod_{i=1}^{\infty}\Lambda^{K}_{\bar{4},\pi^{(i)}}(x_{i}):}={:\Tilde{\mathsf{Z}}^{K}(x)\prod_{\shcube\in\rho}\mathsf{A}^{-1}(\chi_{\four,x}(\hcube)):},
    \end{equation}
    where we used 
    \begin{equation}
        {:\prod_{i=1}^{\infty}\frac{\mathsf{W}_{\bar{4}}(q_{4}^{i-1}x)}{\mathsf{W}_{\bar{4}}(Kq_{4}^{i-1}x)}:}={:\prod_{i=1}^{\infty}\frac{\widetilde{\mathsf{Z}}^{K}(q_{4}^{i-1}x)}{\widetilde{\mathsf{Z}}^{K}(q_{4}^{i}x)}:}=\widetilde{\mathsf{Z}}^{K}(x).
    \end{equation}
    Therefore, we obtain the claim.
\end{proof}

Depending on which $\D6$ $qq$-character $\mathsf{T}_{\bar{a}}(x)$ we use to take the infinite products, we have four possibilities of the $\D8$ $qq$-character. The operator part is invariant under the permutation of the deformation parameters $q_{1},q_{2},q_{3},q_{4}$ but the coefficient part $\mathcal{Z}_{\four;a}^{\D8}[\rho,K]$ is not invariant under the permutation. For example, for lower levels, we have
\begin{subequations}
\begin{align}
    \rho=\{\{\{1\}\}\},&\quad \frac{\mathcal{Z}_{\four;1}^{\D8}[\rho,K]}{\mathcal{Z}_{\four;2}^{\D8}[\rho,K]}=\frac{\mathcal{Z}_{\four;2}^{\D8}[\rho,K]}{\mathcal{Z}_{\four;3}^{\D8}[\rho,K]}=\frac{\mathcal{Z}_{\four;3}^{\D8}[\rho,K]}{\mathcal{Z}_{\four;4}^{\D8}[\rho,K]}=\frac{\mathcal{Z}_{\four;4}^{\D8}[\rho,K]}{\mathcal{Z}_{\four;1}^{\D8}[\rho,K]}=1,\\
    \rho=\{\{\{2\}\}\},&\quad \frac{\mathcal{Z}_{\four;1}^{\D8}[\rho,K]}{\mathcal{Z}_{\four;2}^{\D8}[\rho,K]}=\frac{\mathcal{Z}_{\four;2}^{\D8}[\rho,K]}{\mathcal{Z}_{\four;3}^{\D8}[\rho,K]}=-\frac{\mathcal{Z}_{\four;3}^{\D8}[\rho,K]}{\mathcal{Z}_{\four;4}^{\D8}[\rho,K]}=-\frac{\mathcal{Z}_{\four;4}^{\D8}[\rho,K]}{\mathcal{Z}_{\four;1}^{\D8}[\rho,K]}=1,\\
    \rho=\{\{\{1,1\}\}\},&\quad \frac{\mathcal{Z}_{\four;1}^{\D8}[\rho,K]}{\mathcal{Z}_{\four;2}^{\D8}[\rho,K]}=-\frac{\mathcal{Z}_{\four;2}^{\D8}[\rho,K]}{\mathcal{Z}_{\four;3}^{\D8}[\rho,K]}=-\frac{\mathcal{Z}_{\four;3}^{\D8}[\rho,K]}{\mathcal{Z}_{\four;4}^{\D8}[\rho,K]}=\frac{\mathcal{Z}_{\four;4}^{\D8}[\rho,K]}{\mathcal{Z}_{\four;1}^{\D8}[\rho,K]}=1,\\
    \rho=\{\{\{1\},\{1\}\}\},&\quad -\frac{\mathcal{Z}_{\four;1}^{\D8}[\rho,K]}{\mathcal{Z}_{\four;2}^{\D8}[\rho,K]}=-\frac{\mathcal{Z}_{\four;2}^{\D8}[\rho,K]}{\mathcal{Z}_{\four;3}^{\D8}[\rho,K]}=\frac{\mathcal{Z}_{\four;3}^{\D8}[\rho,K]}{\mathcal{Z}_{\four;4}^{\D8}[\rho,K]}=\frac{\mathcal{Z}_{\four;4}^{\D8}[\rho,K]}{\mathcal{Z}_{\four;1}^{\D8}[\rho,K]}=1,
\end{align}
\end{subequations}
where we described the solid partition as 
\bea
\rho=\{\{\rho_{11},\rho_{12},\ldots\},\ldots,\{\rho_{i1},\rho_{i2}\ldots,\rho_{im_{i}}\},\ldots,\{\rho_{l1},\ldots,\rho_{lm_{l}}\}\},\quad \rho_{ij}=(\rho_{ij1},\ldots,\rho_{ijn_{ij}})
\eea
and the $q$-coordinates of the hypercubes are
\bea
vq_{1}^{a-1}q_{2}^{b-1}q_{3}^{c-1}q_{4}^{d-1},\quad 1\leq a\leq l,\,\,1\leq b\leq m_{a},\,\,1\leq c\leq n_{ab},\,\,1\leq d\leq \rho_{abc}.
\eea
As mentioned in section~\ref{sec:M4partitionfunction}, the partition function of the magnificent four has a sign factor $(-1)^{\sigma_{a}(\rho)}$ depending on the solid partition. In our construction here, this did not appear and thus the appearing coefficient is equal to the $\U(1)$ magnificent four partition function only up to sign factors. One of the reasons this happened is because we naively took the infinite products. For the lower-dimensional $qq$-characters, the infinite products were controlled so that even after taking the infinite products, it still commutes with at least one of the screening charges. To obtain the D8 $qq$-character with the correct sign factors, a different procedure to derive the D8 $qq$-character is necessary.
\subsection{Higher rank magnificent four and D8 \texorpdfstring{$qq$}{qq}-characters}\label{sec:M4rankN}
We still can use the $\D8$ $qq$-characters we constructed in the previous section to derive the partition functions of the rank $N$ magnificent four\footnote{Note that the higher rank version is called the magnificent four with \emph{color} in the original paper. } \cite{Nekrasov:2018xsb} up to sign factors.  The operator products of $\Lambda_{\four,\rho}^{K}(x)$ are 
\beq
    \Lambda_{\four,\rho^{(2)}}^{K_{2}}(x_{2})\Lambda_{\four,\rho^{(1)}}^{K_{1}}(x_{1})=\mathcal{Z}^{\D8\tbar\D8}_{\text{1-loop}}(x_{1},K_{1}\,|\,x_{2},K_{2})\mathcal{Z}^{\D8\tbar\D8}_{K_{1}|K_{2}}(x_{1},\rho^{(1)}\,|\,x_{2},\rho^{(2)}):\Lambda_{\four,\rho^{(2)}}^{K_{2}}(x_{2})\Lambda_{\four,\rho^{(1)}}^{K_{1}}(x_{1}):.
\eeq
The rank $N$ magnificent four partition function is then written using the vertex operators as the following theorem.
\begin{theorem}\label{thm:D8magnificentBPSCFT}
The composition of the $\D8$ $qq$-characters gives the partition function of higher rank magnificent four system up to sign factors:
\bea
    \bra{0}\mathsf{T}_{\four;a_{N}}(x_{N})\cdots \mathsf{T}_{\four;a_{1}}(x_{1})\ket{0}&=\sum_{\rho^{(1)},\cdots ,\rho^{(N)}}\mathfrak{q}^{|\rho|}\prod_{i=1}^{N}\mathcal{Z}_{\four;a_{i}}^{\D8}[\rho^{(i)},K_{i}]\prod_{j>i}\mathcal{Z}_{\text{1-loop}}^{\D8\tbar\D8}(x_{i},K_{i}\,|\,x_{j},K_{j})\\
    &\qquad \qquad\times\prod_{j>i}\mathcal{Z}^{\D8\tbar\D8}_{K_{i}|K_{j}}(x_{i},\rho^{(i)}\,|\,x_{j},\rho^{(j)}).
\eea
This is the BPS/CFT correspondence of the magnificent four.
\end{theorem}
\begin{remark}
    One can also use the other $(2,2),(3,1)$-type descriptions and use infinite products of screening charges or $\D4$ $qq$-characters to rewrite the partition functions. Note that the fusion process we used was essentially the $(1,3)$-type description.
\end{remark}


\section{Quantum toroidal algebras and BPS \texorpdfstring{$qq$}{qq}-characters}\label{sec:toroidal_alg}
In this section, we review the quantum toroidal $\mathfrak{gl}_{1}$ and point out observations regarding the $qq$-characters we introduced in the previous sections. We will show that the $\D2,\D4,\D6$ $qq$-characters are related with the vertical representations of the quantum toroidal $\mathfrak{gl}_{1}$. At the end, we have the following correspondence:
\begin{equation*}
    \renewcommand\arraystretch{1.2}{
    \begin{tabular}{|c|c|c|}\hline
       section & $qq$-characters &  quantum toroidal $\mathfrak{gl}_{1}$\\
     \hline \ref{sec:vectorrep}  & D2 $qq$-character & vector representation \\
      \ref{sec:verticalFockrep} &D4 $qq$-character & Fock representation \\
      \ref{sec:MacMahonrep} &D6 $qq$-character & MacMahon representation \\\hline
    \end{tabular}}
\end{equation*}
We then move on to the gauge origami system on toric CY$_{4}$ $Z=X\times \mathbb{C}$ where $X$ is a toric CY$_3$ and show that the BPS $qq$-characters introduced in Conj.~\ref{conj:BPSqq} are related to the vertical representations of general quantum toroidal algebras associated with toric CY$_3$, which were recently introduced in \cite{Noshita:2021ldl,Galakhov:2021vbo}.

\subsection{Quantum toroidal \texorpdfstring{$\mathfrak{gl}_{1}$}{gl(1)}}\label{sec:QTgl1}
The quantum toroidal $\mathfrak{gl}_{1}$ is an infinite-dimensional quantum algebra with two independent deformation parameters \cite{ding1997generalization,miki2007q,FFJMM1,Feigin2011plane,Feigin2011}. We follow the notations in \cite{DIMreview} (see also \cite[section 5.2, 5.3]{Noshita:2022otp} for a review).

\begin{definition}
Let $\mathsf{q}_{1},\mathsf{q}_{2},\mathsf{q}_{3}$ be the deformation parameters\footnote{In the literature, the deformation parameters of the algebra are denoted $q_{1},q_{2},q_{3}$. We use a different notation to prevent confusion with the parameters $q_{1},q_{2},q_{3},q_{4}$ introduced in this paper. As mentioned in footnote \ref{footnote:structure-function}, the structure function $\sfg(z)$ is related to the structure function $g_{\bar{4}}(z)$ after taking the limit $q_{4}\rightarrow 1$.}  with the condition $\sfq_{1}\sfq_{2}\sfq_{3}=1$. The quantum toroidal $\mathfrak{gl}_{1}$, which is denoted $\mathcal{E}$, is generated by three Drinfeld currents 
\begin{equation}
    E(z)=\sum_{m\in\mathbb{Z}}E_{m}z^{-m},\quad F(z)=\sum_{m\in\mathbb{Z}}F_{m}z^{-m},\quad K^{\pm}(z)=K^{\pm}\exp\left(\sum_{r>0}\mp\frac{\kappa_{r}}{r}H_{\pm r}z^{\mp r}\right)
\end{equation}
and central elements 
\begin{equation}
    C,\quad K^{-}=(K^{+})^{-1}.
\end{equation}
The defining relations are
\begin{align}
\begin{split}
    E(z)E(w)=\sfg(z/w)E(w)E(z),&\quad F(z)F(w)=\sfg(z/w)^{-1}F(w)F(z),\\
    K^{\pm}(z)K^{\pm}(w)=K^{\pm}(w)K^{\pm}(z),&\quad K^{-}(z)K^{+}(w)=\frac{\sfg(C^{-1}z/w)}{\sfg(Cz/w)}K^{+}(w)K^{-}(z),\\
    K^{\pm}(C^{(1\mp 1)/2}z)E(w)&=\sfg(z/w)E(w)K^{\pm}(C^{(1\mp1)/2}z),\\
    K^{\pm}(C^{(1\pm 1)/2}z)F(w)&=\mathsf{g}(z/w)^{-1}F(w)K^{\pm}(C^{(1\pm 1)/2}z),\\
    [E(z),F(w)]=\tilde{g}&\left(\delta\left(\frac{Cw}{z}\right)K^{+}(z)-\delta\left(\frac{Cz}{w}\right)K^{-}(w)\right)
\end{split}
\end{align}
where 
\begin{equation}\label{eq:gl1structurefunction}
    \sfg(z)=\frac{\prod_{i=1}^{3}(1-\sfq_{i}z)}{\prod_{i=1}^{3}(1-\sfq_{i}^{-1}z)},\quad \kappa_{r}=\prod_{i=1}^{3}(\sfq_{i}^{r/2}-\sfq_{i}^{-r/2}),
\end{equation}
and $\tilde{g}=1/\kappa_{1}$.
\end{definition}

Additionally, one needs the so-called Serre relations which are cubic relations of $E(z),F(z)$. We do not need them so we do not write the explicit form (see \cite{DIMreview}). The function $\sfg(z)$ is called the \emph{structure function} of the quantum toroidal $\mathfrak{gl}_{1}$.

The quantum toroidal $\mathfrak{gl}_{1}$ has a Hopf algebraic structure. We only list down the coproduct structure:
\begin{align}\label{eq:coproduct}
\begin{split}
\Delta E(z)&=E(z)\otimes 1+K^{-}(C_{1}z)\otimes E(C_{1}z),\\
\Delta F(z)&=F(C_{2}z)\otimes K^{+}(C_{2}z)+1\otimes F(z),\\
\Delta K^{+}(z)&=K^{+}(z)\otimes K^{+}(C_{1}^{-1}z),\\
\Delta K^{-}(z)&=K^{-}(C_{2}^{-1}z)\otimes K^{-}(z),\\
\Delta(X)&=X\otimes X,\quad X=C,K^{-},
\end{split}
\end{align}
where $C_{1}=C\otimes 1$ and $C_{2}=1\otimes C$. Using this coproduct, we can construct tensor product representations.

The representations of the quantum toroidal $\mathfrak{gl}_{1}$ are obtained by determining the values of the central elements $C,K^{-}$. We have two classes of representations called \emph{vertical representations} and \emph{horizontal representations}. Vertical representations are representations when the central element $C$ is trivial: $C=1$. We have three types of them
\begin{itemize}
    \item vector representation \cite{FFJMM1}: $(C,K^{-})=(1,1)$
    \item Fock representation \cite{Feigin2011}: $(C,K^{-})=(1,\sfq_{c}^{1/2})\,(c=1,2,3)$
    \item MacMahon representation \cite{Feigin2011plane}: $(C,K^{-})=(1,K^{1/2})\,(K\in\mathbb{C}^{\times})$
\end{itemize}
Multi-dimensional partitions appear as the bases of the representation spaces of these representations. For the vector representation, 1d partitions labeled by integers appear. For the Fock and MacMahon representations, 2d and 3d partitions appear respectively (see \cite{DIMreview} and \cite[section 5.3.1]{Noshita:2022otp} for the derivations). 

On the other hand, horizontal representations are representations where the central charges are $(C,K^{-})=(\sfq_{c}^{1/2},1)\,(c=1,2,3)$ \cite{miki2007q,bershtein2018plane,FHSSY:2010,Kojima2019,Kojima2021,Harada:2021xnm}. Drinfeld currents are represented in vertex operators for these representations. See for example \cite{DIMreview} and \cite[section 5.3.2]{Noshita:2022otp} for the explicit derivation of these representations.
\subsection{Vector representation and D2 \texorpdfstring{$qq$}{qq}-character}\label{sec:vectorrep}
There are three types of vector representations with central charges $(C,K^{-})=(1,1)$ and the actions of the Drinfeld currents are
\begin{align}\label{eq:vectorrep}
  \begin{split}
        K^{\pm}(z)[u]^{(c)}_{j}=&\left[\Psi_{[u]^{(c)}_{j}}(z)\right]^{z}_{\pm}[u]^{(c)}_{j}\eqqcolon[S_{c}\left(u\sfq_{c}^{j}/z\right)]_{\pm}[u]^{(c)}_{j},\\
        E(z)[u]^{(c)}_{j}=&\mathcal{E}\delta\left(u\sfq_{c}^{j}/z\right)[u]^{(c)}_{j+1},\\
        F(z)[u]^{(c)}_{j}=&\mathcal{F}\delta\left(u\sfq_{c}^{j-1}/z\right)[u]^{(c)}_{j-1},\quad c=1,2,3,\quad j\in\mathbb{Z}
   \end{split}
\end{align}
where
\begin{equation}
    \mathcal{E}\mathcal{F}=\Tilde{g}\frac{(1-\sfq_{c+1}^{-1})(1-\sfq_{c-1}^{-1})}{(1-\sfq_{c})},\quad S_{c}(z)=\frac{(1-\sfq_{c-1}z)(1-\sfq_{c+1}z)}{(1-z)(1-\sfq_{c-1}\sfq_{c+1}z)}.
\end{equation}
We denote these representations $\mathcal{V}_{c}(u),\,(c=1,2,3)$. The bases $\{[u]_{j}^{(c)}\}_{j\in\mathbb{Z}}$ are represented by 1d partitions (see also section \ref{sec:qcoordinate}):
\begin{equation}
    \begin{tikzpicture}[scale=1.5]
\node at (-2,0.35) {$[u]_{j}^{(c)}=$};
\draw[thick] (0,1)--(0,-0.2);
    \draw[->] (-1,0)--(4.5,0);
    \node at (4.5,0) [right] {$\sfq_{c}$};
    \draw (-1,0.7)--(3.5,0.7);
    \draw (3.5,0.7)--(3.5,0);
    \draw (-0.7,0.7)--(-0.7,0);
    \draw (0.7,0.7)--(0.7,0);
    \draw (1.4,0.7)--(1.4,0);
    \draw (2.1, 0.7)--(2.1,0);
    \node at (1.8, 0.35) {$\cdots$};
    \node at (2.5, 0.35) {$\cdots$};
    \draw (2.8, 0.7)--(2.8,0);
    \node at (0.35, 0) [below] {$1$};
    \node at (1.05,0)[below] {$2$};
    \node at (1.75,0)[below] {$\cdots$};
    \node at (2.45,0)[below] {$\cdots$};
    \node at (3.15,0)[below] {$j$};
    \draw[->] (3.85,1.05)--(3.15,0.35);
    \node at (3.85,1.05)[right] {$u\sfq_{c}^{j-1}$};
    \draw[->] (0.35, 1.05)--(0.35,0.45);
    \node at (0.35, 1.05)[above] {$u$};
    \node at (0.35, 0.20) [right]{$\longrightarrow$};
    \node at (1.05, 0.20) [right,above] {$\sfq_{c}$};
\end{tikzpicture}
\end{equation}
The operator $K^{\pm}(z)$ acts diagonally, $E(z)/F(z)$ adds/removes boxes to/from the configuration. 

To relate the $\D2$ $qq$-characters with the vector representations, we choose one specific direction, which is $\mathbb{C}_{4}$, in the gauge origami system. The motivation for choosing this direction will be explained in section \ref{sec:Betheansatz}. Let us study the relation of the $qq$-characters included in the $\mathbb{C}^{3}_{123}\times \mathbb{S}^{1}$. The operator products of $\mathsf{S}_{c}(z)\,(c=1,2,3)$ with $\mathsf{S}_{4}(q_{4}z)$ are
\begin{subequations}
\begin{align}
    \mathsf{S}_{c}(q_{c}^{j}u)\mathsf{S}_{4}(q_{4}z)&=\left[\mathscr{S}_{\overbar{c4}}(uq_{1}^{j}/z)\right]^{z}_{-}:\mathsf{S}_{c}(q_{c}^{j}u)\mathsf{S}_{4}(q_{4}z):\\
    \mathsf{S}_{4}(q_{4}z)\mathsf{S}_{c}(uq_{c}^{j})&=\left[\mathscr{S}_{\overbar{c4}}(uq_{c}^{j}/z)\right]^{z}_{+}:\mathsf{S}_{c}(q_{c}^{j}u)\mathsf{S}_{4}(q_{4}z):.
\end{align}
\end{subequations}
After taking the limit $q_{4}\rightarrow 1$, we can see that we have
\begin{equation}
    \mathscr{S}_{\overbar{c4}}(uq_{c}^{j}/z)\rightarrow S_{c}(u\sfq_{c}^{j}/z),\quad q_{1},q_{2},q_{3}\rightarrow \sfq_{1},\sfq_{2},\sfq_{3}.
\end{equation}
Comparing with \eqref{eq:vectorrep}, after taking the limit $q_{4}\rightarrow 1$, we can relate the monomial terms of the $\D2$ $qq$-characters with the bases of the vector representation of quantum toroidal $\mathfrak{gl}_{1}$ as\footnote{Note that this is not a strict correspondence. In the limit $q_{4}\rightarrow 1$, some of the vertex operators will diverge and they do not obey the defining relations of the quantum toroidal $\mathfrak{gl}_{1}$. Moreover, for the moment, we do not know how to relate the other operators $E(z),\,F(z)$.}
\begin{align}
    \mathsf{S}_{4}(q_{4}z)\rightarrow K^{\pm}(z),\quad \mathsf{S}_{1}(uq_{1}^{j})\rightarrow [u]_{j}^{(1)}.
\end{align}
This correspondence strengthens the interpretation in \eqref{eq:D2vectorfigure1}. Due to this observation, we can call the $\D2$ $qq$-characters the \textbf{vector $qq$-characters}.

\subsection{Fock representation and D4 \texorpdfstring{$qq$}{qq}-character}\label{sec:verticalFockrep}


Fock representations are representations with central charges $(C,K^{-})=(1,\sfq_{c}^{1/2}),\,(c=1,2,3)$. We denote them $\mathcal{F}_{c}(u),\, (c=1,2,3)$, respectively. The actions of the Drinfeld currents are given
\begin{align}\label{eq:Fockrep}
\begin{split}
    K^{\pm}(z)\ket{u,\lambda}^{(c)}=&\left[\Psi_{\lambda,u}^{(c)}(z)\right]^{z}_{\pm}\ket{u,\lambda}^{(c)}=\left[\sfq_{c}^{-1/2}\frac{\mathcal{Y}^{(c)}_{\lambda,u}(\sfq_{c}^{-1}z)}{\mathcal{Y}^{(c)}_{\lambda,u}(z)}\right]^{z}_{\pm}\ket{u,\lambda}^{(c)},\\
E(z)\ket{u,\lambda}^{(c)}=&\frac{1-\sfq_{c}}{\kappa_{1}}\sum_{\Abox\in A(\lambda)}\delta\left(\chi_{u}^{(c)}(\Bbox)/z\right)\underset{z=\chi_{u}^{(c)}(\Abox)}{\Res}z^{-1}\mathcal{Y}^{(c)}_{\lambda,u}(z)^{-1}\ket{u,\lambda+\Bbox}^{(c)},\\
    F(z)\ket{u,\lambda}^{(c)}=&-\frac{1-\sfq_{c}^{-1}}{\kappa_{1}}\sfq_{c}^{-1/2}\sum_{\Abox\in R(\lambda)} \delta\left(\chi_{u}^{(c)}(\Bbox)/z\right)\underset{z=\chi_{u}^{(c)}(\Abox)}{\Res}z^{-1}\mathcal{Y}_{\lambda,u}^{(c)}(\sfq_{c}^{-1}z)\ket{u,\lambda-\Bbox}^{(c)}
\end{split}
\end{align}
where
\begin{equation}
    \mathcal{Y}_{\lambda,u}^{(c)}(z)=(1-u/z)\prod_{\Abox\in\lambda}S_{c}(\chi^{(c)}_{u}(\Bbox)/z),\quad \chi^{(c)}_{u}(\Bbox)=u\sfq_{c+1}^{i-1}\sfq_{c-1}^{j-1}\,\,(i,j\geq1 ).\label{def-Yc}
\end{equation}
Note that the eigenvalue $\Psi_{\lambda,u}^{(c)}(z)$ can be rewritten as 
\begin{equation}\label{eq:FockCartan}
\Psi^{(c)}_{\lambda,u}(z)=\sfq_{c}^{-1/2}\frac{1-\sfq_{c}u/z}{1-u/z}\prod_{\Abox\in\lambda}\sfg\left(\frac{z}{\chi_{u}^{(c)}(\Bbox)}\right)
\end{equation}
The bases are represented by 2d partitions (see also section \ref{sec:qcoordinate} for the notation):
\begin{equation}
        \ket{u,\lambda}^{(c)}=\qquad \adjustbox{valign=c}{\begin{tikzpicture}[scale=0.7]
        \draw[->] (-1,0)--(4,0);
        \node[above] at (-0.5,4){$\sfq_{c-1}$};
        \node [right] at (4,0){$\sfq_{c+1}$};
        \node [below] at (1.25,0){$i$};
         \draw[->]   (-0.5,-0.5)--(-0.5,4);
         \draw (0.2,3.5)--(0.2,0.7);
         \draw (0.9,2.8)--(0.9,0.7);
         \draw (1.6,2.1)--(1.6,0.7);
         \draw (2.3,1.4)--(2.3,0.7);
         \draw (2.3,0.7)--(-0.5,0.7);
         \draw (2.3,1.4)--(-0.5,1.4);
         \draw (1.6,2.1)--(-0.5,2.1);
         \draw (0.9,2.8)--(-0.5,2.8);
         \draw (0.2,3.5)--(-0.5,3.5);
        \draw (-0.5,0.7)--(3,0.7);
        \draw (3,0)--(3,0.7);
        \draw (0.2,0)--(0.2,0.7);
        \draw (0.9,0)--(0.9,0.7);
         \draw (1.6,0)--(1.6,0.7);
          \draw (2.3,0)--(2.3,0.7);
          \draw (-0.15,0.35)--++(-0.7,-1);
          \node[left] at (-0.85,-0.65){$u$};
          \node [left] at (-0.5,1.75){$j$};
          \draw  (1.25,1.75)--++(0.9,0.9);
          \node[right] at (2.2,2.65) {$u\sfq_{c+1}^{i-1}\sfq_{c-1}^{j-1}$};
        \end{tikzpicture}
        }
    \end{equation}

Similar to the D2 case, we choose $\mathbb{C}_{4}$ to be a specific direction in the gauge origami system. Focusing on $\mathsf{T}_{12}(x)$ and using \eqref{eq:D4-D2contraction}, we obtain
\begin{subequations}
\begin{align}
     :\mathsf{X}_{12}(u)\prod_{\Abox\in\lambda}\mathsf{A}^{-1}(\chi_{12,u}(\Bbox)):\mathsf{S}_{4}(q_{4}z)&=\left[q_{3}^{-1}\frac{\mathscr{Y}^{12}_{\lambda,u}(q_{3}^{-1}z)}{\mathscr{Y}^{12}_{\lambda,u}(z)}\right]^{z}_{-} :\mathsf{X}_{12}(u)\prod_{\Abox\in\lambda}\mathsf{A}^{-1}(\chi_{12,u}(\Bbox))\mathsf{S}_{4}(q_{4}z):,\\
    \mathsf{S}_{4}(q_{4}z):\mathsf{X}_{12}(u)\prod_{\Abox\in\lambda}\mathsf{A}^{-1}(\chi_{12,u}(\Bbox)):&=\left[q_{3}^{-1}\frac{\mathscr{Y}^{12}_{\lambda,u}(q_{3}^{-1}z)}{\mathscr{Y}^{12}_{\lambda,u}(z)}\right]^{z}_{+} :\mathsf{X}_{12}(u)\prod_{\Abox\in\lambda}\mathsf{A}^{-1}(\chi_{12,u}(\Bbox))\mathsf{S}_{4}(q_{4}z):
\end{align}
\end{subequations}
and at the limit $q_{4}\rightarrow 1$, we have 
\begin{equation}
    \mathscr{Y}_{\lambda,u}^{12}(z)\rightarrow \mathcal{Y}^{(3)}_{\lambda,u}(z),\quad \chi_{12,u}(\Bbox)\rightarrow \chi_{u}^{(3)}(\Bbox).
\end{equation}
Thus, at the limit $q_{4}\rightarrow 1$, we can relate the monomial terms of the $\D4$ $qq$-characters with the bases of the Fock representation of quantum toroidal $\mathfrak{gl}_{1}$ as 
\begin{equation}
    \mathsf{S}_{4}(q_{4}z)\longrightarrow K^{\pm}(z),\quad 
    {:\mathsf{X}_{12}(u)\prod_{\Abox\in\lambda}\mathsf{A}^{-1}(\chi_{12,u}(\Bbox)):}\longrightarrow \ket{u,\lambda}^{(3)}.
\end{equation}
In this sense, we can call the $\D4$ $qq$-characters the \textbf{Fock $qq$-characters}.

\begin{remark}
    Actually, the Fock representation can be derived by an infinite number of tensor products of the vector representation \cite{FFJMM1}. The coproduct structure enables us to consider the action of the Drinfeld currents on tensor product representations $\otimes_{i=1}^{N}\mathcal{V}_{c-1}(u_{i})$. We tune the spectral parameter as $u_{i}=u\sfq_{c+1}^{i-1}$ and take the limit $N\rightarrow \infty$. After proper regularization of the infinite products, we obtain $\otimes_{i=1}^{\infty}\mathcal{V}_{c-1}(u\sfq_{c+1}^{i-1})\simeq \mathcal{F}_{c}(u) $. This property corresponds with the fact that the D4 $qq$-character can be obtained by the fusion process of the D2 $qq$-characters as discussed in Thm.~\ref{thm:D2toD4fusion}.
\end{remark}

\paragraph{Higher rank $qq$-characters}
Let us show that the higher rank $\D4$ $qq$-characters correspond to the tensor product representations of the vertical Fock representations. We only focus on the $qq$-characters with no negative weights. The monomial terms appearing in $\mathsf{T}_{12:23:13}^{(\vec{n}|\vec{0})}(\underline{\vec{x}}\,|\,\underline{\vec{0}})$ (see Thm.~\ref{thm:D4generalqqcharacter}) are 
\begin{equation}
    :\prod_{\alpha=1}^{n_{12}}\Lambda_{12,\lambda_{\alpha}}(x_{12,\alpha})\prod_{\beta=1}^{n_{13}}\Lambda_{13,\mu_{\beta}}(x_{13,\beta})\prod_{\gamma=1}^{n_{23}}\Lambda_{23,\nu_{\gamma}}(x_{23,\gamma}):
\end{equation}
where $\lambda_{\alpha},\mu_{\beta},\nu_{\gamma}$ are Young diagrams. The coefficient appearing after taking the operator product with $\mathsf{S}_{4}(q_{4}z)$ is 
\begin{equation}
    \prod_{\alpha=1}^{n_{12}}q_{3}^{-1}\frac{\mathscr{Y}^{12}_{\lambda_{\alpha},x_{12,\alpha}}(q_{3}^{-1}z)}{\mathscr{Y}^{12}_{\lambda_{\alpha},x_{12,\alpha}}(z)}\prod_{\beta=1}^{n_{13}}q_{2}^{-1}\frac{\mathscr{Y}^{13}_{\mu_{\beta},x_{13,\beta}}(q_{2}^{-1}z)}{\mathscr{Y}^{13}_{\mu_{\beta},x_{13,\beta}}(z)}\prod_{\gamma=1}^{n_{23}}q_{1}^{-1}\frac{\mathscr{Y}^{23}_{\nu_{\gamma},x_{23,\gamma}}(q_{1}^{-1}z)}{\mathscr{Y}^{23}_{\nu_{\gamma},x_{23,\gamma}}(z)}.
\end{equation}
After taking the limit $q_{4}\rightarrow 1$, it corresponds with the Cartan eigenvalue of the tensor product representation $\bigotimes_{\alpha=1}^{n_{12}}\mathcal{F}_{3}(x_{12,\alpha})\otimes \bigotimes_{\beta=1}^{n_{13}}\mathcal{F}_{2}(x_{13,\beta})\otimes\bigotimes_{\gamma=1}^{n_{23}}\mathcal{F}_{1}(x_{23,\gamma})$. Note that the Cartan eigenvalues of tensor product representations are simply the products of the Cartan eigenvalue of each representation. This is because at $C=1$, using the coproduct structure \eqref{eq:coproduct}, we have $\Delta K^{\pm}(z)=K^{\pm}(z)\otimes K^{\pm}(z)$. Note also that the ordering of the tensor products does not matter because of the existence of the universal R-matrix\footnote{Recently there have been attempts to construct the $qq$-character using the R-matrix of the quantum toroidal $\mathfrak{gl}_1$~\cite{Liu:2022gwf,Bayindirli:2023byn}.
} \cite{miki2007q,Feigin:2015raa,Feigin:2016pld}.

We can do the same analysis for the negative weights where supergroup analogs appear. The corresponding representations in the quantum toroidal $\mathfrak{gl}_{1}$ were constructed in \cite{Noshita:2022dxv,Bourgine:2018fjy,Feigin:2016pld}, where vertical Fock representations with negative levels were introduced. One can show that the representations in \cite{Noshita:2022dxv} correspond with the $\D4$ $qq$-character generated by $\mathsf{X}_{A}(z)^{-1}$.

\subsection{MacMahon representation and D6 \texorpdfstring{$qq$}{qq}-character}\label{sec:MacMahonrep}
MacMahon representations are representations with central charges $(C,K^{-})=(1,K^{1/2})$ where $K\in\mathbb{C}^{\times}$ is a generic parameter. The action of the Drinfeld currents is given as
\bea\label{eq:MacMahonrep}
    K^{\pm}(z)|u,\pi\rangle&=\left[\Psi_{\pi,u}(z)\right]^{z}_{\pm}|u,\pi\rangle,\\
    E(z)|u,\pi\rangle=&\sum_{\scube\in A(\pi)}\#\delta\left(\frac{z}{\chi_{u}(\cube)}\right)
    \sqrt{\underset{z=\chi_{u}(\scube)}{\mathrm{Res}}z^{-1}\Psi_{\pi,u}(z)}\,
    |u,\pi+\cube\rangle,\\
    F(z)|u,\pi\rangle&=\sum_{\scube\in R(\pi)}\#\delta\left(\frac{z}{\chi_{u}(\cube)}\right)
    \sqrt{\underset{z=\chi_{u}(\scube)}{\mathrm{Res}} z^{-1}\Psi_{\pi,u}(z)}\,
    |u,\pi-\cube\rangle,
\eea 
where $\#$ is some coefficient factor and 
\begin{equation}\label{eq:CartanMacMahon}
    \chi_{u}(\cube)=u\sfq_{1}^{i-1}\sfq_{2}^{j-1}\sfq_{3}^{k-1},\quad  \Psi_{\pi,u}(z)=K^{-1/2}\frac{1-Ku/z}{1-u/z}\prod_{\scube\in\pi}\sfg\left(\frac{z}{\chi_{u}(\cube)}\right).
\end{equation}
 We denote this representation $\mathcal{M}(u,K)$. The explicit coefficients of the right-hand side of $E(z),F(z)$ are omitted. The $qq$-character $\mathsf{T}_{123}(x)$ and screening current $\mathsf{S}_{4}(x')$ gives 
\begin{subequations}
\begin{align}
  :\frac{\mathsf{W}_{\bar{4}}(u)}{\mathsf{W}_{\bar{4}}(Ku)}\prod_{\scube\in\pi}\mathsf{A}^{-1}(\chi_{\overbar{4},u}(\cube)):\mathsf{S}_{4}(q_{4}z)
  &=-q_{4}u\left[\mathscr{W}^{\bar{4},K}_{\pi,u}(z)^{-1}\right]^{z}_{-}:\frac{\mathsf{W}_{\bar{4}}(u)}{\mathsf{W}_{\bar{4}}(Ku)}\prod_{\scube\in\pi}\mathsf{A}^{-1}(\chi_{\overbar{4},u}(\cube))\mathsf{S}_{4}(q_{4}z):,\\
  \mathsf{S}_{4}(q_{4}z):\frac{\mathsf{W}_{\bar{4}}(u)}{\mathsf{W}_{\bar{4}}(Ku)}\prod_{\scube\in\pi}\mathsf{A}^{-1}(\chi_{\overbar{4},u}(\cube)):&=-q_{4}u\left[\mathscr{W}_{\pi,x}^{\bar{4},K}(z)^{-1}\right]^{z}_{+}:\frac{\mathsf{W}_{\bar{4}}(u)}{\mathsf{W}_{\bar{4}}(Ku)}\prod_{\scube\in\pi}\mathsf{A}^{-1}(\chi_{\overbar{4},u}(\cube))\mathsf{S}_{4}(q_{4}z):.
\end{align}
\end{subequations}
Taking the limit $q_{4}\rightarrow 1$, we have $\chi_{\bar{4},u}(\cube)\rightarrow \chi_{u}(\cube)$ and 
\begin{equation}
g_{\bar{4}}(z)\longrightarrow \sfg(z),\quad 
\mathcal{W}_{\pi,u}^{\bar{4},K}(z)^{-1}\longrightarrow \frac{1-Ku/z}{1-u/z}\prod_{\scube\in\pi}\sfg\left(\frac{z}{\chi_{u}(\cube)}\right)
\end{equation}
which gives the identification
\begin{equation}
\mathsf{S}_{4}(q_{4}z)\longrightarrow K^{\pm}(z),\quad 
:\mathsf{W}_{\bar{4}}(u)\prod_{\scube\in\pi}\mathsf{A}^{-1}(\chi_{\overbar{4},u}(\cube)):\,\longrightarrow \ket{u,\pi}.
\end{equation}
Thus, at the limit $q_{4}\rightarrow 1$, we can relate the monomial terms of the $\D6$ $qq$-characters with the bases of the MacMahon representation of the quantum toroidal 
$\mathfrak{gl}_{1}$. In this sense, we call the $\D6$ $qq$-characters the \textbf{MacMahon $qq$-characters}. Under this identification, we can see that the distance between the $\D6$ and $\overline{\D6}$ branes, denoted as $K$, appear as the central charge of the MacMahon representation.

\begin{remark}
    Instead of considering the $\U(1|1)$ theory of $\D6$-branes, we can consider the $\U(1)$ theory by taking the limit $K\rightarrow 0,\,\infty$. After taking the limit $q_{4}\rightarrow 1$, the eigenvalue $\Psi_{\pi,u}(z)|_{K\rightarrow 0,\infty}$ is no longer a rational function with the same degrees in the numerator and the denominator. Note that such kind of representations are not the representation of quantum toroidal $\mathfrak{gl}_{1}$ but of the \emph{shifted} quantum toroidal $\mathfrak{gl}_{1}$ (see \cite{Bourgine:2022scz,Galakhov:2021xum,Noshita:2021dgj}). 
\end{remark}
\begin{remark}
    Similar to the relation between the vector representations and the Fock representations, the MacMahon representation also can be obtained by infinite tensor products of the Fock representations \cite{Feigin2011plane}. Consider the tensor product $\otimes_{i=1}^{N}\mathcal{F}_{c}(u_{i})$ and tune the spectral parameters as $u_{i}=u\sfq_{c}^{i-1}$. We then take the limit $N\rightarrow \infty$ and regularize it. After this process, we get $\otimes_{i=1}^{\infty}\mathcal{F}_{c}(u\sfq_{c}^{i-1})\simeq \mathcal{M}(K,u)$. Note that the nontrivial parameter $K$ appears from the regularization process. This property corresponds to the fusion process of the D4 $qq$-characters to $\D6$ $qq$-characters discussed in Thm.~\ref{thm:D4toD6fusion}.
\end{remark}

\paragraph{General $\D6$ $qq$-characters}
Similar to the $\D4$ case, a higher rank version of $\D6$ $qq$-characters corresponds to the tensor products of the MacMahon representations:
\begin{align}
\begin{split}
    :\prod_{i=1}^{N}\Lambda^{K_{i}}_{\bar{4},\pi^{(i)}}(x_{i}):\mathsf{S}_{4}(q_{4}x)=\prod_{i=1}^{N}(-q_{4}x_{i})\prod_{i=1}^{N}\left[\mathscr{W}^{\bar{4},K_{i}}_{\pi^{(i)},x_{i}}(z)^{-1}\right]^{z}_{-}:\prod_{i=1}^{N}\Lambda^{K_{i}}_{\bar{4},\pi^{(i)}}(x_{i})\mathsf{S}_{4}(q_{4}x):,\\
    \mathsf{S}_{4}(q_{4}x){:\prod_{i=1}^{N}\Lambda^{K_{i}}_{\bar{4},\pi^{(i)}}(x_{i}):}=\prod_{i=1}^{N}(-q_{4}x_{i})\prod_{i=1}^{N}\left[\mathscr{W}^{\bar{4},K_{i}}_{\pi^{(i)},x_{i}}(z)^{-1}\right]^{z}_{+}:\prod_{i=1}^{N}\Lambda^{K_{i}}_{\bar{4},\pi^{(i)}}(x_{i})\mathsf{S}_{4}(q_{4}x):.
\end{split}
\end{align}
After taking the limit $q_{4}\rightarrow 1$, one can see that it matches with the Cartan eigenvalue of the tensor product $\bigotimes_{i=1}^{N}\mathcal{M}(x_{i},K_{i})$.

As discussed in section \ref{sec:generalD6qq}, after specifying the value $K$, we can obtain the pit reduction of the $\D6$ $qq$-characters, which eventually gives the $\D4$ $qq$-character. This situation is the same in the quantum toroidal $\mathfrak{gl}_{1}$. Setting $K=\sfq_{c}$ in \eqref{eq:CartanMacMahon}, the Cartan eigenvalue will be \eqref{eq:FockCartan}. The residue of $\Psi_{\pi,u}(z)|_{K=\sfq_{c}}$ at $z=\sfq_{c}u$ will vanish and thus the action of $E(z)$ will stop the growth of the plane partition in the direction $\sfq_{c}$. We then obtain the $\mathcal{F}_{c}(u)$ representation.\footnote{The coefficients of the action of $E(z),F(z)$ on the bases in the MacMahon representation after setting $K=\mathsf{q}_{c}$ is different from the coefficients in the Fock representation. This comes from the degree of freedom to rescale the bases.} For a general pit located at $(L,M,N)$ in the MacMahon representation, the central charge is $K=\sfq_{1}^{L-1}\sfq_{2}^{M-1}\sfq_{3}^{N-1}$. A similar analysis can be done and we will see that the MacMahon representation will be reduced to $\mathcal{F}_{1}^{\otimes L}\otimes \mathcal{F}_{2}^{\otimes M}\otimes \mathcal{F}_{3}^{\otimes N}$ with the spectral parameters tuned properly (see for example \cite[section 5.1.4]{DIMreview}). 

We can also do the same analysis for the case when the plane partitions appearing have nontrivial boundary Young diagrams. For these cases, the eigenvalue of the vacuum configuration is determined as
\begin{equation}
    \Psi_{\lambda\mu\nu,u}^{\text{vac}}(z)=K^{-1/2}\frac{1-Ku/z}{1-u/z}\prod_{\scube\in \mathcal{P}_{\lambda,\mu,\nu}}\sfg\left(\frac{z}{\chi_{u}(\cube)}\right).
\end{equation}
After taking the $q_{4}\rightarrow 1$ limit, the coefficient appearing after taking the operator product of \eqref{eq:D6hw_boundarypartition} and $\mathsf{S}_{4}(q_{4}z)$ will become this vacuum Cartan eigenvalue.

\subsection{BPS \texorpdfstring{$qq$}{qq}-characters and quiver quantum toroidal algebras}\label{sec:BPSqqQTA}
Up to the previous section, we managed to relate the $qq$-characters with the vertical representation of the quantum toroidal $\mathfrak{gl}_{1}$. To construct the explicit representation of $\mathcal{E}$, the essential part is how to determine the bases and the eigenvalue of $K^{\pm}(z)$. The action of the other Drinfeld currents is determined schematically as \eqref{eq:MacMahonrep}, where the coefficients are proportional to the residue of the eigenvalue. Moreover, the eigenvalue has a general structure
\bea
\Psi_{\Lambda,u}(x)=\psi_{\emptyset,u}(x)\prod_{\scube\in \Lambda}\mathsf{g}\left(\frac{z}{\chi_{u}(\cube)}\right)
\eea
where $\Lambda$ is either the 1d, 2d, 3d partitions or general truncations of the plane partition. The eigenvalue is a product of the vacuum function $\psi_{\emptyset,u}(z)$, and the contributions come from the boxes in the configurations. Actually, one can show that the zero and pole structure of the vacuum function $\psi_{\emptyset,u}(z)$ determines how the bases should be constructed \cite{Feigin2011plane,Prochazka:2015deb,Galakhov:2021xum}.

This property is similar to the procedure to obtain $qq$-characters. To determine the $qq$-characters, we first started by choosing a screening current. We then choose a highest weight and impose the commutativity with the screening current. The $qq$-character is automatically determined in this way. The highest weight therefore has a one-to-one correspondence with the vacuum function of the vertical representations of $\mathcal{E}$:
\bea
\mathsf{S}_{c}(u)\quad &\longleftrightarrow \quad \frac{(1-\sfq_{c-1}u/z)(1-\sfq_{c+1}u/z)}{(1-u/z)(1-\sfq_{c-1}\sfq_{c+1}u/z)},\\
\mathsf{X}_{\overline{c4}}(u)\quad &\longleftrightarrow \quad \sfq_{c}^{-1/2}\frac{1-\sfq_{c}u/z}{1-u/z},\\
\frac{\mathsf{W}_{123}(u)}{\mathsf{W}_{123}(Ku)} \quad &\longleftrightarrow \quad K^{-1/2}\frac{1-Ku/z}{1-u/z}.
\eea
Note that this is just a rewriting of the conclusion obtained in sections \ref{sec:vectorrep}, \ref{sec:verticalFockrep}, \ref{sec:MacMahonrep}.
 
Recently a large class of quantum toroidal algebras associated with toric Calabi--Yau three-folds\footnote{There are also quantum toroidal algebras associated with non-toric CYs. They are related to affine D and E-type algebras. The discussion here can also be generalized to those cases.} were constructed \cite{Galakhov:2021vbo,Noshita:2021ldl}. They are called quiver quantum toroidal algebras or toroidal quiver BPS algebras\footnote{We use the former terminology.}. We denote such algebras as $\mathcal{E}_{(Q,W)}$, where $Q$ is the corresponding quiver and $W$ is the superpotential. The above correspondence between the $qq$-characters and the vertical representations implies that we should have generalizations to such cases too. Let us show that the BPS $qq$-characters introduced in Conj.~\ref{conj:BPSqq} are the $qq$-characters corresponding to the representations of such algebras.

\subsubsection{Quiver quantum toroidal algebra}
Let us introduce the definition of the algebra $\mathcal{E}_{(Q,W)}$ following \cite{Noshita:2021ldl}. Let $X$ be a toric Calabi--Yau three-fold. To $X$, one can associate a quiver $Q=(Q_{0},Q_{1})$ and a superpotential $W$. For each arrow of the quiver, we can associate parameters $\{\sfq_{I}\}_{I\in Q_{1}}$. The superpotential imposes nontrivial conditions on them and the independent parameters are reduced up to two. We assume such conditions are imposed implicitly in the parameters.
\begin{definition}[\cite{Noshita:2021dgj,Noshita:2021ldl,Galakhov:2021xum,Galakhov:2021vbo}]
    The algebra $\mathcal{E}_{(Q,W)}$ is generated by the Drinfeld currents:
    \bea
    E_{i}(z)=\sum_{m\in\mathbb{Z}}E_{i,m}z^{-m},\quad F_{i}(z)=\sum_{m\in\mathbb{Z}}F_{i,m}z^{-m},\quad K_{i}^{\pm}(z)=K_{i}^{\pm}z^{r^{\pm}_{i}}\exp\left(\pm \sum_{r=1}^{\infty}H_{i,\pm r}z^{\mp r}\right),
    \eea
   for $i\in Q_{0}$, where $r^{\pm}_{i}\in \mathbb{Z}$. For each node $i\in Q_{0}$, we assign a $\mathbb{Z}_{2}$ grading:
   \bea
   |i|=\begin{cases}
       0\quad (\exists I\in Q_{1}\quad \text{s.t.}\quad \text{$I$ is a loop}),\\
       1\quad (\text{otherwise}).
   \end{cases}
   \eea
   We have one universal central element which is denoted as $C$. The defining relations are 
    \bea
    K_{i}^{\pm}(z)K_{j}^{\pm}(w)&=K_{j}^{\pm}(w)K_{i}^{\pm}(z),\\
    K_{i}^{-}(z)K_{j}^{+}(w)&=\frac{\widetilde{\varphi}^{j\Rightarrow i}(z,Cw)}{\widetilde{\varphi}^{j\Rightarrow i}(Cz,w)}K_{j}^{+}(w)K_{i}^{-}(z),\\
    K_{i}^{\pm}(C^{\frac{1\mp 1}{2}}z)E_{j}(w)&=\widetilde{\varphi}^{j\Rightarrow i}(z,w)E_{j}(w)K_{i}^{\pm}(C^{\frac{1\mp 1}{2}}z),\\
    K_{i}^{\pm}(C^{\frac{1\pm 1}{2}}z)F_{j}(w)&=\widetilde{\varphi}^{j\Rightarrow i}(z,w)^{-1}F_{j}(w)K_{i}^{\pm}(C^{\frac{1\pm 1}{2}}z),\\
    [E_{i}(z),F_{j}(w)]=\delta_{ij}&\left(\delta\left(\frac{Cw}{z}\right)K^{+}_{i}(z)-\delta\left(\frac{Cz}{w}\right)K_{i}^{-}(w)\right),\\
    E_{i}(z)E_{j}(w)&=(-1)^{|i||j|}\widetilde{\varphi}^{j\Rightarrow i}(z,w)E_{j}(w)E_{i}(z),\\
    F_{i}(z)F_{j}(w)&=(-1)^{|i||j|}\widetilde{\varphi}^{j\Rightarrow i}(z,w)^{-1}F_{j}F_{i}(z),
    \eea
    where the structure functions $\widetilde{\varphi}^{i\Rightarrow j}(z,w)$ are defined as 
    \bea
    \widetilde{\varphi}^{i\Rightarrow j}(z,w)&=(-1)^{\chi_{i\rightarrow j}}\frac{\prod_{I\in\{j\rightarrow i \}}(\sfq_{I}^{1/2}z-\sfq_{I}^{-1/2}w)}{\prod_{I\in\{i\rightarrow j\}}(\sfq_{I}^{-1/2}z-\sfq_{I}^{1/2}w)}\\
    &=(-1)^{\chi_{i\rightarrow j}}\frac{\prod_{I\in\{j\rightarrow i\}}(-\sfq_{I}^{-1/2}w)}{\prod_{I\in\{i\rightarrow j\}}(-\sfq_{I}^{1/2}w)}\varphi^{i\Rightarrow j}(z/w)
    \eea
    where $\chi_{i\rightarrow j}$ is a factor determined by imposing the associativity condition
    \bea
    \widetilde{\varphi}^{j\Rightarrow i}(z,w)\widetilde{\varphi}^{i\Rightarrow j}(w,z)=1
    \eea
    and $\varphi^{i\Rightarrow j}(z)$ is defined in \eqref{eq:CY3timesCstructurefunct}.
\end{definition}
\begin{remark}
    We add the factor $z^{r_{i}^{\pm}}$ to the mode expansion of $K_{i}^{\pm}(z)$ and relaxed the algebra. Such modifications will be important only when we consider the mode expansions of representations. Moreover, they are related to shifted quantum algebras. They are not important in this paper so we will not discuss them.
\end{remark}

Setting $C=1$, one can see that $K_{i}^{\pm}(z)$ commutes with each other and we have diagonal bases. Such diagonal bases are labeled by 3d BPS crystals \cite{Ooguri:2009ijd} and we have a natural action of the algebra on them (see \cite{Li:2020rij} for details). The 3d BPS crystals are a set of colored atoms where the colors are labeled by $Q_{0}$ (see Conj.~\ref{conj:BPSqq} for the notations). The color projection map $\text{c}:\Lambda\rightarrow Q_{0}$ maps the color of the crystal to the quiver nodes. Similar to the discussion in Conj.~\ref{conj:BPSqq}, we can define a coordinate function:
\bea
\chi_{u}(\cube)\coloneqq u\times \prod_{I\in\text{path}[\mathfrak{o}\rightarrow\,\scube]}\sfq_{I}.
\eea
The difference with the coordinate function in Conj.~\ref{conj:BPSqq} is that we use the $q$-deformation parameters $\sfq_{I}$. There, we have an additional parameter $q_{4}$, while here we do not have such a parameter.

Under this situation, the representation is schematically given as
\bea
K_{i}^{\pm}(z)\ket{u,\Lambda}&=\left[\Psi_{\Lambda,u}^{(i)}(z)\right]^{z}_{\pm}\ket{u,\Lambda},\\
E_{i}(z)\ket{u,\Lambda}&=\sum_{\scube\in A(\Lambda)}\#\sqrt{\underset{z=\chi_{u}(\scube)}{\Res}z^{-1}\Psi^{(i)}_{\Lambda,u}(z)}\delta\left(\frac{z}{\chi_{u}(\cube)}\right)\ket{u,\Lambda+\cube},\\
F_{i}(z)\ket{u,\Lambda}&=\sum_{\scube\in R(\Lambda)}\#\sqrt{\underset{z=\chi_{u}(\scube)}{\Res}z^{-1}\Psi^{(i)}_{\Lambda,u}(z)}\delta\left(\frac{z}{\chi_{u}(\cube)}\right)\ket{u,\Lambda-\cube}
\eea
where $\Lambda$ is the 3d BPS crystal configuration and $A(\Lambda),R(\Lambda)$ are addable and removable atoms from the configuration and $\#$ is some coefficient factor. The eigenvalue $\Psi_{\Lambda,u}^{(i)}(z)$ schematically has the form
\bea\label{eq:QQTAeigenvalue}
\Psi^{(i)}_{\Lambda,u}(z)=\psi^{(i)}_{\emptyset,u}(z)\prod_{j\in Q_{0}}\prod_{\substack{\scube\,\in\Lambda,\\\text{c}(\scube)=j}}\widetilde{\varphi}^{j\Rightarrow i}(z,\chi_{u}(\cube))
\eea
where $\psi^{(i)}_{\emptyset,u}(z)$ is the vacuum function. Namely, the eigenvalue is a product of the vacuum function and contributions coming from all of the atoms in the crystal configuration. The contribution to $\Psi^{(i)}_{\Lambda,u}(z)$ coming from an atom at the coordinate $\chi_{u}(\cube)$ with color $j\in Q_{0}$ comes from $\widetilde{\varphi}^{j\Rightarrow i}(z,\chi_{u}(\cube))$. The poles of the eigenvalue $\Psi^{(i)}_{\Lambda,u}(z)$ are the $q$-coordinates of the addable and removable atoms with color $i\in Q_{0}$ and thus the action of $E_{i}(z),\,F_{i}(z)$ are determined by the residue of them. 

The possible configuration is uniquely determined by the vacuum function $\psi^{(i)}_{\emptyset,u}(z)$ (see \cite{Galakhov:2021xum} for examples). We only focus on 3d BPS crystals with only one atom at the origin. For such cases, the vacuum function is given 
\bea
\psi^{(i)}_{\emptyset,u}(z)=\left(\frac{K^{-1/2}z-K^{1/2}u}{z-u}\right)^{\delta_{ia}}.
\eea
This vacuum function gives the BPS crystals where the atom at the origin has color $a$. We denote such crystal configurations as $\Lambda^{(a)}$ for later use.

\begin{remark}
    When $X=\mathbb{C}^{2}/\mathbb{Z}_{n}\times \mathbb{C}$, the corresponding quantum toroidal algebra is the quantum toroidal $\mathfrak{gl}_{n}$ \cite{Ginzburg,Feigin:2013fga,Feigin:2013JA}, which we denote $\mathcal{E}_{n}$ for later use. Applications of these algebras to supersymmetric gauge theories were done in \cite{Awata:2017lqa}. For general abelian orbifold cases such as $X=\mathbb{C}^{3}/\mathbb{Z}_{n}$ or $X=\mathbb{C}^{3}/(\mathbb{Z}_{m}\times \mathbb{Z}_{n})$, the representations and some applications were studied in \cite{Jeong:2018qpc,Noshita:2021dgj,Bourgine:2019phm,Bao:2022oyn}.
\end{remark}

\subsubsection{BPS $qq$-characters}
Let $Z$ be a toric CY$_{4}$ with the form $Z=X\times \mathbb{C}$ where $X$ is a toric CY$_3$. We denote the corresponding quiver of $X$ as $Q=(Q_{0},Q_{1})$ (see section \ref{sec:CY3timesC} for the notations). As discussed in section \ref{sec:CY3timesC} and Conj.~\ref{conj:BPSqq}, we have a screening current $\mathsf{S}_{i}(x)$ and a screening charge corresponding to the $\mathbb{C}$ part of $Z$. The $qq$-character $\mathsf{T}_{j}^{K}(x)$ is obtained from the highest weight $:\frac{\mathsf{W}_{i}(x)}{\mathsf{W}_{j}(Kx)}:$ and the monomial terms are expanded in the 3d BPS crystal $\Lambda^{(j)}$. Let us study the operator products of the monomial terms with the screening current:
\bea
&:\frac{\mathsf{W}_{j}(u)}{\mathsf{W}_{j}(Ku)}\prod_{\scube\in \Lambda^{(j)}}\mathsf{A}^{-1}_{\text{c}(\scube)}(\chi_{X,u}(\cube)):\mathsf{S}_{i}(q_{4}z)\\
&=\mathsf{h}_{ji}(u,z)\left(\frac{1-K^{-1}z/u}{1-z/u}\right)^{\delta_{ij}}\prod_{\scube\in\Lambda^{(j)}}\varphi_{X,\text{c}(\scube)i}\left(\frac{q_{4}z}{\chi_{X,u}(\cube)}\right):\frac{\mathsf{W}_{j}(u)}{\mathsf{W}_{j}(Ku)}\prod_{\scube\in \Lambda^{(j)}}\mathsf{A}^{-1}_{\text{c}(\scube)}(\chi_{X,u}(\cube))\mathsf{S}_{i}(q_{4}z):,\\
&\mathsf{S}_{i}(q_{4}z):\frac{\mathsf{W}_{j}(u)}{\mathsf{W}_{j}(Ku)}\prod_{\scube\in \Lambda^{(j)}}\mathsf{A}^{-1}_{\text{c}(\scube)}(\chi_{X,u}(\cube)):\\
&=\tilde{\mathsf{h}}_{ji}(u,z)\left(\frac{1-Ku/z}{1-u/z}\right)^{\delta_{ij}}\prod_{\scube\in \Lambda^{(j)}}\varphi_{X,i\,\text{c}(\scube)}\left(\frac{\chi_{X,u}(\cube)}{z}\right)^{-1}:\frac{\mathsf{W}_{j}(u)}{\mathsf{W}_{j}(Ku)}\prod_{\scube\in \Lambda^{(j)}}\mathsf{A}^{-1}_{\text{c}(\scube)}(\chi_{X,u}(\cube))\mathsf{S}_{i}(q_{4}z):,
\eea
where $\mathsf{h}_{ji}(z,u),\tilde{\mathsf{h}}_{ji}(z,u)$ are extra factors coming from the zero-modes.

After taking the limit $q_{4}\rightarrow 1$, one can see that the contraction becomes the eigenvalue in \eqref{eq:QQTAeigenvalue} up to extra factors which do not affect the zero and pole structure:
\bea
\left(\frac{1-K^{-1}z/u}{1-z/u}\right)^{\delta_{ij}}\prod_{\scube\in\Lambda^{(j)}}\varphi_{X,\text{c}(\scube)i}\left(\frac{q_{4}z}{\chi_{X,u}(\cube)}\right)\xrightarrow{q_{4}\rightarrow 1}\left(\frac{1-K^{-1}z/u}{1-z/u}\right)^{\delta_{ij}}\prod_{\scube\in\Lambda^{(j)}}\varphi^{\text{c}(\scube)\Rightarrow i}\left(\frac{z}{\chi_{u}(\cube)}\right)
\eea
where $\chi_{X,u}(\cube)\xrightarrow{q_{4}\rightarrow 1} \chi_{u}(\cube)$. We thus have the following identification
\bea
\mathsf{S}_{i}(q_{4}z)\longrightarrow K^{\pm}_{i}(z),\quad :\frac{\mathsf{W}_{j}(u)}{\mathsf{W}_{j}(Ku)}\prod_{\scube\in \Lambda^{(j)}}\mathsf{A}^{-1}_{\text{c}(\scube)}(\chi_{X,u}(\cube)):\,\longrightarrow \ket{u,\Lambda^{(j)}}.
\eea

Therefore, we conclude that the BPS $qq$-characters associated with $Z=X\times \mathbb{C}$ are related to the vertical representations, the BPS crystal representations, of quiver quantum toroidal algebras. Actually, the terminology BPS $qq$-characters came from the fact that they are related to these BPS crystals.


\section{Semi-classical analysis and Bethe ansatz equation}\label{sec:Betheansatz}
We have discussed the gauge origami system on $\mathbb{C}^4$ with four associated equivariant parameters, $\epsilon_{1,2,3,4}$ obeying the CY$_4$ condition, $\sum_{a \in \underline{\textbf{4}}} \epsilon_a = 0$.
In the multiplicative notation, this is written as $q_1 q_2 q_3 q_4 = 1$.
It has been known that the asymptotic behavior of the partition function in the semi-classical limit, $\epsilon_a \to 0$, provides various physical information.
For example, for 4d $\mathcal{N}=2$ theory on $\mathbb{C}_{12}$, one can extract the Seiberg--Witten prepotential from the partition function in the limit $\epsilon_{1,2} \to 0$~\cite{Nekrasov:2002qd}.
Moreover, one can observe the relation between gauge theory and quantum integrable system in the limit, $\epsilon_1$ finite, $\epsilon_2 \to 0$, which is called the Nekrasov--Shatashvili limit (Bethe/Gauge correspondence)~\cite{Nekrasov:2009rc}.
In particular, the saddle point equation obtained from the gauge theory partition function in the NS limit is identified with the Bethe ansatz equation (BAE) in this context. 
In this section, we explore the NS-type limit of the gauge origami system and examine the corresponding BAE.

In general, the partition function defined on $\times_{a \in S} \mathbb{C}_a$ ($S \subset \underline{\textbf{4}}$) behaves in the semi-classical limit as follows,
\begin{equation}
    \mathcal{Z}_S \ \xrightarrow{\epsilon_{\forall a} \to 0} \ \exp \left( \frac{1}{\epsilon_S} \mathcal{F}_S + \cdots \right)
    \, , \qquad 
    \epsilon_S = \prod_{a \in S} \epsilon_a
    \, ,
\end{equation}
where $\mathcal{F}_S$ is identified with the ``prepotential'' associated with gauge theory on $\times_{a \in S} \mathbb{C}_a$.
We may also consider a partial semi-classical limit to obtain the following,
\begin{equation}
    \mathcal{Z}_S \ \xrightarrow{\epsilon_{a} \to 0} \ \exp \left( \frac{1}{\epsilon_a} \widetilde{\mathcal{W}}_a + \cdots \right)
    \, ,
\end{equation}
where $\widetilde{\mathcal{W}}_a (= \mathcal{F}_a)$ is the twisted superpotential associated with 2d (resp. 3d) gauge theory on $\mathbb{C}_a$ ($\mathbb{C}_a \times \mathbb{S}^1$).
In the context of Bethe/Gauge correspondence, this twisted superpotential is identified with the Yang--Yang functional, and the corresponding critical point condition (twisted F-term condition) provides the BAE,
\begin{equation}
    \exp \left( \frac{\partial \widetilde{\mathcal{W}}_a}{\partial x} \right) = 1
    \, .
\end{equation}
In fact, this is immediately obtained from the behavior of the partition function under the instanton adding/removing operation presented in this paper.
Applying this process to the gauge origami system on $\mathbb{C}^4$, we obtain the following BAE in the NS-type limit.

\begin{theorem}\label{thm:generalBetheansatz}%
Let $\underline{\mathbf{3}} = \{1,2,3\}$.
We denote $q_{a} \to \mathsf{q}_{a}$ ($a \in \underline{\mathbf{3}}$) in the limit $q_4 \to 1$ obeying $\mathsf{q}_1 \mathsf{q}_2 \mathsf{q}_3 = 1$. Consider a gauge origami system having the support on the $\mathbb{C}_4$-plane, where D-branes wrap the subspace $\mathcal{S}\times \mathbb{S}^{1}$, where $\mathcal{S} = \tilde{\mathcal{S}} \times \mathbb{C}_4$ with $\tilde{\mathcal{S}} \subsetneq \underline{\mathbf{3}}$ is generally 
\bea
\mathcal{S}=n\mathbb{C}_{4}+(L_{1}\mathbb{C}^{2}_{14}+M_{1}\mathbb{C}^{2}_{24}+N_{1}\mathbb{C}^{2}_{34})+(L_{2}\mathbb{C}^{3}_{234}+M_{2}\mathbb{C}^{3}_{134}+N_{2}\mathbb{C}^{3}_{124}).
\eea
The corresponding BAE of this gauge origami system obtained in the limit $q_4 \to 1$ is
\beq\label{eq:generalBetheansatz}
        1=-\mathfrak{q}\frac{\mathsf{Q}_{\mathcal{S}}(\mathsf{q}_{1}^{-1}x)\mathsf{Q}_{\mathcal{S}}(\mathsf{q}_{2}^{-1}x)\mathsf{Q}_{\mathcal{S}}(\mathsf{q}_{3}^{-1}x)}{\mathsf{Q}_{\mathcal{S}}(\mathsf{q}_{1}x)\mathsf{Q}_{\mathcal{S}}(\mathsf{q}_{2}x)\mathsf{Q}_{\mathcal{S}}(\mathsf{q}_{3}x)},
\eeq
    where the corresponding $\mathsf{Q}$-functions for the D2/D4/D6 system are given by \eqref{eq:Q-func_D2}, \eqref{eq:Q-func_D4}, \eqref{eq:Q-func_D6}, \eqref{eq:D4foldedBethe}, \eqref{eq:D6foldedBethe}.
\end{theorem}
We provide the proof of this Theorem for each system in the following.

\subsection{D4 system}
We start with the D4 system, which realizes 5d $\mathcal{N}=1$ gauge theory compactified on a circle with the adjoint matter, i.e., 5d $\mathcal{N} = 1^*$ theory.
The corresponding BAE was obtained in \cite{Nekrasov:2009rc,Chen:2012we}.
We here rederive the BAE from our analysis presented in Appendix~\ref{app:U1partition}.

Consider 5d U(1) gauge theories defined on $\mathbb{C}^{2}_{a4}\times \mathbb{S}^{1}$, where $a \in \underline{\mathbf{3}}$. See section~\ref{sec:BAE_general} for $\U(n)$ case. Namely, the gauge theories share a subspace $\mathbb{C}_{4}\times \mathbb{S}^{1}$. Under generic $\Omega$-background, the structure functions $\mathscr{Y}^{a4}_{\lambda,v}(x)\,\,(a\in\underline{\mathbf{3}})$ can be written using the $(1,1)$-type description in section~\ref{sec:decomp5d} as 
\begin{equation}
    \mathscr{Y}^{a4}_{\lambda,v}(x)=\prod_{i=1}^{\infty}\frac{1-x^{(a)}_{v;i}/x}{1-q_{a}x^{(a)}_{v;i}/x}=\prod_{i=1}^{\infty}\mathscr{V}_{a}\left(\frac{x_{v;i}^{(a)}}{x}\right)^{-1}
\end{equation}
where 
\begin{equation}
    x^{(a)}_{v;i}=vq_{a}^{i-1}q_{4}^{\lambda_{i}}.
\end{equation}
Introducing the $\mathsf{Q}$-functions~\cite{Nekrasov:2013xda} under generic instanton background
\begin{equation}
    \mathsf{Q}_{a4}(x)=\prod_{i=1}^{\infty}\left(1-x^{(a)}_{v;i}/x\right)
    \label{eq:Q-func_D4}
\end{equation}
we may rewrite the $\mathscr{Y}$-functions as 
\begin{equation}
    \mathscr{Y}^{a4}_{\lambda,v}(x)=\frac{\mathsf{Q}_{a4}(x)}{\mathsf{Q}_{a4}(q_{a}^{-1}x)}.
\end{equation}

The $\D4$ partition function behaves as Thm.~\ref{app:thm-D4recursion} and \eqref{eq:app-D4U1recursionformula}:
\begin{equation}
\begin{split}
\frac{\mathfrak{q}^{|\lambda+\Abox|}\widetilde{\mathcal{Z}}^{\D4}_{a4}[\lambda+\Bbox]}{\mathfrak{q}^{|\lambda|}\widetilde{\mathcal{Z}}_{a4}^{\D4}[\lambda]}&=-\mathfrak{q}\frac{\mathscr{Y}_{\lambda,v}^{a4}(q_{b}^{-1}\chi_{a4,v}(\Bbox))\mathscr{Y}^{a4}_{\lambda+\Abox,v}(q_{c}^{-1}\chi_{a4,v}(\Bbox))}{\mathscr{Y}^{a4}_{\lambda,v}(\chi_{a4,v}(\Bbox))\mathscr{Y}^{a4}_{\lambda+\Abox,v}(q_{bc}^{-1}\chi_{a4,v}(\Bbox))}\\
    &=-\mathfrak{q}\frac{\mathscr{S}_{a4}\left(q_{c4}\chi_{a4,v}(\Bbox)/x\right)}{\mathscr{S}_{a4}(q_{bc4}\chi_{a4,v}(\Bbox)/x)}\frac{\mathscr{Y}_{\lambda,v}^{a4}(q_{ac}x)\mathscr{Y}^{a4}_{\lambda,v}(q_{ab}x)}{\mathscr{Y}^{a4}_{\lambda,v}(q_{abc}x)\mathscr{Y}^{a4}_{\lambda,v}(q_{a}x)}
\end{split}
\end{equation}
where $a,b,c\in\underline{\mathbf{3}}$ and $x=q_{4}\chi_{a4,v}(\Bbox)$. In the NS limit $q_{4}\rightarrow 1$, the left hand side becomes $\exp\left(\partial\widetilde{\mathcal{W}}_{4}/\partial x\right)$ and $\mathscr{S}_{a4}(x)\rightarrow 1$. Hence, the twisted F-term condition reads
\begin{equation}
    1=-\mathfrak{q}\frac{\mathscr{Y}^{a4}_{\lambda,v}(q_{ac}x)\mathscr{Y}^{a4}_{\lambda,v}(q_{ab}x)}{\mathscr{Y}^{a4}_{\lambda,v}(q_{abc}x)\mathscr{Y}^{a4}_{\lambda,v}(q_{a}x)}.
\end{equation}
 Denoting the solution to this twisted F-term condition as $\lambda_{\ast}$ (critical configuration; limit shape), we can rewrite the $\mathscr{Y}$-functions in terms of the $\mathsf{Q}$-functions, where the $x$-variables will take values $x_{v;i}^{(a)}|_{\lambda=\lambda_{\ast}}$. The BAE is then written as 
\begin{equation} \label{eq:Bethe_eq_D4}
    1=-\mathfrak{q}\frac{\mathsf{Q}_{a4}(\mathsf{q}_{a}^{-1}x)\mathsf{Q}_{a4}(\mathsf{q}_{b}^{-1}x)\mathsf{Q}_{a4}(\mathsf{q}_{c}^{-1}x)}{\mathsf{Q}_{a4}(\mathsf{q}_{a}x)\mathsf{Q}_{a4}(\mathsf{q}_{b}x)\mathsf{Q}_{a4}(\mathsf{q}_{c}x)},\quad \{a,b,c\}=\three
\end{equation}
where $q_{a,b,c}\rightarrow \mathsf{q}_{a,b,c}$ under $q_{4}\rightarrow 1$ (see also section \ref{sec:QTgl1}). Using the structure function of the quantum toroidal $\mathfrak{gl}_{1}$ in \eqref{eq:gl1structurefunction}, we can rewrite this BAE as
\begin{equation}\label{eq:D4QTgl1Bethe}
    1=-\mathfrak{q}\prod_{i=1}^{\infty}\mathsf{g}\left(x_{v;i}^{(a)}/x\right).
\end{equation}
Note that we have the condition $\mathsf{q}_{a}\mathsf{q}_{b}\mathsf{q}_{c}=1$. This is the BAE for the elliptic Calogero--Moser/Ruijsenaars--Schneider model, which is the integrable system associated with $\widehat{A}_{0}$ quiver gauge theory~\cite{Donagi:1995cf,DHoker:1997hut,Nekrasov:2012xe}. Note that in the quantum integrable model side, the parameter $\mathsf{q}_{a}$ corresponds to the Planck constant, while one of the parameters $\mathsf{q}_{b},\mathsf{q}_{c}$ correspond to the mass parameter.
The instanton counting parameter $\mathfrak{q}$ becomes the elliptic parameter of the potential function. The weak coupling limit corresponds to the trigonometric limit.
See, for example,~\cite{Kimura:2022zsx}.

\subsection{D6 system}
We then consider the D6 system and derive the corresponding BAE. We consider 7d U(1) gauge theories defined on $\mathbb{C}^{3}_{ab4}\times \mathbb{S}^{1}\,(ab\in\{12,23,13\})$. Under generic $\Omega$-background, the structure functions $\mathscr{W}^{ab4}_{\pi,v}(x)\,(ab\in\{12,13,23\})$ associated with the plane partition configuration $\pi$ can be written using the $(2,1)$-type description in section~\ref{sec:decomp7d} as 
\begin{equation}
    \mathscr{W}^{ab4}_{\pi,v}(x)=\prod_{i,j=1}^{\infty}\frac{(1-x^{(ab)}_{v;i,j}/x)(1-q_{ab}x^{(ab)}_{v;i,j}/x)}{(1-q_{a}x^{(ab)}_{v;i,j}/x)(1-q_{b}x^{(ab)}_{v;i,j}/x)}=\prod_{i,j=1}^{\infty}\mathscr{S}_{ab}\left(\frac{x_{v;i,j}^{(ab)}}{x}\right)^{-1}
\end{equation}
where we introduced 
\begin{equation}
    x^{(ab)}_{v;i,j}=vq_{a}^{i-1}q_{b}^{j-1}q_{4}^{\pi_{ij}}.
\end{equation}
Similar to the D4-case, we introduce the $\mathsf{Q}$-functions as
\begin{equation}
    \mathsf{Q}_{ab4}(x)=\prod_{i,j=1}^{\infty}\left(1-x^{(ab)}_{v;i,j}/x\right)
    \label{eq:Q-func_D6}
\end{equation}
and then the $\mathscr{W}$-functions are rewritten as 
\begin{equation}
    \mathscr{W}^{ab4}_{\pi,v}(x)=\frac{\mathsf{Q}_{ab4}(x)\mathsf{Q}_{ab4}(q_{ab}^{-1}x)}{\mathsf{Q}_{ab4}(q_{a}^{-1}x)\mathsf{Q}_{ab4}(q_{b}^{-1}x)}.
\end{equation}
The recursion formula of the $\D6$-partition function comes from \eqref{eq:app-D6U1recursionformula} and Thm.~\ref{eq:app-thm-D6U1recursionformula}:
\begin{equation}
\begin{split}
    \frac{\mathfrak{q}^{|\pi+\cube|}\widetilde{\mathcal{Z}}^{\D6}_{ab4}[\pi+\cube]}{\mathfrak{q}^{|\pi|}\widetilde{\mathcal{Z}}^{\D6}_{ab4}[\pi]}&=-\mathfrak{q}\frac{\mathscr{W}^{ab4}_{\pi+\cube,v}(q_{ab4}\chi_{ab4,v}(\cube))}{\mathscr{W}^{ab4}_{\pi,v}(\chi_{ab4,v}(\cube))}\\
    &=-\mathfrak{q}\,g_{ab4}\left(\frac{\chi_{ab4,v}(\cube)}{q_{ab}x}\right)\frac{\mathscr{W}^{ab4}_{\pi,v}(q_{ab}x)}{\mathscr{W}^{ab4}_{\pi,v}(q_{4}^{-1}x)}
\end{split}
\end{equation}
where $x=q_{4}\chi_{ab4,v}(\cube)$. Taking the NS limit $q_{4}\rightarrow 1$ and using $g_{ab4}(x)\xrightarrow{q_4 \to 1} 1$, we obtain the twisted F-term condition as
\begin{equation}
    1=-\mathfrak{q}\frac{\mathscr{W}^{ab4}_{\pi,v}(q_{ab}x)}{\mathscr{W}^{ab4}_{\pi,v}(q_{4}^{-1}x)}.
\end{equation}
Denoting the critical configuration $\pi_{\ast}$ and using the $\mathsf{Q}$-functions, we have the following BAE:
\begin{equation} \label{eq:Bethe_eq_D6}
    1=-\mathfrak{q}\frac{\mathsf{Q}_{ab4}(\mathsf{q}_{a}^{-1}x)\mathsf{Q}_{ab4}(\mathsf{q}_{b}^{-1}x)\mathsf{Q}_{ab4}(\mathsf{q}_{c}^{-1}x)}{\mathsf{Q}_{ab4}(\mathsf{q}_{a}x)\mathsf{Q}_{ab4}(\mathsf{q}_{b}x)\mathsf{Q}_{ab4}(\mathsf{q}_{c}x)},\quad a,b,c\in\underline{\mathbf{3}}
\end{equation}
where after taking the limit $q_{4}\rightarrow 1$, we used $\mathsf{q}_{1,2,3}$ with the condition $\mathsf{q}_{1}\mathsf{q}_{2}\mathsf{q}_{3}=1$. Using the structure function of quantum toroidal $\mathfrak{gl}_{1}$ in \eqref{eq:gl1structurefunction}, this is rewritten as 
\begin{equation}\label{eq:D6QTgl1Bethe}
    1=-\mathfrak{q}\prod_{i,j=1}^{\infty}\mathsf{g}\left(x^{(ab)}_{v;i,j}/x\right).
\end{equation}

\subsection{D2 system}\label{sec:D2_BetheEq}
Let us consider next the $\D2$ system and derive the corresponding BAE. We consider the 3d gauge theories on $\mathbb{C}_{4}\times \mathbb{S}^{1}$. First, let us consider the $\U(1)$ gauge theory. The $\mathscr{U}$-function is written as \eqref{eq:D2Ufunctiondef}
\beq
\mathscr{U}^{4}_{k,v}(x)=\left(1-x_{v}/x\right),\quad x_{v}=vq_{4}^{k}.
\eeq
We can define the $\mathsf{Q}$-function with the $\U(1)$ instanton (vortex) background as 
\bea
\mathsf{Q}_4(x)=(1-x_{v}/x).
\eea
We can do the same procedure we have done in previous sections, by studying the recursion relation of the $\U(1)$ partition function. However, for the $\U(1)$ case, the partition function is actually trivial (see Appendix~\ref{app:D2U1partitionfunction}, \eqref{eq:D2U1partition1}, \eqref{eq:D2U1partition2}). Hence, we will obtain a trivial BAE for such a case. 

To obtain a nontrivial BAE, we need to consider the $\U(n)$ gauge theory on $\mathbb{C}_{4}\times \mathbb{S}^{1}$. The $\mathsf{Q}$-function is generalized to 
\bea
\mathsf{Q}_4(x)=\prod_{\alpha=1}^{n}\left(1-x_{v,\alpha}/x\right),\quad x_{v_{\alpha}}=v_{\alpha}q_{4}^{k_{\alpha}}.
\label{eq:Q-func_D2}
\eea
The recursion formula of the $\D2$ $\U(n)$ partition function comes from \eqref{eq:app-D2Unrecursion}:
\bea
\frac{\mathfrak{q}^{|\vec{k}+\Abox|}\mathcal{Z}^{\D2}_{4}[\vec{v},\vec{k}+\Bbox]}{\mathfrak{q}^{|\vec{k}|}\mathcal{Z}_{4}^{\D2}[\vec{v},\vec{k}]}&=\prod_{\beta\neq \alpha}\frac{\mathscr{U}^{4}_{k_{\beta},v_{\beta}}(x)\mathscr{U}^{4}_{k_{\beta},v_{\beta}}(q_{ij}x)\mathscr{U}^{4}_{k_{\beta},v_{\beta}}(q_{i4}^{-1}x)\mathscr{U}^{4}_{k_{\beta},v_{\beta}}(q_{j4}^{-1}x)}{\mathscr{U}^{4}_{k_{\beta},v_{\beta}}(q_{4}^{-1}x)\mathscr{U}^{4}_{k_{\beta},v_{\beta}}(q_{ij4}^{-1})\mathscr{U}^{4}_{k_{\beta},v_{\beta}}(q_{i}x)\mathscr{U}^{4}_{k_{\beta},v_{\beta}}(q_{j}x)}\\
&\times -\frac{\mathscr{U}^{4}_{k_{\alpha},v_{\alpha}}(q_{i4}^{-1}x)}{\mathscr{U}^{4}_{k_{\alpha}+\Abox,v_{\alpha}}(q_{i}x)}\frac{\mathscr{U}^{4}_{k_{\alpha},v_{\alpha}}(q_{j4}^{-1})}{\mathscr{U}^{4}_{k_{\alpha}+\Abox,v_{\alpha}}(q_{j}x)}\frac{\mathscr{U}^{4}_{k_{\alpha}+\Abox,v_{\alpha}}(q_{ij}x)}{\mathscr{U}^{4}_{k_{\alpha},v_{\alpha}}(q_{ij4}^{-1})}\frac{\mathscr{U}^{4}_{k_{\alpha}+\Abox,v_{\alpha}}(x)}{\mathscr{U}^{4}_{k_{\alpha},v_{\alpha}}(q_{4}^{-1}x)}
\eea
where $x=q_{4}\chi_{4,v_{\alpha}}(\Bbox)$. Noticing that $\mathscr{U}^{4}_{k+\Bbox,v}(x)/\mathscr{U}^{4}_{k,v}(x)=\mathscr{V}_{4}(\chi_{4,v}(\Bbox)/x)$ and that we have $\mathscr{V}_{4}(x)\rightarrow 1$ under the limit $q_{4}\rightarrow 1$, we obtain the following BAE:
\bea
1=-\mathfrak{q}\frac{\mathsf{Q}_4(\mathsf{q}_{1}^{-1}x)\mathsf{Q}_4(\mathsf{q}_{2}^{-1}x)\mathsf{Q}_4(\mathsf{q}_{3}^{-1}x)}{\mathsf{Q}_4(\mathsf{q}_{1}x)\mathsf{Q}_4(\mathsf{q}_{2}x)\mathsf{Q}_4(\mathsf{q}_{3}x)}.
\eea
Using the structure function of the quantum toroidal $\mathfrak{gl}_{1}$, we have 
\begin{equation}\label{eq:D2QTgl1Bethe}
    1=-\mathfrak{q}\prod_{\alpha=1}^{n}\mathsf{g}\left(x_{v_{\alpha}}/x\right).
\end{equation}
We note that this BAE appeared in \cite{Chen:2012we} from a 3d $\mathcal{N} = 2$ theory with 3 adjoint chiral matters with the twisted mass $\epsilon_{1,2,3}$, where the gauge group rank corresponds to the number of Bethe roots.

\paragraph{Flavor branes}
Based on previous sections, we expect that the vortex partition function on $\mathbb{C}_{a}\times \mathbb{S}^{1}$ including a flavor D8-$\overline{\D8}$ brane is obtained from the operator product
\bea
\bra{0}\widetilde{\mathsf{Z}}^{K}(\mu)\prod_{\alpha=1}^{n_{a}}\mathsf{S}_{a}(v_{a,\alpha}q_{a}^{k_{a}^{(\alpha)}})\ket{0}.
\eea
The non-perturbative partition function is then modified to 
\bea
\mathcal{Z}^{\D2}_{\text{vor.}}\rightarrow \prod_{\alpha=1}^{n_{a}}\prod_{\Abox\in k_{a}^{(\alpha)}}\left(K\frac{1-K^{-1}\chi_{a,v_{a,\alpha}}(\Bbox)/\mu}{1-\chi_{a,v_{a,\alpha}}(\Bbox)/\mu}\right)\times \mathcal{Z}^{\D2}_{\text{vor.}}
\eea
We may interpret that the $\D8$ and $\overline{\D8}$-branes give the fundamental chiral and anti-chiral multiplets. Since the D2 theory we are considering is a reduction from a 16 SUSY theory, it is natural to consider the same number of chiral and anti-chiral multiplets, that form a hypermultiplet.

Using this, the recursion formula is modified to 
\bea
\frac{\mathfrak{q}^{|\vec{k}+\Abox|}\mathcal{Z}^{\D2}_{4}[\vec{v},\vec{k}+\Bbox]}{\mathfrak{q}^{|\vec{k}|}\mathcal{Z}_{4}^{\D2}[\vec{v},\vec{k}]}&=\prod_{\beta\neq \alpha}\frac{\mathscr{U}^{4}_{k_{\beta},v_{\beta}}(x)\mathscr{U}^{4}_{k_{\beta},v_{\beta}}(q_{ij}x)\mathscr{U}^{4}_{k_{\beta},v_{\beta}}(q_{i4}^{-1}x)\mathscr{U}^{4}_{k_{\beta},v_{\beta}}(q_{j4}^{-1}x)}{\mathscr{U}^{4}_{k_{\beta},v_{\beta}}(q_{4}^{-1}x)\mathscr{U}^{4}_{k_{\beta},v_{\beta}}(q_{ij4}^{-1})\mathscr{U}^{4}_{k_{\beta},v_{\beta}}(q_{i}x)\mathscr{U}^{4}_{k_{\beta},v_{\beta}}(q_{j}x)}\\
&\times -\frac{\mathscr{U}^{4}_{k_{\alpha},v_{\alpha}}(q_{i4}^{-1}x)}{\mathscr{U}^{4}_{k_{\alpha}+\Abox,v_{\alpha}}(q_{i}x)}\frac{\mathscr{U}^{4}_{k_{\alpha},v_{\alpha}}(q_{j4}^{-1})}{\mathscr{U}^{4}_{k_{\alpha}+\Abox,v_{\alpha}}(q_{j}x)}\frac{\mathscr{U}^{4}_{k_{\alpha}+\Abox,v_{\alpha}}(q_{ij}x)}{\mathscr{U}^{4}_{k_{\alpha},v_{\alpha}}(q_{ij4}^{-1})}\frac{\mathscr{U}^{4}_{k_{\alpha}+\Abox,v_{\alpha}}(x)}{\mathscr{U}^{4}_{k_{\alpha},v_{\alpha}}(q_{4}^{-1}x)}\\
&\times \frac{1-K\mu/\chi_{4,v_{\alpha}}(\Bbox)}{1-\mu/\chi_{4,v_{\alpha}}(\Bbox)}.
\eea
After taking the NS limit $q_{4}\rightarrow 1$, we have the BAE involving the additional polynomials, 
\bea\label{eq:MacMahon1Bethe}
1=-\mathfrak{q}\frac{1-\widetilde{K}\mu/x}{1-\mu/x}\frac{\mathsf{Q}_4(\mathsf{q}_{1}^{-1}x)\mathsf{Q}_4(\mathsf{q}_{2}^{-1}x)\mathsf{Q}_4(\mathsf{q}_{3}^{-1}x)}{\mathsf{Q}_4(\mathsf{q}_{1}x)\mathsf{Q}_4(\mathsf{q}_{2}x)\mathsf{Q}_4(\mathsf{q}_{3}x)},
\eea
where we used $K\rightarrow \widetilde{K}$ at the limit\footnote{We use different symbols for the parameters $K$ and $\widetilde{K}$, because for example when $K=q_{1,2,3}^{\pm1}$, we have $\widetilde{K}=\sfq_{1,2,3}^{\pm1}$ after taking the limit. When $K$ is generic, they are identified as $K=\widetilde{K}$.} $q_{4}\rightarrow 1$. This corresponds to the trigonometric version of the expression shown in \cite[Prop. 7.13]{Cao:2023lon}, where they obtained the BAE associated with the MacMahon representation.

After specializing the parameter as $\widetilde{K}=\mathsf{q}_{1},\mathsf{q}_{2},\mathsf{q}_{3},\sfq_{1}^{-1},\sfq_{2}^{-1},\sfq_{3}^{-1}$, we obtain a physical setup where the $\D6_{\bar{i}}\,(i=1,2,3)$ branes or $\overline{\D6}_{\bar{i}}\,(i=1,2,3)$ branes will play the role of a flavor branes:
\bea\label{eq:FockBethe}
1&=-\mathfrak{q}\frac{1-\mathsf{q}_{i}\mu/x}{1-\mu/x}\frac{\mathsf{Q}_4(\mathsf{q}_{1}^{-1}x)\mathsf{Q}_4(\mathsf{q}_{2}^{-1}x)\mathsf{Q}_4(\mathsf{q}_{3}^{-1}x)}{\mathsf{Q}_4(\mathsf{q}_{1}x)\mathsf{Q}_4(\mathsf{q}_{2}x)\mathsf{Q}_4(\mathsf{q}_{3}x)},\\
1&=-\mathfrak{q}\frac{1-\mathsf{q}_{i}^{-1}\mu/x}{1-\mu/x}\frac{\mathsf{Q}_4(\mathsf{q}_{1}^{-1}x)\mathsf{Q}_4(\mathsf{q}_{2}^{-1}x)\mathsf{Q}_4(\mathsf{q}_{3}^{-1}x)}{\mathsf{Q}_4(\mathsf{q}_{1}x)\mathsf{Q}_4(\mathsf{q}_{2}x)\mathsf{Q}_4(\mathsf{q}_{3}x)}.
\eea
We may also add multiple flavor branes. Generally, the BAE is modified to 
\bea\label{eq:MacMahonNBethe}
1=-\mathfrak{q}\prod_{\alpha=1}^{n_\text{F}}\frac{1-\widetilde{K}_{\alpha}\mu_{\alpha}/x}{1-\mu_{\alpha}/x}\frac{\mathsf{Q}_4(\mathsf{q}_{1}^{-1}x)\mathsf{Q}_4(\mathsf{q}_{2}^{-1}x)\mathsf{Q}_4(\mathsf{q}_{3}^{-1}x)}{\mathsf{Q}_4(\mathsf{q}_{1}x)\mathsf{Q}_4(\mathsf{q}_{2}x)\mathsf{Q}_4(\mathsf{q}_{3}x)}.
\eea
Specializing the parameters as $\{K_{\alpha}\}\rightarrow \{q_{1},q_{2},q_{3}\}$ (which means $\{\widetilde{K}_{\alpha}\}\rightarrow \{\sfq_{1},\sfq_{2},\sfq_{3}\}$) give
\bea\label{eq:FockLMN}
1&=-\mathfrak{q}\prod_{\alpha=1}^{L}\frac{1-\mathsf{q}_{1}u_{\alpha}/x}{1-u_{\alpha}/x}\prod_{\beta=1}^{M}\frac{1-\mathsf{q}_{2}v_{\beta}/x}{1-v_{\beta}/x}\prod_{\gamma=1}^{N}\frac{1-\mathsf{q}_{3}w_{\gamma}/x}{1-w_{\gamma}/x}\\
&\qquad\qquad  \times\frac{\mathsf{Q}_4(\mathsf{q}_{1}^{-1}x)\mathsf{Q}_4(\mathsf{q}_{2}^{-1}x)\mathsf{Q}_4(\mathsf{q}_{3}^{-1}x)}{\mathsf{Q}_4(\mathsf{q}_{1}x)\mathsf{Q}_4(\mathsf{q}_{2}x)\mathsf{Q}_4(\mathsf{q}_{3}x)}
\eea
where we relabeled $\{\mu_{\alpha}\}_{\alpha=1}^{n_\text{F}=L+M+N}\rightarrow \{u_{\alpha}\}_{\alpha=1}^{L}\cup \{v_{\beta}\}_{\beta=1}^{M}\cup\{w_{\gamma}\}_{\gamma=1}^{N} $. The most general case is 
\bea\label{eq:Fockgeneral}
1&=-\mathfrak{q}\prod_{\alpha=1}^{L_{+}}\frac{1-\mathsf{q}_{1}u^{+}_{\alpha}/x}{1-u^{+}_{\alpha}/x}\prod_{\alpha=1}^{L_{-}}\frac{1-\mathsf{q}_{1}^{-1}u^{-}_{\alpha}/x}{1-u^{-}_{\alpha}/x}\prod_{\beta=1}^{M_{+}}\frac{1-\mathsf{q}_{2}v^{+}_{\beta}/x}{1-v^{+}_{\beta}/x}\prod_{\beta=1}^{M_{-}}\frac{1-\mathsf{q}_{2}^{-1}v^{-}_{\beta}/x}{1-v^{-}_{\beta}/x}\\
&\qquad \times \prod_{\gamma=1}^{N_{+}}\frac{1-\mathsf{q}_{3}w^{+}_{\gamma}/x}{1-w^{+}_{\gamma}/x}\prod_{\gamma=1}^{N_{-}}\frac{1-\mathsf{q}_{3}^{-1}w^{-}_{\gamma}/x}{1-w^{-}_{\gamma}/x}\frac{\mathsf{Q}_4(\mathsf{q}_{1}^{-1}x)\mathsf{Q}_4(\mathsf{q}_{2}^{-1}x)\mathsf{Q}_4(\mathsf{q}_{3}^{-1}x)}{\mathsf{Q}_4(\mathsf{q}_{1}x)\mathsf{Q}_4(\mathsf{q}_{2}x)\mathsf{Q}_4(\mathsf{q}_{3}x)}
\eea
which comes from $\{K_{a}\}\rightarrow \{q_{1}^{\pm1},q_{2}^{\pm1},q_{3}^{\pm1}\}$ and $\{\mu_{\alpha}\}_{\alpha=1}^{n_{\text{F}}}\rightarrow \{u_{\alpha}^{\pm}\}_{\alpha=1}^{L_{\pm}}\cup\{v_{\beta}^{\pm}\}_{\beta=1}^{M_{\pm}}\cup\{w_{\gamma}^{\pm}\}_{\gamma=1}^{N_{\pm}}$.

This is a gauge theoretic derivation and generalizations of the BAEs of \cite[eq.~(5.5), (6.1)]{Feigin:2015raa}, \cite{Litvinov:2013zda}. The equations \eqref{eq:MacMahon1Bethe} and \eqref{eq:MacMahonNBethe} are the BAEs of MacMahon representation and its tensor product representations. Other equations \eqref{eq:FockBethe}, \eqref{eq:FockLMN}, \eqref{eq:Fockgeneral} correspond to the BAEs of the representations $\mathcal{F}_{i}\,\,(\overline{\mathcal{F}}_{i})$, $\mathcal{F}_{1}^{\otimes L}\otimes \mathcal{F}_{2}^{\otimes M}\otimes \mathcal{F}_{3}^{\otimes N}$, and $\mathcal{F}_{1}^{\otimes L_{+}}\otimes\overbar{\mathcal{F}}_{1}^{\otimes L_{-}}\otimes  \mathcal{F}_{2}^{\otimes M_{+}}\otimes \overline{\mathcal{F}}_{2}^{\otimes M_{-}} \otimes \mathcal{F}_{3}^{\otimes N_{+}}\otimes \overline{\mathcal{F}}_{3}^{\otimes N_{-}}$, respectively, where we denoted the representation with negative central charges as $\overline{\mathcal{F}}_{i}$. The prefactors are determined from the representations.

Another specialization is to take the limit ${K}_{\alpha}\rightarrow 0,\infty$, which gives the following BAE:
\bea
1=-\mathfrak{q}\prod_{\alpha=1}^{n_\text{F}}\frac{1}{1-\mu_{\alpha}/x}\frac{\mathsf{Q}_4(\mathsf{q}_{1}^{-1}x)\mathsf{Q}_4(\mathsf{q}_{2}^{-1}x)\mathsf{Q}_4(\mathsf{q}_{3}^{-1}x)}{\mathsf{Q}_4(\mathsf{q}_{1}x)\mathsf{Q}_4(\mathsf{q}_{2}x)\mathsf{Q}_4(\mathsf{q}_{3}x)}.
\eea
This corresponds to the setup where the $\overline{\D8}$-brane is decoupled from the setup. We expect the underlying quantum algebra to be the shifted quantum toroidal algebras for these cases.

To summarize, we have the following BAE from the D2 system with additional flavor branes.
\begin{theorem}\label{thm:BAEwithPolynom}
In the NS limit of the D2 system with flavor, we obtain the BAE involving the additional polynomials $a(x)$ and $d(x)$ specifying the representation of quantum toroidal $\mathfrak{gl}_1$,
\bea
1=-\mathfrak{q}\frac{a(x)}{d(x)}\frac{\mathsf{Q}_4(\mathsf{q}_{1}^{-1}x)\mathsf{Q}_4(\mathsf{q}_{2}^{-1}x)\mathsf{Q}_4(\mathsf{q}_{3}^{-1}x)}{\mathsf{Q}_4(\mathsf{q}_{1}x)\mathsf{Q}_4(\mathsf{q}_{2}x)\mathsf{Q}_4(\mathsf{q}_{3}x)}.
\eea
\end{theorem}
This is analogous to Yangian and quantum affine algebra, where the representation data appears only in the $a$ and $d$ polynomials.

\subsection{General case}\label{sec:BAE_general}
Generalizations to $\U(n)$ theories of the $\D4$ and $\D6$ theories can be done similarly to the $\D2$-case. The BAE will have a general structure as in \eqref{eq:generalBetheansatz} of Thm.~\ref{thm:generalBetheansatz}. The explicit form of the $\mathsf{Q}$-functions is only modified. For the 5d $\U(n)$ theory on $\mathbb{C}_{a4}^{2}\times \mathbb{S}^{1}$, equation~\eqref{eq:Q-func_D4} is modified to
\bea
\mathsf{Q}_{a4}(x)=\prod_{i=1}^{\infty}\prod_{\alpha=1}^{n}\left(1-x^{(a)}_{\alpha,i}/x\right),\quad x_{\alpha,i}^{(a)}=v_{\alpha}q_{a}^{i-1}q_{4}^{\lambda_{i}^{(\alpha)}}
\eea
We can also consider more general configurations where folded instantons also appear. Consider a gauge origami system where the D4-branes wrap the subspace $\mathcal{S}\times \mathbb{S}^{1}$, where
\beq
\mathcal{S}=L\mathbb{C}^{2}_{14}+M\mathbb{C}^{2}_{24}+N\mathbb{C}^{2}_{34}. 
\eeq
Namely, we have $\U(L),\U(M),\U(N)$ affine quiver gauge theories on $\mathbb{C}^{2}_{14}\times \mathbb{S}^{1},\mathbb{C}^{2}_{24}\times \mathbb{S}^{1},\mathbb{C}^{2}_{34}\times \mathbb{S}^{1}$ respectively, with the folded instantons configurations connecting the different gauge theories. The $\mathsf{Q}$-function is generalized to 
\bea\label{eq:D4foldedBethe}
\mathsf{Q}_{\mathcal{S}}(x)=\prod_{i=1}^{\infty}\prod_{\alpha=1}^{L}(1-x_{\alpha,i}^{(1)}/x)\prod_{j=1}^{\infty}\prod_{\beta=1}^{M}(1-x_{\beta,i}^{(2)}/x)\prod_{k=1}^{\infty}\prod_{\gamma=1}^{N}(1-x_{\gamma,k}^{(3)}/x).
\eea

For a 7d $\U(n)$ theory on $\mathbb{C}^{3}_{ab4}\times\mathbb{S}^{1}$, equation \eqref{eq:Q-func_D6} is modified to
\bea
\mathsf{Q}_{ab4}(x)=\prod_{i,j=1}^{\infty}\prod_{\alpha=1}^{n}\left(1-x_{\alpha;i,j}^{(ab)}/x\right),\quad x_{\alpha;i,j}^{(ab)}=v_{\alpha}q_{a}^{i-1}q_{b}^{j-1}q_{4}^{\pi_{ij}^{(\alpha)}}.
\eea
The folded instanton system where $\D6$-branes wrapping
\bea
\mathcal{S}=L\mathbb{C}^{3}_{234}+M\mathbb{C}^{3}_{134}+N\mathbb{C}^{3}_{124}
\eea
gives the $\mathsf{Q}$-function
\bea\label{eq:D6foldedBethe}
\mathsf{Q}_{\mathcal{S}}(x)=\prod_{i,j=1}^{\infty}\prod_{\alpha=1}^{L}(1-x_{\alpha;i,j}^{(23)}/x)\prod_{k,l=1}^{\infty}\prod_{\beta=1}^{M}(1-x_{\beta;k,l}^{(13)}/x)\prod_{m,n=1}^{\infty}\prod_{\gamma=1}^{N}(1-x_{\gamma;m,n}^{(12)}/x).
\eea

For all cases, by direct computation, one can show that the arising BAE has the form 
\bea
1=-\mathfrak{q}\frac{\mathsf{Q}_{\mathcal{S}}(\mathsf{q}_{1}^{-1}x)\mathsf{Q}_{\mathcal{S}}(\mathsf{q}_{2}^{-1}x)\mathsf{Q}_{\mathcal{S}}(\mathsf{q}_{3}^{-1}x)}{\mathsf{Q}_{\mathcal{S}}(\mathsf{q}_{1}x)\mathsf{Q}_{\mathcal{S}}(\mathsf{q}_{2}x)\mathsf{Q}_{\mathcal{S}}(\mathsf{q}_{3}x)}.
\eea
Therefore, we obtain Thm.~\ref{thm:generalBetheansatz}.

\begin{remark}
We note that when only one of $L,M,N$ is non-zero in \eqref{eq:D4foldedBethe}, the dual integrable system is the elliptic Calogero--Moser/Ruijsenaars--Schneider model. When only one of $L,M,N$ is zero, the dual integrable system is the double elliptic Calogero--Moser system (and its trigonometric version) corresponding to the Calogero system associated with superalgebras~\cite{Kerov:1998IMRN,Sergeev:2001JNMP,Sergeev:2002TMP,Sergeev:2005AM} as discussed in the context of gauge origami~\cite{Nekrasov:2017gzb,Chen:2019vvt} (see also a related paper \cite{Jeong:2021rll}). We expect that the most general case when $L,M,N\neq 0$ gives a \emph{triple} elliptic Calogero--Moser system which is an elliptic generalization of the triple Calogero--Sutherland system~\cite[eq.~(2.15)]{Gaiotto:2020dsq}. For the \eqref{eq:D6foldedBethe}, we do not know the corresponding integrable systems and they are yet to be studied.
\end{remark}

\subsection{Quantum toroidal \texorpdfstring{$\mathfrak{gl}_1$}{gl(1)} and \texorpdfstring{$q$}{q}-characters}

In general, the $qq$-character is reduced to the $q$-character in the NS limit, which can be identified with the T-operator of the corresponding quantum integrable system~\cite{Knight:1995JA,Frenkel:1998ojj}.
We have discussed in section~\ref{sec:toroidal_alg} that the $qq$-characters of D2/D4/D6 systems correspond to the vector/Fock/MacMahon representations of quantum toroidal $\mathfrak{gl}_1$ denoted by $\mathscr{E}$, and actually they are reduced to the $q$-characters of the corresponding representations of $\mathscr{E}$ considered in~\cite{Feigin:2016pld}. The BAE obtained for $\mathscr{E}$ by~\cite{Feigin:2015raa} is based on the Fock representation, which is consistent with our analysis~\eqref{eq:Bethe_eq_D4}.
Hence, we propose that the BAE obtained from the D6 system~\eqref{eq:Bethe_eq_D6} is of the MacMahon representation of $\mathscr{E}$. 

Let us reproduce the $q$-characters of \cite{Feigin:2016pld} from a gauge theoretic perspective. Consider a gauge origami system where there are D-branes on subspaces including the $\mathbb{C}_{4}\times \mathbb{S}^{1}$:
\begin{equation}
\left(\text{pt}\,\,\cup\,\,\bigcup\limits_{a\in\underline{\mathbf{3}}}\mathbb{C}_{a}\cup \bigcup_{A\in\{12,13,23\}}\mathbb{C}^{2}_{A}\right)\times \mathbb{C}_{4}\times \mathbb{S}^{1}.
\end{equation}
Namely, we are considering a generalized gauge theory where we have D2-branes on $\mathbb{C}_{4}\times \mathbb{S}^{1}$, D4-branes on $\mathbb{C}^{2}_{a4}\times \mathbb{S}^{1}\,\,(a\in\underline{\mathbf{3}})$, and D6-branes on $\mathbb{C}^{3}_{A4}\times \mathbb{S}^{1}\,\,(A=12,13,23)$. Using \eqref{eq:D2screening}, \eqref{eq:D4screening}, and \eqref{eq:D6screening} (see also sections \ref{sec:decomp3d}, \ref{sec:decomp5d}, and \ref{sec:decomp7d}), the partition function of such theory is written using the screening current $\mathsf{S}_{4}(x)$ as
\beq
\mathcal{Z}[\vec{k},\underline{\vec{\lambda}},\underline{\vec{\pi}}]=\mathcal{Z}[\mathcal{X}]=\bra{0}\prod_{x\in\mathcal{X}}\mathsf{S}_{4}(x)\ket{0}\eqqcolon \bra{0}\ket{\mathcal{Z}}
\eeq
where 
\bea
&\mathcal{X}=\mathcal{X}_{4}\,\cup\,\bigcup_{a\in\underline{\mathbf{3}}}\mathcal{X}_{a4}\,\,\cup\bigcup_{ab\in\{12,23,13\}}\mathcal{X}_{ab4},\qquad \ket{\mathcal{Z}}=\prod_{x\in\mathcal{X}}\mathsf{S}_{4}(x)\ket{0}
\eea
and
\bea
&\mathcal{X}_{4}=\left\{v_{4,\alpha}q_{4}^{k_{4}^{(\alpha)}}\middle|\alpha=1,\ldots,n_{4}\right\},\quad \mathcal{X}_{a4}=\left\{v_{a4,\alpha}q_{a}^{i-1}q_{4}^{\lambda^{(\alpha)}_{a4,i}}\middle|\substack{\alpha=1,\ldots,n_{a4}\\i=1,\ldots,\infty}\right\}\quad (a\in\underline{\mathbf{3}})\\
&\mathcal{X}_{ab4}=\left\{v_{ab4,\alpha}q_{a}^{i-1}q_{b}^{j-1}q_{4}^{\pi^{(\alpha)}_{abc,ij}}\middle|\substack{\alpha=1,\ldots,n_{abc}\\i,j=1,\ldots,\infty}\right\}\quad (ab=12,13,23).
\eea
The partition function is given as $\mathcal{Z}=\sum_{\mathcal{X}}\mathcal{Z}[\mathcal{X}]$. Let us consider the expectation value of the $qq$-characters associated with the D-branes spanning the subspace $\mathbb{C}^{3}_{123}\times \mathbb{S}^{1}$ respectively: $\mathscr{Q}_{1,2,3}(x)$, $\mathsf{T}_{12,13,23}(x)$, $\mathsf{T}_{123}(x)$. 

For the $\D2$ $qq$-characters, the expectation value is
\bea
\mathcal{T}_{a}(x)=\left\langle\mathscr{Q}_{a}(x)\right\rangle&=\sum_{k\in\mathbb{Z}}\mathfrak{q}^{k}\left\langle\mathsf{S}_{a}(q_{a}^{k}x)\right\rangle=\frac{1}{\mathcal{Z}}\sum_{\mathcal{X}}\left(\sum_{k\in\mathbb{Z}}\mathfrak{q}^{k}\prod_{x'\in\mathcal{X}}\left[\mathscr{S}_{\overline{a4}}\left(xq_{a}^{k}q_{4}/x'\right)\right]^{x'}_{-}\right)\mathcal{Z}[\mathcal{X}]
\eea
for $a\in\underline{\mathbf{3}}$, where the expectation value of an operator $\mathcal{O}$ is defined as
\beq
\left\langle \mathcal{O}\right\rangle=\frac{\bra{0}\mathcal{O}\ket{\mathcal{Z}}}{\bra{0}\ket{\mathcal{Z}}}.
\eeq

For the $\D4$ $qq$-character, the expectation value is 
\bea
\mathcal{T}_{A}(x)=\left\langle\mathsf{T}_{A}(x)\right\rangle&=\sum_{\lambda}\mathfrak{q}^{|\lambda|}\widetilde{\mathcal{Z}}_{A}^{\D4}[\lambda]\left\langle\Lambda_{A,\lambda}(x)\right\rangle\\
&=\frac{1}{\mathcal{Z}}\sum_{\mathcal{X}}\left(\sum_{\lambda}\mathfrak{q}^{|\lambda|}\widetilde{\mathcal{Z}}_{A}^{\D4}[\lambda]\prod_{x'\in\mathcal{X}}\left[q_{A}q_{4}\frac{\mathscr{Y}^{A}_{\lambda,x}(q_{A}x')}{\mathscr{Y}^{A}_{\lambda,x}(q_{4}^{-1}x')}\right]^{x'}_{-}\right)\mathcal{Z}[\mathcal{X}],
\eea
for $A=12,13,23$. For the $\D6$ $qq$-character, we have 
\bea
\mathcal{T}_{123}(x)=\langle\mathsf{T}_{123}(x)\rangle&=\sum_{\pi}\mathfrak{q}^{|\pi|}\widetilde{\mathcal{Z}}^{\D6}_{123}[K,\pi]\left\langle\Lambda_{123,\pi}^{K}(x)\right\rangle\\
&=\frac{1}{\mathcal{Z}}\sum_{\mathcal{X}}\left(\sum_{\pi}\mathfrak{q}^{|\pi|}\widetilde{\mathcal{Z}}^{\D6}_{123}[K,\pi]\prod_{x'\in\mathcal{X}}\left[(-q_{4}x)\mathscr{W}^{\bar{4},K}_{\pi,x}(q_{4}^{-1}x')\right]^{x'}_{-}\right)\mathcal{Z}[\mathcal{X}].
\eea

After taking the NS limit, the fixed point configurations $\mathcal{X}$ transforms to a saddle point configuration\footnote{After taking the NS limit and denoting the saddle point configuration by $\mathcal{X}_*$, we have
\beq
    \left< 0 | Z \right> = \sum_{\mathcal{X}} Z_{\mathcal{X}} \approx Z_{\mathcal{X}_*} , \qquad
    \left< 0 | \mathcal{O} | Z \right> = \sum_{\mathcal{X}} \mathcal{O}_{\mathcal{X}} Z_{\mathcal{X}} \approx \mathcal{O}_{\mathcal{X}_*} Z_{\mathcal{X}_*} ,
\eeq
and then the expectation value is given by the on-shell value, $\frac{\left< 0 | \mathcal{O} | Z \right>}{\left< 0 | Z \right>} \approx \mathcal{O}_{\mathcal{X}_*}$.
} denoted as $\mathcal{X}_{\ast}$. Note also that under this limit, one can show that the $\U(1)$ contribution of the $\D4$ and $\D6$ partition functions become trivial:
\beq
    \widetilde{\mathcal{Z}}^{\D4}_{A}[\lambda]\xrightarrow{q_{4}\rightarrow 1}1,\quad A\in\{12,13,23\},\qquad
    \widetilde{\mathcal{Z}}^{\D6}_{123}[\pi]\xrightarrow{q_{4}\rightarrow 1}1.
\eeq 
The expectation values of the $qq$-characters then become
\begin{subequations}
\begin{align}
\mathcal{T}_{a}(x)\,\,&\xrightarrow{q_{4}\rightarrow 1}\,\, \sum_{k\in\mathbb{Z}}\mathfrak{q}^{k}\left(\prod_{x'\in\mathcal{X}_{\ast}}S_{a}(x\mathsf{q}_{a}^{k}/x')\right), \qquad a\in\underline{\mathbf{3}}\label{eq:vectorqch}\\
\mathcal{T}_{ij}(x)\,\,&\xrightarrow {q_{4}\rightarrow 1}\,\, \sum_{\lambda}\mathfrak{q}^{|\lambda|}\left(\prod_{x'\in\mathcal{X}_{\ast}}\mathsf{q}_{k}^{-1}\frac{\mathcal{Y}^{(k)}_{\lambda,x}(\mathsf{q}_{k}^{-1}x')}{\mathcal{Y}^{(k)}_{\lambda,x}(x')}\right), \qquad \{i,j,k\}=\underline{\mathbf{3}}\label{eq:Fockqch}\\
\mathcal{T}_{123}(x)\,\,&\xrightarrow {q_{4}\rightarrow 1}\,\, (-x)\sum_{\pi}\mathfrak{q}^{|\pi|}\prod_{x'\in\mathcal{X}_{\ast}}\left(\frac{1-Kx/x'}{1-x/x'}\prod_{\scube\in\pi}\mathsf{g}\left(\frac{x'}{\chi_{x}(\cube)}\right)\right)  ,\label{eq:MacMahonqch}
\end{align}
\end{subequations}
where we used the results in section \ref{sec:QTgl1}. The right-hand sides on \eqref{eq:vectorqch}, \eqref{eq:Fockqch}, \eqref{eq:MacMahonqch} are the vector, Fock, MacMahon $q$-characters of \cite[eq.~(3.15), (3.17)]{Feigin:2016pld} respectively. We can do the same analysis for general $\D4$, $\D6$ $qq$-characters using the results in sections~\ref{sec:verticalFockrep}, \ref{sec:MacMahonrep} and obtain a large class of $q$-characters after taking the NS limit. 

Even before taking the expectation values of the $qq$-characters, we have interesting properties under the NS limit. Under the limit $q_{4}\rightarrow 1$, the vertex operators $\mathsf{S}_{a}(x)\,(a\in\underline{\mathbf{3}})$, $\mathsf{X}_{A}(x)\,(A=12,23,13)$, $\mathsf{W}_{123}(x)$, $\mathsf{A}(x)$ all commute with each other. This is because generally these operators are expressed using the modes $\mathsf{a}_{n}$ as \eqref{eq:relationwithD0}:
\bea
\frac{\mathsf{a}_{n}}{\bfP_{\mathcal{S}}^{[-n]}}
\eea
where $\mathcal{S}\subseteq\underline{\mathbf{3}}$, $\mathcal{S}\neq\emptyset$: $\mathcal{S}\in\{1,2,3,12,13,23,123\}$. The commutation relation is 
\beq\label{eq:qchcommutation}
\left[\frac{\mathsf{a}_{n}}{\bfP_{\mathcal{S}}^{[-n]}},\frac{\mathsf{a}_{m}}{\bfP_{\mathcal{S}'}^{[-m]}}\right]=-\frac{1}{n}\delta_{n+m,0}\frac{\bfP_{\four}^{[n]}}{\bfP_{\mathcal{S}}^{[-n]}\bfP_{\mathcal{S'}}^{[n]}}\propto (1-q_{4}^{n}) \xrightarrow{q_{4}\rightarrow 1} 0
\eeq
where we used $\mathcal{S},\mathcal{S'}\in\underline{\mathbf{3}}$. Denoting the NS limit of the $qq$-characters themselves as 
\bea
\mathscr{Q}_{a}(x)\rightarrow \hat{\mathscr{Q}}_{a}(x)\,\,(a\in\underline{\mathbf{3}}),\quad \mathsf{T}_{A}(x)\rightarrow \hat{\mathsf{T}}_{A}(x)\,\,(A=12,13,23),\quad \mathsf{T}_{123}(x)\rightarrow \hat{\mathsf{T}}_{123}(x) 
\eea
we have the following result.
\begin{theorem}\label{thm:q-ch_commute}
For any $q$-characters obtained in the NS limit $q_4 \to 1$, $\mathcal{T}(x),\mathcal{T}'(x)\in\{\hat{\mathscr{Q}}_{1,2,3}(x),\hat{\mathsf{T}}_{12,23,13}(x),\hat{\mathsf{T}}^{K}_{123}(x)\}$, we have
\beq\label{eq:commutingH}
[\mathcal{T}(x),\mathcal{T}'(x')]=0.
\eeq    
\end{theorem}
The commutativity in \eqref{eq:commutingH} implies the existence of commuting Hamiltonians. Constructing the explicit Hamiltonians of the corresponding integrable models for the most general gauge origami system on the lines of \cite{Nekrasov:2017gzb,Chen:2019vvt,Jeong:2021rll,Prochazka:2023zdb} are left for future work. Moreover, it would be also interesting to relate the triple bilinear identity discussed in~\cite{Grekov:2023} to the triality of the $q$-characters.

\begin{remark}
Despite the coincidence of the BAEs in \eqref{eq:Bethe_eq_D4} and \eqref{eq:Bethe_eq_D6}, we would have different commuting Hamiltonians and spectra for these two cases since they have different $q$-characters as in \eqref{eq:Fockqch} and \eqref{eq:MacMahonqch}. Moreover, we see the difference between the perturbative part of the $\mathsf{Q}$-functions, \eqref{eq:Q-func_D4} and \eqref{eq:Q-func_D6},
\beq
    \mathsf{Q}_{a4}(x) \xrightarrow{\lambda \to \emptyset} \left( v/x;q_a\right)_\infty \, , \qquad
    \mathsf{Q}_{ab4}(x) \xrightarrow{\pi \to \emptyset} \left( v/x;q_a,q_b\right)_\infty \, ,
\eeq
which are given by the $q$-deformed gamma function for the D4 system, while it is given by the $q$-deformed double gamma function for the D6 system. 

\end{remark}

\paragraph{Generalizations to other geometries}
Based on the correspondence with BPS $qq$-characters and quantum toroidal algebras in section~\ref{sec:BPSqqQTA}, we can also study the $q$-character limit and the BAE. To do a detailed discussion, we need the explicit form of the coefficients $\mathcal{Z}_{i}^{\D6}[K,\Lambda^{(i)}]$ in Conj.~\ref{conj:BPSqq}. We hope to report a detailed analysis in a future work~\cite{Kimura-Noshita}. Based on the expression using the structure function of quantum toroidal $\mathfrak{gl}_{1}$ in \eqref{eq:D4QTgl1Bethe}, \eqref{eq:D6QTgl1Bethe}, \eqref{eq:D2QTgl1Bethe}, we propose the following BAE.
\begin{conjecture}\label{conj:BetheBPS}
Let $Z=X\times \mathbb{C}$ where $X$ is a toric CY$_3$ be the space-time of the gauge origami system. Consider a gauge theory including $\mathbb{C}$. The recursion relation of the partition function of this gauge theory gives the following BAE after the NS limit ($q_{4}\rightarrow 1$):
\bea
    1=-\mathfrak{q}_{i}\prod_{x'}\varphi^{\text{c}(x)\Rightarrow \text{c}(x')}(x'/x),
\eea
where the product is taken over some set of variables. We may add flavors that gives
\bea
1=-\mathfrak{q}_{i}\frac{a_{i}(x)}{d_{i}(x)}\prod_{x'}\varphi^{\text{c}(x)\Rightarrow \text{c}(x')}(x'/x)
\eea
where $a_{i}(x),d_{i}(x)$ are additional polynomials specifying the representations of the quiver quantum toroidal algebra.
\end{conjecture}
We note that this BAE also appeared in \cite[eq.~(2.53)]{Galakhov:2022uyu} (see also \cite{Cao:2023lon}). We expect that the BPS $qq$-characters we introduced in this paper give a different derivation of the Bethe/Gauge correspondence. We also note that after taking the NS limit, the BPS $qq$-characters will transform to $q$-characters, which we call BPS $q$-characters. Similar to the $q$-characters of the quantum toroidal $\mathfrak{gl}_{1}$, the BPS $q$-characters obtained in this limit are the $q$-characters of the quiver quantum toroidal algebras. We note that they commute with each other and thus imply the existence of commuting Hamiltonians. We expect that the BPS $qq$-characters and $q$-characters we have introduced will be a tool to systematically derive Hamiltonians for new integrable models.



\section{Geometric realization of \texorpdfstring{$qq$}{qq}-character}\label{sec:geometryqq}

In addition to the algebraic construction based on the vertex operators, the $qq$-character allows geometric construction associated with the geometric representation theory of the corresponding quantum algebra~\cite{Nekrasov:2015wsu,KPfractional}.

\subsection{Quantum affine algebra}

Let $\Gamma = (\Gamma_0,\Gamma_1)$ be a finite-type Dynkin quiver with $\Gamma_0$ the set of nodes and $\Gamma_1$ the set of edges. Let $\mathfrak{g}_\Gamma$ be the corresponding simple Lie algebra.
We consider finite dimensional modules of quantum affine algebra $U_{q_1}(\widehat{\mathfrak{g}}_\Gamma)$ specified by $\Gamma_0$-tuple of polynomials with constant term one (Drinfeld polynomials)~\cite{Drinfeld:1987sy,Chari:1991CMP}. 
We denote the set of degrees of polynomials by $w = (w_i)_{i \in \Gamma_0} \in \mathbb{Z}_{\ge 0}^{\Gamma_0}$, the set of their roots by $\underline{x} = (x_{i,\alpha})_{i \in \Gamma_0,\alpha \in [w_i]}$, and the corresponding finite dimensional module by $M_{w;\underline{x}}$.
As $M_{w;\underline{x}}$ is known to be a type 1 module, we have the weight space decomposition as $U_{q_1}(\mathfrak{g}_\Gamma)$-module,
\begin{equation}
    M_{w;\underline{x}} = \bigoplus_{v} M_{w,v;\underline{x}} \, ,
\end{equation}
where $v = (v_i)_{i \in \Gamma_0} \in \mathbb{Z}_{\ge 0}^{\Gamma_0}$ parametrizes the weight lattice, $\sum_{i \in \Gamma_0} (w_i \varpi_i - v_i \alpha_i) \in P(\mathfrak{g}_\Gamma)$ with $(\varpi_i)_{i \in \Gamma_0}$ the fundamental weights and $(\alpha_i)_{i \in \Gamma_0}$ the simple roots of $\mathfrak{g}_\Gamma$.
Then, $qq$-character is defined as a map from the Grothendieck ring $\operatorname{Rep} U_{q_1}(\widehat{\mathfrak{g}}_\Gamma)$ of finite dimensional $U_{q_1}(\widehat{\mathfrak{g}}_\Gamma)$-modules of type 1 to the formal power series of the form, $(\partial^\bullet \mathscr{Y}_i(x)^{\pm 1})_{i \in \Gamma_0, x \in \mathbb{C}^\times}$.
\begin{proposition}[\cite{Nekrasov:2015wsu}]\label{prop:qq_ch_geom1}
    Let $M_{w;\underline{x}}$ be the $U_{q_1}(\mathfrak{g}_\Gamma)$-module defined above.
    The $qq$-character associated with $M_{w;\underline{x}}$ is given by the equivariant integral over the quiver variety $\mathfrak{M}_{w,v}$ parametrized by $(w,v)$ as follows,
    \begin{equation}
    \mathsf{T}_{w;\underline{x}}[\mathbf{Y}] = \sum_v \mathsf{T}_{w,v;\underline{x}}[\mathbf{Y}]
    \, ,
    \label{eq:qq-ch_w}
    \end{equation}
    where
    \begin{equation}
    \mathsf{T}_{w,v;\underline{x}}[\mathbf{Y}] = q_2^{-\frac{1}{2} \dim \mathfrak{M}_{w,v}} \int_{\mathfrak{M}_{w,v}} \operatorname{ch} \wedge^\bullet \mathbf{Y}_{w,v}^\vee \mathbf{Y} \operatorname{ch} \wedge^\bullet_{q_2} T^\vee \mathfrak{M}_{w,v} \operatorname{td} \left( T \mathfrak{M}_{w,v} \right)
    \, .
    \label{eq:qq-ch_wv}
    \end{equation}
    We denote $\mathbf{Y}_{w,v}^\vee \mathbf{Y} = \sum_{i \in \Gamma_0} \mathbf{Y}_{w,v,i}^\vee \mathbf{Y}_i$ where $\mathbf{Y}_{w,v} = (\mathbf{Y}_{w,v,i})_{i \in \Gamma_0}$ and $\mathbf{Y} = (\mathbf{Y}_i)_{i \in \Gamma_0}$ are the observable bundle and the formal bundle over $\mathfrak{M}_{w,v}$ that defines the $\mathscr{Y}$-function by
    \begin{equation}
    \mathscr{Y}_i(x) = \operatorname{ch} \wedge_{x^{-1}}^\bullet \mathbf{Y}_i
    \, , \qquad 
    i \in \Gamma_0
    \, .
    \end{equation}
    The dimension of the quiver variety is given by $\frac{1}{2} \dim \mathfrak{M}_{w,v} = \sum_{i \in \Gamma_0} w_i v_i -  \sum_{i, j \in \Gamma_0} v_i c_{ij}^+ v_j$ where $c^+ = (c_{ij})^+_{i, j \in \Gamma_0}$ is the half Cartan matrix associated with quiver $\Gamma$.
\end{proposition}

In fact, this formula is understood as the equivariant $\chi_{q_2}$-genus of the quiver variety with the additional insertion $\mathbf{Y}_{w,v}^\vee \mathbf{Y}$, which is physically interpreted as the coupling with the defect brane as discussed in section~\ref{sec:qqpartitionfunct}.
From this point of view, in the limit $q_2 \to 1$, it is reduced to the Euler characteristics of the corresponding quiver variety, which provides the geometric realization of the $q$-character.
In order to consider the non-simply-laced algebras, we need to consider the fractional quiver variety~\cite{KPfractional}. See also~\cite{Nakajima:2019olw}.

\begin{remark}
    In the geometric formula, two parameters $q_1$ and $q_2$ play different roles: $q_1$ is the quantum deformation parameter of the algebra $U_{q_1}(\widehat{\mathfrak{g}})$, while $q_2$ is the twist parameter for the cotangent bundle insertion. 
    Although they have different meanings, the resulting expression is in the end symmetric for $q_1$ and $q_2$ for the simply-laced algebra~\cite{Nekrasov:2015wsu}.
    On the other hand, the $qq$-character is not symmetric under $q_1$ and $q_2$ for non-simply-laced cases~\cite{Kimura:2017hez,KPfractional}.
    See also~\cite{Frenkel:1997CMP}.
\end{remark}

\begin{remark}
    It has been known that there exists another (non-commutative) deformation of the $q$-character, which is called the $t$-analog of $q$-character~\cite{Nakajima:1999JAMS,Nakajima:2001PC,Nakajima:2004AM}.
    From the geometric point of view, the $t$-analog is the deformation based on the Poincaré polynomial compared with the Euler characteristics, while the $qq$-character corresponds to the $\chi_{q_2}$-genus of the quiver variety.
\end{remark}

\subsection{Quantum toroidal \texorpdfstring{$\mathfrak{gl}_1$}{gl(1)}}

Prop.~\ref{prop:qq_ch_geom1} also applies to affine quivers, which gives rise to the $qq$-character of quantum toroidal algebra.
For example, for $\Gamma = \widehat{A}_0$, we obtain the $qq$-character of the Fock representation of quantum toroidal $\mathfrak{gl}_1$ that we denote by $\mathscr{E}$.
Motivated by our analysis of the $qq$-character based on the vertex operator formalism, we have the geometric formula for the $qq$-character of the MacMahon module (and its tensor product) of $\mathscr{E}$.
\begin{theorem}\label{thm:qq_ch_geom2}
    Let $w \in \mathbb{Z}_+$ and $(x_\alpha)_{\alpha = 1, \ldots, w} \in (\mathbb{C}^\times)^w$.
    The $qq$-character of degree-$w$ tensor product of the MacMahon module of quantum toroidal $\mathfrak{gl}_1$ is given by the following equivariant integral,
    \begin{equation}
    \mathsf{T}_{w;\underline{x}}[\mathbf{Y}] = \sum_{v \ge 0} \mathsf{T}_{w,v;\underline{x}}[\mathbf{Y}]
    \, , \qquad
    \mathsf{T}_{w,v;\underline{x}}[\mathbf{Y}] =  \int_{[\mathfrak{M}_{w,v}]^\text{vir}} \operatorname{ch} \wedge^\bullet \mathbf{Y}_{w,v}^\vee \mathbf{Y} \operatorname{td} \left( T \mathfrak{M}_{w,v} \right)
    \, ,
    \label{eq:qq-ch_wv_MacMahon}
    \end{equation}
    where the integral is taken over the virtual fundamental cycle of $\mathfrak{M}_{w,v}$, the moduli space of $v$ D0 and $w$ D6 system (rank $w$ tetrahedron instanton on $\mathbb{C}_{123}^3$), which is isomorphic to the Quot scheme, $\mathfrak{M}_{w,v} \cong \operatorname{Quot}_{\mathbb{C}^3}^v(\mathcal{O}^{\oplus w})$.
    We denote the observable bundle and the formal bundle over $\mathfrak{M}_{w,v}$ by $\mathbf{Y}_{w,v}$ and $\mathbf{Y}$ that defines the $\mathscr{Y}$-function as before.
\end{theorem}
\begin{proof}
    We consider the case $w = 1$ for simplicity. 
    As this integral is understood as the equivariant integral, it localizes on the equivariant fixed points, which are characterized by plane partitions $\pi$ with $|\pi| = v$. 
    Trivializing the formal bundle, $\operatorname{ch}\mathbf{Y} \to 1$, it is reduced to the equivariant integral over the moduli space without any insertion, which gives rise to the K-theoretic tetrahedron instanton partition function on $\mathbb{C}^3_{123}$, corresponding to $n_{\bar{4}} = w = 1$ and $n_{\bar{a}} = 0$ for $a \neq 4$ in our notation (see section~\ref{sec:Tetrahedron_inst}). 
    Hence, we have the coefficient $\widetilde{\mathcal{Z}}^\text{D6}_{\bar{4}}[\vec{\pi}]$ in Def.~\ref{def:D6_qq-ch} as the contribution from each fixed point.
    On the other hand, the character of the observable bundle~\eqref{eq:D6Ybundle} is given by $\operatorname{ch}\mathbf{Y}_{w,v}\Big|_{{\pi}} = x\left(1-\prod_{i=1}^{3}(1-q_{i})\sum_{(i,j,k) \in \pi}q_1^{i-1} q_2^{j-1} q_3^{k-1})\right)$.
    Hence, from Lem.~\ref{lemma:D6_iWeyl_ref}, $\operatorname{ch} \wedge^\bullet \mathbf{Y}_{w,v}^\vee \mathbf{Y}$ provides  $\Lambda_{\bar{4},\pi}(x)$ at the fixed point ${\pi}$ under the replacement of $\mathsf{W}_{\bar{4}}(x)$ with $\mathscr{Y}(x) = \operatorname{ch} \wedge_{x^{-1}}^\bullet \mathbf{Y}$.
    The higher rank case $w > 1$ works totally in the same way.   
\end{proof}

The moduli space $\mathfrak{M}_{w,v}$ has more constraints than the ordinary D0-D4 instanton moduli space, which are implemented by the potential of the quiver variety~\cite{Pomoni:2021hkn,Cao:2023lon}. 
As a result, the virtual dimension of the moduli space becomes zero in this case, and hence, compared to the previous formula~\eqref{eq:qq-ch_wv}, we have no additional $q_2$-twisted cotangent bundle insertion.

\subsection{Quantum toroidal \texorpdfstring{$\mathfrak{gl}_n$}{gl(n)}}

Let us also comment on the geometric formula of the $qq$-character of quantum toroidal $\mathfrak{gl}_n$ that we denote by $\mathscr{E}_n$.
We also have vector/Fock/MacMahon representations for $\mathscr{E}_n$~\cite{Feigin:2013JA}.
We can apply the previous formula \eqref{eq:qq-ch_wv} to affine quiver $\widehat{A}_{n-1}$ to obtain the Fock representation of $\mathscr{E}_n$~\cite{Kimura:2022spi}.
For the MacMahon representation, we apply the formula \eqref{eq:qq-ch_wv_MacMahon} with the moduli space of D0-D6 system on $\mathbb{C}^3/\mathbb{Z}_n$ to construct the $qq$-character.
\begin{conjecture}
    Let $\underline{i} = (i_j)_{j \in 0,\ldots,n-1} \in (\mathbb{Z}/n\mathbb{Z})^n$ with $i_j = \delta_{i,j}$.
    Let $M_{\underline{i};x}$ be the MacMahon module of color $i$ of $\mathscr{E}_n$.
    For $w = (w_i)_{i \in 0,\ldots,n-1} \in \mathbb{Z}_{\ge 0}^n$ and $\underline{x} = (x_{i,\alpha})_{i \in 0,\ldots,n-1, \alpha = 1,\ldots,w_i} \in (\mathbb{C}^\times)^{|w|}$, we consider the tensor product module of the MacMahon modules,
    \begin{equation}
        M_{w;\underline{x}} = \bigotimes_{i \in 0,\ldots,n-1} \bigotimes_{\alpha = 1,\ldots,w_i} M_{\underline{i};x_{i,\alpha}} \, .
    \end{equation}
    Then, the $qq$-character of $\mathscr{E}_n$-module $M_{w;\underline{x}}$ is given by the equivariant integral over the moduli space of D0-D6 system on $\mathbb{C}^3/\mathbb{Z}_n$ denoted by $\mathfrak{M}_{w,v} \cong \operatorname{Quot}_{\mathbb{C}^3/\mathbb{Z}_n}^{|v|}(\mathcal{O}^{\oplus |w|})$, $v = (v_i)_{i = 0,\ldots,n-1} \in \mathbb{Z}_{\ge 0}^n$, as follows,
    \begin{equation}
    \mathsf{T}_{w;\underline{x}}[\mathbf{Y}] = \sum_{v \in \mathbb{Z}_{\ge 0}^n} \mathsf{T}_{w,v;\underline{x}}[\mathbf{Y}]
    \, , \qquad
    \mathsf{T}_{w,v;\underline{x}}[\mathbf{Y}] =  \int_{[\mathfrak{M}_{w,v}]^\text{vir}} \operatorname{ch} \wedge^\bullet \mathbf{Y}_{w,v}^\vee \mathbf{Y} \operatorname{td} \left( T \mathfrak{M}_{w,v} \right)
    \, .
    \end{equation}
\end{conjecture}
In this case, the integral localizes on the fixed points parametrized by the colored plane partition, which characterizes the MacMahon module of $\mathscr{E}_n$~\cite{Feigin:2013JA}.

The geometric formula presented here could be straightforwardly extended to various $qq$-characters discussed in previous sections by replacing the moduli space with the corresponding one.


\section{Quiver elliptic W-algebra}\label{sec:ellipticWorigami}
As mentioned in section \ref{sec:physicalsetup}, instead of considering the gauge origami system on $\mathbb{C}^{4}\times \mathbb{R}\times \mathbb{S}^{1}$, we can consider it also on $\mathbb{C}^{4}\times \mathbb{R}^{2}$ or $\mathbb{C}^{4}\times \mathbb{T}^{2}$. For the former case, the arising algebra is the rational version of the quiver W-algebra \cite{Nieri:2019mdl} (see also \cite{Bourgine:2018uod}), while for the latter case, the arising algebra is the quiver elliptic W-algebra \cite{Kimura:2016dys}. For the rational case, one needs to take care of the perturbative one-loop part carefully but a similar computation can be done. In this paper, we only discuss the elliptic counterpart of the $qq$-characters introduced in previous sections.

We denote the elliptic parameter $p=e^{2\pi i \tau}\in\mathbb{C}^\times$. The elliptic deformation is obtained by modifying the index as in \eqref{eq:rat_trig_ell}:
\beq
    \mathbb{I}[x]=1-x^{-1}\rightarrow \mathbb{I}_{p}[x]=\theta(x^{-1};p),
\eeq
where the theta function is defined in \eqref{eq:ellipticthetadef}. Note that in the limit $p\rightarrow 0$, we obtain the trigonometric index. 

At the operator level, we need to double the number of modes and introduce two independent operators \cite{Kimura:2016dys}. For example, the root current will be modified as 
\bea
    \mathsf{A}(x)=\mathsf{a}_{0}(x):\exp\left(\sum_{n\neq 0}\mathsf{a_{n}}x^{-n}\right):\,\,&\rightarrow \,\, \mathsf{A}^{\theta}(x)=\mathsf{a}_{0}(x):\exp\left(\sum_{n\neq 0}\left(\mathsf{a}_{\theta,n}^{\splus}x^{-n}+\mathsf{a}_{\theta,n}^{\sminus}x^{n}\right)\right):,
\eea
where
\beq
[\mathsf{a}_{\theta,n}^{\spm},\mathsf{a}_{\theta,m}^{\spm}]=\mp\frac{1}{n}\frac{\mathbf{P}_{\four}^{[\pm n]}}{1-p^{\pm n}}\delta_{n+m,0}.
\eeq
In this process, the modes will be elliptically deformed. We only deform the nonzero modes and leave the zero modes undeformed.

The operators for $\D2$, $\D4$, $\D6$, $\D8$ branes are similarly defined as 
\begin{subequations}
\begin{align}
\D2:&\quad\mathsf{S}^{\theta}_{a}(x)=\mathsf{s}_{a,0}(x):\exp\left(\sum_{n\neq 0}\left(\mathsf{s}_{\theta,a,n}^{\splus}x^{-n}+\mathsf{s}^{\sminus}_{\theta,a,n}x^{n}\right)\right):,\quad \mathsf{s}_{\theta,a,n}^{\spm}=\frac{\mathsf{a}_{\theta,n}^{\spm}}{\mathbf{P}_{a}^{[\mp n]}},\quad a\in\four,\\
\D4:&\quad \mathsf{X}_{A}^{\theta}(x)=\mathsf{x}_{A,0}(x):\exp\left(\sum_{n\neq0}\left(\mathsf{x}_{\theta,A,n}^{\splus}x^{-n}+\mathsf{x}_{\theta,A,n}^{\sminus}x^{n}\right)\right):,\quad
\mathsf{x}_{\theta,A,n}^{\spm}=\frac{\mathsf{a}_{\theta,n}^{\spm}}{\mathbf{P}_{A}^{[\mp n]}},\quad A\in\six,\\
\D6:&\quad \mathsf{W}^{\theta}_{\bar{a}}(x)=\mathsf{w}_{\bar{a},0}(x):\exp\left(\sum_{n\neq 0}\left(\mathsf{w}^{\splus}_{\theta,\bar{a},n}x^{-n}+\mathsf{w}_{\theta,\bar{a},n}^{\sminus}x^{n}\right)\right):,\quad \mathsf{w}_{\theta,\bar{a},n}^{\spm}=\frac{\mathsf{a}_{\theta,n}^{\spm}}{\mathbf{P}_{\bar{a}}^{[\mp n]}},\quad a\in\four,\\
\D8:&\quad \mathsf{Z}^{\theta}(x)=\mathsf{z}_{0}(x):\exp\left(\sum_{n\neq 0}\left(\mathsf{z}^{\splus}_{\theta,n}x^{-n}+\mathsf{z}^{\sminus}_{\theta,n}x^{n}\right)\right):,\quad \mathsf{z}^{\spm}_{\theta,n}=\frac{\mathsf{a}_{\theta,n}^{\spm}}{\bfP_{\four}^{[\mp n]}}.
\end{align}   
\end{subequations}
Similarly, elliptic analogs of the structure functions \eqref{eq:struct_funct} are
\bea
\mathscr{V}^{\theta}_{a}(x)&=\mathbb{I}_{p}[-\bfP_{a}^{\vee}x^{\vee}]=\frac{\theta(q_{a}x;p)}{\theta(x;p)},\\
\mathscr{S}^{\theta}_{ab}(x)&=\mathbb{I}_{p}[-\mathbf{P}_{ab}^{\vee}x^{\vee}]=\frac{\theta(q_{a}x;p)\theta(q_{b}x;p)}{\theta(x;p)\theta(q_{a}q_{b}x;p)},\\
g_{\bar{a}}^{\theta}(x)&=\mathbb{I}_{p}[-\bfP_{\bar{a}}^{\vee}x^{\vee}]=\frac{\prod_{i\neq a}\theta(q_{i}x;p)\theta(q_{\bar{a}}x;p)}{\theta(x;p)\prod_{i\neq a}\theta(q_{a}^{-1}q_{i}^{-1}x;p)},\\
\mathcal{A}_{\mathbb{C}^{4}}^{\theta}(x)&=\mathbb{I}_{p}[-\mathbf{P}_{\four}^{\vee}x^{\vee}]=\frac{\prod_{a\in\four}\theta(q_{a}x;p)\prod_{a\in\four}\theta(q_{a}^{-1}x;p)}{\theta(x;p)^{2}\prod_{A\in\six}\theta(q_{A}x;p)}.
\eea
The D-brane $qq$-characters are defined by changing the operators to the elliptic deformed vertex operators:
\begin{subequations}
\begin{align}
\D2\,\,\text{elliptic $qq$-character}:&\quad \mathscr{Q}^{\theta}_{a}(x)=\sum_{k\in\mathbb{Z}}\mathfrak{q}^{k}\mathsf{S}^{\theta}_{a}(q_{a}^{k}x),\\
\D4\,\,\text{elliptic $qq$-character}:&\quad \mathsf{T}^{\theta}_{A}(x)=\sum_{\lambda}\mathfrak{q}^{|\lambda|}\mathcal{Z}_{\theta,A}^{\D4}[\lambda]:\mathsf{X}_{A}^{\theta}(x)\prod_{\Abox\in\lambda}\mathsf{A}^{\theta}(\chi_{A,x}(\Bbox))^{-1}:,\\
\D6\,\,\text{elliptic $qq$-character}:&\quad \mathsf{T}_{\bar{4}}^{\theta}(x)=\sum_{\pi}\mathfrak{q}^{|\pi|}\mathcal{Z}_{\bar{4},\text{CS}}^{\D6}\mathcal{Z}_{\theta,\bar{4}}^{\D6}[\pi]:\mathsf{W}_{\bar{4}}^{\theta}(x)\prod_{\scube\in\pi}\mathsf{A}^{\theta}(\chi_{\bar{4},x}(\cube))^{-1}:,\\
\D8\,\,\text{elliptic $qq$-character}:&\quad \mathsf{T}^{K\,\theta}_{\four;a}(x)=\sum_{\rho\in\mathcal{SP}}\mathfrak{q}^{|\rho|}\mathcal{Z}^{\D8}_{\theta,\four;a}[\rho,K]:\widetilde{\mathsf{Z}}^{K,\theta}(x)\prod_{\shcube\in\rho}\mathsf{A}^{\theta}(\chi_{\four,x}(\hcube))^{-1}:.
\end{align}    
\end{subequations}
One can show that the commutativity with the elliptic deformed screening charges also holds similar to the trigonometric $qq$-characters. We also can obtain elliptic deformations of the general D4, D6 $qq$-characters and even the BPS $qq$-characters introduced in the previous sections. Since it is a straightforward generalization, we omit the results of them.


\section{Conclusion and discussions}\label{sec:conclusion}
We introduced vertex operators corresponding to the D$(2p)$-branes $(p=0,1,2,3,4)$ and gave the free field realizations of the contour integral formulas of the gauge origami partition function in $\mathbb{C}^{4}$. Interestingly, they can be understood in terms of graded quivers associated with 2d $\mathcal{N}=(0,2)$ quiver gauge theories, which enabled us to generalize the conventional quiver W-algebra. Based on this free field realization, we managed to construct $\D2,\D4,\D6$ $qq$-characters and show the BPS/CFT correspondence for the coupled vortex, spiked instanton, and tetrahedron instanton systems. We also managed to give many generalizations and conjectures of the gauge origami setup from the quantum algebraic viewpoint. 

The D2 $qq$-characters play the roles of screening charges and the $\D4,\D6$ $qq$-characters are uniquely determined after setting the highest weight and imposing the commutativity condition with the screening charge. An interesting property was that the monomial terms are classified by truncations of plane partitions and the patterns of them were determined from the highest weight. Moreover, we have a one-to-one correspondence with the MacMahon representation and its generalizations (e.g. truncations, boundary conditions, etc.) of the quantum toroidal $\mathfrak{gl}_{1}$, and the Bethe ansatz equations were derived in a natural way.

Despite the fact that all processes in D2, D4, and D6 were relatively successful, the D8 case was not so successful because we could not construct the screening charge associated with it, and as a result, the D8 $qq$-character we constructed could only reproduce the partition function of the magnificent four model up to sign factors. The construction of the complete D8 $qq$-character is one of the ultimate goals because this D8 $qq$-character is expected to be the mother of all other D-brane $qq$-characters in the gauge origami system.

Let us list down possible directions for future work.
\vspace{-0.2cm}
\paragraph{Gauge theoretical interpretation of general $\D6$ $qq$-characters} As discussed in section~\ref{sec:generalD6qq}, general D6 $qq$-characters have a one-to-one correspondence with truncations of plane partitions. Using the BPS/CFT correspondence (see Thm.~\ref{thm:BPS/CFTintro}), we can obtain generalizations of the tetrahedron instanton partition functions by constructing the $qq$-characters and studying their compositions. Let $\pi_{1},\pi_{2}$ be the two different truncations of the plane partition. After constructing the associated D6 $qq$-characters $\mathsf{T}_{\pi_{1,2}}(x)$, the vacuum expectation value of the composition of them takes the form:
\bea
\mathcal{Z}=\sum_{\pi_{1,2}}\mathfrak{q}^{|\pi_{1}|+|\pi_{2}|}\prod_{i=1}^{2}\mathcal{Z}[\pi_{i}]\mathcal{Z}[\pi_{1},\pi_{2}]\mathcal{Z}[\pi_{2},\pi_{1}].
\eea
The physical interpretation of both $\mathcal{Z}[\pi_{i}]$ and $\mathcal{Z}[\pi_{i},\pi_{j}]$ is necessary. Studying the relation with open BPS states might help to solve this problem \cite{Nagao:2009ky,Nagao:2009rq,Sulkowski:2010eg}. Recently, the relation with a different algebra (shifted quiver Yangian) was studied in \cite{Galakhov:2021xum}. Understanding the connection with the representations appearing there might help.

\vspace{-0.2cm}
\paragraph{Gauge origami of general toric Calabi--Yau four-folds}
As mentioned in section~\ref{sec:introduction} and~\ref{sec:physicalsetup}, the gauge origami system is understood as a setup where D1-branes are probing intersecting D-branes wrapping cycles in the Calabi--Yau four-fold. In \cite{Pomoni:2021hkn}, the authors defined the elliptic genus \cite{Benini:2013nda,Benini:2013xpa,Benini:2018hjy} of the tetraheron instanton system and evaluated the poles by using the JK-residue techniques \cite{Jeffrey1993LocalizationFN}. Generalization of this computation to general gauge origami setup is necessary. Since toric Calabi--Yau four-folds are related to brane brick models, the computation in \cite{Franco:2017cjj} might help. Comparison with the contour integral formulas we proposed in Conj.~\ref{conj:CY4}, \ref{conj:CY3} might be interesting too.

\vspace{-0.2cm}
\paragraph{BPS $qq$-characters and crystal melting}
In the main section, we gave multiple conjectures (Conj.~\ref{conj:CY4}, \ref{conj:CY3}, \ref{conj:BPSqq}) related to generalizations to toric Calabi--Yau four-folds. We hope to come back and solve these conjectures by discussing the relation with brane brick models and brane tilings in a future work \cite{Kimura-Noshita}. Since all of the BPS $qq$-characters associated with the gauge origami system of a Calabi--Yau four-fold $Z$ are expected to be understood as truncations of a four-dimensional analog of the 3d BPS crystals \cite{Ooguri:2009ijd}, it is necessary to generalize the discussion of~\cite{Ooguri:2009ijd} to Calabi--Yau four-folds. Truncations of them are also necessary to be studied.

\vspace{-0.2cm}
\paragraph{Finite type quivers and $qq$-characters}
Although the $qq$-characters we introduced in this paper are associated with ``affine"-type quivers which give an infinite number of monomial terms appearing in the $qq$-character, we can consider $qq$-characters with a finite number of monomial terms. Such kind of $qq$-character should be understood as a singular limit of the $q$-parameters appearing in the $q$-Cartan matrix. For example, consider the $\mathbb{C}^{2}_{12}\times \mathbb{C}_{34}^{2}$ gauge origami system and place a D4$_{12}$-brane and a D4$_{34}$-brane spanning transversely. The D4$_{12}$ theory is a 5d $\mathcal{N}=1^{\ast}$ theory with adjoint mass $q_{3}$. We can take the limit $q_{3}\rightarrow 0,\infty$ while keeping the product $q_{3}q_{4}$ fixed. In this process, the 5d $\mathcal{N}=1^{\ast}$ theory on $\mathbb{C}^{2}_{12}\times \mathbb{S}^{1}$ becomes the pure SYM theory whose quiver structure is the $A_{1}$. Namely, the quiver structure changes $\widehat{A}_{0}\rightarrow A_{1}$ under this limit. Actually, one can show that the two $qq$-characters $\mathsf{T}_{12}(x),\mathsf{T}_{34}(x)$ will be deformed as
\bea
\mathsf{T}_{12}(x)&\longrightarrow \mathsf{T}_{\text{pure SYM}}(x),\\
\mathsf{T}_{34}(x)&\longrightarrow \mathsf{T}_{\text{$q$-Vir}}(x),
\eea
where $\mathsf{T}_{\text{pure SYM}}(x)$ is the $qq$-character reproducing the instanton partition functions of the pure SYM, while $\mathsf{T}_{\text{$q$-Vir}}(x)$ is the generator of the $q$-Virasoro algebra \cite{Kimura-Noshita}. Generalizations of this discussion to other quiver W-algebras is a possible future work.

\vspace{-0.2cm}
\paragraph{Relation with intertwiners and quantum toroidal algebras}
A recent work \cite{Konno:2021zvl} derived the D4 $qq$-character from a new elliptic quantum toroidal algebra denoted as $U_{q,t,p}(\mathfrak{gl}_{1,\text{tor}})$ by studying the intertwiner formalism based on \cite{Awata:2011ce} (see also \cite{Awata:2016riz, Awata:2016mxc,Awata:2016bdm,Bourgine:2016vsq,Awata:2017cnz,Awata:2017lqa, Bourgine:2017jsi, Bourgine:2017rik,Bourgine:2018uod,Zenkevich:2018fzl,Bourgine:2019phm,Zenkevich:2020ufs,Ghoneim:2020sqi,Bourgine:2021nyw,Bourgine:2022scz} for related works). The elliptic quantum toroidal algebra $U_{q,t,p}(\mathfrak{gl}_{1,\text{tor}})$ is an algebra with four parameters $q,t,p,p^{\ast}=p\gamma^{-2}$, where $\gamma$ is the central element of the algebra. In their paper, they studied the vertical representation with $\gamma=1$ and the horizontal representation with $\gamma=(t/q)^{1/2}$. In the horizontal representation, we have $p^{\ast}/p=q/t$ which implies the identification $(q_{1},q_{2},q_{3},q_{4})=(q,t^{-1},p,t/(qp))$. Note at the limit $p^{\ast}\rightarrow 1$, we have $p=t/q$. We also note that their elliptic parameter $p$ is related to the $\Omega$-deformation parameter of $\mathbb{C}^{4}$ and is \emph{different} from the elliptic parameter in section~\ref{sec:ellipticWorigami}. Our elliptic parameter $p$ is related to $\mathcal{C}=\mathbb{T}^{2}$ but not $\mathbb{C}^{4}$ and from the algebraic viewpoint, it is rather related to~\cite{Saito2013EllipticDA,Zhu:2017ysu,Foda:2018sce}. 

The two screening currents in \cite[section 5.1]{Konno:2021zvl} are identified as 
\bea
s_{n}^{+}\leftrightarrow \mathsf{s}_{4,n},\quad s_{n}^{-}\leftrightarrow \mathsf{s}_{3,n}
\eea
and the $T(u)$ operator \cite[section 5.2]{Konno:2021zvl} is identified with the $\D4_{12}$ $qq$-character $\mathsf{T}_{12}(u)$. As the authors studied one type of D4 $qq$-characters spanning the $\mathbb{C}_{12}^{2}$ out of the other six possible cases, the BPS/CFT correspondence of the most general spiked instanton system could not be reproduced. It would be interesting to consider other types of $qq$-characters in the context of elliptic quantum toroidal algebra as well. Since the algebra still has a triality symmetry which is the subgroup of the quadrality symmetry, we expect we can construct three types of vertical and horizontal representations (see \cite{Frenkel:1997CMP,bershtein2018plane,Kojima2019,Feigin2020DeformationsO,Kojima2021,Harada:2021xnm}).  

Studying interwiners relating these three different types of representations following the line of \cite{Zenkevich:2018fzl,Zenkevich:2020ufs} should give the partition function of intersecting 5d gauge theories. However, we note that since the elliptic parameter $p$ (or $p^{\ast}$) is already breaking the quadrality symmetry, the partition function obtained in this way might be the gauge origami partition function of D4-branes wrapping the subspaces of the $\mathbb{C}^{3}$ part and thus we still can not obtain the full spiked instanton partition function. This is a situation similar to what happened in this paper, where we can choose one screening charge $\mathscr{Q}_{4}(x)$ and the commutativity with this screening charge gives general $qq$-characters reproducing the gauge origami system in the $\mathbb{C}^{3}_{123}$ subspace. To obtain the complete partition function, we may need a new algebra symmetric in $q_{1,2,3,4}$ with $q_{1}q_{2}q_{3}q_{4}=1$.

\vspace{-0.2cm}
\paragraph{Commuting Hamiltonians and quantum integrable models}
In section~\ref{sec:Betheansatz}, we discussed the Bethe ansatz equation of the gauge origami system and related it with the $q$-characters and the associate quantum toroidal algebras. Thm.~\ref{thm:generalBetheansatz}, \ref{thm:BAEwithPolynom}, \ref{thm:q-ch_commute}, and Conj.~\ref{conj:BetheBPS}, imply new types of quantum integrable models. Explicitly deriving the commuting Hamiltonians is an interesting topic. Recent studies \cite{Nekrasov:2017gzb,Jeong:2017pai,Jeong:2021rll,Chen:2019vvt,Jeong:2023qdr} gave a way to derive the commuting Hamiltonians by studying surface defects in the gauge origami system. Generalizations of them might help solve this problem. 

We also note that the similar Bethe ansatz equations in Conj.~\ref{conj:BetheBPS} were proposed in~\cite{Galakhov:2022uyu} (see also \cite{Bao:2022fpk,Cao:2023lon,Prochazka:2023zdb} for related works) by studying the R-matrices of the quiver Yangians following previous works \cite{Maulik2012QuantumGA,Fukuda:2017qki,Litvinov:2020zeq,Litvinov:2021phc,Chistyakova:2021yyd,Kolyaskin:2022tqi}. It was shown in \cite{Galakhov:2022uyu} that under suitable assumptions, there is no universal R-matrix for some representations, and thus there are obstructions in the Bethe/Gauge correspondence. However, the BPS $q$-characters obtained under the NS limit of the BPS $qq$-characters look to commute with each other and therefore imply commuting Hamiltonians. Therefore, naively we expect to still have Bethe/Gauge correspondence. To fill in this gap, detailed discussions on the BPS $qq$-characters and their truncated versions are necessary. In the historical paper~\cite{Frenkel:1998ojj}, the $q$-characters associated with quantum affine algebras and their relations with deformed W-algebras \cite{Frenkel:1997CMP} and the R-matrices were discussed. Generalization of these works to quantum toroidal $\mathfrak{gl}_{1}$ was done in \cite{Feigin:2015raa,Feigin:2016pld}. Understanding the BPS $qq$-characters, BPS $q$-characters, and Bethe ansatz equations from the R-matrix of quiver quantum toroidal algebras following \cite{Frenkel:1998ojj,Frenkel:1997CMP,Feigin:2015raa,Feigin:2016pld} is also a topic that must be studied.

\section*{Acknowledgements}

We would like to thank Yalong Cao, Hiroaki Kanno, Hitoshi Konno, Norton Lee, and Hiraku Nakajima for fruitful discussions on the project.
The work of TK was in part supported by EIPHI Graduate School (No.~ANR-17-EURE-0002) and Bourgogne-Franche-Comté region. A part of the results in this paper has been presented by TK at \href{https://indico.sissa.it/event/98/}{XIII Workshop on Geometric Correspondences of Gauge Theories}, June 2023, Trieste, Italy, and \href{https://sites.google.com/view/msj-si-2023/home}{The 16th MSJ-SI: Elliptic Integrable Systems, Representation Theory and Hypergeometric Functions}, July 2023, Tokyo, Japan. He is grateful to the organizers for the invitation and hospitality.
GN is supported by JSPS KAKENHI Grant-in-Aid for JSPS fellows Grant No.~JP22J20944, JSR Fellowship, and FoPM (WINGS Program), the University of Tokyo. GN is also grateful for the hospitality during his stay in Université de Bourgogne where a part of this work was initiated.



\appendix
\allowdisplaybreaks
\newpage
\section{Special functions}\label{app:specialfunct}

\subsection{\texorpdfstring{$q$}{q}-functions}\label{sec:q-functions}
We define the $q$-shifted factorial as
\begin{equation}
    (z;q)_{n}=\prod_{m=0}^{n-1}(1-zq^{m}).
\end{equation}
Taking the limit $n\rightarrow \infty$ for $|q|<1$, we have 
\beq
(z;q)_{\infty}=\prod_{m=0}^{\infty}(1-zq^{m})=\exp\left(-\sum_{m=1}^{\infty}\frac{z^{m}}{m(1-q^{m})}\right).
\eeq
For the analytic region $|q|>1$, it is given through analytic continuation
\begin{equation}\label{eq:app-qPochreflec}
    (z;q)_{\infty}=(zq^{-1};q^{-1})_{\infty}^{-1}.
\end{equation}
Note that we have the property 
\beq
(x;q)_{\infty}=(1-x)(xq;q)_{\infty}
\eeq
and 
\beq
(z;q)_{n}=\frac{(z;q)_{\infty}}{(zq^{n};q)_{\infty}}=\exp\left(-\sum_{m=1}^{\infty}\frac{z^{m}}{m}\frac{1-q^{mn}}{1-q^{m}}\right).
\eeq
We also define the multiple version of the $q$-shifted factorial for $|q_{1}|,\ldots,|q_{k}|<1$,
\footnote{
This multiple $q$-shifted factorial is interpreted as a $q$-analog of the multiple gamma function, $\Gamma_k(x;q_1,\ldots,q_k) = (z;q_{1},\ldots,q_{k})_{\infty}^{(-1)^k}$ obeying the relation,
\beq
    \frac{\Gamma_k(xq_i;q_1,\ldots,q_k)}{\Gamma_k(q_i;q_1,\ldots,q_k)} = \Gamma_{k-1}(x;q_1,\ldots,q_i\hspace{-.65em}/\hspace{.3em},\ldots,q_k) .
\eeq
}
\beq
(z;q_{1},\ldots,q_{k})_{\infty}=\prod_{0\leq n_{1},\ldots,n_{k}\leq \infty}(1-zq_{1}^{n_{1}}\cdots q_{k}^{n_{k}}).
\eeq
For other analytic regions, we use 
\beq
\exp\left(-\sum_{m=1}^{\infty}\frac{z^{m}}{m}\frac{1}{(1-q_{1}^{m})\cdots (1-q_{k}^{m})}\right)
\eeq
and do a similar process as \eqref{eq:app-qPochreflec}. 

The one-loop perturbative functions appearing in \eqref{eq:D8oneloop}, \eqref{eq:D6oneloop}, \eqref{eq:D4oneloop}, \eqref{eq:D2oneloop} can be rewritten using the $q$-shifted factorial. For example, for the $\D2$-case in \eqref{eq:D2oneloop}, we have
\begin{align}\label{eq:D2oneloop-qgamma}
\begin{split}
\mathcal{Z}^{\D2\tbar\D2}_{\text{1-loop}}(x,a|x',a)&=\exp\left(-\sum_{n=1}^{\infty}\frac{1}{n}\frac{\bfP_{\bar{a}}^{[n]}\bfP_{\bar{a}}^{[-n]}}{\bfP_{\four}^{[n]}}\left(\frac{x}{x'}\right)^{n}\right)\\
&=\begin{dcases}
  \frac{(x/x';q_{a})_{\infty}\prod\limits_{i\neq a}(q_{i}q_{a}x/x';q_{a})_{\infty}}{(q_{a}x/x';q_{a})_{\infty}\prod\limits_{i\neq a}(q_{i}^{-1}x/x';q_{a})_{\infty}} ,\quad |q_{a}|<1,\\
  \frac{(x/x';q_{a}^{-1})_{\infty}\prod\limits_{i\neq a}(q_{i}^{-1}q_{a}^{-1}x/x';q_{a}^{-1})_{\infty}}{(q_{a}^{-1}x/x';q_{a}^{-1})_{\infty}\prod\limits_{i\neq a}(q_{i}x/x';q_{a}^{-1})_{\infty}} ,\quad |q_{a}|>1.
\end{dcases}
\end{split}
\end{align}

For the $\D4$-case when $A=ab$, $\bar{A}=cd$ and $|q_{a}|,|q_{b}|<1$
\bea\label{eq:D4oneloop-qgamma}
\mathcal{Z}^{\D4\tbar\D4}_{\text{1-loop}}(x,A\,|\,x',A)&=\exp\left(-\sum_{n=1}^{\infty}\frac{1}{n}\frac{\bfP_{\bar{A}}^{[n]}\bfP_{\bar{A}}^{[-n]}}{\bfP_{\four}^{[n]}}\left(\frac{x}{x'}\right)^{n}\right)\\
&=\frac{(x/x';q_{a},q_{b})_{\infty}(q_{c}q_{d}x/x';q_{a},q_{b})_{\infty}}{(q_{c}x/x';q_{a},q_{b})_{\infty}(q_{d}x/x';q_{a},q_{b})_{\infty}}.
\eea
Other formulas in other analytic regions are written using the reflection formula in \eqref{eq:app-qPochreflec}.

Similarly, for the $\D6$-case with $|q_{i}|<1\,\,(i\in\bar{a})$, we have 
\bea\label{eq:D6oneloop-qgamma}
    \mathcal{Z}^{\D6\tbar\D6}_{\text{1-loop}}(x,\bar{a}\,|\,x',\bar{a})&=\exp\left(-\sum_{n>0}\frac{1}{n}\frac{\bfP_{a}^{[n]}\bfP_{a}^{[-n]}}{\bfP_{\four}^{[n]}}\left(\frac{x}{x'}\right)^{n}\right)\\
    &=\frac{(x/x';q_{b},q_{c},q_{d})_{\infty}}{(q_{a}x/x';q_{b},q_{c},q_{d})_{\infty}}
\eea
where $b,c,d\in\bar{a}$.

For the $\D8$-case, we have to take care of the analytic region. Since we are imposing the Calabi--Yau condition $q_{1}q_{2}q_{3}q_{4}=1$, we can not simply take all of the parameters in the same analytic region. We choose $|q_{1}|,|q_{2}|,|q_{3}|<1$ and $|q_{4}|>1$. For simplicity, we set $K_{1},K_{2}=0$. The one-loop perturbative factor is then written as
\bea\label{eq:D8oneloop-qgamma}
\mathcal{Z}_{\text{1-loop}}^{\D8\tbar\D8}(x,0|x',0)&=\exp\left(-\sum_{n>0}\frac{1}{n}\frac{1}{\bfP_{\four}^{[n]}}\left(\frac{x}{x'}\right)^{n}\right)\\
&=\exp\left(\sum_{n>0}\frac{1}{n}\frac{q_{4}^{-n}}{(1-q_{1}^{n})(1-q_{2}^{n})(1-q_{3}^{n})(1-q_{4}^{-n})}\left(\frac{x}{x'}\right)^{n}\right)\\
&=(q_{4}^{-1}x/x';q_{1},q_{2},q_{3},q_{4}^{-1})_{\infty}^{-1}.
\eea
Other analytic regions can be obtained similarly using the reflection formula \eqref{eq:app-qPochreflec}.

\subsection{Elliptic Formulas}\label{sec:ellipticformula}
Let us also introduce the elliptic formulas used in the main section:
\bea\label{eq:ellipticthetadef}
\theta(x;p)&=(x;p)_{\infty}(px^{-1};p)_{\infty}=\exp\left(-\sum_{m\neq0}\frac{x^{m}}{m(1-p^{m})}\right),\quad |p|<1.
\eea
We have the following reflection property
\beq
\theta(x^{-1};p)=-x^{-1}\theta(x;p).
\eeq
To study the commutativity of the elliptic $qq$-characters and the elliptic deformed screening currents, one can use the following formulas.
\begin{theorem}
    \bea
    \delta(x)&=\frac{1}{1-x}+\frac{x^{-1}}{1-x^{-1}},\quad 
    \frac{\delta(x)}{(p;p)^{2}_{\infty}}&=\frac{1}{\theta(x;p)}+\frac{x^{-1}}{\theta(x^{-1};p)}
    \eea
\end{theorem}
\begin{proof}
The first equation is well known. The second one is obtained using the first one as
\bea
    \frac{1}{\theta(x;p)}&=\frac{1}{(x;p)_{\infty}(px^{-1};p)_{\infty}}=\frac{1}{(1-x)(x;p)_{\infty}(px^{-1};p)_{\infty}}\\
&=\left(\delta(x)-\frac{x^{-1}}{1-x^{-1}}\right)\frac{1}{(xp;p)_{\infty}(px^{-1};p)_{\infty}}=\frac{\delta(x)}{(p;p)^{2}_{\infty}}-\frac{x^{-1}}{\theta(x^{-1};p)}.
\eea
\end{proof}

\section{\texorpdfstring{$\U(1)$}{U(1)} partition functions and structure functions}\label{app:U1partition}
In this section, we summarize some properties of the functions frequently used in the main text. We also give the explicit formulas of the $\U(1)$ partition functions for the $\D2$, $\D4$, $\D6$ theories.
\subsection{Structure functions}\label{app:structurefunct}
The structure functions in \eqref{eq:struct_funct} were defined as 
\bea\label{app-eq:struct_funct}
\mathscr{V}_{a}(x)&=\mathbb{I}[-\bfP_{a}^{\vee}x^{\vee}]=\frac{1-q_{a}x}{1-x},\\
\mathscr{S}_{ab}(x)&=\mathbb{I}[-\bfP_{ab}^{\vee}x^{\vee}]=\frac{(1-q_{a}x)(1-q_{b}x)}{(1-x)(1-q_{a}q_{b}x)},\\
g_{\bar{a}}(x)&=\mathbb{I}[-\bfP_{\bar{a}}^{\vee}x^{\vee}]=\frac{\prod_{i\neq a}(1-q_{i}x)(1-q_{\bar{a}}x)}{(1-x)\prod_{i\neq a}(1-q_{a}^{-1}q_{i}^{-1}x)},\\
\mathcal{A}_{\mathbb{C}^{4}}(x)&=\mathbb{I}[-\bfP_{\four}^{\vee}x^{\vee}]=\frac{\prod_{i=1}^{4}(1-q_{i}x)\prod^{4}_{i=1}(1-q_{i}^{-1}x)}{(1-x)^{2}\prod_{i\neq j}(1-q_{i}q_{j}x)}.
\eea
We have the following properties
\bea
    \mathscr{S}_{ab}(x)=\frac{\mathscr{V}_{a}(x)}{\mathscr{V}_{a}(q_{b}x)},\quad g_{abc}(x)=\frac{\mathscr{S}_{ab}(x)}{\mathscr{S}_{ab}(q_{c}x)},\quad \mathcal{A}_{\mathbb{C}^{4}}(x)=\frac{g_{\bar{a}}(x)}{g_{\bar{a}}(q_{a}x)}
\eea
and the reflection formulas
\bea
\mathscr{V}_{a}(x)=q_{a}\mathscr{V}_{a}(q_{a}^{-1}x^{-1})^{-1},\quad \mathscr{S}_{ab}(x)=\mathscr{S}_{ab}(q_{ab}^{-1}x^{-1}),\quad g_{\bar{a}}(x)=g_{\bar{a}}(q_{a}x^{-1})^{-1},\quad \mathcal{A}_{\mathbb{C}^{4}}(x)=\mathcal{A}_{\mathbb{C}^{4}}(x^{-1}).
\eea

The following functions were also defined to study the recursion formulas of the Nekrasov factors for $\D2,\D4,\D6$ theories. They are determined from the index of the observable bundles $\bfY_{a}$ in \eqref{eq:D2Ybundle}, $\bfY_{A}$ in \eqref{eq:D4Ybundle}, $\bfY_{\bar{a}}$ in \eqref{eq:D6Ybundle} for each theory after localization:
\bea\label{eq:app-structurefunct}
\mathscr{U}^{a}_{k,v}(x)&=\mathbb{I}\left[\bfY_{a}^{\vee}x\right]=\left(1-\frac{v}{x}\right)\prod_{\Abox \in k}\mathscr{V}_{a}\left(\frac{\chi_{a,v}(\Bbox)}{x}\right)=\left(1-\frac{vq_{a}^{k}}{x}\right),\\
\mathscr{Y}^{A}_{\lambda,v}(x)&=\mathbb{I}\left[\bfY_{A}^{\vee}x\right]=\left(1-\frac{v}{x}\right)\prod_{\Abox\in\lambda}\mathscr{S}_{A}\left(\frac{\chi_{A,v}(\Bbox)}{x}\right)=\frac{\prod_{\Abox\in A(\lambda)}\left(1-\chi_{A,v}(\Bbox)/x\right)}{\prod_{\Abox\in R(\lambda)} (1-q_{A}\chi_{A,v}(\Bbox)/x)},\\
\mathscr{W}^{\bar{a}}_{\pi,v}(x)&=\mathbb{I}[\bfY_{\bar{a}}^{\vee}x]=\left(1-\frac{v}{x}\right)\prod_{\scube\in \pi}g_{\bar{a}}\left(\frac{\chi_{\bar{a},v}(\cube)}{x}\right)\propto \prod_{\scube\in A(\pi)}\left(1-\frac{\chi_{\bar{a},v}(\cube)}{x}\right)\prod_{\scube\in R(\pi)}\left(1-q_{a}^{-1}\frac{\chi_{\bar{a},v}(\cube)}{x}\right).
\eea
Although we did not use it in the main text, we formally can define a similar function for the $\D8\tbar\overline{\D8}$ case using $\bfY$ in \eqref{eq:D8Ybundle}
\bea
\mathscr{M}^{K}_{\rho,v}(x)=\mathbb{I}[\bfY^{\vee}x]=\frac{(1-v/x)}{(1-Kv/x)}\prod_{\shcube\in\rho}\mathcal{A}_{\mathbb{C}^{4}}\left(\frac{\chi_{\four,v}(\hcube)}{x}\right).
\eea

We can also define the dual functions as
\bea\label{eq:app-dualfunctions}
\mathscr{U}^{a\,\vee}_{k,v}(x)&=\mathbb{I}\left[\bfY_{a}x^{\vee}\right]=\left(1-\frac{x}{v}\right)\prod_{\Abox\in k}\mathscr{V}_{a}\left(q_{a}^{-1}\frac{x}{\chi_{a,v}(\Bbox)}\right)^{-1}=\left(1-\frac{x}{vq_{a}^{k}}\right),\\
\mathscr{Y}^{A\,\vee}_{\lambda,v}(x)&=\mathbb{I}\left[\bfY_{A}x^{\vee}\right]=\left(1-\frac{x}{v}\right)\prod_{\Abox\in \lambda}\mathscr{S}_{A}\left(q_{A}^{-1}\frac{x}{\chi_{A,v}(\Bbox)}\right),\\
\mathscr{W}^{\bar{a}\,\vee}_{\pi,v}(x)&=\mathbb{I}\left[\bfY_{\bar{a}}x^{\vee}\right]=\left(1-\frac{x}{v}\right)\prod_{\scube\in\pi}g_{\bar{a}}\left(q_{a}\frac{x}{\chi_{\bar{a},v}(\cube)}\right)^{-1},\\
\mathscr{M}^{K\,\vee}_{\rho,v}(x)&=\mathbb{I}[\bfY_{\four}x^{\vee}]=\frac{(1-x/v)}{(1-K^{-1}x/v)}\prod_{\shcube\in\rho}\mathcal{A}_{\mathbb{C}^{4}}\left(\frac{x}{\chi_{\four,v}(\hcube)}\right).
\eea
Note that we have the following reflection formulas 
\bea\label{eq:app-dualreflection}
\mathscr{U}^{a\,\vee}_{k,v}(x)=\left(-\frac{x}{vq_{a}^{k}}\right)\mathscr{U}^{a}_{k,v}(x),&\qquad
\mathscr{Y}^{A\,\vee}_{\lambda,v}(x)=\left(-\frac{x}{v}\right)\mathscr{Y}^{A}_{\lambda,v}(x),\\
\mathscr{W}^{\bar{a}\,\vee}_{\pi,v}(x)=\left(-\frac{x}{v}\right)\mathscr{W}^{\bar{a}}_{\pi,v}(x),&\qquad
\mathscr{M}^{K\,\vee}_{\rho,v}(x)=K\mathscr{M}^{K}_{\rho,v}(x).
\eea

\subsection{Examples of structure functions}
Let us list some explicit formulas for low levels. For the $\D2$ case, one can easily use the formula in \eqref{eq:app-structurefunct} so we omit it.

\paragraph{$\D4$ structure function}
We focus on $\mathscr{Y}^{12}_{\lambda,v}(x)$, and other cases are simply obtained by the symmetries of the parameters. The 2d partitions are denoted as \bea\{\lambda_{1},\lambda_{2},\ldots,\lambda_{\ell(\lambda)}\}\eea
where the $q$-coordinates are
\begin{equation}
    vq_{1}^{i-1}q_{2}^{j-1},\quad 1\leq i\leq \ell(\lambda),\quad 1\leq j\leq \lambda_{i}.
\end{equation}
Up to level $|\lambda|=3$, we have the following functions.
\begin{align*}
    \begin{array}{|c|c|}\hline
       \text{ Young diagram $\lambda$ }  & \,\,\mathscr{Y}^{12}_{\lambda,v}(x)\,\, \rule[-4mm]{0mm}{10mm}\\
    \hline\hline  \emptyset  &  \left(1-v/x\right)      \rule[-4mm]{0mm}{10mm}\\ 
    \hline \{1\}  & \frac{\left(1-q_1 v/x\right) \left(1-q_2 v/x\right)}{\left(1-q_1 q_2
   v/x\right)}\rule[-4mm]{0mm}{10mm} \\
   \hline \{1,1\}   &  \frac{\left(1-q_1^2 v/x\right) \left(1-q_2 v/x\right)}{\left(1-q_1^2 q_2
   v/x\right)} \rule[-4mm]{0mm}{10mm} \\
   \{2\}    &   \frac{\left(1-q_1 v/x\right) \left(1-q_2^2 v/x\right)}{\left(1-q_1 q_2^2
   v/x\right)} \rule[-4mm]{0mm}{10mm} \\
   \hline \{1,1,1\}   &  \frac{\left(1-q_1^3 v/x\right) \left(1-q_2 v/x\right)}{\left(1-q_1^3 q_2
   v/x\right)}\rule[-4mm]{0mm}{10mm}\\
   \{2,1\}  & \frac{\left(1-q_1^2 v/x\right) \left(1-q_1 q_2 v/x\right)
   \left(1-q_2^2 v/x\right)}{\left(1-q_1^2 q_2 v/x\right)   \left(1-q_1 q_2^2 v/x\right)} \rule[-4mm]{0mm}{10mm} \\
   \{3\}   & \frac{\left(1-q_1 v/x\right) \left(1-q_2^3 v/x\right)}{\left(1-q_1 q_2^3
   v/x\right)} \rule[-4mm]{0mm}{10mm} \\ \hline
    \end{array}
\end{align*}

\paragraph{$\D6$ structure function}
For the $\D6$-case, we focus only on $\mathscr{W}^{\bar{4}}_{\pi,v}(x)$. The plane partitions are denoted as lists of numbers as
\begin{equation}
    \{\{\pi_{11},\pi_{12},\ldots \pi_{1n_{1}} \},\{\pi_{21},\pi_{22},\ldots \pi_{2n_{2}}\}\ldots \{\pi_{m1},\pi_{m2},\ldots \pi_{mn_{m}}\}\}
\end{equation}
where we used the $(2,1)$-type description. The $q$-coordinates of the boxes in the plane partition will be 
\begin{equation}
    vq_{1}^{i-1}q_{2}^{j-1}q_{3}^{k-1},\quad 1\leq i\leq m,\quad 1\leq j\leq n_{i},\quad 1\leq k\leq \pi_{ij}.
\end{equation} Denoting the number of boxes as $k=|\pi|$, we have the following functions up to level $k=3$.
\begin{align*}
\renewcommand\arraystretch{3.1}{
\begin{array}{|c|c|}\hline
\raisebox{5pt}{\text{ plane partition $\pi$ }}&\,\,\raisebox{5pt}{$\mathscr{W}^{\bar{4}}_{\pi,v}\left(x\right)$}\,\,\\
\hline\hline \raisebox{5pt}{$\emptyset$} & \raisebox{5pt}{$\left(1-\frac{v}{x}\right)$}\\
\hline\raisebox{5pt}{$ \{\{1\}\}$} & \raisebox{5pt}{$\frac{\left(1-\frac{q_1 v}{x}\right) \left(1-\frac{q_2 v}{x}\right) \left(1-\frac{q_3
   v}{x}\right) \left(1-\frac{q_1 q_2 q_3 v}{x}\right)}{\left(1-\frac{q_1 q_2
   v}{x}\right) \left(1-\frac{q_1 q_3 v}{x}\right) \left(1-\frac{q_2 q_3 v}{x}\right)}$}\\
\hline \{\{1\},\{1\}\} & \raisebox{3pt}{$\frac{\left(1-\frac{q_1^2 v}{x}\right) \left(1-\frac{q_2 v}{x}\right) \left(1-\frac{q_3
   v}{x}\right) \left(1-\frac{q_1^2 q_2 q_3 v}{x}\right)}{\left(1-\frac{q_1^2 q_2
   v}{x}\right) \left(1-\frac{q_1^2 q_3 v}{x}\right) \left(1-\frac{q_2 q_3 v}{x}\right)}$}\\
   \{\{2\}\} & \raisebox{3pt}{$\frac{\left(1-\frac{q_1 v}{x}\right) \left(1-\frac{q_2 v}{x}\right) \left(1-\frac{q_3^2
   v}{x}\right) \left(1-\frac{q_1 q_2 q_3^2 v}{x}\right)}{\left(1-\frac{q_1 q_2
   v}{x}\right) \left(1-\frac{q_1 q_3^2 v}{x}\right) \left(1-\frac{q_2 q_3^2
   v}{x}\right)}$}\\
   \{\{1,1\}\} & \raisebox{3pt}{$\frac{\left(1-\frac{q_1 v}{x}\right) \left(1-\frac{q_2^2 v}{x}\right) \left(1-\frac{q_3
   v}{x}\right) \left(1-\frac{q_1 q_2^2 q_3 v}{x}\right)}{\left(1-\frac{q_1 q_2^2
   v}{x}\right) \left(1-\frac{q_1 q_3 v}{x}\right) \left(1-\frac{q_2^2 q_3 v}{x}\right)}$}\\
   \hline \{\{3\}\} & \frac{\left(1-\frac{q_1 v}{x}\right) \left(1-\frac{q_2 v}{x}\right) \left(1-\frac{q_3^3
   v}{x}\right) \left(1-\frac{q_1 q_2 q_3^3 v}{x}\right)}{\left(1-\frac{q_1 q_2
   v}{x}\right) \left(1-\frac{q_1 q_3^3 v}{x}\right) \left(1-\frac{q_2 q_3^3
   v}{x}\right)}\\
   \{\{2,1\}\} & \frac{\left(1-\frac{q_1 v}{x}\right) \left(1-\frac{q_2^2 v}{x}\right) \left(1-\frac{q_2
   q_3 v}{x}\right) \left(1-\frac{q_1 q_2^2 q_3 v}{x}\right) \left(1-\frac{q_3^2
   v}{x}\right) \left(1-\frac{q_1 q_2 q_3^2 v}{x}\right)}{\left(1-\frac{q_1 q_2^2
   v}{x}\right) \left(1-\frac{q_1 q_2 q_3 v}{x}\right) \left(1-\frac{q_2^2 q_3
   v}{x}\right) \left(1-\frac{q_1 q_3^2 v}{x}\right) \left(1-\frac{q_2 q_3^2
   v}{x}\right)}\\
   \{\{1,1,1\}\} & \frac{\left(1-\frac{q_1 v}{x}\right) \left(1-\frac{q_2^3 v}{x}\right) \left(1-\frac{q_3
   v}{x}\right) \left(1-\frac{q_1 q_2^3 q_3 v}{x}\right)}{\left(1-\frac{q_1 q_2^3
   v}{x}\right) \left(1-\frac{q_1 q_3 v}{x}\right) \left(1-\frac{q_2^3 q_3 v}{x}\right)}\\
   \{\{2\},\{1\}\} & \frac{\left(1-\frac{q_1^2 v}{x}\right) \left(1-\frac{q_2 v}{x}\right) \left(1-\frac{q_1
   q_3 v}{x}\right) \left(1-\frac{q_1^2 q_2 q_3 v}{x}\right) \left(1-\frac{q_3^2
   v}{x}\right) \left(1-\frac{q_1 q_2 q_3^2 v}{x}\right)}{\left(1-\frac{q_1^2 q_2
   v}{x}\right) \left(1-\frac{q_1^2 q_3 v}{x}\right) \left(1-\frac{q_1 q_2 q_3
   v}{x}\right) \left(1-\frac{q_1 q_3^2 v}{x}\right) \left(1-\frac{q_2 q_3^2
   v}{x}\right)}\\
   \{\{1,1\},\{1\}\} & \frac{\left(1-\frac{q_1^2 v}{x}\right) \left(1-\frac{q_1 q_2 v}{x}\right)
   \left(1-\frac{q_2^2 v}{x}\right) \left(1-\frac{q_3 v}{x}\right) \left(1-\frac{q_1^2
   q_2 q_3 v}{x}\right) \left(1-\frac{q_1 q_2^2 q_3 v}{x}\right)}{\left(1-\frac{q_1^2 q_2
   v}{x}\right) \left(1-\frac{q_1 q_2^2 v}{x}\right) \left(1-\frac{q_1^2 q_3 v}{x}\right)
   \left(1-\frac{q_1 q_2 q_3 v}{x}\right) \left(1-\frac{q_2^2 q_3 v}{x}\right)}\\
   \{\{1\},\{1\},\{1\}\} &\raisebox{3pt}{$\frac{\left(1-\frac{q_1^3 v}{x}\right) \left(1-\frac{q_2 v}{x}\right) \left(1-\frac{q_3
   v}{x}\right) \left(1-\frac{q_1^3 q_2 q_3 v}{x}\right)}{\left(1-\frac{q_1^3 q_2
   v}{x}\right) \left(1-\frac{q_1^3 q_3 v}{x}\right) \left(1-\frac{q_2 q_3 v}{x}\right)}$}\\
   \hline
\end{array}
}
\end{align*}

\subsection{D6 partition functions}\label{app:D6U1partitionfunction}
The partition function of the 7d $\U(1)$ theory on $\mathbb{C}^{3}_{\bar{a}}\times \mathbb{S}^{1}$ is given from
\bea
\mathcal{Z}^{\D6}_{\bar{a}}[\pi]=\frac{1}{\mathsf{N}_{\bar{a}}(v,\pi\,|\,v,\pi)},\quad \widetilde{Z}^{\D6}_{\bar{a}}[\pi]=\prod_{\scube\in\pi}\left(-\frac{q_{a}v}{\chi_{\bar{a},v}(\cube)}\right)\mathcal{Z}^{\D6}_{\bar{a}}[\pi],
\eea
where 
\beq
\mathsf{N}_{\bar{a}}(v,\pi\,|\,v,\pi)=\prod_{\scube\in\pi}\frac{(1-v/\chi_{\bar{a},v}(\cube))}{(1-q_{a}^{-1}\chi_{\bar{a},v}(\cube)/v)}\prod_{\substack{\scube\in\pi\\\scubeF\in\pi}}g_{\bar{a}}\left(\frac{\chi_{\bar{a},v}(\cube)}{\chi_{\bar{a},v}(\cubeF)}\right).
\eeq
For simplicity, we only consider $a=4$ and $\mathcal{Z}^{\D6}_{\bar{4}}[\pi]$. Note that this comes from the following character:
\bea\label{eq:app-D6U1character}
    &\mathbf{v}_{\text{inst}}=-\bfN_{123}^{\vee}\bfK_{123}+q_{4}\bfK_{123}^{\vee}\bfN_{123}+\bfP^{\vee}_{123}\bfK_{123}^{\vee}\bfK_{123},\\
    &\bfN_{123}=v,\quad \bfK_{123}=\sum_{\scube\,\in\pi}\chi_{123,v}(\cube)=\sum_{(x,y,z)\in\pi}vq_{1}^{x-1}q_{2}^{y-1}q_{3}^{z-1}.
\eea

The generating function of the plane partition is referred to as the MacMahon function, which is given as 
\bea
Z_{3}(\mathfrak{q})=\sum_{\pi}\mathfrak{q}^{|\pi|}=\prod_{n=1}^{\infty}\frac{1}{(1-\mathfrak{q}^{n})^{n}}=1+\mathfrak{q}+3\mathfrak{q}^{2}+6\mathfrak{q}^{3}+13\mathfrak{q}^{4}+\cdots .
\eea
For levels of instanton $k=|\pi|\leq 3$, the $\U(1)$ contribution is given as the following.  
\begin{itemize}
\item $k=0,\,\,\pi=\emptyset $:
\beq
1
\eeq
    \item $k=1,\,\,\pi=\{\{1\}\}$:
    \beq
        \frac{(1-q_{1}q_{2})(1-q_{1}q_{3})(1-q_{2}q_{3})}{(1-q_{1})(1-q_{2})(1-q_{3})}
    \eeq
    \item $k=2$
    \bea
        &\pi=\{\{1\},\{1\}\}:\quad \frac{q_{1}(1-q_{1}q_{2})(1-q_{1}^{2}q_{2})(1-q_{1}q_{3})(1-q_{1}^{2}q_{3})(q_{1}-q_{2}q_{3})}{(1-q_{1})(1-q_{1}^2)(1-q_{2})(q_{1}-q_{2})(1-q_{3})(q_{1}-q_{3})},\\
        &\pi=\{\{2\}\}:\quad \frac{(1 - q_{1} q_{2}) q_{3} (-q_{1} q_{2} + q_{3}) (1 - q_{1} q_{3}) (1 - q_{2} q_{3}) (1 - q_{1} q_{3}^{2}) (1 - q_{2} q_{3}^{2})}{(1 - q_{1}) (1 - q_{2}) (1 - q_{3}) (-q_{1} + q_{3}) (-q_{2} + q_{3}) (1 - q_{3}^{2})},\\
        &\pi=\{\{1,1,\}\}:\quad \frac{q_{2} (1 - q_{1}q_{2}) (1 - q_{1}q_{2}^{2}) (1 - q_{1}q_{3}) (q_{2} - q_{1}q_{3}) (1 - q_{2}q_{3}) (1 - q_{2}^{2}q_{3})}{(1 - q_{1})(1 - q_{2})(-q_{1} + q_{2})(1 - q_{2}^{2})(1 - q_{3})(q_{2} - q_{3})}
    \eea
    \item $k=3$
    \bea
        &\pi=\{\{3\}\}:\quad \\
        &\frac{(1 - q_{1}q_{2})q_{3}^{3}(-q_{1}q_{2} + q_{3})(1 - q_{1}q_{3})(1 - q_{2}q_{3})(-q_{1}q_{2} + q_{3}^{2})(1 - q_{1}q_{3}^{2})(1 - q_{2}q_{3}^{2})(1 - q_{1}q_{3}^{3})(1 - q_{2}q_{3}^{3})}{(1 - q_{1})(1 - q_{2})(1 - q_{3})(-q_{1} + q_{3})(-q_{2} + q_{3})(1 - q_{3}^{2})(-q_{1} + q_{3}^{2})(-q_{2} + q_{3}^{2})(1 - q_{3}^{3})}\\
        &\pi=\{\{2,1\}\}:\quad\\
        &\frac{q_{2} \left(1-q_{1} q_{2}\right)^{2} q_{3} \left(-q_{1} q_{2}^{2}+q_{3}\right) \left(1-q_{1} q_{3}\right)^{2} \left(1-q_{2} q_{3}\right) \left(1-q_{2}^{2} q_{3}\right) \left(q_{2}-q_{1} q_{3}^{2}\right)}{\left(1-q_{1}\right) \left(1-q_{2}\right)^{2} \left(-q_{1}+q_{2}\right) \left(1-q_{3}\right)^{2} \left(-q_{1}+q_{3}\right) \left(-q_{2}^{2}+q_{3}\right) \left(q_{2}-q_{3}^{2}\right)}\\
        &\pi=\{\{1,1,1\}\}:\\
        &\frac{q_{2}^{3} (1 - q_{1} q_{2}) (1 - q_{1} q_{2}^{2}) (1 - q_{1} q_{2}^{3}) (1 - q_{1} q_{3}) (q_{2} - q_{1} q_{3}) (q_{2}^{2} - q_{1} q_{3}) (1 - q_{2} q_{3}) (1 - q_{2}^{2} q_{3}) (1 - q_{2}^{3} q_{3})}{(1 - q_{1}) (1 - q_{2}) (-q_{1} + q_{2}) (1 - q_{2}^{2}) (-q_{1} + q_{2}^{2}) (1 - q_{2}^{3}) (1 - q_{3}) (q_{2} - q_{3}) (q_{2}^{2} - q_{3})}\nonumber
\eea
\bea
&\pi=\{\{2\},\{1\}\}:\\
&\frac{q_1 (1 - q_1 q_2)^2 q_3 (-q_1^2 q_2 + q_3) (1 - q_1 q_3) (1 - q_1^2 q_3) (1 - 
   q_2 q_3)^2 (1 - q_1 q_3^2) (q_1 - q_2 q_3^2)}{(1 - q_1)^2 (1 - q_2) (q_1 - 
   q_2) (1 - q_3)^2 (-q_1^2 + q_3) (-q_2 + q_3) (q_1 - q_3^2)} \\
&\pi=\{\{1,1\}\{1\}\}:\\
&\frac{q_{1} q_{2} (1 - q_{1} q_{2}) (1 - q_{1}^{2} q_{2}) (1 - q_{1} q_{2}^{2}) (1 - q_{1} q_{3})^{2} (q_{2} - q_{1}^{2} q_{3}) (1 - q_{2} q_{3})^{2} (q_{1} - q_{2}^{2} q_{3})}{(1 - q_{1})^{2} (1 - q_{2})^{2} (-q_{1}^{2} + q_{2}) (q_{1} - q_{2}^{2}) (1 - q_{3}) (q_{1} - q_{3}) (q_{2} - q_{3})}\\
&\pi=\{\{1\},\{1\},\{1\}\}:\\
&\frac{q_{1}^{3} (1 - q_{1} q_{2}) (1 - q_{1}^{2} q_{2}) (1 - q_{1}^{3} q_{2}) (1 - q_{1} q_{3}) (1 - q_{1}^{2} q_{3}) (1 - q_{1}^{3} q_{3}) (1 - q_{2} q_{3}) (q_{1} - q_{2} q_{3}) (q_{1}^{2} - q_{2} q_{3})}{(1 - q_{1}) (1 - q_{1}^{2}) (1 - q_{1}^{3}) (1 - q_{2}) (q_{1} - q_{2}) (q_{1}^{2} - q_{2}) (1 - q_{3}) (q_{1} - q_{3}) (q_{1}^{2} - q_{3})}
\eea
\end{itemize}
Therefore, in the simplest case when there is only one box in the plane partition, we have 
\bea\label{eq:app-D6level1}
\mathcal{Z}_{\bar{4}}^{\D6}[\{\{1\}\}]=\frac{(1-q_{1}q_{2})(1-q_{1}q_{3})(1-q_{2}q_{3})}{(1-q_{1})(1-q_{2})(1-q_{3})},\quad \widetilde{\mathcal{Z}}_{\bar{4}}^{\D6}[\{\{1\}\}]=-q_{4}\frac{(1-q_{1}q_{2})(1-q_{1}q_{3})(1-q_{2}q_{3})}{(1-q_{1})(1-q_{2})(1-q_{3})}.
\eea

Under the NS limit $q_{4}\rightarrow 1$ and $q_{1,2,3}\rightarrow \sfq_{1,2,3}$, the partition function becomes trivial: $\widetilde{\mathcal{Z}}_{\bar{4}}^{\D6}[\pi]=1$. For example, we have
\bea
\widetilde{\mathcal{Z}}_{\bar{4}}^{\D6}[\{\{1\}\}]\,\,\longrightarrow\,\,-\frac{(1-\sfq_{1}\sfq_{2})(1-\sfq_{1}\sfq_{3})(1-\sfq_{2}\sfq_{3})}{(1-\sfq_{1})(1-\sfq_{2})(1-\sfq_{3})}=1.
\eea
Hence, under this limit the total $\U(1)$ partition function becomes the MacMahon function:
\bea
\mathcal{Z}^{\D6}_{\bar{4},\text{inst.}}[\U(1)]=\sum_{\pi}\mathfrak{q}^{|\pi|}\widetilde{\mathcal{Z}}^{\D6}_{\bar{4}}[\pi]\,\,\xrightarrow{q_{4}\rightarrow 1}\,\, Z_{3}(\mathfrak{q})=\sum_{\pi}\mathfrak{q}^{|\pi|}.
\eea

\paragraph{Recursion formulas}
For low levels, using the explicit form of the partition functions, we can write down the recursion formula for $\mathcal{Z}^{\D6}_{\bar{4}}[\pi]$ as
\begin{align*}
\begin{array}{|c|c|c|c|}\hline
\text{ $\pi$ } & \,\,\pi+\cube\,\,&\text{ $\chi_{\bar{4},v}(\cube)$ }& \,\,\mathcal{Z}^{\D6}_{\bar{4}}[\pi+\cube]/\mathcal{Z}^{\D6}_{\bar{4}}[\pi] \rule[-4mm]{0mm}{10mm}\\
\hline\hline\emptyset&\{\{1\}\} & v &  \frac{(1-q_{1}q_{2})(1-q_{1}q_{3})(1-q_{2}q_{3})}{(1-q_{1})(1-q_{2})(1-q_{3})}\rule[-4mm]{0mm}{10mm}\\
\hline \{\{1\}\} &\{\{1\},\{1\}\} & vq_{1} & \frac{q_1 \left(1-q_1^2 q_2\right) \left(1-q_1^2 q_3\right) \left(q_1-q_2
   q_3\right)}{\left(1-q_1^2\right) \left(q_1-q_2\right) \left(q_1-q_3\right)}\rule[-4mm]{0mm}{10mm}\\
   \hline \{\{1\}\} & \{\{2\}\}  & vq_{3} & \frac{q_3 \left(q_3-q_1 q_2\right) \left(1-q_1 q_3^2\right) \left(1-q_2
   q_3^2\right)}{\left(q_3-q_1\right) \left(q_3-q_2\right) \left(1-q_3^2\right)}\rule[-4mm]{0mm}{10mm}\\
   \hline \{\{1\}\} &\{\{1,1\}\} &vq_{2} &\frac{q_2 \left(1-q_1 q_2^2\right) \left(q_2-q_1 q_3\right) \left(1-q_2^2
   q_3\right)}{\left(q_2-q_1\right) \left(1-q_2^2\right) \left(q_2-q_3\right)}\rule[-4mm]{0mm}{10mm}\\
  \hline \{\{1\},\{1\}\} & \{\{1\},\{1\},\{1\}\}& vq_{1}^{2}&\frac{q_1^2 \left(1-q_1^3 q_2\right) \left(1-q_1^3 q_3\right) \left(q_1^2-q_2
   q_3\right)}{\left(1-q_1^3\right) \left(q_1^2-q_2\right) \left(q_1^2-q_3\right)}\rule[-4mm]{0mm}{10mm}\\
  \hline\{\{1\},\{1\}\} &\{\{1,1\},\{1\}\}& vq_{2}&\frac{\left(1-q_1^2\right) \left(q_1-q_2\right) q_2 \left(1-q_1 q_2^2\right) \left(1-q_1
   q_3\right) \left(q_2-q_1^2 q_3\right) \left(1-q_2 q_3\right) \left(q_1-q_2^2
   q_3\right)}{\left(1-q_1\right) \left(1-q_2\right) \left(q_2-q_1^2\right)
   \left(q_1-q_2^2\right) \left(q_2-q_3\right) \left(1-q_1^2 q_3\right) \left(q_1-q_2
   q_3\right)}\rule[-4mm]{0mm}{10mm}\\
  \hline \{\{1\},\{1\}\}&\{\{2\},\{1\}\}& vq_{3}&\frac{\left(1-q_1^2\right) \left(1-q_1 q_2\right) \left(q_1-q_3\right) q_3
   \left(q_3-q_1^2 q_2\right) \left(1-q_2 q_3\right) \left(1-q_1 q_3^2\right)
   \left(q_1-q_2 q_3^2\right)}{\left(1-q_1\right) \left(1-q_1^2 q_2\right)
   \left(1-q_3\right) \left(q_3-q_1^2\right) \left(q_3-q_2\right) \left(q_1-q_2
   q_3\right) \left(q_1-q_3^2\right)}\rule[-4mm]{0mm}{10mm}\\
  \hline \{\{2\}\} & \{\{2\},\{1\}\}& vq_{1}&\frac{q_1 \left(1-q_1 q_2\right) \left(q_3-q_1\right) \left(q_3-q_1^2 q_2\right)
   \left(1-q_1^2 q_3\right) \left(1-q_2 q_3\right) \left(1-q_3^2\right) \left(q_1-q_2
   q_3^2\right)}{\left(1-q_1\right) \left(q_1-q_2\right) \left(1-q_3\right)
   \left(q_3-q_1^2\right) \left(q_3-q_1 q_2\right) \left(q_1-q_3^2\right) \left(1-q_2
   q_3^2\right)}\rule[-4mm]{0mm}{10mm}\\
  \hline\{\{2\}\} &\{\{2,1\}\}& vq_{2}&\frac{q_2 \left(1-q_1 q_2\right) \left(q_3-q_2\right) \left(q_3-q_1 q_2^2\right)
   \left(1-q_1 q_3\right) \left(1-q_2^2 q_3\right) \left(1-q_3^2\right) \left(q_2-q_1
   q_3^2\right)}{\left(1-q_2\right) \left(q_2-q_1\right) \left(1-q_3\right) \left(q_3-q_1
   q_2\right) \left(q_3-q_2^2\right) \left(q_2-q_3^2\right) \left(1-q_1 q_3^2\right)}\rule[-4mm]{0mm}{10mm}\\
  \hline \{\{2\}\}&\{\{3\}\}& vq_{3}^{2}&\frac{q_3^2 \left(q_3^2-q_1 q_2\right) \left(1-q_1 q_3^3\right) \left(1-q_2
   q_3^3\right)}{\left(q_3^2-q_1\right) \left(q_3^2-q_2\right) \left(1-q_3^3\right)}\rule[-4mm]{0mm}{10mm}\\
  \hline \{\{1,1\}\} & \{\{1,1,1\}\}& vq_{2}^{2}&\frac{q_2^2 \left(1-q_1 q_2^3\right) \left(q_2^2-q_1 q_3\right) \left(1-q_2^3
   q_3\right)}{\left(q_2^2-q_1\right) \left(1-q_2^3\right) \left(q_2^2-q_3\right)}\rule[-4mm]{0mm}{10mm}\\
  \hline\{\{1,1\}\} &\{\{1,1\},\{1\}\}& vq_{1}&\frac{q_1 \left(q_2-q_1\right) \left(1-q_1^2 q_2\right) \left(1-q_2^2\right) \left(1-q_1
   q_3\right) \left(q_2-q_1^2 q_3\right) \left(1-q_2 q_3\right) \left(q_1-q_2^2
   q_3\right)}{\left(1-q_1\right) \left(1-q_2\right) \left(q_2-q_1^2\right)
   \left(q_1-q_2^2\right) \left(q_1-q_3\right) \left(q_2-q_1 q_3\right) \left(1-q_2^2
   q_3\right)}\rule[-4mm]{0mm}{10mm}\\
  \hline \{\{1,1\}\}&\{\{2,1\}\}& vq_{3}&\frac{\left(1-q_1 q_2\right) \left(1-q_2^2\right) \left(q_2-q_3\right) q_3 \left(q_3-q_1
   q_2^2\right) \left(1-q_1 q_3\right) \left(q_2-q_1 q_3^2\right) \left(1-q_2
   q_3^2\right)}{\left(1-q_2\right) \left(1-q_1 q_2^2\right) \left(1-q_3\right)
   \left(q_3-q_1\right) \left(q_3-q_2^2\right) \left(q_2-q_1 q_3\right)
   \left(q_2-q_3^2\right)}\rule[-4mm]{0mm}{10mm}\\
  \hline
\end{array}
\end{align*}
Using the explicit forms of $\mathscr{W}^{\bar{4}}_{\pi,v}(x)$, we can compute the following function for low levels:
\begin{equation}
    \frac{\underset{x=\chi_{\bar{4},v}(\scube)}{\Res}x^{-1}\mathscr{W}^{\bar{4}}_{\pi,v}(x)^{-1}}{\underset{x=\chi_{\bar{4},v}(\scube)}{\Res}x^{-1}\mathscr{W}^{\bar{4}}_{\pi+\scube,v}(q_{4}^{-1}x)^{-1}}.
\end{equation}
where note that $\underset{x'=q_{4}\chi_{\bar{4},x}(\scube)}{\Res}x'^{-1}\mathscr{W}^{\bar{4}}_{\pi,x}(q_{4}^{-1}x')^{-1}=\underset{x=\chi_{\bar{4},v}(\scube)}{\Res}x^{-1}\mathscr{W}^{\bar{4}}_{\pi,v}(x)^{-1}$. The results are
\begin{align*}
\begin{tabular}{|c|c|c|c|}\hline
 $\pi$  & \,\,$\pi+\cube$\,\,& $\chi_{\bar{4},v}(\cube)$ & $\frac{\underset{x=\chi_{\bar{4},v}(\scube)}{\Res}x^{-1}\mathscr{W}^{\bar{4}}_{\pi,v}(x)^{-1}}{\underset{x=\chi_{\bar{4},v}(\scube)}{\Res}x^{-1}\mathscr{W}^{\bar{4}}_{\pi+\scube,v}(q_{4}^{-1}x)^{-1}}$\rule[0mm]{0mm}{9mm}\\
\hline\hline$\emptyset$&$\{\{1\}\}$ & $v $&  $q_{123}^{-1}\frac{(1-q_{1}q_{2})(1-q_{1}q_{3})(1-q_{2}q_{3})}{(1-q_{1})(1-q_{2})(1-q_{3})}$\rule[-3.5mm]{0mm}{9mm}\\
\hline $\{\{1\}\}$ &$\{\{1\},\{1\}\}$ & $vq_{1}$ & $-\frac{\left(q_1^2 q_2-1\right) \left(q_1^2 q_3-1\right) \left(q_1-q_2 q_3\right)}{q_1
   \left(q_1^2-1\right) \left(q_1-q_2\right) q_2 \left(q_1-q_3\right) q_3}$\rule[-4mm]{0mm}{10mm}\\
   \hline $\{\{1\}\}$ & $\{\{2\}\}$  & $vq_{3}$ & $\frac{\left(q_1 q_2-q_3\right) \left(q_1 q_3^2-1\right) \left(q_2 q_3^2-1\right)}{q_1 q_2
   \left(q_1-q_3\right) \left(q_2-q_3\right) q_3 \left(q_3^2-1\right)}$\rule[-4mm]{0mm}{10mm}\\
   \hline $\{\{1\}\}$ &$\{\{1,1\}\}$ &$vq_{2}$ &$-\frac{\left(q_1 q_2^2-1\right) \left(q_1 q_3-q_2\right) \left(q_2^2 q_3-1\right)}{q_1
   \left(q_1-q_2\right) q_2 \left(q_2^2-1\right) \left(q_2-q_3\right) q_3}$\rule[-4mm]{0mm}{10mm}\\
  \hline $\{\{1\},\{1\}\}$ & $\{\{1\},\{1\},\{1\}\}$& $vq_{1}^{2}$&$-\frac{\left(q_1^3 q_2-1\right) \left(q_1^3 q_3-1\right) \left(q_1^2-q_2 q_3\right)}{q_1
   \left(q_1^3-1\right) \left(q_1^2-q_2\right) q_2 \left(q_1^2-q_3\right) q_3}$\rule[-4mm]{0mm}{10mm}\\
  \hline$\{\{1\},\{1\}\}$ &$\{\{1,1\},\{1\}\}$& $vq_{2}$&$-\frac{\left(q_1^2-1\right) \left(q_1-q_2\right) \left(q_1 q_2^2-1\right) \left(q_1
   q_3-1\right) \left(q_1^2 q_3-q_2\right) \left(q_2 q_3-1\right) \left(q_1-q_2^2
   q_3\right)}{\left(q_1-1\right) q_1 \left(q_1^2-q_2\right) \left(q_2-1\right) q_2
   \left(q_1-q_2^2\right) \left(q_2-q_3\right) q_3 \left(q_1^2 q_3-1\right) \left(q_1-q_2
   q_3\right)}$\rule[-4mm]{0mm}{10mm}\\
  \hline $\{\{1\},\{1\}\}$&$\{\{2\},\{1\}\}$& $vq_{3}$&$\frac{\left(q_1^2-1\right) \left(q_1 q_2-1\right) \left(q_1-q_3\right) \left(q_1^2
   q_2-q_3\right) \left(q_2 q_3-1\right) \left(q_1 q_3^2-1\right) \left(q_1-q_2
   q_3^2\right)}{\left(q_1-1\right) q_1 q_2 \left(q_1^2 q_2-1\right)
   \left(q_1^2-q_3\right) \left(q_2-q_3\right) \left(q_3-1\right) q_3 \left(q_1-q_2
   q_3\right) \left(q_1-q_3^2\right)}$\rule[-4mm]{0mm}{10mm}\\
  \hline $\{\{2\}\}$ & $\{\{2\},\{1\}\}$& $vq_{1}$&$-\frac{\left(q_1 q_2-1\right) \left(q_1-q_3\right) \left(q_1^2 q_2-q_3\right) \left(q_1^2
   q_3-1\right) \left(q_2 q_3-1\right) \left(q_3^2-1\right) \left(q_1-q_2
   q_3^2\right)}{\left(q_1-1\right) q_1 \left(q_1-q_2\right) q_2 \left(q_1^2-q_3\right)
   \left(q_1 q_2-q_3\right) \left(q_3-1\right) q_3 \left(q_1-q_3^2\right) \left(q_2
   q_3^2-1\right)}$\rule[-4mm]{0mm}{10mm}\\
  \hline$\{\{2\}\}$ &$\{\{2,1\}\}$& $vq_{2}$&$-\frac{\left(q_1 q_2-1\right) \left(q_2-q_3\right) \left(q_1 q_2^2-q_3\right) \left(q_1
   q_3-1\right) \left(q_2^2 q_3-1\right) \left(q_3^2-1\right) \left(q_1
   q_3^2-q_2\right)}{q_1 \left(q_1-q_2\right) \left(q_2-1\right) q_2 \left(q_1
   q_2-q_3\right) \left(q_2^2-q_3\right) \left(q_3-1\right) q_3 \left(q_2-q_3^2\right)
   \left(q_1 q_3^2-1\right)}$\rule[-4mm]{0mm}{10mm}\\
  \hline $\{\{2\}\}$&$\{\{3\}\}$& $vq_{3}^{2}$&$\frac{\left(q_1 q_2-q_3^2\right) \left(q_1 q_3^3-1\right) \left(q_2 q_3^3-1\right)}{q_1
   q_2 q_3 \left(q_1-q_3^2\right) \left(q_2-q_3^2\right) \left(q_3^3-1\right)}$\rule[-4mm]{0mm}{10mm}\\
  \hline $\{\{1,1\}\}$ & $\{\{1,1,1\}\}$ &$vq_{2}^{2}  $&$-\frac{\left(q_1 q_2^3-1\right) \left(q_1 q_3-q_2^2\right) \left(q_2^3 q_3-1\right)}{q_1
   q_2 \left(q_1-q_2^2\right) \left(q_2^3-1\right) \left(q_2^2-q_3\right) q_3}$\rule[-4mm]{0mm}{10mm}\\
  \hline$\{\{1,1\}\}$ &$\{\{1,1\},\{1\}\}$&$vq_{1}$ &$-\frac{\left(q_1-q_2\right) \left(q_1^2 q_2-1\right) \left(q_2^2-1\right) \left(q_1
   q_3-1\right) \left(q_1^2 q_3-q_2\right) \left(q_2 q_3-1\right) \left(q_1-q_2^2
   q_3\right)}{\left(q_1-1\right) q_1 \left(q_1^2-q_2\right) \left(q_2-1\right) q_2
   \left(q_1-q_2^2\right) \left(q_1-q_3\right) q_3 \left(q_1 q_3-q_2\right) \left(q_2^2
   q_3-1\right)}$\rule[-4mm]{0mm}{10mm}\\
  \hline $\{\{1,1\}\}$&$\{\{2,1\}\}$&$vq_{3}$ &$\frac{\left(q_1 q_2-1\right) \left(q_2^2-1\right) \left(q_2-q_3\right) \left(q_1
   q_2^2-q_3\right) \left(q_1 q_3-1\right) \left(q_1 q_3^2-q_2\right) \left(q_2
   q_3^2-1\right)}{q_1 \left(q_2-1\right) q_2 \left(q_1 q_2^2-1\right)
   \left(q_1-q_3\right) \left(q_2^2-q_3\right) \left(q_3-1\right) q_3 \left(q_1
   q_3-q_2\right) \left(q_2-q_3^2\right)}$\rule[-4mm]{0mm}{10mm}\\
  \hline
\end{tabular}
\end{align*}
We then have the following relation.
\begin{theorem}\label{eq:app-thm-D6U1recursionformula}
    The recursion formula of the $\U(1)$ partition function $\widetilde{\mathcal{Z}}^{\D6}_{\bar{a}}[\pi]$ is related with the functions $\mathscr{W}^{\bar{a}}_{\pi,v}(x)$:
    \bea
        \frac{\widetilde{\mathcal{Z}}^{\D6}_{\bar{a}}[\pi+\cube]}{\widetilde{\mathcal{Z}}^{\D6}_{\bar{a}}[\pi]}=-\frac{\underset{x'=q_{a}\chi_{\bar{a},v}(\scube)}{\Res}x'^{-1}\mathscr{W}^{\bar{a}}_{\pi,v}(q_{a}^{-1}x')^{-1}}{\underset{x'=\chi_{\bar{a},v}(\scube)}{\Res}x'^{-1}\mathscr{W}^{\bar{a}}_{\pi+\scube,v}(q_{a}^{-1}x')^{-1}}
    \eea
\end{theorem}
\begin{proof}
  The nontrivial part of this recursion formula is the minus sign appearing on the right-hand side. To obtain this minus sign, one needs to take care of the collision terms carefully. A general proof for this recursion formula can be obtained following the discussion in \cite[section 5.1.1]{Kimura:2020jxl}.
\end{proof}
The above formula is written in a simple form as 
\bea\label{eq:app-D6U1recursionformula}
    \frac{\widetilde{\mathcal{Z}}^{\D6}_{\bar{a}}[\pi+\cube]}{\widetilde{\mathcal{Z}}^{\D6}_{\bar{a}}[\pi]}&=-\frac{\underset{x'=q_{a}\chi_{\bar{a},v}(\scube)}{\Res}x'^{-1}\mathscr{W}^{\bar{a}}_{\pi,v}(q_{a}^{-1}x')^{-1}}{\underset{x'=\chi_{\bar{a},v}(\scube)}{\Res}x'^{-1}\mathscr{W}^{\bar{a}}_{\pi+\scube,v}(q_{a}^{-1}x')^{-1}}\\
    &=-\frac{\underset{x'\rightarrow \chi_{\bar{a},v}(\scube)}{\lim}\left(1-\chi_{\bar{a},v}(\cube)/x'\right)\mathscr{W}_{\pi,v}^{\bar{a}}(x')^{-1}}{\underset{x'\rightarrow \chi_{\bar{a},v}(\scube)}{\lim}\left(1-\chi_{\bar{a},v}(\cube)/x'\right)\mathscr{W}_{\pi+\scube,v}^{\bar{a}}(q_{a}^{-1}x')^{-1}}\\
    &=-\frac{\mathscr{W}^{\bar{a}}_{\pi+\scube,v}(q_{a}^{-1}\chi_{\bar{a},v}(\cube))}{\mathscr{W}_{\pi,v}^{\bar{a}}(\chi_{\bar{a},v}(\cube))}.
\eea
The numerator and denominator are both singular and the inserting $\chi_{\bar{a},v}(\cube)$ diverges. However, such singular terms cancel with each other so the total formula itself is well-defined.

\subsection{D4 partition functions}\label{app:D4U1partitionfunction}
The $\U(1)$ contribution of the 5d affine quiver gauge theory is given from 
\begin{equation}
    \mathcal{Z}_{A}^{\D4}[\lambda;q_{\text{inf}(\bar{A})}]=\frac{\mathsf{N}_{A}(q_{\text{inf}(\bar{A})}v,\lambda\,|\,v,\lambda)}{\mathsf{N}_{A}(v,\lambda\,|\,v,\lambda)},\quad \widetilde{\mathcal{Z}}_{A}^{\D4}[\lambda]=q_{\text{inf}(\bar{A})}^{-|\lambda|}\mathcal{Z}_{A}^{\D4}[\lambda;q_{\text{inf}(\bar{A})}],
\end{equation}
where 
\beq
    \mathsf{N}_{A}(v,\lambda\,|\,v,\lambda)=\prod_{\scube\in\lambda}\left(1-\frac{q_{A}\chi_{A,v}(\cube)}{v}\right)\prod_{\scube\in\lambda}\left(1-\frac{v}{\chi_{A,v}(\cube)}\right)\prod_{\substack{\scube\in\lambda\\\scubeF\in\lambda}}\mathscr{S}_{A}\left(\frac{\chi_{A,v}(\cube)}{\chi_{A,v}(\cubeF)}\right).
\eeq
For simplicity, we only consider $A=12$ and $\mathcal{Z}^{\D4}_{12}[\lambda;q_{3}]$. Note that this form comes from the character
\bea
    &\mathbf{v}_{\text{inst}}=-\bfN_{12}^{\vee}\bfP_{3}^{\vee}\bfK_{12}-\bfN_{12} q_{12}^{-1}\bfP_{3}^{\vee}\bfK_{12}^{\vee}+\bfK_{12}^{\vee}\bfP_{123}^{\vee}\bfK_{12},\\
    &\bfN_{12}=v,\quad \bfK_{12}=\sum_{i=1}^{\ell(\lambda)}\sum_{j=1}^{\lambda_{i}}vq_{1}^{i-1}q_{2}^{j-1}.
\eea

The generating function for the 2d partition is given as 
\bea
Z_{2}(\mathfrak{q})=\sum_{\lambda}\mathfrak{q}^{|\lambda|}=\prod_{n=1}^{\infty}\frac{1}{1-\mathfrak{q}^{n}}=1+\mathfrak{q}+2\mathfrak{q}^{2}+3\mathfrak{q}^{3}+5\mathfrak{q}^{4}+\cdots.
\eea
For low levels, the $\U(1)$ contribution $\mathcal{Z}^{\D4}_{12}[\lambda;q_{3}]$ is given as follows.
\begin{itemize}
    \item $k=1,\quad \lambda=\{1\}$:
    \bea
    \frac{\left(1-q_1 q_3\right) \left(1-q_2 q_3\right)}{\left(1-q_1\right)
   \left(1-q_2\right)}
    \eea
    \item $k=2$
    \bea
    \lambda=\{2\}:&\quad\frac{\left(1-q_1 q_3\right) \left(q_2-q_1 q_3\right) \left(1-q_2 q_3\right)
   \left(1-q_2^2 q_3\right)}{\left(1-q_1\right) \left(1-q_2\right) \left(q_2-q_1\right)
   \left(1-q_2^2\right)}\\
    \lambda=\{1,1\}:&\quad\frac{\left(1-q_1 q_3\right) \left(1-q_1^2 q_3\right) \left(1-q_2 q_3\right)
   \left(q_1-q_2 q_3\right)}{\left(1-q_1\right) \left(1-q_1^2\right) \left(1-q_2\right)
   \left(q_1-q_2\right)}
    \eea
    \item $k=3$
    \bea
    \lambda=\{3\}:&\quad\frac{\left(1-q_1 q_3\right) \left(q_2-q_1 q_3\right) \left(q_2^2-q_1 q_3\right)
   \left(1-q_2 q_3\right) \left(1-q_2^2 q_3\right) \left(1-q_2^3
   q_3\right)}{\left(1-q_1\right) \left(1-q_2\right) \left(q_2-q_1\right)
   \left(1-q_2^2\right) \left(q_2^2-q_1\right) \left(1-q_2^3\right)}\\
    \lambda=\{2,1\}:&\quad\frac{\left(1-q_1 q_3\right)^2 \left(q_2-q_1^2 q_3\right) \left(1-q_2 q_3\right)^2
   \left(q_1-q_2^2 q_3\right)}{\left(1-q_1\right)^2 \left(1-q_2\right)^2
   \left(q_2-q_1^2\right) \left(q_1-q_2^2\right)}\\
    \lambda=\{1,1,1\}:&\quad\frac{\left(1-q_1 q_3\right) \left(1-q_1^2 q_3\right) \left(1-q_1^3 q_3\right)
   \left(1-q_2 q_3\right) \left(q_1-q_2 q_3\right) \left(q_1^2-q_2
   q_3\right)}{\left(1-q_1\right) \left(1-q_1^2\right) \left(1-q_1^3\right)
   \left(1-q_2\right) \left(q_1-q_2\right) \left(q_1^2-q_2\right)}
    \eea
\end{itemize}
Therefore, in the simplest case when there is only one instanton, we have 
\bea
    \mathcal{Z}^{\D4}_{12}[\{1\};q_{3}]&=\frac{(1-q_{1}q_{3})(1-q_{2}q_{3})}{(1-q_{1})(1-q_{2})},\\
    \widetilde{\mathcal{Z}}^{\D4}_{12}[\{1\}]&=q_{3}^{-1}\frac{(1-q_{1}q_{3})(1-q_{2}q_{3})}{(1-q_{1})(1-q_{2})}=\frac{(1-q_{1}q_{3})(1-q_{1}q_{4})}{(1-q_{1})(1-q_{1}q_{3}q_{4})}=\mathscr{S}_{34}(q_{1}).
\eea
Note that we have $\mathscr{S}_{34}(q_{1})=\mathscr{S}_{34}(q_{2})$.

Under the NS limit $q_{4}\rightarrow 1$ and $q_{1,2,3}\rightarrow \sfq_{1,2,3}$, the $\U(1)$ instanton partition function trivialize: $\widetilde{\mathcal{Z}}_{12}^{\D4}[\lambda]=1$. For example, we have 
\bea
\widetilde{\mathcal{Z}}^{\D4}_{12}[\{1\}]=\frac{(1-q_{1}q_{3})(1-q_{1}q_{4})}{(1-q_{1})(1-q_{1}q_{3}q_{4})}\,\,\longrightarrow \,\,\frac{(1-\sfq_{1}\sfq_{3})(1-\sfq_{1})}{(1-\sfq_{1})(1-\sfq_{1}\sfq_{3})}=1.
\eea
The $\U(1)$ instanton partition function becomes the generating function of the 2d partition:
\bea
\mathcal{Z}_{12,\text{inst.}}^{\D4}[\U(1)]=\sum_{\lambda}\mathfrak{q}^{|\lambda|}\widetilde{\mathcal{Z}}^{\D4}_{12}[\lambda]\xrightarrow{q_{4}\rightarrow 1} Z_{2}(\mathfrak{q})=\sum_{\lambda}\mathfrak{q}^{|\lambda|}.
\eea

\paragraph{Recursion formulas for $\U(1)$ contribution}
Using the explicit formulas, let us write down the recursion formula for the $\U(1)$ contribution coming from $\mathcal{Z}^{\D4}_{12}[\lambda;q_{3}]$.
\begin{align}
    \begin{array}{|c|c|c|c|}\hline
       \lambda  & \lambda+\Bbox & \chi_{12,v}(\Bbox) & \mathcal{Z}_{12}^{\D4}[\lambda+\Bbox;q_{3}]/\mathcal{Z}_{12}^{\D4}[\lambda;q_{3}] \rule[-4mm]{0mm}{10mm}\\
     \hline\hline  \emptyset  & \{1\}  & v & \frac{(1-q_{1}q_{3})(1-q_{2}q_{3})}{(1-q_{1})(1-q_{2})}\rule[-4mm]{0mm}{10mm}\\
     \hline\{1\}  &  \{1,1\} & vq_{1} &\frac{\left(1-q_1^2 q_3\right) \left(q_1-q_2 q_3\right)}{\left(1-q_1^2\right)
   \left(q_1-q_2\right)}\rule[-4mm]{0mm}{10mm}\\
     \hline \{1\}  & \{2\}  & vq_{2} & \frac{\left(q_2-q_1 q_3\right) \left(1-q_2^2 q_3\right)}{\left(q_2-q_1\right)
   \left(1-q_2^2\right)}\rule[-4mm]{0mm}{10mm}\\
   \hline \{1,1\} & \{2,1\} & vq_{2}& \frac{\left(1-q_1^2\right) \left(q_1-q_2\right) \left(1-q_1 q_3\right) \left(q_2-q_1^2
   q_3\right) \left(1-q_2 q_3\right) \left(q_1-q_2^2 q_3\right)}{\left(1-q_1\right)
   \left(1-q_2\right) \left(q_2-q_1^2\right) \left(q_1-q_2^2\right) \left(1-q_1^2
   q_3\right) \left(q_1-q_2 q_3\right)}\rule[-4mm]{0mm}{10mm}\\
   \hline \{1,1\} & \{1,1,1\} & vq_{1}^{2} &\frac{\left(1-q_1^3 q_3\right) \left(q_1^2-q_2 q_3\right)}{\left(1-q_1^3\right)
   \left(q_1^2-q_2\right)} \rule[-4mm]{0mm}{10mm}\\
   \hline \{2\} & \{3\} & vq_{2}^{2}&\frac{\left(q_2^2-q_1 q_3\right) \left(1-q_2^3 q_3\right)}{\left(q_2^2-q_1\right)
   \left(1-q_2^3\right)}  \rule[-4mm]{0mm}{10mm}\\
   \hline \{2\} & \{2,1\} &vq_{1} & \frac{\left(q_2-q_1\right) \left(1-q_2^2\right) \left(1-q_1 q_3\right) \left(q_2-q_1^2
   q_3\right) \left(1-q_2 q_3\right) \left(q_1-q_2^2 q_3\right)}{\left(1-q_1\right)
   \left(1-q_2\right) \left(q_2-q_1^2\right) \left(q_1-q_2^2\right) \left(q_2-q_1
   q_3\right) \left(1-q_2^2 q_3\right)}   \rule[-4mm]{0mm}{10mm}\\
   \hline
    \end{array}
\end{align}

Using the explicit forms of $\mathscr{Y}^{12}_{\lambda,v}(x)$, we can also compute the following function for low levels:
\begin{equation}
  \frac{\underset{x=q_{3}^{-1}\chi_{12,v}(\Abox)}{\Res}{x}^{-1}\frac{\mathscr{Y}_{\lambda,v}^{12}(x)}{\mathscr{Y}_{\lambda,v}^{12}(q_{3}x)}}{\underset{x=q_{34}^{-1}\chi_{12,v}(\Abox)}{\Res}{x}^{-1}\frac{\mathscr{Y}_{\lambda+\Abox,v}^{12}(x)}{\mathscr{Y}_{\lambda+\Abox,v}^{12}(q_{3}x)}}.
\end{equation}
The results are
\begin{align}
    \begin{array}{|c|c|c|c|}\hline
        \lambda & \lambda+\Bbox &\chi_{12,v}(\Bbox)& \frac{\underset{x=q_{3}^{-1}\chi_{12,v}(\Abox)}{\Res}{x}^{-1}\frac{\mathscr{Y}_{\lambda,v}^{12}(x)}{\mathscr{Y}_{\lambda,v}^{12}(q_{3}x)}}{\underset{x=q_{34}^{-1}\chi_{12,v}(\Abox)}{\Res}{x}^{-1}\frac{\mathscr{Y}_{\lambda+\Abox,v}^{12}(x)}{\mathscr{Y}_{\lambda+\Abox,v}^{12}(q_{3}x)}} \rule[-4mm]{0mm}{15mm}\\
     \hline\hline \emptyset   & \{1\} & v&  -q_{3}^{-1}\frac{(1-q_{1}q_{3})(1-q_{2}q_{3})}{(1-q_{1})(1-q_{2})}\rule[-4mm]{0mm}{10mm}\\
     \hline \{1\}& \{1,1\} & vq_{1}&-\frac{\left(q_1^2 q_3-1\right) \left(q_1-q_2 q_3\right)}{\left(q_1^2-1\right)
   \left(q_1-q_2\right) q_3}\rule[-4mm]{0mm}{10mm}\\
   \hline \{1\}  & \{2\}  & vq_{2} & \frac{\left(q_2-q_1 q_3\right) \left(q_2^2 q_3-1\right)}{\left(q_1-q_2\right)
   \left(q_2^2-1\right) q_3}\rule[-4mm]{0mm}{10mm}\\
   \hline \{1,1\} & \{2,1\} & vq_{2}& -\frac{\left(q_1+1\right) \left(q_1-q_2\right) \left(q_1 q_3-1\right) \left(q_1^2
   q_3-q_2\right) \left(q_2 q_3-1\right) \left(q_1-q_2^2
   q_3\right)}{\left(q_1^2-q_2\right) \left(q_2-1\right) \left(q_1-q_2^2\right) q_3
   \left(q_1^2 q_3-1\right) \left(q_1-q_2 q_3\right)}\rule[-4mm]{0mm}{10mm}\\
   \hline \{1,1\} & \{1,1,1\} & vq_{1}^{2} & -\frac{\left(q_1^3 q_3-1\right) \left(q_1^2-q_2 q_3\right)}{\left(q_1^3-1\right)
   \left(q_1^2-q_2\right) q_3}\rule[-4mm]{0mm}{10mm}\\
   \hline \{2\} & \{3\} & vq_{2}^{2}&\frac{\left(q_2^2-q_1 q_3\right) \left(q_2^3 q_3-1\right)}{\left(q_1-q_2^2\right)
   \left(q_2^3-1\right) q_3}\rule[-4mm]{0mm}{10mm}\\
   \hline \{2\} & \{2,1\} &vq_{1} & -\frac{\left(q_1-q_2\right) \left(q_2+1\right) \left(q_1 q_3-1\right) \left(q_1^2
   q_3-q_2\right) \left(q_2 q_3-1\right) \left(q_1-q_2^2 q_3\right)}{\left(q_1-1\right)
   \left(q_1^2-q_2\right) \left(q_1-q_2^2\right) q_3 \left(q_1 q_3-q_2\right)
   \left(q_2^2 q_3-1\right)}  \rule[-4mm]{0mm}{10mm}\\
   \hline
    \end{array}
\end{align}
Comparing the two relations, one will see we have the following relation
\begin{equation}
     \frac{\underset{x=q_{3}^{-1}\chi_{12,v}(\Abox)}{\Res}{x}^{-1}\frac{\mathscr{Y}_{\lambda,v}^{12}(x)}{\mathscr{Y}_{\lambda,v}^{12}(q_{3}x)}}{\underset{x=q_{34}^{-1}\chi_{12,v}(\Abox)}{\Res}{x}^{-1}\frac{\mathscr{Y}_{\lambda+\Abox,v}^{12}(x)}{\mathscr{Y}_{\lambda+\Abox,v}^{12}(q_{3}x)}}=-q_{3}^{-1}\frac{\mathcal{Z}_{12}^{\D4}[\lambda+\Bbox\,;q_{3}]}{\mathcal{Z}_{12}^{\D4}[\lambda\,;q_{3}]}=-\frac{\widetilde{\mathcal{Z}}_{12}^{\D4}[\lambda+\Bbox]}{\widetilde{\mathcal{Z}}_{12}^{\D4}[\lambda]}.
\end{equation}
Generally, we have the following theorem.
\begin{theorem}\label{app:thm-D4recursion}
The recursion relation of the $\U(1)$ partition function of the 5d $\mathcal{N}=1^{\ast}$ theory is written using the residues of the $\mathscr{Y}$-functions as
\begin{equation}
    \frac{\widetilde{\mathcal{Z}}^{\D4}_{A}[\lambda+\Bbox]}{\widetilde{\mathcal{Z}}^{\D4}_{A}[\lambda]}=-\frac{\underset{x=q_{\text{inf}(\bar{A})}^{-1}\chi_{A,v}(\Abox)}{\Res}{x}^{-1}\frac{\mathscr{Y}_{\lambda,v}^{A}(x)}{\mathscr{Y}_{\lambda,v}^{A}(q_{\text{inf}(\bar{A})}x)}}{\underset{x=q_{A}\chi_{A,v}(\Abox)}{\Res}{x}^{-1}\frac{\mathscr{Y}_{\lambda+\Abox,v}^{A}(x)}{\mathscr{Y}_{\lambda+\Abox,v}^{A}(q_{\text{inf}(\bar{A})}x)}}.
\end{equation}    
At one-instanton level, we explicitly have
\beq\label{app-eq:one-instanton}
\widetilde{\mathcal{Z}}^{\D4}_{ab}[\{1\}]=\mathscr{S}_{\overline{ab}}(q_{a})=\mathscr{S}_{\overline{ab}}(q_{b})
\eeq
\end{theorem}
\begin{proof}
    The nontrivial part is the minus sign arising on the right-hand side. This formula can be obtained following the discussion in \cite[section~5.1.1]{Kimura:2020jxl}.
\end{proof}
The recursion relation above can be written in a rather symmetric way: 
\bea\label{eq:app-D4U1recursionformula}
    \frac{\widetilde{\mathcal{Z}}^{\D4}_{ab}[\lambda+\Bbox]}{\widetilde{\mathcal{Z}}
^{\D4}_{ab}[\lambda]}&=-\underset{x\rightarrow q_{c}^{-1}\chi_{ab,v}(\Abox)}{\lim}\left(1-\frac{q_{c}^{-1}\chi_{ab,v}(\Bbox)}{x}\right)\frac{\mathscr{Y}_{\lambda,v}^{ab}(x)}{\mathscr{Y}^{ab}_{\lambda,v}(q_{c}x)}\\
&\qquad \times\left(\underset{x\rightarrow q_{ab}\chi_{ab,v}(\Abox)}{\lim}\left(1-\frac{q_{ab}\chi_{ab,v}(\Bbox)}{x}\right)\frac{\mathscr{Y}^{ab}_{\lambda+\Abox,v}(x)}{\mathscr{Y}^{ab}_{\lambda+\Abox,v}(q_{c}x)}\right)^{-1}\\
&=-\frac{\mathscr{Y}_{\lambda,v}^{ab}(q_{c}^{-1}\chi_{ab,v}(\Bbox))\mathscr{Y}^{ab}_{\lambda+\Abox,v}(q_{d}^{-1}\chi_{ab,v}(\Bbox))}{\mathscr{Y}^{ab}_{\lambda,v}(\chi_{ab,v}(\Bbox))\mathscr{Y}^{ab}_{\lambda+\Abox,v}(q_{cd}^{-1}\chi_{ab,v}(\Bbox))}
\eea
where $A=ab, \bar{A}=cd\,(c<d)$ and $\Bbox\in A(\lambda)$. The numerators come from the fact that they have no pole when taking the limit. For the denominators, both of them are singular but the singular part will cancel with each other and the above formula itself is well-defined.

\subsection{D2 partition functions}\label{app:D2U1partitionfunction}
The $\U(1)$ contribution of the 3d gauge theory is given from
\begin{equation}
     \mathcal{Z}^{\D2}_{a}[k\,;q_{i},q_{j}]=\frac{\mathsf{N}_{a}(q_{i}v,k\,|\,v,k)\mathsf{N}_{a}(q_{j}v,k\,|\,v,k)}{\mathsf{N}_{a}(v,k\,|\,v,k)\mathsf{N}_{a}(q_{i}q_{j}v,k\,|\,v,k)}
\end{equation}
where 
\begin{equation}
    \mathsf{N}_{a}(v_{1},k_{1}|v_{2},k_{2})=\frac{\prod\limits_{\AboxF\in k_{2}}\left(1-v_{1}/\chi_{a,v_{2}}(\BboxF)\right)}{\prod\limits_{\Abox\in k_{1}}\left(1-q_{a}\chi_{a,v_{1}}(\Bbox)/v_{2}\right)}\prod_{\substack{\Abox\in k_{1}\\\AboxF\in k_{2}}}\mathscr{V}_{a}\left(\frac{\chi_{a,v_{1}}(\Bbox)}{\chi_{a,v_{2}}(\BboxF)}\right).
\end{equation}
The $\D2$ Nekrasov factor explicitly is given 
\bea\label{eq:D2Nekrasovexplicit}
\mathsf{N}_{a}(v_{1},k_{1}\,|\,v_{2},k_{2})&=\mathbb{I}\left[\frac{v_{2}}{v_{1}}\frac{1-q_{a}^{k_{2}-k_{1}}}{1-q_{a}}\right]=\begin{dcases}
    \prod_{l=1}^{k_{2}-k_{1}}\left(1-\frac{v_{1}}{v_{2}}q_{a}^{-l+1}\right),\quad &k_{2}>k_{1},\\
    1,\quad &k_{1}=k_{2},\\
    \prod_{l=1}^{k_{1}-k_{2}}\left(1-\frac{v_{1}}{v_{2}}q_{a}^{l-1}\right),\quad &k_{
1}>k_{2}.
\end{dcases}
\eea

For simplicity, we only consider the case $a=4$. Note that the character for the $\U(1)$ case is
\bea
&\mathbf{v}_{\text{inst.}}=\bfP_{A}^{\vee}\left(\bfN_{4}^{\vee}\bfK_{4}-q_{4}^{-1}\bfN_{4}\bfK_{4}^{\vee}-\bfP_{4}^{\vee}\bfK_{4}^{\vee}\bfK_{4}\right),\quad A\in \{12,13,23\},\\
& \bfN_{4}=v,\quad  \bfK_{4}=\sum_{j=1}^{k}vq_{4}^{j-1}=v\frac{1-q_{4}^{k}}{1-q_{4}}.
\eea
Explicitly, we can show that the partition is trivial:
\bea
&\bfN_{4}^{\vee}\bfK_{4}-q_{4}^{-1}\bfN_{4}\bfK_{4}^{\vee}-\bfP_{4}^{\vee}\bfK_{4}^{\vee}\bfK_{4}\\
=&\frac{1}{1-q_{4}}\left(1-q_{4}^{k}+1-q_{4}^{-k}-(1-q_{4}^{-k})(1-q_{4}^{k})\right)\\
=&0.
\eea
Thus, the $\U(1)$ partition function is given
\bea\label{eq:D2U1partition1}
\mathcal{Z}_{4}^{\D2}[k;q_{i},q_{j}]=1.
\eea
Generally, we have 
\bea\label{eq:D2U1partition2}
\mathcal{Z}_{a}^{\D2}[k;q_{i},q_{j}]=1.
\eea
The recursion formula of the Nekrasov factor $\mathsf{N}_{a}(v,\lambda\,|\,v,\lambda)$ in terms of the $\mathscr{U}$-functions is 
\begin{equation}
\begin{split}
    \frac{\mathsf{N}_{a}(v,k+\Bbox\,|\,v,k+\Bbox)}{\mathsf{N}_{a}(v,k\,|\,v,k)}&=-\left(-\frac{vq_{a}^{k}}{\chi_{a,v}(\Bbox)}\right)\frac{\underset{x'=\chi_{a,v}(\Abox)}{\Res}{x'}^{-1}\mathscr{U}^{a}_{k+\Abox,v}(q_{a}x')^{-1}}{\underset{x'=\chi_{a,v}(\Abox)}{\Res}{x'}^{-1}\mathscr{U}^{a}_{k,v}(x')^{-1}}\\
    &=-\left(-\frac{vq_{a}^{k}}{\chi_{a,v}(\Bbox)}\right)\frac{\mathscr{U}^{a}_{k,v}(\chi_{a,v}(\Bbox))}{\mathscr{U}^{a}_{k+\Abox,v}(q_{a}\chi_{a,v}(\Bbox))}.
\end{split}
\end{equation}
Note that the second equation is still well-defined although the numerator and denominator are singular themselves. This is because the singularities cancel with each other. Note also that $k+\Bbox=k+1$ for the vortex partition function. 

For the adjoint contributions, we have 
\bea
\frac{\mathsf{N}_{a}(mv,k+\Bbox\,|\,v,k+\Bbox)}{\mathsf{N}_{a}(mv,k\,|\,v,k)}&=
\frac{\mathscr{U}^{a}_{k,v}(m^{-1}\chi_{a,v}(\Bbox))}{\mathscr{U}^{a\vee}_{k+\Abox,v}(mq_{a}\chi_{a,v}(\Bbox))}\\
&=\left(-\frac{vq_{a}^{k}}{m\chi_{a,v}(\Bbox)}\right)\frac{\mathscr{U}^{a}_{k,v}(m^{-1}\chi_{a,v}(\Bbox))}{\mathscr{U}^{a}_{k+\Abox,v}(mq_{a}\chi_{a,v}(\Bbox))}
\eea
In this case, both the numerator and denominator themselves are already non-singular and well-defined. Combining these, we have 
\bea\label{eq:app-D2U1recursionformula}
&\frac{\mathcal{Z}_{a}^{\D2}[k+\Bbox;q_{i},q_{j}]}{\mathcal{Z}^{\D2}_{a}[k;q_{i},q_{j}]}\\
=&-\frac{\mathscr{U}^{a}_{k,v}(q_{i}^{-1}\chi_{a,v}(\Bbox))}{\mathscr{U}^{a}_{k+\Abox,v}(q_{i}q_{a}\chi_{a,v}(\Bbox))}\frac{\mathscr{U}^{a}_{k,v}(q_{j}^{-1}\chi_{a,v}(\Bbox))}{\mathscr{U}^{a}_{k+\Abox,v}(q_{j}q_{a}\chi_{a,v}(\Bbox))}\frac{\mathscr{U}^{a}_{k+\Abox,v}(q_{ij}q_{a}\chi_{a,v}(\Bbox))}{\mathscr{U}^{a}_{k,v}(q_{ij}^{-1}\chi_{a,v}(\Bbox))}\frac{\mathscr{U}^{a}_{k+\Abox,v}(q_{a}\chi_{a,v}(\Bbox))}{\mathscr{U}^{a}_{k,v}(\chi_{a,v}(\Bbox))}.
\eea
Although, indeed we can rewrite using the $\mathscr{U}$-functions, the right-hand side will trivialize after computation.

Since the D2 Nekrasov factor \eqref{eq:D2Nekrasovexplicit} does not disappear generally, the partition function for the $\U(n_{a})\,(a\in\four)$ case does not disappear. Denoting the $\U(n_{a})$ partition function as $\mathcal{Z}^{\D2}_{a}[\vec{v}_{a},\vec{k}_{a}]$, we have 
\bea
\mathcal{Z}^{\D2}_{a}[\vec{v}_{a},\vec{k}_{a}]&=\prod_{\alpha=1}^{n_{a}}\mathcal{Z}^{\D2}_{a}[k_{a}^{(\alpha)};q_{\text{i}(a)},q_{\text{j}(a)}]\prod_{\alpha<\beta}\mathcal{Z}^{\D2\tbar\D2}_{a|a}(v_{a,\alpha},k_{a}^{(\alpha)}|v_{a,\beta},k_{a}^{(\beta)})\\
&=\prod_{\alpha,\beta}\frac{\mathsf{N}_{a}(q_{\text{i}(a)}v_{a,\alpha},k_{a}^{(\alpha)}|v_{a,\beta},k_{a}^{(\beta)})\mathsf{N}_{a}(q_{\text{j}(a)}v_{a,\alpha},k_{a}^{(\alpha)}|v_{a,\beta},k_{a}^{(\beta)})}{\mathsf{N}_{a}(q_{\text{i}(a)}q_{\text{j}(a)}v_{a,\alpha},k_{a}^{(\alpha)}|v_{a,\beta},k_{a}^{(\beta)})\mathsf{N}_{a}(v_{a,\alpha},k_{a}^{(\alpha)}|v_{a,\beta},k_{a}^{(\beta)})}\\
&=\prod_{\alpha\neq\beta}\frac{\mathsf{N}_{a}(q_{\text{i}(a)}v_{a,\alpha},k_{a}^{(\alpha)}|v_{a,\beta},k_{a}^{(\beta)})\mathsf{N}_{a}(q_{\text{j}(a)}v_{a,\alpha},k_{a}^{(\alpha)}|v_{a,\beta},k_{a}^{(\beta)})}{\mathsf{N}_{a}(q_{\text{i}(a)}q_{\text{j}(a)}v_{a,\alpha},k_{a}^{(\alpha)}|v_{a,\beta},k_{a}^{(\beta)})\mathsf{N}_{a}(v_{a,\alpha},k_{a}^{(\alpha)}|v_{a,\beta},k_{a}^{(\beta)})},
\eea
where we used $\mathsf{N}_{a}(v,\lambda|v,\lambda)=1$.

Let us focus on the $a=4$ case:
\bea
\mathcal{Z}_{4}^{\D2}[\vec{v},\vec{k}]=\prod_{\alpha,\beta}\frac{\mathsf{N}_{4}(q_{i}v_{\alpha},k^{(\alpha)}|v_{\beta},k^{(\beta)})\mathsf{N}_{4}(q_{j}v_{\alpha},k^{(\alpha)}|v_{\beta},k^{(\beta)})}{\mathsf{N}_{4}(v_{\alpha},k^{(\alpha)}|v_{\beta},k^{(\beta)})\mathsf{N}_{4}(q_{i}q_{j}v_{\alpha},k^{(\alpha)}|v_{\beta},k^{(\beta)})},
\eea
where we kept the $\U(1)$ part for convention. The recursive relation is given as 
\bea\label{eq:app-D2Unrecursion}
&\frac{\mathcal{Z}^{\D2}_{4}[\vec{v},\vec{k}+\Bbox\,]}{\mathcal{Z}^{\D2}_{4}[\vec{v},\vec{k}]}\\
&=\prod_{\beta\neq \alpha}\frac{\mathscr{U}^{4}_{k_{\beta},v_{\beta}}(q_{4}\chi_{4,v_{\alpha}}(\Bbox))\mathscr{U}^{4}_{k_{\beta},v_{\beta}}(q_{ij4}\chi_{4,v_{\alpha}}(\Bbox))\mathscr{U}^{4}_{k_{\beta},v_{\beta}}(q_{i}^{-1}\chi_{4,v_{\alpha}}(\Bbox))\mathscr{U}^{4}_{k_{\beta},v_{\beta}}(q_{j}^{-1}\chi_{4,v_{\alpha}}(\Bbox))}{\mathscr{U}^{4}_{k_{\beta},v_{\beta}}(\chi_{4,v_{\alpha}}(\Bbox))\mathscr{U}^{4}_{k_{\beta},v_{\beta}}(q_{ij}^{-1}\chi_{4,v_{\alpha}}(\Bbox))\mathscr{U}^{4}_{k_{\beta},v_{\beta}}(q_{4i}\chi_{4,v_{\alpha}}(\Bbox))\mathscr{U}^{4}_{k_{\beta},v_{\beta}}(q_{4j}\chi_{4,v_{\alpha}}(\Bbox))}\\
&\times -\frac{\mathscr{U}^{4}_{k_{\alpha},v_{\alpha}}(q_{i}^{-1}\chi_{4,v_{\alpha}}(\Bbox))}{\mathscr{U}^{4}_{k_{\alpha}+\Abox,v_{\alpha}}(q_{i}q_{4}\chi_{4,v_{\alpha}}(\Bbox))}\frac{\mathscr{U}^{4}_{k_{\alpha},v_{\alpha}}(q_{j}^{-1}\chi_{4,v_{\alpha}}(\Bbox))}{\mathscr{U}^{4}_{k_{\alpha}+\Abox,v_{\alpha}}(q_{j}q_{4}\chi_{4,v_{\alpha}}(\Bbox))}\frac{\mathscr{U}^{4}_{k_{\alpha}+\Abox,v_{\alpha}}(q_{ij}q_{4}\chi_{4,v_{\alpha}}(\Bbox))}{\mathscr{U}^{4}_{k_{\alpha},v_{\alpha}}(q_{ij}^{-1}\chi_{4,v_{\alpha}}(\Bbox))}\frac{\mathscr{U}^{4}_{k_{\alpha}+\Abox,v_{\alpha}}(q_{4}\chi_{4,v_{\alpha}}(\Bbox))}{\mathscr{U}^{4}_{k_{\alpha},v_{\alpha}}(\chi_{4,v_{\alpha}}(\Bbox))}
\eea
where we shortly wrote $\vec{v}=(v_{\alpha})_{\alpha=1}^{n},\,\,\vec{k}=(k_{\alpha})_{\alpha=1}^{n}$ and assumed $\Bbox\in A(k^{(\alpha)}),\exists\alpha$.

\section{Zero modes of vertex operators}
Let us check that the zero modes discussed in \eqref{eq:zeromodes1} and \eqref{eq:zeromodes2} indeed satisfy the conditions in \eqref{eq:D0D0zero}, \eqref{eq:D0D2zero}, \eqref{eq:D0D4zero}, \eqref{eq:D0D6zero}, \eqref{eq:D0D8zero}, \eqref{eq:D2D2zero}, \eqref{eq:D2D4zero}, \eqref{eq:D2D6zero}. Moreover, the notation we use also obeys \eqref{eq:zeromodesrelation}.

\begin{lemma}[Zero modes]\label{app:lem:zero-modes-prf1}
The free field realizations
\bea
&\mathsf{a}_{0}(x)=e^{\mathsf{t}_{0}},\quad \mathsf{s}_{a,0}(x)=x^{-(\log q_{a})^{-1}\mathsf{t}_{0}}e^{-(\log q_{a})^{-1}\tilde{\partial}_{\mathsf{t}}},\quad \mathsf{x}_{\overbar{ab},0}(x)=e^{\log q_{a}\log q_{b}\tilde{\mathsf{t}}_{0}},\\
&\mathsf{w}_{\bar{a},0}(x)=x^{-\log q_{a}\tilde{\mathsf{t}}_{0}}e^{-\log q_{a}\log(-q_{a})\tilde{\mathsf{t}}_{0}}e^{-\log q_{a}\partial_{\mathsf{t}}},\quad \tilde{\mathsf{z}}^{K}_{0}(x)=x^{-\log K\tilde{\mathsf{t}}_{0}}e^{-\log K\log(-K)\tilde{\mathsf{t}}_{0}}e^{-\log K\partial_{\mathsf{t}}}
\eea
obey \eqref{eq:D0D0zero}, \eqref{eq:D0D2zero}, \eqref{eq:D0D4zero}, \eqref{eq:D0D6zero}, \eqref{eq:D0D8zero}, \eqref{eq:D2D2zero}, \eqref{eq:D2D4zero}, \eqref{eq:D2D6zero}.
\end{lemma}

\begin{proof}
    Let us consider the commutation relations of the zero modes with $\mathsf{a}_{0}(x)$. Obviously, using \eqref{eq:zeromodesdef}, we have \eqref{eq:D0D0zero}, \eqref{eq:D0D2zero}, \eqref{eq:D0D4zero}:
    \beq
        \mathsf{a}_{0}(x)\mathsf{a}_{0}(x')=\mathsf{a}_{0}(x')\mathsf{a}_{0}(x),\quad \mathsf{a}_{0}(x)\mathsf{s}_{0}(x')=\mathsf{s}_{0}(x')\mathsf{a}_{0}(x),\quad \mathsf{a}_{0}(x)\mathsf{x}_{A,0}(x')=\mathsf{x}_{A,0}(x')\mathsf{a}_{0}(x).
    \eeq
    For the $\D0$-$\D6$ and $\D0$-$\D8$ relations, we have
    \beq
        \mathsf{w}_{\bar{a},0}(x)\mathsf{a}_{0}(x')=q_{a}^{-1}\mathsf{a}_{0}(x')\mathsf{w}_{\bar{a},0}(x),\quad
        \tilde{\mathsf{z}}^{K}_{0}(x)\mathsf{a}_{0}(x')=K^{-1}\mathsf{a}_{0}(x')\tilde{\mathsf{z}}^{K}_{0}(x).
    \eeq
    which give \eqref{eq:D0D6zero}, \eqref{eq:D0D8zero}. Let us next consider the commutation relations of the zero modes with $\mathsf{s}_{a,0}(x)$. For the $\D2_{a}$-$\D2_{b}$\,$(a\neq b)$ case, we obviously have 
    \beq
        \mathsf{s}_{a,0}(x)\mathsf{s}_{b,0}(x')=\mathsf{s}_{b,0}(x')\mathsf{s}_{a,0}(x)
    \eeq
    which gives \eqref{eq:D2D2zero}. The $\D2_{c}$-$\D4_{A}$ case where $A=ab,\,\bar{A}=cd$ is given as 
    \bea
        \mathsf{s}_{c,0}(x)\mathsf{x}_{ab,0}(x')&=e^{-\frac{\log q_{c}\log q_{d}}{\log q_{c}}[\tilde{\partial}_{\mathsf{t}},\tilde{\mathsf{t}_{0}}]}\mathsf{x}_{ab,0}(x')\mathsf{s}_{c,0}(x)\\
        &=q_{d}^{-1}\mathsf{x}_{ab,0}(x')\mathsf{s}_{c,0}(x).
    \eea
    Switching $c\leftrightarrow d$, we also obtain
    \beq
        \mathsf{s}_{d,0}(x)\mathsf{x}_{ab,0}(x') =q_{c}^{-1}\mathsf{x}_{ab,0}(x')\mathsf{s}_{d,0}(x),
    \eeq
    and obtain \eqref{eq:D2D4zero}. The $\D2_{a}$-$\D6_{\bar{a}}$ case \eqref{eq:D2D6zero} comes from 
    \beq
        \mathsf{s}_{a,0}(x)\mathsf{w}_{\bar{a},0}(x')=-q_{a}x':\mathsf{s}_{a,0}(x)\mathsf{w}_{\bar{a},0}(x'):,\quad
        \mathsf{w}_{\bar{a},0}(x')\mathsf{s}_{a,0}(x)=x:\mathsf{s}_{a,0}(x)\mathsf{w}_{\bar{a},0}(x'):
    \eeq
\end{proof}

\begin{corollary}\label{app:cor:zero-modes-prf2}
  The zero modes obey the relations in \eqref{eq:zeromodesrelation}.  
\end{corollary}

\begin{proof}
The relation between $\mathsf{s}_{a,0}(x)$ and $\mathsf{a}_{0}(x)$ is given
\beq
    :\frac{\mathsf{s}_{a,0}(x)}{\mathsf{s}_{a,0}(q_{a}x)}:=\frac{x^{-(\log q_{a})^{-1}\mathsf{t}_{0}}}{(q_{a}x)^{-(\log q_{a})^{-1}\mathsf{t}_{0}}}=e^{\mathsf{t}_{0}}=\mathsf{a}_{0}(x).
\eeq
The relation between $\mathsf{w}_{abc,0}(x)$ and $\mathsf{x}_{ab,0}(x)$ is given
\beq
    :\frac{\mathsf{w}_{abc,0}(x)}{\mathsf{w}_{abc,0}(q_{c}x)}:=q_{c}^{\log q_{d}\tilde{\mathsf{t}}_{0}}=e^{\log q_{c}\log q_{d}\tilde{\mathsf{t}}_{0}}=\mathsf{x}_{ab,0}(x)
\eeq
where $\{a,b,c,d\}=\four$. For $\tilde{\mathsf{z}}_{0}^{K}(x)$ and $\mathsf{w}_{\bar{a},0}(x)$, we have 
\beq
    \tilde{\mathsf{z}}_{0}^{q_{a}}(x)=\mathsf{w}_{\bar{a},0}(x).
\eeq
Using this relation, one can also show that we have the relation
\beq
    :\frac{\mathsf{w}_{\bar{a},0}(x)}{\mathsf{w}_{\bar{a},0}(Kx)}:=:\frac{\tilde{\mathsf{z}}_{0}^{K}(x)}{\tilde{\mathsf{z}}_{0}^{K}(q_{a}x)}:
\eeq
\end{proof}

\begin{proposition}[Contraction formulas]\label{app:prop:contraction_formula}
Under the above free field realizations of the zero modes, the contraction formulas are
\begin{subequations}
\begin{align}
\mathsf{A}(x)\mathsf{S}_{a}(x')&=g_{\bar{a}}\left(x'/x\right)^{-1}: \mathsf{A}(x)\mathsf{S}_{a}(x'):,\\
\mathsf{S}_{a}(x')\mathsf{A}(x)&=g_{\bar{a}}(q_{a}x/x'):\mathsf{A}(x)\mathsf{S}_{a}(x'):,\\
\StepSubequations
\mathsf{A}(x)\mathsf{X}_{A}(x')&=\mathscr{S}_{\bar{A}}(x'/x)^{-1}:\mathsf{A}(x)\mathsf{X}_{A}(x'):,\\
    \mathsf{X}_{A}(x')\mathsf{A}(x)&=\mathscr{S}_{\bar{A}}(q_{A}x/x')^{-1}:\mathsf{X}_{A}(x')\mathsf{A}(x):,\\
\StepSubequations
    \mathsf{A}(x)\mathsf{W}_{\bar{a}}(x')&=\mathscr{V}_{a}\left(x'/x\right)^{-1}:\mathsf{A}(x)\mathsf{W}_{\bar{a}}(x'):,\\
    \mathsf{W}_{\bar{a}}(x')\mathsf{A}(x)&=q_{a}^{-1}\mathscr{V}_{a}(q_{a}^{-1}x/x'):\mathsf{W}_{\bar{a}}(x')\mathsf{A}(x):\\
\StepSubequations
    \mathsf{A}(x)\widetilde{\mathsf{Z}}^{K}(x')&=\frac{1-x'/x}{1-Kx'/x}:\mathsf{A}(x)\widetilde{\mathsf{Z}}^{K}(x'):,\\
    \widetilde{\mathsf{Z}}^{K}(x')\mathsf{A}(x)&=K^{-1}\frac{1-x/x'}{1-K^{-1}x/x'}:\mathsf{A}(x)\widetilde{\mathsf{Z}}^{K}(x'):,\\
\StepSubequations
    \mathsf{S}_{a}(x)\mathsf{S}_{b}(x')&=\mathscr{S}_{\overbar{ab}}(q_{a}x'/x):\mathsf{S}_{a}(x)\mathsf{S}_{b}(x'):,\\
        \mathsf{S}_{b}(x')\mathsf{S}_{a}(x)&=\mathscr{S}_{\overbar{ab}}(q_{b}x/x'):\mathsf{S}_{a}(x)\mathsf{S}_{b}(x'):,\\
\StepSubequations
        \mathsf{X}_{A}(x)\mathsf{S}_{c}(x')&=\mathscr{V}_{d}\left(q_{A}x'/x\right)^{-1} : \mathsf{X}_{A}(x)\mathsf{S}_{c}(x'):  ,\\
        \mathsf{S}_{c}(x')\mathsf{X}_{A}(x)&=q_{d}^{-1}\mathscr{V}_{d}\left(q_{d}^{-1}q_{A}^{-1}x/x'\right): \mathsf{X}_{A}(x)\mathsf{S}_{c}(x'):,\\
\StepSubequations
\mathsf{W}_{\bar{a}}(x)\mathsf{S}_{a}(x')&=x'\frac{1}{1-q_{a}^{-1}x'/x}:\mathsf{W}_{\bar{a}}(x)\mathsf{S}_{a}(x'):,\\
\mathsf{S}_{a}(x')\mathsf{W}_{\bar{a}}(x)&=(-q_{a}x)\frac{1}{1-q_{a}x/x'}:\mathsf{W}_{\bar{a}}(x)\mathsf{S}_{a}(x'):
\end{align}\label{eq:app-contractions}
\end{subequations}
\end{proposition}

\section{Supergroup gauge theory}\label{app:supergroup}
In this section, we review the instanton partition function of supergroup gauge theories in 5d theories. From the D-brane perspective, we need to include \textit{ghost} D4-branes, denoted by $\D4^{-}$, in the system. The partition functions were first derived by \cite{Kimura:2019msw} (see also \cite{Kimura:2023iup} for a review). The quantum algebraic perspective was studied in \cite{Noshita:2022dxv} for A, D-quiver gauge theories. The discussion here is complementary to the discussion in \cite{Noshita:2022dxv}.
\subsection{Contour integral formula and partition functions}
We consider 5d gauge theories on $\mathbb{C}^{2}_{12}\times \mathbb{S}^{1}$. The instanton partition functions of superunitary groups $\U(n_{+}|n_{-})$ are described by changing the character to super characters:
\bea
&\bfN=\bfN^{+}-\bfN^{-},\quad \bfK=\bfK^{+}-\bfK^{-},\\
&\bfN^{\pm}=\sum_{\alpha=1}^{n_{\pm}}e^{a_{\alpha}^{\pm}}=\sum_{\alpha=1}^{n_{\pm}}v_{\alpha}^{\pm},\quad \bfK^{\pm}=\sum_{I=1}^{k_{\pm}}e^{\phi_{I}^{\pm}}=\sum_{I=1}^{k_{\pm}}x_{I}^{\pm}.
\eea
The partition function is given by inserting these characters to \eqref{eq:affinequiver} or generally to \eqref{eq:D4squareroot}. Let us focus on the affine quiver gauge theory case whose character is\footnote{We use the character in \eqref{eq:D4squareroot} where the topological term is slightly modified.} 
\bea
\mathbf{v}_{\text{inst.}}&=-\bfP_{34}^{\vee}\bfN^{\vee}\bfK+\bfP_{123}^{\vee}\bfK^{\vee}\bfK\\
&=-\sum_{\sigma,\sigma'=\pm}\sigma\sigma'\bfP_{34}^{\vee}\bfN^{\sigma\,\vee}\bfK^{\sigma'}+\sum_{\sigma,\sigma'=\pm}\sigma\sigma'\bfP_{123}^{\vee}\bfK^{\sigma\vee}\bfK^{\sigma'}.
\eea
The contour integral is then given by
\bea\label{eq:app-supergroupLMNS}
\mathcal{Z}^{\D4}_{12,k_{+}|k_{-}}=\mathbb{I}[\mathbf{v}_{\text{inst.}}]=&\frac{\mathcal{G}_{\bar{4}}^{k_{+}}\mathcal{G}_{\bar{4}}^{k_{-}}}{k_{+}!\,k_{-}!}\oint\prod_{\sigma=\pm}\prod_{I=1}^{k_{\sigma}}\frac{dx_{I}^{\sigma}}{2\pi\iota x_{I}^{\sigma}}\prod_{\sigma\sigma'=\pm}\prod_{\alpha=1}^{n^{\sigma}}\prod_{I=1}^{k_{\sigma'}}\mathscr{S}_{34}\left(\frac{v_{\alpha}^{\sigma}}{x_{I}^{\sigma'}}\right)^{\sigma\sigma'}\\
&\qquad\times\prod_{\sigma,\sigma'=\pm}\prod_{\substack{I=1,\ldots,k_{\sigma}\\J=1,\ldots,k_{\sigma'}\\(I,\sigma)\neq (J,\sigma')}}g_{\bar{4}}\left(\frac{x_{I}^{\sigma}}{x_{J}^{\sigma'}}\right)^{-\sigma\sigma'}.
\eea
The instanton partition function is schematically written as
\bea
\mathcal{Z}_{\text{inst.}}=\sum_{k_{\pm}\geq 0}\mathfrak{q}^{k_{+}-k_{-}}\mathcal{Z}^{\D4}_{12,k_{+}|k_{-}}.
\eea
Note that the topological term is in the opposite power compared to the normal group case. The supergroup analogue of the gauge origami of the spiked instantons is obtained in a similar way.

The contour integral will localize on the poles classified by $(n_{+},n_{-})$-tuples of Young diagrams:
\bea
&\vec{v}=(v_{\alpha}^{\sigma})_{\alpha=1,\ldots,n_{\sigma}}^{\sigma=\pm},\quad \vec{\lambda}=(\lambda^{(\alpha)}_{\sigma})_{\alpha=1,\ldots,n_{\sigma}}^{\sigma=\pm},\quad \sum_{\alpha=1}^{n_{\sigma}}|\lambda^{(\alpha)}_{\sigma}|=k_{\sigma}\\
&\{x_{I}^{\sigma}\}_{I=1,\ldots,k_{\sigma}}\rightarrow \{\chi_{12,v_{\alpha}^{+}}(\Bbox)\}_{\alpha=1,\ldots,n_{+},\Abox\in\lambda^{(\alpha)}_{+}}\cup\{\bar{\chi}_{12,v_{\alpha}^{-}}(\Bbox)\}_{\alpha=1,\ldots,n_{-},\Abox\in\lambda_{-}^{(\alpha)}}
\eea
where
\bea
\bar{\chi}_{12,v}(\Bbox)=vq_{1}^{-i}q_{2}^{-j},\quad \Bbox=(i,j).
\eea
The character of the instanton bundle for the negative instanton contribution is
\bea
\bfK^{-}=\sum_{\alpha=1}^{n_{-}}\sum_{\Abox\in\lambda^{(\alpha)}_{-}}\bar{\chi}_{12,v_{\alpha}^{-}}(\Bbox).
\eea

The explicit form of the instanton partition functions is expressed by the Nekrasov factors for supergroup gauge theories (see \cite{Noshita:2022dxv}):
\bea
\mathsf{N}_{A}^{\sigma\sigma'}(v_{1},\lambda^{(1)}\,|\,v_{2},\lambda^{(2)})=\prod_{\Abox\in\lambda^{(1)}}\left(1-q_{A}\frac{\chi_{A,v_{1}}^{(\sigma)}(\Bbox)}{v_{2}}\right)\prod_{\Abox\in\lambda^{(2)}}\left(1-\frac{v_{1}}{\chi_{A,v_{2}}^{(\sigma')}(\Bbox)}\right)\prod_{\substack{\Abox\in\lambda^{(1)}\\\AboxF\in\lambda^{(2)}}}\mathscr{S}_{A}\left(\frac{\chi_{A,v_{1}}^{(\sigma)}(\Bbox)}{\chi_{A,v_{2}}^{(\sigma')}(\BboxF)}\right)
\eea
where 
\bea
\chi_{A,v}^{(\sigma)}=\begin{dcases}
    \chi_{A,v}(\Bbox),\quad \sigma=+,\\
    \bar{\chi}_{A,v}(\Bbox),\quad \sigma=-.
\end{dcases}
\eea
These are the supergroup analogs of the Nekrasov factors introduced in \eqref{eq:D4Nekrasovfactor}. See \cite{Noshita:2022dxv} for the explicit form of how the partition function looks like using the Nekrasov factors. 

For later use, let us write down the $\U(0|1)$ partition function:
\bea
\widetilde{\mathcal{Z}}^{\D4^{-}}_{A}[\lambda]=q_{\text{inf}(\bar{A})}^{-|\lambda|}\frac{\mathsf{N}^{--}_{A}(q_{\text{inf}(\bar{A})}x,\lambda\,|\,x,\lambda)}{\mathsf{N}^{--}_{A}(x,\lambda\,|\,x,\lambda)}.
\eea
For example, the one instanton contribution for $A=12$ is given by the 
\bea\label{eq:negativeoneinstanton}
\widetilde{\mathcal{Z}}^{\D4^{-}}_{12}[\{1\}]=\mathscr{S}_{34}(q_{1}),
\eea
which is equivalent to the one instanton contribution for the $\U(1|0)$ case.

The recursion formulas of these Nekrasov factors are given as
\bea
\frac{\mathsf{N}_{A}^{\sigma\sigma'}(v_{1},\lambda^{(1)}+\Bbox\,|\,v_{2},\lambda^{(2)})}{\mathsf{N}_{A}^{\sigma\sigma'}(v_{1},\lambda^{(1)}\,|\,v_{2},\lambda^{(2)})}&=\left(-\frac{q_{A}\chi_{A,v_{1}}^{(\sigma)}}{v_{2}}\right)\mathscr{Y}^{A,\sigma'}_{\lambda^{(2)},v_{2}}(q_{A}\chi_{A,v}^{(\sigma)}(\Bbox)),\\
\frac{\mathsf{N}_{A}^{\sigma\sigma'}(v_{1},\lambda^{(1)}\,|\,v_{2},\lambda^{(2)}+\Bbox)}{\mathsf{N}_{A}^{\sigma\sigma'}(v_{1},\lambda^{(1)}\,|\,v_{2},\lambda^{(2)})}&=\mathscr{Y}_{\lambda^{(1)},v_{1}}^{A,\sigma}(\chi_{A,v_{2}}^{(\sigma')}(\Bbox))
\eea
where
\bea
\mathscr{Y}^{A,\sigma}_{\lambda,v}(x)=\left(1-\frac{v}{x}\right)\prod_{\Abox \in \lambda}\mathscr{S}_{A}\left(\frac{\chi_{A,v}^{(\sigma)}(\Bbox)}{x}\right).
\eea
For $\sigma=-$, we shortly use (see \cite{Noshita:2022dxv} for the derivations)
\bea
\bar{\mathscr{Y}}_{\lambda,v}^{A}(x)\coloneqq\mathscr{Y}^{A,-}_{\lambda,v}(x)=\frac{\prod_{\Abox\in A(\lambda)}(1-q_{A}\bar{\chi}_{A,v}(\Bbox)/x)}{\prod_{\Abox \in R(\lambda)}(1-\bar{\chi}_{A,v}(\Bbox)/x)}.
\eea

\subsection{Supergroup analog of affine quiver W-algebra}
Let us give a proof of Thm.~\ref{thm:D4negativeqqch}. It is enough to consider the $qq$-character generated by $\mathsf{X}_{12}(x)^{-1}$. Let $F^{\D4^{-}}_{12}[\lambda]$ be an undetermined function and let us find the condition it should obey so that the $qq$-character commutes:
\begin{equation}
    \mathsf{T}_{12}^{(0|1)}(x)=\sum_{\lambda} F^{\D4^{-}}_{12}[\lambda]\overbar{\Lambda}_{12,\lambda}(x),\quad  \overbar{\Lambda}_{12,\lambda}(x)=:\mathsf{X}_{12}(x)^{-1}\prod_{\Abox\in\lambda}\mathsf{A}(\overbar{\chi}_{12,x}(\Bbox)):.
\end{equation}
First of all, let us check the commutativity for the first term. The operator product is given 
\bea
\mathsf{X}_{12}^{-1}(q_{12}x)\mathsf{S}_{4}(x')&=\frac{1-q_{3}x'/x}{1-x'/x}:\mathsf{X}_{12}(x)^{-1}\mathsf{S}_{4}(x'):,\\
:\mathsf{X}_{12}^{-1}(q_{12}x)\mathsf{A}(x):\mathsf{S}_{4}(x')&=\frac{(1-q_{12}x'/x)(1-q_{13}x'/x)(1-q_{23}x'/x)}{(1-q_{4}^{-1}x'/x)(1-q_{1}x'/x)(1-q_{2}x'/x)}:\mathsf{X}_{12}^{-1}(q_{12}x)\mathsf{A}(x)\mathsf{S}_{4}(x'):
\eea
and thus we have
\bea
\relax[\mathsf{X}_{12}^{-1}(q_{12}x),\mathsf{S}_{4}(x')]&=(1-q_{3})\delta\left(x'/x\right):\mathsf{X}_{12}(x)^{-1}\mathsf{S}_{4}(x'):,\\
[{:\mathsf{X}_{12}^{-1}(q_{12}x)\mathsf{A}(x):},\mathsf{S}_{4}(x')]&=(1-q_{3}^{-1})\frac{(1-q_{2}^{-1})(1-q_{1}^{-1})}{(1-q_{14})(1-q_{24})}\delta\left(q_{4}x/x'\right):\mathsf{X}_{12}^{-1}(q_{12}x)\mathsf{A}(x)\mathsf{S}_{4}(x'):\\
&\qquad+\cdots.
\eea
Therefore, if we do the i-Weyl reflection as 
\bea
\mathsf{X}_{12}(x)\rightarrow q_{3}\frac{(1-q_{13}^{-1})(1-q_{23}^{-1})}{(1-q_{1}^{-1})(1-q_{2}^{-1})}:\mathsf{X}_{12}^{-1}(x)\mathsf{A}(q_{12}^{-1}x):=\mathscr{S}_{34}(q_{1}):\mathsf{X}_{12}^{-1}(x)\mathsf{A}(q_{12}^{-1}x):,
\eea
the terms coming from the first pole cancel up to total difference. Note that the coefficient $\mathscr{S}_{34}(q_{1})$ is just the one instanton contribution coming from $\U(0|1)$ (see \eqref{eq:negativeoneinstanton}).

Generally, using
\bea
    &\left[\frac{\overbar{\mathscr{Y}}^{12}_{\lambda,x}(q_{123}x')}{\overbar{\mathscr{Y}}^{12}_{\lambda,x}(q_{12}x')}\right]_{+}-\left[\frac{\overbar{\mathscr{Y}}^{12}_{\lambda,x}(q_{123}x')}{\overbar{\mathscr{Y}}^{12}_{\lambda,x}(q_{12}x')}\right]_{-}\\
    =&\sum_{\Abox \in R(\lambda)}\delta\left(q_{4}\frac{\bar{\chi}_{12,x}(\Bbox)}{x'}\right)\underset{x'=q_{4}\bar{\chi}_{12,x}(\Abox)}{\Res}{x'}^{-1}\frac{\overbar{\mathscr{Y}}^{12}_{\lambda,x}(q_{123}x')}{\overbar{\mathscr{Y}}^{12}_{\lambda,x}(q_{12}x')}\\
    &+\sum_{\Abox\in A(\lambda)}\delta\left(\frac{\overbar{\chi}_{12,x}(\Bbox)}{x'}\right)\underset{x'=\bar{\chi}_{12,x}(\Abox)}{\Res}{x'}^{-1}\frac{\overbar{\mathscr{Y}}^{12}_{\lambda,x}(q_{123}x')}{\overbar{\mathscr{Y}}^{12}_{\lambda,x}(q_{12}x')}.
\eea
and 
\bea
    \mathsf{S}_{4}(x')\overbar{\Lambda}_{12,\lambda}(x)&=\left[q_{3}\frac{\overbar{\mathscr{Y}}^{12}_{\lambda,x}(q_{123}x')}{\overbar{\mathscr{Y}}^{12}_{\lambda,x}(q_{12}x')}\right]_{+}:\mathsf{S}_{4}(x')\overbar{\Lambda}_{12,\lambda}(x):,\\
    \overbar{\Lambda}_{12,\lambda}(x)\mathsf{S}_{4}(x')&=\left[q_{3}\frac{\overbar{\mathscr{Y}}^{12}_{\lambda,x}(q_{123}x')}{\overbar{\mathscr{Y}}^{12}_{\lambda,x}(q_{12}x')}\right]_{-}:\mathsf{S}_{4}(x')\overbar{\Lambda}_{12,\lambda}(x)
\eea
Then, the commutation relation with the screening current is 
\begin{align}
\begin{split}
    &\left[\mathsf{T}_{12}^{(0|1)}(x),\mathsf{S}_{4}(x')\right]\\
    =&(-q_{3})\sum_{\lambda\in\mathcal{P}}F^{\D4^{-}}_{12}[\lambda]\left[ \sum_{\Abox \in R(\lambda)}\delta\left(q_{4}\frac{\bar{\chi}_{12,x}(\Bbox)}{x'}\right)\underset{x'=q_{4}\bar{\chi}_{12,x}(\Abox)}{\Res}{x'}^{-1}\frac{\overbar{\mathscr{Y}}^{12}_{\lambda,x}(q_{123}x')}{\overbar{\mathscr{Y}}^{12}_{\lambda,x}(q_{12}x')}\right.\\
    &\quad +\left.\sum_{\Abox\in A(\lambda)}\delta\left(\frac{\overbar{\chi}_{12,x}(\Bbox)}{x'}\right)\underset{x'=\bar{\chi}_{12,x}(\Abox)}{\Res}{x'}^{-1}\frac{\overbar{\mathscr{Y}}^{12}_{\lambda,x}(q_{123}x')}{\overbar{\mathscr{Y}}^{12}_{\lambda,x}(q_{12}x')}\right]:\overbar{\Lambda}_{12,\lambda}(x)\mathsf{S}_{4}(x'):.
\end{split}
\end{align}
After shifting the first term using $\lambda=\lambda'+\Bbox$, the first term transforms to
\bea
=&(-q_{3})\sum_{\lambda'\in\mathcal{P}}\sum_{\Abox\in A(\lambda')}F^{\D4^{-}}_{12}[\lambda'+\Bbox]\delta\left(q_{4}\frac{\overbar{\chi}_{12,x}(\Bbox)}{x'}\right)\underset{x'=q_{4}\bar{\chi}_{12,x}(\Abox)}{\Res}{x'}^{-1}\frac{\overbar{\mathscr{Y}}^{12}_{\lambda'+\Abox,x}(q_{123}x')}{\overbar{\mathscr{Y}}^{12}_{\lambda'+\Abox,x}(q_{12}x')}\\
&\qquad\times :\overbar{\Lambda}_{12,\lambda'+\Abox}(x)\mathsf{S}_{4}(q_{4}\overbar{\chi}_{12,x}(\Bbox)):\\
\eea
Using 
\begin{equation}
    :\overbar{\Lambda}_{12,\lambda'+\Abox}(x)\mathsf{S}_{4}(q_{4}\overbar{\chi}_{12,x}(\Bbox)):=:\overbar{\Lambda}_{12,\lambda'}(x)\mathsf{A}(\overbar{\chi}_{12,x}(\Bbox))\mathsf{S}_{4}(q_{4}\overbar{\chi}_{12,x}(\Bbox)):=:\overbar{\Lambda}_{12,\lambda'}(x)\mathsf{S}_{4}(\overbar{\chi}_{12,x}(\Bbox)):
\end{equation}
and imposing the condition
\bea
\frac{F^{\D4^{-}}_{12}[\lambda+\Bbox]}{F^{\D4^{-}}_{12}[\lambda]}=-\frac{\underset{x'=\bar{\chi}_{12,x}(\Abox)}{\Res}{x'}^{-1}\frac{\overbar{\mathscr{Y}}^{12}_{\lambda,x}(q_{123}x')}{\overbar{\mathscr{Y}}^{12}_{\lambda,x}(q_{12}x')} }{\underset{x'=q_{4}\bar{\chi}_{12,x}(\Abox)}{\Res}{x'}^{-1}\frac{\overbar{\mathscr{Y}}^{12}_{\lambda+\Abox,x}(q_{123}x')}{\overbar{\mathscr{Y}}^{12}_{\lambda+\Abox,x}(q_{12}x')}}
\eea
we obtain
\begin{equation}
    [\mathsf{T}^{(0|1)}_{12}(x),\mathscr{Q}_{4}(x')]=0.
\end{equation}
The condition actually implies $F^{\D4^{-}}_{12}[\lambda]=\widetilde{\mathcal{Z}}_{12}^{\D4^{-}}[\lambda]$. This can be shown by doing a similar analysis in Appendix \ref{app:D4U1partitionfunction}.


\bibliographystyle{ytamsalpha.bst}
\bibliography{Worigami}
\end{document}